\documentclass[journal]{IEEEtran}

\usepackage{bm}
\usepackage{amsmath}
\usepackage{epsf}
\usepackage{graphics}
\usepackage{ amssymb }
\usepackage[dvips]{graphicx}
\usepackage{mathtools}
\usepackage{epsfig}
\usepackage{cite}
\usepackage[linesnumbered,ruled,vlined]{algorithm2e}
\usepackage{colortbl}
\usepackage{color}
\usepackage{enumitem}
\usepackage{makecell}
\usepackage{soul,xcolor}
\usepackage{bm}
\usepackage{amsmath}
\usepackage{epsf}
\usepackage{graphics}
\usepackage{ amssymb }
\usepackage[dvips]{graphicx}
\usepackage{epsfig}
\usepackage{cite}
\usepackage[linesnumbered,ruled,vlined]{algorithm2e}
\usepackage{graphicx}
\usepackage{epsfig}
\usepackage{latexsym}
\usepackage{amsfonts}
\usepackage{here}
\usepackage{rawfonts}

\usepackage[utf8]{inputenc}
\usepackage[english]{babel}
\usepackage{amsmath}
\usepackage{amsfonts}
\usepackage{amssymb}
\usepackage{color} 
\usepackage{bm}
\usepackage{listings}
\usepackage{caption}
    \usepackage{multicol}

\usepackage{tabularx,booktabs}
\newcolumntype{Y}{>{\centering\arraybackslash}X}
\makeatletter
\newcommand\notsotiny{\@setfontsize\notsotiny{6.31415}{7.1828}}
\makeatother
\usepackage{amssymb}
\usepackage{amsthm}
\usepackage{graphicx}
\usepackage{epstopdf}
\usepackage{listings}
\usepackage{float}
\usepackage{amsmath}
\usepackage{amssymb}
\usepackage{amsfonts}
\usepackage{epstopdf}

\usepackage{multirow}
\usepackage{amscd}
\usepackage{mathrsfs}
\usepackage{graphicx}
\usepackage{color}
\usepackage{url}
\usepackage{bm}
\usepackage{setspace}
\usepackage{footnote}
\usepackage{xcolor}
\usepackage{svg}
\usepackage{dsfont}

\DeclareMathOperator*{\argmin}{arg\,min}
\newtheorem{theorem}{Theorem}
\newtheorem{lemma}{Lemma}

\newtheorem{definition}{Definition}
\newtheorem{proposition}{Proposition}
\newtheorem{corollary}{Corollary}

\newtheorem{remark}{Remark}

\newtheorem{assumption}{Assumption}

\usepackage[]{algorithm2e}

\addto\captionsenglish{}
\usepackage{bbm}

\usepackage{multicol}
\usepackage{mathtools}

\usepackage{lipsum,graphicx,subcaption}

\usepackage{diagbox}
\usepackage{hyperref}
\usepackage[protrusion=true,expansion=true]{microtype}
\pdfoutput=1
\usepackage[font=footnotesize]{caption}
\allowdisplaybreaks[4]

\usepackage{graphicx}
\usepackage{grffile}
\usepackage{tabularx}
\newcounter{term}[section]

\renewcommand\theterm{\alph{term}}
\makeatletter
\newcommand{\vast}{\bBigg@{4}}

\newcommand{\Vast}{\bBigg@{5}}
\makeatother

\usepackage{soul,xcolor}

\usepackage[most]{tcolorbox}
\usepackage{nicefrac, xfrac}

\definecolor{myblue}{RGB}{240,255,240}
\tcbset{enhanced,
      colframe=green!10!black,
      colback=myblue,
      boxsep=0pt}
   
        \allowdisplaybreaks

\begin{document}
\title{Dynamic D2D-Assisted Federated Learning  over O-RAN:  
Performance Analysis, MAC Scheduler,\\ and Asymmetric User Selection}
\author{Payam Abdisarabshali,~\IEEEmembership{Student~Member,~IEEE}, Kwang Taik Kim,~\IEEEmembership{Senior~Member,~IEEE},
\\ Michael Langberg,~\IEEEmembership{Fellow,~IEEE}, Weifeng Su,~\IEEEmembership{Fellow,~IEEE}, and Seyyedali Hosseinalipour,~\IEEEmembership{Senior Member,~IEEE}
 \IEEEcompsocitemizethanks{\IEEEcompsocthanksitem P. Abdisarabshali, M. Langberg, W. Su, and S. Hosseinalipour are with the Department of Electrical Engineering, University at Buffalo-SUNY, NY, USA (emails: \{payamabd, mikel, weifeng, alipour\}@buffalo.edu). K. T. Kim is with the department of Electrical and Computer Engineering, Purdue University, IN, USA (email:kimkt@purdue.edu))\\
 This work was supported by the U.S. National Science Foundation (NSF) under Grant ECCS 2512911.\\
Corresponding Author: S. Hosseinalipour}
}

\maketitle
\setulcolor{red}
\setul{red}{2pt}
\setstcolor{red}


\begin{abstract}
Existing studies on federated learning (FL) are mostly focused on system orchestration for~\textit{static~snapshots}~of~the~network~and~making~\textit{static control decisions} (e.g., spectrum allocation). 
However, real-world wireless networks are susceptible to \textit{temporal variations} of wireless channel capacity and users' datasets. In this paper, we study the impacts of the dynamics of (i) wireless channels and (ii) users' datasets on the FL execution. The former is captured by introducing a set of discrete time events while the latter is characterized by a novel~\textit{ordinary differential equation} and the metric of~\textit{dynamic~model~drift},~formulated~via~a~\textit{partial~differential~inequality}, drawing concrete analytical connections between
the dynamics of users' datasets and FL accuracy. 
We then propose \underline{d}ynamic \underline{c}ooperative F\underline{L} with dedicated \underline{M}AC schedulers ({\tt DCLM}), exploiting the unique features of open radio access network (O-RAN) to execute FL. {\tt DCLM} entails (i) a hierarchical device-to-device (D2D)-assisted model training, (ii) dynamic~control~decisions~through~dedicated~O-RAN~MAC~schedulers,~and~(iii) asymmetric~user~selection.~We provide extensive theoretical analysis to study the convergence~of~{\tt DCLM} and then aim to optimize its degrees of freedom (e.g., user selection and spectrum allocation) through a  non-convex optimization problem. 
We develop a systematic and generic approach to obtain the solution for this problem.
We finally show the efficiency of {\tt DCLM} via numerical simulations and provide a series of future directions.
\end{abstract}
\vspace{-3.8mm}
\begin{IEEEkeywords}
Federated learning, system dynamics, user selection, open RAN, MAC scheduler, performance analysis.
\end{IEEEkeywords}

\IEEEpeerreviewmaketitle
\section{Introduction}\label{sec:intro}
\noindent\IEEEPARstart{F}{ederated} learning (FL)
is an emerging distributed machine learning (ML) paradigm that enables mobile users to collaboratively train models \textit{without sharing raw data}, making it valuable for various Internet-of-Things (IoT) applications, such as smart healthcare and autonomous driving~\cite{xu2023edge}.
FL involves model training through multiple global rounds, each comprising: (i) a server broadcasts an ML model, called \textit{global model} (GM), to users; (ii) users perform \textit{local model} (LM) training using their  data and the received GM (e.g., via the stochastic gradient descent (SGD)); (iii) users transfer their trained LMs to the server, and the server aggregates the received LMs (e.g., via  weighted averaging) into a new GM.\par 


\subsection{Motivation and Solution Overview} 
{
FL is expected to be implemented on resource-constrained devices (e.g., smartphones and IoT sensors)~\cite{8737464}, 
posing challenges due to modern ML models' large parameter sizes, which demand significant local computation power and incur high communication overhead during model transfers.
Motivated by this, research has focused on (i) handling device' computation/communication heterogeneities~\cite{yang2022federated, su2025joint}; (ii) optimizing resource (e.g., spectrum) allocation~\cite{10061474,9337204,8664630,8737464} (iii) optimizing the  periodicity of client-to-server communications~\cite{10061474,9345723}; (iv) smart client sampling~\cite{9345723, 9337204}; and (v) leveraging inter-user cooperation to reduce training costs (e.g., energy usage) \cite{10061474}. 

Despite the significance of the above works, they share a common yet restrictive assumption, which we term \textit{snapshot-based analysis}: they assume that users' datasets and client-to-server wireless channels remain static throughout each global round of FL. Consequently, they perform snapshot-based analysis of the network, treating each snapshot as a representation of the network and users at the beginning of each global round. They then perform control decisions, such as spectrum allocation and user recruitment for each snapshot.

However, such snapshot-based analysis undermines the practicality of the existing studies, as in real-world scenarios, wireless channels and user datasets experience \textit{temporal fluctuations} during each FL global round, demanding \textit{adaptive ML and wireless control strategies}.
In this paper, for the first time, we relax this key assumption by incorporating (i) dynamic wireless channels of users and (ii) dynamic datasets of users into FL. We capture the former by introducing a set of events that occur at discrete time instants, which we refer to as $\mathscr{D}$-Events. To capture the later, we develop two models. First, we model the dynamics of each user's dataset, in terms of dataset size, via an \textit{ordinary differential equation}. Second, through a \textit{partial differential inequality}, we introduce a new metric called \textit{dynamic model drift} to measure the impact of temporal variation of users' datasets on FL accuracy. 
\textit{We demonstrate theoretically that these dynamics impact the resource allocation decisions, user sampling/recruitment, training latency, and energy consumption, as well as FL model convergence.} These theoretical results highlight the importance of performing dynamic resource allocation -- a task conducted by MAC schedulers of radio access networks (RANs) -- tailored to the requirements of FL. Based on these findings, we introduce \underline{d}ynamic \underline{c}ooperative F\underline{L} with dedicated \underline{M}AC schedulers ({\tt DCLM}), which will be deployed at RAN, to perform dynamic resource (e.g., spectrum/power) allocation and user recruitment.



 We note that traditional RANs, such as distributed RAN used in 4G networks rely on \textit{black-box MAC schedulers} to perform dynamic spectrum/power allocation for all types of users~\cite{9076124}. These black-box schedulers are proprietary, meaning only the hardware vendors can modify them, which limits the network’s flexibility and adaptability to diverse service requirements. In the context of FL, effective resource management must account for FL-specific factors, such as model prediction accuracy and client recruitment costs, alongside QoS metrics such as energy consumption and model training latency. These FL-specific considerations and limitations of traditional RANs in addressing them motivate us to develop our {\tt DCLM} over Open RAN (O-RAN), an emerging RAN used in 5G-and-beyond networks.

O-RAN (see Fig. \ref{fig:ORAN}) disaggregates elements (i.e., 3GPP protocol stack) of each base station (BS) into (i) radio unit (O-RU), (ii) distributed unit (O-DU), and (iii) centralized unit (O-CU)~\cite{9846950}. O-RAN provides a software-based infrastructure that allows the design of network functionalities tailored to the specific resource requirements of various services, a feature known as \textit{programmability}. For instance, users with ultra-reliable low-latency communication requirement can benefit from a tailored MAC scheduler. O-RAN enables this level of flexibility through (i) RAN slicing (creating virtual RANs for each service) and (ii) RAN intelligent controllers, which host applications known as rApps and xApps. These applications define network management policies based on service requirements of  users of each slice and the network conditions (e.g., channel gains). A detailed explanation on the evolution of RAN  is provided in Appendix~\ref{apx:ran_evolution}. In this paper, we formulate {\tt DCLM} as an optimization that can be encapsulated into rApps and xApps and deployed at O-RAN to establish MAC scheduling and user recruitment policies tailored to the requirements of FL (e.g., FL model accuracy), dynamics of wireless channels, and time-varying datasets of users.}


{\tt DCLM} orchestrates FL at two levels of \textit{agility} through breaking down each FL global round into multiple smaller time instants,
called \textit{fine granular time instants} (FGTIs). In particular, (i) ML-related control decisions are performed at each global round, e.g., recruiting FLUs, and (ii) dynamic wireless spectrum/power allocation are performed at each FGTI. 

\vspace{-2.5mm}
\subsection{Related Work}\label{related_work}
\subsubsection{FL and O-RAN}
To predict data traffic of O-RAN slices, \cite{9771700} proposes a slicing-based user selection FL algorithm. Further, a hierarchical FL-based resource allocation scheme and network slicing for O-RAN has been proposed in~\cite{qiao2025resource}.
However, these works differ from ours since they \textit{utilize FL for O-RAN orchestration} while we use the \textit{potentials of O-RAN to facilitate the execution of FL} in dynamic wireless networks.
\vspace{-0.15mm}

\subsubsection{MAC schedulers} 
Authors in~\cite{zhao2025adaslicing} propose AdaSlicing, an adaptive O-RAN slicing framework with real-time AI-driven orchestration. Further, \cite{dai2024ran} proposes an O‑RAN scheduling framework that uses inter-slice control via xApps and slice-aware MAC scheduling to meet users' service level agreements. Extending this literature, we are the first to design dedicated MAC schedulers for the execution of FL over O-RAN. 
\vspace{-0.15mm}
\subsubsection{User recruitment}
Several studies \cite{9345723, 9337204} investigate improving the efficiency of FL via proper user selection. Specifically, \cite{9345723} chooses users that exhibit a low degree of data heterogeneity and \cite{9337204} groups users with similar data types under one user to achieve near i.i.d. data distribution.
 Our work contributes to this literature by proposing a methodology for FL user recruitment over O-RAN under system dynamics. 
\vspace{-0.15mm}
\subsubsection{FL system orchestration}
Authors in \cite{10061474} proposed an FL system orchestrator aimed at minimizing energy consumption via spectrum allocation and device-to-device (D2D) communications.
Another research \cite{guo2021dynamic} employs a sequential wireless resource scheduling by establishing a queue to determine the order in which users transmit their models to the server.
To our knowledge, the existing FL orchestrations, including those mentioned, neither consider O-RAN nor orchestrate FL under system dynamics, which we achieve both in this work.
\subsubsection{Decentralized FL}
{ Authors in \cite{hu2019decentralized} propose a gossip approach for decentralized FL, where each user partitions its large ML model and shares different segments with neighboring users during FL training rounds. Similarly, \cite{Heged} highlights the effectiveness of gossip learning under uniform data distribution across the clients. 
Our work introduces the first D2D-assisted FL orchestration framework over O-RAN, which extends the existing decentralized FL approaches by proposing multi-timescale system orchestration and dedicated MAC schedulers while accounting for fine-granular system dynamics.

\newcommand\sbullet[1][.5]{\mathbin{\vcenter{\hbox{\scalebox{#1}{$\bullet$}}}}}

\begin{figure}[t]
\centering
\includegraphics[width=.40\textwidth,trim=20 5 20 25,clip]{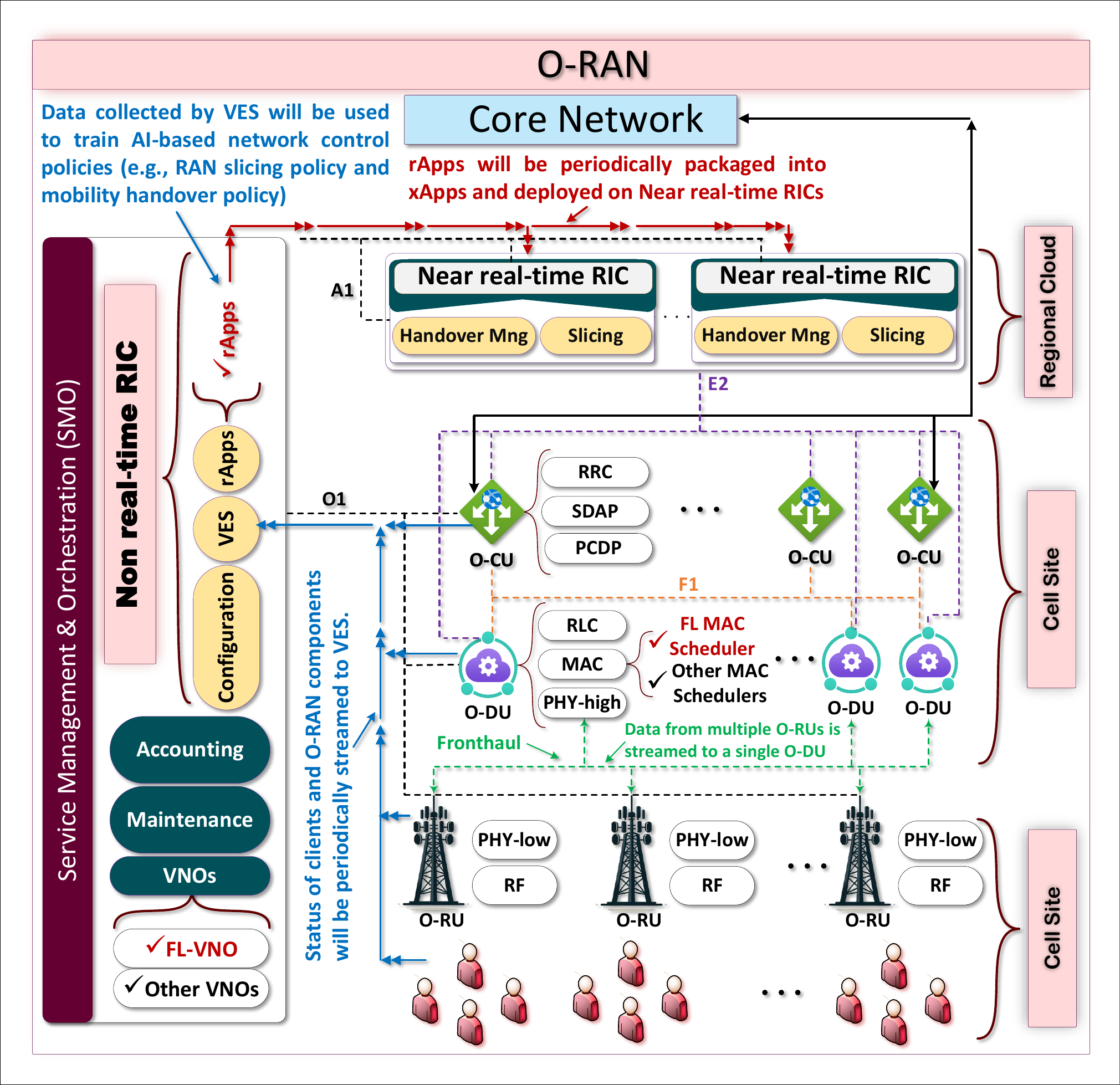}
\vspace{-2.5mm}
\caption{O-RAN architecture and integration of dedicated FL-related functionalities.  O-RAN components interact via standard open interfaces, labeled with E2, F1, open fronthaul, A1, and O1 interfaces, which facilitate interoperability between network elements from different manufacturers.}
\label{fig:ORAN}
\vspace{-1.1mm}
\end{figure}
 
\vspace{-2.2mm}
\subsection{Summary of Contributions}\label{sec:summary_of_contribution}
Our main contributions can be summarized as follows:
\begin{description}[font=$\bullet$~\normalfont,leftmargin=4.3mm]
 \itemsep-0.04em 
    \item We incorporate the dynamics of wireless channels of FL users (FLUs) and their datasets into FL orchestration. To this end, we first introduce a discrete-time model comprising a set of $\mathscr{D}$-Events. Afterwards, we model the dynamics of each FLU's dataset size as an \textit{ordinary differential equation} and introduce a metric named \textit{dynamic model drift}, which is modeled as a \textit{partial differential inequality}. These analytical models  draw one of the first connections between the impact of temporal variation of FLUs' datasets and FL accuracy.
    \item Utilizing the potential of O-RAN, we introduce {\tt DCLM}, a D2D-assisted multi time-scale system orchestrator for FL operating at two levels of agility: (i) global FL training rounds (coarse time-scale) and (ii) FGTIs (fine time-scale).
     \item In {\tt DCLM}, we introduce (i) the dynamic hierarchical FL, (ii) an asymmetric FLU recruitment strategy, and (iii) dispersed collaborative/cooperative communication (DCC) for FL. 
    \item Leveraging the programmability of O-RAN, we design dedicated MAC schedulers for {\tt DCLM} to conduct dynamic spectrum and transmit power allocation for FLU-to-O-RU (i.e., uplink/downlink) and FLU-to-FLU (i.e., D2D) communications tailored to dynamic wireless channel and datasets of FLUs during FGTIs.
    \item We conduct extensive theoretical analysis of {\tt DCLM} and derive a set of new FL convergence bounds.
    \item We provide both theoretical and experimental answers to the following fundamental open research questions:
    \begin{enumerate}[label={\textbf{(Q\arabic*})},leftmargin=6.3mm]
        \item \label{Q1} How dynamic wireless channel, time-varying datasets of FLUs, non-i.i.d. FLUs' datasets, and the dissimilarity of data points within the dataset of an FLU jointly influence decisions regarding FLU recruitment?
        \item \label{Q3} What is the impact of dynamic wireless control decisions on (i) the FL model performance, (ii) model training latency, and (iii) energy consumption?
    \end{enumerate}
    \item We~demonstrate~the~efficiency~of~{\tt DCLM}~in~terms~of~ML model performance and network resource savings via numerical simulations considering real-world ML datasets. 
\end{description}


\vspace{-1mm}
\section{System Model and Design Consideration}\label{sec:system_model}
\noindent  Below, we describe the main aspects of our system model. 
The major notations and acronyms used in the paper are listed in Table~\ref{table:notations} and Table~\ref{table:acronyms} in Appendix~\ref{app:notaions}. Also, one of our goals is to introduce a variety of future directions for this underexplored field of research, which is done in Appendix~\ref{app:future_research}. 

\begin{figure}[t]
 \centering
\includegraphics[width=8.0cm,trim=8 8 8 8,clip]{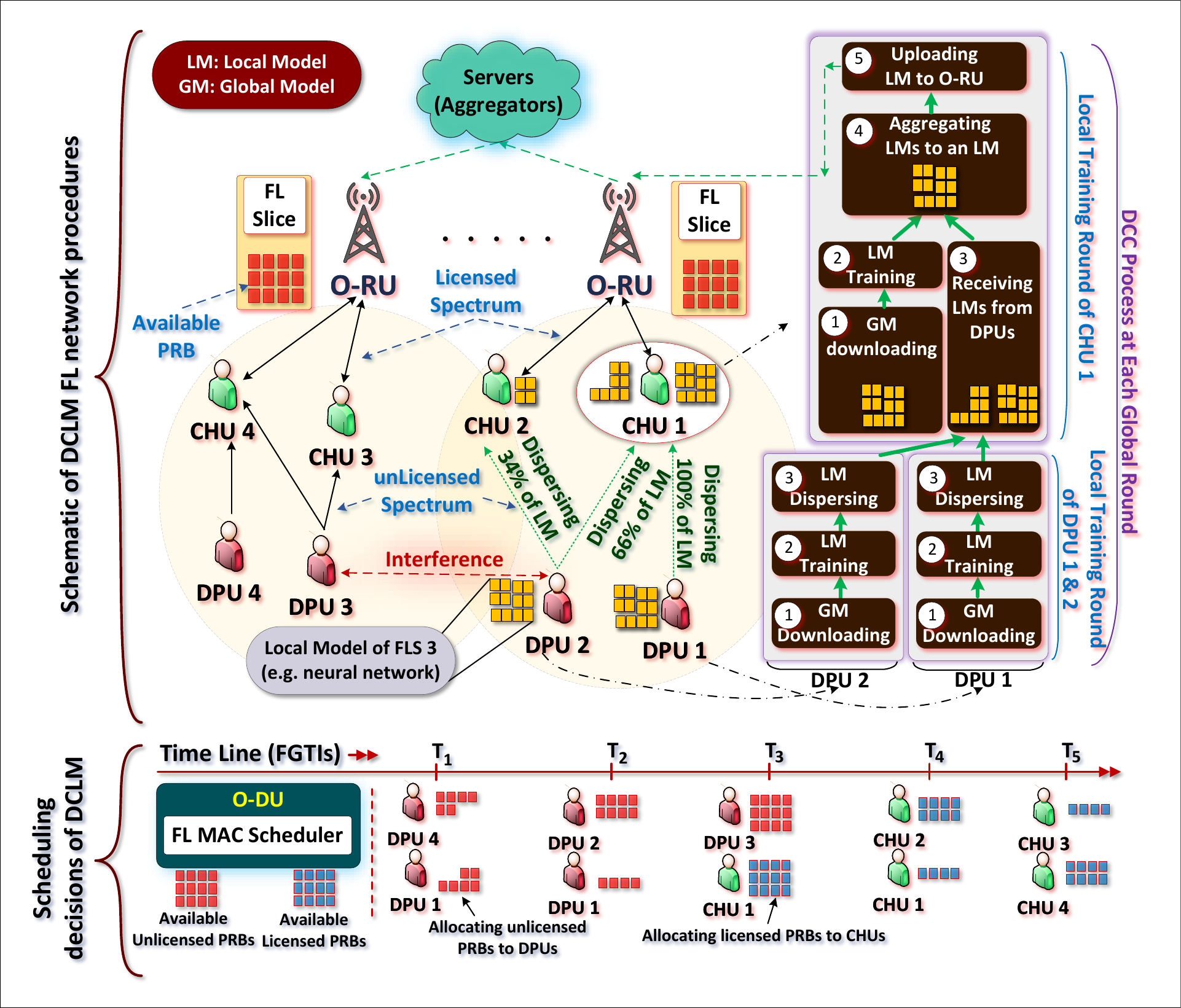} 
\caption{The FL MAC scheduler establishes multi-channel connections between DPUs, CHUs, and O-RUs. {\tt DCLM} uses DCC to perform a hierarchical ML model transmission/aggregation from DPUs to CHUs to O-RUs, and the server.}
\label{fig:systemModel}
\end{figure}

\vspace{-1mm}
\subsection{System Architecture}\label{sec:service_heterogeneity}

{ Referring to Fig.~\ref{fig:systemModel}, we consider an \textit{FL service}\footnote{For research directions on the coexistence of multiple FL services with conflicting interests refer to Appendix~\ref{app:future_research}.} \cite{9082655} performed by a set of geo-distributed FLUs across various O-RUs in O-RAN, gathered by set ${\Omega}$. The model training is coordinated by a server through a sequence of global rounds~$k\in\mathcal{K}=\{0,\cdots,K\}$. The set of FLUs within the coverage of each O-RU $b{\in}\Omega$ is collected via  $\mathcal{U}_{b}$. We assume a prior knowledge on the presence of FLUs in the system, supported by prior works \cite{8737464, su2025joint, 10061474, guo2021dynamic}. Further extensions to stochastic FLU arrival and departure are discussed in Appendix~\ref{app:future_research}. We consider \textit{heterogeneous communication/computation} capabilities across FLUs, where an FLU $u{\in}\mathcal{U}_{b}$ has a maximum transmit power of $P^{\mathsf{max}}_{u}$ and time-varying available CPU frequency of $f^{(k)}_{u}$.

We consider data transmissions scheduled in group(s) of $12$ sub-carriers, known as physical resource blocks (PRBs). To support  heterogeneous services, 5G new radio adopts a set of \textit{numerologies}~\cite{7744816} that tune the bandwidth of PRBs. Each numerology is characterized by an integer~$\gamma$, enforcing the subcarrier spacing to be $2^{\gamma}\times 15 \textrm{ KHz}$, where $15$ KHz is the reference subcarrier spacing. In this paper, we consider two types of PRBs, including (i) licensed PRBs used for communication between O-RUs and FLUs and (ii) unlicensed PRBs exploited for  device-to-device (D2D) communications (i.e., FLU-to-FLU communications). We refer to the bandwidth of licensed and unlicensed PRBs of each O-RU by $B{=}12{\times} 15 {\times} 2^{\gamma_1}$ KHz and $\overline{B}{=}12{\times} 15 {\times} 2^{\gamma_2}$ KHz, respectively, where $\gamma_1$ and $\gamma_2$ are their respective numerologies. Moreover, we consider a realistic scenario where multiple PRBs can be assigned to each user, enabling high-speed communications. 


We consider a RAN slice for FL at each O-RU $b$, provisioned with a maximum downlink transmit power of $P^{\mathsf{max}}_b$. Further, we presume the worst-case scenario where FL slices of all O-RUs have access to the same sets of licensed and unlicensed PRBs denoted by $\mathcal{R}_{b}$ and $\overline{\mathcal{R}}_{b}$, $b{\in} \Omega$. Thus, FLUs experience (i) intra-cell interference from other FLUs in the same O-RU that use the same PRBs and (ii) inter-cell interference from FLUs using the same PRBs in nearby O-RUs. This calls for adaptive interference management and transmit power control, which we develop in this work.\footnote{Research directions on \textit{elastic RAN slicing} (i.e., scale down/up of resources provisioned to O-RAN slices) are provided in Appendix~\ref{app:future_research}.} 
We model both downlink and uplink communications for FL. In the downlink, the server broadcasts the GM to recruited FLUs via licensed PRBs of O-RUs. In the uplink, to upload FLUs' LMs, we propose the \textit{dispersed collaborative/cooperative communication (DCC)} mode (see Fig.~\ref{fig:systemModel}), splitting FLUs into (i) communication head FLUs (CHUs) and (ii) deprived FLUs (DPUs). In particular, DPUs disperse unequal fractions of their LMs to multiple CHUs via low-power D2D communications on unlicensed PRBs instead of communicating directly with O-RUs. Afterward, CHUs aggregate their LMs with those received from DPUs into a single LM, which is transferred to O-RUs via licensed PRBs. DCC conserves DPUs' resources (e.g., transmit power) and network resources (e.g., licensed PRBs), as shown in Sec.~\ref{numerical_evaluation}.}

\vspace{-1mm}
\subsection{System Dynamic}\label{sec:system_dynamic}

In {\tt DCLM}, we migrate from the snapshot-based analysis of FL, described in Sec.~\ref{sec:intro}, by taking the first step toward incorporating system dynamics into FL, which in turn calls for \textit{dynamic control decisions}. Specifically, to capture the system dynamics in {\tt DCLM}, for each global round $k{\in}\mathcal{K}$, let $[T^{(k{-}1)}, T^{(k)}]$ be a wall-clock time window (measured in seconds), where $T^{(k)}$ denotes the time when global round $k$ is finished and $T^{(-1)}{=}0$. 
We consider two types of multi-granular dynamics that impact the performance of FL: (M1) discrete-time dynamic FLUs' resource requirements (e.g., wireless spectrum) and (M2) continuous-time dynamic FLUs' local datasets. Below, we focus on (M1) and  model (M2) later in Sec.~\ref{sec:local_dataset}.


\textbf{Discrete-time dynamic FLUs' resource requirements.}
For brevity, we refer to the processes associated with transmitting ML models (i.e., uploading LMs of FLUs and downloading the GM) as communication tasks (CTs).
Further, we define the following variables: (i) $\tau_{b}^{\downarrow,{(k)}}$: Broadcast (downlink) latency of O-RU $b$ to transmit the GM to the recruited FLUs over the licensed PRBs; (ii) $\tau_{u}^{\mathsf{LC},{(k)}}$: Local computation/training time of FLU $u{\in} \mathcal{U}_b$; (iii) $\overline{\tau}_{u}^{\uparrow,{(k)}}$: LM upload latency of DPU $u{\in} \mathcal{U}_b$ over the unlicensed PRBs; (iv) $\tau_{u}^{\mathsf{W},{(k)}}$: Waiting time of CHU $u{\in} \mathcal{U}_b$ for receiving all LMs from its associated DPUs in DCC; (v) $\tau_{u}^{\uparrow,{(k)}}$: LM upload latency of CHU $u{\in} \mathcal{U}_b$ over the licensed PRBs. Accordingly, below, we introduce a set of \textit{discrete-time events}, called $\mathscr{D}$-Events, which cover all the discrete-time events in which the resource requirement of the users may experience a change.\footnote{We presume a deterministic scenario, where $\mathscr{D}$-Events are non-probabilistic and discuss research directions on probabilistic systems in Appendix~\ref{app:future_research}.} We further define a function $\Psi(\cdot){\in}[T^{(k{-}1)},T^{(k)})$, which takes a $\mathscr{D}$-Event as an input and outputs its occurrence time. 
\vspace{-0.5mm}
\begin{enumerate}[label={\textbf{(E\arabic*)}}] 
\itemsep-0.05em 
    \item \label{E4} $\mathscr{D}_{u}^{\downarrow,(k)}$: Downlink CT arrival event of FLU $u{\in} \mathcal{U}_b$ (i.e., when $u$ starts receiving the GM from its O-RU $b$), the occurrence time of which is $\Psi(\mathscr{D}_{u}^{\downarrow,(k)}){=}T^{(k{-}1)}$ since it only happens at the start of each global round. 
    \item \label{E5} $\mathscr{D}_{u}^{\ndownarrow,(k)}$: Downlink CT departure event of FLU $u{\in} \mathcal{U}_b$ (i.e., when $u$ receives the GM completely from its O-RU), which occurs at time $\Psi(\mathscr{D}_{u}^{\ndownarrow,(k)}){=}\Psi(\mathscr{D}_{u}^{\downarrow,(k)}){+}\tau_{b}^{\downarrow,{(k)}}$.
    \item \label{E6} $\overline{\mathscr{D}}_{u}^{\uparrow,(k)}$: Uplink CT arrival event of DPU $u\in \mathcal{U}_b$ over the unlicensed PRBs (i.e., when the DPU starts uploading its LM to its associated CHU), the occurrence time of which is $\Psi(\overline{\mathscr{D}}_{u}^{\uparrow,(k)})=\Psi(\mathscr{D}_{u}^{\ndownarrow,(k)})+\tau_{u}^{\mathsf{LC},{(k)}}$.
    \item \label{E7} $\overline{\mathscr{D}}_{u}^{\nuparrow,(k)}$: Uplink CT departure event of DPU $u\in \mathcal{U}_b$ over the unlicensed PRBs (i.e., when $u$ finishes uploading its LM to its associated CHU), the occurrence time of which is $\Psi(\overline{\mathscr{D}}_{u}^{\nuparrow,(k)})=\Psi(\overline{\mathscr{D}}_{u}^{\uparrow,(k)})+\overline{\tau}_{u}^{\uparrow,{(k)}}$.
    \item \label{E8} $\mathscr{D}_{u}^{\uparrow,\mathsf{C},(k)}$: Uplink CT arrival event of CHU $u\in \mathcal{U}_b$ over the licensed PRBs (i.e., when $u$ starts uploading its aggregated LM to its associated O-RU), which occurs at time $\Psi(\mathscr{D}_{u}^{\uparrow,(k)})=\Psi(\mathscr{D}_{u}^{\ndownarrow,(k)})+\max\{\tau_{u}^{\mathsf{LC},{(k)}}, \tau_{u}^{\mathsf{W},{(k)}}\}$.\footnote{$\max\{\tau_{u}^{\mathsf{LC},{(k)}},\tau_{u}^{\mathsf{W},{(k)}}\}$ captures the wait bottleneck: CHU $u$ must finish LM training and wait for all DPUs’ LMs before uploading it aggregated LM.}
    \item \label{E9} $\mathscr{D}_{u}^{\nuparrow,\mathsf{C},(k)}$: Uplink CT departure event of CHU $u{\in} \mathcal{U}_b$ over the licensed PRBs (i.e., when the CHU finishes uploading its aggregated LM to its associated O-RU), the occurrence time of which is $\Psi(\mathscr{D}_{u}^{\nuparrow,(k)})=\Psi(\mathscr{D}_{u}^{\uparrow,(k)})+\tau_{u}^{\uparrow,{(k)}}$.
\end{enumerate}
$\mathscr{D}$-Events build the pillars of analysis of the system dynamics in {\tt DCLM}. These events can cause fluctuations in the resource requirements of FL slices during model training. For example, occurrence of $\mathscr{D}_{u}^{\downarrow,(k)}$ can increase the licensed PRB requirement of FL slices, while $\mathscr{D}_{u}^{\ndownarrow,(k)}$ may decrease it. 
$\mathscr{D}$-Events further motivate us to incorporate a novel notion into FL, where \textit{virtualization} and \textit{RAN slicing} in O-RAN provide a unique platform for harnessing $\mathscr{D}$-Events. This notion is \textit{dynamic MAC scheduler}, which allocates resources (e.g., PRB and transmit power) of each virtually crafted RAN slice to the FLUs. Details of formalizing this notion will be presented in Sec.~\ref{sec:dynamic_control_decisions}.


While $\mathscr{D}$-Events make the modeling more realistic, they significantly increase the complexity of controlling network elements for FL.
{\tt DCLM} 
addresses this complexity by employing a multi-time scale management system. Specifically, in
{\tt DCLM}, we split the duration of each global round $k$ (i.e., {\small$[T^{(k{-}1)}, T^{(k)})$}) into multiple \textit{fine granular time instants} (FGTIs) to tackle system dynamics caused by $\mathscr{D}$-Events, allowing~for~a~continuum~of decisions (e.g., dynamic PRB allocation by MAC schedulers).
Formally, we divide {\small$[T^{(k{-}1)}, T^{(k)})$} into multiple FGTIs and introduce decision variable {\small$t_x{\in}[T^{(k{-}1)}, T^{(k)})$} to refer to a specific FGTI. Here, {\small$x{\in}\mathcal{N}^{(k)}$}, where {\small $\mathcal{N}^{(k)}{=}\{\hspace{-0.5mm}(N^{(k{-}1)}{+}1),\cdots\hspace{-0.5mm},N^{(k)}\hspace{-0.5mm}\}$}, is the set of {\small$N^{(k)}{-}N^{(k{-}1)}{>}0$} integer values (i.e., the indices of FGTIs); {\small$N^{(-1)}=0$}. 
Also {\small$t_{x{+}1}>t_{x}$}, {\small$\forall x\in\cup_{k\in\mathcal{K}}\mathcal{N}^{(k)}\setminus\{N^{(K)}\}$}. The impact of $t_x$ on system performance and optimization of its occurrence time are presented in Secs.~\ref{sec:dynamic_control_decisions}\&\ref{sec:computation_communication_latencies_energy} and~\ref{sec:PF}.

\vspace{-1mm}
\subsection{ Integration of {\tt DCLM} into O-RAN}\label{sec:intergation}
\vspace{-.01mm}
O-RAN enables the development of dedicated MAC schedulers through \textit{non-real-time RAN intelligent controllers (Non-RT RICs)}, which define network management policies. In particular, data (e.g., users' positions), generated by users and the RAN nodes (e.g., O-RUs) stream periodically via the O1 interface to the virtual event streaming (VES). The collected data is utilized to build non-RT applications known as rApps (see Fig.~\ref{fig:ORAN}), which handle long-term tasks such as training ML models and solving optimization problems. These rApps will be packaged into applications known as xApps and deployed on near-RT RICs. A key use case of xApp is defining resource allocation policy of the MAC schedulers at O-DUs. Specifically, unlike black box MAC schedulers used in traditional RANs (e.g., distributed RAN) for all services, xApps offer a flexible white box, software-based infrastructure for designing a dedicated MAC scheduler tailored to a specific network service. \textit{In {\tt DCLM}, we formulate a dedicated MAC scheduler for FL as an optimization problem in Sec.~\ref{sec:PF}, which will be packaged within an rApp in non-RT RIC (see Fig.~\ref{fig:ORAN}). The solution of this problem will subsequently be packaged into an xApp to define MAC scheduling policy}. This MAC scheduler manages three tasks: (i) scheduling O-RUs at different FGTIs and assigning licensed PRBs for broadcasting the GM to FLUs; (ii) managing D2D operations by scheduling DPUs at different FGTIs and allocating unlicensed PRBs to the scheduled DPUs for dispersing their LMs across CHUs; and (iii) scheduling CHUs at different FGTIs and allocating licensed PRBs to the scheduled CHUs for transmitting their aggregated LMs to O-RUs. To facilitate understanding, we visualize the protocol flow of {\tt DCLM} in Fig.~\ref{fig:protocol_flow} in Appendix~\ref{app:protocol_flow}, illustrating the interactions among the {\tt DCLM} MAC scheduler, O-RUs, CHUs, and DPUs.
In the following, we present the steps involved in {\tt DCLM}.

\section{{\tt DCLM} Distributed Machine Learning Model}\label{sec:FL_model}
\noindent 
In {\tt DCLM}, each global round $k\in\mathcal{K}$ engages/recruits a subset of FLUs (including DPUs and CHUs). To determine the recruitment status of FLU $u$, we utilize $\widehat{\lambda}_{u}^{(k)}{=}\lambda_{u}^{(k)}{+}\overline{\lambda}_{u}^{(k)},~\forall u{\in} \mathcal{U}_{b}$, where $\lambda_{u}^{(k)}{\in}\{0,1\}$ and $\overline{\lambda}_{u}^{(k)}{\in}\{0,1\}$ are  decision variables (optimized in~Sec.~\ref{sec:PF}): FLU $u$ is a CHU if $\lambda_{u}^{(k)}{=}1$, it is a DPU if $\overline{\lambda}_{u}^{(k)}{=}1$, and
it is not recruited at global round $k$ if $\widehat{\lambda}_{u}^{(k)}{=}0$. To prevent FLU $u$ from being recruited as both a CHU and DPU simultaneously, we impose $\widehat{\lambda}_{u}^{(k)}{=}\lambda_{u}^{(k)} {+} \overline{\lambda}_{u}^{(k)} {\le} 1$. 

In {\tt DCLM}, the following seven steps are conducted in each global round $k$. {(s-i)\label{FL_step1}} \textit{Broadcasting a GM}: the server broadcasts a GM, characterized by model parameter $\bm{\omega}^{(k)}\in \mathbb{R}^M$, to the recruited FLUs through O-RUs. {(s-ii)\label{FL_step2}} \textit{LM training}: each recruited CHUs and DPUs initializes its LM with the received GM and starts training its LM using its dataset via stochastic gradient descent (SGD). {(s-iii)\label{FL_step3}} \textit{Dispersion of gradient vector (GV) of DPUs}: each recruited DPU disperses unequal fractions/parts of its GV among CHUs (optimized in Sec. \ref{sec:GM_GPS_transmissions}). {(s-iv)\label{FL_step4}} \textit{GV aggregation at CHUs}: each recruited CHU aggregates the GVs received from its associated DPUs and its own GV into one GV. {(s-v)\label{FL_step5}} \textit{Dispatching of GVs of CHUs}: each recruited CHU uploads different fractions of its aggregated GV to an O-RU over multiple PRBs (CHU to O-RU assignment and PRB allocation are formulated in Sec. \ref{sec:GM_GPS_transmissions}). {(s-vi)\label{FL_step6}} \textit{GV aggregation at O-RUs}: each O-RU aggregates the received GVs from its associated CHUs into one GV and transfers it to the server. {(s-vii)\label{FL_step7}} \textit{GM aggregation}: the server aggregates the received GVs from O-RUs to a new GM $\bm{\omega}^{(k+1)}$, utilized for round $k{+}1$.


Below, we formulate different aspects of {\tt DCLM}: dynamics of FLUs' local datasets (Sec. \ref{sec:local_dataset}), distributed training objective of {\tt DCLM} (Sec.~\ref{sec:GM}), LM training (Sec.~\ref{sec:LM}), LM dispersion by DPUs (Sec.~\ref{sec:dispersion}), and aggregation procedure (Sec.~\ref{sec:aggregation}).
\vspace{-1mm}
\subsection{Continuous-Time Dynamic FLUs' Local Datasets}\label{sec:local_dataset}
In contrast to the majority of works on FL that assume FLUs' datasets are static\cite{10061474,9337204,8664630,8737464}, we consider a realistic scenario where FLUs' datasets change over time when not engaged in model training. Let $T^{\mathsf{train},(k)}_u{\triangleq} [\Psi(\mathscr{D}_{u}^{\ndownarrow,(k)}),\Psi(\mathscr{D}_{u}^{\ndownarrow,(k)}){+}\tau_{u}^{\mathsf{LC},{(k)}}]$ be the training time window of FLU $u$ at global round~$k$ and $T^{\mathsf{idle},(k)}_u{\triangleq} [T^{(k{-}1)}, T^{(k)}) {\setminus} T^{\mathsf{train},(k)}_u$ denote the times that FLU $u$ is idle (i.e., it does not conduct training). Below, we characterize the continuous-time dynamics of FLUs' datasets (in terms of size and distribution) (Sec. \ref{sec:dynamic_dataset}), dynamic global/local loss functions (Sec. \ref{sec:global_local_loss}), and dynamic model drift (Sec. \ref{sec:model_drift}).
\subsubsection{Dynamic dataset}\label{sec:dynamic_dataset}
We model the dynamics of dataset size of FLU $u$ during global round $k$ (i.e., during time window {\small$[T^{(k{-}1)}, T^{(k)})$}) via the following differential equation:
\begin{equation}\label{eq:dynamic_dataset_size}
\frac{d|\Upsilon_{u}(t)|}{dt}{=}
\begin{cases}
   G^{(k)}_u(t)&t{\in}T^{\mathsf{idle},(k)}_u,\\
   0&t{\in}T^{\mathsf{train},(k)}_u,
\end{cases}
\end{equation}
where {\small$\Upsilon_{u}(t)$} refers to the collected dataset of FLU $u$ with {\small$|\Upsilon_{u}(t)|$} data points at time {\small$t{\in}[T^{(k{-}1)}, T^{(k)})$}. Further, in \eqref{eq:dynamic_dataset_size}, {\small$G^{(k)}_u(t){<}0$} ({\small$G^{(k)}_u(t){>}0$}) implies a decrease (increase) in the dataset size over time. We also use {\small$\Upsilon(t){=}\cup_{b\in \Omega,u\in\mathcal{U}_{b}}\Upsilon_u(t)$} and {\small$\Upsilon^{\mathsf{s}}(t){=}\cup_{b\in \Omega,u\in\mathcal{U}_{b}}\widehat{\lambda}_{u}^{(k)} \Upsilon_u(t)$} to refer to the cumulative dataset of all FLUs and recruited FLUs at time $t$, respectively. 

We also model the change in the distribution of data points in the FLUs' datasets over time (e.g., changes in the distribution of collected data of sensors across different days). To this end, we use {\small$\xi{=}(\bm{\xi}, y)\in \Upsilon_{u}(t)$} to denote an arbitrary data point with feature {\small$\bm{\xi}$} and label $y$, and then define {\small$\sigma_{u}(t){=}\sqrt{\frac{1}{|\Upsilon_{u}(t)|{-}1}\sum_{\xi\in \Upsilon_{u}(t)} \Vert \bm{\xi}{-} \bm{\mu}_{u}(t)\Vert^2}$} and {\small $\bm{\mu}_{u}(t){=}\sum_{\xi\in \Upsilon_{u}(t)} {\bm{\xi}}\big/{|\Upsilon_{u}(t)|}$} as the time-varying standard deviation and the mean of data points inside $\Upsilon_{u}(t)$, respectively.

\subsubsection{Dynamic local/global loss functions}\label{sec:global_local_loss}
Let ${f}(\bm{\omega},\xi)$ denote the \textit{loss function} (see Table 1 in~\cite{8664630}) that quantifies the model performance for data point {\small$\xi {\in} \Upsilon_{u}(t)$} under ML model {\small$\bm{\omega}{\in}\mathbb{R}^{M}$} with $M$ parameters. Also let {\small$\mathfrak{L}_{u}(\bm{\omega}| \Upsilon_u(t)){=}\sum_{\xi {\in} \Upsilon_{u}(t)} {{f}(\bm{\omega},\xi)}\big/{|\Upsilon_{u}(t)|}$} denote the local loss at FLU $u$. The global loss function at time $t$ is then given by:
\begin{equation}
\begin{aligned}
  \mathfrak{L}(\bm{\omega}| \Upsilon(t)){=}\sum_{b\in\Omega}\sum_{u\in \mathcal{U}_{b}} \frac{|\Upsilon_{u}(t)|}{|\Upsilon(t)|} \mathfrak{L}_{u}(\bm{\omega}| \Upsilon_u(t)).
\end{aligned}
\end{equation}

We next use a measure called model drift to capture how the dynamics of datasets formulated in Sec. \ref{sec:dynamic_dataset} affect the accuracy of an ML model in terms of local loss $\mathfrak{L}_{u}(\bm{\omega}| \Upsilon_u(t))$.
\subsubsection{Dynamic model drift}\label{sec:model_drift}
Few works on FL have looked into model drift \cite{ganguly2023online} (e.g., to capture the variation of data between two consecutive global rounds). However, none of them can be applied to {\tt DCLM} as the continuous-time dynamics of FLUs' datasets (Sec. \ref{sec:dynamic_dataset}) lead to a continuous-time variation of data throughout the entire FL processes (i.e., within $[0, T^{(K)}]$). Motivated by this, we introduce a new definition for model drift, which is in the form of a \textit{partial differential inequality}. This allows us to accurately quantify the impact of fluctuations in the data distribution \textit{during} each global round $k$. Our later analysis in Sec.~\ref{sec:conv} will highlight the impact of this metric on the performance of {\tt DCLM} and reveal a set of novel insights.

\vspace{-0.8mm}
\begin{definition} [Dynamic Model Drift]\label{def:cons}
Let $\bm{\omega}^{(k)}{\in}\mathbb{R}^M$ be a GM with $M$ elements. We formulate dynamic model drift of FLU $u$ at global round $k$ via a partial differential inequality:
\begin{equation}\label{eq:conceptDrift}
\hspace{-4mm}
 \displaystyle \frac{\partial}{\partial t}\left(\frac{|\Upsilon_u(t)|}{|\Upsilon(t)|}\mathfrak{L}_u(\bm{\omega}^{(k)}|\Upsilon_u{(t)})\right){\leq} \mathfrak{D}_u(t),~t{\in} T^{\mathsf{idle},(k)}_u.
\hspace{-4mm}
\end{equation}
\end{definition}
As model drift value {\small$\mathfrak{D}_u(t)$} increases, it becomes difficult to track the best model parameters since the loss function rapidly increases. Moreover, \eqref{eq:conceptDrift} also includes situations where the old model becomes more fitting for current data when {\small$\mathfrak{D}_u(t){<}0$}. The definition in \eqref{eq:conceptDrift} generalizes the discrete model drift in~\cite{ganguly2023online} and extend it to capture the rate of change in model performance over continuous time, allowing for the use of continuous time function approximations in capturing the model drift. This is particularly suitable for applications that generate streaming data with rapidly evolving distributions~\cite{jin2021budget}.


Considering the above system dynamics, in {\tt DCLM}, we aim to perform dynamic ML training as described below.


\vspace{-1mm}
\subsection{ Distributed Training Objective of {\tt DCLM}}\label{sec:GM}
At the $k$-th global round, let {\small $\mathfrak{L}_{u}^{(k)}(\bm{\omega}){\triangleq}  \mathfrak{L}_{u}(\bm{\omega}| \Upsilon_u(T^{(k)}))$} denote the local loss of FLU $u$  and {\small$\mathfrak{L}^{(k)}(\bm{\omega}){\triangleq}  \mathfrak{L}(\bm{\omega}| \Upsilon(T^{(k)}))$} represent the global loss for an ML model {\small$\bm{\omega}$} at time {\small$T^{(k)}$}. Due to the dynamics of local datasets, the optimal global model is time-varying and forms a sequence {\small $\big\{\bm{\omega}^{{(k)}^\star}\big\}_{k=1}^{K}$}, where {\small$\bm{\omega}^{{(k)}^\star}={\argmin} \;  \mathfrak{L}^{(k)}(\bm{\omega}), \bm{\omega}\in \mathbb{R}^M,~\forall k\in\mathcal{K}$}. To track this sequence while taking into account the energy consumption and training latency of FL,  {\tt DCLM} leverages O-RAN programmability and introduces a dedicated MAC scheduler for FL to perform optimal dynamic resource (i.e., PRB and transmit power) allocation. Note that, as dynamic model drift {\small$\mathfrak{D}_u(t)$} across FLUs increases, the training procedure (including LM training and model transmitting) must be conducted faster to track the optimal global model. Also, devices with larger values of {\small$\mathfrak{D}_u(t)$} should be recruited more frequently since neglecting them will lead to having an outdated global model. \textit{These intuitions are later concretized and verified by our theoretical analysis} in Sec.~\ref{sec:conv}.


\vspace{-1mm}
\subsection{Local Model Training} \label{sec:LM}
To obtain~{\small$\big\{\bm{\omega}^{{(k)}^\star}\big\}_{k=1}^{K}$}, 
the server broadcasts a 
GM (i.e., {\small$\bm{\omega}^{(0)}\in\mathbb{R}^{M}$}) among recruited FLUs, followed by broadcasting of {\small$\bm{\omega}^{(k)}\in\mathbb{R}^{M}$} at the beginning of each global round $k\in\mathcal{K}$.
Let {\small$\widetilde{\tau}_{b}^{\downarrow,{(k)}}{=}T^{(k{-}1)}{+}\tau_{b}^{\downarrow,{(k)}}$} refer to the time when GM broadcasting of O-RU $b$ at global round $k$ is finished. At  {\small$\widetilde{\tau}_{b}^{\downarrow,{(k)}}$}, each recruited FLU $u\in\mathcal{U}_b$ first synchronizes its LM with {\small$\bm{\omega}^{(k)}$}. It then trains it LM through {\small$\ell^{(k)}_{u}$} SGD iterations over mini-batches of data points from its dataset denoted by {\small$\mathcal{B}^{\ell}_{u}(\widetilde{\tau}_{b}^{\downarrow,{(k)}})\subseteq \Upsilon_{u}(\widetilde{\tau}_{b}^{\downarrow,{(k)}})$},\footnote{Note that according to \eqref{eq:dynamic_dataset_size}, FLU $u$ has a stationary dataset {\small$\Upsilon_{u}(\widetilde{\tau}_{b}^{\downarrow,{(k)}})$} during model training (i.e., {\small$T^{\mathsf{train}}_u$}).} where {\small$\ell\in\{1,\cdots,\ell^{(k)}_{u}\}$} is the SGD iteration index.\footnote{To cope with computation heterogeneities of FLUs, we presume that FLUs conduct different numbers of SGD iterations under various mini-batch sizes.} Each mini-batch {\small$\mathcal{B}^{\ell}_{u}(\widetilde{\tau}_{b}^{\downarrow,{(k)}})$} is obtained via random data sampling without replacement. Further, the mini-batch size is determined by {\small$|\mathcal{B}^{\ell}_{u}(\widetilde{\tau}_{b}^{\downarrow,{(k)}})|{=}{B}_{u}(\widetilde{\tau}_{b}^{\downarrow,{(k)}})=\varsigma^{(k)}_u |\Upsilon_{u}(\widetilde{\tau}_{b}^{\downarrow,{(k)}})|$},~$\forall \ell$, where {\small$\varsigma^{(k)}_u\in(0,1]$} represents the fraction of local dataset contained in the mini-batch. Letting {\small$\eta_{_k}$} denote the step-size, the evolution of LM of FLU $u$ is given by: 
\begin{equation}\label{eq:WeightupdateStrat}
\hspace{-13.5mm}
\resizebox{0.42\textwidth}{!}{  $
   \bm{\omega}_{u}^{(k),\ell}\hspace{-0.5mm}{=}\bm{\omega}^{(k),\ell-1}_{u}\hspace{-0.5mm} {-} {\eta_{_k}} \hspace{-0.5mm}\sum_{\xi\in \mathcal{B}^{\ell}_{u}(\widetilde{\tau}_{b}^{\downarrow,{(k)}})} \hspace{-1mm} {{\nabla  f_{u}(\bm{\omega}^{(k),\ell-1}_{u},\xi)}\big/{{B}_{u}(\widetilde{\tau}_{b}^{\downarrow,{(k)}})}}\hspace{-0.1mm},
    \hspace{-9mm}$}
    \vspace{-0.5mm}
\end{equation}
where $\bm{\omega}^{(k),0}_{u}=\bm{\omega}^{(k)}$ is the GM received from the server. 

\vspace{-1mm}
\subsection{Dynamic Gradient Vector (GV) Dispersion by DPU}\label{sec:dispersion}

After completing {\small$\ell^{(k)}_{u}$} SGD iterations, as described above, each FLU $u$ first calculates its cumulative GV as: {\small$\widetilde{\nabla {\mathfrak{L}}}_{u}^{(k)} {=} \hspace{-.1mm}(\bm{\omega}_u^{(k),0}{-}\bm{\omega}_{u}^{(k),\ell^{(k)}_{u}})\big/\eta_{_k}=\hspace{-.1mm}(\bm{\omega}^{(k)}{-}\bm{\omega}_{u}^{(k),\ell^{(k)}_{u}})\big/\eta_{_k}$}. 
Each DPU $u'$ then \textit{partitions} its GV {\small$\widetilde{\nabla {\mathfrak{L}}}_{u'}^{(k)}$} into multiple \textit{chunks}, parallelly transferred to multiple CHUs in D2D mode. Specifically, let {\small$\widetilde{\nabla {\mathfrak{L}}}_{u',u,r}(t_x){\in}\mathbb{R}^{M}$} denote a gradient chunk of DPU $u'\in\mathcal{U}_{b}$ received at CHU $u$ over PRB $r$ at FGTI $t_x$, with {\small$\overline{\psi}_{u',u,r}(t_x) M$} elements of which selected from {\small$\widetilde{\nabla {\mathfrak{L}}}_{u'}^{(k)}$} and zero elsewhere.\footnote{Note that when transmitting {\small$\widetilde{\nabla {\mathfrak{L}}}_{u',u,r}(t_x)$}, only non-zero elements and their indices are sent, reducing communication overhead.} 
 Here, {\small$\overline{\psi}_{u',u,r}(t_x){\in}[0,1]$} is a decision variable (optimized in Sec.~\ref{sec:PF}) indicating the proportion of the gradient of DPU $u'$ that will be sent to CHU $u$ over PRB $r$.
Note that, for each $u'\in\mathcal{U}_{b}$, we have {\small$\sum_{x\in\mathcal{N}^{(k)}}\sum_{u\in \mathcal{U}_b}\sum_{r\in \overline{R}_{b}}\widetilde{\nabla {\mathfrak{L}}}_{u',u,r}(t_x)=\widetilde{\nabla {\mathfrak{L}}}_{u}^{(k)}$}, and {\small$\sum_{x\in\mathcal{N}^{(k)}}\sum_{u\in \mathcal{U}_b}\sum_{r\in \overline{R}_{b}} \overline{\psi}_{u',u,r}(t_x) =1$}.



\vspace{-1mm}
\subsection{Temporal and Hierarchical FL Model Aggregation}\label{sec:aggregation}
Considering DCC depicted in Fig.~\ref{fig:systemModel}, {\tt DCLM} establishes a 4-tier uplink communication hierarchy. The server is situated at the top (layer~1), O-RUs are located at layer~2, followed by CHUs at layer~3, and DPUs at layer~4. Subsequently, at each global round $k$, the process of uploading local GVs in {\tt DCLM} occurs in three hierarchical stages, starting with DPUs to CHUs, then from CHUs to O-RUs, and finally from O-RUs to the server. In this scheme, components of each layer utilize a designated aggregation strategy to aggregate GVs received from the lower layer into a single GV, which will be transmitted to the upper layer. Specifically, as explained below, the aggregation process starts with performing \textit{dynamic fragmented local GVs aggregation} at CHUs during FGTIs of global round $k$ (i.e., layer~3).\footnote{Here, the term ``{dynamic}" refers to gradual GVs aggregation carried out by CHUs during FGTIs, while the term ``{fragmented}" refers to the fact that GVs of each DPU is spread across multiple CHUs.} Then, \textit{integrated local GVs aggregation} is done at O-RUs (i.e., layer~2). Lastly, \textit{global model aggregation} is performed at the server (i.e., layer~1). 

\subsubsection{Dynamic fragmented local GVs aggregation}\label{sec:model_aggregation}
Each CHU $u$ aggregates its own GV and the GVs received from its DPUs into an aggregated GV  as follows:
\begin{equation}\label{eq:normalLocalCHU}
\hspace{-6mm}
\begin{aligned}
    &\widetilde{\nabla {\mathfrak{L}}}_{u}^{\mathsf{frag},(k)}\hspace{-1mm}=\underbrace{\big(\lambda_{u}^{(k)}\times|\Upsilon_{u}(\widetilde{\tau}_{b}^{\downarrow,{(k)}})|\times\widetilde{\nabla {\mathfrak{L}}}_{u}^{(k)}\big)\big/\ell^{(k)}_{u}}_{(a)}\\[-.3em]
    &{+}\hspace{-2.5mm}\underbrace{\sum_{x{\in}\mathcal{N}^{(k)}}\hspace{-0.5mm}\sum_{u'\in \mathcal{U}_{b}}\hspace{-0.5mm}\sum_{r{\in}\overline{\mathcal{R}}_{b}}\hspace{-2mm}\Big(\hspace{-0.5mm}\overline{\lambda}_{u'}^{(k)} {\times}|\Upsilon_{u'}(\widetilde{\tau}_{b}^{\downarrow,{(k)}})|{\times}\widetilde{\nabla {\mathfrak{L}}}_{u',u,r}(t_x)\hspace{-0.5mm}\Big)\hspace{-0.7mm}{
    \Big/}\ell^{(k)}_{u'}}_{(b)}, 
\end{aligned}
\hspace{-4.5mm}
\vspace{-0.3mm}
\end{equation}
where $(a)$ is the scaled cumulative GV of CHU $u$ during local training and $(b)$ accounts for the scaled cumulative GV of DPUs that transmit a portion of their GV to CHU $u$ during FGTIs of global round $k$. In \eqref{eq:normalLocalCHU}, a normalization is done based on SGD iterations $\ell^{(k)}_{u}$ and $\ell^{(k)}_{u'}$ to compensate for the variety in the numbers of FLUs' SGDs. This prevents model bias towards FLUs that conduct higher SGD iterations~\cite{9521822}. Each CHU $u$ then transmits {\small$\widetilde{\nabla \mathfrak{L}}_{u}^{\mathsf{frag},(k)}$} to its associated O-RU.
\subsubsection{Integrated local GV aggregation}
Each O-RU $b$ aggregates its CHUs' GVs into an integrated\footnote{The term ``integrated" captures that all fragments of GVs of DPUs dispersed among CHUs will be bring together at O-RUs.} GV as follows:
\begin{equation}\label{eq:normalLocalORU}
\resizebox{0.36\textwidth}{!}{  $
        \widetilde{\nabla {\mathfrak{L}}}_{b}^{(k)}=(1/|\Upsilon^{\mathsf{s}}(\widetilde{\bm{\tau}}^{\downarrow,{(k)}})|)\sum_{u\in \mathcal{U}_{b}}\lambda_{u}^{(k)}\widetilde{\nabla {\mathfrak{L}}}_{u}^{\mathsf{frag},(k)},$}
\end{equation}
where {\small$|\Upsilon^{\mathsf{s}}(\widetilde{\bm{\tau}}^{\downarrow,{(k)}})|{=}\sum_{b\in \Omega}\sum_{u\in \mathcal{U}_{b}}\widehat{\lambda}_{u}^{(k)}|\Upsilon_{u}(\widetilde{\tau}_{b}^{\downarrow,{(k)}})|$} is the cumulative dataset size of recruited FLUs, where  {\small$\widetilde{\bm{\tau}}^{\downarrow,{(k)}}{=}[\widetilde{\tau}_{b}^{\downarrow,{(k)}}]_{b\in \Omega}$}. Each O-RU $b$ then transmits {\small$\widetilde{\nabla {\mathfrak{L}}}_{b}^{(k)}$} to the server.
\subsubsection{Global model aggregation}
The server performs global aggregation, to form a global GV as {\small$\widetilde{\nabla {\mathfrak{L}}}^{(k)}{=} \mathfrak{B}_{k}\sum_{b{\in} \Omega}\widetilde{\nabla {\mathfrak{L}}}_{b}^{(k)}$}, where {\small$\mathfrak{B}_{k}$} is a boosting term which can accelerate the training. The server then obtains the new GM {\small$\bm{\omega}^{(k+1)}$} as follows:
\begin{equation}\label{eq:mainupdateWeight}
    \bm{\omega}^{(k+1)}=\bm{\omega}^{(k)}-\eta_{_k}\widetilde{\nabla {\mathfrak{L}}}^{(k)},
    \vspace{-1.1mm}
\end{equation}
which is broadcast to recruited FLUs of round $k{+}1$. We next focus on the wireless control decisions in {\tt DCLM}, detailing \textit{how} model transmissions occur during the above processes.

\vspace{-0.5mm}
\section{Dynamic Wireless Control Decisions}\label{sec:dynamic_control_decisions}
\noindent 
To orchestrate the procedure outlined in the previous section while taking into account the latency and energy overhead of FL in a dynamic environment, dynamic wireless control decisions become essential. However, $\mathscr{D}$-Events complicate these decisions by perturbing the instantaneous resource demands of FLUs. To address this challenge, {\tt DCLM} formalizes a suite of dynamic wireless control decisions (Fig. \ref{fig:dynamic_control_decision}), including dynamic MAC scheduler (Sec.~\ref{sec:mac_scheduler}) and dynamic GM/GV transmissions (Sec.~\ref{sec:GM_GPS_transmissions}). These decisions are later  optimized in Sec.~\ref{sec:PF}.

In a nutshell, the MAC scheduler is in charge of performing fine-granular uplink/downlink and D2D resource allocation at each FGTI (i.e., licensed/unlicensed PRB and transmit power allocation). Also, dynamic GM/GV transmissions using DCC mode enable controlling the broadcast of GM from O-RUs to all  FLUs, dispersion of local GVs of DPUs across CHUs, and dispatching of local GVs of CHUs to O-RUs. 

\begin{figure}[t]
 \centering
\noindent\includegraphics[width=8.9cm,trim= 25 30 5 90, clip]{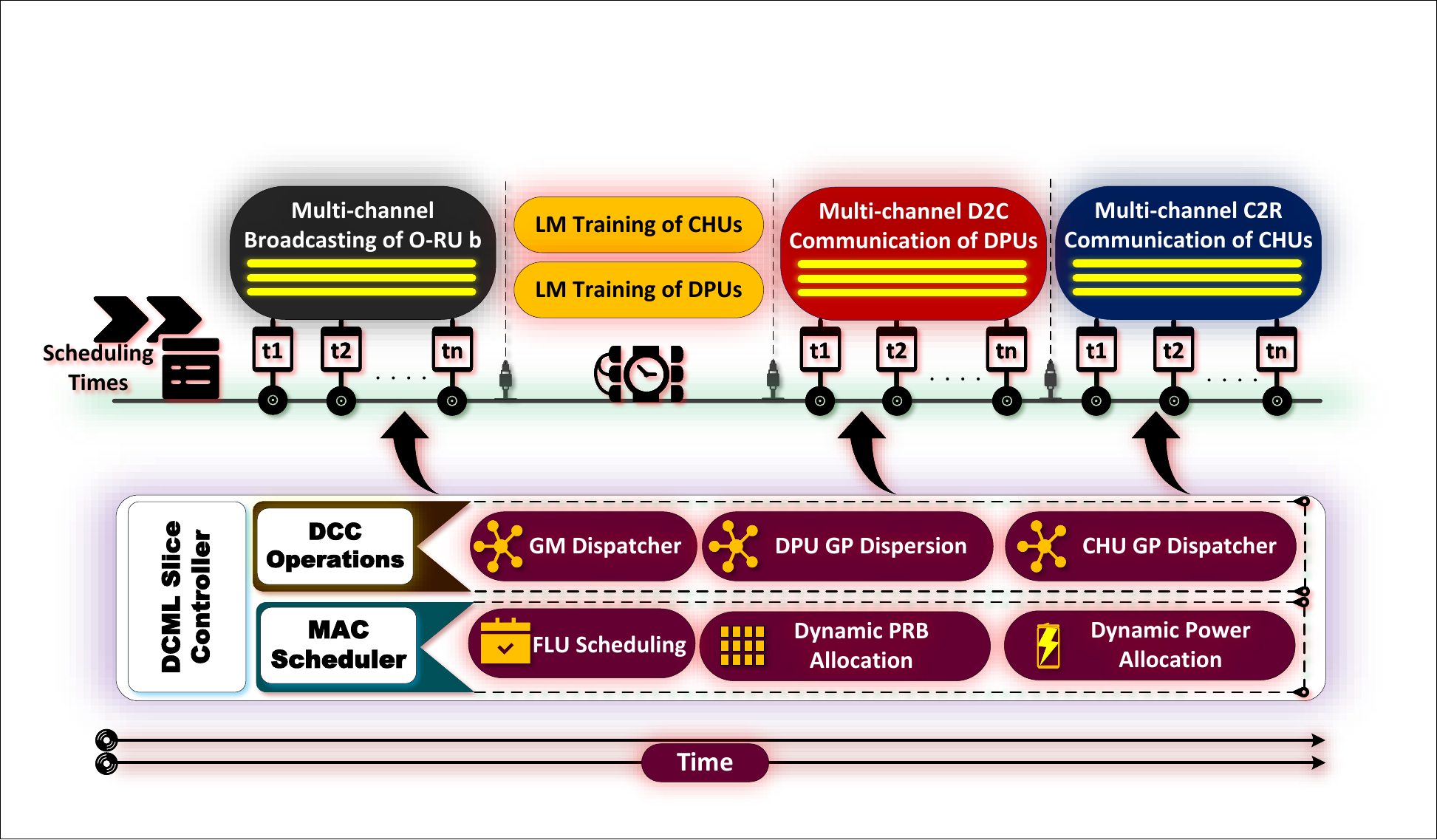} 
\vspace{-5.75mm}
\caption{The {\tt DCLM} slice controller utilizes DCC and a dedicated MAC scheduler over O-RAN to effectively schedule FLUs for model transmission.}
\label{fig:dynamic_control_decision}
\vspace{-2.1mm}
\end{figure}
\vspace{-1mm}
\subsection{Dynamic MAC Scheduler for an FL Slice}\label{sec:mac_scheduler}
In this section, we develop one of the first dedicated MAC scheduler for an FL slice $\mathcal{S}^{(k)}_{b}$ at O-RU $b$, which allocates licensed/unlicensed PRBs to FLUs and transmit power of FLUs and O-RUs to PRBs tailored to dynamic channel capacities during FGTIs of global round $k$. Specifically, the MAC scheduler operates in three main steps. First, it makes scheduling decisions (SDs) for downlink/uplink communication modes at each FGTI (Sec.~\ref{sec:scheduling_decisions}). Second, it allocates PRBs to the scheduled FLUs (Sec.~\ref{sec:resource_allocation_mac}). Third, it allocates transmit power to PRBs (Sec.~\ref{sec:power_allocation}). We next discuss these  steps.

\subsubsection{Scheduling decisions (SDs)}\label{sec:scheduling_decisions} 
We define three SD variables. (i) \textit{Broadcast scheduling decision variable $\beta^{\downarrow}_{b}\hspace{-0.3mm}(t_{x}){\in}\{0,1\}$}, determining whether O-RU $b$ is scheduled to broadcast the GM to the recruited FLUs over licensed PRBs of slice $\mathcal{S}^{(k)}_{b}$ at FGTI $t_{x}$. (ii) \textit{DPU scheduling decision variable $\overline{\beta}^{\uparrow}_u(t_x){\in}\{0,1\}$}, which determines whether DPU $u$ is scheduled to upload its LM over unlicensed PRBs at FGTI $t_x$. (iii) \textit{CHU scheduling decision variable $\beta^{\uparrow}_u(t_{x}){\in} \{0,1\}$}, specifying whether CHU $u$ is scheduled to upload its LM over licensed PRBs at FGTI~$t_{x}$.

To ensure proper scheduling, we next impose important SD constraints (SDCs) on the above SD variables.
\vspace{-0.5mm}
\begin{remark}
In the following, we will face conditions that can be expressed as if/else statements. Nevertheless, designing optimal control decisions with respect to if/else statements pushes the problem to the realm of integer programming, which is often intractable. Instead, as part of our contributions, we will demonstrate how to transform these conditions to a set of \textit{continuous non-convex representations} (CNRs), enabling us to develop a tractable approach to optimize the control decisions.    
\end{remark}
\vspace{-0.5mm}
\textbf{$\bm{\mathsf{SDC}_1}$: GM broadcasting by O-RUs.}  
SD variable $\beta^{\downarrow}_{b}\hspace{-0.3mm}(t_{x})$, $\forall x{\in}\mathcal{N}^{(k)}$, can only be $1$ (i.e., O-RU $b$ can broadcast the GM to the recruited FLUs) if at least one FLU is recruited from O-RU $b$ (i.e., $\max_{u\in\mathcal{U}_b}\{\widehat{\lambda}^{(k)}_{u}\}=1$) and $t_{x}{\in}[T^{(k{-}1)},T^{(k{-}1)}+\tau_{b}^{\downarrow,{(k)}})$.
We describe these conditions through the following CNR: 
\begin{equation}\label{cons:mac_1}
   \hspace{-1mm}
   \underbrace{\big(((T^{(k{-}1)}{+}\tau_{b}^{\downarrow,{(k)}}){/}t_{x})\max_{u\in\mathcal{U}_b}\{\widehat{\lambda}^{(k)}_{u}\}{-}1\big)}_{(a)}{\times} {\big(1{-}\beta^{\downarrow}_{b}\hspace{-0.3mm}(t_{x})\big)} {\leq}  0.
   \hspace{-1mm}
   \vspace{-1mm}
\end{equation}
In \eqref{cons:mac_1}, if {\small$t_{x}{\in}[T^{(k{-}1)},T^{(k{-}1)}{+}\tau_{b}^{\downarrow,{(k)}})$} and {\small$\max_{u{\in}\mathcal{U}_b}\{\widehat{\lambda}^{(k)}_{u}\}{=}1$}, term $(a)$ becomes positive, forcing {\small$\beta^{\downarrow}_{b}\hspace{-0.3mm}(t_{x}){=}1$}. However, \eqref{cons:mac_1} is not capable of imposing {\small$\beta^{\downarrow}_{b}\hspace{-0.3mm}(t_{x}){=}0$} if either {\small$t_x\ge T^{(k{-}1)}{+}\tau_{b}^{\downarrow,{(k)}}$ or $\max_{u{\in}\mathcal{U}_b}\{\widehat{\lambda}^{(k)}_{u}\}{=}0$}. To resolve this issue, we complement \eqref{cons:mac_1} with the following CNR if {\small$t_x{>}T^{(k{-}1)}{+}\tau_{b}^{\downarrow,{(k)}}$}:
\begin{equation}\label{cons:mac_2}
    \big(1{-}((T^{(k{-}1)}+\tau_{b}^{\downarrow,{(k)}})/t_{x})\max_{u\in\mathcal{U}_b}\{\widehat{\lambda}^{(k)}_{u}\}\big)\times\beta^{\downarrow}_{b}\hspace{-0.3mm}(t_{x})\le 0.
\vspace{-0.5mm}
\end{equation}
We also set $\beta^{\downarrow}_{b}\hspace{-0.3mm}(t_{x}){=}0$ if $t_{x}=T^{(k{-}1)}+\tau_{b}^{\downarrow,{(k)}}$.


\textbf{$\bm{\mathsf{SDC}_2}$: Local GV uploading by DPUs.} SD variable $\overline{\beta}^{\uparrow}_u(t_x)$, $\forall x{\in} \mathcal{N}^{(k)}$, can only be $1$ if FLU $u{\in} \mathcal{U}_{b}$ is recruited as a DPU (i.e., $\overline{\lambda}_{u}^{(k)}{=}1$) and $t_x$ is in the interval of LM uploading of DPU $u$ (i.e., $t_x\in(\Psi(\overline{\mathscr{D}}_{u}^{\uparrow,(k)}), \Psi(\overline{\mathscr{D}}_{u}^{\nuparrow,(k)}))$). We represent these conditions as the following CNR:
\begin{equation}\label{cons:mac_4}
\hspace{-2mm}
 \resizebox{0.45\textwidth}{!}{  $
   \big(\overline{\lambda}_{u}^{(k)}\Psi(\overline{\mathscr{D}}_{u}^{\nuparrow,(k)}){/}t_x{-}1\big)\big(1{-}\overline{\lambda}_{u}^{(k)}\Psi(\overline{\mathscr{D}}_{u}^{\uparrow,(k)}){/}t_x\big)\big(1{-}\overline{\beta}^{\uparrow}_u(t_x)\big){\leq} 0,$}
\hspace{-2mm}
\end{equation}
which can guarantee that $\overline{\beta}^{\uparrow}_u(t_x)=1$ at correct FGTIs $t_x$ and when $\overline{\lambda}_{u}^{(k)}{=}1$; however, it cannot enforce $\overline{\beta}^{\uparrow}_u(t_x)=0$ if DPU $u$ has not been recruited as a DPU ($\overline{\lambda}_{u}^{(k)}{=}0$), 
has not started its local training ($t_x\leq \Psi(\overline{\mathscr{D}}_{u}^{\uparrow,(k)})$), or
has already finished its LM uploading ($t_x\geq \Psi(\overline{\mathscr{D}}_{u}^{\nuparrow,(k)})$). We ensure these conditions through the following succinct CNR:
\begin{equation}\label{cons:mac_5}
\hspace{-5mm}
 \resizebox{0.46\textwidth}{!}{  $
 \max\big\{\overline{\lambda}_{u}^{(k)}\Psi(\overline{\mathscr{D}}_{u}^{\uparrow,(k)}){/}t_x{-}1,1{-}\overline{\lambda}_{u}^{(k)}\Psi(\overline{\mathscr{D}}_{u}^{\nuparrow,(k)}){/}t_x\big\}\overline{\beta}^{\uparrow}_u(t_x){\leq} 0.$}
 \hspace{-4mm}
\end{equation}
We also set $\overline{\beta}^{\uparrow}_u(t_x)=0$ if $t_x=\Psi(\overline{\mathscr{D}}_{u}^{\uparrow,(k)})$ or $t_x=\Psi(\overline{\mathscr{D}}_{u}^{\nuparrow,(k)})$.

\par\textbf{$\bm{\mathsf{SDC}_3}$: Local GV uploading by CHUs.} SD variable $\beta^{\uparrow}_u(t_{x})$, $\forall x{\in}\mathcal{N}^{(k)}$, can only be $1$ if FLU $u{\in} \mathcal{U}_{b}$ (i) has been recruited as a CHU (i.e., $\lambda_{u}^{(k)}{=}1$), (ii) has finished its local training, (iii) has received all of the LMs from its associated DPUs, and (iv) has not finished its upload operation, which we capture via
\begin{equation}\label{cons:mac_7}
   \hspace{-5mm} \resizebox{0.46\textwidth}{!}{$
   \big(\hspace{-0.5mm}\lambda_{u}^{(k)}\Psi(\mathscr{D}_{u}^{\nuparrow,(k)}){/}t_{x}{-}1\hspace{-0.5mm}\big)\big(\hspace{-0.5mm}1{-}\lambda_{u}^{(k)}\Psi(\mathscr{D}_{u}^{\uparrow,(k)}){/}t_{x}\hspace{-0.5mm}\big)\left(\hspace{-0.5mm}1{-}\beta^{\uparrow}_u(t_{x})\hspace{-0.5mm}\right){\leq}0,$}
   \hspace{-4mm}
\end{equation}
using the same interpretations as \eqref{cons:mac_4}. Similar to \eqref{cons:mac_5}, we restrict CHU $u$ from being scheduled at time $t_{x}$ via CNR
\begin{equation}\label{cons:mac_8}
\hspace{-4mm}
\resizebox{0.45\textwidth}{!}{  $
 \max\big\{\lambda_{u}^{(k)}\Psi(\mathscr{D}_{u}^{\uparrow,(k)}){/}t_{x}{-}1,1{-}\lambda_{u}^{(k)}\Psi(\mathscr{D}_{u}^{\nuparrow,(k)}){/}t_{x}\big\}\beta^{\uparrow}_u(t_{x}){\leq} 0.$}
\hspace{-4mm}
\end{equation}
We also set $\beta^{\uparrow}_u(t_x)=0$ if $t_x=\Psi(\mathscr{D}_{u}^{\uparrow,(k)})$ or $t_x=\Psi(\mathscr{D}_{u}^{\nuparrow,(k)})$.

Considering the above SDs and SDCs, we next present dynamic PRB allocation strategy of the MAC scheduler.

\subsubsection{Dynamic PRB allocation}\label{sec:resource_allocation_mac}
We define decision variable $\overline{\varrho}_{u,u',r}(t_x){\in}\{0,1\}$, $\forall x\in \mathcal{N}^{(k)}$, to denote whether unlicensed PRB $r {\in} \overline{\mathcal{R}}_{b}$ is allocated to DPU $u{\in}\mathcal{U}_{b}$ to upload a fraction of its LM to CHU $u'{\in}\mathcal{U}_{b}$ at time $t_x$ (i.e., $\overline{\varrho}_{u,u',r}(t_x){=}1$) or not (i.e., $\overline{\varrho}_{u,u',r}(t_x){=}0$). $\overline{\varrho}_{u,u',r}(t_x)$ can be $1$ only if (i) $u$ is recruited as a DPU (i.e., $\overline{\lambda}_{u}^{(k)}=1$), (ii) $u$ is scheduled at time $t_x$ for uplink  (i.e., $\overline{\beta}^{\uparrow}_u(t_x)=1$), and (iii) $u'$ is recruited as a CHU (i.e., $\lambda_{u}^{(k)}=1$), which we collectively capture via 
\begin{equation}\label{cons:mac_r_1}
        \overline{\varrho}_{u,u',r}(t_x){-}\min\big\{\overline{\lambda}_{u}^{(k)},\overline{\beta}^{\uparrow}_u(t_x),\lambda_{u}^{(k)}\big\} {\le} 0,~x{\in}\mathcal{N}^{(k)}.
\end{equation}
Further, {\small$\overline{\varrho}_{u,u',r}(t_x)$} can be $1$ if CHU $u'$ has not started its upload (i.e., when {\small$\beta^{\uparrow}_u(t_{z}){=}0, \forall z{\in}[\mathcal{N}^{(k{-}1)}{+}1,x]$}), thus we get:
\begin{equation}\label{cons:mac_r_2}
        \overline{\varrho}_{u,u',r}(t_x){+}\max_{z{\in}[\mathcal{N}^{(k{-}1)}{+}1,x]}\{\beta^{\uparrow}_u(t_{z})\} {\le} 1,~x{\in}\mathcal{N}^{(k)}.
\end{equation}
Moreover, to guarantee that there is at least one DPU dispersing its LM to CHU~$u'$, we introduce the following CNR:
\begin{equation}\label{cons:mac_r_3}
\resizebox{0.38\textwidth}{!}{  $
     \sum_{u{\in} \mathcal{U}_{b}}\sum_{r{\in}\overline{\mathcal{R}}_{b}} \overline{\varrho}_{u,u',r}(t_x)\ge 1,~x{\in}\mathcal{N}^{(k)},~ u'{\in} \mathcal{U}_{b}.
$}
\end{equation}
It is worth mentioning that $\overline{\varrho}_{u,u',r}(t_x)=1$ implicitly dictates the connectivity between DPU $u$ and CHU $u'$. 

We also introduce a decision variable $\varrho_{u,r}\hspace{-0.3mm}(t_{x}){\in}\{0,1\}$ to determine whether licensed PRB $r {\in} \mathcal{R}_{b}$ is allocated to CHU $u$ to transmit a fraction of its LM to O-RU $b$ at time $t_{x}$. For $\varrho_{u,r}\hspace{-0.3mm}(t_{x})$ to be $1$, $u$ must be recruited as a CHU (i.e., $\lambda_{u}^{(k)}{=}1$) and scheduled for uplink (i.e., $\beta^{\uparrow}_u(t_{x}){=}1$), captured via
\begin{equation}\label{cons:mac_r_4}
        \varrho_{u,r}(t_{x}){-}\min\big\{\lambda_{u}^{(k)},\beta^{\uparrow}_u(t_{x})\big\}{\le}0,~x{\in}\mathcal{N}^{(k)}.
\end{equation}
We next present power allocation strategy of the MAC scheduler.
\vspace{-1mm}
\subsubsection{Dynamic power allocation}\label{sec:power_allocation}
We formulate dynamic transmit power allocation under DCC for three communication modes utilized in {\tt DCLM}, including (i) O-RU to FLU (R2F), (ii) CHU to O-RU (C2R), and (iii) DPU to CHU (D2C). 

\textbf{O-RU to FLU ($\bm{\mathsf{R2F}}$).} We use a decision variable $\rho^{\downarrow}_{b,r}\hspace{-0.5mm}(t_{x}){\in} [0,1]$ to capture the percentage of the maximum transmit power of O-RU $b$, i.e., $P^{\mathsf{max}}_{b}$, utilized for GM broadcasting over PRB $r {\in} \mathcal{R}_{b}$ at global round $k$, where {\small$\sum_{r {\in} \mathcal{R}_{b}} \rho^{\downarrow}_{b,r}\hspace{-0.5mm}(t_{x}) {\le} 1$}. We define the following CNR to ensure that power allocation for R2F communication mode is done if and only if O-RU $b$ is scheduled for broadcasting the GM to the FLUs: 
\vspace{-.5mm}
\begin{equation}\label{cons:mac_p_1}
    \begin{aligned}
        0\le \beta^{\downarrow}_{b}\hspace{-0.3mm}(t_{x})-\rho^{\downarrow}_{b,r}\hspace{-0.5mm}(t_{x})<1,~x{\in}\mathcal{N}^{(k)}.
    \end{aligned}
\end{equation}

\textbf{CHU to O-RU ($\bm{\mathsf{C2R}}$).} We introduce a decision variable $\rho^{\uparrow}_{u,r}\hspace{-0.3mm}(t_{x}){\in}[0,1]$ to refer to the percentage of the transmit power of CHU $u{\in} \mathcal{U}_{b}$ (i.e., $P^{\mathsf{max}}_{u}$) allocated to PRB $r {\in} \mathcal{R}_{b}$ for uplink  at time $t_{x}$, where $\sum_{r {\in} \mathcal{R}_{b}} \rho^{\uparrow}_{u,r}\hspace{-0.3mm}(t_{x}) {\le} 1$. $\rho^{\uparrow}_{u,r}\hspace{-0.3mm}(t_{x})$ can take values larger than $0$ if PRB $r$ is allocated to CHU $u$ at FGTI $t_{x}$ (i.e., $\varrho_{u,r}\hspace{-0.3mm}(t_{x}){=}1$), captured via 
\vspace{-.5mm}
\begin{equation}\label{cons:mac_p_2}
    \begin{aligned}
      0\le \varrho_{u,r}\hspace{-0.3mm}(t_{x})-\rho^{\uparrow}_{u,r}\hspace{-0.3mm}(t_{x})< 1,~x{\in}\mathcal{N}^{(k)}.
    \end{aligned}
\end{equation}

\textbf{DPU to CHU ($\bm{\mathsf{D2C}}$).} We use a continuous decision variable $\overline{\rho}^{\uparrow}_{u,r}(t_x)\in[0,1]$ to capture the percentage of transmit power of DPU $u\in \mathcal{U}_{b}$ (i.e., $P^{\mathsf{max}}_{u}$) allocated to unlicensed PRB $r {\in} \overline{\mathcal{R}}_{b}$ for communication at time $t_x$, where $\sum_{r {\in} \overline{\mathcal{R}}_{b}} \overline{\rho}^{\uparrow}_{u,r}(t_x) {\le} 1$. $\overline{\rho}^{\uparrow}_{u,r}(t_x)$ can be larger than $0$ if and only if PRB $r$ is allocated to DPU $u$ (i.e., $\overline{\varrho}_{u,u',r}(t_x){=}1$), captured via
\vspace{-0.5mm}
\begin{equation}\label{cons:mac_p_3}
    \begin{aligned}
      0\le \overline{\varrho}_{u,u',r}(t_x)-\overline{\rho}^{\uparrow}_{u,r}(t_x)< 1,~x{\in}\mathcal{N}^{(k)}.
    \end{aligned}
\end{equation}
\vspace{-4mm}
\subsection{Dynamic GM and GV Transmissions}\label{sec:GM_GPS_transmissions}
\textbf{O-RUs GM broadcasting/dispatching.} We introduce a decision variable $\varphi_{b,r}\hspace{-0.3mm}(t_{x}){\in}[0,1]$ to capture the percentage of the GM that will be broadcast to FLUs over PRB $r$ from O-RU $b$ at~$t_{x}$. We define the following CNR to ensure that $\varphi_{b,r}\hspace{-0.3mm}(t_{x})>0$ if and only if O-RU $b$ is scheduled for broadcasting the GM: 
\vspace{-0.5mm}
\begin{equation}\label{cons:LGPD0}
    0\le \beta^{\downarrow}_{b}\hspace{-0.3mm}(t_{x})-\varphi_{b,r}\hspace{-0.3mm}(t_{x})<1,~x{\in}\mathcal{N}^{(k)}.
\end{equation}

\textbf{DPUs local GV dispersion.} We utilize a continuous decision variable $\overline{\psi}_{u,u',r}(t_x){\in}[0,1]$ to capture the percentage of GV of DPU $u{\in} \mathcal{U}_{b}$ offloaded to CHU $u'{\in} \mathcal{U}_{b}$ over unlicensed PRB $r {\in} \overline{\mathcal{R}}_{b}$ at time $t_x$. $\overline{\psi}_{u,u',r}(t_x)$ can be greater than zero only if PRB $r$ is allocated to DPU $u$, i.e., $\overline{\varrho}_{u,u',r}(t_x){=}1$. Similarly, if $\overline{\varrho}_{u,u',r}(t_x){=}1$, DPU $u$ should disperse a fraction of its LM to CHU $u'$ over PRB $r$, i.e., $\overline{\psi}_{u,u',r}(t_x)>0$. We capture these conditions via defining the following CNR:
\vspace{-0.5mm}
\begin{equation}\label{cons:LGPD1}
\begin{aligned}
   0 \le \overline{\varrho}_{u,u',r}(t_x)-\overline{\psi}_{u,u',r}(t_x)< 1,~x{\in}\mathcal{N}^{(k)}.
\end{aligned}
\end{equation}
Since DPU $u$ cannot disperse its GV to itself, we also impose
\begin{equation}\label{cons:LGPD2}
\begin{aligned}
   \overline{\psi}_{u,u,r}(t_x)= 0,~~x{\in}\mathcal{N}^{(k)},\forall r \in \overline{\mathcal{R}}_{b}.
\end{aligned}
\vspace{-.5mm}
\end{equation}

\textbf{CHUs local GV dispatching.} We use a decision variable $\psi_{u,r}\hspace{-0.3mm}(t_{x}){\in}[0,1]$ to denote the percentage of GV of CHU $u{\in} \mathcal{U}_{b}$ dispatched to O-RU $b$ over PRB $r {\in} \mathcal{R}_{b}$~at time~$t_{x}$. Further, with the same interpretations as~\eqref{cons:LGPD1}, we impose
\vspace{-.1mm}
\begin{equation}\label{cons:LMD1}
        0\le \varrho_{u,r}\hspace{-0.3mm}(t_{x})-\psi_{u,r}\hspace{-0.3mm}(t_{x})< 1,~x{\in}\mathcal{N}^{(k)}.
\end{equation}

\vspace{-1mm}
\section{Latency and Energy Consumption}\label{sec:computation_communication_latencies_energy}
\noindent In this section, we calculate (i) dynamic licensed/unlicensed data  rates (Sec.~\ref{sec:data_transmission_rate}), (ii) licensed/unlicensed uplink and downlink latencies (Sec. \ref{sec:computation_communication_latencies}), and (iii) consumed energy of FLUs and O-RUs (Sec. \ref{sec:energy_consumption}). 
In nutshell, these derivations are introduced to obtain closed-form expression of the latency and energy consumption of ML model training with respect to the dynamic wireless control decisions at each FGTI $t_x$, including dynamic MAC scheduler (e.g., scheduling decision $\beta^{\downarrow}_{b}\hspace{-0.3mm}(t_{x})$) and GM/GV  transmission decisions (e.g., GM broadcasting decision $\varphi_{b,r}\hspace{-0.3mm}(t_{x})$). These closed-form expressions will be later integrated into the ML convergence bounds (Sec.~\ref{sec:conv}) and enable a coherent formulation of DCLM optimization (Sec.~\ref{sec:PF}).


\vspace{-1.0mm}
\subsection{Dynamic Data Transmission Rates}\label{sec:data_transmission_rate}
\textbf{$\bm{\mathsf{R2F}}$ transmission rate.}
Let $\mathfrak{R}^{\downarrow}_{b,u,r}\hspace{-0.5mm}(t_{x})$ be the maximum downlink data rate of FLU $u{\in} \mathcal{U}_{b}$ over PRB $r$ through O-RU $b$ at FGTI $t_x$. The broadcasting data rate of O-RU $b$ over PRB $r$ at FGTI $t_x$, denoted as $ \mathfrak{R}^{\downarrow}_{b,r}\hspace{-0.5mm}(t_{x})$, is then given by
\vspace{-0.5mm}
\begin{equation}\label{eq:broadcast_data_rate}
\hspace{-5mm}
    \mathfrak{R}^{\downarrow}_{b,r}\hspace{-0.5mm}(t_{x}){=}\min_{u\in\mathcal{U}_b}\big\{\mathfrak{R}^{\downarrow}_{b,u,r}\hspace{-0.5mm}(t_{x})|\widehat{\lambda}^{(k)}_u{=} 1,x{\in}\mathcal{N}^{(k)},\forall u{\in} \mathcal{U}_{b}\big\},
    \vspace{-0.3mm}
\hspace{-5mm}
\end{equation}
where the $\min\{\cdot\}$ function is used to ensure that the FLUs do not receive data at a rate exceeding their channel capacity and
\vspace{-0.6mm}
\begin{equation} \label{cons:DR_R2F1}
\mathfrak{R}^{\downarrow}_{b,u,r}\hspace{-0.5mm}(t_{x})= B \log_2\big(1+\Gamma^{\downarrow}_{b,u,r}\hspace{-0.5mm}(t_{x})\big),
\end{equation}
with $B$ denoting the bandwidth of each licensed PRB {\small $r {\in} \mathcal{R}_{b}$}, and {\small$\Gamma^{\downarrow}_{b,u,r}\hspace{-0.5mm}(t_{x})$} representing the signal-to-interference-plus-noise-ratio (SINR) of FLU $u$ when downloading over $r$, given by
\vspace{-0.5mm}
\begin{equation}\label{eq:SINR_broadcast}
\hspace{-5mm}
    \Gamma^{\downarrow}_{b,u,r}\hspace{-0.5mm}(t_{x}){=}\frac{\vert\xi_{b,u}\hspace{-0.3mm}(t_{x})\vert^2 \rho^{\downarrow}_{b,r}\hspace{-0.5mm}(t_{x}) P^{\mathsf{max}}_{b}}{\displaystyle\sum_{b'{\in}\Omega\setminus \{b\}}\hspace{-2mm} |\xi_{b',u}\hspace{-0.3mm}(t_{x})|^2 \hspace{-0.5mm}\rho^{\downarrow}_{b',r}\hspace{-0.5mm}(t_{x}) P^{\mathsf{max}}_{b'}{+} BN_{0}}.
    \vspace{-0.5mm}
\hspace{-5mm}
\end{equation}
In \eqref{eq:SINR_broadcast}, the summation is interference caused by GM broadcasting of all O-RUs $b' {\in} \Omega{\setminus} \{b\}$ at FGTI $t_{x}$, and $N_{0}$ is the power spectral density of Gaussian noise. Additionally, $|\xi_{b,u}\hspace{-0.3mm}(t_{x})|^2$ is the channel gain between FLU $u{\in}\mathcal{U}_{b}$ and O-RU $b$ at time $t_x$.

\textbf{$\bm{\mathsf{C2R}}$ transmission rate.}
Dynamic uplink data rate of CHU $u{\in} \mathcal{U}_{b}$ over PRB $r {\in} \mathcal{R}_{b}$ at time $t_{x}$ is given by
 \vspace{-0.5mm}
\begin{equation}\label{cons:DR_F2R1} 
\mathfrak{R}^{\uparrow}_{u,r}\hspace{-0.3mm}(t_{x})= B \log_2\big(1+\Gamma^{\uparrow}_{u,r}\hspace{-0.3mm}(t_{x})\big),~x{\in}\mathcal{N}^{(k)},
\end{equation}
where $\Gamma^{\uparrow}_{u,r}\hspace{-0.3mm}(t_{x})$ is the uplink SINR of CHU $u{\in}\mathcal{U}_{b}$ to upload data to O-RU $b$ over PRB $r {\in} \mathcal{R}_{b}$ at time $t_{x}$, which is
\begin{equation}\label{eq:SINR_uplink}
\hspace{-5mm}
    \Gamma^{\uparrow}_{u,r}\hspace{-0.3mm}(t_{x}){=}\frac{\vert\xi_{b,u}\hspace{-0.3mm}(t_{x})\vert^2 \rho^{\uparrow}_{u,r}\hspace{-0.3mm}(t_{x})P^{\mathsf{max}}_{u}}{\hspace{-1mm}\displaystyle\sum_{b' \in \Omega} \sum_{u'\in \mathcal{U}_{b'}\setminus \{u\}}\hspace{-1mm}|\xi_{b,u'}\hspace{-0.3mm}(t_{x})|^2 \hspace{-0.8mm}\rho^{\uparrow}_{u',r}\hspace{-0.3mm}(t_{x})P^{\mathsf{max}}_{u'}{+} B N_{0}}\hspace{-0.5mm}.
    \vspace{-.1mm}
\hspace{-4mm}
\end{equation}
In \eqref{eq:SINR_uplink}, the summation in the denominator is the interference caused by upload operations of CHUs of all O-RUs $b' {\in} \Omega$. To prevent interference/overlap between R2F and C2R communications, we impose the following constraint:
 \vspace{-0.5mm}
\begin{equation}\label{cons:r2cc2r1}
T^{(k-1)}+\tau_{b}^{\downarrow,{(k)}}{\leq }~\min_{b\in\Omega}\min_{u\in \mathcal{U}_b}\{\Psi(\mathscr{D}_{u}^{\uparrow,(k)})\}.
\end{equation}

\textbf{$\bm{\mathsf{D2C}}$ transmission rate.} 
The  data rate between DPU $u{\in} \mathcal{U}_{b}$ and CHU $u'{\in} \mathcal{U}_{b}$ over PRB $r {\in} \overline{\mathcal{R}}_{b}$ at $t_x$ is 
 \vspace{-0.7mm}
\begin{equation} \label{cons:DR_Fh2F}
\overline{\mathfrak{R}}^{\uparrow}_{u,u',r}(\hspace{-0.5mm}t_x\hspace{-0.5mm}){=}\overline{B} \log_2\left(1+\overline{\Gamma}^{\uparrow}_{u,u',r}(t_x)\right),~x{\in}\mathcal{N}^{(k)},
\end{equation}
where $\overline{B}$ is the bandwidth of unlicensed PRB $r {\in} \overline{\mathcal{R}}_{b}$, and $\overline{\Gamma}^{\uparrow}_{u,u',r}(\hspace{-0.2mm}t_x\hspace{-0.5mm})$ is the uplink SINR of DPU $u{\in} \mathcal{U}_{b}$ to CHU $u'$ over unlicensed PRB $r {\in} \overline{\mathcal{R}}_{b}$ at time $t_x$, which is given by
\begin{equation}\label{eq:SINR_up_DPU}
 \hspace{-4.3mm} \resizebox{0.46\textwidth}{!}{  $
    \overline{\Gamma}^{\uparrow}_{u,u',r}(\hspace{-0.2mm}t_x\hspace{-0.5mm}){=}\frac{{\displaystyle\vert\xi_{u,u'}\hspace{-0.3mm}(t_{x})\vert^2 \overline{\rho}^{\uparrow}_{u,r}(\hspace{-0.2mm}t_x\hspace{-0.5mm})P^{\mathsf{max}}_{u}}}{\hspace{-0.2mm}\displaystyle\sum_{b' \in \Omega} \sum_{\widehat{u},\widehat{u}'\in \mathcal{U}_{b'}\setminus \{u,u'\}}\hspace{-6mm}|\xi_{\widehat{u},u'}\hspace{-0.3mm}(t_{x})|^2 \hspace{-0.7mm}\rho_{\widehat{u},\widehat{u}',r}(\hspace{-0.2mm}t_x\hspace{-0.5mm})P^{\mathsf{max}}_{\widehat{u}}{+}\overline{B}N_{0}}.$}
     \hspace{-3mm} 
    \vspace{-0.1mm}
\end{equation}
In \eqref{eq:SINR_up_DPU}, $|\xi_{u,u'}\hspace{-0.3mm}(t_{x})|^2$ is the channel gain between DPU $u$ and CHU $u'$ at FGTI $t_{x}$, and the denominator is the interference caused by DPUs of other O-RUs (i.e., $b' {\in} \Omega$) over PRB $r$. 

\vspace{-1mm}
\subsection{Computation and Communication Latencies}\label{sec:computation_communication_latencies}
We next calculate computation and communication latencies, utilized to calculate the occurrence time of $\mathscr{D}$-Events. 
\subsubsection{Local computation}\label{sec:computation_latency}
Let us define $a_{u}$ as the number of CPU cycles used to process one data sample at FLU~$u{\in}\mathcal{U}_{b}$. At global round $k$, the \textit{computation time} of  FLU~$u$ to execute $\ell^{(k)}_{u}$ local SGD iterations based on~\eqref{eq:WeightupdateStrat} is given by
\vspace{-.5mm}
\begin{equation}\label{cons:LT1}
    \begin{aligned}
     \tau_{u}^{\mathsf{LC},{(k)}}= \ell^{(k)}_{u} a_{u} {B}_{u}(\widetilde{\tau}_{b}^{\downarrow,{(k)}}) {/}f^{(k)}_{u},
    \end{aligned}
    \vspace{-.15mm}
\end{equation}
where ${B}_{u}(\widetilde{\tau}_{b}^{\downarrow,{(k)}})$ denotes the mini-batch size of FLU $u$ (see Sec.~\ref{sec:LM}), and $f^{(k)}_{u}$ is the CPU frequency of FLU $u$.

\subsubsection{Communication latencies}\label{sec:communication_latency}
In the following, we present four propositions to calculate $\tau_{b}^{\downarrow,{(k)}}$, $\overline{\tau}_{u}^{\uparrow,(k)}$, $\tau_{u}^{\uparrow,(k)}$, and $\tau_{u}^{\mathsf{W},{(k)}}$. 


\vspace{-0.5mm}
\begin{proposition}[GM Broadcasting Latency]\label{propo:broadcast}
Let $\alpha_{\bm{\omega}} M$ refer to the total size of the model that needs to be broadcast through O-RUs, where $\alpha_{\bm{\omega}} $ is the number of bits required to represent one element of the GM vector $\bm{\omega}^{(k)}$ with size $M$. Further, assume $\sum_{r \in \mathcal{R}_{b}} \varphi_{b,r}\hspace{-0.3mm}(t_{x})= 1$. The downlink latency of O-RU $b$ to broadcast the GM $\bm{\omega}^{(k)}$ to its recruited FLUs is given by
\begin{equation}\label{eq:TDownload}
\hspace{-6mm}\resizebox{0.46\textwidth}{!}{  $
   \tau_{b}^{\downarrow,{(k)}}=\hspace{-1mm}\sum_{x{=}N^{(k{-}1)}{+}1}^{N^{(k)}}\max_{r{\in}\mathcal{R}_{b}}\Big\{ \min\big\{\tau_{b,r}^{\downarrow}(t_{x}),\left(t_{x{+}1}-t_{x}\right)\big\}\Big\},$}
\hspace{-5mm}
\end{equation}    
where {\small$t_{N^{(K)}{+}1}{>}t_{N^{(K)}}$}, and {\small$\tau_{b,r}^{\downarrow}\hspace{-0.3mm}(t_{x})$} is calculated as follows:
\begin{equation}\label{eq:broadcast_communication_latency}
\hspace{-6mm}\resizebox{0.39\textwidth}{!}{  $
\hspace{-5mm}
\tau_{b,r}^{\downarrow}\hspace{-0.3mm}(t_{x}){=}\big(\hspace{-0.3mm}\alpha_{\bm{\omega}}  M{-}M^{\downarrow}_{b}(x)\hspace{-0.3mm}\big)\varphi_{b,r}\hspace{-0.3mm}(t_{x}){/}\big(\mathfrak{R}^{\downarrow}_{b,r}\hspace{-0.3mm}(t_{x}){+}1{-}\beta^{\downarrow}_{b}\hspace{-0.3mm}(t_{x})\hspace{-0.3mm}\big)\hspace{-0.3mm}.
\hspace{-5mm}$}
\vspace{-.5mm}
 \end{equation}
In \eqref{eq:broadcast_communication_latency}, $M^{\downarrow}_{b}(x)$ is the transmitted data until $t_{x}$ given by 

\vspace{-3.5mm}
\begin{equation}
\hspace{-5mm}
   M^{\downarrow}_{b}(x){=}\hspace{-5mm}\sum_{z{=}N^{(k{-}1)}{+}1}^{x}\sum_{r'
   {\in}\mathcal{R}_{b}}\hspace{-2mm}\min\hspace{-1mm}\big\{\tau_{b,r'}^{\downarrow}(t_{z})\hspace{-0.5mm},\hspace{-0.5mm}\left(t_{z{+}1}{-}t_{z}\right)\hspace{-0.5mm}\big\}\mathfrak{R}^{\downarrow}_{b,r'}(t_{z}).
\hspace{-5mm}
\end{equation} 
\end{proposition}
\vspace{-1.5mm}
\begin{proof}
   The proof is presented in Appendix~\ref{app:propo:broadcast}.
\end{proof}
\vspace{-2mm}
\vspace{-0.8mm}
\begin{proposition}[GV Dispersion Latency of DPUs]\label{propo:DPU_uplink}
Assume $\sum_{u'{\in} \mathcal{U}_{b}}\sum_{r{\in}\overline{\mathcal{R}}_{b}}\overline{\psi}_{u,u',r}(t_x){=}1,~x\in\mathcal{N}^{(k)}$. The dispersion latency of DPU $u{\in}\mathcal{U}_{b}$ is given by
\vspace{-.5mm}
\begin{equation}
\hspace{-5mm}
    \overline{\tau}_{u}^{\uparrow,(k)}{=}\hspace{-3mm}\sum_{x{=}N^{(k{-}1)}{+}1}^{N^{(k)}}\max_{\substack{u'{\in}\mathcal{U}_{b},\\r{\in}\overline{\mathcal{R}}_{b}}}\Big\{\min\big\{\overline{\tau}_{u,u',r}^{\uparrow}(t_{x}),\left(t_{x{+}1}{-}t_{x}\right)\big\}\hspace{-0.7mm}\Big\},
\hspace{-5mm}
\vspace{-.5mm}
\end{equation}   
where {\small$t_{N^{(K)}{+}1}{>}t_{N^{(K)}}$}, and {\small$\overline{\tau}_{u,u',r}^{\uparrow}(t_x)$} is calculated as follows:
\vspace{-1.5mm}
\begin{equation}\label{eq:DPU_communication_latency}
\overline{\tau}_{u,u',r}^{\uparrow}(t_x){=}\frac{(\alpha_{\bm{\omega}} M{-}\overline{M}^{\uparrow}_{u}(x))\overline{\psi}_{u,u',r}(t_x)}{\overline{\mathfrak{R}}^{\uparrow}_{u,u',r}\hspace{-0.5mm}(t_x\hspace{-0.5mm}){+}1{-}\overline{\beta}^{\uparrow}_u(t_x)}.
\vspace{-1.5mm}
\end{equation}
 In \eqref{eq:DPU_communication_latency}, $\overline{M}^{\uparrow}_{u}(x)$ is the transmitted data until time $t_{x}$ given by
 \vspace{-0.5mm}
\begin{align}
    \overline{M}^{\uparrow}_{u}(x)=\hspace{-3mm}\sum_{z=N^{(k{-}1)}{+}1}^{x}&\sum_{u'{\in}\mathcal{U}_{b}}\sum_{r'{\in}\overline{\mathcal{R}}_{b}}\min\Big\{\overline{\tau}_{u,u',r'}^{\uparrow}(t_{z})\nonumber\\
    &,\left(t_{z{+}1}{-}t_{z}\right)\hspace{-1mm}\Big\}\overline{\mathfrak{R}}^{\uparrow}_{u,u',r'}(t_{z}).
\end{align}
\end{proposition}
\vspace{-3.5mm}
\begin{proof}
    The proof is presented in Appendix~\ref{app:propo:DPU_uplink}.
\end{proof}
\vspace{-1mm}
\begin{proposition}[GV Dispatching Latency of CHUs]\label{propo:CHU_uplink}
Assume $\sum_{r {\in} \mathcal{R}_{b}} \psi_{u,r}\hspace{-0.3mm}(t_{x})=1,~x\in\mathcal{N}^{(k)}$. The dispatching latency of CHU $u\in\mathcal{U}_{b}$ to dispatch its GV to O-RU $b$ is given by
\begin{equation}
\hspace{-2.5mm}\resizebox{0.46\textwidth}{!}{  $
    \tau_{u}^{\uparrow,{(k)}}{=}\hspace{-1mm}\sum_{x{=}N^{(k{-}1)}{+}1}^{N^{(k)}}\max_{r{\in}\mathcal{R}_{b}}\Big\{\hspace{-1mm}\min\left\{\tau_{u,r}^{\uparrow}\hspace{-0.3mm}(t_{x}),\left(t_{x{+}1}-t_{x}\right)\right\}\hspace{-1mm}\Big\},$}
\hspace{-2mm}
\end{equation}
where {\small$t_{N^{(K)}{+}1}{>}t_{N^{(K)}}$}, and {\small$\tau_{u,r}^{\uparrow}\hspace{-0.3mm}(t_{x})$} is calculated as follows:
\begin{equation}\label{eq:CHU_communication_latency}
\hspace{-5mm}
\tau_{u,r}^{\uparrow}\hspace{-0.3mm}(t_{x}){=}\hspace{-0.5mm}\big(\hspace{-0.2mm}\alpha_{\bm{\omega}} M{-}M^{\uparrow}_{u}(x)\hspace{-0.2mm}\big)\psi_{u,r}\hspace{-0.3mm}(t_{x})\hspace{-0.5mm}{/}\hspace{-0.5mm}\big(\hspace{-0.5mm}\mathfrak{R}^{\uparrow}_{u,r}\hspace{-0.3mm}(t_{x}){+}1{-}\beta^{\uparrow}_u(t_x)\hspace{-0.5mm}\big)\hspace{-0.5mm}.
\hspace{-5mm}
\vspace{-0.15mm}
 \end{equation}
In \eqref{eq:CHU_communication_latency}, $M^{\uparrow}_{u}(x)$ is transmitted data until time $t_{x}$ given by
\vspace{-.0mm}

\hspace{-3.5mm}\begin{minipage}{1\linewidth}
\hspace{-7mm}
\begin{equation}
\hspace{-7mm}
    M^{\uparrow}_{u}(x){=}\hspace{-5mm}\sum_{z{=}N^{(k{-}1)}{+}1}^{x}\sum_{r'{\in}\mathcal{R}_{b}}\hspace{-2mm}\min\hspace{-1mm}\big\{\hspace{-0.5mm}\tau_{u,r'}^{\uparrow}\hspace{-0.5mm}(t_{z})\hspace{-0.5mm},\hspace{-0.5mm}\left(t_{z{+}1}{-}t_{z}\right)\hspace{-1mm}\big\}\mathfrak{R}^{\uparrow}_{u,r'}\hspace{-0.5mm}(t_{z}\hspace{-0.3mm}).
\hspace{-6mm}
\vspace{-1mm}
\end{equation}
\end{minipage}

\end{proposition}
\begin{proof}
    The proof is presented in Appendix~\ref{app:propo:CHU_uplink}.
\end{proof}
\vspace{-2mm}
\begin{proposition}[Waiting Time of CHUs]\label{propo:waiting_time}
    The waiting time of CHU $u{\in}\mathcal{U}_{b}$ to receive LMs from its associated DPUs is
    \vspace{-1mm}
\begin{equation}\label{waiting_time}
\hspace{-5mm}
    \tau_{u}^{\mathsf{W},{(k)}}=\max_{u'\in\mathcal{U}_{b}}\hspace{-0.5mm}\Big\{ \tau_{u'}^{\mathsf{LC},{(k)}}{+}\hspace{-2mm}\sum_{x{=}N^{(k{-}1)}{+}1}^{N^{(k)}}\max_{r{\in}\overline{\mathcal{R}}_{b}}\left\{L_{u,u',r}(t_x)\right\}\Big\},
\hspace{-5mm}
\end{equation}
where $L_{u,u',r}(t_x){=}\hspace{-0.5mm}\min\hspace{-1mm}\big\{\hspace{-0.5mm}\overline{\tau}_{u',u,r}^{\uparrow}(t_{x})\hspace{-0.5mm},\hspace{-0.5mm}\left(\hspace{-0.5mm}t_{x{+}1}{-}t_{x}\hspace{-0.5mm}\right)\hspace{-1mm}\big\}$, and $\overline{\tau}_{u',u,r}^{\uparrow}(\hspace{-0.5mm}t_{x}\hspace{-0.5mm})$ is given by \eqref{eq:DPU_communication_latency}.
\end{proposition}
\vspace{-2mm}
\begin{proof}
    The proof is presented in Appendix~\ref{app:propo:waiting_time}.
\end{proof}\vspace{-2mm}

\textbf{Completion time.} Considering the above propositions, the completion time of global round $k$ is given by
\begin{equation}\label{TK}
    T^{(k)}= T^{(k{-}1)}+\max_{b\in \Omega,u\in \mathcal{U}_{b}}\{\Psi(\mathscr{D}_{u}^{\nuparrow,(k)})\},
\end{equation}
where {\small$\Psi(\mathscr{D}_{u}^{\nuparrow,(k)}){=}\lambda_{u}^{(k)}\left(\tau_{b}^{\downarrow,{(k)}}{+}\max\{\tau_{u}^{\mathsf{LC},{(k)}}, \tau_{u}^{\mathsf{W},{(k)}}\}{+} \tau_{u}^{\uparrow,{(k)}}\right)$} is the occurrence time of event {\small$\mathscr{D}_{u}^{\nuparrow,(k)}$} introduced in Sec.~\ref{sec:system_dynamic}.{ We subsequently impose the following constraint:
\begin{equation}\label{const:latency_max}
    T^{(k)}\le T^{\mathsf{max}},
\end{equation}
where $\small T^{\mathsf{max}}$ is maximum allowable latency per global round.}

\vspace{-1mm}
\subsection{Energy Consumption}\label{sec:energy_consumption}
We next calculate energy consumption of {\tt DCLM}, encompassing (i) GM broadcasting of O-RUs (Sec.~\ref{sec:GM_broadcasting}), (ii) LM training of FLUs (Sec.~\ref{sec:En_LM_training}), (iii) GV dispersing of DPUs (Sec.~\ref{sec:En_DPU}), and (iv) GV dispatching of CHUs (Sec.~\ref{sec:En_CHU}). 
\subsubsection{GM broadcasting of O-RUs}\label{sec:GM_broadcasting}
Total energy consumed by O-RU $b$ at global round $k$ to broadcast the GM to FLUs can be calculated as the cumulative consumed energy for broadcasting during the FGTIs of global round $k$, which is given by
  \vspace{-.7mm}
\begin{equation}\label{eq:EN_broad}
\hspace{-5mm}
    E_{b}^{\downarrow,{(k)}}{=}\hspace{-6mm}\sum_{x{=}N^{(k{-}1)}{+}1}^{N^{(k)}}\sum_{r{\in} \mathcal{R}_{b}}\hspace{-1.4mm}{\min\hspace{-1mm}\big\{\hspace{-0.3mm}\tau_{b,r}^{\downarrow}(\hspace{-0.3mm}t_{x}\hspace{-0.3mm}),\hspace{-0.5mm}\left(\hspace{-0.3mm}t_{x{+}1}{-}t_{x}\hspace{-0.3mm}\right)
     \hspace{-0.7mm}\big\}\rho^{\downarrow}_{b,r}\hspace{-1mm}\left(\hspace{-0.3mm}t_{x}\hspace{-0.3mm}\right)\hspace{-1mm} P^{\mathsf{max}}_{b}},
\hspace{-5mm}
\vspace{-.17mm}
\end{equation}
where {\small$t_{N^{(K)}{+}1}{>}t_{N^{(K)}}$}, and the summand term is the energy usage of broadcasting a fraction of the GM to FLUs over PRB $r$ during $[t_{x},t_{x{+}1})$. Further, $\tau_{b,r}^{\downarrow}\hspace{-0.3mm}(t_{x})$ is given in \eqref{eq:broadcast_communication_latency}.
\subsubsection{LM training} \label{sec:En_LM_training}
Let $\frac{\alpha_{u}}{2}$ be the effective chipset capacitance of FLUs~\cite{8737464}. The total energy consumed by FLU $u$ to perform LM training at global round $k$ is given by 
\begin{equation}\label{eq:EN_LC}
    E_{u}^{\mathsf{LC},{(k)}}= (f^{(k)}_{u})^{3}\tau_{u}^{\mathsf{LC},{(k)}}{\alpha_{u}}/{2}, \forall u\in\mathcal{U}_{b}, \forall b\in\Omega,
\end{equation}
where $\tau_{u}^{\mathsf{LC},{(k)}}$ is the training latency of FLU $u$ given in~\eqref{cons:LT1}. 
\subsubsection{GV dispersion by DPUs}\label{sec:En_DPU}
Total energy consumed by DPU $u\in\mathcal{U}_{b}$ to disperse its GV among CHUs at global round $k$ is 
\vspace{-.7mm}
\begin{equation}\label{eq:EN_DPU_UP}
\hspace{-5mm}
    \overline{E}_{u}^{\uparrow,(k)}{=}\hspace{-4mm}\sum_{x{=}N^{(k{-}1)}{+}1}^{N^{(k)}}\sum_{u'{\in}\mathcal{U}_{b}}\sum_{r{\in} \overline{\mathcal{R}}_{b}}\hspace{-0.3mm}{H_{u,u',r}(t_x)\overline{\rho}^{\uparrow}_{u,r}\hspace{-0.5mm}(\hspace{-0.5mm}t_x\hspace{-0.5mm})P^{\mathsf{max}}_{u}},
\hspace{-5mm}
\vspace{-.0mm}
\end{equation}
where {\small$H_{u,u',r}(t_x){=}\hspace{-0.5mm}\min\hspace{-1mm}\big\{\hspace{-0.1mm}\overline{\tau}_{u,u',r}^{\uparrow}(\hspace{-0.1mm}t_x\hspace{-0.1mm}),\hspace{-0.1mm}\left(\hspace{-0.1mm}t_{x{+}1}{-}t_{x}\hspace{-0.1mm}\right)\hspace{-1mm}\big\}$}, {\small$t_{N^{(K)}{+}1}{>}t_{N^{(K)}}$} ({\small$\overline{\tau}_{u,u',r}^{\uparrow}(t_x)$} is given in \eqref{eq:DPU_communication_latency}), and the summand captures the consumed energy of DPU $u$ to upload a fraction of its GV to CHU $u'$ over PRB $r$ during time window {\small$[t_{x},t_{x{+}1})$}. We impose the following constraint on the energy consumption of DPUs:
\begin{equation}\label{const:energy:max_DPU}
    \overline{E}_{u}^{\uparrow,(k)}{+}E_{u}^{\mathsf{LC},{(k)}}\le E^{\mathsf{max}}_{u},
\end{equation}
where $E^{\mathsf{max}}_{u}$ is the maximum battery capacity of FLU $u$.

\subsubsection{GV dispatchment by CHUs}\label{sec:En_CHU}
Following the same analogy as DPU LM uploading, the consumed energy of CHU $u$ to dispatch its GV to O-RU $b$ at global round $k$ is 
\vspace{-0.75mm}
\begin{equation}\label{eq:EN_CHU_UP}
\hspace{-5mm}
    E_{u}^{\uparrow,(k)}{=}\hspace{-5.5mm}\sum_{x{=}N^{(k{-}1)}{+}1}^{N^{(k)}}\sum_{r{\in}\mathcal{R}_{b}}\hspace{-1.2mm}\min\hspace{-.9mm}\big\{\hspace{-0.7mm}\tau_{u,r}^{\uparrow}\hspace{-0.4mm}(\hspace{-0.5mm}t_{x}\hspace{-0.5mm}),\hspace{-0.7mm}\left(t_{x{+}1}{-}t_{x}\right)\hspace{-1.4mm}\big\}\hspace{-0.5mm}\rho^{\uparrow}_{u,r}\hspace{-0.4mm}(\hspace{-0.5mm}t_{x}\hspace{-0.5mm})\hspace{-0.4mm}P^{\mathsf{max}}_{u}\hspace{-0.5mm},
\hspace{-5mm}
\end{equation}
where {\small$t_{N^{(K)}{+}1}{>}t_{N^{(K)}}$}, and {\small$\tau_{u,r}^{\uparrow}\hspace{-0.3mm}(t_{x})$} is given in \eqref{eq:CHU_communication_latency}. We impose the constraint below on energy consumption of CHUs:
\begin{equation}\label{const:energy:max_CHU}
    E_{u}^{\uparrow,(k)}{+}E_{u}^{\mathsf{LC},{(k)}}\le E^{\mathsf{max}}_{u}.
\end{equation}

\begin{remark}
The expressions above and the four propositions in Sec.~\ref{sec:communication_latency}, along with the derivations in Sec.~\ref{sec:dynamic_control_decisions}, provide answers to parts (ii) and (iii) of question \ref{Q3} in Sec. \ref{sec:summary_of_contribution}, pertaining to the impact of dynamic wireless control decisions on ML training latency and energy consumption. In particular, the above-derived expressions show a nested connection between data rates, transmit powers, PRB allocations, scheduling decisions, and the occurrence time of $\mathscr{D}$-Events. These expressions are functions of the downlink and uplink data rates of FLUs and O-RUs (Sec.~\ref{sec:data_transmission_rate}), which are in turn determined by the transmit power decision variables  (Sec.~\ref{sec:data_transmission_rate}) --- themselves functions of the PRB allocation decisions made by the MAC scheduler (Sec.~\ref{sec:power_allocation}). Further, the MAC scheduler conducts PRB allocations according to the scheduling decision variables (Sec.~\ref{sec:resource_allocation_mac}), which are in turn dependent on the occurrence time of $\mathscr{D}$-Events and recruitment indicators of FLUs (Sec.~\ref{sec:scheduling_decisions}). These results are all later used in Sec.~\ref{sec:PF} to obtain the optimal orchestration of network elements.
\end{remark}
Below, we focus on assessing the convergence behavior of the global ML model as a means of gauging its performance.

\vspace{-1mm}
\section{Convergence Analysis of {\tt DCLM}}\label{sec:conv}
\noindent Before conducting convergence analysis of {\tt DCLM}, it is imperative to make two standard assumptions~\cite{8737464,8664630}.
\vspace{-1mm}
\begin{assumption}[$\beta$-smoothness of Global/Local Loss Functions]\label{Assup:lossFun}
    Local loss function $\mathfrak{L}^{({k})}_{u}$ is $\beta$-smooth,~$\forall u,k$:
    \footnote{\vspace{-1mm}For brevity in notations we utilize $\nabla$ to denote $\nabla_{\bm{\omega}}$}
    \begin{equation}
       \hspace{-1.5mm} \Vert \nabla \mathfrak{L}^{({k})}_{u}(\bm{\omega})- \nabla \mathfrak{L}^{({k})}_{u}(\bm{\omega}') \Vert \hspace{-.4mm}\leq \hspace{-.4mm} \beta \Vert \bm{\omega}-\bm{\omega}' \Vert,~ \hspace{-.8mm}\forall \bm{\omega},\bm{\omega}'\hspace{-.4mm}\in\hspace{-.4mm}\mathbb{R}^M\hspace{-.7mm}, \hspace{-2.3mm} 
    \end{equation}
which results in $\beta$-smoothness of global loss function $\mathfrak{L}^{({k})}$.  
\end{assumption}
 \vspace{-3mm} 
\begin{assumption}[Heterogeneity of FLUs' Datasets]\label{Assup:Dissimilarity} 
For any set of coefficients {\small$\{s_{u}{\geq} 0\}_{u{\in} \mathcal{U}_{b},b\in\Omega,}$}, where {\small$\sum_{b\in\Omega}\sum_{u{\in} \mathcal{U}_{b}} s_{u}{=}1$}, there exist constants {\small$\mathfrak{X}_1{\geq} 1,\mathfrak{X}_2 {\geq} 0$}, such that, {\small$\forall k,\bm{\omega}$}, we have
   \begin{equation}
   \hspace{-3mm}
        \sum_{b\in\Omega}\hspace{-0.5mm}\sum_{u\in \mathcal{U}_{b}}\hspace{-2mm} s_{u} \Vert \nabla \mathfrak{L}^{({k})}_{u}(\bm{\omega}) \Vert^2 \hspace{-1.2mm} \leq\hspace{-0.4mm} \mathfrak{X}_1 \hspace{-0.1mm} \Big\Vert  \sum_{b\in\Omega}\hspace{-0.5mm}\sum_{u\in \mathcal{U}_{b}}\hspace{-2mm}s_{u} \nabla \mathfrak{L}^{({k})}_{u}(\bm{\omega}) \Big\Vert^2 \hspace{-1.2mm} {+}\hspace{-0.2mm}\mathfrak{X}_2.
    \hspace{-4mm}
   \end{equation}
\end{assumption}

Assumption~\ref{Assup:Dissimilarity} uses $\mathfrak{X}_1$ and $\mathfrak{X}_2$ to determine the degree of dissimilarity (i.e., non-i.i.d.-ness) among the datasets of the FLUs. Larger values of these parameters indicate greater data heterogeneity, while, in the case of homogeneous/i.i.d datasets, the minimum values of $\mathfrak{X}_1{=}1$ and $\mathfrak{X}_2{=}0$ are achieved.




\begin{table*}[t!]
\vspace{-1mm}
\begin{minipage}{0.99\textwidth}
{\scriptsize
\begin{align}\label{eq:gen_conv}
    &\frac{1}{K} \sum_{k=0}^{K-1}\mathbb{E}\left[\left\Vert\nabla{\mathfrak{L}^{({k})}(\bm{\omega}^{(k)})}\right\Vert^2\right] {\leq} \frac{4}{K} \sum_{k=0}^{K-1}\Bigg(\underbrace{\frac{\mathbb{E}_k\left[\mathfrak{L}^{(k-1)}(\bm{\omega}^{(k)})\right]-\mathbb{E}_k\left[\mathfrak{L}^{({k})}(\bm{\omega}^{(k+1)})\right]}{\eta_{_k}\mathfrak{B}_k\left(1-\zeta^{(k)}\right)}}_{(a)}+\underbrace{\frac{\sum_{b \in \Omega}\sum_{u\in \mathcal{U}_{b}} \mathfrak{D}^{(k)}_u\left(\left(\Delta T^{(k)}_u-\widehat{\lambda}_{u}^{(k)}\tau_{u}^{\mathsf{LC},{(k)}}\right)\right) }{\eta_{_k}\mathfrak{B}_k\left(1-\zeta^{(k)}\right)}}_{(b)}\Bigg)\nonumber\\[-.25em]
    &+\frac{8}{K} \sum_{k=0}^{K-1}\Bigg[\frac{1}{{\left(1-\zeta^{(k)}\right)}}\underbrace{\Bigg({\beta^2\Theta^2 \eta_k^2}\sum_{b \in \Omega} \sum_{u\in \mathcal{U}_{b}}\frac{|\Upsilon_{u}(\widetilde{\tau}_{b}^{\downarrow,{(k)}})|}{|\Upsilon(\widetilde{\bm{\tau}}^{\downarrow,{(k)}})|}\frac{\left(\ell_u^{(k)}-1\right)}{1- 4\eta_k^2\beta^2 \ell_u^{(k)}\left(\ell_u^{(k)}-1\right)}\left(1-\frac{{B}_{u}(\widetilde{\tau}_{b}^{\downarrow,{(k)}})}{|\Upsilon_{u}(\widetilde{\tau}_{b}^{\downarrow,{(k)}})|} \right)  \frac{{(|\Upsilon_{u}(\widetilde{\tau}_{b}^{\downarrow,{(k)}})|-1)}\left(\sigma_{u}(\widetilde{\tau}_{b}^{\downarrow,{(k)}})\right)^2}{|\Upsilon_{u}(\widetilde{\tau}_{b}^{\downarrow,{(k)}})|{B}_{u}(\widetilde{\tau}_{b}^{\downarrow,{(k)}})}\Bigg)}_{(c)}\nonumber\\[-.25em]
    & + \frac{1}{\left(1-\zeta^{(k)}\right)}\Bigg(\underbrace{\frac{\mathfrak{X}_2 \eta_k^2\beta^2 \left(\ell_{\mathsf{max}}^{(k)}\right)\left(\ell_{\mathsf{max}}^{(k)}-1\right)}{1- 4\eta_k^2\beta^2\ell_{\mathsf{max}}^{(k)}\left(\ell_{\mathsf{max}}^{(k)}-1\right)}}_{(d)}+ \underbrace{\frac{\Theta^2\beta\eta_{_k}\mathfrak{B}_k}{2} \sum_{b \in \Omega} \sum_{u\in \mathcal{U}_{b}}\frac{\left(\widehat{\lambda}_{u}^{(k)}|\Upsilon_{u}(\widetilde{\tau}_{b}^{\downarrow,{(k)}})|\right)^2}{\left(|\Upsilon^{\mathsf{s}}(\widetilde{\bm{\tau}}^{\downarrow,{(k)}})|\right)^2 \ell^{(k)}_{u}}\left(1-\frac{{B}_{u}(\widetilde{\tau}_{b}^{\downarrow,{(k)}})}{|\Upsilon_{u}(\widetilde{\tau}_{b}^{\downarrow,{(k)}})|} \right) \frac{(|\Upsilon_{u}(\widetilde{\tau}_{b}^{\downarrow,{(k)}})|-1)\left(\sigma_{u}(\widetilde{\tau}_{b}^{\downarrow,{(k)}})\right)^2}{|\Upsilon_{u}(\widetilde{\tau}_{b}^{\downarrow,{(k)}})|{B}_{u}(\widetilde{\tau}_{b}^{\downarrow,{(k)}})}}_{(e)}\Bigg)\nonumber\\[-.25em]
    &+\frac{24}{\left(1-\zeta^{(k)}\right)}\underbrace{\left (\frac{|\Upsilon(\widetilde{\bm{\tau}}^{\downarrow,{(k)}})|-|\Upsilon^{\mathsf{s}}(\widetilde{\bm{\tau}}^{\downarrow,{(k)}})|}{|\Upsilon(\widetilde{\bm{\tau}}^{\downarrow,{(k)}})|}\right)^2\sum_{b \in \Omega} \sum_{u\in \mathcal{U}_{b}}\Bigg(\frac{\Theta^2 \beta^2 \eta_k^2 \left(\ell_u^{(k)}-1\right)}{1- 4\eta_k^2\beta^2 \ell_u^{(k)}\left(\ell_u^{(k)}-1\right)} \left(1-\frac{{B}_{u}(\widetilde{\tau}_{b}^{\downarrow,{(k)}})}{|\Upsilon_{u}(\widetilde{\tau}_{b}^{\downarrow,{(k)}})|} \right)  \frac{{(|\Upsilon_{u}(\widetilde{\tau}_{b}^{\downarrow,{(k)}})|-1)}\left(\sigma_{u}(\widetilde{\tau}_{b}^{\downarrow,{(k)}})\right)^2}{|\Upsilon_{u}(\widetilde{\tau}_{b}^{\downarrow,{(k)}})|{B}_{u}(\widetilde{\tau}_{b}^{\downarrow,{(k)}})}\Bigg)}_{(f)}\nonumber\\[-.25em]
    &+\frac{6}{\left(1-\zeta^{(k)}\right)}\underbrace{\Bigg (\frac{|\Upsilon(\widetilde{\bm{\tau}}^{\downarrow,{(k)}})|-|\Upsilon^{\mathsf{s}}(\widetilde{\bm{\tau}}^{\downarrow,{(k)}})|}{|\Upsilon(\widetilde{\bm{\tau}}^{\downarrow,{(k)}})|}\Bigg)^2\frac{|\Upsilon(\widetilde{\bm{\tau}}^{\downarrow,{(k)}})|}{|\Upsilon_{\mathsf{min}}(\widetilde{\bm{\tau}}^{\downarrow,{(k)}})|}\Bigg(\frac{\mathfrak{X}_2}{1- 4\eta_k^2\beta^2 \ell_{\mathsf{max}}^{(k)}\left(\ell_{\mathsf{max}}^{(k)}-1\right)}\Bigg)}_{(g)}\Bigg]
 \end{align}
  \vspace{-.1mm}
 }
 \hrule
 \end{minipage}
 \vspace{-6.5mm}
\end{table*}
We next introduce \textit{local data dissimilarity} to evaluate the level of dissimilarity of data points in the dataset of each FLU:
\vspace{-2mm}
\begin{definition}[Local Data Dissimilarity]\label{Assump:DataVariabilit}
    Local data dissimilarity at FLU $u{\in}\mathcal{U}_{b}$ is measured via $\Theta_{u}{\geq} 0$, which, $\forall \bm{\omega}\in\mathbb{R}^M$ and $\forall \xi,\xi'\in\Upsilon_{u}(\widetilde{\tau}_{b}^{\downarrow,{(k)}})$, satisfies
    \begin{equation}
    \hspace{-5mm}
   \Vert \nabla f_{u}\hspace{-.4mm}(\bm{\omega},\xi) \hspace{-.4mm}-\hspace{-.4mm} \nabla f_{u}\hspace{-.4mm}(\bm{\omega},\xi')\Vert  \hspace{-.4mm} \leq \hspace{-.4mm} \Theta_{u} \Vert \xi\hspace{-.4mm}-\hspace{-.4mm}\xi' \Vert.
   \hspace{-5mm}
    \end{equation}
\end{definition}

We next obtain the general convergence behavior of {\tt DCLM}.
\vspace{-1mm}
\begin{theorem}[General Convergence Behavior of {\tt DCLM}]\label{th:main}
Let us define {\small$\widetilde{\tau}_{b}^{\downarrow,{(k)}}{=}T^{(k{-}1)}{+}\tau_{b}^{\downarrow,{(k)}}$},~{\small$\widetilde{\bm{\tau}}^{\downarrow,{(k)}}{=}[\widetilde{\tau}_{b}^{\downarrow,{(k)}}]_{b\in \Omega}$}, {\small$|\Upsilon_{\mathsf{min}}(\widetilde{\bm{\tau}}^{\downarrow,{(k)}})|=\min_{b{\in}\Omega,u{\in} \mathcal{U}_{b}}\{|\Upsilon_{u}(\widetilde{\tau}_{b}^{\downarrow,{(k)}})|\}$}, {\small$\Theta = \max_{b\in\Omega,u\in \mathcal{U}_{b}}\{\Theta_{u} \}$}, {\small$\ell^{(k)}_{\mathsf{max}}=\max_{b\in\Omega,u{\in} \mathcal{U}_{b}}\{\ell^{(k)}_u\}$}, {\small$\mathfrak{D}_u^{(k)}=\max_{t{\in}T_u^{\mathsf{Idle},(k)}} \mathfrak{D}_u(t)$}. Let {\small $\Delta T_u^{(k)} \triangleq \big(T^{(k)} - T^{(k-1)}\big) 
+ \Psi\!\big(\mathscr{D}_u^{\ndownarrow,(k+1)}\big) - \Psi\!\big(\mathscr{D}_u^{\ndownarrow,(k)}\big)$}. Assume that {\small$\eta_k$} satisfies {\small$\eta_k \leq \min \big\{\frac{1}{2\beta} \sqrt{(\zeta^{(k)}{-} 48\mathfrak{X}_1\mathscr{R}^{(k)}){/}\Delta}, \frac{1}{2\beta}\big\}$}, where {\small$\zeta^{(k)}<1$}, {\small$\Delta=\ell_{\mathsf{max}}^{(k)}\big(\ell_{\mathsf{max}}^{(k)}{-}1\big)\big(2\mathfrak{X}_1{+}\zeta^{(k)}\big)$}, and 
{
\begin{equation}
\resizebox{0.49\textwidth}{!}{  $
    \mathscr{R}^{(k)}={\left(|\Upsilon(\widetilde{\bm{\tau}}^{\downarrow,{(k)}})|{-}|\Upsilon^{\mathsf{s}}(\widetilde{\bm{\tau}}^{\downarrow,{(k)}})|\right)^2}\Big/\left({|\Upsilon(\widetilde{\bm{\tau}}^{\downarrow,{(k)}})||\Upsilon_{\mathsf{min}}(\widetilde{\bm{\tau}}^{\downarrow,{(k)}})|}\right)\nonumber.$}
\end{equation}
}
The upper bound in~\eqref{eq:gen_conv} is satisfied by the cumulative average of the gradient of the global loss function in {\tt DCLM}.
\end{theorem}
\vspace{-2mm}
\begin{proof}
The proof is carried out in Appendix~\ref{app:th:main}.
\end{proof}
\vspace{-2mm}
Note that {\small$\Upsilon(\hspace{-0.2mm}T^{(k)}\hspace{-0.3mm}){=}\hspace{-0.3mm}\cup_{b{\in} \Omega,u{\in}\mathcal{U}_{b}}\hspace{-1.3mm}\Upsilon_u(\hspace{-0.3mm}T^{(k)}\hspace{-0.3mm})$} refers to the cumulative datasets of all FLUs after performing $k$-th global aggregation (i.e., at time $T^{(k)}$), which is different from the cumulative dataset of recruited FLUs utilized for obtaining global model $\bm{\omega}^{(k)}$, referred to as {\small$\Upsilon^{\mathsf{s}}(\widetilde{\bm{\tau}}^{\downarrow,{(k)}}){=}\hspace{-0.5mm}\cup_{b{\in}\Omega,u{\in}\mathcal{U}_{b}}\hspace{-0.5mm}\lambda^{(k)}_u \Upsilon_u(\widetilde{\tau}_{b}^{\downarrow,{(k)}})$}. Accordingly, in term $(a)$ of ~\eqref{eq:gen_conv}, {\small${\mathfrak{L}^{({k-1})}(\bm{\omega}^{(k)})}{=}{\mathfrak{L}(\bm{\omega}^{(k)}|\Upsilon(T^{(k{-}1)})})$} denotes the global loss under GM $\bm{\omega}^{(k)}$ and dataset $\Upsilon(T^{(k{-}1)})$, and {\small$\mathfrak{L}^{({k})}(\bm{\omega}^{(k{+}1)}){=}{\mathfrak{L}(\bm{\omega}^{(k{+}1)}|\Upsilon(T^{(k)})})$} denotes the global loss under GM $\bm{\omega}^{(k{+}1)}$ and dataset $\Upsilon(T^{(k)})$. Also, $T^{(-1)}{=}0$, and {\small$\mathfrak{L}^{(-1)}(\bm{\omega}^{(0)}){=}\mathfrak{L}(\bm{\omega}^{(0)}|\Upsilon(T^{(-1)}))$} denotes the \textit{initial} loss of the algorithm before the start of any model training.

Below, we explain how the ML-related control decisions (Sec.~\ref{sec:FL_model}) and dynamic wireless control decisions (Sec.~\ref{sec:dynamic_control_decisions}) utilized in {\tt DCLM} affect the convergence bound in~\eqref{eq:gen_conv}.

\textbf{Interpretation 1: ML-related control decisions.}
Referring to \eqref{eq:gen_conv},  term $(a)$ captures the effect of consecutive loss function gains during ML training. This term is inversely proportional to $\eta_k$, while $(c)$, $(d)$, and $(f)$ are quadratically proportional to $\eta_k$, and $(e)$ is linearly proportional to $\eta_k$. Terms $(c)$, $(e)$, $(f)$ capture the impact of mini-batch sizes {\small${B}_{u}(\widetilde{\tau}_{b}^{\downarrow,{(k)}})$} used at the FLUs: these terms become zero when full batch size is employed (i.e., {\small ${B}_{u}(\widetilde{\tau}_{b}^{\downarrow,{(k)}}){=}|\Upsilon_{u}(\widetilde{\tau}_{b}^{\downarrow,{(k)}})|$}). The heterogeneity of the dataset is captured by $\mathfrak{X}_2$ in $(d)$ and $(g)$, and by $\mathfrak{X}_1$ in the condition on $\eta_{_k}$. Larger Dataset heterogeneity leads to a larger upper bound due to $(d)$ and $(g)$. Also, larger local data dissimilarity implies a larger bound ($\Theta$ in $(c)$, $(e)$, and $(f)$). Terms $(c)$, $(d)$, $(e)$, $(f)$, and $(g)$ capture the impact of the number of SGD iterations among FLUs (via $\ell^{(k)}_{\mathsf{max}}$ and $\ell^{(k)}_n$). If  $\ell^{(k)}_n{=}1$, $\forall n,k$, $(c)$, $(d)$, and $(f)$ become zero, resulting in a bound that show 1-epoch distributed ML with non-uniform SGD sampling and asymmetric FLU recruitment (term $(g)$). 

\textbf{Interpretation 2: Asymmetric FLU recruitment.} The bound reveals the relationship between FLU recruitment (via {\small$|\Upsilon(\widetilde{\bm{\tau}}^{\downarrow,{(k)}})|{-}|\Upsilon^{\mathsf{s}}(\widetilde{\bm{\tau}}^{\downarrow,{(k)}})|$}) and both the dataset heterogeneity of FLUs  (captured via $\mathfrak{X}_2$ in $(g)$) and their local data dissimilarity (captured via $\Theta$ in $(f)$): terms $(g)$ and $(f)$ imply that if $\mathfrak{X}_2$  and $\Theta$ are large, more FLUs with larger dataset must be recruited. Moreover, if {\small$|\Upsilon^{\mathsf{s}}(\widetilde{\bm{\tau}}^{\downarrow,{(k)}})|{=}|\Upsilon(\widetilde{\bm{\tau}}^{\downarrow,{(k)}})|$}, terms $(f)$ and $(g)$ become zero, obtaining the bound for full participant of FLUs. Further, the relationship between FLU recruitment and model drift is captured via $\widetilde{\lambda}^{(k)}_u$ and $\mathfrak{D}_{u}^{(k)}$ in  $(b)$, which implies that recruiting more FLUs reduces the error caused by model drift. Also, $(b)$ reveals that recruiting FLUs with larger model drift leads to a smaller bound (i.e., faster convergence). 
\vspace{-.25mm}

\textbf{Interpretation 3: Dynamic dataset and model drift.}  Terms $(f)$ and $(g)$ demonstrate that to reduce the error caused by FLU recruitment, we need to recruit FLUs whose datasets' size are larger at time $T^{(k{-}1)}$ and have a higher growth rate during GM broadcasting (i.e., for {\small$t{\in}[T^{(k{-}1)}, \widetilde{\tau}_{b}^{\downarrow,{(k)}}]$}); otherwise {\small$|\Upsilon(t)|$} grows faster than {\small$|\Upsilon^{\mathsf{s}}(t)|$}, leading to a larger value of {\small$\frac{|\Upsilon(t)|-|\Upsilon^{\mathsf{s}}(t)|}{|\Upsilon(t)|}$} at time {\small$t=\widetilde{\bm{\tau}}^{\downarrow,{(k)}}$} (i.e., the start time of LM training at global round $k$) and  increased errors caused by terms $(f)$ and $(g)$. Bound~\eqref{eq:gen_conv}, via {\small$T^{(k)}{-}T^{(k{-}1)}$} in term $(b)$, also implies that to mitigate the error caused by model drift (i.e., {\small$\mathfrak{D}^{(k)}_u, \forall u\in\mathcal{U}_{b}$}), the GM broadcasting and FLUs' GV uploading must be performed faster as model drift {\small$\mathfrak{D}^{(k)}_u$} increases. 

Interpretations 2 and 3 answer question \eqref{Q1} in Sec. \ref{sec:summary_of_contribution}. 
\vspace{-.25mm}

\textbf{Interpretation 4: Dynamic wireless control decisions.} 
As evident from \eqref{eq:gen_conv}, terms $(c)$, $(d)$, $(e)$, $(f)$, and $(g)$ are functions of the completion time of GM broadcasting, i.e., $\widetilde{\tau}_{b}^{\downarrow,{(k)}}{=}T^{(k{-}1)}{+}\tau_{b}^{\downarrow,{(k)}}$. Also, term $(b)$ is a function of $T^{(k)}$ and $T^{(k{-}1)}$ calculated in \eqref{TK}. $T^{(k)}$ and $T^{(k{-}1)}$ incorporate the occurrence times of $\mathscr{D}$-events, calculated through propositions in Sec.~\ref{sec:communication_latency}, into bound~\eqref{eq:gen_conv}. These propositions are directly dependent on the dynamic wireless control decisions conducted at FGTIs (see Sec. \ref{sec:dynamic_control_decisions}). Subsequently, bound~\eqref{eq:gen_conv} shows that DCC communication mode, scheduling decisions and resource allocation of MAC scheduler, transmit power allocations, and interference on PRBs experienced from nearby O-RUs directly affect the convergence of {\tt DCLM}. This interpretation answers part (i) of question \eqref{Q3} in Sec. \ref{sec:summary_of_contribution} regarding the impact of dynamic wireless control decisions on ML model accuracy.

\textbf{Conclusion of Theorem~\ref{th:main}.} Considering the above interpretations, Theorem~\ref{th:main} provides a deep understanding of the impact of wireless network conditions on FL.
While existing studies~\cite{ganguly2023multi,9261995, 8737464} have investigated how wireless control decisions affect model training latency and energy consumption, our findings are the first to reveal that such decisions also significantly influence the convergence behavior of FL. This is mainly because prior works have focused on static snapshots of the network, without considering how the network and its users evolve over time.


In the following theorem, we provide the sufficient conditions under which convergence of {\tt DCLM} is guaranteed.

\vspace{-2mm}
\begin{theorem}[Convergence under Sufficient Conditions]\label{th:sufficient_conditions}
In addition to the conditions in Theorem~\ref{th:main}, further assume that (i) {\small$\widehat{\ell}_{\mathsf{min}} \leq \ell^{(k)}_{\mathsf{sum}}\leq  \widehat{\ell}_{\mathsf{max}}$} for two finite positive constants {\small$\widehat{\ell}_{\mathsf{min}}$} and {\small$\widehat{\ell}_{\mathsf{max}}$}, {\small$\forall k$}, (ii) {\small$\max_{k{\in}\mathcal{K}}\{ \ell^{(k)}_{\mathsf{max}}\}\leq \ell_{\mathsf{max}}$} and {\small$\max_{k{\in}\mathcal{K}}\{ \ell^{(k)}_{\mathsf{min}}\}\leq \ell_{\mathsf{min}}$}, (iii)  {\small$\max_{k{\in}\mathcal{K}} \big\{\zeta^{(k)}\big\} \leq \zeta_{\mathsf{max}}<1$}, (iv) {\small$N=\sum_{b\in\Omega}U_{b}$}, (v) {\small$\max_{k{\in}\mathcal{K}} \big\{\zeta^{(k)}\big\} \leq \zeta_{\mathsf{max}}<1$}, (vi) {\small$\widehat{\zeta}^{(k)}\le\zeta^{(k)}$}, and (vii) {\small$\max_{k{\in}\mathcal{K}}\{(|\Upsilon_{\mathsf{min}}(\widetilde{\bm{\tau}}^{\downarrow,{(k)}})|)/(|\Upsilon(\widetilde{\bm{\tau}}^{\downarrow,{(k)}})|)\}{\le} |\Upsilon_{\mathsf{max}}|$} for a constant {\small$|\Upsilon_{\mathsf{max}}|$}. Also, consider two constants {\small$|\varpi|{<}1$} and {\small$|\vartheta|{<}1$}. {\tt DCLM} satisfies {\small$\frac{1}{K} \sum_{k=0}^{K-1}\mathbb{E}\big[\big\Vert\nabla{\mathfrak{L}^{({k})}(\bm{\omega}^{(k)})}\big\Vert^2\big]=\mathcal{O}\big(1/\sqrt{K}\big)$}, 
implying  {\small$\lim_{K\rightarrow \infty}\frac{1}{K} \sum_{k=0}^{K-1}\mathbb{E}\big[\big\Vert\nabla{\mathfrak{L}^{({k})}(\bm{\omega}^{(k)})}\big\Vert^2\big]{=}0$}, 
if the following \textit{sufficient conditions} are met:
\begin{enumerate}[leftmargin=4.5mm]
    \item $\bm{\mathfrak{C}^{(L)}}\rightarrow$ \textbf{Sufficient Condition on Training Latency:} 
        \begin{align}\label{suf:main:gamma_upper}
            \hspace{-5mm}&T^{(k)}-T^{(k{-}1)} \le \frac{\eta_{_k}\mathfrak{B}_k\left(1-\zeta^{(k)}\right)\vartheta^k}{4\sum_{b \in \Omega}\sum_{u\in \mathcal{U}_{b}} \mathfrak{D}^{(k)}_u}+\nonumber\\
            \hspace{-5mm}&\frac{\sum_{b \in \Omega}\sum_{u\in \mathcal{U}_{b}} \mathfrak{D}^{(k)}_u(\widehat{\lambda}_{u}^{(k)}\tau_{u}^{\mathsf{LC},{(k)}}{-}\Psi\!\big(\mathscr{D}_u^{\ndownarrow,(k+1)}\big) {+} \Psi\!\big(\mathscr{D}_u^{\ndownarrow,(k)}\big))}{\sum_{b \in \Omega}\sum_{u\in \mathcal{U}_{b}} \mathfrak{D}^{(k)}_u}.
        \end{align}
    \item $\bm{\mathfrak{C}^{(\Upsilon)}}\rightarrow$ \textbf{Sufficient Condition on $|\Upsilon^{\mathsf{s}}(\widetilde{\bm{\tau}}^{\downarrow,{(k)}})|$:}
        \begin{equation}\label{suf:main:recruitment222}
        \hspace{-8mm}
        \resizebox{0.43\textwidth}{!}{$
        \begin{aligned}
            &\left(\hspace{-0.5mm}|\Upsilon(\widetilde{\bm{\tau}}^{\downarrow,{(k)}})|{-}|\Upsilon^{\mathsf{s}}(\widetilde{\bm{\tau}}^{\downarrow,{(k)}})|\hspace{-0.5mm}\right)^2\hspace{-2mm}\\
            &{<}\hspace{-0.3mm}|\Upsilon(\widetilde{\bm{\tau}}^{\downarrow,{(k)}})|\min\hspace{-0.5mm}\left\{\hspace{-0.5mm}\widehat{\zeta}^{(k)}|\Upsilon_{\mathsf{min}}(\widetilde{\bm{\tau}}^{\downarrow,{(k)}})|{/}(48\mathfrak{X}_1),\mathfrak{N}^{(k)}{/}(48\mathfrak{X}_2)\hspace{-0.5mm}\right\}\hspace{-0.5mm},
            \hspace{-4mm}
        \end{aligned}
        $}
        \end{equation}
        where
       $
            \mathfrak{N}^{(k)}=\varpi^k \left(1-\zeta^{(k)}\right)|\Upsilon_{\mathsf{min}}(\widetilde{\bm{\tau}}^{\downarrow,{(k)}})|\times\left(1- 4\eta_k^2\beta^2 \ell_{\mathsf{max}}^{(k)}\left(\ell_{\mathsf{max}}^{(k)}-1\right)\right).
        $
    \item $\bm{\mathfrak{C}^{(\eta)}}\rightarrow$ \textbf{Sufficient Condition on Step Size ($\eta_k$):}
        \begin{equation}\label{suf:main:eta}
        \eta_k = \alpha \big /{\sqrt{K \ell^{(k)}_{\mathsf{sum}}/N}},
        \end{equation}
        where $\ell^{(k)}_{\mathsf{sum}}=\sum_{b\in\Omega}\sum_{u\in \mathcal{U}_{b}} \ell_u^{(k)}$. Moreover, $\alpha$ must be chosen to satisfy {\small $\alpha<\sqrt{(\widehat{\ell}_{\mathsf{min}} K)/(4 N\beta^2 \ell_{\mathsf{max}}\left(\ell_{\mathsf{max}}-1\right))}$} and {\small $\eta_k \leq \min \{\frac{1}{2\beta} \sqrt{(\zeta^{(k)}- \widehat{\zeta}^{(k)})/\Delta}, \frac{1}{2\beta}\}$}, where {\small$\Delta=\ell_{\mathsf{max}}^{(k)}(\ell_{\mathsf{max}}^{(k)}-1)(2\mathfrak{X}_1+\zeta^{(k)})$}.
    \item $\bm{\mathfrak{C}^{(\sigma)}}\rightarrow$ \textbf{Sufficient Condition on Gradient Sampling Noise:}
        \begin{equation}\label{suf:main:sigma}
     \hspace{-7mm}\resizebox{0.45\textwidth}{!}{$
\max_{\substack{k\in\mathcal{K},\\b{\in}\Omega,\\u{\in}\mathcal{U}_{b}}}\vast\{\hspace{-1mm}\frac{\Big(1{-}\frac{{B}_{u}(\widetilde{\tau}_{b}^{\downarrow,{(k)}})}{|\Upsilon_{u}(\widetilde{\tau}_{b}^{\downarrow,{(k)}})|} \Big){(|\Upsilon_{u}(\widetilde{\tau}_{b}^{\downarrow,{(k)}})|{-}1)}\left(\sigma_{u}(\widetilde{\tau}_{b}^{\downarrow,{(k)}})\right)^2}{|\Upsilon_{u}(\widetilde{\tau}_{b}^{\downarrow,{(k)}})|{B}_{u}(\widetilde{\tau}_{b}^{\downarrow,{(k)}})}\hspace{-1mm}\vast\}{\leq} \sigma_{\mathsf{max}}.$}\hspace{-3mm}
        \end{equation}
\end{enumerate}
\end{theorem}
\begin{proof}
The proof is carried out in Appendix~\ref{app:th:sufficient_conditions}.
\end{proof}
\vspace{-2mm}
\textbf{Key interpretations.}
One of the key findings of Theorem~\ref{th:sufficient_conditions} is that for {\tt DCLM} to converge, we must recruit FLUs whose cumulative dataset size (i.e., $|\Upsilon^{\mathsf{s}}(\widetilde{\bm{\tau}}^{\downarrow,{(k)}})|$) satisfies inequality~\eqref{suf:main:recruitment222} at each global round $k$. Another key finding pertains to sufficient condition $\bm{\mathfrak{C}^{(L)}}$. Specifically, if $\mathfrak{D}^{(k)}_u{\rightarrow} 0,~\forall u{\in}\mathcal{U}_{b}$, the right-hand side of~\eqref{suf:main:gamma_upper} approaches infinity. However, for a large value of $\mathfrak{D}^{(k)}_u,~\forall u{\in}\mathcal{U}_{b}$, the second fraction in the right-hand side of~\eqref{suf:main:gamma_upper} becomes small. In this situation, for \eqref{suf:main:gamma_upper} to hold, as the model drifts of FLUs increase, more FLUs should be recruited and the communication latencies, integrated in the left hand side of \eqref{suf:main:gamma_upper} via~\eqref{TK}, should decrease, implying faster communications.


Corollary \ref{cor:sufficient_condition} in Appendix~\ref{app:cor:sufficient_condition} provides the convergence bound of {\tt DCLM} under sufficient conditions given in Theorem~\ref{th:sufficient_conditions}.

\vspace{-1mm}
\section{{\tt DCLM} Network Orchestration}\label{sec:PF}
 \vspace{-1.1mm}
\noindent In {\tt DCLM}, we aim to jointly optimize the \textit{ML-related control decisions} (Sec.~\ref{sec:FL_model}) and \textit{dynamic wireless control decisions} (Sec.~\ref{sec:dynamic_control_decisions}), which, to our knowledge, is among the most generic network optimization objectives for FL. Specifically, we formulate {\tt DCLM} as the following  optimization problem $\bm{\mathcal{P}}$:
 \vspace{-0.5mm}
 \begin{align}\label{opt:main}
     &\footnotesize (\bm{\mathcal{P}}):\min \bigg[\underbrace{c_1 \hspace{-.1mm}\frac{1}{K}\sum_{k=0}^{K-1}\mathbb E\left[ \big\Vert \nabla \mathfrak{L}^{({k})}(\bm{\omega}^{({k})})\big\Vert^2\right]}_{(a)}+\underbrace{c_2\sum_{k=0}^{K-1}\sum_{b \in \Omega} E_{b}^{\downarrow,{(k)}}}_{(b)}\nonumber\\[-.2em] 
     &\footnotesize {+}\underbrace{c_2\sum_{k=0}^{K-1}\sum_{b \in \Omega}\sum_{u{\in}\mathcal{U}_{b}}\overline{\lambda}_{u}^{(k)}\big( E_{u}^{\mathsf{LC},{(k)}}+\overline{E}_{u}^{\uparrow,(k)}\big)+\lambda_{u}^{(k)}\big(E_{u}^{\mathsf{LC},{(k)}}+E_{u}^{\uparrow,(k)}\big)}_{(c)}\bigg] \\[-.7em]
     &\hspace{-1mm}\textrm{{\textbf{s.t.}}}\nonumber\\[-0.2em]
&\hspace{-1mm}\textrm{\underline{\textit{\textbf{Constraints}}}:}\nonumber\\[-0.2em]
     &\sbullet[0.7]\text{\small Scheduling Decisions} \textrm{: }{\small\eqref{cons:mac_1},\eqref{cons:mac_2},  \eqref{cons:mac_4}, \eqref{cons:mac_5},  \eqref{cons:mac_7}, \eqref{cons:mac_8}}\nonumber\\[-0.25em]
     &\sbullet[0.7]\text{\small PRB/Power Allocation} \textrm{: } {\small\eqref{cons:mac_r_1}, \eqref{cons:mac_r_2}, \eqref{cons:mac_r_3}, \eqref{cons:mac_r_4}, \eqref{cons:mac_p_1}, \eqref{cons:mac_p_2}, \eqref{cons:mac_p_3}}\nonumber\\[-0.25em]
     &\sbullet[0.7]\text{\small GM/GV Transmissions} \textrm{: }{\small\eqref{cons:LGPD1},\eqref{cons:LGPD2},\eqref{cons:LMD1}},\eqref{cons:r2cc2r1}\nonumber\\[-0.25em]
     &\sbullet[0.7]\text{\small Training Latency and Energy Consumption} \textrm{: }
     \eqref{const:latency_max},\eqref{const:energy:max_DPU},\eqref{const:energy:max_CHU}\nonumber\\[-0.25em]
     &\sbullet[0.7]\text{\small Sufficient ML Convergence Conditions} \textrm{: }{\small \eqref{suf:main:gamma_upper}, \eqref{suf:main:recruitment222}, \eqref{suf:main:eta}, \eqref{suf:main:sigma} } \nonumber\\[-0.25em]
&\hspace{-1mm}\textrm{\underline{\textit{\textbf{Variables}}}:}\nonumber\\[-0.25em]
     &\sbullet[0.7]\textrm{\small\textit{FLU Recruitment, SGD Mini-batch and Iterations, CPU Frequency}:}\nonumber\\[-0.2em]
     &{\footnotesize \hspace{-1mm}\Big\{\hspace{-0.5mm}[\lambda_{u}^{(k)}]_{\hspace{-0.3mm}u{\in} \mathcal{U}_{b}},\hspace{-0.5mm}[\overline{\lambda}_{u}^{(k)}]_{\hspace{-0.3mm}u{\in} \mathcal{U}_{b}}, \hspace{-0.5mm}[\varsigma^{(k)}_u]_{\hspace{-0.3mm}u{\in} \mathcal{U}_{b}}, \hspace{-0.5mm}[\ell^{(k)}_{u}]_{\hspace{-0.3mm}u{\in} \mathcal{U}_{b}}, \hspace{-0.5mm}[f^{(k)}_{u}]_{u{\in} \mathcal{U}_{b}}\hspace{-0.5mm}\Big\}_{\hspace{-0.5mm}b{\in}\Omega,k{\in}\mathcal{K}}}\nonumber\\[-0.1em]
     &\sbullet[0.7]\textrm{\small\textit{Scheduling Decisions of MAC scheduler}:}\nonumber\\[-0.4em]
     &~~~~~~\footnotesize \bigg\{t_x,\Big\{\beta^{\downarrow}_{b}\hspace{-0.3mm}(t_{x}),[\overline{\beta}^{\uparrow}_u(t_x)]_{u{\in} \mathcal{U}_{b}},[\beta^{\uparrow}_u(t_x)]_{u{\in} \mathcal{U}_{b}}\Big\}_{b{\in}\Omega}\bigg\}_{x{\in}\mathcal{N}^{(k)},k{\in}\mathcal{K}}\nonumber\\[-0.0em]
     &\sbullet[0.7]\textrm{\small\textit{PRB and Transmit Power Allocation through MAC Scheduler}:}\nonumber\\[-0.4em]
     &\Big\{[\overline{\varrho}_{u,u',r}(t_x)]_{u,u'{\in}\mathcal{U}_{b},r{\in}\overline{\mathcal{R}}_{b}},[\varrho_{u,r}(t_x)]_{u{\in}\mathcal{U}_{b},r{\in}\mathcal{R}_{b}},[\rho^{\downarrow}_{b,r}\hspace{-0.5mm}(t_{x})]_{r{\in}\mathcal{R}_{b}}\nonumber\\[-0.2em]
     &~~~~,[\overline{\rho}^{\uparrow}_{u,r}\hspace{-0.5mm}(t_{x})]_{u{\in\mathcal{U}_{b}},r{\in}\overline{\mathcal{R}}_{b}},[\rho^{\uparrow}_{u,r}\hspace{-0.5mm}(t_{x})]_{u{\in\mathcal{U}_{b}},r{\in}\mathcal{R}_{b}}\Big\}_{b{\in}\Omega,x{\in}\mathcal{N}^{(k)},k{\in}\mathcal{K}}\nonumber\\[-0.2em]
     &\sbullet[0.7]\textrm{\small\textit{The GM and GV Transmissions}:}\nonumber\\[-0.4em]
     &~~~~\Big\{[\varphi_{b,r}\hspace{-0.3mm}(t_{x})]_{r{\in}\mathcal{R}_{b}},[\overline{\psi}_{u,u',r}(t_x)]_{u,u'{\in}\mathcal{U}_{b},r{\in}\overline{\mathcal{R}}_{b}}\nonumber\\[-0.2em]
     &~~~~~~~~~~~~~~~~~~~~~~~,[\psi_{u,r}(t_x)]_{u{\in}\mathcal{U}_{b},r{\in}\mathcal{R}_{b}}\hspace{-1mm}\Big\}_{b{\in}\Omega,x{\in}\mathcal{N}^{(k)},k{\in}\mathcal{K}}\nonumber\hspace{-10mm}
 \end{align}

\textbf{Objective Function of $\bm{\mathcal{P}}$.} The objective function of $\bm{\mathcal{P}}$ draws a tradeoff between the ML performance (term $(a)$ given by \eqref{eq:gen_conv}), the total energy consumption of O-RAN (term $(b)$, where $E_{b}^{\downarrow,{(k)}}$ is given by \eqref{eq:EN_broad}), and the total energy consumption of DPUs and CHUs (term $(c)$, where $E_{u}^{\mathsf{LC},{(k)}}$, $\overline{E}_{u}^{\uparrow,(k)}$, and $E_{u}^{\uparrow,(k)}$, are given by \eqref{eq:EN_LC}, \eqref{eq:EN_DPU_UP}, and \eqref{eq:EN_CHU_UP}). These (potentially) competing objectives are weighted via non-negative coefficients $c_1\&c_2$, respectively. To achieve this tradeoff, $\bm{\mathcal{P}}$ aims to find the optimal values for the  variables mentioned above while satisfying a set of constraints classified into five categories, pertaining to \textit{scheduling decisions}, \textit{PRB and power allocation}, \textit{GM and GV transmissions}, \textit{training latency and energy consumption}, and \textit{sufficient ML convergence conditions}. 

\vspace{-2mm}
\begin{remark}
   Optimization problem $\bm{\mathcal{P}}$ is highly non-trivial to solve, and it will require numerous mathematical steps to obtain its solution. To improve the paper's readability, we do not delve into all the mathematics involved in solving the problem and refer the interested reader to Appendix~\ref{app:optTransform}. Instead, below, we discuss different aspects of $\bm{\mathcal{P}}$ that are of particular interest.
\end{remark} 

 \vspace{-2mm}
\textbf{Challenges Faced in Solving $\bm{\mathcal{P}}$.} The objective function and constraints of $\bm{\mathcal{P}}$ are highly non-convex and intricate. This is because of several factors, including (i) terms with negative signs in the convergence bound \eqref{eq:gen_conv}, (ii) multiplication between different terms appearing in the objective function and the constraints of $\bm{\mathcal{P}}$ (e.g., in the convergence bound \eqref{eq:gen_conv} and energy consumption calculations in Sec.~\ref{sec:energy_consumption}), (iii) logarithmic functions in the calculations of the data  rates (see Sec.~\ref{sec:data_transmission_rate}), (iv) recursive functions such as the global round completion time (see Sec.~\ref{sec:communication_latency}) and scheduling decision constraints (see Sec.~\ref{sec:mac_scheduler}), (v) and non-convex nested min/max functions involved in the constraints of $\bm{\mathcal{P}}$ and the calculations of the communication latencies of {\tt DCLM} (see Sec.~\ref{sec:communication_latency}).

\begin{figure}[t!]
    \centering
    \begin{subfigure}[t]{0.23\textwidth}
       \noindent\includegraphics[height=1.26in]{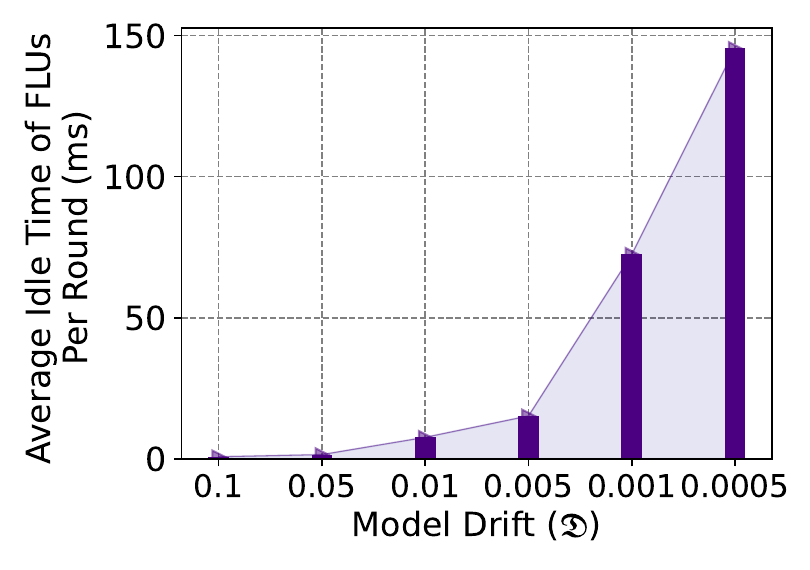} 
    \end{subfigure}%
    ~ 
    \begin{subfigure}[t]{0.23\textwidth}
        \centering
        \noindent\includegraphics[height=1.26in]{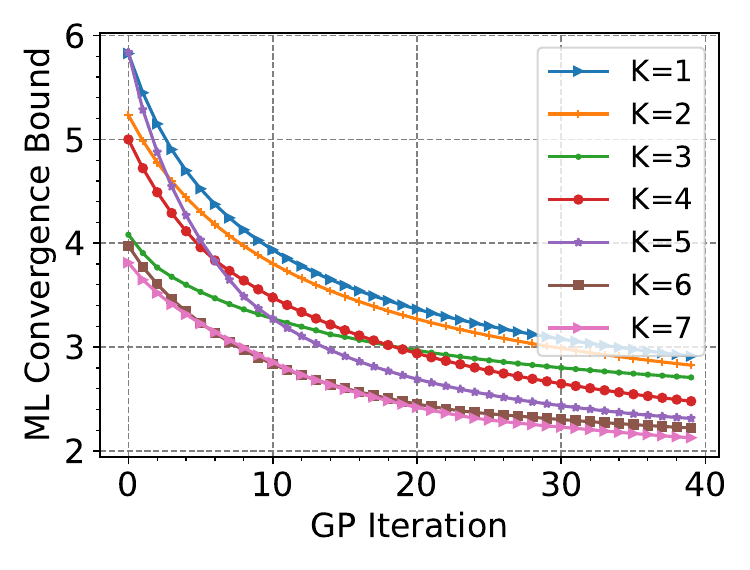} 
    \end{subfigure}
    \vspace{-1.5mm}
    \caption{Left plot: { As model drift decreases (e.g., $\mathfrak{D} = 0.0005$), data variations among FLUs diminish, allowing for longer idle times for FLUs (captured by $\sum_{b \in \Omega}\sum_{u\in \mathcal{U}_{b}} \mathfrak{D}^{(k)}_u$ in the condition $\bm{\mathfrak{C}^{(L)}}$ in Theorem~\ref{th:sufficient_conditions}}). Right plot: { The relationship between the number of global rounds $K$ and the ML convergence upper bound. An increase in $K$ results in a decrease in convergence bound}.}
    \label{fig:MD_omega_K}
    \vspace{-1.5mm}
\end{figure}

\textbf{Solution Design.} In Appendix~\ref{app:optTransform}, we draw a resemblance between $\bm{\mathcal{P}}$ and \textit{signomial programming} (SP) problems, which are highly non-convex and, in general, NP-hard. We then introduce and exploit a variety of techniques and approximations to transform the objective function and the constraints of $\bm{\mathcal{P}}$ into a set of \textit{monomials} and \textit{posynomials}, which transforms the problem into the format of \textit{geometric programming} (GP).  Using a logarithmic change of variables, GP, in its standard form, can be transformed into a convex optimization that can be efficiently solved with software such as CVXPY~\cite{diamond2016cvxpy}. \textit{Given the generality of $\bm{\mathcal{P}}$, the steps involved in our methodology are generalizable and can be applied to a wider range of optimization formulations under the umbrella of network-aware FL.} Nevertheless, they are quite lengthy and are moved to Appendix~\ref{app:optTransform}. In a nutshell, the optimization solver discussed in Appendix~\ref{app:optTransform} is an iterative solver, which leverages a combination of the arithmetic-geometric mean inequality \eqref{eq:approxPosMonMain}, Taylor-power approximation \eqref{eq:taylor_e}, and sum-power approximations, e.g., $\min\{A, B\}\approx (A^{-p}+B^{-p})^{-\frac{1}{p}}$ and $\max\{A, B\}\approx (A^{p}+B^{p})^{-\frac{1}{p}}$. 
Furthermore, we describe the dynamics of the dataset size of FLU $u$ through the following piece-wise ordinary differential equation, the closed form solution of which is derived in \eqref{app:eq:dataset_size_compact_closed_form} in Appendix~\ref{app:optTransform}:
\begin{equation}\label{app:eq:dynamic_dataset_size_ODE_main}
\frac{d|\Upsilon_{u}(t)|}{dt}{=}
\begin{cases}
   C_{u}^{\downarrow,(k)},&t{\in}[T^{(k{-}1)}, T^{(k{-}1)}{+}\tau_{b}^{\downarrow,{(k)}}),\\
   0,&t{\in}[T^{(k{-}1)}{+}\tau_{b}^{\downarrow,{(k)}},\Psi^{(k)}_{u}],\\
   C_{u}^{\uparrow,(k)},&t{\in}(\Psi^{(k)}_{u},T^{(k)}].
\end{cases}
\end{equation}
where $C_{u}^{\downarrow,(k)}$ and $C_{u}^{\uparrow,(k)}$ are the rate of growth and decay, respectively. Further, $\Psi^{(k)}_{u}$ is the time at which LM training of FLU $u$ is completed, which is given in \eqref{app:eq:psi_u_k} in Appendix~\ref{app:optTransform}. We also demonstrate that our GP solver (i.e. Algorithm~\ref{alg:cent} in Appendix~\ref{app:cons:sudo}) generates a series of solutions that converge to Karush–Kuhn–Tucker (KKT) conditions of $\bm{\mathcal{P}}$ (see Proposition~\ref{propo:KKT} in Appendix~\ref{app:optTransform}). { Also, complexity analysis and feasibility of solving our optimization problem $\bm{\mathcal{P}}$ using CVXPY are presented in Appendix~\ref{app:complexity}.}
\begin{figure}[t]
\centering
\noindent\includegraphics[width=9cm]{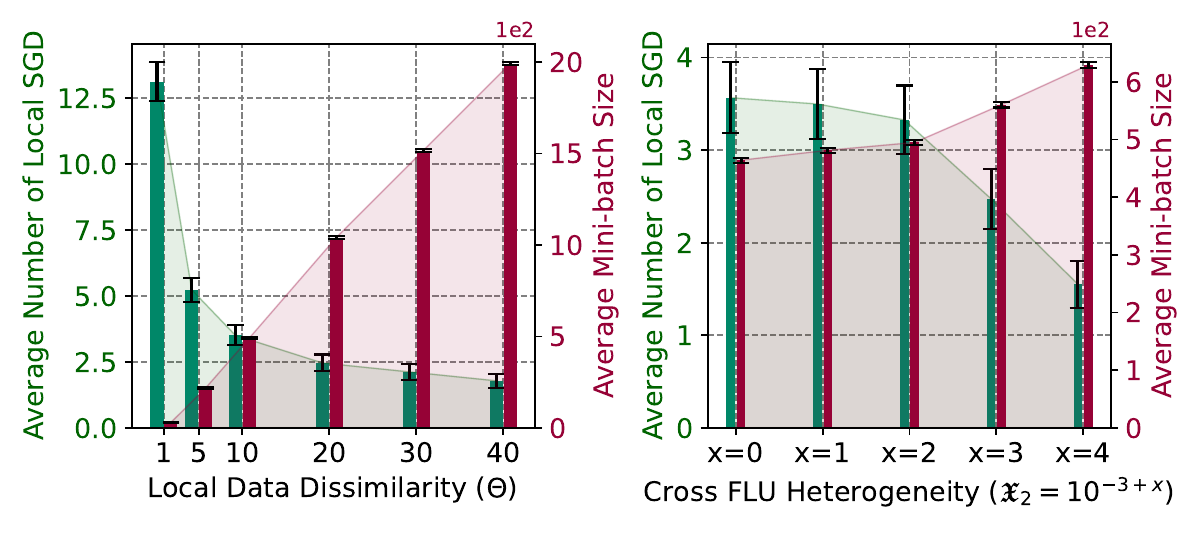} 
\vspace{-7.5mm}
\caption{Impact of $\Theta$ and $\mathfrak{X}_2$  on SGD iterations and mini-batch size of FLUs. Small values of $\Theta$ and $\mathfrak{X}_2$ lead to more local SGD iterations and a smaller mini-batch. Conversely, increasing $\Theta$ and $\mathfrak{X}_2$ prompts {\tt DCLM} to reduce SGD iterations and increase mini-batch size to avoid model bias.  { This aligns with the theoretical results discussed in \textit{Interpretation 1} in Sec.~\ref{sec:conv}.}}
\label{fig:zeta2_epochs}
\vspace{-1.5mm}
\end{figure}

In the following, we conduct numerical evaluations on the solution of $\bm{\mathcal{P}}$ obtained by our methodology.

\vspace{-1.0mm}
\section{Numerical Evaluation}\label{numerical_evaluation}
\noindent 
To conduct simulations, we use PyTorch \cite{paszke2019pytorch} and three Nvidia Tesla V100 GPUs, each equipped with 32 GB VRAM. The solution of $\bm{\mathcal{P}}$, derived in \eqref{opt:main}, is obtained using primal–dual interior-point method implemented in CVXPY \cite{diamond2016cvxpy}. We consider ML tasks on (i) MNIST \cite{9492755} dataset and (ii) CIFAR10 \cite{9492755} dataset. { The network configuration and the channel models used in the simulations are provided in Appendix \ref{app:network_settings}.}

\vspace{-1mm}
\subsection{Dynamic Network-Aware {\tt DCLM}: Ablation Study}
Examining $\boldsymbol{\mathcal{P}}$ directly is challenging due to the interdependencies between the optimization variables and the bound in~\eqref{eq:gen_conv}. We thus conduct an ablation study to systematically assess the isolated impacts of key optimization and scaling variables in $\boldsymbol{\mathcal{P}}$. The results in Secs. \ref{subsubsub:2}-\ref{sec:DCC} represent the averaged outcomes over $10$ Monte-Carlo iterations. To avoid repetition of results, in this section, we only consider MNIST dataset and use a 1200-kilobit ML model (convolutional neural network), later utilized in Sec.~\ref{sec:training_performance} for model training.


\subsubsection{Impacts of dynamic model drift}\label{subsubsub:2}
In the left plot of Fig.~\ref{fig:MD_omega_K}, we examine the influence of dynamic model drift {\small $\mathfrak{D}^{(k)}_{u}$} on the time that FLUs are idle. Here, we assume a constant model drift across all FLUs (i.e., $\mathfrak{D}^{(k)}_{u} = \mathfrak{D}, \forall u, k$). The figure shows that when model drift is high (e.g., $\mathfrak{D} = 0.1$), {\tt DCLM} reduces the idle time of FLUs in each global round to capture large changes in data distributions, as governed by term $(b)$ in \eqref{eq:gen_conv}. Conversely, lower model drifts (e.g., $\mathfrak{D} = 0.0005$) allow for longer idle times for FLUs.

\vspace{-.1mm}
\subsubsection{Impacts of number of global rounds}\label{subsubsub:4}
In the right plot of Fig.~\ref{fig:MD_omega_K}, we illustrate how an increase in the number of global rounds ($K$) reduces the ML convergence bound, aligning with the convergence guarantee of {\tt DCLM} from Theorem \ref{th:sufficient_conditions}. 

\begin{figure}[t]
\centering
\noindent\includegraphics[width=8cm]{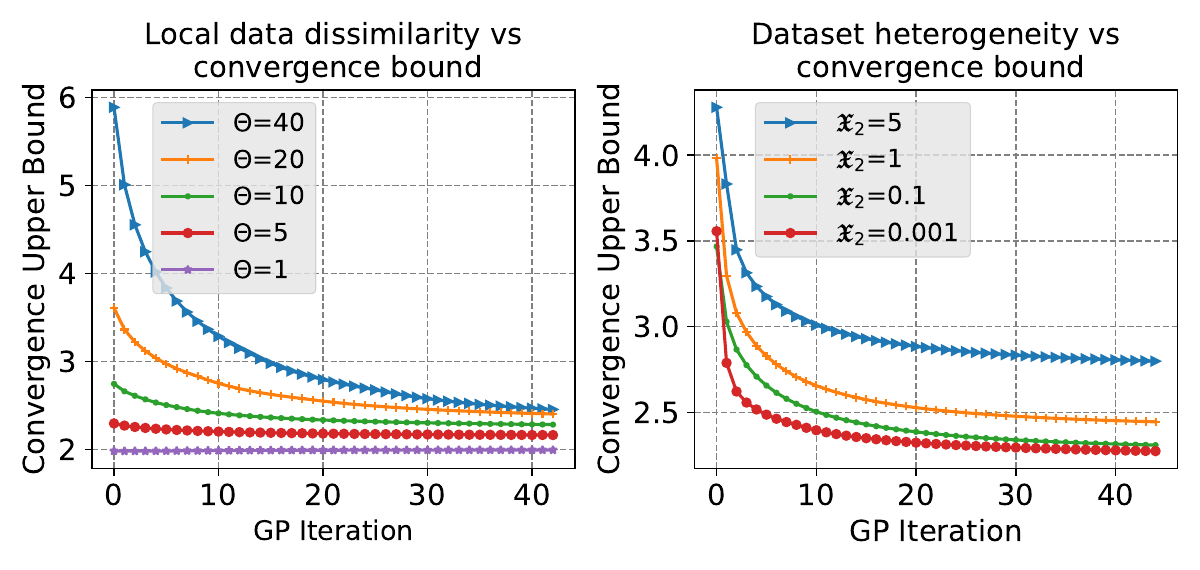} 
\vspace{-1.5mm}
\caption{Impact of $\Theta$ and $\mathfrak{X}_2$ on the ML convergence bound over the GP solver iterations. Larger values of $\Theta$ and $\mathfrak{X}_2$  are associated with slower convergence of the GP solver and an increased ML convergence upper bound. { This is because convergence bound \eqref{eq:gen_conv} is directly proportional to  $\Theta$ and $\mathfrak{X}_2$ as reflected in terms $(c), (d), (e), (f)$ and $(g)$ in the bound.} }
\label{fig:zeta2_epochs2}
\vspace{-1.1mm}
\end{figure}

\subsubsection{Local data dissimilarity and cross-FLU data heterogeneity}\label{subsubsub:3}
 In Fig. \ref{fig:zeta2_epochs}, we vary local data dissimilarity and cross-FLU data heterogeneity to demonstrate their impact on the number of SGD iterations and mini-batch size across FLUs. This figure shows that when local data dissimilarity of FLUs $\Theta$ and dataset heterogeneity across FLUs $\mathfrak{X}_2$ are small, {\tt DCLM} maximizes ML performance by having more local SGD iterations at FLUs and a smaller mini-batch. Also, as $\Theta$ and $\mathfrak{X}_2$ increase, {\tt DCLM} reduces the number of SGD iterations and increases the mini-batch size of FLUs to avoid local model bias. Further, the influence of $\mathfrak{X}_2$ and $\Theta$ on the ML convergence upper bound and the convergence behavior of our GP solver is depicted in Fig. \ref{fig:zeta2_epochs2}, showing that larger values of $\mathfrak{X}_2$ and $\Theta$ result in slower convergence of the GP solver and a larger convergence bound.

\vspace{-.1mm}
\subsubsection{Impacts of MAC scheduler}\label{subsubsub:5}
In the left plot of Fig. \ref{fig:scheduling_recruitment}, we demonstrate how  MAC scheduler can enhance the overall system performance: an increase in the number of FGTIs significantly reduces energy consumption of FLUs. This is because upon  increasing the number of FGTIs, fewer FLUs transfer their models simultaneously on a shared PRB, thereby reducing interference on communication channels, which, in turn, reduces energy consumption.

\subsubsection{Cross FLU data heterogeneity vs FLU recruitment}\label{subsubsub:6}
Considering the ML convergence sufficient condition $\bm{\mathfrak{C}^{(\Upsilon)}}$ in \eqref{suf:main:recruitment222}, cross-FLU data heterogeneity, captured via $\mathfrak{X}_1$, directly influences the maximum number of FLUs allowed to remain unrecruited. We show this in the right plot of Fig.~\ref{fig:scheduling_recruitment} by varying $\mathfrak{X}_1$ for different total numbers of FLUs in each O-RU (with $\mathcal{U}_b\in[10, 20,40]$). As observed, an increase in $\mathfrak{X}_1$ reduces the number of unrecruited FLUs, aligning with \eqref{suf:main:recruitment222}. Further, an increase in the total number of FLUs in each O-RU raises the maximum allowable number of unrecruited FLUs. { Additionally, in Appendix~\ref{app:fairness}, we study the fairness of recruiting FLUs as CHUs in terms of energy consumption.}

\begin{figure}[t!]
    \centering
    \begin{subfigure}[t]{0.23\textwidth}
       \noindent\includegraphics[height=1.26in,trim=5 5 5 5, clip]{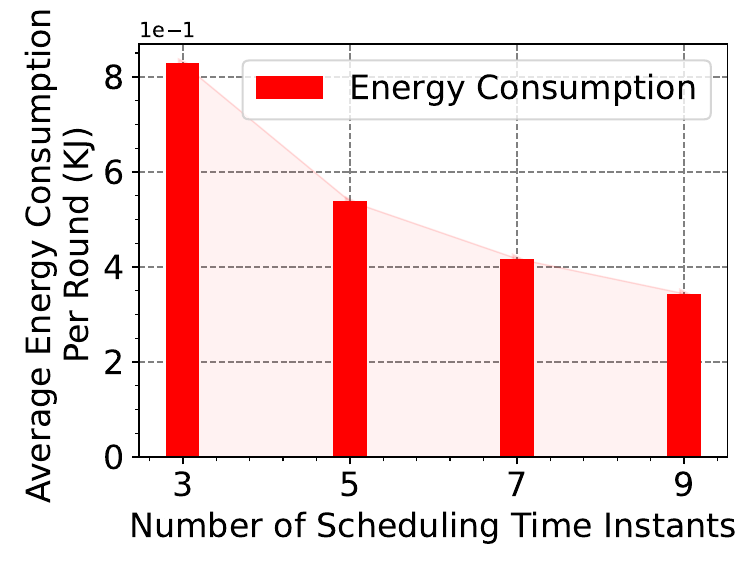} 
    \end{subfigure}%
    ~ 
    \begin{subfigure}[t]{0.23\textwidth}
        \centering
        \noindent\includegraphics[height=1.26in,trim=5 5 5 5, clip]{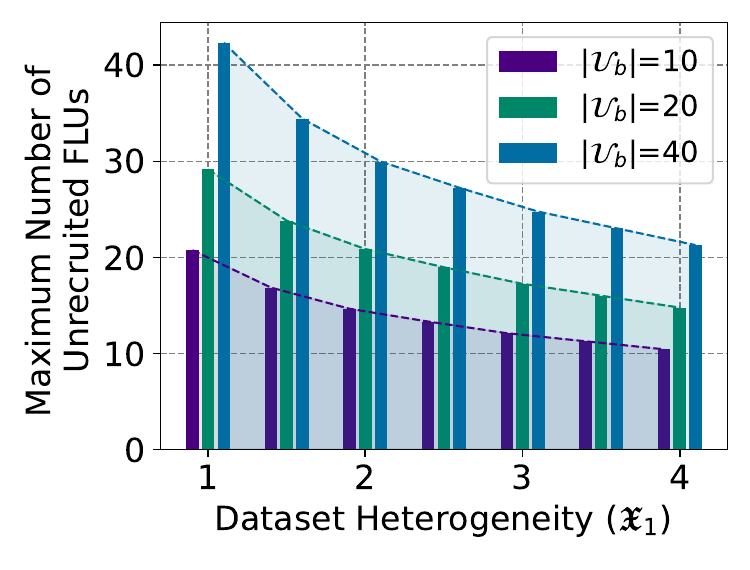} 
    \end{subfigure}
    \vspace{-1.5mm}
    \caption{Left plot: Increasing the number of FGTIs leads to a substantial reduction in average energy consumption of FLUs per global round { as highers number of FGTIs allows for scheduling FLUs at non-overlapping FGTIs.} Right plot: An increase in the number of FLUs in each O-RU ($|\mathcal{U}_b|$) allows for a higher number of unrecruited FLUs. Conversely, an increase in $\mathfrak{X}_1$ leads to a reduction in the number of unrecruited FLUs. { This is because larger values of $\mathfrak{X}_1$ (i.e., higher data heterogeneity) calls for a larger number of recruited FLUs as reflected in condition $\bm{\mathfrak{C}^{(\Upsilon)}}$ in Theorem~\ref{th:sufficient_conditions}.}}
    \label{fig:scheduling_recruitment}
\end{figure}
\subsubsection{Impacts of DCC on energy consumption}\label{sec:DCC}
{ To demonstrate the effectiveness of DCC, we evaluate {\tt DCLM} under different combinations of the numbers of CHUs and DPUs  while keeping the total number of FLUs constant at $30$ (i.e., six FLUs per O-RU).
Specifically, in Fig.~\ref{fig:ext:no_DPU}, we examine six FLU combinations per O-RU: (a) 1 CHU and 5 DPUs, (c) 2 CHUs and 4 DPUs, (d) 3 CHUs and 3 DPUs, (e) 4 CHUs and 2 DPUs, (f) 5 CHUs and 1 DPU, (g) 6 CHUs and 0 DPU.
The extreme case with 6 CHUs and no DPU (i.e., the right-most  bar) significantly increases energy consumption. This is primarily because of two reasons. First, without DPUs communicating over unlicensed PRBs, the licensed PRBs must be shared among more CHUs, resulting in increased traffic and interference over the uplink channels to O-RUs. Second, with the inclusion of DPUs, users located far from O-RUs can transmit their local models to nearby CHUs with lower-power D2D communications over unlicensed PRBs; however, without DPUs, these users must use more transmit power to transfer their local models to distant O-RUs. A similar result is observed in the extreme case with only 1 CHU and 5 DPUs, where all DPUs -- regardless of their closeness to the CHU -- must transmit their models to the selected CHU in D2D communication mode, which is not efficient upon having long distances between the CHU and DPUs. The best combination is when there are 3 CHUs and 3 DPUs, which efficiently distributes the communication burden over the licensed and unlicensed PRBs via taking the full advantage of the closeness of DPUs to CHUs and the proximity of CHUs to the O-RUs.}
\begin{figure}[!t]
\centering
\noindent\includegraphics[height=1.26in]{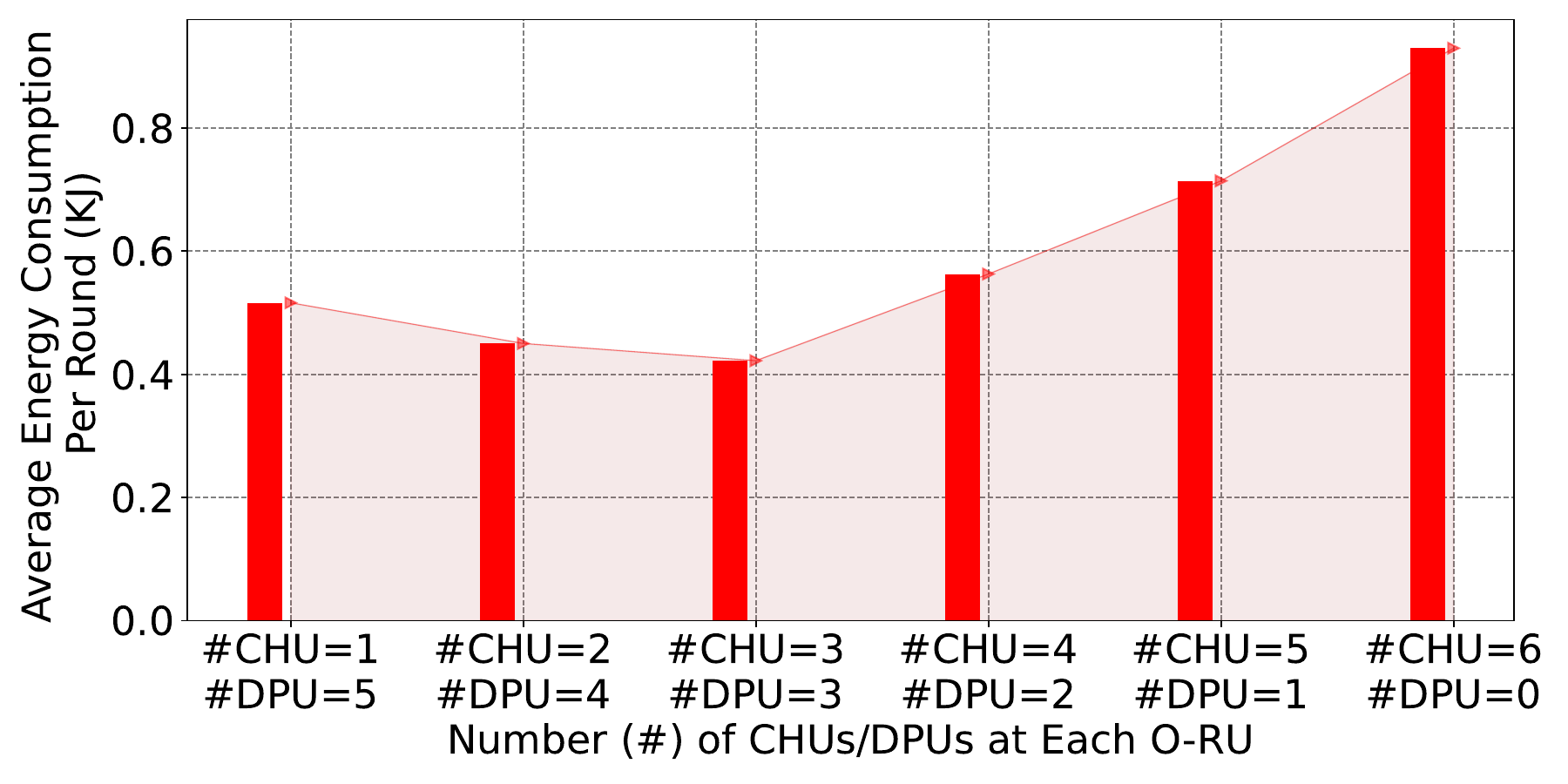} 
\vspace{-1mm}
\caption{ Effect of the CHU and DPU configuration on the energy consumption. The extreme case with 6 CHUs and no DPU leads to the highest energy consumption due to increased traffic and interference on licensed PRBs. The best configuration is having 3 CHUs and 3 DPUs, which effectively distributes the communication load to licensed and unlicensed PRBs through DCC.}
\label{fig:ext:no_DPU}
\vspace{-1.1mm}
\end{figure}

\vspace{-1mm}
\subsection{{\tt DCLM} Model Training Performance}\label{sec:training_performance}
\subsubsection{Simulation setup} We compare {\tt DCLM} with three baselines: FL-EOCD~\cite{10061474}, DISCO~\cite{guo2021dynamic}, and AHFLP~\cite{su2025joint}. FL-EOCD is an optimization theory-based approach that minimizes energy consumption through PRB allocation and a D2D communication model but lacks a MAC scheduler and multi-channel ML model dispersion. DISCO employs a sequential wireless resource scheduling policy by establishing a queue using the Lyapunov optimization framework to determine the order in which users transmit their LMs to the server (the  available wireless spectrum is allocated to the users in the queue, one at a time).  AHFLP is also an optimization theory-based approach and adopts a hierarchical aggregation strategy, applying different aggregation frequencies $l_1$ and $l_2$ at the edge and cloud levels, respectively. It performs joint optimization of $l_1$, $l_2$, bandwidth allocation, and CPU frequency of users, to minimize FL loss under time and energy constraints, but does not support D2D communication. We note that FL-EOCD, DISCO, and AHFLP are not directly comparable to our method: FL-EOCD and DISCO do not account for communication interference and consider only a single O-RU, while AHFLP supports multiple O-RUs but neglects interference. To ensure a fair comparison, we evaluate {\tt DCLM} against extended versions of baselines, adapted to support communication through multiple O-RUs in the presence of interference. Also, to demonstrate the effectiveness of the dedicated FL MAC scheduler of {\tt DCLM}, we consider two different variants of {\tt DCLM}, including (i) {\tt DCLM(1)}:  resource allocation is performed at only one FGTI at the beginning of each global round, (ii) {\tt DCLM(7)}: resource allocation is performed at $7$ FGTIs during each global round.

\subsubsection{Energy consumption}
{ Fig.~\ref{fig:energy_consumption} shows the energy consumption of {\tt DCLM(7)}, {\tt DCLM(1)}, FL-EOCD, and DISCO for MNIST (left plot) and CIFAR10 (right plot). As shown, both {\tt DCLM(1)} and {\tt DCLM(7)} consume less energy compared to FL-EOCD and DISCO. This is because the {\tt DCLM} variants effectively minimize interference, thereby reducing energy consumption, through DCC multi-channel ML model dispersion over licensed and unlicensed PRBs. Additionally, {\tt DCLM(7)} has the lowest energy consumption since, in addition to DCC, it uses a dedicated MAC scheduler to schedule FLUs at multiple non-overlapping FGTIs, further reducing interference.}
\begin{figure}[t!]
    \centering
    \begin{subfigure}[t]{0.235\textwidth}
       \noindent\includegraphics[height=1.26in,trim=5 5 5 5, clip]{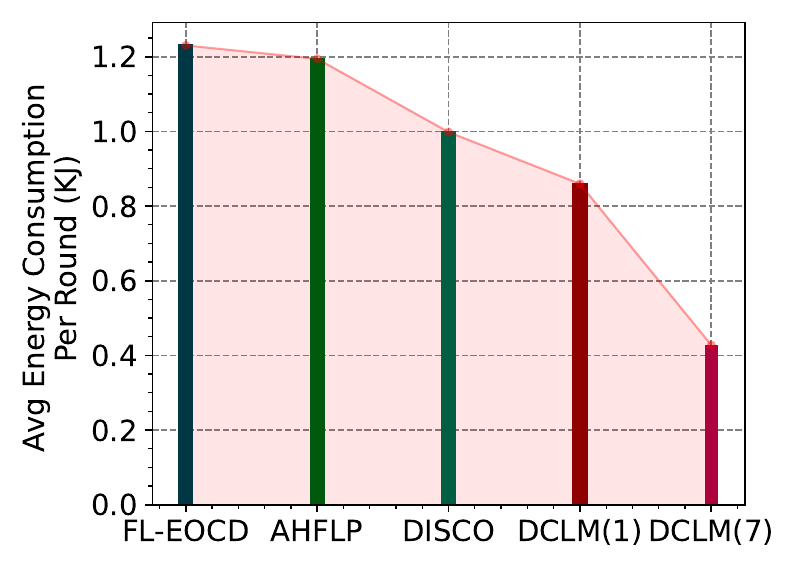} 
    \end{subfigure}%
    ~ 
    \begin{subfigure}[t]{0.235\textwidth}
        \centering
        \noindent\includegraphics[height=1.26in,trim=5 5 5 5, clip]{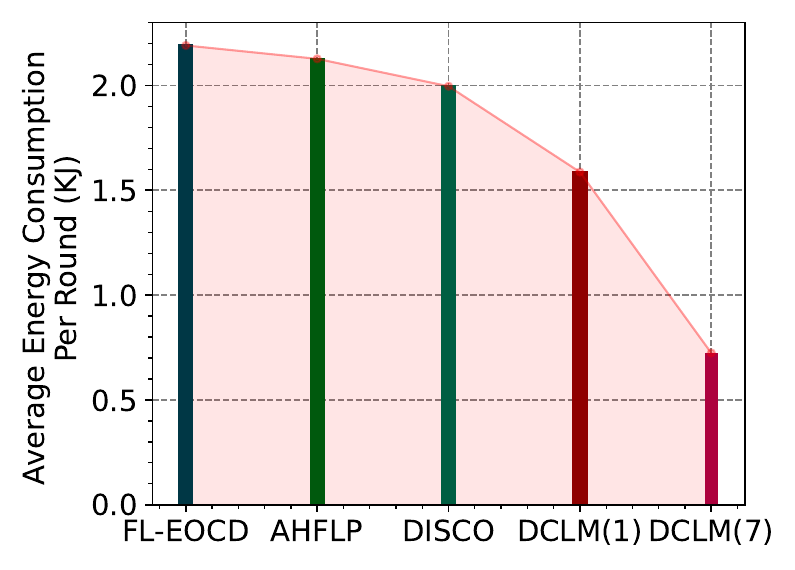} 
    \end{subfigure}
    \vspace{-1mm}
    \caption{ Energy consumption of {\tt DCLM(7)}, {\tt DCLM(1)}, FL-EOCD, and DISCO for MNIST (left) and CIFAR10 (right). {\tt DCLM} variants consume less energy by minimizing interference through DCC model dispersion, with {\tt DCLM(7)} achieving the lowest consumption due to its multi-time MAC scheduler.}
    \label{fig:energy_consumption}
    \vspace{-1.5mm}
\end{figure}

\subsubsection{FL model accuracy}
{ Tables~\ref{tbl:training_proof_mnist}\&\ref{tbl:training_proof_CFIAR10} present the average ML test accuracy of {\tt DCLM(7)}, {\tt DCLM(1)}, FL-EOCD, and DISCO on MNIST and CIFAR10 datasets after $40$ global rounds. We employ a Dirichlet distribution~\cite{li2021model} with varying concentration parameters $\alpha^{\mathsf{D}} \in \{5, 0.5, 0.05\}$ to partition the MNIST and CIFAR-10 datasets among FLUs, allowing us to study how different levels of data heterogeneity affect {\tt DCLM}'s performance. Here, $\alpha^{\mathsf{D}}{=}0.05$ introduces extreme data heterogeneity~\cite{li2021model}. Fig.~\ref{fig:ext:het} in Appendix~\ref{app:datahetrogenity} visualizes class distributions under each $\alpha^{\mathsf{D}}$ value.
In this simulation, following the first realization of FLUs' datasets, in the subsequent global rounds, FLUs obtain a new dataset of size dictated by the piece-wise differential equation in~\eqref{app:eq:dynamic_dataset_size_ODE_main}.
The second column of both tables shows the average accuracy and standard deviation across 20 runs of {\tt DCLM(7)}. For {\tt DCLM(1)}, AHFLP, FL-EOCD, and DISCO, each column provides: (i) the average accuracy and standard deviation and (ii) the \textit{P-value from the Wilcoxon test}\footnote{The Wilcoxon test is a non-parametric statistical test used to compare paired samples and assess whether their population distributions differ significantly~\cite{10750029}. In this paper, it evaluates the significance of differences in ML test accuracy between {\tt DCLM(7)} and the other three algorithms.} comparing {\tt DCLM(7)} with the respective algorithm at a significance level of 0.05~\cite{10750029}. Referring to Tables~\ref{tbl:training_proof_mnist}\&\ref{tbl:training_proof_CFIAR10}, {\tt DCLM(7)} statistically outperforms  {\tt DCLM(1)}, AHFLP, FL-EOCD, and DISCO for the three levels of data heterogeneity on both MNIST and CIFAR10. This is because, in {\tt DCLM(7)}, FLUs consume significantly less energy and can participate in more global rounds. For instance, AHFLP and FL-EOCD consume twice as much energy as {\tt DCLM(7)} (see Fig.~\ref{fig:energy_consumption}). This highlights the effectiveness of (i) multi-channel model transmission over licensed/unlicensed PRBs using the DCC, and (ii) scheduling FLUs at non-overlapping FGTIs through the dedicated MAC scheduler of {\tt DCLM}.}

\setlength{\tabcolsep}{2.6pt}
\begin{table}[t]
 
\centering
\scriptsize
\caption{  The table presents the average ML test accuracy and standard deviation over 20 independent runs on \textbf{MNIST} for {\tt DCLM(7)}, {\tt DCLM(1)}, FL-EOCD, and DISCO. Additionally, it includes the Wilcoxon test comparing {\tt DCLM(7)} with the other three algorithms. An algorithm with a mean test accuracy lower than {\tt DCLM(7)} and $p \le 0.05$ is statistically worse than {\tt DCLM(7)}.}
\vspace{-1.5mm}
\begin{tabular}{|c||c|c|c|c|c|}
\hline
 {\cellcolor[HTML]{c7e8ff} $\alpha^{\mathsf{D}}$ }  &{\cellcolor[HTML]{c7e8ff}DCLM(7)} & \cellcolor[HTML]{c7e8ff}DCLM(1)& \cellcolor[HTML]{c7e8ff}AHFLP~\cite{su2025joint}& \cellcolor[HTML]{c7e8ff}FL-EOCD~\cite{10061474} & \cellcolor[HTML]{c7e8ff}DISCO~\cite{guo2021dynamic} \\ \hline
$5$ & \makecell{\bm{$97\%$}\\\bm{${\pm}0.0013$}} & \makecell{$95\%{\pm}0.0018$\\$(p=0.0019)$}& \makecell{$94\%{\pm}0.0021$\\$(p=0.0019)$}&\makecell{$94\%{\pm}0.0055$\\$(p=0.0019)$} & \makecell{$94\%{\pm}0.0025$\\$(p=0.0019)$} \\ \hline
$0.5$ & \makecell{$\bm{96\%}$\\$\bm{{\pm}0.0018}$}&\makecell{$94\%{\pm}0.0035$\\$(p=0.0019)$}&\makecell{$93\%{\pm}0.0012$\\$(p=0.0019)$}  & \makecell{$92\%{\pm}0.0037$\\$(p=0.0019)$}& \makecell{$93\%{\pm}0.0052$\\$(p=0.0019)$}  \\ \hline
$0.05$ & \makecell{$\bm{90\%}$\\$\bm{{\pm}0.022}$}&\makecell{$85\%{\pm}0.031$\\$(p=0.0488)$}&\makecell{$81\%{\pm}0.062$\\$(p=0.0253)$}   &\makecell{$81\%{\pm}0.0325$\\$(p=0.0018)$} & \makecell{$82\%{\pm}0.0425$\\$(p=0.005)$}\\
\hline  
\end{tabular}
\label{tbl:training_proof_mnist}
\vspace{-1mm}
\end{table}
\begin{table}[t]
\centering
\scriptsize
\caption{ The table presents the average ML test accuracy and standard deviation over 20 independent runs on \textbf{CIFAR10} for {\tt DCLM(7)}, {\tt DCLM(1)}, FL-EOCD, and DISCO. Additionally, it includes the Wilcoxon test comparing {\tt DCLM(7)} with the other three algorithms. An algorithm with a mean test accuracy lower than {\tt DCLM(7)} and $p \le 0.05$ is statistically worse than {\tt DCLM(7)}. }
\vspace{-1.5mm}
\begin{tabular}{|c||c|c|c|c|c|}
\hline
 {\cellcolor[HTML]{c7e8ff} } $\alpha^{\mathsf{D}}$&{\cellcolor[HTML]{c7e8ff}DCLM(7)} & \cellcolor[HTML]{c7e8ff}DCLM(1)& \cellcolor[HTML]{c7e8ff}AHFLP~\cite{su2025joint}& \cellcolor[HTML]{c7e8ff}FL-EOCD~\cite{10061474} & \cellcolor[HTML]{c7e8ff}DISCO~\cite{guo2021dynamic} \\ \hline
$5$ & \makecell{\bm{$58\%$}\\\bm{${\pm}0.0063$}} & \makecell{$43\%{\pm}0.0098$\\$(p=0.0019)$}& \makecell{$35\%{\pm}0.0028$\\$(p=0.0019)$}&\makecell{$34\%{\pm}0.0075$\\$(p=0.0019)$} & \makecell{$37\%{\pm}0.0054$\\$(p=0.0019)$} \\ \hline
$0.5$ & \makecell{$\bm{54\%}$\\$\bm{{\pm}0.0088}$}&\makecell{$40\%{\pm}0.0085$\\$(p=0.0012)$}&\makecell{$27\%{\pm}0.0015$\\$(p=0.0024)$}  & \makecell{$24\%{\pm}0.0072$\\$(p=0.0003)$}& \makecell{$32\%{\pm}0.0097$\\$(p=0.0019)$}  \\ \hline
$0.05$ & \makecell{$\bm{31\%}$\\$\bm{{\pm}0.003}$}&\makecell{$23\%{\pm}0.016$\\$(p=0.0132)$}&\makecell{$17\%{\pm}0.036$\\$(p=0.0012)$}   &\makecell{$16\%{\pm}0.0225$\\$(p=0.0020)$} & \makecell{$20\%{\pm}0.0150$\\$(p=0.0014)$}  \\ \hline
\end{tabular}
\label{tbl:training_proof_CFIAR10}
\vspace{-1.5mm}
\end{table}

\section{Conclusion and Future Work}\label{conclusion}
\noindent We introduced {\tt DCLM}, incorporating dynamic channels and time-varying FLUs' datasets into FL. {\tt DCLM} introduces dynamic ML and wireless control decisions, which are optimized via introducing \textit{dedicated FL MAC schedulers} over O-RAN and efficient implementation of DCC mode. This was carried out through formulating a non-convex optimization problem and proposing a generalizable solution methodology. Our theoretical analysis and analytical optimization solver form the baseline for studying dynamic resource allocation and MAC scheduler for FL over O-RAN. Specifically, with AI/ML techniques integrated into O-RAN intelligent controllers (i.e., RICs), we expect them to be compared to our analytical approach. 
Furthermore, a comprehensive set of future works were presented in the paper and Appendix~\ref{app:future_research}.


 

\bibliographystyle{IEEEtran}
\bibliography{ref}

\vspace{-12.15mm}
\begin{IEEEbiographynophoto}{Payam Abdisarabshali}[S] is currently a PhD student at EE Department at the University at Buffalo--SUNY, USA.
\end{IEEEbiographynophoto}
\vspace{-12.15mm}
\begin{IEEEbiographynophoto}{Kwang Taik Kim}[SM] is currently a Research Assistant Professor with the Elmore Family School of ECE at Purdue University, USA.
\end{IEEEbiographynophoto}
\vspace{-12.15mm}
\begin{IEEEbiographynophoto}{Michael Langberg}[F] is currently a Professor with the EE Department at the University at Buffalo--SUNY, USA.
\end{IEEEbiographynophoto}
\vspace{-12.15mm}
\begin{IEEEbiographynophoto}{Weifeng Su}[F] is currently a Professor with the EE Department at the University at Buffalo--SUNY, USA.
\end{IEEEbiographynophoto}
\vspace{-12.15mm}
\begin{IEEEbiographynophoto}{Seyyedali Hosseinalipour}[SM] is currently an Assistant Professor with the EE Department at the University at Buffalo--SUNY, USA.
\end{IEEEbiographynophoto}


\vfill
\begingroup
\onecolumn

\appendices
\newpage
\section{Notations Table}\label{app:notaions}

\setlength{\tabcolsep}{4pt}
\begin{table}[h]
\small
\centering
\caption{ { Major notations used in the paper.}}
\label{table:notations}
{ 
\begin{tabular}{|c|p{6.4cm}||c|p{6.4cm}|}
\hline
\rowcolor[HTML]{848884} 
\textbf{Notation} & \textbf{Description} & \textbf{Notation} & \textbf{Description}  \\ 
\hline
\hline
{ $\Omega$} & { Set of O-RUs}  & $b$ & O-RU index \\ \hline
 \rowcolor[HTML]{D3D3D3} { $r$} & { PRB index} &{ $u$} & { FLU index} \\ \hline
{ $k$} & { Global round index} & $\mathcal{R}_{b}$ & Set of licensed PRBs of O-RU $b$ \\ \hline
\rowcolor[HTML]{D3D3D3}$\overline{\mathcal{R}}_{b}$ & Set of unlicensed PRBs of O-RU $b$ & { $P^{\mathsf{max}}_{u}$} & { Maximum transmit power of FLU $u$} \\ \hline
{$P^{\mathsf{max}}_b$} & { Maximum downlink transmit power of O-RU b} & { $f^{(k)}_{u}$} & { Available CPU frequency of FLU $u$} \\ \hline
\rowcolor[HTML]{D3D3D3}$\tau_{b}^{\downarrow,{(k)}}$ & Broadcast (downlink) latency of O-RU $b$ & {  $B$} & { Bandwidth of each licensed PRB} \\ \hline
{ $\overline{B}$} & { Bandwidth of each unlicensed PRB} & { $\gamma_1$} & { Numerologies of licensed PRBs} \\ \hline
\rowcolor[HTML]{D3D3D3} { $\gamma_2$} & { Numerologies of unlicensed PRBs} & $\mathcal{U}_b$ & Set of FLUs of O-RU $b$ \\ \hline
$\tau_{u}^{\mathsf{LC},{(k)}}$ & Local computation/training time of FLU $u$ & $\mathcal{K}$ & Set of global training rounds \\ \hline
\rowcolor[HTML]{D3D3D3}$T^{(k)}$ & Completion time of global round $k$ & { $T^{\mathsf{max}}$} & { Maximum allowable latency per global round} \\ \hline
{ $\Psi(\cdot)$} & { Occurrence time of a $\mathscr{D}$-Event} & $\overline{\tau}_{u}^{\uparrow,{(k)}}$ & LM upload latency of DPU $u$ \\ \hline
\rowcolor[HTML]{D3D3D3}$\tau_{u}^{\mathsf{W},{(k)}}$ & Waiting time of CHU $u$ & $\tau_{u}^{\uparrow,{(k)}}$ & LM upload latency of CHU $u$ \\ \hline
$t_x$ & Fine granular time instant (FGTI) & $\mathcal{N}^{(k)}$ & Set of the indices of FGTIs \\ \hline
\rowcolor[HTML]{D3D3D3}$\widehat{\lambda}_{u}^{(k)}$ & Recruitment status of FLU $u$ & {$\lambda_{u}^{(k)}$} & { Recruitment decision variable of CHU $u$} \\ \hline
{$\overline{\lambda}_{u}^{(k)}$} & { Recruitment decision variable of DPU $u$} & $\bm{\omega}^{(k)}$ & Global model parameters \\ \hline
\rowcolor[HTML]{D3D3D3}$\Upsilon_{u}(t)$ & Collected dataset of FLU $u$ & { $G^{(k)}_u(t)$} & { Rate of change in dataset size for FLU $u$} \\ \hline
{ $T^{\mathsf{train},(k)}_u$} & { Local training time window of FLU $u$} & { $T^{\mathsf{idle},(k)}_u$} & { Idle time of FLU $u$} \\ \hline
$\Upsilon(t)$ & Cumulative dataset of all FLUs & $\Upsilon^{\mathsf{s}}(t)$ & Cumulative dataset of recruited FLUs \\ \hline
\rowcolor[HTML]{D3D3D3}$\mathfrak{L}_{u}^{(k)}(\bm{\omega})$ & Local loss of FLU $u$ & $\mathfrak{L}^{(k)}(\bm{\omega})$ & Global loss \\ \hline
$\mathfrak{D}_u(t)$ & Model drift value of FLU $u$ & {$\widetilde{\tau}_{b}^{\downarrow,{(k)}}$} & { Completion time of model broadcasting at O-RU $b$} \\ \hline
\rowcolor[HTML]{D3D3D3}$\ell^{(k)}_{u}$ & Number of SGD iterations of FLU $u$ & ${B}_{u}(\widetilde{\tau}_{b}^{\downarrow,{(k)}})$ & Mini-batch size of FLU $u$ \\ \hline
$\eta_{_k}$ & Step size at global round $k$ & $\widetilde{\nabla {\mathfrak{L}}}_{u}^{(k)}$ & Cumulative gradient of FLU $u$ \\ \hline
\rowcolor[HTML]{D3D3D3}$\widetilde{\nabla {\mathfrak{L}}}^{(k)}$ & Global gradient vector & {$\beta^{\downarrow}_{b}\hspace{-0.3mm}(t_{x})$} & { Broadcast scheduling decision variable of O-RU $b$} \\ \hline
{$\overline{\beta}^{\uparrow}_u(t_x)$} & { Uplink scheduling decision variable of DPU $u$} & {$\beta^{\uparrow}_u(t_{x})$} &{ Uplink scheduling decision variable of CHU $u$} \\ \hline
\rowcolor[HTML]{D3D3D3}{$\overline{\varrho}_{u,u',r}(t_x)$} &{ Allocation decision of assigning unlicensed PRB $r$ to DPU $u$ to transfer a fraction of its LM to CHU $u'$} & {$\varrho_{u,r}\hspace{-0.3mm}(t_{x})$} & {Allocation decision to assign licensed PRB $r$ to CHU $u$ to transmit a fraction of its LM to its corresponding O-RU} \\ \hline
{$\rho^{\downarrow}_{b,r}\hspace{-0.5mm}(t_{x})$} & {Power allocation decision of PRB $r$ at O-RU $b$} & {$\rho^{\uparrow}_{u,r}\hspace{-0.3mm}(t_{x})$} & {Power allocation decision of PRB $r$  at CHU $u$} \\ \hline
\rowcolor[HTML]{D3D3D3}{$\overline{\rho}^{\uparrow}_{u,r}(t_x)$} &{ Power allocation decision of PRB $r$ at DPU $u$} & {$\varphi_{b,r}\hspace{-0.3mm}(t_{x})$} & {Dispatching decision of O-RU $b$ to broadcast a fraction of the GM to FLUs over PRB $r$ } \\ \hline
{$\overline{\psi}_{u,u',r}(t_x)$} & { Dispersion decision of DPU $u$ to transmit a fraction of its LM to CHU $u'$ over PRB $r$} & {$\psi_{u,r}\hspace{-0.3mm}(t_{x})$} & { Dispatching decision of CHU $u$ to transmit a fraction of its LM to its corresponding O-RU over PRB $r$} \\ \hline
{ $\mathfrak{R}^{\downarrow}_{b,u,r}\hspace{-0.5mm}(t_{x})$} & { Maximum downlink transmission rate from O-RU $b$ to FLU $u$ over PRB $r$} & { $\mathfrak{R}^{\downarrow}_{b,r}\hspace{-0.5mm}(t_{x})$} & { Broadcast transmission rate of O-RU $b$ over PRB $r$} \\ \hline
\rowcolor[HTML]{D3D3D3} { $\mathfrak{R}^{\uparrow}_{u,r}\hspace{-0.3mm}(t_{x})$} & { Uplink transmission rate from CHU $u$ to its corresponding O-RU over PRB $r$} & { $\overline{\mathfrak{R}}^{\uparrow}_{u,u',r}(\hspace{-0.5mm}t_x\hspace{-0.5mm})$} & { Uplink transmission rate from DPU $u$ to CHU $u'$ over PRB $r$} \\ \hline
{ $|\xi_{b,u}\hspace{-0.3mm}(t_{x})|^2$} & { Channel gain between FLU $u$ and O-RU $b$}& { $|\xi_{u,u'}\hspace{-0.3mm}(t_{x})|^2$} & { Channel gain between DPU $u$ and CHU $u'$} \\ \hline
\rowcolor[HTML]{D3D3D3}{ $\alpha_{\bm{\omega}} M$} & { Global model size} & { $E_{b}^{\downarrow,{(k)}}$} & { Energy consumption of model broadcasting at O-RU $b$} \\ \hline
{ $E_{u}^{\mathsf{LC},{(k)}}$} & { Energy consumption of local training at FLU $u$} & { $\overline{E}_{u}^{\uparrow,(k)}$} & { Energy consumption of model dispersion at DPU $u$} \\ \hline
\rowcolor[HTML]{D3D3D3}{ $E_{u}^{\uparrow,(k)}$} & { Energy consumption of model dispatching at CHU $u$} & { $E^{\mathsf{max}}_{u}$} & {Maximum battery capacity of FLU $u$} \\ \hline
\end{tabular}%
}
\end{table}

\begin{table}[h]
\centering
\caption{ {  Acronyms used in the paper.}}
\small
\label{table:acronyms}
{  
\begin{tabular}{|c|c||c|c|}
\hline
\rowcolor[HTML]{848884} 
\textbf{Acronym} & \textbf{Description} & \textbf{Acronym} & \textbf{Description}  \\ 
\hline
\hline
\rowcolor[HTML]{D3D3D3} FL &  Federated learning & {\tt DCLM} & Dynamic cooperative FL with dedicated MAC schedulers  \\ \hline
ML &   Machine learning &  O-RAN &   Open radio access network\\ \hline
\rowcolor[HTML]{D3D3D3} D2D &  device-to-device &  SGD &  stochastic gradient descent\\ \hline
LM &    Local model &  GM &  Global model\\ \hline
\rowcolor[HTML]{D3D3D3} BS &   Base station &  O-RU &  Radio unit\\ \hline
O-DU &  Distributed unit &  O-CU & Central unit\\ \hline
\rowcolor[HTML]{D3D3D3} FLU &    Federated learning user&  CHU&  Communication head FLU \\ \hline
DPU&  Deprived FLU &  DCC &  Dispersed cooperative communication\\ \hline
\rowcolor[HTML]{D3D3D3} FGTI &  Fine granular time instant & PRB  &  Physical resource block  \\ \hline
CT & Communication task &  Non-RT RIC &  Non real-time RAN intelligent controller \\ \hline
\rowcolor[HTML]{D3D3D3} Near-RT RIC  &  Near real-time RAN intelligent controller &  VES &  Virtual event streaming \\ \hline
FL-VNO &  FL virtual network operator &  GV & Gradient vector \\ \hline
\rowcolor[HTML]{D3D3D3} SD &  Scheduling decision &  SDC & Scheduling decision constraint \\ \hline
CNR&  Continuous non-convex representation &  R2F & O-RU to FLU\\ \hline
\rowcolor[HTML]{D3D3D3} C2R&  CHU to O-RU &  D2C & DPU to CHU \\ \hline
SINR &  Signal-to-interference-plus-noise-ratio &  SP & Signomial programming \\ \hline
\rowcolor[HTML]{D3D3D3} GP &  Geometric programming &  KKT & Karush–Kuhn–Tucker \\ \hline
\end{tabular}%
}
\end{table}
\newpage
\section{Future/Open Research Directions}\label{app:future_research}
\begin{remark}[Multiple FL Services and Concurrent Task/Service Processing] \label{remark:non_FLS}
    Our model can be easily extended to cover coexistence of multiple FLSs and other types of services, such as URLLC and mMTC, studying which is left as future work.
\end{remark}
\begin{remark}[FL Service Competition and Unbalanced Execution]\label{remark:multi_FLS}
    This paper builds the theoretical pillars for an unexplored research direction on modeling the coexistence of multiple FLSs in the system with conflicts of interest, which have to compete with each other to recruit FLUs and acquire network resources (e.g., wireless spectrum) from different O-RUs~\cite{abdisarabshali2023synergies}. This situation may further lead to overcrowdedness in specific O-RUs that contain low-cost (e.g., in terms of recruitment cost) and high-quality (e.g., in terms of the quality of taken images for image-oriented FLSs) FLUs~\cite{abdisarabshali2023synergies}. Further modeling and investigations of such scenarios are left as future work.
\end{remark}
\begin{remark}[Stochastic System Analysis]\label{remark:stochastic1}
    If either users' locations or their computation capabilities are stochastic, the values of $\Psi(\mathscr{D}_{u}^{\downarrow,(k)})$, $\Psi(\mathscr{D}_{u}^{\ndownarrow,(k)})$, $\Psi(\overline{\mathscr{D}}_{u}^{\uparrow,(k)})$, $\Psi(\overline{\mathscr{D}}_{u}^{\nuparrow,(k)})$, $\Psi(\mathscr{D}_{u}^{\uparrow,(k)})$, and $\Psi(\mathscr{D}_{u}^{\nuparrow,(k)})$ will become stochastic, which calls for stochastic dynamic optimal control strategies studying of which are left as future work. 
\end{remark}
\begin{remark}[Non-Deterministic MAC Scheduler Design]\label{remark:stochastic2}
   As mentioned in Sec \ref{sec:system_dynamic}, for deterministic systems, the resource requirements (i.e., PRBs and transmit power requirements) of FLS slices are precisely known as the channel conditions change at known discrete-time instants, which can be accomplished through the channel quality estimation procedure commonly used in contemporary cellular systems. This enables both provisioning wireless resources (i.e., RAN slice scaling) to FLS slices and allocating these resources by MAC schedulers of each slice to their users according to the users' exact requirements at each FGTI. However, if the system is stochastic (e.g., abrupt environment changes, stochastic arrival/departure of FLUs, and stochastic availability of FLUs' computation resources) the future state of the system becomes uncertain, and the resource requirements of the slices cannot be accurately determined. In such cases, resource allocation by MAC schedulers and slice scaling must be performed at different time scales. This necessitates a coarse granular time scale for slice scaling operation comprising multiple FGTIs. This situation requires stochastic analysis and ML-assisted methods to predict the future state of the system at predetermined time instants, studying which is a promising research direction and left for future work. 
\end{remark}
\begin{remark}[Considering More Comprehensive O-RAN Features in the Design]\label{remark:NFV}
In general, an O-RAN slice consists of (i) a set of PRBs, (ii) transmit power, and (iii) a set of virtual network functions (VNFs) executed on computation resources of O-DUs and O-CUs to process the traffic flow of PRBs. In this paper, we focus on a scenario in which an FLS slice consists of unlicensed/licensed PRBs and transmit power. Considering scenarios with VNFs are left for future work. In particular, designing of VNFs will require mathematical modeling of the traffic flow over the fronthaul and midhaul of the network.
\end{remark}

\begin{remark}[Convergence Acceleration under D2D Data Transmission]\label{remark:D2D}
In situations where privacy concerns are not paramount for data collected at the FLUs (such as visual data of traffic signs gathered by smart cars), there is potential for data transmission between FLUs over O-RAN via D2D communications. In such cases, optimizing data transmission is essential to minimize data heterogeneity across FLUs. This necessitates the {\tt DCML} MAC scheduler to handle data transmission while considering system dynamics alongside ML model parameter transmission, requiring dynamic control decision variables. Exploring these scenarios is a subject for future research.
\end{remark}

\begin{remark}[Convergence Acceleration under D2D Model Consensus]\label{remark:D2S}
{ Building on Remark~\ref{remark:D2D}, we can also consider the possibility of model censuses among DPUs through D2D communications. In this framework, DPUs can create localized networks and conduct local consensus (e.g., through distributed averaging mechanisms used in decentralized FL literature) to de-bias their local models, and then only a small fraction of them engage in model transmission to CHUs. Further, CHUs can also create such localized networks and engage in D2D-assisted model consensus, after which only a small fraction of CHUs will engage in uplink transmissions to the O-RU. Detailed modeling and formulation of these aspects are deferred to future work.}
\end{remark}

\begin{remark}[Asynchronous Model Training]\label{remark:Async}
In this study, our focus is on synchronized FL training, the most widely recognized structure for ML model training in FL. In this setup, the server waits for models from a large subset of FLUs before aggregating them to form a global model. However, one can envision a more flexible architecture where the server updates the global model more rapidly. For instance, upon receiving each model from every FLU, it could immediately update the global model. Such update rules shift the focus to asynchronous FL over O-RAN requiring redesigning dedicated FL orchestrator system over O-RAN slices, a topic left for future investigation.
\end{remark}


\begin{remark}[Distributed Multi-Modal Learning]\label{remark:multi-modal}
   In this paper, similar to the prevailing approach in FL research, our focus was on scenarios where FLUs possess data in a single modality. However, FLUs can gather data across multiple modalities (e.g., video, image, and audio) in dynamic environments, necessitating appropriate scheduling strategies and data processing load balancing. Exploring such scenarios will extend the design to the domain of multi-modal FL over O-RAN, a topic that remains largely unexplored. Notably, the data modality collected by users can influence recruitment decisions and network operations. Investigating these scenarios presents intriguing avenues for future research.
\end{remark}
\newpage
\section{Network Characteristics Utilized for the Simulations}\label{app:network_settings}
We use the following network characteristics for our simulations, parameter of which are summarized in Table. \ref{tab:sim_network_params}.\\
{ \textbf{Stochastic channel model and user mobility:}
The channel between FLU $u$ and O-RU $b$ during FGTI $t_x$, referred to by $\xi_{b,u}\hspace{-0.3mm}(t_{x})$, used in \eqref{eq:SINR_broadcast} and \eqref{eq:SINR_uplink}, is assumed to experience the following evolution over time

\begin{equation}\label{eq:channel}
    \xi_{b,u}\hspace{-0.3mm}(t_{x}) = \sqrt{\beta_{b,u}(t_x)}\times h_{b,u}(t_x),
\end{equation}
where $h_{b,u}(t_x) = \mu_u h_{b,u}(t_{x-1}) + \sqrt{(1-\mu_u^2)} \times n_{b,u}(t_x)$ represents the small-scale fading factor modeled as a first-order time-varying Gauss-Markov process~\cite{10330597}. The perturbation terms {$n_{b,u}(t_x)$} over various FGTIs $t_x$ are zero-mean, unit-variance complex Gaussian random variables that are i.i.d. over $u$, $b$, and $t_x$, i.e., they follow $CN(0,1)$ distribution. 
At time $t_0$, we assume $h_{b,u}(0) \sim CN(0,1)$ and independent of $n_{b,u}(t_1)$. The correlation coefficient $\mu_u$ for FLU $u$ is given by $\mu_u = J_0\left(2\pi\frac{v_u}{c}f_cT_e\right)$, where $J_0(\cdot)$ is the Bessel function of the first kind of order zero, $v_k$ is the velocity of FLU $u$, $f_c$ is the carrier frequency, $T_e=t_x-t_{x-1}$ is the channel instantiation interval, and $c=3\times 10^8$ m/s is the speed of light. The magnitude of $h_{b,u}(t_x)$ is designed to follow a Rayleigh distribution, effectively capturing the characteristics of a dense scattering wireless environment. The term $\beta_{b,u}(t_x)$ in \eqref{eq:channel} represents the large-scale fading factor, which is inversely proportional to the distance between FLU $u$ and O-RU $b$ at time $t_x$. A similar formulation is applied to compute the channel gain between DPU $u$ and CHU $u'$ at FGTI $t_{x}$, denoted by $|\xi_{u,u'}(t_{x})|^2$ and used in \eqref{eq:SINR_up_DPU}.

\noindent\textbf{FL-related parameters:}
The growth and decay rates of FLU dataset size (i.e., $C_{u}^{\uparrow,(k)}$ and $C_{u}^{\downarrow,(k)}$) are randomly selected from the range $[3,5]$. Local data dissimilarity $\Theta$ and cross-FLU data heterogeneity captured via $\mathfrak{X}_1$ and $\mathfrak{X}_2$ are $3$, $1$, and $10^{-3}$, respectively. The effective chipset capacitance of FLUs, $\frac{\alpha_{u}}{2}$, is chosen to be $\frac{10^{-27}}{2}$~\cite{8737464}. We choose the step size coefficient $\alpha$ in \eqref{suf:main:eta} to be $0.1$. The number of CPU cycles needed to process one data sample at FLUs, $a_{u}$, is set to $4000$. 
The importance weights for the ML convergence bound $c_1$ and energy consumption $c_2$ are set to $10$ and $10^{-2}$, respectively. We use two convolutional neural networks (CNNs) with sizes of 1,200 kilobits for MNIST and 2,000 kilobits for CIFAR10, which translate to the representation of 37,500 and 62,500 model parameters, respectively, under 32 bit quantization.

\noindent\textbf{Network and FLU parameters:}
We consider an O-RAN comprising five O-RUs, one O-DU, one O-CU, one near-RT RIC, and one non-RT RIC, where all O-RUs are connected to the O-DU, which is connected to the O-CU. O-RUs utilize the same set of 10 licensed and 10 unlicensed PRBs. For unlicensed and licensed PRBs, we consider numerologies $\gamma_1=1$ and $\gamma_2=0$, resulting in a bandwidth allocation of $360$~KHz and $180$~KHz, respectively~\cite{7744816}. The power spectral density of white Gaussian noise $N_0$ is $-174$dBm/Hz.  Maximum transmit powers of O-RUs $P^{\mathsf{max}}_{b}$, maximum transmit powers of FLUs $P^{\mathsf{max}}_{u}$, and CPU frequencies of FLUs $f^{(k)}_{u}$ are uniformly drawn from [$3$,~$4$]~W, [$500$,~$800$]~mW, and [$1.5$,~$2$]~GHz, respectively~\cite{10330597}. 
The available battery capacities of FLUs allocated/dedicated for FL training, denoted as $E^{\mathsf{max}}_{u}$ in~\eqref{const:energy:max_DPU}, are randomly chosen from the range $[0.5, 1]$ KJ. We consider the five O-RUs to be located at coordinates [(0, 0), (200, 200), (400, 400), (600, 600), (800, 800)], measured in meters~\cite{10330597}, and consider the presence of $30$ FLUs in the system. Centered around each O-RU, six FLUs are initially distributed randomly within a circular area with a radius of 50m. We set the maximum allowable latency per global round ($T^{\mathsf{max}}$ in~\eqref{const:latency_max}) to $2$ seconds. }

\begin{table*}[h]
\vspace{1.5mm}
\caption{Network characteristics used for ablation study. The experiments all use the network values herein, unless indicated otherwise.}
\vspace{-1.5mm}
\centering
\label{tab:sim_network_params} 
{
\begin{tabular}{|c c|c c|c c|c c|} 
\hline
\textbf{Param} & \textbf{Value} & \textbf{Param} & \textbf{Value} & \textbf{Param} & \textbf{Value}& \textbf{Param} & \textbf{Value} \\
\hline
$\alpha$ & $0.1$ & $P^{\mathsf{max}}_{b}$ & $[3,4]$W & $C_{u}^{\uparrow,(k)}$ & $[3,5]$ & $c_1$ & $10$  \\
$f^{(k)}_{u}$ & $[1.5, 2.0]$GHz & $P^{\mathsf{max}}_{u}$ & $[500,800]$mW & $C_{u}^{\downarrow,(k)}$ & $[3,5]$& $c_2$ & $10^{-2}$\\
$a_{u}$ & $3960$ & $N_0$ & $-174$dBm/Hz & $\mathfrak{X}_1$ & $1$ & $E^{\mathsf{max}}_{u}$ & [0.5~ 1]~ KJ  \\
$\alpha_u$& $10^{-27}$ & $\gamma_1$ & $1$ & $\mathfrak{X}_2$ & $10^{-3}$ & $T^{\mathsf{max}}$ & 2 Sec \\
$\gamma_2$ & $0$ & $\Theta$& $3$&   &  & &  \\
\hline
\end{tabular}
}
\vspace{-2mm}
\end{table*}

\newpage
\section{The Protocol Flow of {\tt DCLM}}\label{app:protocol_flow}

\begin{figure*}[!h]
    \centering 
    \includegraphics[width=\linewidth,trim=80 5 5 5,clip]{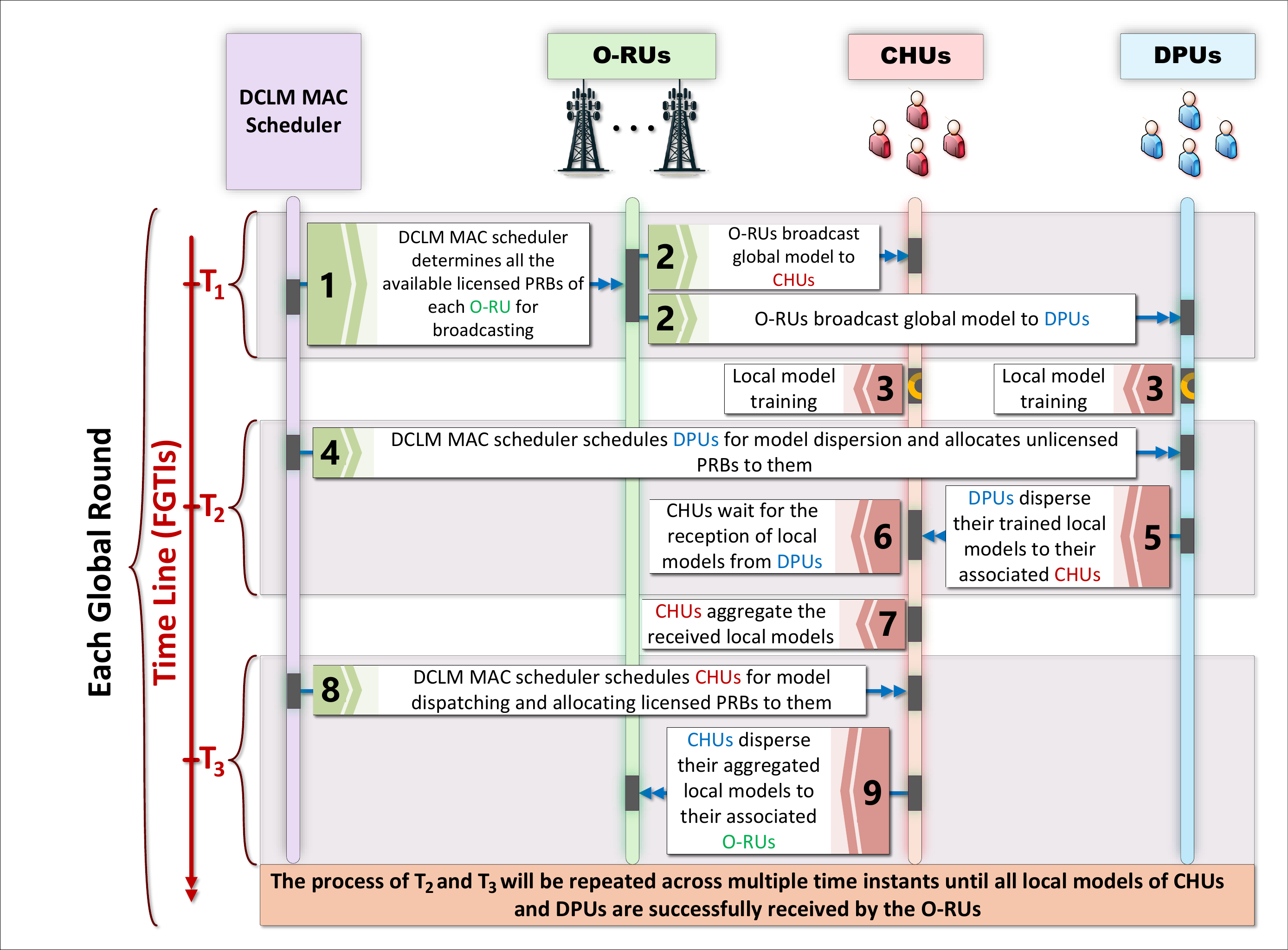}
    \caption{{ The protocol flow of the {\tt DCLM} MAC scheduler, showcasing the interactions between {\tt DCLM} MAC scheduler, O-RUs, CHUs, and DPUs.}}
    \label{fig:protocol_flow}
\end{figure*}

\newpage
\section{Fairness of Recruiting FLUs as CHU}\label{app:fairness}

\begin{figure}[!h]
\centering
\noindent\includegraphics[width=7cm]{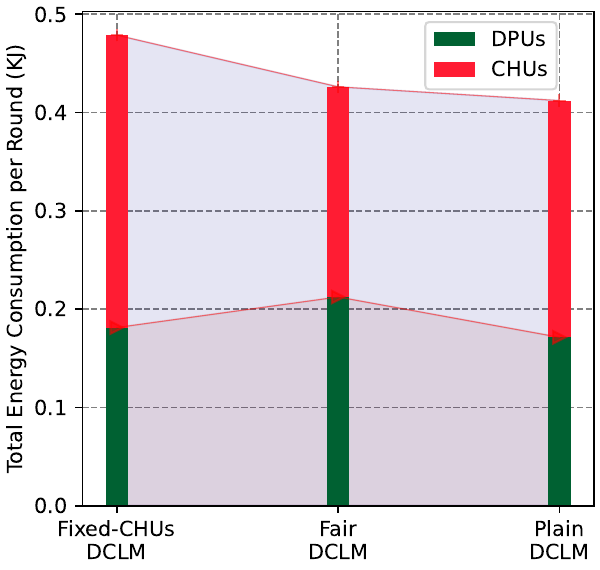} 
\caption{ {\tt Plain DCLM} achieves the lowest overall energy consumption compared to both {\tt fixed-CHUs DCLM} and {\tt fair DCLM}. Additionally, {\tt fixed-CHUs DCLM} results in the highest total energy consumption and the most unfair energy consumption of CHUs vs DPUs due to its inability to adapt the selection of FLUs as CHUs to the evolving network conditions. An interesting insight is that while {\tt fair DCLM} achieves the most balanced energy consumption between CHUs and DPUs (i.e., the fairest energy consumption among the FLUs), it does not lead to the most energy-efficient operation of the overall system. In other words, the superior energy efficiency of {\tt Plain DCLM} is achieved by slightly compromising the fairness of energy consumption between DPUs and CHUs.}
\label{fig:ext:fairness2}
\vspace{-.1mm}
\end{figure}
{ In this section, we examine the fairness of {\tt DCLM} by comparing {\tt DCLM(7)}, here referred to as {\tt plain DCLM}, with two modified versions. The first is {\tt Fixed-CHUs DCLM}, where the optimization is constrained to select the same set of FLUs as CHUs across all global rounds. The second is {\tt fair DCLM}, which incorporates a fairness term -- defined as the modified standard deviation of the energy consumption of FLUs -- into the objective function of the optimization problem $\bm{\mathcal{P}}$. Mathematically, we add the following term to the objective function of $\bm{\mathcal{P}}$:
\begin{equation}\label{eq:fairness}
   \mathscr{F}= \frac{\sum_{b \in \Omega,u{\in}\mathcal{U}_{b}} \left((\sum_{k=0}^{K-1}\lambda_{u}^{(k)}E_{u}^{\uparrow,(k)}+\overline{\lambda}_{u}^{(k)}\overline{E}_{u}^{\uparrow,(k)}) - M_{E}\right)^2}{\sum_{b \in \Omega}|\mathcal{U}_{b}|},
\end{equation}
where $M_{E}$ is the average energy consumption of FLUs, given by: 
\begin{equation}
   M_{E}= \frac{ \sum_{b \in \Omega,u{\in}\mathcal{U}_{b}}  (\sum_{k=0}^{K-1}\lambda_{u}^{(k)}E_{u}^{\uparrow,(k)}+\overline{\lambda}_{u}^{(k)}\overline{E}_{u}^{\uparrow,(k)})}{\sum_{b \in \Omega}|\mathcal{U}_{b}|}.
\end{equation}
Here, $\overline{E}_{u}^{\uparrow,(k)}$ and $E_{u}^{\uparrow,(k)}$ are computed in \eqref{eq:EN_DPU_UP} and \eqref{eq:EN_CHU_UP}, respectively.
Note that for each FLU $u$, we either have $\lambda_{u}^{(k)}=1$ or $\overline{\lambda}_{u}^{(k)}=1$ and not both (since the FLU is either a CHU or a DPU and not both). Subsequently, inspecting the numerator of $\mathscr{F}$,  minimization of this term results in CHUs and DPUs having the same level of energy consumption (i.e., only when the overall energy consumption of CHUs is equivalent to the overall energy consumption of DPUs, which makes them both equivalent to the average energy consumption of FLUs, term $\mathscr{F}$ would be equal to zero).

We would like to emphasize that the default version of our optimization, i.e., {\tt plain DCLM}, naturally results in changing/flipping the CHUs over time since the goal of the optimization is to optimize the selection of CHUs to minimize the overall energy consumption of the system, while taking into account the time-varying channels of FLUs.

Fig.~\ref{fig:ext:fairness2} illustrates that {\tt plain DCLM} achieves the lowest overall energy consumption compared to both {\tt fixed-CHUs DCLM} and {\tt fair DCLM}, while slightly compromising the fairness of energy consumption between DPUs and CHUs. In this context, fairness implies that the average energy consumption per round for both DPUs and CHUs is nearly equal (i.e., although CHUs conduct uplink transmissions to the O-RUs, which are often more energy demanding than D2D communications conducted by DPUs, the overall energy consumption of CHUs is kept at the same level of DPUs). 

The reduction in the overall energy consumption of FLUs achieved via {\tt plain DCLM} is primarily because the channel quality of FLUs fluctuates over FGTIs, and {\tt plain DCLM} minimizes the overall energy consumption of FLUs -- dictated by terms (c) in the objective function of $\bm{\mathcal{P}}$ -- by dynamically selecting FLUs as CHUs with favorable channel quality to both DPUs and O-RUs in each global round (see Appendix~C for details on the dynamic channel model used in this paper).

Note that the total energy consumption of FLUs in {\tt fair DCLM} is slightly higher than in {\tt plain DCLM}, as the optimization problem favors balancing the total energy consumption of FLUs with the standard deviation of their energy consumption (i.e., it favors having equal amounts of energy consumption over DPUs and CHUs although this may result in an increase in the overall energy consumption). Interestingly, this highlights that inducing fairness in terms of energy consumption across the FLUs does not always align with the optimal energy-efficient solution for the overall system. Additionally, the figure shows that {\tt fixed-CHUs DCLM} results in the highest total energy consumption and the most unfair strategy for selecting CHUs in terms of energy consumption, due to its inability to adapt the selection of FLUs as CHUs to the evolving network conditions.}

\newpage
\section{The Evolution of RAN: From Distributed RAN to O-RAN}\label{apx:ran_evolution}
{ RAN technologies have evolved significantly since their invention as traditional RANs to accommodate the ever-changing wireless network demands (Fig. \ref{fig:ORANEvolution} depicts the evolution of RAN technologies) \cite{chih14}. In particular, the 3GPP 4G LTE RAN (Long-Term Evolution Radio Access Network), e.g., distributed RAN (D-RAN) shown in the left part of Fig.~\ref{fig:ORANEvolution}, is a traditional, centralized network architecture that relies on proprietary hardware and software from specific vendors, focusing on enhanced mobile broadband (eMBB) services. It employs a tightly integrated infrastructure with base stations (eNBs) communicating directly with the core network. Specifically, in D-RAN, the functionalities of the  3GPP protocol stack -- that is, (i) physical (PHY) low in the remote radio head (RRH), and (ii) PHY-high, medium access control (MAC), radio link control (RLC), packet data convergence protocol (PCDP), and radio resource control (RRC) in the base band unit (BBU) -- are placed at each base station near the users. While effective for 4G networks, this architecture is costly (i.e., in terms of management and energy consumption costs) and lacks the scalability required for 5G-and-beyond networks \cite{chih14}. Compared with D-RAN, cloud RAN (C-RAN), middle part of Fig.~\ref{fig:ORANEvolution}, exploits cloud computing services and realizes the sharing of processing resources by pooling BBUs into centralized units named BBU pools, which can effectively reduce energy consumption and management costs~\cite{chih14}. However, C-RAN requires a tremendous amount of fronthaul capacity from the BBU pools to the RRHs~\cite{chih14}.
\begin{figure*}[!t]
    \centering 
    \includegraphics[width=\linewidth,trim=5 5 5 5,clip]{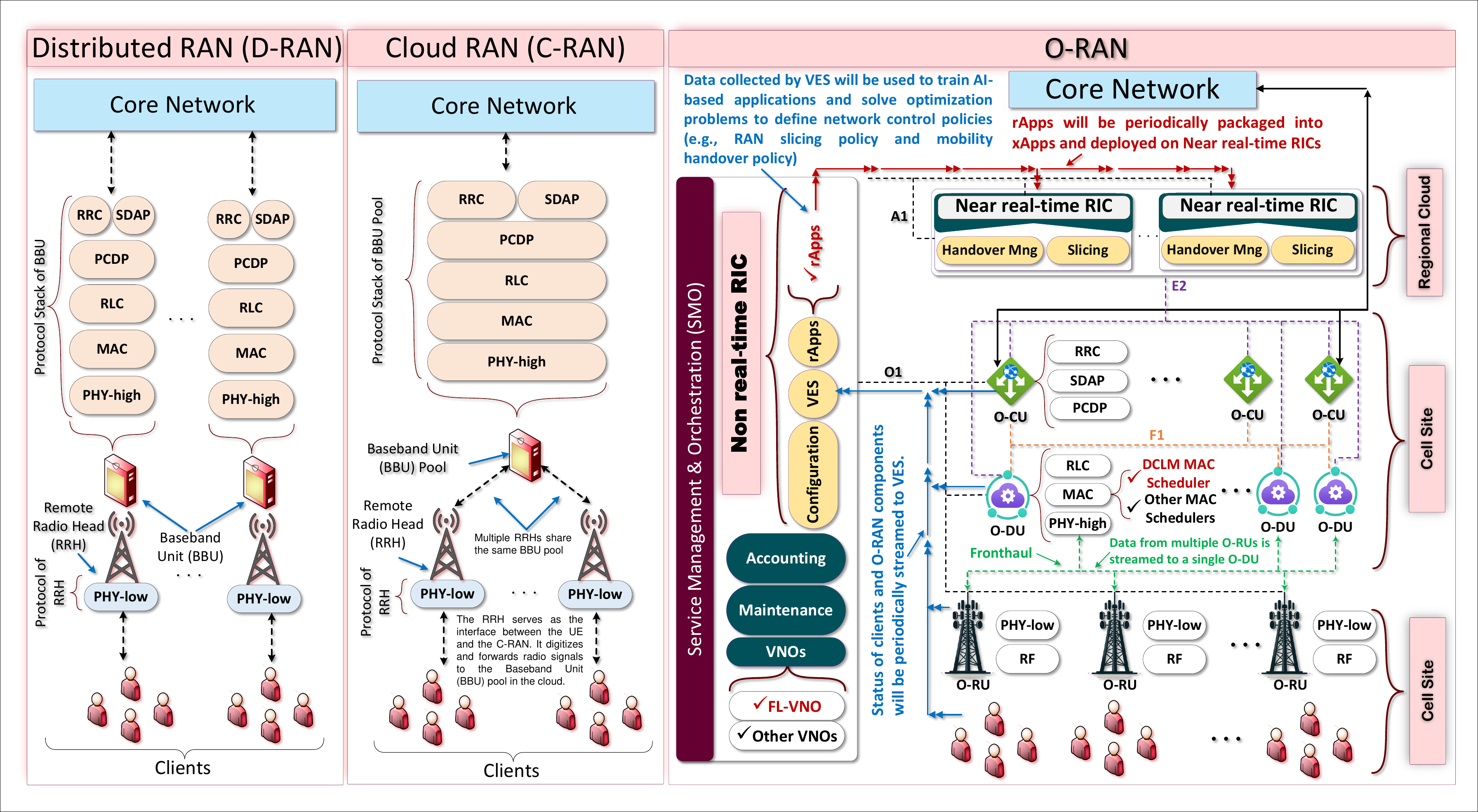}
    \caption{\footnotesize The evolution of RAN: Distributed RAN (D-RAN) $\rightarrow$ Cloud RAN (C-RAN) $\rightarrow$ Open RAN (O-RAN)}
    \label{fig:ORANEvolution}
\end{figure*}

In contrast to the aforementioned two traditional RAN architectures, in O-RAN, right part of Fig.~\ref{fig:ORANEvolution}, the elements of the base station are disaggregated into three distinct components: (i) radio unit (O-RU), responsible for the PHY-low and radio frequency (RF) functionalities, (ii) distributed unit (O-DU), responsible for the PHY-high, MAC, and RLC functionalities, and (iii)  centralized unit (O-CU) responsible for the remainder of the 3GPP stack (i.e., RRC, service data adaption protocol (SDAP), and PCDP) \cite{9846950}. O-RAN also introduces the concept of non-real-time (non-RT) RAN intelligent controllers (RICs) and near-RT RICs which orchestrate RAN operations (e.g., resource allocation and interference management) via software applications called rApps and xApps~\cite{9846950}.

The O-RAN components (i.e., O-RUs, O-DUs, O-CUs, and RICs) interact with each other via standard open interfaces, facilitating interoperability between network elements from different manufacturers. Specifically, each O-CU provides RRC, SDAP, and PCDP functionalities for multiple O-DUs through standard F1 interface. Each O-DU is also associated with multiple O-RUs via standard open fronthaul interface and provides PHY-high, MAC, and RLC functionalities for them. Correspondingly, each O-RU delivers network connection for users of various network services (i.e., FL, mobile broadband, ultra-reliable low-latency communications, and massive machine-type communications). The interaction between RICs and other elements of O-RAN is explained below.

\textbf{O-RAN intelligent orchestration system.} In O-RAN, data (e.g., key performance measurements (KPMs)~\cite{9846950}), generated by the RAN nodes, i.e., O-CUs, O-DUs, and O-RUs, stream periodically via the O1 interface to the virtual event streaming (VES). The data collected by VES is utilized by rApps in non-RT RICs to establish policies for wireless network management. These policies are packaged into xApps and instantiated on demand on the near-RT RICs, utilized to fine-tune the behavior of RAN via performing near-RT RAN operations, such as mobility management and RAN slicing. In particular,  RAN slicing refers to partitioning the physical RAN resources (e.g., radio resources), governed by a network operator, into various isolated logical slices, each leased to a virtual network operator (VNO). VNOs tailor their RAN slices to meet their service requirements by adopting appropriate RAN control strategies (e.g., MAC scheduling and mobility management). Here,  \textit{a MAC scheduler is an essential part of O-DU}, where real-time radio resources allocation is performed according to the current network conditions \cite{9846950}. MAC scheduler also enforces the necessary QoS requirements for the users' various services.

In traditional closed RANs (e.g., 3GPP LTE D-RAN), network operator cannot form/alter the network functionalities (e.g., a rigid black box MAC scheduler in 4G is utilized for all the users' services). To address this limitation, O-RAN offers a flexible infrastructure for designing dedicated functionalities for each virtual RAN slice, a feature referred to as \textit{programmability}. In {\tt DCLM}, we take advantage of the programmability feature of O-RAN to design a dedicated MAC scheduler for dynamic resource allocation in FL. This scheduler, deployed at the MAC layer (as labeled {\tt DCLM} MAC scheduler in the right part of Fig.~\ref{fig:ORANEvolution}), schedules CHUs and DPUs across different fine-granular time instants (FGTIs). By dynamically allocating licensed and unlicensed PRBs to CHUs and DPUs at different FGTIs, the {\tt DCLM} MAC scheduler effectively reduces interference across FLUs, enhancing resource utilization and reducing energy consumption.
It is worth mentioning that D2D communications have been proposed as a part of O-RAN~\cite{linsalata2024addressing}; however, this work is the first to take a unique advantage of D2D communications in O-RAN for FL. In particular, in this work, D2D communications is used between the DPUs and CHUs to save the licensed spectrum, mitigate the interference, and reduce the energy consumption and latency of model aggregations in FL over O-RAN.
}

\section{Data Heterogeneity Visualization}\label{app:datahetrogenity}

{ Fig.~\ref{fig:ext:het} illustrates how varying degrees of data heterogeneity, modeled through the Dirichlet distribution parameter $\alpha^{\mathsf{D}}$, influence the class distribution across five O-RUs and their covered $30$ FLUs. The color-coded bars represent the relative class distribution per O-RU, with the vertical axis indicating the number of samples and the horizontal axis covering normalized class frequency. 
The three sub-figures correspond to different $\alpha^{\mathsf{D}}$ values: (i) Left subfigure (Low Heterogeneity, $\alpha^{\mathsf{D}} = 5$): Each O-RU and its FLs receive a relatively balanced mix of all ten classes, indicating a nearly i.i.d. (independent and identically distributed) setting. (ii) Middle subfigure (Moderate Heterogeneity, $\alpha^{\mathsf{D}} = 0.5$): The data becomes more skewed across O-RUs and their FLUs. This setup reflects a moderately non-IID scenario. (iii) Right subfigure (Extreme Heterogeneity, $\alpha^{\mathsf{D}} = 0.05$): Most O-RUs and FLUs receive data samples from only a small subset of the classes, showcasing an extremely non-IID distribution.}

\begin{figure}[!b]
\centering
\noindent\includegraphics[width=.79\textwidth]{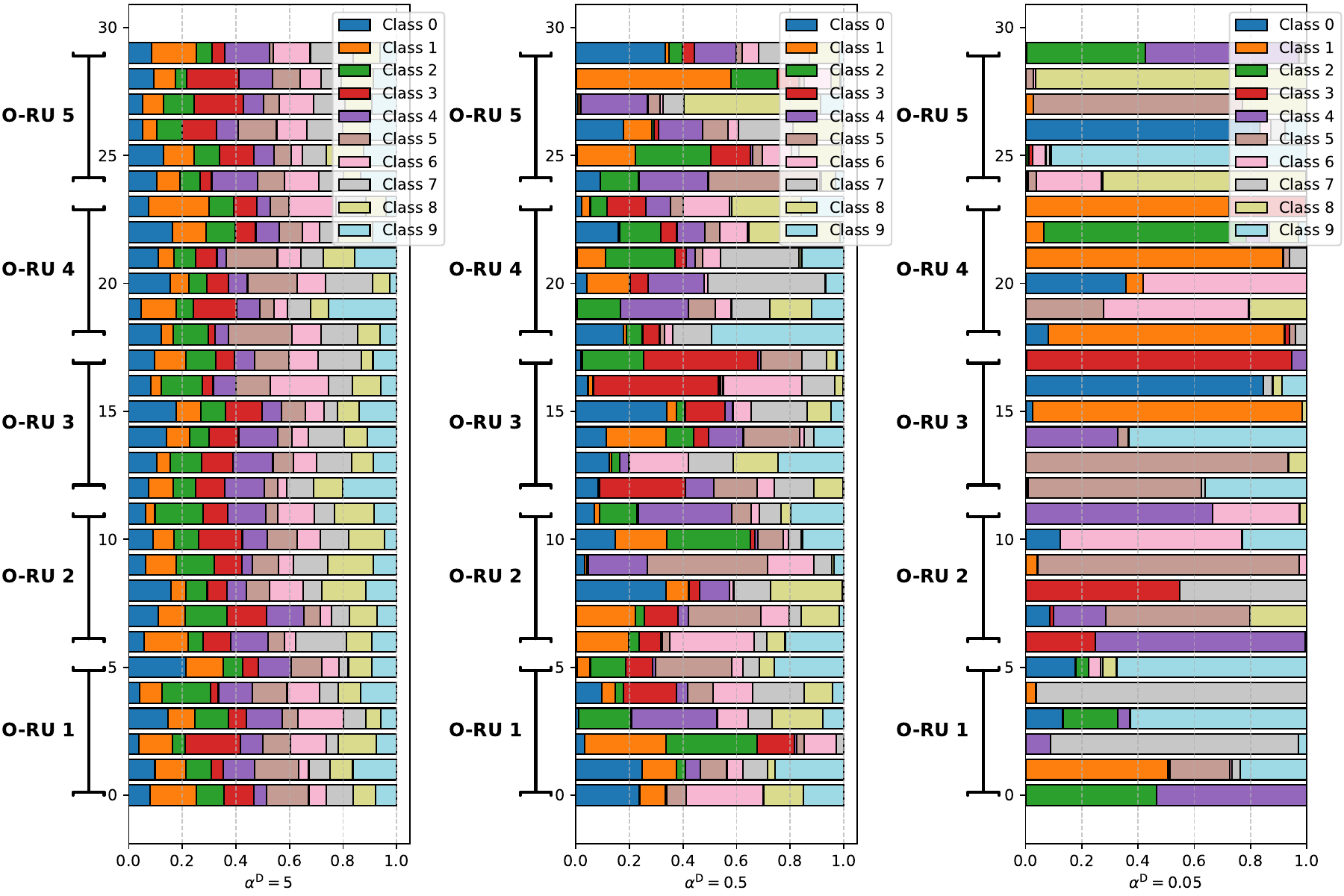} 
\caption{ The selected values of $\alpha^{\mathsf{D}}$ represent three levels of data heterogeneity: low ($\alpha^{\mathsf{D}}{=}5$), moderate ($\alpha^{\mathsf{D}}{=}0.5$), and extreme ($\alpha^{\mathsf{D}}{=}0.05$). Smaller $\alpha^{\mathsf{D}}$ values lead to more non-IID (skewed) distributions. In particular, as shown in the rightmost subfigure, at $\alpha^{\mathsf{D}} = 0.05$, most FLUs receive data from only a few classes, reflecting an extreme non-IID setting.}
\label{fig:ext:het}
\end{figure}
\newpage
\section{Proof of Proposition \ref{propo:broadcast}}\label{app:propo:broadcast}
\noindent Let $\alpha_{\bm{\omega}}M$ refer to the total size of the data that needs to be broadcast through O-RUs, where $\alpha_{\bm{\omega}}$ is the number of bits required to represent one element of the GM vector $\bm{\omega}^{(k)}$ with size $M$. Further, let $M^{\downarrow}_{b}(x)$ be the total size of the transmitted data of GM vector $\bm{\omega}^{(k)}$ by O-RU $b$ until $t_{x}$. Thus, $\alpha_{\bm{\omega}}-M^{\downarrow}_{b}(x)$ is the size of the remaining data at time $t_{x}$ that O-RU $b$ must broadcast to its recruited FLUs. Moreover, assume $\sum_{r \in \mathcal{R}_{b}} \varphi_{b,r}\hspace{-0.3mm}(t_{x})= 1$ (i.e., O-RU $b$ attempts to broadcast all its remaining GM among its recruited FLUs at time $t_{x}$). Accordingly, at time $t_{x}$, the size of data that O-RU $b$ will broadcast to its recruited FLUs over PRB $r$ is $\big(\alpha_{\bm{\omega}} M{-}M^{\downarrow}_{b}(x)\big)\varphi_{b,r}\hspace{-0.3mm}(t_{x})$, the latency of which can be computed as follows:
\begin{equation}\label{eq:broadcast_communication_latency1}
\hspace{-5mm}
\tau_{b,r}^{\downarrow}\hspace{-0.3mm}(t_{x}){=}\frac{\big(\alpha_{\bm{\omega}} M{-}M^{\downarrow}_{b}(x)\big)\varphi_{b,r}\hspace{-0.3mm}(t_{x})}{\big(\mathfrak{R}^{\downarrow}_{b,r}\hspace{-0.3mm}(t_{x}){+}1{-}\beta^{\downarrow}_{b}\hspace{-0.3mm}(t_{x})\big)},
\hspace{-5mm}
\end{equation}
where $1{-}\beta^{\downarrow}_{b}\hspace{-0.3mm}(t_{x})$ is added to the denominator of \eqref{eq:broadcast_communication_latency1} to avoid division by zero. At time $t_{x{+}1}$, MAC schedulers may perform rescheduling of users and reallocation of available resources (i.e., licensed PRBs and transmit power). Therefore, we have to recalculate the latency of broadcasting remaining data over PRBs. To this end, We consider the two potential scenarios that could occur at time $t_{x{+}1}$: 

\begin{enumerate}[label={\textbf{\arabic*})}]
    \item $\tau_{b,r}^{\downarrow}\hspace{-0.3mm}(t_{x})<\left(t_{x{+}1}-t_{x}\right)$: in this situation, the size of data transmitted by O-RU $b$ to its recruited FLUs over PRB $r$ during time window $[t_{x},t_{x{+}1}]$ is given by $\tau_{b,r}^{\downarrow}\hspace{-0.3mm}(t_{x})\mathfrak{R}^{\downarrow}_{b,r}\hspace{-0.3mm}(t_{x})$
    \item $\tau_{b,r}^{\downarrow}\hspace{-0.3mm}(t_{x})\ge\left(t_{x{+}1}-t_{x}\right)$: in this situation, the size of data transmitted by O-RU $b$ to its recruited FLUs over PRB $r$ during time window $[t_{x},t_{x{+}1}]$ is given by $\left(t_{x{+}1}-t_{x}\right)\mathfrak{R}^{\downarrow}_{b,r}\hspace{-0.3mm}(t_{x})$. 
\end{enumerate}
These two cases can be captured compactly through $\min\left\{\tau_{b,r}^{\downarrow}\hspace{-0.3mm}(t_{x}),t_{x{+}1}-t_{x}\right\} \mathfrak{R}^{\downarrow}_{b,r}\hspace{-0.3mm}(t_{x})$. On account of this, the total size of data  transmitted by O-RU $b$ to its recruited FLUs until FGTI $t_{x{+}1}$ is given by
\begin{equation}
\hspace{-5mm}
   M^{\downarrow}_{b}(x{+}1){=}\hspace{-5mm}\sum_{z{=}N^{(k{-}1)}{+}1}^{x{+}1}\sum_{r'
   {\in}\mathcal{R}_{b}}\min\hspace{-1mm}\big\{\tau_{b,r'}^{\downarrow}(t_{z}),\left(t_{z{+}1}{-}t_{z}\right)\big\}\mathfrak{R}^{\downarrow}_{b,r'}(t_{z}).
\hspace{-5mm}
\end{equation} 
Similarly, $M^{\downarrow}_{b}(x)$ is given by
\begin{equation}
\hspace{-5mm}
   M^{\downarrow}_{b}(x){=}\hspace{-5mm}\sum_{z{=}N^{(k{-}1)}{+}1}^{x}\sum_{r'
   {\in}\mathcal{R}_{b}}\min\hspace{-1mm}\big\{\tau_{b,r'}^{\downarrow}(t_{z}),\left(t_{z{+}1}{-}t_{z}\right)\big\}\mathfrak{R}^{\downarrow}_{b,r'}(t_{z}).
\hspace{-5mm}
\end{equation} 
Consequently, the downlink communication latency of O-RU $b$ to broadcast the GM $\bm{\omega}^{(k)}$ to its recruited FLUs can be computed as the cumulative broadcasting latency during FGTIs of global round $k$, which is given by
\begin{equation}\label{eq:TDownload}
\hspace{-5mm}
   \tau_{b}^{\downarrow,{(k)}}=\hspace{-3mm}\sum_{x{=}N^{(k{-}1)}{+}1}^{N^{(k)}{-}1}\max_{r{\in}\mathcal{R}_{b}}\Big\{\min\big\{\tau_{b,r}^{\downarrow}(t_{x}),\left(t_{x{+}1}-t_{x}\right)\big\}\Big\},
\hspace{-5mm}
\end{equation}    
where $\tau_{b,r}^{\downarrow}\hspace{-0.3mm}(t_{x})$ is given in \eqref{eq:broadcast_communication_latency1}, which concludes the proof.

\newpage
\section{Proof of Proposition~\ref{propo:DPU_uplink}}\label{app:propo:DPU_uplink}
\noindent 
The proof is similar to that of Appendix \ref{app:propo:broadcast}. In particular, let $\alpha_{\bm{\omega}}M$ be the size of the data that must be dispersed among CHUs by DPU $u$, where $\alpha_{\bm{\omega}}$ is the number of bits required to represent one element of the gradient vector $\widetilde{\nabla {\mathfrak{L}}}_u^{(k)}$ with size $M$. Further, let $\overline{M}^{\uparrow}_{u}(x)$ be the size of data (i.e., $\alpha_{\bm{\omega}}M$) of DPU $u$ dispersed among its associated CHUs until FGTI $t_{x}$. Thus, $\alpha_{\bm{\omega}}M-\overline{M}^{\uparrow}_{u}(x)$ is the size of remaining data of DPU $u$ at time $t_{x}$ that must be dispersed to its associated CHUs. Moreover, assume $\sum_{u'{\in} \mathcal{U}_{b}}\sum_{r{\in}\overline{\mathcal{R}}_{b}}\overline{\psi}_{u,u',r}(t_x){=}1,~x\in\mathcal{N}^{(k)}$ (i.e., DPU $u$ attempts to disperse all its remaining GPs among its associated CHUs at time $t_{x}$). Accordingly, at FGTI $t_{x}$, the size of data that DPU $u$ will upload to CHU $u'$ over PRB $r$ is $\big(\alpha_{\bm{\omega}}M-\overline{M}^{\uparrow}_{u}(x)\big)\overline{\psi}_{u,u',r}(t_x)$, the latency of which can be computed as follows:
\begin{equation}\label{eq:DPU_communication_latency1}
\overline{\tau}_{u,u',r}^{\uparrow}(t_x){=}\frac{(\alpha_{\bm{\omega}} M{-}\overline{M}^{\uparrow}_{u}(t_x))\overline{\psi}_{u,u',r}(t_x)}{\overline{\mathfrak{R}}^{\uparrow}_{u,u',r}(t_x){+}1{-}\overline{\beta}^{\uparrow}_u(t_x)},
\end{equation}
where $1{-}\overline{\beta}^{\uparrow}_u(t_x)$ is added to the denominator of \eqref{eq:DPU_communication_latency1} to avoid division by zero.
At time $t_{x{+}1}$, MAC schedulers may perform rescheduling of users and reallocation of available resources (i.e., unlicensed PRBs and transmit power). Therefore, we have to recalculate the latencies of dispersing remaining data of DPU $u$ over PRBs. To this end, we have to consider two potential scenarios that could occur at time $t_{x{+}1}$:
\begin{enumerate}[label={\textbf{\arabic*})}]
    \item $\overline{\tau}_{u,u',r}^{\uparrow}(t_x)<\left(t_{x{+}1}-t_{x}\right)$: in this situation, the size of the data of DPU $u$ transmitted to CHU $u'$ over PRB $r$ during time window $\left(t_{x{+}1}-t_{x}\right)$ is calculated as $\overline{\tau}_{u,u',r}^{\uparrow}(t_x)\overline{\mathfrak{R}}^{\uparrow}_{u,u',r}(t_x)$
    \item $\overline{\tau}_{u,u',r}^{\uparrow}(t_x)\ge\left(t_{x{+}1}-t_{x}\right)$: in this situation, the size of the data of DPU $u$ transmitted to CHU $u'$ over PRB $r$ during time window $\left(t_{x{+}1}-t_{x}\right)$ is $\left(t_{x{+}1}-t_{x}\right)\overline{\mathfrak{R}}^{\uparrow}_{u,u',r}(t_x)$.
\end{enumerate}


These two cases can be captured compactly through $\min\left\{\overline{\tau}_{u,u',r}^{\uparrow}(t_x),\left(t_{x{+}1}-t_{x}\right)\right\}\overline{\mathfrak{R}}^{\uparrow}_{u,u',r}(t_x)$. On account of this, the total size of data of DPU $u$ disperses among its associated CHUs until FGTI $t_{x{+}1}$ is given by
\begin{align}
    \overline{M}^{\uparrow}_{u}(t_{x+1})=\hspace{-3mm}\sum_{z=N^{(k{-}1)}{+}1}^{x{+}1}&\sum_{u'{\in}\mathcal{U}_{b}}\sum_{r'{\in}\overline{\mathcal{R}}_{b}}\min\Big\{\overline{\tau}_{u,u',r'}^{\uparrow}(t_{z}),\left(t_{z{+}1}-t_{z}\right)\hspace{-1mm}\Big\}\overline{\mathfrak{R}}^{\uparrow}_{u,u',r'}(t_{z}).
\end{align}
Similarly, $M^{\downarrow}_{b}(x)$ is given by
\begin{align}
    \overline{M}^{\uparrow}_{u}(t_x)=\hspace{-3mm}\sum_{z=N^{(k{-}1)}{+}1}^{x}&\sum_{u'{\in}\mathcal{U}_{b}}\sum_{r'{\in}\overline{\mathcal{R}}_{b}}\min\Big\{\overline{\tau}_{u,u',r'}^{\uparrow}(t_{z}),\left(t_{z{+}1}-t_{z}\right)\hspace{-1mm}\Big\}\overline{\mathfrak{R}}^{\uparrow}_{u,u',r'}(t_{z}).
\end{align}
Consequently, the dispersion latency of DPU $u{\in}\mathcal{U}_{b}$ can be calculated as its cumulative dispersion latency during FGTIs of global round $k$, given by
\begin{equation}
\hspace{-5mm}
    \overline{\tau}_{u}^{\uparrow,(k)}{=}\hspace{-3mm}\sum_{x{=}N^{(k{-}1)}{+}1}^{N^{(k)}{-}1}\max_{\substack{u'{\in}\mathcal{U}_{b},\\r{\in}\overline{\mathcal{R}}_{b}}}\Big\{\min\big\{\overline{\tau}_{u,u',r}^{\uparrow}(t_{x}),\left(t_{x{+}1}{-}t_{x}\right)\big\}\hspace{-0.7mm}\Big\},
\hspace{-5mm}
\end{equation}   
where $\overline{\tau}_{u,u',r}^{\uparrow}(t_x)$ is given in \eqref{eq:DPU_communication_latency1}. This concludes the proof.

\newpage
\section{Proof of Proposition~\ref{propo:CHU_uplink}}\label{app:propo:CHU_uplink}
\noindent The proof is similar to that of Appendix \ref{app:propo:broadcast}. In particular, let $\alpha_{\bm{\omega}}M$ be the size of the data that must be dispatched to O-RU $b$, where $\alpha_{\bm{\omega}}$ is the number of bits required to represent one element of the gradient vector $\widetilde{\nabla {\mathfrak{L}}}_u^{(k)}$ with size $M$. Further, let $M^{\uparrow}_{u}(x)$ be the size of data (i.e., $\alpha_{\bm{\omega}}M$) of CHU $u$ dispatched to O-RU $b$ until FGTI $t_{x}$. Thus, $\alpha_{\bm{\omega}}M-M^{\uparrow}_{u}(x)$ is the size of the remaining data of CHU $u$ at time $t_{x}$ that must be dispatched to O-RU $b$. Moreover, assume $\sum_{r {\in} \mathcal{R}_{b}} \psi_{u,r}\hspace{-0.3mm}(t_{x})=1,~x\in\mathcal{N}^{(k)}$ (i.e., CHU $u$ attempts to dispatch all its remaining GPs to O-RU $b$ at time $t_{x}$). Accordingly, at FGTI $t_{x}$, the size of data that CHU $u$ will upload to O-RU $b$ over PRB $r$ is $\big(\alpha_{\bm{\omega}}M-M^{\uparrow}_{u}(x)\big)\psi_{u,r}\hspace{-0.3mm}(t_{x})$, the latency of transmission of which is given by
\begin{equation}\label{eq:CHU_communication_latency1}
\hspace{-5mm}
\tau_{u,r}^{\uparrow}\hspace{-0.3mm}(t_{x}){=}\frac{\big(\alpha_{\bm{\omega}} M{-}M^{\uparrow}_{u}(x)\big)\psi_{u,r}\hspace{-0.3mm}(t_{x})}{\mathfrak{R}^{\uparrow}_{u,r}\hspace{-0.3mm}(t_{x}){+}1{-}\beta^{\uparrow}_u(t_x)},
\hspace{-5mm}
 \end{equation}
where $1{-}\beta^{\uparrow}_u(t_x)$ is added to the denominator of \eqref{eq:CHU_communication_latency1} to avoid division by zero. At time $t_{x{+}1}$, MAC schedulers may perform rescheduling of users and reallocation of available resources (i.e., licensed PRBs and transmit powers). Therefore, we have to recalculate the latencies of dispatching the remaining data of CHU $u$ over PRBs. To this end, we have to consider two potential scenarios that could occur at time $t_{x{+}1}$:
\begin{enumerate}[label={\textbf{\arabic*})}]
    \item $\tau_{u,r}^{\uparrow}\hspace{-0.3mm}(t_{x})<\left(t_{x{+}1}-t_{x}\right)$: in this situation, the size of the data of CHU $u$ transmitted to O-RU $b$ over PRB $r$ during time window $\left(t_{x{+}1}-t_{x}\right)$ is calculated as $\tau_{u,r}^{\uparrow}\hspace{-0.3mm}(t_{x})\mathfrak{R}^{\uparrow}_{u,r}\hspace{-0.3mm}(t_{x})$.
    \item $\tau_{u,r}^{\uparrow}\hspace{-0.3mm}(t_{x})\ge\left(t_{x{+}1}-t_{x}\right)$: in this situation, the size of the data of CHU $u$ transmitted to O-RU $b$ over PRB $r$ during time window $\left(t_{x{+}1}-t_{x}\right)$ is $\left(t_{x{+}1}-t_{x}\right)\mathfrak{R}^{\uparrow}_{u,r}\hspace{-0.3mm}(t_{x})$.
\end{enumerate}
These two cases can be captured compactly through $\min\left\{\tau_{u,r}^{\uparrow}\hspace{-0.3mm}(t_{x}),\left(t_{x{+}1}-t_{x}\right)\right\}\mathfrak{R}^{\uparrow}_{u,r}\hspace{-0.3mm}(t_{x})$. On account of this, the total size of data of CHU $u$ dispatched to O-RU $b$ until FGTI $t_{x{+}1}$ is given by
\begin{equation}
\hspace{-5mm}
    M^{\uparrow}_{u}(x{+}1){=}\hspace{-5mm}\sum_{z{=}N^{(k{-}1)}{+}1}^{x{+}1}\sum_{r'{\in}\mathcal{R}_{b}}\hspace{-2mm}\min\hspace{-1mm}\big\{\tau_{u,r'}^{\uparrow}(t_{z}),\left(t_{z{+}1}{-}t_{z}\right)\hspace{-1mm}\big\}\mathfrak{R}^{\uparrow}_{u,r'}(t_{z}\hspace{-0.3mm}).
\hspace{-5mm}
\end{equation}
Similarly, $M^{\uparrow}_{u}(x)$ is given by
\begin{equation}
\hspace{-5mm}
    M^{\uparrow}_{u}(x){=}\sum_{z{=}N^{(k{-}1)}{+}1}^{x}\sum_{r'{\in}\mathcal{R}_{b}}\min\hspace{-1mm}\big\{\tau_{u,r'}^{\uparrow}(t_{z}),\left(t_{z{+}1}{-}t_{z}\right)\hspace{-1mm}\big\}\mathfrak{R}^{\uparrow}_{u,r'}(t_{z}\hspace{-0.3mm}).
\hspace{-5mm}
\end{equation}
Consequently, the dispatch latency of CHU $u\in\mathcal{U}_{b}$ to transmit its GPs to O-RU $b$ through multiple PRBs can be calculated as the cumulative dispatch latency during FGTIs of global round $k$, given by 
\begin{equation}
    \tau_{u}^{\uparrow,{(k)}}{=}\hspace{-3mm}\sum_{x{=}N^{(k{-}1)}{+}1}^{N^{(k)}{-}1}\max_{r{\in}\mathcal{R}_{b}}\Big\{\hspace{-1mm}\min\left\{\tau_{u,r}^{\uparrow}\hspace{-0.3mm}(t_{x}),\left(t_{x{+}1}-t_{x}\right)\right\}\hspace{-1mm}\Big\},
\end{equation}
where $\tau_{u,r}^{\uparrow}\hspace{-0.3mm}(t_{x})$ is given in \eqref{eq:CHU_communication_latency1}. This concludes the proof.

\newpage
\section{Proof of Proposition~\ref{propo:waiting_time}}\label{app:propo:waiting_time}
\noindent The waiting time of CHU $u$ to receive the LM of DPU $u'$ (if it is connected to CHU $u$) is the cumulative latency of transmitting LM from DPU $u'$ to CHU $u$ over multiple PRBs during all FGTIs of global round $k$, captured through 
\begin{equation}
    \sum_{x{=}N^{(k{-}1)}{+}1}^{N^{(k)}{-}1}\max_{r{\in}\overline{\mathcal{R}}_{b}}\left\{\min\hspace{-1mm}\big\{\overline{\tau}_{u',u,r}^{\uparrow}(t_{x}),\left(t_{x{+}1}{-}t_{x}\right)\hspace{-1mm}\big\}\right\},
\end{equation}
where $\overline{\tau}_{u',u,r}^{\uparrow}(t_{x})$ is given in \eqref{eq:DPU_communication_latency1}. Furthermore, since each CHU $u$ can act as the communication head for multiple DPUs, the waiting time of CHU $u$, referred to as $\tau_{u}^{\mathsf{W},{(k)}}$, is the maximum of  waiting times needed to receive the LMs of all DPUs associated with CHU $u$, given by
\begin{equation}
\hspace{-5mm}
    \tau_{u}^{\mathsf{W},{(k)}}=\max_{u'\in\mathcal{U}_{b}}\Big\{ \tau_{u'}^{\mathsf{LC},{(k)}}{+}\hspace{-2mm}\sum_{x{=}N^{(k{-}1)}{+}1}^{N^{(k)}{-}1}\max_{r{\in}\overline{\mathcal{R}}_{b}}\left\{L_{u,u',r}(t_x)\right\}\Big\},
\hspace{-5mm}
\end{equation}
where $L_{u,u',r}(t_x){=}\min\hspace{-1mm}\big\{\overline{\tau}_{u',u,r}^{\uparrow}(t_{x}),\left(t_{x{+}1}{-}t_{x}\right)\hspace{-1mm}\big\}$, and $x{\in}\mathcal{N}^{(k)}_b$.

\newpage
\newpage
\section{Proof of Theorem \ref{th:main}}\label{app:th:main}
\noindent Recalling from Sec.~\ref{sec:model_aggregation}, we have 
\begin{align}\label{app:eq:normalLocalCHU}
    &\widetilde{\nabla {\mathfrak{L}}}_{u}^{\mathsf{Frag},(k)}=\frac{\lambda_{u}^{(k)}|\Upsilon_{u}(\widetilde{\tau}_{b}^{\downarrow,{(k)}})|}{\ell^{(k)}_{u}} \widetilde{\nabla {\mathfrak{L}}}_{u}^{(k)}+\sum_{x{\in}\mathcal{N}^{(k)}}\sum_{u'\in \mathcal{U}_{b}}\sum_{r{\in}\overline{\mathcal{R}}_{b}}\frac{\overline{\lambda}_{u'}^{(k)} |\Upsilon_{u'}(\widetilde{\tau}_{b}^{\downarrow,{(k)}})|}{ \ell^{(k)}_{u'}}\widetilde{\nabla {\mathfrak{L}}}_{u',u,r}^{(k)}(t_{x}), 
\end{align}
and 
\begin{equation}\label{app:eq:normalLocalORU}
        \widetilde{\nabla {\mathfrak{L}}}_{b}^{(k)}=\frac{1}{|\Upsilon^{\mathsf{s}}(\widetilde{\bm{\tau}}^{\downarrow,{(k)}})|}\sum_{u\in \mathcal{U}_{b}}\widetilde{\nabla {\mathfrak{L}}}_{u}^{\mathsf{Frag}(k)}.
\end{equation}
Substituting \eqref{app:eq:normalLocalCHU} in \eqref{app:eq:normalLocalORU} gives us
\begin{equation}\label{eq:normalLocal}
    \widetilde{\nabla {\mathfrak{L}}}_{b}^{(k)}=\sum_{u\in \mathcal{U}_{b}}\frac{\widehat{\lambda}_{u}^{(k)}|\Upsilon_{u}(\widetilde{\tau}_{b}^{\downarrow,{(k)}})|}{|\Upsilon^{\mathsf{s}}(\widetilde{\bm{\tau}}^{\downarrow,{(k)}})| \ell^{(k)}_{u}}\widetilde{\nabla {\mathfrak{L}}}_{u}^{(k)}.
\end{equation}
Consequently, $\widetilde{\nabla {\mathfrak{L}}}^{(k)}$ can be rewritten as follows:
\begin{equation}
    \widetilde{\nabla {\mathfrak{L}}}^{(k)}=\mathfrak{B}_k\times\sum_{b\in \Omega}\sum_{u\in \mathcal{U}_{b}}\frac{\widehat{\lambda}_{u}^{(k)}|\Upsilon_{u}(\widetilde{\tau}_{b}^{\downarrow,{(k)}})|}{|\Upsilon^{\mathsf{s}}(\widetilde{\bm{\tau}}^{\downarrow,{(k)}})| \ell^{(k)}_{u}}\widetilde{\nabla {\mathfrak{L}}}_{u}^{(k)},
\end{equation}
where
\begin{equation}
    \mathfrak{B}_k = \sum_{b'\in \Omega}\sum_{u'\in \mathcal{U}_{b'}}\frac{|\Upsilon^{\mathsf{s}}(\widetilde{\bm{\tau}}^{\downarrow,{(k)}})| \ell^{(k)}_{u'}}{\widehat{\lambda}_{u'}^{(k)}|\Upsilon_{u'}(\widetilde{\tau}_{b'}^{\downarrow,{(k)}})|}.
\end{equation}
Furthermore, consider the following lemma which is a direct result of $\beta$-smoothness of a function.
\begin{lemma}[Smooth Function Characteristics~\cite{2200000050}]\label{lemma:smooth}
    Let $f$ be a $\beta$-smooth function on $\mathbb{R}^n$. Then for any $\bm{x},\bm{y}\in \mathbb{R}^n$, we have
    \begin{equation}
        f(\bm{x})\le f(\bm{y})+\left\langle\nabla f(\bm{y}), (\bm{x}-\bm{y})\right\rangle +\frac{\beta}{2}\left\Vert \bm{x}-\bm{y} \right\Vert^2.
    \end{equation}
\end{lemma}

\noindent Using the $\beta$-smoothness of the global loss function (Assumption~\ref{Assup:lossFun}) and considering Lemma~\ref{lemma:smooth}, we have 
\begin{equation}
 \mathfrak{L}^{({k})}(\bm{\omega}^{(k+1)}) \leq \mathfrak{L}^{({k})}(\bm{\omega}^{(k)}) +  \left\langle\nabla{\mathfrak{L}^{({k})}(\bm{\omega}^{(k)})},\left( \bm{\omega}^{(k+1)} - \bm{\omega}^{(k)}\right)\right\rangle+ \frac{\beta}{2} \left\Vert \bm{\omega}^{(k+1)} - \bm{\omega}^{(k)}\right\Vert^2.
\end{equation}
Replacing the updating rule for $\bm{\omega}^{(k+1)}$ and taking the conditional expectation (with respect to randomized data sampling at the last aggregation) from both hand sides yields
\begin{align}\label{ineq:main1}
 \mathbb{E}_k\left[\mathfrak{L}^{({k})}(\bm{\omega}^{(k+1)})\right] &\leq \mathbb{E}_k\left[\mathfrak{L}^{({k})}(\bm{\omega}^{(k)}) -  \left\langle\nabla{\mathfrak{L}^{({k})}(\bm{\omega}^{(k)})},\eta_{_k}\widetilde{\nabla \mathfrak{L}}^{(k)}\right\rangle+ \frac{\beta}{2} \left\Vert \eta_{_k}\widetilde{\nabla \mathfrak{L}}^{(k)}\right\Vert^2\right]\nonumber\\
 &= \mathfrak{L}^{({k})}(\bm{\omega}^{(k)}) - \eta_{_k}\mathfrak{B}_k\mathbb{E}_k\left[\left\langle\nabla{\mathfrak{L}^{({k})}(\bm{\omega}^{(k)})},\sum_{b \in \Omega} \sum_{u\in \mathcal{U}_{b}}\frac{\widehat{\lambda}_{u}^{(k)}|\Upsilon_{u}(\widetilde{\tau}_{b}^{\downarrow,{(k)}})|}{|\Upsilon^{\mathsf{s}}(\widetilde{\bm{\tau}}^{\downarrow,{(k)}})| \ell^{(k)}_{u}}\widetilde{\nabla \mathfrak{L}}_{u}^{(k)}\right\rangle\right]\nonumber\\
 &~~~~~~~~~+ \frac{\beta\eta_{_k}^2\mathfrak{B}_k^2}{2} \mathbb{E}_k\left[\left\Vert\sum_{b \in \Omega} \sum_{u\in \mathcal{U}_{b}}\frac{\widehat{\lambda}_{u}^{(k)}|\Upsilon_{u}(\widetilde{\tau}_{b}^{\downarrow,{(k)}})|}{|\Upsilon^{\mathsf{s}}(\widetilde{\bm{\tau}}^{\downarrow,{(k)}})| \ell^{(k)}_{u}}\widetilde{\nabla \mathfrak{L}}_{u}^{(k)}\right\Vert^2\right].
\end{align}
Since $\widetilde{\nabla \mathfrak{L}}_{u}^{(k)} = \frac{1}{\eta_{_k}}\left(\bm{\omega}^{(k)}-\bm{\omega}_{u}^{(k),\ell^{(k)}_{u}}\right)$, via recursive expansion of the updating rule in~\eqref{eq:WeightupdateStrat}, we get
\begin{equation}\label{eq:nablabarF}
\widetilde{\nabla \mathfrak{L}}_{u}^{(k)} = \sum_{\ell=1}^{\ell^{(k)}_{u}} \sum_{\xi\in \mathcal{B}^{(k),\ell}_{u}} \hspace{-3mm} {\frac{\nabla  f_u(\bm{\omega}^{(k),\ell-1}_{u},\xi)}{{B}_{u}(\widetilde{\tau}_{b}^{\downarrow,{(k)}})}}.
\end{equation}
Replacing the above result back in \eqref{ineq:main1} leads to
\begin{equation}
\begin{aligned}
    \mathbb{E}_k\left[\mathfrak{L}^{({k})}(\bm{\omega}^{(k+1)})\right] &\leq  \mathfrak{L}^{({k})}(\bm{\omega}^{(k)}) - \eta_{_k} \mathfrak{B}_k\mathbb{E}_k\left[\left\langle\nabla{\mathfrak{L}^{({k})}(\bm{\omega}^{(k)})},\sum_{b \in \Omega} \sum_{u\in \mathcal{U}_{b}}\frac{\widehat{\lambda}_{u}^{(k)}|\Upsilon_{u}(\widetilde{\tau}_{b}^{\downarrow,{(k)}})|}{|\Upsilon^{\mathsf{s}}(\widetilde{\bm{\tau}}^{\downarrow,{(k)}})| \ell^{(k)}_{u}}\sum_{\ell=1}^{\ell^{(k)}_{u}} \sum_{\xi\in \mathcal{B}^{(k),\ell}_{u}} \hspace{-3mm} {\frac{\nabla  f_u(\bm{\omega}^{(k),\ell-1}_{u},\xi)}{{B}_{u}(\widetilde{\tau}_{b}^{\downarrow,{(k)}})}}\right\rangle\right]\nonumber\\
    &+ \frac{\beta\eta_{_k}^2\mathfrak{B}_k^2}{2} \mathbb{E}_k\left[\left\Vert \sum_{b \in \Omega} \sum_{u\in \mathcal{U}_{b}}\frac{\widehat{\lambda}_{u}^{(k)}|\Upsilon_{u}(\widetilde{\tau}_{b}^{\downarrow,{(k)}})|}{|\Upsilon^{\mathsf{s}}(\widetilde{\bm{\tau}}^{\downarrow,{(k)}})| \ell^{(k)}_{u}}\widetilde{\nabla \mathfrak{L}}_{u}^{(k)}\right\Vert^2\right].
\end{aligned}
\end{equation}
Let $\mathscr{N}_u^{\mathsf{G},(k),\ell}=\sum_{\xi\in \mathcal{B}^{(k),\ell}_{u}}{\frac{\nabla  f_u(\bm{\omega}^{(k),\ell-1}_{u},\xi)}{{B}_{u}(\widetilde{\tau}_{b}^{\downarrow,{(k)}})}}-{\nabla  \mathfrak{L}^{(k)}_u(\bm{\omega}^{(k),\ell-1}_{u})}$ denote the noise of SGD of user $u$, where ${\nabla  \mathfrak{L}^{(k)}_u(\bm{\omega}^{(k),\ell-1}_{u})}=\sum_{\xi\in\Upsilon_{u}(\widetilde{\tau}_{b}^{\downarrow,{(k)}})}\frac{\nabla  f_u(\bm{\omega}^{(k),\ell-1}_{u},\xi)}{|\Upsilon_{u}(\widetilde{\tau}_{b}^{\downarrow,{(k)}})|}$. The above inequality can be written as follows:
\begin{align}\label{ineq:main2}
    \mathbb{E}_k&\left[\mathfrak{L}^{({k})}(\bm{\omega}^{(k+1)})\right] \leq  \mathfrak{L}^{({k})}(\bm{\omega}^{(k)})- \eta_{_k}\mathfrak{B}_k \mathbb{E}_k\left[\left\langle\nabla{\mathfrak{L}^{({k})}(\bm{\omega}^{(k)})},\sum_{b \in \Omega} \sum_{u\in \mathcal{U}_{b}}\frac{\widehat{\lambda}_{u}^{(k)}|\Upsilon_{u}(\widetilde{\tau}_{b}^{\downarrow,{(k)}})|}{|\Upsilon^{\mathsf{s}}(\widetilde{\bm{\tau}}^{\downarrow,{(k)}})| \ell^{(k)}_{u}}\sum_{\ell=1}^{\ell^{(k)}_u}  \left({\nabla  \mathfrak{L}^{(k)}_u(\bm{\omega}^{(k),\ell-1}_{u})}+\mathscr{N}_u^{\mathsf{G},(k),\ell}\right)\right\rangle\right]\nonumber\\
    &~~~~~~~~+ \frac{\beta\eta_{_k}^2\mathfrak{B}_k^2}{2} \mathbb{E}_k\left[\left\Vert \sum_{b \in \Omega} \sum_{u\in \mathcal{U}_{b}}\frac{\widehat{\lambda}_{u}^{(k)}|\Upsilon_{u}(\widetilde{\tau}_{b}^{\downarrow,{(k)}})|}{|\Upsilon^{\mathsf{s}}(\widetilde{\bm{\tau}}^{\downarrow,{(k)}})| \ell^{(k)}_{u}}\widetilde{\nabla \mathfrak{L}}_{u}^{(k)}\right\Vert^2\right]\nonumber\\
    &=\mathfrak{L}^{({k})}(\bm{\omega}^{(k)}) - \eta_{_k}\mathfrak{B}_k \mathbb{E}_k\left[\left\langle\nabla{\mathfrak{L}^{({k})}(\bm{\omega}^{(k)})},\sum_{b \in \Omega} \sum_{u\in \mathcal{U}_{b}}\frac{\widehat{\lambda}_{u}^{(k)}|\Upsilon_{u}(\widetilde{\tau}_{b}^{\downarrow,{(k)}})|}{|\Upsilon^{\mathsf{s}}(\widetilde{\bm{\tau}}^{\downarrow,{(k)}})| \ell^{(k)}_{u}}\sum_{\ell=1}^{\ell^{(k)}_u}  {\nabla  \mathfrak{L}^{(k)}_u(\bm{\omega}^{(k),\ell-1}_{u})}\right\rangle\right]\nonumber\\
    &~~~~~~~-\underbrace{\eta_{_k}\mathfrak{B}_k \mathbb{E}_k\left[\left\langle\nabla{\mathfrak{L}^{({k})}(\bm{\omega}^{(k)})},\sum_{b \in \Omega} \sum_{u\in \mathcal{U}_{b}}\frac{\widehat{\lambda}_{u}^{(k)}|\Upsilon_{u}(\widetilde{\tau}_{b}^{\downarrow,{(k)}})|}{|\Upsilon^{\mathsf{s}}(\widetilde{\bm{\tau}}^{\downarrow,{(k)}})| \ell^{(k)}_{u}}\sum_{\ell=1}^{\ell^{(k)}_u}\mathscr{N}_u^{\mathsf{G},(k),\ell} \right\rangle\right]}_{(x_1)}+ \frac{\beta\eta_{_k}^2\mathfrak{B}_k^2}{2} \mathbb{E}_k\left[\left\Vert \sum_{b \in \Omega} \sum_{u\in \mathcal{U}_{b}}\frac{\widehat{\lambda}_{u}^{(k)}|\Upsilon_{u}(\widetilde{\tau}_{b}^{\downarrow,{(k)}})|}{|\Upsilon^{\mathsf{s}}(\widetilde{\bm{\tau}}^{\downarrow,{(k)}})| \ell^{(k)}_{u}}\widetilde{\nabla \mathfrak{L}}_{u}^{(k)}\right\Vert^2\right]\nonumber\\
    &\overset{(i)}{=}\mathfrak{L}^{({k})}(\bm{\omega}^{(k)}) - \eta_{_k}\mathfrak{B}_k \mathbb{E}_k\left[\left\langle\nabla{\mathfrak{L}^{({k})}(\bm{\omega}^{(k)})},\sum_{b \in \Omega} \sum_{u\in \mathcal{U}_{b}}\frac{\widehat{\lambda}_{u}^{(k)}|\Upsilon_{u}(\widetilde{\tau}_{b}^{\downarrow,{(k)}})|}{|\Upsilon^{\mathsf{s}}(\widetilde{\bm{\tau}}^{\downarrow,{(k)}})| \ell^{(k)}_{u}}\sum_{\ell=1}^{\ell^{(k)}_u}  {\nabla  \mathfrak{L}^{(k)}_u(\bm{\omega}^{(k),\ell-1}_{u})}\right\rangle\right]\nonumber\\
    &~~~~~~~+ \frac{\beta \eta_{_k}^2\mathfrak{B}_k^2}{2} \mathbb{E}_k\left[\left\Vert \sum_{b \in \Omega} \sum_{u\in \mathcal{U}_{b}}\frac{\widehat{\lambda}_{u}^{(k)}|\Upsilon_{u}(\widetilde{\tau}_{b}^{\downarrow,{(k)}})|}{|\Upsilon^{\mathsf{s}}(\widetilde{\bm{\tau}}^{\downarrow,{(k)}})| \ell^{(k)}_{u}}\widetilde{\nabla \mathfrak{L}}_{u}^{(k)}\right\Vert^2\right],
\end{align}
where $(i)$ uses the fact that $(x_1)$ is zero since the noise of gradient estimation across the mini-batches are independent and zero mean. Further, let 
\begin{equation}\label{eq:user_selection_noise}
    \mathscr{U}^{(k)}_{b}= \sum_{u\in \mathcal{U}_{b}}\frac{\widehat{\lambda}_{u}^{(k)}|\Upsilon_{u}(\widetilde{\tau}_{b}^{\downarrow,{(k)}})|}{|\Upsilon^{\mathsf{s}}(\widetilde{\bm{\tau}}^{\downarrow,{(k)}})| \ell^{(k)}_{u}}\sum_{\ell=1}^{\ell^{(k)}_{u}} {\nabla  \mathfrak{L}^{(k)}_u(\bm{\omega}^{(k),\ell-1}_{u})}- \sum_{u\in \mathcal{U}_{b}}\frac{|\Upsilon_{u}(\widetilde{\tau}_{b}^{\downarrow,{(k)}})|}{|\Upsilon(\widetilde{\bm{\tau}}^{\downarrow,{(k)}})| \ell^{(k)}_{u}}\sum_{\ell=1}^{\ell^{(k)}_{u}} {\nabla  \mathfrak{L}^{(k)}_u(\bm{\omega}^{(k),\ell-1}_{u})} 
\end{equation}
denote the noise of user recruitment at O-RU $b$. Inequality \eqref{ineq:main2} can be rewritten as follows:
\begin{align}
    \mathbb{E}_k&\left[\mathfrak{L}^{({k})}(\bm{\omega}^{(k+1)})\right] \leq \mathfrak{L}^{({k})}(\bm{\omega}^{(k)}) - \eta_{_k} \mathfrak{B}_k\mathbb{E}_k\left[\left\langle\nabla{\mathfrak{L}^{({k})}(\bm{\omega}^{(k)})},\sum_{b \in \Omega} \sum_{u\in \mathcal{U}_{b}}\frac{|\Upsilon_{u}(\widetilde{\tau}_{b}^{\downarrow,{(k)}})|}{|\Upsilon(\widetilde{\bm{\tau}}^{\downarrow,{(k)}})| \ell^{(k)}_{u}}\sum_{\ell=1}^{\ell^{(k)}_u}  {\nabla  \mathfrak{L}^{(k)}_u(\bm{\omega}^{(k),\ell-1}_{u})}+ \mathscr{U}^{(k)}_{b} \right\rangle\right]\nonumber\\
    &~~~~~~~~~~~+ \frac{\beta\eta_{_k}^2\mathfrak{B}_k^2}{2} \mathbb{E}_k\left[\left\Vert \sum_{b \in \Omega} \sum_{u\in \mathcal{U}_{b}}\frac{\widehat{\lambda}_{u}^{(k)}|\Upsilon_{u}(\widetilde{\tau}_{b}^{\downarrow,{(k)}})|}{|\Upsilon^{\mathsf{s}}(\widetilde{\bm{\tau}}^{\downarrow,{(k)}})| \ell^{(k)}_{u}}\widetilde{\nabla \mathfrak{L}}_{u}^{(k)}\right\Vert^2\right]\nonumber\\
    &=\mathfrak{L}^{({k})}(\bm{\omega}^{(k)}) - \eta_{_k} \mathfrak{B}_k\mathbb{E}_k\left[\left\langle\nabla{\mathfrak{L}^{({k})}(\bm{\omega}^{(k)})},\sum_{b \in \Omega} \sum_{u\in \mathcal{U}_{b}}\frac{|\Upsilon_{u}(\widetilde{\tau}_{b}^{\downarrow,{(k)}})|}{|\Upsilon(\widetilde{\bm{\tau}}^{\downarrow,{(k)}})| \ell^{(k)}_{u}}\sum_{\ell=1}^{\ell^{(k)}_u}  {\nabla  \mathfrak{L}^{(k)}_u(\bm{\omega}^{(k),\ell-1}_{u})}\right\rangle\right]\nonumber\\
    &+ \eta_{_k}\mathfrak{B}_k \mathbb{E}_k\left[\left\langle\nabla{\mathfrak{L}^{({k})}(\bm{\omega}^{(k)})},-\sum_{b \in \Omega}\mathscr{U}^{(k)}_{b} \right\rangle\right]+ \frac{\beta\eta_{_k}^2\mathfrak{B}_k^2}{2} \mathbb{E}_k\Bigg[\Bigg\Vert \sum_{b \in \Omega} \sum_{u\in \mathcal{U}_{b}}\frac{\widehat{\lambda}_{u}^{(k)}|\Upsilon_{u}(\widetilde{\tau}_{b}^{\downarrow,{(k)}})|}{|\Upsilon^{\mathsf{s}}(\widetilde{\bm{\tau}}^{\downarrow,{(k)}})| \ell^{(k)}_{u}}\widetilde{\nabla \mathfrak{L}}_{u}^{(k)} \Bigg\Vert^2\Bigg].
\end{align}
Similarly, writing the last term of the above inequality in terms of $\mathscr{N}_u^{\mathsf{G},(k),\ell}$ gives us
\begin{align}
    \mathbb{E}_k&\left[\mathfrak{L}^{({k})}(\bm{\omega}^{(k+1)})\right] \leq \mathfrak{L}^{({k})}(\bm{\omega}^{(k)}) - \eta_{_k} \mathfrak{B}_k\mathbb{E}_k\left[\left\langle\nabla{\mathfrak{L}^{({k})}(\bm{\omega}^{(k)})},\sum_{b \in \Omega} \sum_{u\in \mathcal{U}_{b}}\frac{|\Upsilon_{u}(\widetilde{\tau}_{b}^{\downarrow,{(k)}})|}{|\Upsilon(\widetilde{\bm{\tau}}^{\downarrow,{(k)}})| \ell^{(k)}_{u}}\sum_{\ell=1}^{\ell^{(k)}_u}  {\nabla  \mathfrak{L}^{(k)}_u(\bm{\omega}^{(k),\ell-1}_{u})}\right\rangle\right]\nonumber\\
    &+ \eta_{_k}\mathfrak{B}_k \mathbb{E}_k\left[\left\langle\nabla{\mathfrak{L}^{({k})}(\bm{\omega}^{(k)})},-\sum_{b \in \Omega}\mathscr{U}^{(k)}_{b} \right\rangle\right]+ \frac{\beta\eta_{_k}^2\mathfrak{B}_k^2}{2} \mathbb{E}_k\Bigg[\Bigg\Vert \sum_{b \in \Omega} \sum_{u\in \mathcal{U}_{b}}\frac{\widehat{\lambda}_{u}^{(k)}|\Upsilon_{u}(\widetilde{\tau}_{b}^{\downarrow,{(k)}})|}{|\Upsilon^{\mathsf{s}}(\widetilde{\bm{\tau}}^{\downarrow,{(k)}})| \ell^{(k)}_{u}}\sum_{\ell=1}^{\ell^{(k)}_{u}} \sum_{\xi\in \mathcal{B}^{(k),\ell}_{u}} \hspace{-3mm} {\frac{\nabla  f_u(\bm{\omega}^{(k),\ell-1}_{u},\xi)}{{B}_{u}(\widetilde{\tau}_{b}^{\downarrow,{(k)}})}} \Bigg\Vert^2\Bigg]\nonumber\\
    &=\mathfrak{L}^{({k})}(\bm{\omega}^{(k)}) - \eta_{_k} \mathfrak{B}_k\mathbb{E}_k\left[\left\langle\nabla{\mathfrak{L}^{({k})}(\bm{\omega}^{(k)})},\sum_{b \in \Omega} \sum_{u\in \mathcal{U}_{b}}\frac{|\Upsilon_{u}(\widetilde{\tau}_{b}^{\downarrow,{(k)}})|}{|\Upsilon(\widetilde{\bm{\tau}}^{\downarrow,{(k)}})| \ell^{(k)}_{u}}\sum_{\ell=1}^{\ell^{(k)}_u}  {\nabla  \mathfrak{L}^{(k)}_u(\bm{\omega}^{(k),\ell-1}_{u})}\right\rangle\right]\nonumber\\
    &+ \eta_{_k}\mathfrak{B}_k \mathbb{E}_k\left[\left\langle\nabla{\mathfrak{L}^{({k})}(\bm{\omega}^{(k)})},-\sum_{b \in \Omega}\mathscr{U}^{(k)}_{b} \right\rangle\right]+ \frac{\beta\eta_{_k}^2\mathfrak{B}_k^2}{2} \mathbb{E}_k\Bigg[\Bigg\Vert \sum_{b \in \Omega} \sum_{u\in \mathcal{U}_{b}}\frac{\widehat{\lambda}_{u}^{(k)}|\Upsilon_{u}(\widetilde{\tau}_{b}^{\downarrow,{(k)}})|}{|\Upsilon^{\mathsf{s}}(\widetilde{\bm{\tau}}^{\downarrow,{(k)}})| \ell^{(k)}_{u}}\sum_{\ell=1}^{\ell^{(k)}_{u}} \left({\nabla  \mathfrak{L}^{(k)}_u(\bm{\omega}^{(k),\ell-1}_{u})}+\mathscr{N}_u^{\mathsf{G},(k),\ell}\right) \Bigg\Vert^2\Bigg]\nonumber\\
    &=\mathfrak{L}^{({k})}(\bm{\omega}^{(k)}) - \eta_{_k}\mathfrak{B}_k \mathbb{E}_k\left[\left\langle\nabla{\mathfrak{L}^{({k})}(\bm{\omega}^{(k)})},\sum_{b \in \Omega} \sum_{u\in \mathcal{U}_{b}}\frac{|\Upsilon_{u}(\widetilde{\tau}_{b}^{\downarrow,{(k)}})|}{|\Upsilon(\widetilde{\bm{\tau}}^{\downarrow,{(k)}})| \ell^{(k)}_{u}}\sum_{\ell=1}^{\ell^{(k)}_u}  {\nabla  \mathfrak{L}^{(k)}_u(\bm{\omega}^{(k),\ell-1}_{u})}\right\rangle\right]\nonumber\\
    &~~+ \eta_{_k}\mathfrak{B}_k \mathbb{E}_k\left[\left\langle\nabla{\mathfrak{L}^{({k})}(\bm{\omega}^{(k)})},-\sum_{b \in \Omega}\mathscr{U}^{(k)}_{b} \right\rangle\right]\nonumber\\
    &~~+ \frac{\beta\eta_{_k}^2\mathfrak{B}_k^2}{2} \mathbb{E}_k\left[\left\Vert \sum_{b \in \Omega} \sum_{u\in \mathcal{U}_{b}}\frac{\widehat{\lambda}_{u}^{(k)}|\Upsilon_{u}(\widetilde{\tau}_{b}^{\downarrow,{(k)}})|}{|\Upsilon^{\mathsf{s}}(\widetilde{\bm{\tau}}^{\downarrow,{(k)}})| \ell^{(k)}_{u}}\sum_{\ell=1}^{\ell^{(k)}_{u}} {\nabla  \mathfrak{L}^{(k)}_u(\bm{\omega}^{(k),\ell-1}_{u})}+\left(\sum_{b \in \Omega} \sum_{u\in \mathcal{U}_{b}}\frac{\widehat{\lambda}_{u}^{(k)}|\Upsilon_{u}(\widetilde{\tau}_{b}^{\downarrow,{(k)}})|}{|\Upsilon^{\mathsf{s}}(\widetilde{\bm{\tau}}^{\downarrow,{(k)}})| \ell^{(k)}_{u}}\sum_{\ell=1}^{\ell^{(k)}_{u}}\mathscr{N}_u^{\mathsf{G},(k),\ell}\right) \right\Vert^2\right]\nonumber\\
    &\overset{(i)}{=}\mathfrak{L}^{({k})}(\bm{\omega}^{(k)}) - \eta_{_k}\mathfrak{B}_k \mathbb{E}_k\left[\left\langle\nabla{\mathfrak{L}^{({k})}(\bm{\omega}^{(k)})},\sum_{b \in \Omega} \sum_{u\in \mathcal{U}_{b}}\frac{|\Upsilon_{u}(\widetilde{\tau}_{b}^{\downarrow,{(k)}})|}{|\Upsilon(\widetilde{\bm{\tau}}^{\downarrow,{(k)}})| \ell^{(k)}_{u}}\sum_{\ell=1}^{\ell^{(k)}_u}  {\nabla  \mathfrak{L}^{(k)}_u(\bm{\omega}^{(k),\ell-1}_{u})}\right\rangle\right]\nonumber\\
    &~~+ \eta_{_k}\mathfrak{B}_k \mathbb{E}_k\left[\left\langle\nabla{\mathfrak{L}^{({k})}(\bm{\omega}^{(k)})},-\sum_{b \in \Omega}\mathscr{U}^{(k)}_{b} \right\rangle\right]\nonumber\\
    &~~+ \frac{\beta\eta_{_k}^2\mathfrak{B}_k^2}{2} \mathbb{E}_k\Bigg[\Bigg\Vert \sum_{b \in \Omega} \sum_{u\in \mathcal{U}_{b}}\frac{\widehat{\lambda}_{u}^{(k)}|\Upsilon_{u}(\widetilde{\tau}_{b}^{\downarrow,{(k)}})|}{|\Upsilon^{\mathsf{s}}(\widetilde{\bm{\tau}}^{\downarrow,{(k)}})| \ell^{(k)}_{u}}\sum_{\ell=1}^{\ell^{(k)}_{u}} {\nabla  \mathfrak{L}^{(k)}_u(\bm{\omega}^{(k),\ell-1}_{u})}\Bigg\Vert^2\Bigg]+\frac{\beta\eta_{_k}^2\mathfrak{B}_k^2}{2}\mathbb{E}_k\Bigg[\Bigg\Vert\sum_{b \in \Omega} \sum_{u\in \mathcal{U}_{b}}\frac{\widehat{\lambda}_{u}^{(k)}|\Upsilon_{u}(\widetilde{\tau}_{b}^{\downarrow,{(k)}})|}{|\Upsilon^{\mathsf{s}}(\widetilde{\bm{\tau}}^{\downarrow,{(k)}})| \ell^{(k)}_{u}}\sum_{\ell=1}^{\ell^{(k)}_{u}}\mathscr{N}_u^{\mathsf{G},(k),\ell}\Bigg\Vert^2\Bigg]\nonumber\\
    &~~+\underbrace{\beta\eta_{_k}^2\mathfrak{B}_k^2 \mathbb{E}_k\Bigg[\left\langle \sum_{b \in \Omega} \sum_{u\in \mathcal{U}_{b}}\frac{\widehat{\lambda}_{u}^{(k)}|\Upsilon_{u}(\widetilde{\tau}_{b}^{\downarrow,{(k)}})|}{|\Upsilon^{\mathsf{s}}(\widetilde{\bm{\tau}}^{\downarrow,{(k)}})| \ell^{(k)}_{u}}\sum_{\ell=1}^{\ell^{(k)}_{u}} {\nabla  \mathfrak{L}^{(k)}_u(\bm{\omega}^{(k),\ell-1}_{u})} ,\sum_{b \in \Omega} \sum_{u\in \mathcal{U}_{b}}\frac{\widehat{\lambda}_{u}^{(k)}|\Upsilon_{u}(\widetilde{\tau}_{b}^{\downarrow,{(k)}})|}{|\Upsilon^{\mathsf{s}}(\widetilde{\bm{\tau}}^{\downarrow,{(k)}})| \ell^{(k)}_{u}}\sum_{\ell=1}^{\ell^{(k)}_{u}}\mathscr{N}_u^{\mathsf{G},(k),\ell}\right\rangle\Bigg]}_{(x_2)}\nonumber\\
    &\overset{(ii)}{=}\mathfrak{L}^{({k})}(\bm{\omega}^{(k)}) - \eta_{_k}\mathfrak{B}_k \mathbb{E}_k\left[\left\langle\nabla{\mathfrak{L}^{({k})}(\bm{\omega}^{(k)})},\sum_{b \in \Omega} \sum_{u\in \mathcal{U}_{b}}\frac{|\Upsilon_{u}(\widetilde{\tau}_{b}^{\downarrow,{(k)}})|}{|\Upsilon(\widetilde{\bm{\tau}}^{\downarrow,{(k)}})| \ell^{(k)}_{u}}\sum_{\ell=1}^{\ell^{(k)}_u}  {\nabla  \mathfrak{L}^{(k)}_u(\bm{\omega}^{(k),\ell-1}_{u})}\right\rangle\right]\nonumber\\
    &~+ \eta_{_k}\mathfrak{B}_k \mathbb{E}_k\left[\left\langle\nabla{\mathfrak{L}^{({k})}(\bm{\omega}^{(k)})},-\sum_{b \in \Omega}\mathscr{U}^{(k)}_{b} \right\rangle\right]\nonumber\\
    &~+ \frac{\beta\eta_{_k}^2\mathfrak{B}_k^2}{2} \mathbb{E}_k\Bigg[\Bigg\Vert \sum_{b \in \Omega} \sum_{u\in \mathcal{U}_{b}}\frac{\widehat{\lambda}_{u}^{(k)}|\Upsilon_{u}(\widetilde{\tau}_{b}^{\downarrow,{(k)}})|}{|\Upsilon^{\mathsf{s}}(\widetilde{\bm{\tau}}^{\downarrow,{(k)}})| \ell^{(k)}_{u}}\sum_{\ell=1}^{\ell^{(k)}_{u}} {\nabla  \mathfrak{L}^{(k)}_u(\bm{\omega}^{(k),\ell-1}_{u})}\Bigg\Vert^2\Bigg]+\frac{\beta\eta_{_k}^2\mathfrak{B}_k^2}{2}\mathbb{E}_k\Bigg[\Bigg\Vert\sum_{b \in \Omega} \sum_{u\in \mathcal{U}_{b}}\frac{\widehat{\lambda}_{u}^{(k)}|\Upsilon_{u}(\widetilde{\tau}_{b}^{\downarrow,{(k)}})|}{|\Upsilon^{\mathsf{s}}(\widetilde{\bm{\tau}}^{\downarrow,{(k)}})| \ell^{(k)}_{u}}\sum_{\ell=1}^{\ell^{(k)}_{u}}\mathscr{N}_u^{\mathsf{G},(k),\ell}\Bigg\Vert^2\Bigg],
\end{align}
where $(i)$ uses the fact that for any two real valued vectors $\bm{a}$ and $\bm{b}$ with the same length, we have $\Vert \bm{a}+\bm{b}\Vert^2=\Vert \bm{a}\Vert^2+\Vert \bm{b}\Vert^2+2\bm{a}^\top\bm{b}$. Further, $(ii)$ in the above inequality is due to the fact that $(x_2)$ is zero since the noise of gradient estimation of users are independent and zero mean. Using $\Vert \sum_{i=1}^{n}\bm{x}_i\Vert^2=\sum_{i=1}^{n}\Vert\bm{x}_i\Vert^2+\sum_{\substack{j=1\\j\neq i}}^{n}\left\langle\bm{x}_i,\bm{x}_j\right\rangle$, where $\bm{x}\in \mathbb{R}^m$, and performing some algebraic manipulation, we expand the last term of the above inequality as follows: 
\begin{align}\label{ineq:main3}
    \mathbb{E}_k&\left[\mathfrak{L}^{({k})}(\bm{\omega}^{(k+1)})\right] \leq \mathfrak{L}^{({k})}(\bm{\omega}^{(k)}) - \eta_{_k}\mathfrak{B}_k \mathbb{E}_k\left[\left\langle\nabla{\mathfrak{L}^{({k})}(\bm{\omega}^{(k)})},\sum_{b \in \Omega} \sum_{u\in \mathcal{U}_{b}}\frac{|\Upsilon_{u}(\widetilde{\tau}_{b}^{\downarrow,{(k)}})|}{|\Upsilon(\widetilde{\bm{\tau}}^{\downarrow,{(k)}})| \ell^{(k)}_{u}}\sum_{\ell=1}^{\ell^{(k)}_u}  {\nabla  \mathfrak{L}^{(k)}_u(\bm{\omega}^{(k),\ell-1}_{u})}\right\rangle\right]\nonumber\\
    &~~~~+ \eta_{_k}\mathfrak{B}_k \mathbb{E}_k\left[\left\langle\nabla{\mathfrak{L}^{({k})}(\bm{\omega}^{(k)})},-\sum_{b \in \Omega}\mathscr{U}^{(k)}_{b} \right\rangle\right]+ \frac{\beta\eta_{_k}^2\mathfrak{B}_k^2}{2} \sum_{b \in \Omega} \sum_{u\in \mathcal{U}_{b}}\mathbb{E}_k\Bigg[\left(\frac{\widehat{\lambda}_{u}^{(k)}|\Upsilon_{u}(\widetilde{\tau}_{b}^{\downarrow,{(k)}})|}{|\Upsilon^{\mathsf{s}}(\widetilde{\bm{\tau}}^{\downarrow,{(k)}})| \ell^{(k)}_{u}}\right)^2\Bigg\Vert \sum_{\ell=1}^{\ell^{(k)}_{u}}\mathscr{N}_u^{\mathsf{G},(k),\ell}\Bigg\Vert^2\Bigg]\nonumber\\
    &+\underbrace{\sum_{b \in \Omega} \sum_{u\in \mathcal{U}_{b}} \sum_{b' \in \Omega} \sum_{u'\in \mathcal{U}_{b'}\setminus u}\frac{\beta\eta_{_k}^2\mathfrak{B}_k^2}{2}\mathbb{E}_k\Bigg[\left\langle \frac{\widehat{\lambda}_{u}^{(k)}|\Upsilon_{u}(\widetilde{\tau}_{b}^{\downarrow,{(k)}})|}{|\Upsilon^{\mathsf{s}}(\widetilde{\bm{\tau}}^{\downarrow,{(k)}})| \ell^{(k)}_{u}}\sum_{\ell=1}^{\ell^{(k)}_{u}}\mathscr{N}_u^{\mathsf{G},(k),\ell}, \frac{\widehat{\lambda}_{u'}^{(k)}|\Upsilon_{u'}(\widetilde{\tau}_{b'}^{\downarrow,{(k)}})|}{|\Upsilon^{\mathsf{s}}(\widetilde{\bm{\tau}}^{\downarrow,{(k)}})| \ell^{(k)}_{u'}}\sum_{\ell=1}^{\ell^{(k)}_{u'}}\mathscr{N}_{u'}^{\mathsf{G},(k),\ell}\right\rangle\Bigg]}_{(x_3)}\nonumber\\
    &~~~~+ \frac{\beta\eta_{_k}^2\mathfrak{B}_k^2}{2} \mathbb{E}_k\Bigg[\Bigg\Vert \sum_{b \in \Omega} \sum_{u\in \mathcal{U}_{b}}\frac{\widehat{\lambda}_{u}^{(k)}|\Upsilon_{u}(\widetilde{\tau}_{b}^{\downarrow,{(k)}})|}{|\Upsilon^{\mathsf{s}}(\widetilde{\bm{\tau}}^{\downarrow,{(k)}})| \ell^{(k)}_{u}}\sum_{\ell=1}^{\ell^{(k)}_{u}} {\nabla  \mathfrak{L}^{(k)}_u(\bm{\omega}^{(k),\ell-1}_{u})}\Bigg\Vert^2\Bigg]\nonumber\\
    &\overset{(i)}{=}\mathfrak{L}^{({k})}(\bm{\omega}^{(k)}) - \eta_{_k}\mathfrak{B}_k \mathbb{E}_k\left[\left\langle\nabla{\mathfrak{L}^{({k})}(\bm{\omega}^{(k)})},\sum_{b \in \Omega} \sum_{u\in \mathcal{U}_{b}}\frac{|\Upsilon_{u}(\widetilde{\tau}_{b}^{\downarrow,{(k)}})|}{|\Upsilon(\widetilde{\bm{\tau}}^{\downarrow,{(k)}})| \ell^{(k)}_{u}}\sum_{\ell=1}^{\ell^{(k)}_u}  {\nabla  \mathfrak{L}^{(k)}_u(\bm{\omega}^{(k),\ell-1}_{u})}\right\rangle\right]\nonumber\\
    &~~~~+ \eta_{_k}\mathfrak{B}_k \underbrace{\mathbb{E}_k\left[\left\langle\nabla{\mathfrak{L}^{({k})}(\bm{\omega}^{(k)})},-\sum_{b \in \Omega}\mathscr{U}^{(k)}_{b} \right\rangle\right]}_{(a)}+ \frac{\beta\eta_{_k}^2\mathfrak{B}_k^2}{2} \sum_{b \in \Omega} \sum_{u\in \mathcal{U}_{b}}\mathbb{E}_k\Bigg[\left(\frac{\widehat{\lambda}_{u}^{(k)}|\Upsilon_{u}(\widetilde{\tau}_{b}^{\downarrow,{(k)}})|}{|\Upsilon^{\mathsf{s}}(\widetilde{\bm{\tau}}^{\downarrow,{(k)}})| \ell^{(k)}_{u}}\right)^2\Bigg\Vert \sum_{\ell=1}^{\ell^{(k)}_{u}}\mathscr{N}_u^{\mathsf{G},(k),\ell}\Bigg\Vert^2\Bigg]\nonumber\\
    &~~~~+ \frac{\beta\eta_{_k}^2\mathfrak{B}_k^2}{2} \mathbb{E}_k\Bigg[\Bigg\Vert \underbrace{\sum_{b \in \Omega} \sum_{u\in \mathcal{U}_{b}}\frac{\widehat{\lambda}_{u}^{(k)}|\Upsilon_{u}(\widetilde{\tau}_{b}^{\downarrow,{(k)}})|}{|\Upsilon^{\mathsf{s}}(\widetilde{\bm{\tau}}^{\downarrow,{(k)}})| \ell^{(k)}_{u}}\sum_{\ell=1}^{\ell^{(k)}_{u}} {\nabla  \mathfrak{L}^{(k)}_u(\bm{\omega}^{(k),\ell-1}_{u})}}_{(b)}\Bigg\Vert^2\Bigg],
\end{align}
where $(i)$ uses the fact that the noise of gradient estimation of a specific user is independent and zero mean. Using Cauchy-Schwartz and Young's inequalities, for two real valued vectors $\bm{a}$ and $\bm{b}$, we have
\begin{equation}
    \bm{a}^\top\bm{b}\le \frac{\alpha}{2}\Vert \bm{a}\Vert^2+\frac{1}{2\alpha}\Vert \bm{b}\Vert^2,~ \alpha \in \mathbb{R}^{++},
\end{equation}
implying $\bm{a}^\top\bm{b}\le \Vert \bm{a}\Vert^2+\frac{1}{4}\Vert \bm{b}\Vert^2$. Using this inequality to expand $(a)$ in \eqref{ineq:main3} and also rewriting $(b)$ in \eqref{ineq:main3} in terms of $\mathscr{U}^{(k)}_{b}$, we can further bound \eqref{ineq:main3} as follows:
\begin{align}\label{ineq:main3_1}
    \mathbb{E}_k&\left[\mathfrak{L}^{({k})}(\bm{\omega}^{(k+1)})\right] \leq \mathfrak{L}^{({k})}(\bm{\omega}^{(k)}) - \eta_{_k}\mathfrak{B}_k \mathbb{E}_k\left[\Bigg\langle\nabla{\mathfrak{L}^{({k})}(\bm{\omega}^{(k)})},\sum_{b \in \Omega} \sum_{u\in \mathcal{U}_{b}}\frac{|\Upsilon_{u}(\widetilde{\tau}_{b}^{\downarrow,{(k)}})|}{|\Upsilon(\widetilde{\bm{\tau}}^{\downarrow,{(k)}})| \ell^{(k)}_{u}}\sum_{\ell=1}^{\ell^{(k)}_u}  {\nabla  \mathfrak{L}^{(k)}_u(\bm{\omega}^{(k),\ell-1}_{u})}\Bigg\rangle\right]\nonumber\\
    &+ \frac{\beta\eta_{_k}^2\mathfrak{B}_k^2}{2} \sum_{b \in \Omega} \sum_{u\in \mathcal{U}_{b}}\left(\frac{\widehat{\lambda}_{u}^{(k)}|\Upsilon_{u}(\widetilde{\tau}_{b}^{\downarrow,{(k)}})|}{|\Upsilon^{\mathsf{s}}(\widetilde{\bm{\tau}}^{\downarrow,{(k)}})| \ell^{(k)}_{u}}\right)^2\mathbb{E}_k\Bigg[\Bigg\Vert \sum_{\ell=1}^{\ell^{(k)}_{u}}\mathscr{N}_u^{\mathsf{G},(k),\ell}\Bigg\Vert^2\Bigg]\nonumber\\
    &+ \frac{\beta\eta_{_k}^2\mathfrak{B}_k^2}{2} \mathbb{E}_k\Bigg[\Bigg\Vert \sum_{b \in \Omega} \sum_{u\in \mathcal{U}_{b}}\frac{|\Upsilon_{u}(\widetilde{\tau}_{b}^{\downarrow,{(k)}})|}{|\Upsilon(\widetilde{\bm{\tau}}^{\downarrow,{(k)}})| \ell^{(k)}_{u}}\sum_{\ell=1}^{\ell^{(k)}_{u}} {\nabla  \mathfrak{L}^{(k)}_u(\bm{\omega}^{(k),\ell-1}_{u})}+\sum_{b \in \Omega}\mathscr{U}^{(k)}_{b}\Bigg\Vert^2\Bigg]+ \eta_{_k}\mathfrak{B}_k \mathbb{E}_k\left[\left\langle\nabla{\mathfrak{L}^{({k})}(\bm{\omega}^{(k)})},-\sum_{b \in \Omega}\mathscr{U}^{(k)}_{b} \right\rangle\right]\nonumber\\
    &\le \mathfrak{L}^{({k})}(\bm{\omega}^{(k)}) - \eta_{_k}\mathfrak{B}_k \mathbb{E}_k\left[\Bigg\langle\nabla{\mathfrak{L}^{({k})}(\bm{\omega}^{(k)})},\sum_{b \in \Omega} \sum_{u\in \mathcal{U}_{b}}\frac{|\Upsilon_{u}(\widetilde{\tau}_{b}^{\downarrow,{(k)}})|}{|\Upsilon(\widetilde{\bm{\tau}}^{\downarrow,{(k)}})| \ell^{(k)}_{u}}\sum_{\ell=1}^{\ell^{(k)}_u}  {\nabla  \mathfrak{L}^{(k)}_u(\bm{\omega}^{(k),\ell-1}_{u})}\Bigg\rangle\right]\nonumber\\
    &+ \frac{\beta\eta_{_k}^2\mathfrak{B}_k^2}{2} \sum_{b \in \Omega} \sum_{u\in \mathcal{U}_{b}}\left(\frac{\widehat{\lambda}_{u}^{(k)}|\Upsilon_{u}(\widetilde{\tau}_{b}^{\downarrow,{(k)}})|}{|\Upsilon^{\mathsf{s}}(\widetilde{\bm{\tau}}^{\downarrow,{(k)}})| \ell^{(k)}_{u}}\right)^2\mathbb{E}_k\Bigg[\Bigg\Vert \sum_{\ell=1}^{\ell^{(k)}_{u}}\mathscr{N}_u^{\mathsf{G},(k),\ell}\Bigg\Vert^2\Bigg]\nonumber\\
    &+ \frac{\beta\eta_{_k}^2\mathfrak{B}_k^2}{2} \mathbb{E}_k\Bigg[\underbrace{\Bigg\Vert \sum_{b \in \Omega} \sum_{u\in \mathcal{U}_{b}}\frac{|\Upsilon_{u}(\widetilde{\tau}_{b}^{\downarrow,{(k)}})|}{|\Upsilon(\widetilde{\bm{\tau}}^{\downarrow,{(k)}})| \ell^{(k)}_{u}}\sum_{\ell=1}^{\ell^{(k)}_{u}} {\nabla  \mathfrak{L}^{(k)}_u(\bm{\omega}^{(k),\ell-1}_{u})}+\sum_{b \in \Omega}\mathscr{U}^{(k)}_{b}\Bigg\Vert^2}_{(a)}\Bigg]\nonumber\\
    &~~~~~~~~~~~~~~~~~~~~~~~~~~~~~+ \frac{\eta_{_k}\mathfrak{B}_k}{4} \mathbb{E}_k\left[\left\Vert\nabla{\mathfrak{L}^{({k})}(\bm{\omega}^{(k)})} \right\Vert^2\right]+ \eta_{_k}\mathfrak{B}_k \mathbb{E}_k\left[\left\Vert\sum_{b \in \Omega}\mathscr{U}^{(k)}_{b} \right\Vert^2\right]\nonumber\\
    &\overset{(i)}{\le}\mathfrak{L}^{({k})}(\bm{\omega}^{(k)}) - \eta_{_k} \mathfrak{B}_k\mathbb{E}_k\left[\underbrace{\left\langle\nabla{\mathfrak{L}^{({k})}(\bm{\omega}^{(k)})},\sum_{b \in \Omega} \sum_{u\in \mathcal{U}_{b}}\frac{|\Upsilon_{u}(\widetilde{\tau}_{b}^{\downarrow,{(k)}})|}{|\Upsilon(\widetilde{\bm{\tau}}^{\downarrow,{(k)}})| \ell^{(k)}_{u}}\sum_{\ell=1}^{\ell^{(k)}_u}  {\nabla  \mathfrak{L}^{(k)}_u(\bm{\omega}^{(k),\ell-1}_{u})}\right\rangle}_{(b)}\right]\nonumber\\
    &+ \frac{\beta\eta_{_k}^2\mathfrak{B}_k^2}{2} \sum_{b \in \Omega} \sum_{u\in \mathcal{U}_{b}}\left(\frac{\widehat{\lambda}_{u}^{(k)}|\Upsilon_{u}(\widetilde{\tau}_{b}^{\downarrow,{(k)}})|}{|\Upsilon^{\mathsf{s}}(\widetilde{\bm{\tau}}^{\downarrow,{(k)}})| \ell^{(k)}_{u}}\right)^2\mathbb{E}_k\Bigg[\Bigg\Vert \sum_{\ell=1}^{\ell^{(k)}_{u}}\mathscr{N}_u^{\mathsf{G},(k),\ell}\Bigg\Vert^2\Bigg]+ \beta\eta_{_k}^2 \mathfrak{B}_k^2\mathbb{E}_k\Bigg[\Bigg\Vert \sum_{b \in \Omega} \sum_{u\in \mathcal{U}_{b}}\frac{|\Upsilon_{u}(\widetilde{\tau}_{b}^{\downarrow,{(k)}})|}{|\Upsilon(\widetilde{\bm{\tau}}^{\downarrow,{(k)}})| \ell^{(k)}_{u}}\sum_{\ell=1}^{\ell^{(k)}_{u}} {\nabla  \mathfrak{L}^{(k)}_u(\bm{\omega}^{(k),\ell-1}_{u})}\Bigg\Vert^2\Bigg]\nonumber\\
    &+ \beta\eta_{_k}^2 \mathfrak{B}_k^2\mathbb{E}_k\left[\left\Vert\sum_{b \in \Omega}\mathscr{U}^{(k)}_{b}\right\Vert^2\right]+ \frac{\eta_{_k}\mathfrak{B}_k}{4} \mathbb{E}_k\left[\left\Vert\nabla{\mathfrak{L}^{({k})}(\bm{\omega}^{(k)})} \right\Vert^2\right]+ \eta_{_k} \mathfrak{B}_k\mathbb{E}_k\left[\left\Vert\sum_{b \in \Omega}\mathscr{U}^{(k)}_{b} \right\Vert^2\right],
\end{align}
where in inequality $(i)$ we have used the Cauchy-Schwarz inequality $\Vert \mathbf{a}+\mathbf{b} \Vert^2\leq 2 \Vert \mathbf{a} \Vert^2+2\Vert \mathbf{b} \Vert^2$ to bound $(a)$ in \eqref{ineq:main3_1}. Using $2\bm{a}^\top\bm{b}=\Vert \bm{a}\Vert^2+\Vert \bm{b}\Vert^2-\Vert \bm{a}-\bm{b}\Vert^2$ to expand $(b)$ in \eqref{ineq:main3_1} and performing some algebraic manipulation, we simplify the above inequality as follows:
\begin{align}
    \mathbb{E}_k&\left[\mathfrak{L}^{({k})}(\bm{\omega}^{(k+1)})\right] \leq \mathfrak{L}^{({k})}(\bm{\omega}^{(k)}) - \frac{\eta_{_k}\mathfrak{B}_k}{2} \mathbb{E}_k\Bigg[\Bigg\Vert\nabla{\mathfrak{L}^{({k})}(\bm{\omega}^{(k)})}\Bigg\Vert^2+\Bigg\Vert\sum_{b \in \Omega} \sum_{u\in \mathcal{U}_{b}}\frac{|\Upsilon_{u}(\widetilde{\tau}_{b}^{\downarrow,{(k)}})|}{|\Upsilon(\widetilde{\bm{\tau}}^{\downarrow,{(k)}})| \ell^{(k)}_{u}}\sum_{\ell=1}^{\ell^{(k)}_u}  {\nabla  \mathfrak{L}^{(k)}_u(\bm{\omega}^{(k),\ell-1}_{u})}\Bigg\Vert^2\nonumber\\
    &~~~~~~~~~~~~~~-\Bigg\Vert \nabla{\mathfrak{L}^{({k})}(\bm{\omega}^{(k)})}- \sum_{b \in \Omega} \sum_{u\in \mathcal{U}_{b}}\frac{|\Upsilon_{u}(\widetilde{\tau}_{b}^{\downarrow,{(k)}})|}{|\Upsilon(\widetilde{\bm{\tau}}^{\downarrow,{(k)}})| \ell^{(k)}_{u}}\sum_{\ell=1}^{\ell^{(k)}_u}  {\nabla  \mathfrak{L}^{(k)}_u(\bm{\omega}^{(k),\ell-1}_{u})}\Bigg\Vert^2\Bigg]\nonumber\\
    &+ \frac{\beta\eta_{_k}^2}{2} \sum_{b \in \Omega} \sum_{u\in \mathcal{U}_{b}}\left(\frac{\widehat{\lambda}_{u}^{(k)}|\Upsilon_{u}(\widetilde{\tau}_{b}^{\downarrow,{(k)}})|}{|\Upsilon^{\mathsf{s}}(\widetilde{\bm{\tau}}^{\downarrow,{(k)}})| \ell^{(k)}_{u}}\right)^2\mathbb{E}_k\Bigg[\Bigg\Vert \sum_{\ell=1}^{\ell^{(k)}_{u}}\mathscr{N}_u^{\mathsf{G},(k),\ell}\Bigg\Vert^2\Bigg]+ \beta\eta_{_k}^2\mathfrak{B}_k^2 \mathbb{E}_k\Bigg[\Bigg\Vert \sum_{b \in \Omega} \sum_{u\in \mathcal{U}_{b}}\frac{|\Upsilon_{u}(\widetilde{\tau}_{b}^{\downarrow,{(k)}})|}{|\Upsilon(\widetilde{\bm{\tau}}^{\downarrow,{(k)}})| \ell^{(k)}_{u}}\sum_{\ell=1}^{\ell^{(k)}_{u}} {\nabla  \mathfrak{L}^{(k)}_u(\bm{\omega}^{(k),\ell-1}_{u})}\Bigg\Vert^2\Bigg]\nonumber\\
    &+ \beta\eta_{_k}^2 \mathfrak{B}_k^2\mathbb{E}_k\left[\left\Vert\sum_{b \in \Omega}\mathscr{U}^{(k)}_{b}\right\Vert^2\right]+ \frac{\eta_{_k}\mathfrak{B}_k}{4} \mathbb{E}_k\left[\left\Vert\nabla{\mathfrak{L}^{({k})}(\bm{\omega}^{(k)})} \right\Vert^2\right]+ \eta_{_k}\mathfrak{B}_k  \mathbb{E}_k\left[\left\Vert\sum_{b \in \Omega}\mathscr{U}^{(k)}_{b} \right\Vert^2\right]\nonumber\\
    &=\mathfrak{L}^{({k})}(\bm{\omega}^{(k)}) - \frac{\eta_{_k}\mathfrak{B}_k}{2} \mathbb{E}_k\Bigg[\Bigg\Vert\nabla{\mathfrak{L}^{({k})}(\bm{\omega}^{(k)})}\Bigg\Vert^2\Bigg]- \frac{\eta_{_k}\mathfrak{B}_k}{2} \mathbb{E}_k\Bigg[\Bigg\Vert\sum_{b \in \Omega} \sum_{u\in \mathcal{U}_{b}}\frac{|\Upsilon_{u}(\widetilde{\tau}_{b}^{\downarrow,{(k)}})|}{|\Upsilon(\widetilde{\bm{\tau}}^{\downarrow,{(k)}})| \ell^{(k)}_{u}}\sum_{\ell=1}^{\ell^{(k)}_u}  {\nabla  \mathfrak{L}^{(k)}_u(\bm{\omega}^{(k),\ell-1}_{u})}\Bigg\Vert^2\Bigg]\nonumber\\
    &~~~~~~~~~~~~~~+ \frac{\eta_{_k}\mathfrak{B}_k}{2} \mathbb{E}_k\Bigg[\Bigg\Vert \nabla{\mathfrak{L}^{({k})}(\bm{\omega}^{(k)})}- \sum_{b \in \Omega} \sum_{u\in \mathcal{U}_{b}}\frac{|\Upsilon_{u}(\widetilde{\tau}_{b}^{\downarrow,{(k)}})|}{|\Upsilon(\widetilde{\bm{\tau}}^{\downarrow,{(k)}})| \ell^{(k)}_{u}}\sum_{\ell=1}^{\ell^{(k)}_u}  {\nabla  \mathfrak{L}^{(k)}_u(\bm{\omega}^{(k),\ell-1}_{u})}\Bigg\Vert^2\Bigg]\nonumber\\
    &+ \frac{\beta\eta_{_k}^2\mathfrak{B}_k^2}{2} \sum_{b \in \Omega} \sum_{u\in \mathcal{U}_{b}}\left(\frac{\widehat{\lambda}_{u}^{(k)}|\Upsilon_{u}(\widetilde{\tau}_{b}^{\downarrow,{(k)}})|}{|\Upsilon^{\mathsf{s}}(\widetilde{\bm{\tau}}^{\downarrow,{(k)}})| \ell^{(k)}_{u}}\right)^2\mathbb{E}_k\Bigg[\Bigg\Vert \sum_{\ell=1}^{\ell^{(k)}_{u}}\mathscr{N}_u^{\mathsf{G},(k),\ell}\Bigg\Vert^2\Bigg]+ \beta\eta_{_k}^2\mathfrak{B}_k^2 \mathbb{E}_k\Bigg[\Bigg\Vert \sum_{b \in \Omega} \sum_{u\in \mathcal{U}_{b}}\frac{|\Upsilon_{u}(\widetilde{\tau}_{b}^{\downarrow,{(k)}})|}{|\Upsilon(\widetilde{\bm{\tau}}^{\downarrow,{(k)}})| \ell^{(k)}_{u}}\sum_{\ell=1}^{\ell^{(k)}_{u}} {\nabla  \mathfrak{L}^{(k)}_u(\bm{\omega}^{(k),\ell-1}_{u})}\Bigg\Vert^2\Bigg]\nonumber\\
    &+ \beta\eta_{_k}^2\mathfrak{B}_k^2 \mathbb{E}_k\left[\left\Vert\sum_{b \in \Omega}\mathscr{U}^{(k)}_{b}\right\Vert^2\right]+ \frac{\eta_{_k}\mathfrak{B}_k}{4} \mathbb{E}_k\left[\left\Vert\nabla{\mathfrak{L}^{({k})}(\bm{\omega}^{(k)})} \right\Vert^2\right]+ \eta_{_k} \mathfrak{B}_k \mathbb{E}_k\left[\left\Vert\sum_{b \in \Omega}\mathscr{U}^{(k)}_{b} \right\Vert^2\right]\nonumber\\
    &=\mathfrak{L}^{({k})}(\bm{\omega}^{(k)}) - \frac{\eta_{_k}\mathfrak{B}_k}{4} \mathbb{E}_k\Bigg[\Bigg\Vert\nabla{\mathfrak{L}^{({k})}(\bm{\omega}^{(k)})}\Bigg\Vert^2\Bigg]+\underbrace{\eta_{_k}\left(\beta\eta_{_k}-\frac{1}{2}\right)\mathfrak{B}_k \mathbb{E}_k\Bigg[\Bigg\Vert\sum_{b \in \Omega} \sum_{u\in \mathcal{U}_{b}}\frac{|\Upsilon_{u}(\widetilde{\tau}_{b}^{\downarrow,{(k)}})|}{|\Upsilon(\widetilde{\bm{\tau}}^{\downarrow,{(k)}})| \ell^{(k)}_{u}}\sum_{\ell=1}^{\ell^{(k)}_u}  {\nabla  \mathfrak{L}^{(k)}_u(\bm{\omega}^{(k),\ell-1}_{u})}\Bigg\Vert^2\Bigg]}_{(a)}\nonumber\\
    &~~~~~~~~~~~~~~+ \frac{\eta_{_k}\mathfrak{B}_k}{2} \mathbb{E}_k\Bigg[\Bigg\Vert \nabla{\mathfrak{L}^{({k})}(\bm{\omega}^{(k)})}- \sum_{b \in \Omega} \sum_{u\in \mathcal{U}_{b}}\frac{|\Upsilon_{u}(\widetilde{\tau}_{b}^{\downarrow,{(k)}})|}{|\Upsilon(\widetilde{\bm{\tau}}^{\downarrow,{(k)}})| \ell^{(k)}_{u}}\sum_{\ell=1}^{\ell^{(k)}_u}  {\nabla  \mathfrak{L}^{(k)}_u(\bm{\omega}^{(k),\ell-1}_{u})}\Bigg\Vert^2\Bigg]\nonumber\\
    &+ \frac{\beta\eta_{_k}^2\mathfrak{B}_k^2}{2} \sum_{b \in \Omega} \sum_{u\in \mathcal{U}_{b}}\left(\frac{\widehat{\lambda}_{u}^{(k)}|\Upsilon_{u}(\widetilde{\tau}_{b}^{\downarrow,{(k)}})|}{|\Upsilon^{\mathsf{s}}(\widetilde{\bm{\tau}}^{\downarrow,{(k)}})| \ell^{(k)}_{u}}\right)^2\mathbb{E}_k\Bigg[\Bigg\Vert \sum_{\ell=1}^{\ell^{(k)}_{u}}\mathscr{N}_u^{\mathsf{G},(k),\ell}\Bigg\Vert^2\Bigg]+ \eta_{_k}\left(1+\beta\eta_{_k}\right)\mathfrak{B}_k \mathbb{E}_k\left[\left\Vert\sum_{b \in \Omega}\mathscr{U}^{(k)}_{b}\right\Vert^2\right].
\end{align}
Assuming 
\begin{equation}\label{ineq:eta_cond1}
 \eta_{_k}\le\frac{1}{2\beta}   
\end{equation}
makes $(a)$ in the above expression negative and thus can be removed from the bound. Moreover, $\eta_{_k}\le\frac{1}{2\beta}$ implies $1+\beta\eta_{_k}\le\frac{3}{2}$. Applying these result to the above bound gives us
\begin{align}\label{ineq:main4_1}
    \mathbb{E}_k&\left[\mathfrak{L}^{({k})}(\bm{\omega}^{(k+1)})\right] \leq \mathfrak{L}^{({k})}(\bm{\omega}^{(k)}) - \frac{\eta_{_k}\mathfrak{B}_k}{4} \mathbb{E}_k\Bigg[\Bigg\Vert\nabla{\mathfrak{L}^{({k})}(\bm{\omega}^{(k)})}\Bigg\Vert^2\Bigg]\nonumber\\
    &~~~~~~~~~~~~~~+ \frac{\eta_{_k}\mathfrak{B}_k}{2} \mathbb{E}_k\Bigg[\Bigg\Vert \nabla{\mathfrak{L}^{({k})}(\bm{\omega}^{(k)})}- \sum_{b \in \Omega} \sum_{u\in \mathcal{U}_{b}}\frac{|\Upsilon_{u}(\widetilde{\tau}_{b}^{\downarrow,{(k)}})|}{|\Upsilon(\widetilde{\bm{\tau}}^{\downarrow,{(k)}})| \ell^{(k)}_{u}}\sum_{\ell=1}^{\ell^{(k)}_u}  {\nabla  \mathfrak{L}^{(k)}_u(\bm{\omega}^{(k),\ell-1}_{u})}\Bigg\Vert^2\Bigg]\nonumber\\
    &+ \frac{\beta\eta_{_k}^2\mathfrak{B}_k^2}{2} \sum_{b \in \Omega} \sum_{u\in \mathcal{U}_{b}}\left(\frac{\widehat{\lambda}_{u}^{(k)}|\Upsilon_{u}(\widetilde{\tau}_{b}^{\downarrow,{(k)}})|}{|\Upsilon^{\mathsf{s}}(\widetilde{\bm{\tau}}^{\downarrow,{(k)}})| \ell^{(k)}_{u}}\right)^2\mathbb{E}_k\Bigg[\underbrace{\Bigg\Vert \sum_{\ell=1}^{\ell^{(k)}_{u}}\mathscr{N}_u^{\mathsf{G},(k),\ell}\Bigg\Vert^2}_{(a)}\Bigg]+ \frac{3}{2}\eta_{_k}\mathfrak{B}_k \mathbb{E}_k\left[\left\Vert\sum_{b \in \Omega}\mathscr{U}^{(k)}_{b}\right\Vert^2\right].
\end{align}
Using $\Vert \sum_{i=1}^{n}\bm{x}_i\Vert^2=\sum_{i=1}^{n}\Vert\bm{x}_i\Vert^2+\sum_{\substack{j=1\\j\neq i}}^{n}\left\langle\bm{x}_i,\bm{x}_j\right\rangle$ to expand $(a)$ in \eqref{ineq:main4_1} and performing some algebraic manipulation, we rewrite \eqref{ineq:main4_1} as follows: 
\begin{align}\label{ineq:main4}
    \mathbb{E}_k&\left[\mathfrak{L}^{({k})}(\bm{\omega}^{(k+1)})\right] \leq \mathfrak{L}^{({k})}(\bm{\omega}^{(k)}) - \frac{\eta_{_k}\mathfrak{B}_k}{4} \mathbb{E}_k\Bigg[\Bigg\Vert\nabla{\mathfrak{L}^{({k})}(\bm{\omega}^{(k)})}\Bigg\Vert^2\Bigg]\nonumber\\
    &~~~~~~~~~~~~~~+ \frac{\eta_{_k}\mathfrak{B}_k}{2} \mathbb{E}_k\Bigg[\Bigg\Vert \nabla{\mathfrak{L}^{({k})}(\bm{\omega}^{(k)})}- \sum_{b \in \Omega} \sum_{u\in \mathcal{U}_{b}}\frac{|\Upsilon_{u}(\widetilde{\tau}_{b}^{\downarrow,{(k)}})|}{|\Upsilon(\widetilde{\bm{\tau}}^{\downarrow,{(k)}})| \ell^{(k)}_{u}}\sum_{\ell=1}^{\ell^{(k)}_u}  {\nabla  \mathfrak{L}^{(k)}_u(\bm{\omega}^{(k),\ell-1}_{u})}\Bigg\Vert^2\Bigg]+ \frac{3}{2}\eta_{_k}\mathfrak{B}_k \mathbb{E}_k\left[\left\Vert\sum_{b \in \Omega}\mathscr{U}^{(k)}_{b}\right\Vert^2\right]\nonumber\\
    &+ \frac{\beta\eta_{_k}^2\mathfrak{B}_k^2}{2} \sum_{b \in \Omega} \sum_{u\in \mathcal{U}_{b}}\left(\frac{\widehat{\lambda}_{u}^{(k)}|\Upsilon_{u}(\widetilde{\tau}_{b}^{\downarrow,{(k)}})|}{|\Upsilon^{\mathsf{s}}(\widetilde{\bm{\tau}}^{\downarrow,{(k)}})| \ell^{(k)}_{u}}\right)^2\left(\mathbb{E}_k\Bigg[\sum_{\ell=1}^{\ell^{(k)}_{u}}\Bigg\Vert\mathscr{N}_u^{\mathsf{G},(k),\ell}\Bigg\Vert^2\Bigg]+\underbrace{\mathbb{E}_k\Bigg[\sum_{\ell=1}^{\ell^{(k)}_{u}}\sum_{\substack{\ell'=1,\\\ell'\neq \ell}}^{\ell^{(k)}_{u}}\left\langle\mathscr{N}_u^{\mathsf{G},(k),\ell},\mathscr{N}_u^{\mathsf{G},(k),\ell'}\right\rangle\Bigg]}_{(x_4)}\right)\nonumber\\
    &\overset{(i)}{=}\mathfrak{L}^{({k})}(\bm{\omega}^{(k)}) - \frac{\eta_{_k}\mathfrak{B}_k}{4} \mathbb{E}_k\Bigg[\Bigg\Vert\nabla{\mathfrak{L}^{({k})}(\bm{\omega}^{(k)})}\Bigg\Vert^2\Bigg]\nonumber\\
    &~~~~~~~~~~~~~~+ \frac{\eta_{_k}\mathfrak{B}_k}{2} \mathbb{E}_k\Bigg[\Bigg\Vert \nabla{\mathfrak{L}^{({k})}(\bm{\omega}^{(k)})}- \sum_{b \in \Omega} \sum_{u\in \mathcal{U}_{b}}\frac{|\Upsilon_{u}(\widetilde{\tau}_{b}^{\downarrow,{(k)}})|}{|\Upsilon(\widetilde{\bm{\tau}}^{\downarrow,{(k)}})| \ell^{(k)}_{u}}\sum_{\ell=1}^{\ell^{(k)}_u}  {\nabla  \mathfrak{L}^{(k)}_u(\bm{\omega}^{(k),\ell-1}_{u})}\Bigg\Vert^2\Bigg]+ \frac{3}{2}\eta_{_k}\mathfrak{B}_k \mathbb{E}_k\left[\left\Vert\sum_{b \in \Omega}\mathscr{U}^{(k)}_{b}\right\Vert^2\right]\nonumber\\
    &+ \frac{\beta\eta_{_k}^2\mathfrak{B}_k^2}{2} \sum_{b \in \Omega} \sum_{u\in \mathcal{U}_{b}}\left(\frac{\widehat{\lambda}_{u}^{(k)}|\Upsilon_{u}(\widetilde{\tau}_{b}^{\downarrow,{(k)}})|}{|\Upsilon^{\mathsf{s}}(\widetilde{\bm{\tau}}^{\downarrow,{(k)}})| \ell^{(k)}_{u}}\right)^2\sum_{\ell=1}^{\ell^{(k)}_{u}}\mathbb{E}_k\Bigg[\Bigg\Vert\mathscr{N}_u^{\mathsf{G},(k),\ell}\Bigg\Vert^2\Bigg],
\end{align}
where $(i)$ uses the fact that $(x_4)$ is zero since the noise of gradient estimation of a specific user is independent and zero mean. We next aim to simplify $\mathbb{E}_k\left[\left\Vert\mathscr{N}^{\mathsf{G},(k)}_{u}\right\Vert^2\right]$. Considering the result of~\cite{lohr2019sampling} (Chapter 3, Eq. (3.5)), we have 
\begin{align}\label{eq:varGrad}
    \mathbb{E}_k\left[\left\Vert\mathscr{N}^{\mathsf{G},(k)}_{u}\right\Vert^2\right]=\mathbb{E}_k\left[\left\Vert \sum_{\xi\in \mathcal{B}^{(k),\ell}_{u}} \hspace{-3mm} {\frac{\nabla  f_u(\bm{\omega}^{(k),\ell-1}_{u},\xi)}{{B}_{u}(\widetilde{\tau}_{b}^{\downarrow,{(k)}})}} - \sum_{\xi\in\Upsilon_{u}(\widetilde{\tau}_{b}^{\downarrow,{(k)}})} \hspace{-3mm} {\frac{\nabla  f_u(\bm{\omega}^{(k),\ell-1}_{u},\xi)}{|\Upsilon_{u}(\widetilde{\tau}_{b}^{\downarrow,{(k)}})|}}\right\Vert^2\right]=\left(1-\frac{{B}_{u}(\widetilde{\tau}_{b}^{\downarrow,{(k)}})}{|\Upsilon_{u}(\widetilde{\tau}_{b}^{\downarrow,{(k)}})|} \right) \frac{\left(\sigma_{u}^{\ell-1}(\widetilde{\tau}_{b}^{\downarrow,{(k)}})\right)^2}{{B}_{u}(\widetilde{\tau}_{b}^{\downarrow,{(k)}})},
\end{align}
where $\sigma_{u}^{\ell-1}(\widetilde{\tau}_{b}^{\downarrow,{(k)}})$ denotes the \textit{variance of the gradients} evaluated at the particular local descent iteration $\ell-1$ for the parameter realization $\bm{\omega}^{(k),\ell-1}_{u}$, and $\left(\sigma_{u}^{\ell{-}1}(\widetilde{\tau}_{b}^{\downarrow,{(k)}})\right)^2$ is calculated as follows:
\begin{align}\label{eq:dataVar0}
     \left(\sigma_{u}^{\ell{-}1}(\widetilde{\tau}_{b}^{\downarrow,{(k)}})\right)^2&= \frac{\sum_{\xi\in\Upsilon_{u}(\widetilde{\tau}_{b}^{\downarrow,{(k)}}) } \Big\Vert \nabla  f_u(\bm{\omega}^{(k),\ell{-}1}_{u},\xi){-}{\sum_{\tilde{\xi}\in\Upsilon_{u}(\widetilde{\tau}_{b}^{\downarrow,{(k)}})} }\frac{\nabla  f_u(\bm{\omega}^{(k),\ell{-}1}_{u},\tilde{\xi})}{|\Upsilon_{u}(\widetilde{\tau}_{b}^{\downarrow,{(k)}})|}\Big\Vert^2}{|\Upsilon_{u}(\widetilde{\tau}_{b}^{\downarrow,{(k)}})|{-}1}
     \nonumber\\
     &=\frac{\sum_{\xi\in\Upsilon_{u}(\widetilde{\tau}_{b}^{\downarrow,{(k)}}) }\frac{1}{\left(|\Upsilon_{u}(\widetilde{\tau}_{b}^{\downarrow,{(k)}})|\right)^2} \Big\Vert |\Upsilon_{u}(\widetilde{\tau}_{b}^{\downarrow,{(k)}})|\nabla  f_u(\bm{\omega}^{(D),\ell{-}1}_{u},\xi){-}{\sum_{\tilde{\xi}\in\Upsilon_{u}(\widetilde{\tau}_{b}^{\downarrow,{(k)}})} }{\nabla  f_u(\bm{\omega}^{(k),\ell{-}1}_{u},\tilde{\xi})}\Big\Vert^2}{|\Upsilon_{u}(\widetilde{\tau}_{b}^{\downarrow,{(k)}})|{-}1}.
\end{align}
Using the Cauchy-Schwarz inequality, we can bound \eqref{eq:dataVar0} as follows:
 \begin{align}\label{eq:dataVar}
     &\left(\sigma_{u}^{\ell{-}1}(\widetilde{\tau}_{b}^{\downarrow,{(k)}})\right)^2\leq \frac{\sum_{\xi\in\Upsilon_{u}(\widetilde{\tau}_{b}^{\downarrow,{(k)}}) }\frac{|\Upsilon_{u}(\widetilde{\tau}_{b}^{\downarrow,{(k)}})|{-}1}{\left(|\Upsilon_{u}(\widetilde{\tau}_{b}^{\downarrow,{(k)}})|\right)^2} \sum_{\tilde{\xi}\in\Upsilon_{u}(\widetilde{\tau}_{b}^{\downarrow,{(k)}})} \Big\Vert \nabla  f_u(\bm{\omega}^{(k),\ell{-}1}_{u},\xi){-}{\nabla  f_u(\bm{\omega}^{(k),\ell{-}1}_{u},\tilde{\xi})}\Big\Vert^2}{|\Upsilon_{u}(\widetilde{\tau}_{b}^{\downarrow,{(k)}})|{-}1}
     \nonumber\\&
     \leq \frac{\sum_{\xi\in\Upsilon_{u}(\widetilde{\tau}_{b}^{\downarrow,{(k)}}) }\frac{(|\Upsilon_{u}(\widetilde{\tau}_{b}^{\downarrow,{(k)}})|{-}1)\Theta^2}{\left(|\Upsilon_{u}(\widetilde{\tau}_{b}^{\downarrow,{(k)}})|\right)^2} \sum_{\tilde{\xi}\in\Upsilon_{u}(\widetilde{\tau}_{b}^{\downarrow,{(k)}})} \Big\Vert \bm{\xi}{-}\tilde{\bm{\xi}}\Big\Vert^2}{|\Upsilon_{u}(\widetilde{\tau}_{b}^{\downarrow,{(k)}})|{-}1}\nonumber\\
     &= \frac{(|\Upsilon_{u}(\widetilde{\tau}_{b}^{\downarrow,{(k)}})|{-}1)\Theta^2}{\left(|\Upsilon_{u}(\widetilde{\tau}_{b}^{\downarrow,{(k)}})|\right)^2}\frac{\sum_{\xi\in\Upsilon_{u}(\widetilde{\tau}_{b}^{\downarrow,{(k)}}) } \sum_{\tilde{\xi}\in\Upsilon_{u}(\widetilde{\tau}_{b}^{\downarrow,{(k)}})} \Big\Vert \bm{\xi}{-}\tilde{\bm{\xi}}{+}\bm{\mu}_{u}(\widetilde{\tau}_{b}^{\downarrow,{(k)}}){-}\bm{\mu}_{u}(\widetilde{\tau}_{b}^{\downarrow,{(k)}})\Big\Vert^2}{|\Upsilon_{u}(\widetilde{\tau}_{b}^{\downarrow,{(k)}})|{-}1}\nonumber\\
     &=\frac{(|\Upsilon_{u}(\widetilde{\tau}_{b}^{\downarrow,{(k)}})|{-}1)\Theta^2}{\left(|\Upsilon_{u}(\widetilde{\tau}_{b}^{\downarrow,{(k)}})|\right)^2}\frac{\displaystyle\sum_{\xi\in\Upsilon_{u}(\widetilde{\tau}_{b}^{\downarrow,{(k)}}) } \sum_{\tilde{\xi}\in\Upsilon_{u}(\widetilde{\tau}_{b}^{\downarrow,{(k)}})} \left[\Big\Vert \bm{\xi}{-}  \bm{\mu}_{u}(\widetilde{\tau}_{b}^{\downarrow,{(k)}}) \Big\Vert^2 {+} \Big\Vert \tilde{\bm{\xi}}{-}  \bm{\mu}_{u}(\widetilde{\tau}_{b}^{\downarrow,{(k)}})\Big\Vert^2 {-} 2\left\langle\bm{\xi}{-}  \bm{\mu}_{u}(\widetilde{\tau}_{b}^{\downarrow,{(k)}}),\tilde{\bm{\xi}}{-}  \bm{\mu}_{u}(\widetilde{\tau}_{b}^{\downarrow,{(k)}})\right\rangle \right]}{|\Upsilon_{u}(\widetilde{\tau}_{b}^{\downarrow,{(k)}})|{-}1}\nonumber\\
     &\overset{(ii)}{=}\frac{(|\Upsilon_{u}(\widetilde{\tau}_{b}^{\downarrow,{(k)}})|{-}1)\Theta^2}{\left(|\Upsilon_{u}(\widetilde{\tau}_{b}^{\downarrow,{(k)}})|\right)^2}\frac{ |\Upsilon_{u}(\widetilde{\tau}_{b}^{\downarrow,{(k)}})| \sum_{\xi\in\Upsilon_{u}(\widetilde{\tau}_{b}^{\downarrow,{(k)}}) } \Big\Vert \bm{\xi}{-}  \bm{\mu}_{u}(\widetilde{\tau}_{b}^{\downarrow,{(k)}}) \Big\Vert^2 {+}  |\Upsilon_{u}(\widetilde{\tau}_{b}^{\downarrow,{(k)}})| \sum_{\tilde{\xi}\in\Upsilon_{u}(\widetilde{\tau}_{b}^{\downarrow,{(k)}})} \Big\Vert \tilde{\bm{\xi}}{-}  \bm{\mu}_{u}(\widetilde{\tau}_{b}^{\downarrow,{(k)}})\Big\Vert^2}{|\Upsilon_{u}(\widetilde{\tau}_{b}^{\downarrow,{(k)}})|{-}1}\nonumber\\
     &=\frac{2(|\Upsilon_{u}(\widetilde{\tau}_{b}^{\downarrow,{(k)}})|{-}1)\Theta^2}{|\Upsilon_{u}(\widetilde{\tau}_{b}^{\downarrow,{(k)}})|}\left(\sigma_{u}(\widetilde{\tau}_{b}^{\downarrow,{(k)}})\right)^2, 
\end{align}
where $\bm{\mu}_{u}(\widetilde{\tau}_{b}^{\downarrow,{(k)}})$ and $\sigma_{u}(\widetilde{\tau}_{b}^{\downarrow,{(k)}})$ denote the mean and sample variance of data points in dataset $\mathcal|\Upsilon_{u}(\widetilde{\tau}_{b}^{\downarrow,{(k)}})|$, which are gradient independent. Further, $\Theta {=} \max_{b\in\Omega,u{\in}\mathcal{U}_{b}}\{\Theta_{u} \}$. Also, $\bm{\xi}$ refers to the feature vector of data point $\xi$. Furthermore, $(ii)$ used the fact that $\sum_{\xi\in\Upsilon_{u}(\widetilde{\tau}_{b}^{\downarrow,{(k)}}) } (\bm{\xi}-  \bm{\mu}_{u}(\widetilde{\tau}_{b}^{\downarrow,{(k)}})) =\bm{0}$. Replacing the above result in~\eqref{eq:varGrad}, inequality \eqref{ineq:main4} can be written as follows:
\begin{align}\label{ineq:main5}
    \mathbb{E}_k&\left[\mathfrak{L}^{({k})}(\bm{\omega}^{(k+1)})\right] \leq \mathfrak{L}^{({k})}(\bm{\omega}^{(k)}) - \frac{\eta_{_k}\mathfrak{B}_k}{4} \mathbb{E}_k\Bigg[\Bigg\Vert\nabla{\mathfrak{L}^{({k})}(\bm{\omega}^{(k)})}\Bigg\Vert^2\Bigg]\nonumber\\
    &~~~~+ \frac{\eta_{_k}\mathfrak{B}_k}{2} \underbrace{\mathbb{E}_k\Bigg[\Bigg\Vert \nabla{\mathfrak{L}^{({k})}(\bm{\omega}^{(k)})}- \sum_{b \in \Omega} \sum_{u\in \mathcal{U}_{b}}\frac{|\Upsilon_{u}(\widetilde{\tau}_{b}^{\downarrow,{(k)}})|}{|\Upsilon(\widetilde{\bm{\tau}}^{\downarrow,{(k)}})| \ell^{(k)}_{u}}\sum_{\ell=1}^{\ell^{(k)}_u}  {\nabla  \mathfrak{L}^{(k)}_u(\bm{\omega}^{(k),\ell-1}_{u})}\Bigg\Vert^2\Bigg]}_{(d)}+ \frac{3}{2}\eta_{_k}\mathfrak{B}_k \mathbb{E}_k\left[\left\Vert\sum_{b \in \Omega}\mathscr{U}^{(k)}_{b}\right\Vert^2\right]\nonumber\\
    &~~~~+ \frac{\beta\eta_{_k}^2\mathfrak{B}_k^2}{2} \sum_{b \in \Omega} \sum_{u\in \mathcal{U}_{b}}\frac{\left(\widehat{\lambda}_{u}^{(k)}|\Upsilon_{u}(\widetilde{\tau}_{b}^{\downarrow,{(k)}})|\right)^2}{\left(|\Upsilon^{\mathsf{s}}(\widetilde{\bm{\tau}}^{\downarrow,{(k)}})|\right)^2 \ell^{(k)}_{u}}\left(1-\frac{{B}_{u}(\widetilde{\tau}_{b}^{\downarrow,{(k)}})}{|\Upsilon_{u}(\widetilde{\tau}_{b}^{\downarrow,{(k)}})|} \right) \frac{2(|\Upsilon_{u}(\widetilde{\tau}_{b}^{\downarrow,{(k)}})|-1)\Theta^2\left(\sigma_{u}(\widetilde{\tau}_{b}^{\downarrow,{(k)}})\right)^2}{|\Upsilon_{u}(\widetilde{\tau}_{b}^{\downarrow,{(k)}})|{B}_{u}(\widetilde{\tau}_{b}^{\downarrow,{(k)}})}.
\end{align}
In the following, we aim to bound term $(d)$.
\begin{align}\label{eq:res3}
        (d) &=\mathbb{E}_k\left[\left\Vert \nabla{\mathfrak{L}^{(k)}(\bm{\omega}^{(k)})}-\sum_{b \in \Omega} \sum_{u\in \mathcal{U}_{b}}\frac{|\Upsilon_{u}(\widetilde{\tau}_{b}^{\downarrow,{(k)}})|}{|\Upsilon(\widetilde{\bm{\tau}}^{\downarrow,{(k)}})| \ell^{(k)}_u} \sum_{\ell=1}^{\ell^{(k)}_u}  {\nabla  \mathfrak{L}^{(k)}_u(\bm{\omega}^{(k),\ell-1}_{u})}\right\Vert^2\right]\nonumber\\
    &\overset{(i)}{\leq}   \sum_{b \in \Omega} \sum_{u\in \mathcal{U}_{b}}\frac{|\Upsilon_{u}(\widetilde{\tau}_{b}^{\downarrow,{(k)}})|}{|\Upsilon(\widetilde{\bm{\tau}}^{\downarrow,{(k)}})| }\mathbb{E}_k\left[\left\Vert \nabla{\mathfrak{L}^{(k)}_u(\bm{\omega}^{(k)})}-  \frac{1}{\ell^{(k)}_u} \sum_{\ell=1}^{\ell^{(k)}_u}  {\nabla  \mathfrak{L}^{(k)}_u(\bm{\omega}^{(k),\ell-1}_{u})} \right\Vert^2\right]\nonumber\\
    &\overset{(ii)}{\leq} \sum_{b \in \Omega} \sum_{u\in \mathcal{U}_{b}}\frac{|\Upsilon_{u}(\widetilde{\tau}_{b}^{\downarrow,{(k)}})|}{|\Upsilon(\widetilde{\bm{\tau}}^{\downarrow,{(k)}})| \ell_u^{(k)}}\sum_{\ell=1}^{\ell^{(k)}_u} \mathbb{E}_k\left[\left\Vert
   \nabla{\mathfrak{L}^{(k)}_u(\bm{\omega}^{(k)})}- {\nabla  \mathfrak{L}^{(k)}_u(\bm{\omega}^{(k),\ell-1}_{u})}\right\Vert^2\right] \nonumber\\
   &\overset{(iii)}{\leq} \beta^2\sum_{b \in \Omega} \sum_{u\in \mathcal{U}_{b}}\frac{|\Upsilon_{u}(\widetilde{\tau}_{b}^{\downarrow,{(k)}})|}{|\Upsilon(\widetilde{\bm{\tau}}^{\downarrow,{(k)}})| \ell_u^{(k)}}\sum_{\ell=1}^{\ell^{(k)}_u} \mathbb{E}_k\left[\left\Vert\bm{\omega}^{(k)}-\bm{\omega}^{(k),\ell-1}_{u}\right\Vert^2\right],
\end{align}
where in inequalities $(i)$ and $(ii)$ in~\eqref{eq:res3} we used Jenson's inequality. Inequality $(iii)$ is the result of Assumption \eqref{Assup:lossFun}. To bound $\mathbb{E}_k\left[\left\Vert\bm{\omega}^{(k)}-\bm{\omega}^{(k),\ell-1}_{u}\right\Vert^2\right]$, we take the following steps.
\begin{align}\label{eq:res4}
    &\mathbb{E}_k\left[\left\Vert\bm{\omega}^{(k)}-\bm{\omega}^{(k),\ell-1}_{u}\right\Vert^2\right]\overset{(i)}{=}\eta_k^2  \mathbb{E}_k\left[\left\Vert
    \sum_{\ell'=1}^{\ell-1}\sum_{\xi\in \mathcal{B}^{(k),\ell'}_{u}} \hspace{-1mm} {\frac{\nabla  f_u(\bm{\omega}^{(k),\ell'-1}_{u},\xi)}{{B}_{u}(\widetilde{\tau}_{b}^{\downarrow,{(k)}})}}\right\Vert^2 \right]
    \nonumber\\
    &= \eta_k^2 \mathbb{E}_k\vast[\Bigg\Vert\sum_{\ell'=1}^{\ell-1} \sum_{\xi\in \mathcal{B}^{(k),\ell'}_{u}} \hspace{-1mm} {\frac{\nabla  f_u(\bm{\omega}^{(k),\ell'-1}_{u},\xi)}{{B}_{u}(\widetilde{\tau}_{b}^{\downarrow,{(k)}})}}\nonumber-\frac{1}{|\Upsilon_{u}(\widetilde{\tau}_{b}^{\downarrow,{(k)}})|} \sum_{\ell'=1}^{\ell-1} \sum_{\xi\in\Upsilon_{u}(\widetilde{\tau}_{b}^{\downarrow,{(k)}})}\nabla  f_u(\bm{\omega}^{(k),\ell'-1}_{u},\xi)\nonumber\\
    &~~~~~~~~~~~~~~~~~~~~~~~~~~~~~~~~~~~~~~~~~~~~~~~~~~~~~~+ \frac{1}{|\Upsilon_{u}(\widetilde{\tau}_{b}^{\downarrow,{(k)}})|} \sum_{\ell'=1}^{\ell-1}\sum_{\xi\in\Upsilon_{u}(\widetilde{\tau}_{b}^{\downarrow,{(k)}})}\nabla  f_u(\bm{\omega}^{(k),\ell'-1}_{u},\xi)\Bigg\Vert^2 \vast]\nonumber \nonumber\\
    & \overset{(ii)}{\leq }2 \eta_k^2  \mathbb{E}_k\left[\Bigg\Vert\sum_{\ell'=1}^{\ell-1} \sum_{\xi\in \mathcal{B}^{(k),\ell'}_{u}} \hspace{-1mm} {\frac{\nabla  f_u(\bm{\omega}^{(k),\ell'-1}_{u},\xi)}{{B}_{u}(\widetilde{\tau}_{b}^{\downarrow,{(k)}})}} - \frac{1}{|\Upsilon_{u}(\widetilde{\tau}_{b}^{\downarrow,{(k)}})|} \sum_{\ell'=1}^{\ell-1}\sum_{\xi\in\Upsilon_{u}(\widetilde{\tau}_{b}^{\downarrow,{(k)}})}\nabla f_u(\bm{\omega}^{(k),\ell'-1}_{u},\xi)\Bigg\Vert^2\right] \nonumber\\
    &~~~~~~~~~~~~~~~~~~~~~~~~~~~~~~~~~~~~~~~~~+2 \eta_k^2 \mathbb{E}_k\left[\Bigg\Vert  \frac{1}{|\Upsilon_{u}(\widetilde{\tau}_{b}^{\downarrow,{(k)}})|} \sum_{\ell'=1}^{\ell-1}\sum_{\xi\in\Upsilon_{u}(\widetilde{\tau}_{b}^{\downarrow,{(k)}})}\nabla  f_u(\bm{\omega}^{(k),\ell'-1}_{u},\xi)\Bigg\Vert^2 \right]\nonumber \nonumber\\
    &\overset{(iii)}{=} \underbrace{2 \eta_k^2 \sum_{\ell'=1}^{\ell-1}\mathbb{E}_k\left[\Bigg\Vert\sum_{\xi\in \mathcal{B}^{(k),\ell'}_{u}} \hspace{-1mm} {\frac{\nabla  f_u(\bm{\omega}^{(k),\ell'-1}_{u},\xi)}{{B}_{u}(\widetilde{\tau}_{b}^{\downarrow,{(k)}})}} - \frac{1}{|\Upsilon_{u}(\widetilde{\tau}_{b}^{\downarrow,{(k)}})|} \sum_{\xi\in\Upsilon_{u}(\widetilde{\tau}_{b}^{\downarrow,{(k)}})}\nabla f_u(\bm{\omega}^{(k),\ell'-1}_{u},\xi)\Bigg\Vert^2\right]}_{(e)} \nonumber\\
    &~~~~~~~~~~~~~~~~~~~~~~~~~~~~~~~~~~~~~~~~~+\underbrace{2 \eta_k^2 \mathbb{E}_k\left[\Bigg\Vert  \frac{1}{|\Upsilon_{u}(\widetilde{\tau}_{b}^{\downarrow,{(k)}})|} \sum_{\ell'=1}^{\ell-1}\sum_{\xi\in\Upsilon_{u}(\widetilde{\tau}_{b}^{\downarrow,{(k)}})}\nabla  f_u(\bm{\omega}^{(k),\ell'-1}_{u},\xi)\Bigg\Vert^2\right]}_{(f)},
\end{align}
where to obtain~\eqref{eq:res4}, in equality $(i)$ we used~\eqref{eq:WeightupdateStrat}, in inequality $(ii)$ we used  Cauchy–Schwarz inequality, and $(ii)$
uses the fact that each local gradient estimation is unbiased (i.e., zero mean) conditioned on its own local parameter and the law of total expectation (across the mini-batches $\ell'$). Similar to~\eqref{eq:varGrad}, using~\eqref{eq:dataVar} we upper bound term $(e)$ in~\eqref{eq:res4} as follows:
\begin{equation}\label{eq:A2}
    (e) \leq 4 \Theta^2 \eta_k^2 \sum_{\ell'=1}^{\ell-1} \left(1-\frac{{B}_{u}(\widetilde{\tau}_{b}^{\downarrow,{(k)}})}{|\Upsilon_{u}(\widetilde{\tau}_{b}^{\downarrow,{(k)}})|} \right)  \frac{{(|\Upsilon_{u}(\widetilde{\tau}_{b}^{\downarrow,{(k)}})|-1)}
     \left(\sigma_{u}(\widetilde{\tau}_{b}^{\downarrow,{(k)}})\right)^2}{|\Upsilon_{u}(\widetilde{\tau}_{b}^{\downarrow,{(k)}})|{B}_{u}(\widetilde{\tau}_{b}^{\downarrow,{(k)}})}.
\end{equation}
Also, for term $(f)$, we have
\begin{align}\label{eq:B2}
   (f)&
   \overset{(i)}{\leq} 2 \eta_k^2 (\ell-1)\sum_{\ell'=1}^{\ell-1}\mathbb{E}_k\left[\Bigg\Vert \frac{1}{|\Upsilon_{u}(\widetilde{\tau}_{b}^{\downarrow,{(k)}})|} \sum_{\xi\in\Upsilon_{u}(\widetilde{\tau}_{b}^{\downarrow,{(k)}})}\nabla  f_u(\bm{\omega}^{(k),\ell'-1}_{u},\xi) - \nabla \mathfrak{L}^{(k)}_u(\bm{\omega}^{(k)})+ \nabla \mathfrak{L}^{(k)}_u(\bm{\omega}^{(k)})
    \Bigg\Vert^2 \right]
   \nonumber \nonumber\\
    &\overset{(ii)}{\leq} 4 \eta_k^2 (\ell-1)\sum_{\ell'=1}^{\ell-1}\mathbb{E}_k\left[\Bigg\Vert \nabla  \mathfrak{L}^{(k)}_u(\bm{\omega}^{(k),\ell'-1}_{u}) - \nabla \mathfrak{L}^{(k)}_u(\bm{\omega}^{(k)})\Bigg\Vert^2\right] + 4 \eta_k^2 (\ell-1)\sum_{\ell'=1}^{\ell-1} \Bigg\Vert\nabla \mathfrak{L}^{(k)}_u(\bm{\omega}^{(k)})
    \Bigg\Vert^2 
    \nonumber \nonumber\\
    &\leq 4\eta_k^2\beta^2 (\ell-1) \sum_{\ell'=1}^{\ell-1}\mathbb{E}_k\left[\Bigg\Vert \bm{\omega}^{(k),\ell'-1}_{u} - \bm{\omega}^{(k)}\Bigg\Vert^2\right]+ 4 \eta_k^2 (\ell-1)\sum_{\ell'=1}^{\ell-1} \Bigg\Vert\nabla \mathfrak{L}^{(k)}_u(\bm{\omega}^{(k)})
    \Bigg\Vert^2 ,
\end{align}
where inequalities $(i)$ and $(ii)$ are obtained via Cauchy-Schwarz inequality. Replacing the result of~\eqref{eq:A2} and~\eqref{eq:B2} back in~\eqref{eq:res4} we have
\begin{align}
  \mathbb{E}_k\left[  \left\Vert
  \bm{\omega}^{(k)}-
     \bm{\omega}^{(k),\ell-1}_{u}
    \right\Vert^2\right] \leq& 4 \Theta^2 \eta_k^2 \sum_{\ell'=1}^{\ell-1} \left(1-\frac{{B}_{u}(\widetilde{\tau}_{b}^{\downarrow,{(k)}})}{|\Upsilon_{u}(\widetilde{\tau}_{b}^{\downarrow,{(k)}})|} \right)  \frac{{(|\Upsilon_{u}(\widetilde{\tau}_{b}^{\downarrow,{(k)}})|-1)}
     \left(\sigma_{u}(\widetilde{\tau}_{b}^{\downarrow,{(k)}})\right)^2}{|\Upsilon_{u}(\widetilde{\tau}_{b}^{\downarrow,{(k)}})|{B}_{u}(\widetilde{\tau}_{b}^{\downarrow,{(k)}})} \nonumber \nonumber\\
     & + 4\eta_k^2\beta^2 (\ell-1) \sum_{\ell'=1}^{\ell-1}\mathbb{E}_k\left[\Bigg\Vert \bm{\omega}^{(k),\ell'-1}_{u} - \bm{\omega}^{(k)}\Bigg\Vert^2\right]+ 4 \eta_k^2 (\ell-1)\sum_{\ell'=1}^{\ell-1} \Bigg\Vert\nabla \mathfrak{L}^{(k)}_u(\bm{\omega}^{(k)})
    \Bigg\Vert^2, 
\end{align}
which implies
\begin{align}
   &\sum_{\ell=1}^{\ell_u^{(k)}} \mathbb{E}_k\left[\left\Vert\bm{\omega}^{(k)}-\bm{\omega}^{(k),\ell-1}_{u}\right\Vert^2\right] \leq 4 \Theta^2 \eta_k^2\sum_{\ell=1}^{\ell_u^{(k)}} \sum_{\ell'=1}^{\ell-1} \left(1-\frac{{B}_{u}(\widetilde{\tau}_{b}^{\downarrow,{(k)}})}{|\Upsilon_{u}(\widetilde{\tau}_{b}^{\downarrow,{(k)}})|} \right)  \frac{{(|\Upsilon_{u}(\widetilde{\tau}_{b}^{\downarrow,{(k)}})|-1)} \left(\sigma_{u}(\widetilde{\tau}_{b}^{\downarrow,{(k)}})\right)^2}{|\Upsilon_{u}(\widetilde{\tau}_{b}^{\downarrow,{(k)}})|{B}_{u}(\widetilde{\tau}_{b}^{\downarrow,{(k)}})} \nonumber \\
     & + 4\eta_k^2\beta^2 \sum_{\ell=1}^{\ell_u^{(k)}}(\ell-1) \sum_{\ell'=1}^{\ell-1}\mathbb{E}_k\left[\Bigg\Vert \bm{\omega}^{(k),\ell'-1}_{u} - \bm{\omega}^{(k)}\Bigg\Vert^2\right]+ 4 \eta_k^2 \sum_{\ell=1}^{\ell_u^{(k)}}(\ell-1)\sum_{\ell'=1}^{\ell-1} \Bigg\Vert\nabla \mathfrak{L}^{(k)}_u(\bm{\omega}^{(k)}) \Bigg\Vert^2\nonumber\\
     &\leq 4 \Theta^2 \eta_k^2 \left(\ell_u^{(k)}\right)\left(\ell_u^{(k)}-1\right) \left(1-\frac{{B}_{u}(\widetilde{\tau}_{b}^{\downarrow,{(k)}})}{|\Upsilon_{u}(\widetilde{\tau}_{b}^{\downarrow,{(k)}})|} \right)  \frac{{(|\Upsilon_{u}(\widetilde{\tau}_{b}^{\downarrow,{(k)}})|-1)}\left(\sigma_{u}(\widetilde{\tau}_{b}^{\downarrow,{(k)}})\right)^2}{|\Upsilon_{u}(\widetilde{\tau}_{b}^{\downarrow,{(k)}})|{B}_{u}(\widetilde{\tau}_{b}^{\downarrow,{(k)}})} \nonumber\\
     & + 4\eta_k^2\beta^2 \left(\ell_u^{(k)}\right)\left(\ell_u^{(k)}-1\right) \sum_{\ell=1}^{\ell_u^{(k)}}\mathbb{E}_k\left[\Bigg\Vert \bm{\omega}^{(k),\ell-1}_{u} - \bm{\omega}^{(k)}\Bigg\Vert^2\right]+ 4 \eta_k^2 \left(\ell_u^{(k)}\right)\left(\ell_u^{(k)}-1\right) \sum_{\ell=1}^{\ell_u^{(k)}} \Bigg\Vert\nabla \mathfrak{L}^{(k)}_u(\bm{\omega}^{(k)})\Bigg\Vert^2.
\end{align}
Assuming $\eta_k \leq \left(2\beta\sqrt{\ell_u^{(k)}(\ell_u^{(k)}-1)}\right)^{-1},\forall u$, the above inequality implies
\begin{align}\label{eq:f_bound}
    \sum_{\ell=1}^{\ell_u^{(k)}} \mathbb{E}_k\left[\left\Vert\bm{\omega}^{(k)}- \bm{\omega}^{(k),\ell-1}_{u} \right\Vert^2\right] \leq&\frac{4 \Theta^2 \eta_k^2 \ell_u^{(k)}\left(\ell_u^{(k)}-1\right)}{1- 4\eta_k^2\beta^2 \ell_u^{(k)}\left(\ell_u^{(k)}-1\right)} \left(1-\frac{{B}_{u}(\widetilde{\tau}_{b}^{\downarrow,{(k)}})}{|\Upsilon_{u}(\widetilde{\tau}_{b}^{\downarrow,{(k)}})|} \right)  \frac{{(|\Upsilon_{u}(\widetilde{\tau}_{b}^{\downarrow,{(k)}})|-1)}\left(\sigma_{u}(\widetilde{\tau}_{b}^{\downarrow,{(k)}})\right)^2}{|\Upsilon_{u}(\widetilde{\tau}_{b}^{\downarrow,{(k)}})|{B}_{u}(\widetilde{\tau}_{b}^{\downarrow,{(k)}})}\nonumber\\
     & + \frac{4 \eta_k^2 \left(\ell_u^{(k)}\right)^2\left(\ell_u^{(k)}-1\right)}{1- 4\eta_k^2\beta^2 \ell_u^{(k)}\left(\ell_u^{(k)}-1\right)}  \Bigg\Vert\nabla \mathfrak{L}^{(k)}_u(\bm{\omega}^{(k)})\Bigg\Vert^2.
\end{align}
Replacing this result back in~\eqref{eq:res3} we have
\begin{align}\label{eq:res5}
    (d) \leq &{4 \beta^2\Theta^2 \eta_k^2}\sum_{b \in \Omega} \sum_{u\in \mathcal{U}_{b}}\frac{|\Upsilon_{u}(\widetilde{\tau}_{b}^{\downarrow,{(k)}})|}{|\Upsilon(\widetilde{\bm{\tau}}^{\downarrow,{(k)}})| \ell_u^{(k)}}\frac{\left(\ell_u^{(k)}\right)\left(\ell_u^{(k)}-1\right)}{1- 4\eta_k^2\beta^2 \ell_u^{(k)}\left(\ell_u^{(k)}-1\right)}\left(1-\frac{{B}_{u}(\widetilde{\tau}_{b}^{\downarrow,{(k)}})}{|\Upsilon_{u}(\widetilde{\tau}_{b}^{\downarrow,{(k)}})|} \right)  \frac{{(|\Upsilon_{u}(\widetilde{\tau}_{b}^{\downarrow,{(k)}})|-1)}
    \left(\sigma_{u}(\widetilde{\tau}_{b}^{\downarrow,{(k)}})\right)^2}{|\Upsilon_{u}(\widetilde{\tau}_{b}^{\downarrow,{(k)}})|{B}_{u}(\widetilde{\tau}_{b}^{\downarrow,{(k)}})}\nonumber  \nonumber\\
    &+ {4 \eta_k^2\beta^2 }\sum_{b \in \Omega} \sum_{u\in \mathcal{U}_{b}}\frac{|\Upsilon_{u}(\widetilde{\tau}_{b}^{\downarrow,{(k)}})|}{|\Upsilon(\widetilde{\bm{\tau}}^{\downarrow,{(k)}})| \ell_u^{(k)}}\frac{\left(\ell_u^{(k)}\right)^2\left(\ell_u^{(k)}-1\right)}{{1- 4\eta_k^2\beta^2 \ell_u^{(k)}\left(\ell_u^{(k)}-1\right)}}  \Bigg\Vert\nabla \mathfrak{L}^{(k)}_u(\bm{\omega}^{(k)}) \Bigg\Vert^2.
    \end{align}
Assuming $\ell^{(k)}_{\mathsf{max}}=\max_{b\in\Omega,u\in \mathcal{U}_{b}}\{\ell^{(k)}_u\}$, and imposing 
 \begin{align}\label{eq:firstConStep}
     1- 4\eta_k^2\beta^2 \ell_{\mathsf{max}}^{(k)}\left(\ell_{\mathsf{max}}^{(k)}-1\right) \geq 0\Rightarrow \eta_k \leq \frac{1}{2\beta} \sqrt{\frac{1}{ \ell_{\mathsf{max}}^{(k)}\left(\ell_{\mathsf{max}}^{(k)}-1\right)}}
 \end{align}
on $\eta_k$ result in  
 \begin{align}\label{eq:res09444}
    (d) \leq &{4 \beta^2\Theta^2 \eta_k^2}\sum_{b \in \Omega} \sum_{u\in \mathcal{U}_{b}}\frac{|\Upsilon_{u}(\widetilde{\tau}_{b}^{\downarrow,{(k)}})|}{|\Upsilon(\widetilde{\bm{\tau}}^{\downarrow,{(k)}})| \ell_u^{(k)}}\frac{\left(\ell_u^{(k)}\right)\left(\ell_u^{(k)}-1\right)}{1- 4\eta_k^2\beta^2 \ell_u^{(k)}\left(\ell_u^{(k)}-1\right)}\left(1-\frac{{B}_{u}(\widetilde{\tau}_{b}^{\downarrow,{(k)}})}{|\Upsilon_{u}(\widetilde{\tau}_{b}^{\downarrow,{(k)}})|} \right)  \frac{{(|\Upsilon_{u}(\widetilde{\tau}_{b}^{\downarrow,{(k)}})|-1)}
    \left(\sigma_{u}(\widetilde{\tau}_{b}^{\downarrow,{(k)}})\right)^2}{|\Upsilon_{u}(\widetilde{\tau}_{b}^{\downarrow,{(k)}})|{B}_{u}(\widetilde{\tau}_{b}^{\downarrow,{(k)}})}\nonumber  \nonumber\\
    &+ \frac{4 \eta_k^2\beta^2 \left(\ell_{\mathsf{max}}^{(k)}\right)\left(\ell_{\mathsf{max}}^{(k)}-1\right)}{1- 4\eta_k^2\beta^2 \ell_{\mathsf{max}}^{(k)}\left(\ell_{\mathsf{max}}^{(k)}-1\right)}\left(\sum_{b \in \Omega} \sum_{u\in \mathcal{U}_{b}}\frac{|\Upsilon_{u}(\widetilde{\tau}_{b}^{\downarrow,{(k)}})|}{|\Upsilon(\widetilde{\bm{\tau}}^{\downarrow,{(k)}})| }\Bigg\Vert \nabla \mathfrak{L}^{(k)}_u(\bm{\omega}^{(k)})
    \Bigg\Vert^2\right).
\end{align}
Using the bounded dissimilarity assumption among the local gradients (Assumption~\ref{Assup:Dissimilarity}), we get
\begin{align}\label{eq:res4444}
        (d) 
    \leq &{4 \beta^2\Theta^2 \eta_k^2}\sum_{b \in \Omega} \sum_{u\in \mathcal{U}_{b}}\frac{|\Upsilon_{u}(\widetilde{\tau}_{b}^{\downarrow,{(k)}})|}{|\Upsilon(\widetilde{\bm{\tau}}^{\downarrow,{(k)}})| \ell_u^{(k)}}\frac{\left(\ell_u^{(k)}\right)\left(\ell_u^{(k)}-1\right)}{1- 4\eta_k^2\beta^2 \ell_u^{(k)}\left(\ell_u^{(k)}-1\right)}\left(1-\frac{{B}_{u}(\widetilde{\tau}_{b}^{\downarrow,{(k)}})}{|\Upsilon_{u}(\widetilde{\tau}_{b}^{\downarrow,{(k)}})|} \right)  \frac{{(|\Upsilon_{u}(\widetilde{\tau}_{b}^{\downarrow,{(k)}})|-1)}
    \left(\sigma_{u}(\widetilde{\tau}_{b}^{\downarrow,{(k)}})\right)^2}{|\Upsilon_{u}(\widetilde{\tau}_{b}^{\downarrow,{(k)}})|{B}_{u}(\widetilde{\tau}_{b}^{\downarrow,{(k)}})}\nonumber\\
    &+ \frac{4 \eta_k^2\beta^2 \left(\ell_{\mathsf{max}}^{(k)}\right)\left(\ell_{\mathsf{max}}^{(k)}-1\right)}{1- 4\eta_k^2\beta^2 \ell_{\mathsf{max}}^{(k)}\left(\ell_{\mathsf{max}}^{(k)}-1\right)}\left(\mathfrak{X}_1\Bigg\Vert \sum_{b \in \Omega} \sum_{u\in \mathcal{U}_{b}}\frac{|\Upsilon_{u}(\widetilde{\tau}_{b}^{\downarrow,{(k)}})|}{|\Upsilon(\widetilde{\bm{\tau}}^{\downarrow,{(k)}})| }\nabla \mathfrak{L}^{(k)}_u(\bm{\omega}^{(k)})
    \Bigg\Vert^2+\mathfrak{X}_2\right)\nonumber\\
    = &{4 \beta^2\Theta^2 \eta_k^2}\sum_{b \in \Omega} \sum_{u\in \mathcal{U}_{b}}\frac{|\Upsilon_{u}(\widetilde{\tau}_{b}^{\downarrow,{(k)}})|}{|\Upsilon(\widetilde{\bm{\tau}}^{\downarrow,{(k)}})| \ell_u^{(k)}}\frac{\left(\ell_u^{(k)}\right)\left(\ell_u^{(k)}-1\right)}{1- 4\eta_k^2\beta^2 \ell_u^{(k)}\left(\ell_u^{(k)}-1\right)}\left(1-\frac{{B}_{u}(\widetilde{\tau}_{b}^{\downarrow,{(k)}})}{|\Upsilon_{u}(\widetilde{\tau}_{b}^{\downarrow,{(k)}})|} \right)  \frac{{(|\Upsilon_{u}(\widetilde{\tau}_{b}^{\downarrow,{(k)}})|-1)}\left(\sigma_{u}(\widetilde{\tau}_{b}^{\downarrow,{(k)}})\right)^2}{|\Upsilon_{u}(\widetilde{\tau}_{b}^{\downarrow,{(k)}})|{B}_{u}(\widetilde{\tau}_{b}^{\downarrow,{(k)}})}\nonumber\\
    &+ \frac{4 \eta_k^2\beta^2 \left(\ell_{\mathsf{max}}^{(k)}\right)\left(\ell_{\mathsf{max}}^{(k)}-1\right)}{1- 4\eta_k^2\beta^2\ell_{\mathsf{max}}^{(k)}\left(\ell_{\mathsf{max}}^{(k)}-1\right)}\left(\mathfrak{X}_1   \Bigg\Vert\nabla \mathfrak{L}^{(k)}(\bm{\omega}^{(k)})\Bigg\Vert^2 + \mathfrak{X}_2 \right).
\end{align}
 Replacing the above result back in~\eqref{ineq:main5} and gathering the terms leads to
 \begin{align}\label{ineq:main6}
    \mathbb{E}_k&\left[\mathfrak{L}^{({k})}(\bm{\omega}^{(k+1)})\right] {\leq} \mathfrak{L}^{({k})}(\bm{\omega}^{(k)}) +\frac{\eta_{_k}\mathfrak{B}_k}{4}\left(\frac{8 \mathfrak{X}_1 \eta_k^2\beta^2 \left(\ell_{\mathsf{max}}^{(k)}\right)\left(\ell_{\mathsf{max}}^{(k)}-1\right)}{1- 4\eta_k^2\beta^2\ell_{\mathsf{max}}^{(k)}\left(\ell_{\mathsf{max}}^{(k)}-1\right)}-1\right) \left\Vert\nabla{\mathfrak{L}^{({k})}(\bm{\omega}^{(k)})}\right\Vert^2  \nonumber\\
    &+{2 \beta^2\Theta^2 \eta_k^3 \mathfrak{B}_k}\sum_{b \in \Omega} \sum_{u\in \mathcal{U}_{b}}\frac{|\Upsilon_{u}(\widetilde{\tau}_{b}^{\downarrow,{(k)}})|}{|\Upsilon(\widetilde{\bm{\tau}}^{\downarrow,{(k)}})| \ell_u^{(k)}}\frac{\left(\ell_u^{(k)}\right)\left(\ell_u^{(k)}-1\right)}{1- 4\eta_k^2\beta^2 \ell_u^{(k)}\left(\ell_u^{(k)}-1\right)}\left(1-\frac{{B}_{u}(\widetilde{\tau}_{b}^{\downarrow,{(k)}})}{|\Upsilon_{u}(\widetilde{\tau}_{b}^{\downarrow,{(k)}})|} \right)  \frac{{(|\Upsilon_{u}(\widetilde{\tau}_{b}^{\downarrow,{(k)}})|-1)}\left(\sigma_{u}(\widetilde{\tau}_{b}^{\downarrow,{(k)}})\right)^2}{|\Upsilon_{u}(\widetilde{\tau}_{b}^{\downarrow,{(k)}})|{B}_{u}(\widetilde{\tau}_{b}^{\downarrow,{(k)}})}\nonumber\\
    & + \frac{2 \mathfrak{X}_2 \eta_k^3\beta^2 \mathfrak{B}_k \left(\ell_{\mathsf{max}}^{(k)}\right)\left(\ell_{\mathsf{max}}^{(k)}-1\right)}{1- 4\eta_k^2\beta^2\ell_{\mathsf{max}}^{(k)}\left(\ell_{\mathsf{max}}^{(k)}-1\right)}+\frac{3}{2}\eta_{_k}\mathfrak{B}_k\underbrace{\mathbb{E}_k\left[\left\Vert\sum_{b \in \Omega}\mathscr{U}^{(k)}_{b}\right\Vert^2\right]}_{(g)}\nonumber\\
    &+ \frac{\beta\eta_{_k}^2\mathfrak{B}_k^2}{2} \sum_{b \in \Omega} \sum_{u\in \mathcal{U}_{b}}\frac{\left(\widehat{\lambda}_{u}^{(k)}|\Upsilon_{u}(\widetilde{\tau}_{b}^{\downarrow,{(k)}})|\right)^2}{\left(|\Upsilon^{\mathsf{s}}(\widetilde{\bm{\tau}}^{\downarrow,{(k)}})|\right)^2 \ell^{(k)}_{u}}\left(1-\frac{{B}_{u}(\widetilde{\tau}_{b}^{\downarrow,{(k)}})}{|\Upsilon_{u}(\widetilde{\tau}_{b}^{\downarrow,{(k)}})|} \right) \frac{2(|\Upsilon_{u}(\widetilde{\tau}_{b}^{\downarrow,{(k)}})|-1)\Theta^2\left(\sigma_{u}(\widetilde{\tau}_{b}^{\downarrow,{(k)}})\right)^2}{|\Upsilon_{u}(\widetilde{\tau}_{b}^{\downarrow,{(k)}})|{B}_{u}(\widetilde{\tau}_{b}^{\downarrow,{(k)}})}.
\end{align}
Considering the definition of the noise of FLU recruitment at O-RU $b$ presented in \eqref{eq:user_selection_noise}, we next aim to bound $(g)$ bellow.
\begin{align}
    \mathbb{E}_k\left[\left\Vert\sum_{b \in \Omega}\mathscr{U}^{(k)}_{b}\right\Vert^2\right]&= \mathbb{E}_k\left[\left\Vert\sum_{b \in \Omega} \sum_{u\in \mathcal{U}_{b}}\frac{\widehat{\lambda}_{u}^{(k)}|\Upsilon_{u}(\widetilde{\tau}_{b}^{\downarrow,{(k)}})|}{|\Upsilon^{\mathsf{s}}(\widetilde{\bm{\tau}}^{\downarrow,{(k)}})| \ell^{(k)}_{u}}\sum_{\ell=1}^{\ell^{(k)}_{u}} {\nabla  \mathfrak{L}^{(k)}_u(\bm{\omega}^{(k),\ell-1}_{u})}-\sum_{b \in \Omega} \sum_{u\in \mathcal{U}_{b}}\frac{|\Upsilon_{u}(\widetilde{\tau}_{b}^{\downarrow,{(k)}})|}{|\Upsilon(\widetilde{\bm{\tau}}^{\downarrow,{(k)}})| \ell^{(k)}_{u}}\sum_{\ell=1}^{\ell^{(k)}_{u}} {\nabla  \mathfrak{L}^{(k)}_u(\bm{\omega}^{(k),\ell-1}_{u})} \right\Vert^2\right]\nonumber\\
    &=\mathbb{E}_k\Bigg[\Bigg\Vert\sum_{b \in \Omega} \sum_{u\in \mathcal{U}_{b}}\widehat{\lambda}_{u}^{(k)}\frac{\left(|\Upsilon(\widetilde{\bm{\tau}}^{\downarrow,{(k)}})|-|\Upsilon^{\mathsf{s}}(\widetilde{\bm{\tau}}^{\downarrow,{(k)}})|\right)|\Upsilon_{u}(\widetilde{\tau}_{b}^{\downarrow,{(k)}})|}{|\Upsilon(\widetilde{\bm{\tau}}^{\downarrow,{(k)}})||\Upsilon^{\mathsf{s}}(\widetilde{\bm{\tau}}^{\downarrow,{(k)}})| \ell^{(k)}_{u}}\sum_{\ell=1}^{\ell^{(k)}_{u}} {\nabla  \mathfrak{L}^{(k)}_u(\bm{\omega}^{(k),\ell-1}_{u})}\nonumber\\
    &~~~~~~~~~~~~~~~~~~~~~~~~-\sum_{b \in \Omega} \sum_{u\in \mathcal{U}_{b}}\left(1-\widehat{\lambda}_{u}^{(k)}\right)\frac{|\Upsilon_{u}(\widetilde{\tau}_{b}^{\downarrow,{(k)}})|}{|\Upsilon(\widetilde{\bm{\tau}}^{\downarrow,{(k)}})| \ell^{(k)}_{u}}\sum_{\ell=1}^{\ell^{(k)}_{u}} {\nabla  \mathfrak{L}^{(k)}_u(\bm{\omega}^{(k),\ell-1}_{u})} \Bigg\Vert^2\Bigg]\nonumber\\
    &\overset{(i)}{\le} \mathbb{E}_k\left[\left (\sum_{b \in \Omega} \sum_{u\in \mathcal{U}_{b}}\widehat{\lambda}_{u}^{(k)}\frac{\left(|\Upsilon(\widetilde{\bm{\tau}}^{\downarrow,{(k)}})|-|\Upsilon^{\mathsf{s}}(\widetilde{\bm{\tau}}^{\downarrow,{(k)}})|\right)|\Upsilon_{u}(\widetilde{\tau}_{b}^{\downarrow,{(k)}})|}{|\Upsilon(\widetilde{\bm{\tau}}^{\downarrow,{(k)}})||\Upsilon^{\mathsf{s}}(\widetilde{\bm{\tau}}^{\downarrow,{(k)}})|}\max_{b\in\Omega,u\in\mathcal{U}_{b}}\left(\frac{\left\Vert\sum_{\ell=1}^{\ell^{(k)}_{u}} {\nabla  \mathfrak{L}^{(k)}_u(\bm{\omega}^{(k),\ell-1}_{u})}\right\Vert}{\ell^{(k)}_{u}}\right)\right.\right.\nonumber\\
    &~~~~~~~~~~~~~~~~~~~~~~~~+\left.\left.\sum_{b \in \Omega} \sum_{u\in \mathcal{U}_{b}}\left(1-\widehat{\lambda}_{u}^{(k)}\right)\frac{|\Upsilon_{u}(\widetilde{\tau}_{b}^{\downarrow,{(k)}})|}{|\Upsilon(\widetilde{\bm{\tau}}^{\downarrow,{(k)}})|}\max_{b\in\Omega,u\in\mathcal{U}_{b}}\left(\frac{\left\Vert\sum_{\ell=1}^{\ell^{(k)}_{u}} {\nabla  \mathfrak{L}^{(k)}_u(\bm{\omega}^{(k),\ell-1}_{u})}\right\Vert}{\ell^{(k)}_{u}}\right) \right)^2\right]\nonumber\\
    &= \mathbb{E}_k\left[\left (\frac{\left(|\Upsilon(\widetilde{\bm{\tau}}^{\downarrow,{(k)}})|-|\Upsilon^{\mathsf{s}}(\widetilde{\bm{\tau}}^{\downarrow,{(k)}})|\right)|\Upsilon^{\mathsf{s}}(\widetilde{\bm{\tau}}^{\downarrow,{(k)}})|}{|\Upsilon(\widetilde{\bm{\tau}}^{\downarrow,{(k)}})||\Upsilon^{\mathsf{s}}(\widetilde{\bm{\tau}}^{\downarrow,{(k)}})|}\max_{b\in\Omega,u\in\mathcal{U}_{b}}\left(\frac{\left\Vert\sum_{\ell=1}^{\ell^{(k)}_{u}} {\nabla  \mathfrak{L}^{(k)}_u(\bm{\omega}^{(k),\ell-1}_{u})}\right\Vert}{\ell^{(k)}_{u}}\right)\right.\right.\nonumber\\
    &~~~~~~~~~~~~~~~~~~~~~~~~+\left.\left.\frac{|\Upsilon(\widetilde{\bm{\tau}}^{\downarrow,{(k)}})|-|\Upsilon^{\mathsf{s}}(\widetilde{\bm{\tau}}^{\downarrow,{(k)}})|}{|\Upsilon(\widetilde{\bm{\tau}}^{\downarrow,{(k)}})|}\max_{b\in\Omega,u\in\mathcal{U}_{b}}\left(\frac{\left\Vert\sum_{\ell=1}^{\ell^{(k)}_{u}} {\nabla  \mathfrak{L}^{(k)}_u(\bm{\omega}^{(k),\ell-1}_{u})}\right\Vert}{\ell^{(k)}_{u}}\right) \right)^2\right]\nonumber\\
    &= 4\left (\frac{|\Upsilon(\widetilde{\bm{\tau}}^{\downarrow,{(k)}})|-|\Upsilon^{\mathsf{s}}(\widetilde{\bm{\tau}}^{\downarrow,{(k)}})|}{|\Upsilon(\widetilde{\bm{\tau}}^{\downarrow,{(k)}})|}\right)^2\mathbb{E}_k\left[\left (\max_{b\in\Omega,u\in\mathcal{U}_{b}}\left(\frac{1}{(\ell^{(k)}_{u})}\left\Vert\sum_{\ell=1}^{\ell^{(k)}_{u}} {\nabla  \mathfrak{L}^{(k)}_u(\bm{\omega}^{(k),\ell-1}_{u})}\right\Vert\right) \right)^2\right]\nonumber\\
    &= 4\left (\frac{|\Upsilon(\widetilde{\bm{\tau}}^{\downarrow,{(k)}})|-|\Upsilon^{\mathsf{s}}(\widetilde{\bm{\tau}}^{\downarrow,{(k)}})|}{|\Upsilon(\widetilde{\bm{\tau}}^{\downarrow,{(k)}})|}\right)^2\mathbb{E}_k\left[\max_{b\in\Omega,u\in\mathcal{U}_{b}}\left(\frac{1}{(\ell^{(k)}_{u})^2}\underbrace{\left\Vert\sum_{\ell=1}^{\ell^{(k)}_{u}} {\nabla  \mathfrak{L}^{(k)}_u(\bm{\omega}^{(k),\ell-1}_{u})}\right\Vert^2}_{(a)}\right)\right],
\end{align}
where $(i)$ is obtained using triangle inequality. Using the same technique utilized in \eqref{eq:B2} and performing some algebraic manipulation give us an upper bound for $(a)$ in the above inequality. Using this result, we can further bound $\mathbb{E}_k\left[\left\Vert\sum_{b \in \Omega}\mathscr{U}^{(k)}_{b}\right\Vert^2\right]$ as follows:  
\begin{align}
    \hspace{-2.5mm}\mathbb{E}_k\left[\left\Vert\sum_{b \in \Omega}\mathscr{U}^{(k)}_{b}\right\Vert^2\right]&\le 4\left (\frac{|\Upsilon(\widetilde{\bm{\tau}}^{\downarrow,{(k)}})|-|\Upsilon^{\mathsf{s}}(\widetilde{\bm{\tau}}^{\downarrow,{(k)}})|}{|\Upsilon(\widetilde{\bm{\tau}}^{\downarrow,{(k)}})|}\right)^2\mathbb{E}_k\left[\max_{b\in\Omega,u\in\mathcal{U}_{b}}\left(\frac{2\beta^2}{(\ell^{(k)}_{u})} \sum_{\ell=1}^{\ell^{(k)}_{u}}\Bigg\Vert \bm{\omega}^{(k)}-\bm{\omega}^{(k),\ell-1}_{u}\Bigg\Vert^2+ \frac{2}{(\ell^{(k)}_{u})}\sum_{\ell=1}^{\ell^{(k)}_{u}} \Bigg\Vert\nabla \mathfrak{L}^{(k)}_u(\bm{\omega}^{(k)})\Bigg\Vert^2\right) \right]\nonumber\\
    &\le 8\left (\frac{|\Upsilon(\widetilde{\bm{\tau}}^{\downarrow,{(k)}})|-|\Upsilon^{\mathsf{s}}(\widetilde{\bm{\tau}}^{\downarrow,{(k)}})|}{|\Upsilon(\widetilde{\bm{\tau}}^{\downarrow,{(k)}})|}\right)^2\sum_{b \in \Omega} \sum_{u\in \mathcal{U}_{b}}\left( \frac{\beta^2}{(\ell^{(k)}_{u})}\sum_{\ell=1}^{\ell^{(k)}_{u}}\mathbb{E}_k\left[\Bigg\Vert \bm{\omega}^{(k)}-\bm{\omega}^{(k),\ell-1}_{u}\Bigg\Vert^2 \right]+ \Bigg\Vert\nabla \mathfrak{L}^{(k)}_u(\bm{\omega}^{(k)})\Bigg\Vert^2\right).
\end{align}
Substituting $\sum_{\ell=1}^{\ell^{(k)}_{u}}\mathbb{E}_k\left[\left\Vert \bm{\omega}^{(k)}-\bm{\omega}^{(k),\ell-1}_{u}\right\Vert^2 \right]$ from \eqref{eq:f_bound} into the above inequality and performing some algebraic operations give us 
\begin{align}
    &\mathbb{E}_k\left[\left\Vert\sum_{b \in \Omega}\mathscr{U}^{(k)}_{b}\right\Vert^2\right]\nonumber\\
    &\le 8\left (\frac{|\Upsilon(\widetilde{\bm{\tau}}^{\downarrow,{(k)}})|-|\Upsilon^{\mathsf{s}}(\widetilde{\bm{\tau}}^{\downarrow,{(k)}})|}{|\Upsilon(\widetilde{\bm{\tau}}^{\downarrow,{(k)}})|}\right)^2\sum_{b \in \Omega} \sum_{u\in \mathcal{U}_{b}}\Bigg( \frac{\beta^2}{(\ell^{(k)}_{u})}\frac{4 \Theta^2 \eta_k^2 \ell_u^{(k)}\left(\ell_u^{(k)}-1\right)}{1- 4\eta_k^2\beta^2 \ell_u^{(k)}\left(\ell_u^{(k)}-1\right)} \left(1-\frac{{B}_{u}(\widetilde{\tau}_{b}^{\downarrow,{(k)}})}{|\Upsilon_{u}(\widetilde{\tau}_{b}^{\downarrow,{(k)}})|} \right)  \frac{{(|\Upsilon_{u}(\widetilde{\tau}_{b}^{\downarrow,{(k)}})|-1)}\left(\sigma_{u}(\widetilde{\tau}_{b}^{\downarrow,{(k)}})\right)^2}{|\Upsilon_{u}(\widetilde{\tau}_{b}^{\downarrow,{(k)}})|{B}_{u}(\widetilde{\tau}_{b}^{\downarrow,{(k)}})}\nonumber\\
     & ~~~~~~~~~~~~+ \frac{\beta^2}{(\ell^{(k)}_{u})}\frac{4 \eta_k^2 \left(\ell_u^{(k)}\right)^2\left(\ell_u^{(k)}-1\right)}{1- 4\eta_k^2\beta^2 \ell_u^{(k)}\left(\ell_u^{(k)}-1\right)}  \Bigg\Vert\nabla \mathfrak{L}^{(k)}_u(\bm{\omega}^{(k)})\Bigg\Vert^2+ \Bigg\Vert\nabla \mathfrak{L}^{(k)}_u(\bm{\omega}^{(k)})\Bigg\Vert^2\Bigg)\nonumber\\
     &\le 8\left (\frac{|\Upsilon(\widetilde{\bm{\tau}}^{\downarrow,{(k)}})|-|\Upsilon^{\mathsf{s}}(\widetilde{\bm{\tau}}^{\downarrow,{(k)}})|}{|\Upsilon(\widetilde{\bm{\tau}}^{\downarrow,{(k)}})|}\right)^2\sum_{b \in \Omega} \sum_{u\in \mathcal{U}_{b}}\Bigg( \frac{\beta^2}{(\ell^{(k)}_{u})}\frac{4 \Theta^2 \eta_k^2 \ell_u^{(k)}\left(\ell_u^{(k)}-1\right)}{1- 4\eta_k^2\beta^2 \ell_u^{(k)}\left(\ell_u^{(k)}-1\right)} \left(1-\frac{{B}_{u}(\widetilde{\tau}_{b}^{\downarrow,{(k)}})}{|\Upsilon_{u}(\widetilde{\tau}_{b}^{\downarrow,{(k)}})|} \right)  \frac{{(|\Upsilon_{u}(\widetilde{\tau}_{b}^{\downarrow,{(k)}})|-1)}\left(\sigma_{u}(\widetilde{\tau}_{b}^{\downarrow,{(k)}})\right)^2}{|\Upsilon_{u}(\widetilde{\tau}_{b}^{\downarrow,{(k)}})|{B}_{u}(\widetilde{\tau}_{b}^{\downarrow,{(k)}})}\nonumber\\
     & ~~~~~~~~~~~~+ \left(\frac{4 \eta_k^2 \beta^2 \ell_u^{(k)}\left(\ell_u^{(k)}-1\right)}{1- 4\eta_k^2\beta^2 \ell_u^{(k)}\left(\ell_u^{(k)}-1\right)}+1\right)  \Bigg\Vert\nabla \mathfrak{L}^{(k)}_u(\bm{\omega}^{(k)})\Bigg\Vert^2\Bigg)\nonumber\\
     &\le 8\left (\frac{|\Upsilon(\widetilde{\bm{\tau}}^{\downarrow,{(k)}})|-|\Upsilon^{\mathsf{s}}(\widetilde{\bm{\tau}}^{\downarrow,{(k)}})|}{|\Upsilon(\widetilde{\bm{\tau}}^{\downarrow,{(k)}})|}\right)^2\sum_{b \in \Omega} \sum_{u\in \mathcal{U}_{b}}\Bigg( \frac{\beta^2}{(\ell^{(k)}_{u})}\frac{4 \Theta^2 \eta_k^2 \ell_u^{(k)}\left(\ell_u^{(k)}-1\right)}{1- 4\eta_k^2\beta^2 \ell_u^{(k)}\left(\ell_u^{(k)}-1\right)} \left(1-\frac{{B}_{u}(\widetilde{\tau}_{b}^{\downarrow,{(k)}})}{|\Upsilon_{u}(\widetilde{\tau}_{b}^{\downarrow,{(k)}})|} \right)  \frac{{(|\Upsilon_{u}(\widetilde{\tau}_{b}^{\downarrow,{(k)}})|-1)}\left(\sigma_{u}(\widetilde{\tau}_{b}^{\downarrow,{(k)}})\right)^2}{|\Upsilon_{u}(\widetilde{\tau}_{b}^{\downarrow,{(k)}})|{B}_{u}(\widetilde{\tau}_{b}^{\downarrow,{(k)}})}\nonumber\\
     & ~~~~~~~~~~~~+ \left(\frac{4 \eta_k^2 \beta^2 \ell_u^{(k)}\left(\ell_u^{(k)}-1\right)+1- 4\eta_k^2\beta^2 \ell_u^{(k)}\left(\ell_u^{(k)}-1\right)}{1- 4\eta_k^2\beta^2 \ell_u^{(k)}\left(\ell_u^{(k)}-1\right)}\right)  \Bigg\Vert\nabla \mathfrak{L}^{(k)}_u(\bm{\omega}^{(k)})\Bigg\Vert^2\Bigg)\nonumber\\
     &\le 8\left (\frac{|\Upsilon(\widetilde{\bm{\tau}}^{\downarrow,{(k)}})|-|\Upsilon^{\mathsf{s}}(\widetilde{\bm{\tau}}^{\downarrow,{(k)}})|}{|\Upsilon(\widetilde{\bm{\tau}}^{\downarrow,{(k)}})|}\right)^2\sum_{b \in \Omega} \sum_{u\in \mathcal{U}_{b}}\Bigg( \frac{\beta^2}{(\ell^{(k)}_{u})}\frac{4 \Theta^2 \eta_k^2 \ell_u^{(k)}\left(\ell_u^{(k)}-1\right)}{1- 4\eta_k^2\beta^2 \ell_u^{(k)}\left(\ell_u^{(k)}-1\right)} \left(1-\frac{{B}_{u}(\widetilde{\tau}_{b}^{\downarrow,{(k)}})}{|\Upsilon_{u}(\widetilde{\tau}_{b}^{\downarrow,{(k)}})|} \right)  \frac{{(|\Upsilon_{u}(\widetilde{\tau}_{b}^{\downarrow,{(k)}})|-1)}\left(\sigma_{u}(\widetilde{\tau}_{b}^{\downarrow,{(k)}})\right)^2}{|\Upsilon_{u}(\widetilde{\tau}_{b}^{\downarrow,{(k)}})|{B}_{u}(\widetilde{\tau}_{b}^{\downarrow,{(k)}})}\nonumber\\
     & ~~~~~~~~~~~~+ \left(\frac{1}{1- 4\eta_k^2\beta^2 \ell_u^{(k)}\left(\ell_u^{(k)}-1\right)}\right)  \Bigg\Vert\nabla \mathfrak{L}^{(k)}_u(\bm{\omega}^{(k)})\Bigg\Vert^2\Bigg)\nonumber\\
     &\le 8\left (\frac{|\Upsilon(\widetilde{\bm{\tau}}^{\downarrow,{(k)}})|-|\Upsilon^{\mathsf{s}}(\widetilde{\bm{\tau}}^{\downarrow,{(k)}})|}{|\Upsilon(\widetilde{\bm{\tau}}^{\downarrow,{(k)}})|}\right)^2\sum_{b \in \Omega} \sum_{u\in \mathcal{U}_{b}}\Bigg(\frac{4 \Theta^2 \beta^2 \eta_k^2 \left(\ell_u^{(k)}-1\right)}{1- 4\eta_k^2\beta^2 \ell_u^{(k)}\left(\ell_u^{(k)}-1\right)} \left(1-\frac{{B}_{u}(\widetilde{\tau}_{b}^{\downarrow,{(k)}})}{|\Upsilon_{u}(\widetilde{\tau}_{b}^{\downarrow,{(k)}})|} \right)  \frac{{(|\Upsilon_{u}(\widetilde{\tau}_{b}^{\downarrow,{(k)}})|-1)}\left(\sigma_{u}(\widetilde{\tau}_{b}^{\downarrow,{(k)}})\right)^2}{|\Upsilon_{u}(\widetilde{\tau}_{b}^{\downarrow,{(k)}})|{B}_{u}(\widetilde{\tau}_{b}^{\downarrow,{(k)}})}\Bigg)\nonumber\\
     &+8\left (\frac{|\Upsilon(\widetilde{\bm{\tau}}^{\downarrow,{(k)}})|-|\Upsilon^{\mathsf{s}}(\widetilde{\bm{\tau}}^{\downarrow,{(k)}})|}{|\Upsilon(\widetilde{\bm{\tau}}^{\downarrow,{(k)}})|}\right)^2\sum_{b \in \Omega} \sum_{u\in \mathcal{U}_{b}}\left(\frac{1}{1- 4\eta_k^2\beta^2 \ell_u^{(k)}\left(\ell_u^{(k)}-1\right)}\right)\frac{|\Upsilon_{u}(\widetilde{\tau}_{b}^{\downarrow,{(k)}})|}{|\Upsilon(\widetilde{\bm{\tau}}^{\downarrow,{(k)}})|}\frac{|\Upsilon(\widetilde{\bm{\tau}}^{\downarrow,{(k)}})|}{|\Upsilon_{u}(\widetilde{\tau}_{b}^{\downarrow,{(k)}})|}\Bigg\Vert\nabla \mathfrak{L}^{(k)}_u(\bm{\omega}^{(k)})\Bigg\Vert^2.
\end{align}
Assuming $\ell^{(k)}_{\mathsf{max}}=\max_{b\in\Omega,u\in\mathcal{U}_{b}}\{\ell^{(k)}_u\}$, $|\Upsilon_{\mathsf{min}}(\widetilde{\bm{\tau}}^{\downarrow,{(k)}})|=\min_{b\in\Omega,u\in\mathcal{U}_{b}}\{|\Upsilon_{u}(\widetilde{\tau}_{b}^{\downarrow,{(k)}})|\}$ and considering \eqref{eq:firstConStep}, we get
\begin{align}
    &\mathbb{E}_k\left[\left\Vert\sum_{b \in \Omega}\mathscr{U}^{(k)}_{b}\right\Vert^2\right]\nonumber\\
    &\le 8\left (\frac{|\Upsilon(\widetilde{\bm{\tau}}^{\downarrow,{(k)}})|-|\Upsilon^{\mathsf{s}}(\widetilde{\bm{\tau}}^{\downarrow,{(k)}})|}{|\Upsilon(\widetilde{\bm{\tau}}^{\downarrow,{(k)}})|}\right)^2\sum_{b \in \Omega} \sum_{u\in \mathcal{U}_{b}}\Bigg(\frac{4 \Theta^2 \beta^2 \eta_k^2 \left(\ell_u^{(k)}-1\right)}{1- 4\eta_k^2\beta^2 \ell_u^{(k)}\left(\ell_u^{(k)}-1\right)} \left(1-\frac{{B}_{u}(\widetilde{\tau}_{b}^{\downarrow,{(k)}})}{|\Upsilon_{u}(\widetilde{\tau}_{b}^{\downarrow,{(k)}})|} \right)  \frac{{(|\Upsilon_{u}(\widetilde{\tau}_{b}^{\downarrow,{(k)}})|-1)}\left(\sigma_{u}(\widetilde{\tau}_{b}^{\downarrow,{(k)}})\right)^2}{|\Upsilon_{u}(\widetilde{\tau}_{b}^{\downarrow,{(k)}})|{B}_{u}(\widetilde{\tau}_{b}^{\downarrow,{(k)}})}\Bigg)\nonumber\\
     &+8\left (\frac{|\Upsilon(\widetilde{\bm{\tau}}^{\downarrow,{(k)}})|-|\Upsilon^{\mathsf{s}}(\widetilde{\bm{\tau}}^{\downarrow,{(k)}})|}{|\Upsilon(\widetilde{\bm{\tau}}^{\downarrow,{(k)}})|}\right)^2\sum_{b \in \Omega} \sum_{u\in \mathcal{U}_{b}}\left(\frac{1}{1- 4\eta_k^2\beta^2 \ell_{\mathsf{max}}^{(k)}\left(\ell_{\mathsf{max}}^{(k)}-1\right)}\right)\frac{|\Upsilon_{u}(\widetilde{\tau}_{b}^{\downarrow,{(k)}})|}{|\Upsilon(\widetilde{\bm{\tau}}^{\downarrow,{(k)}})|}\frac{|\Upsilon(\widetilde{\bm{\tau}}^{\downarrow,{(k)}})|}{|\Upsilon_{\mathsf{min}}(\widetilde{\bm{\tau}}^{\downarrow,{(k)}})|}\Bigg\Vert\nabla \mathfrak{L}^{(k)}_u(\bm{\omega}^{(k)})\Bigg\Vert^2\nonumber\\
     &\le 8\left (\frac{|\Upsilon(\widetilde{\bm{\tau}}^{\downarrow,{(k)}})|-|\Upsilon^{\mathsf{s}}(\widetilde{\bm{\tau}}^{\downarrow,{(k)}})|}{|\Upsilon(\widetilde{\bm{\tau}}^{\downarrow,{(k)}})|}\right)^2\sum_{b \in \Omega} \sum_{u\in \mathcal{U}_{b}}\Bigg(\frac{4 \Theta^2 \beta^2 \eta_k^2 \left(\ell_u^{(k)}-1\right)}{1- 4\eta_k^2\beta^2 \ell_u^{(k)}\left(\ell_u^{(k)}-1\right)} \left(1-\frac{{B}_{u}(\widetilde{\tau}_{b}^{\downarrow,{(k)}})}{|\Upsilon_{u}(\widetilde{\tau}_{b}^{\downarrow,{(k)}})|} \right)  \frac{{(|\Upsilon_{u}(\widetilde{\tau}_{b}^{\downarrow,{(k)}})|-1)}\left(\sigma_{u}(\widetilde{\tau}_{b}^{\downarrow,{(k)}})\right)^2}{|\Upsilon_{u}(\widetilde{\tau}_{b}^{\downarrow,{(k)}})|{B}_{u}(\widetilde{\tau}_{b}^{\downarrow,{(k)}})}\Bigg)\nonumber\\
     &+8\left (\frac{|\Upsilon(\widetilde{\bm{\tau}}^{\downarrow,{(k)}})|-|\Upsilon^{\mathsf{s}}(\widetilde{\bm{\tau}}^{\downarrow,{(k)}})|}{|\Upsilon(\widetilde{\bm{\tau}}^{\downarrow,{(k)}})|}\right)^2\frac{|\Upsilon(\widetilde{\bm{\tau}}^{\downarrow,{(k)}})|}{|\Upsilon_{\mathsf{min}}(\widetilde{\bm{\tau}}^{\downarrow,{(k)}})|}\left(\frac{1}{1- 4\eta_k^2\beta^2 \ell_{\mathsf{max}}^{(k)}\left(\ell_{\mathsf{max}}^{(k)}-1\right)}\right)\sum_{b \in \Omega} \sum_{u\in \mathcal{U}_{b}}\frac{|\Upsilon_{u}(\widetilde{\tau}_{b}^{\downarrow,{(k)}})|}{|\Upsilon(\widetilde{\bm{\tau}}^{\downarrow,{(k)}})|}\Bigg\Vert\nabla \mathfrak{L}^{(k)}_u(\bm{\omega}^{(k)})\Bigg\Vert^2.
\end{align}
Using the bounded dissimilarity assumption among the local gradients (Assumption~\ref{Assup:Dissimilarity}) results in
\begin{align}
    &\mathbb{E}_k\left[\left\Vert\sum_{b \in \Omega}\mathscr{U}^{(k)}_{b}\right\Vert^2\right]\nonumber\\
    &\le 8\left (\frac{|\Upsilon(\widetilde{\bm{\tau}}^{\downarrow,{(k)}})|-|\Upsilon^{\mathsf{s}}(\widetilde{\bm{\tau}}^{\downarrow,{(k)}})|}{|\Upsilon(\widetilde{\bm{\tau}}^{\downarrow,{(k)}})|}\right)^2\sum_{b \in \Omega} \sum_{u\in \mathcal{U}_{b}}\Bigg(\frac{4 \Theta^2 \beta^2 \eta_k^2 \left(\ell_u^{(k)}-1\right)}{1- 4\eta_k^2\beta^2 \ell_u^{(k)}\left(\ell_u^{(k)}-1\right)} \left(1-\frac{{B}_{u}(\widetilde{\tau}_{b}^{\downarrow,{(k)}})}{|\Upsilon_{u}(\widetilde{\tau}_{b}^{\downarrow,{(k)}})|} \right)  \frac{{(|\Upsilon_{u}(\widetilde{\tau}_{b}^{\downarrow,{(k)}})|-1)}\left(\sigma_{u}(\widetilde{\tau}_{b}^{\downarrow,{(k)}})\right)^2}{|\Upsilon_{u}(\widetilde{\tau}_{b}^{\downarrow,{(k)}})|{B}_{u}(\widetilde{\tau}_{b}^{\downarrow,{(k)}})}\Bigg)\nonumber\\
     &+8\left (\frac{|\Upsilon(\widetilde{\bm{\tau}}^{\downarrow,{(k)}})|-|\Upsilon^{\mathsf{s}}(\widetilde{\bm{\tau}}^{\downarrow,{(k)}})|}{|\Upsilon(\widetilde{\bm{\tau}}^{\downarrow,{(k)}})|}\right)^2\frac{|\Upsilon(\widetilde{\bm{\tau}}^{\downarrow,{(k)}})|}{|\Upsilon_{\mathsf{min}}(\widetilde{\bm{\tau}}^{\downarrow,{(k)}})|}\left(\frac{1}{1- 4\eta_k^2\beta^2 \ell_{\mathsf{max}}^{(k)}\left(\ell_{\mathsf{max}}^{(k)}-1\right)}\right)\left(\mathfrak{X}_1\Bigg\Vert \sum_{b \in \Omega} \sum_{u\in \mathcal{U}_{b}}\frac{|\Upsilon_{u}(\widetilde{\tau}_{b}^{\downarrow,{(k)}})|}{|\Upsilon(\widetilde{\bm{\tau}}^{\downarrow,{(k)}})| }\nabla \mathfrak{L}^{(k)}_u(\bm{\omega}^{(k)})
    \Bigg\Vert^2+\mathfrak{X}_2\right)\nonumber\\
    &\le 8\left (\frac{|\Upsilon(\widetilde{\bm{\tau}}^{\downarrow,{(k)}})|-|\Upsilon^{\mathsf{s}}(\widetilde{\bm{\tau}}^{\downarrow,{(k)}})|}{|\Upsilon(\widetilde{\bm{\tau}}^{\downarrow,{(k)}})|}\right)^2\sum_{b \in \Omega} \sum_{u\in \mathcal{U}_{b}}\Bigg(\frac{4 \Theta^2 \beta^2 \eta_k^2 \left(\ell_u^{(k)}-1\right)}{1- 4\eta_k^2\beta^2 \ell_u^{(k)}\left(\ell_u^{(k)}-1\right)} \left(1-\frac{{B}_{u}(\widetilde{\tau}_{b}^{\downarrow,{(k)}})}{|\Upsilon_{u}(\widetilde{\tau}_{b}^{\downarrow,{(k)}})|} \right)  \frac{{(|\Upsilon_{u}(\widetilde{\tau}_{b}^{\downarrow,{(k)}})|-1)}\left(\sigma_{u}(\widetilde{\tau}_{b}^{\downarrow,{(k)}})\right)^2}{|\Upsilon_{u}(\widetilde{\tau}_{b}^{\downarrow,{(k)}})|{B}_{u}(\widetilde{\tau}_{b}^{\downarrow,{(k)}})}\Bigg)\nonumber\\
     &+8\left (\frac{|\Upsilon(\widetilde{\bm{\tau}}^{\downarrow,{(k)}})|-|\Upsilon^{\mathsf{s}}(\widetilde{\bm{\tau}}^{\downarrow,{(k)}})|}{|\Upsilon(\widetilde{\bm{\tau}}^{\downarrow,{(k)}})|}\right)^2\frac{|\Upsilon(\widetilde{\bm{\tau}}^{\downarrow,{(k)}})|}{|\Upsilon_{\mathsf{min}}(\widetilde{\bm{\tau}}^{\downarrow,{(k)}})|}\left(\frac{1}{1- 4\eta_k^2\beta^2 \ell_{\mathsf{max}}^{(k)}\left(\ell_{\mathsf{max}}^{(k)}-1\right)}\right)\left(\mathfrak{X}_1   \Bigg\Vert\nabla \mathfrak{L}^{(k)}(\bm{\omega}^{(k)})\Bigg\Vert^2 + \mathfrak{X}_2 \right)\nonumber\\
     &\le 8\left (\frac{|\Upsilon(\widetilde{\bm{\tau}}^{\downarrow,{(k)}})|-|\Upsilon^{\mathsf{s}}(\widetilde{\bm{\tau}}^{\downarrow,{(k)}})|}{|\Upsilon(\widetilde{\bm{\tau}}^{\downarrow,{(k)}})|}\right)^2\sum_{b \in \Omega} \sum_{u\in \mathcal{U}_{b}}\Bigg(\frac{4 \Theta^2 \beta^2 \eta_k^2 \left(\ell_u^{(k)}-1\right)}{1- 4\eta_k^2\beta^2 \ell_u^{(k)}\left(\ell_u^{(k)}-1\right)} \left(1-\frac{{B}_{u}(\widetilde{\tau}_{b}^{\downarrow,{(k)}})}{|\Upsilon_{u}(\widetilde{\tau}_{b}^{\downarrow,{(k)}})|} \right)  \frac{{(|\Upsilon_{u}(\widetilde{\tau}_{b}^{\downarrow,{(k)}})|-1)}\left(\sigma_{u}(\widetilde{\tau}_{b}^{\downarrow,{(k)}})\right)^2}{|\Upsilon_{u}(\widetilde{\tau}_{b}^{\downarrow,{(k)}})|{B}_{u}(\widetilde{\tau}_{b}^{\downarrow,{(k)}})}\Bigg)\nonumber\\
     &~~~~~~~~+8\left (\frac{|\Upsilon(\widetilde{\bm{\tau}}^{\downarrow,{(k)}})|-|\Upsilon^{\mathsf{s}}(\widetilde{\bm{\tau}}^{\downarrow,{(k)}})|}{|\Upsilon(\widetilde{\bm{\tau}}^{\downarrow,{(k)}})|}\right)^2\frac{|\Upsilon(\widetilde{\bm{\tau}}^{\downarrow,{(k)}})|}{|\Upsilon_{\mathsf{min}}(\widetilde{\bm{\tau}}^{\downarrow,{(k)}})|}\left(\frac{\mathfrak{X}_1 }{1- 4\eta_k^2\beta^2 \ell_{\mathsf{max}}^{(k)}\left(\ell_{\mathsf{max}}^{(k)}-1\right)}\right)\Bigg\Vert\nabla \mathfrak{L}^{(k)}(\bm{\omega}^{(k)})\Bigg\Vert^2\nonumber\\
     &~~~~~~~~+8\left (\frac{|\Upsilon(\widetilde{\bm{\tau}}^{\downarrow,{(k)}})|-|\Upsilon^{\mathsf{s}}(\widetilde{\bm{\tau}}^{\downarrow,{(k)}})|}{|\Upsilon(\widetilde{\bm{\tau}}^{\downarrow,{(k)}})|}\right)^2\frac{|\Upsilon(\widetilde{\bm{\tau}}^{\downarrow,{(k)}})|}{|\Upsilon_{\mathsf{min}}(\widetilde{\bm{\tau}}^{\downarrow,{(k)}})|}\left(\frac{\mathfrak{X}_2}{1- 4\eta_k^2\beta^2 \ell_{\mathsf{max}}^{(k)}\left(\ell_{\mathsf{max}}^{(k)}-1\right)}\right).
\end{align}

Replacing the above result back in \eqref{ineq:main6} yields 
 \begin{align}\label{ineq:main7}
     \hspace{-2.5mm} &\mathbb{E}_k\left[\mathfrak{L}^{({k})}(\bm{\omega}^{(k+1)})\right] {\leq} \mathfrak{L}^{({k})}(\bm{\omega}^{(k)}) +\frac{\eta_{_k}\mathfrak{B}_k}{4}\left(\frac{8 \mathfrak{X}_1 \eta_k^2\beta^2 \left(\ell_{\mathsf{max}}^{(k)}\right)\left(\ell_{\mathsf{max}}^{(k)}-1\right)}{1- 4\eta_k^2\beta^2\ell_{\mathsf{max}}^{(k)}\left(\ell_{\mathsf{max}}^{(k)}-1\right)}-1\right) \left\Vert\nabla{\mathfrak{L}^{({k})}(\bm{\omega}^{(k)})}\right\Vert^2  \nonumber\\
    \hspace{-2.5mm}  &+{2 \beta^2\Theta^2 \eta_k^3\mathfrak{B}_k}\sum_{b \in \Omega} \sum_{u\in \mathcal{U}_{b}}\frac{|\Upsilon_{u}(\widetilde{\tau}_{b}^{\downarrow,{(k)}})|}{|\Upsilon(\widetilde{\bm{\tau}}^{\downarrow,{(k)}})| \ell_u^{(k)}}\frac{\left(\ell_u^{(k)}\right)\left(\ell_u^{(k)}-1\right)}{1- 4\eta_k^2\beta^2 \ell_u^{(k)}\left(\ell_u^{(k)}-1\right)}\left(1-\frac{{B}_{u}(\widetilde{\tau}_{b}^{\downarrow,{(k)}})}{|\Upsilon_{u}(\widetilde{\tau}_{b}^{\downarrow,{(k)}})|} \right)  \frac{{(|\Upsilon_{u}(\widetilde{\tau}_{b}^{\downarrow,{(k)}})|-1)}\left(\sigma_{u}(\widetilde{\tau}_{b}^{\downarrow,{(k)}})\right)^2}{|\Upsilon_{u}(\widetilde{\tau}_{b}^{\downarrow,{(k)}})|{B}_{u}(\widetilde{\tau}_{b}^{\downarrow,{(k)}})}\nonumber\\
     \hspace{-2.5mm} & + \frac{2 \mathfrak{X}_2 \eta_k^3\beta^2 \mathfrak{B}_k \left(\ell_{\mathsf{max}}^{(k)}\right)\left(\ell_{\mathsf{max}}^{(k)}-1\right)}{1- 4\eta_k^2\beta^2\ell_{\mathsf{max}}^{(k)}\left(\ell_{\mathsf{max}}^{(k)}-1\right)}+ \frac{\beta\eta_{_k}^2\mathfrak{B}_k^2}{2} \sum_{b \in \Omega} \sum_{u\in \mathcal{U}_{b}}\frac{\left(\widehat{\lambda}_{u}^{(k)}|\Upsilon_{u}(\widetilde{\tau}_{b}^{\downarrow,{(k)}})|\right)^2}{\left(|\Upsilon^{\mathsf{s}}(\widetilde{\bm{\tau}}^{\downarrow,{(k)}})|\right)^2 \ell^{(k)}_{u}}\left(1-\frac{{B}_{u}(\widetilde{\tau}_{b}^{\downarrow,{(k)}})}{|\Upsilon_{u}(\widetilde{\tau}_{b}^{\downarrow,{(k)}})|} \right) \frac{2(|\Upsilon_{u}(\widetilde{\tau}_{b}^{\downarrow,{(k)}})|-1)\Theta^2\left(\sigma_{u}(\widetilde{\tau}_{b}^{\downarrow,{(k)}})\right)^2}{|\Upsilon_{u}(\widetilde{\tau}_{b}^{\downarrow,{(k)}})|{B}_{u}(\widetilde{\tau}_{b}^{\downarrow,{(k)}})}\nonumber\\
    \hspace{-2.5mm}  &+\frac{24\eta_{_k}\mathfrak{B}_k}{2}\left (\frac{|\Upsilon(\widetilde{\bm{\tau}}^{\downarrow,{(k)}})|-|\Upsilon^{\mathsf{s}}(\widetilde{\bm{\tau}}^{\downarrow,{(k)}})|}{|\Upsilon(\widetilde{\bm{\tau}}^{\downarrow,{(k)}})|}\right)^2\sum_{b \in \Omega} \sum_{u\in \mathcal{U}_{b}}\Bigg(\frac{4 \Theta^2 \beta^2 \eta_k^2 \left(\ell_u^{(k)}-1\right)}{1- 4\eta_k^2\beta^2 \ell_u^{(k)}\left(\ell_u^{(k)}-1\right)} \left(1-\frac{{B}_{u}(\widetilde{\tau}_{b}^{\downarrow,{(k)}})}{|\Upsilon_{u}(\widetilde{\tau}_{b}^{\downarrow,{(k)}})|} \right)  \frac{{(|\Upsilon_{u}(\widetilde{\tau}_{b}^{\downarrow,{(k)}})|-1)}\left(\sigma_{u}(\widetilde{\tau}_{b}^{\downarrow,{(k)}})\right)^2}{|\Upsilon_{u}(\widetilde{\tau}_{b}^{\downarrow,{(k)}})|{B}_{u}(\widetilde{\tau}_{b}^{\downarrow,{(k)}})}\Bigg)\nonumber\\
    \hspace{-2.5mm}   &+\frac{24\eta_{_k}\mathfrak{B}_k}{2}\left (\frac{|\Upsilon(\widetilde{\bm{\tau}}^{\downarrow,{(k)}})|-|\Upsilon^{\mathsf{s}}(\widetilde{\bm{\tau}}^{\downarrow,{(k)}})|}{|\Upsilon(\widetilde{\bm{\tau}}^{\downarrow,{(k)}})|}\right)^2\frac{|\Upsilon(\widetilde{\bm{\tau}}^{\downarrow,{(k)}})|}{|\Upsilon_{\mathsf{min}}(\widetilde{\bm{\tau}}^{\downarrow,{(k)}})|}\left(\frac{\mathfrak{X}_1 }{1- 4\eta_k^2\beta^2 \ell_{\mathsf{max}}^{(k)}\left(\ell_{\mathsf{max}}^{(k)}-1\right)}\right)\Bigg\Vert\nabla \mathfrak{L}^{(k)}(\bm{\omega}^{(k)})\Bigg\Vert^2\nonumber\\
     \hspace{-2.5mm}  &+\frac{24\eta_{_k}\mathfrak{B}_k}{2}\left (\frac{|\Upsilon(\widetilde{\bm{\tau}}^{\downarrow,{(k)}})|-|\Upsilon^{\mathsf{s}}(\widetilde{\bm{\tau}}^{\downarrow,{(k)}})|}{|\Upsilon(\widetilde{\bm{\tau}}^{\downarrow,{(k)}})|}\right)^2\frac{|\Upsilon(\widetilde{\bm{\tau}}^{\downarrow,{(k)}})|}{|\Upsilon_{\mathsf{min}}(\widetilde{\bm{\tau}}^{\downarrow,{(k)}})|}\left(\frac{\mathfrak{X}_2}{1- 4\eta_k^2\beta^2 \ell_{\mathsf{max}}^{(k)}\left(\ell_{\mathsf{max}}^{(k)}-1\right)}\right)\nonumber\\
    \hspace{-2.5mm}  &=\mathfrak{L}^{({k})}(\bm{\omega}^{(k)}) +\frac{\eta_{_k}\mathfrak{B}_k}{4}\left(\frac{8 \mathfrak{X}_1 \eta_k^2\beta^2 \left(\ell_{\mathsf{max}}^{(k)}\right)\left(\ell_{\mathsf{max}}^{(k)}-1\right)+48\mathfrak{X}_1\left (\frac{|\Upsilon(\widetilde{\bm{\tau}}^{\downarrow,{(k)}})|-|\Upsilon^{\mathsf{s}}(\widetilde{\bm{\tau}}^{\downarrow,{(k)}})|}{|\Upsilon(\widetilde{\bm{\tau}}^{\downarrow,{(k)}})|}\right)^2\frac{|\Upsilon(\widetilde{\bm{\tau}}^{\downarrow,{(k)}})|}{|\Upsilon_{\mathsf{min}}(\widetilde{\bm{\tau}}^{\downarrow,{(k)}})|}}{1- 4\eta_k^2\beta^2\ell_{\mathsf{max}}^{(k)}\left(\ell_{\mathsf{max}}^{(k)}-1\right)}-1\right) \left\Vert\nabla{\mathfrak{L}^{({k})}(\bm{\omega}^{(k)})}\right\Vert^2  \nonumber\\
     \hspace{-2.5mm} &+{2 \beta^2\Theta^2 \eta_k^3\mathfrak{B}_k}\sum_{b \in \Omega} \sum_{u\in \mathcal{U}_{b}}\frac{|\Upsilon_{u}(\widetilde{\tau}_{b}^{\downarrow,{(k)}})|}{|\Upsilon(\widetilde{\bm{\tau}}^{\downarrow,{(k)}})| \ell_u^{(k)}}\frac{\left(\ell_u^{(k)}\right)\left(\ell_u^{(k)}-1\right)}{1- 4\eta_k^2\beta^2 \ell_u^{(k)}\left(\ell_u^{(k)}-1\right)}\left(1-\frac{{B}_{u}(\widetilde{\tau}_{b}^{\downarrow,{(k)}})}{|\Upsilon_{u}(\widetilde{\tau}_{b}^{\downarrow,{(k)}})|} \right)  \frac{{(|\Upsilon_{u}(\widetilde{\tau}_{b}^{\downarrow,{(k)}})|-1)}\left(\sigma_{u}(\widetilde{\tau}_{b}^{\downarrow,{(k)}})\right)^2}{|\Upsilon_{u}(\widetilde{\tau}_{b}^{\downarrow,{(k)}})|{B}_{u}(\widetilde{\tau}_{b}^{\downarrow,{(k)}})}\nonumber\\
     \hspace{-2.5mm} & + \frac{2 \mathfrak{X}_2 \eta_k^3\beta^2 \mathfrak{B}_k \left(\ell_{\mathsf{max}}^{(k)}\right)\left(\ell_{\mathsf{max}}^{(k)}-1\right)}{1- 4\eta_k^2\beta^2\ell_{\mathsf{max}}^{(k)}\left(\ell_{\mathsf{max}}^{(k)}-1\right)}+ \frac{\beta\eta_{_k}^2\mathfrak{B}_k^2}{2} \sum_{b \in \Omega} \sum_{u\in \mathcal{U}_{b}}\frac{\left(\widehat{\lambda}_{u}^{(k)}|\Upsilon_{u}(\widetilde{\tau}_{b}^{\downarrow,{(k)}})|\right)^2}{\left(|\Upsilon^{\mathsf{s}}(\widetilde{\bm{\tau}}^{\downarrow,{(k)}})|\right)^2 \ell^{(k)}_{u}}\left(1-\frac{{B}_{u}(\widetilde{\tau}_{b}^{\downarrow,{(k)}})}{|\Upsilon_{u}(\widetilde{\tau}_{b}^{\downarrow,{(k)}})|} \right) \frac{2(|\Upsilon_{u}(\widetilde{\tau}_{b}^{\downarrow,{(k)}})|-1)\Theta^2\left(\sigma_{u}(\widetilde{\tau}_{b}^{\downarrow,{(k)}})\right)^2}{|\Upsilon_{u}(\widetilde{\tau}_{b}^{\downarrow,{(k)}})|{B}_{u}(\widetilde{\tau}_{b}^{\downarrow,{(k)}})}\nonumber\\
      \hspace{-2.5mm}&+\frac{24\eta_{_k}\mathfrak{B}_k}{2}\left (\frac{|\Upsilon(\widetilde{\bm{\tau}}^{\downarrow,{(k)}})|-|\Upsilon^{\mathsf{s}}(\widetilde{\bm{\tau}}^{\downarrow,{(k)}})|}{|\Upsilon(\widetilde{\bm{\tau}}^{\downarrow,{(k)}})|}\right)^2\sum_{b \in \Omega} \sum_{u\in \mathcal{U}_{b}}\Bigg(\frac{4 \Theta^2 \beta^2 \eta_k^2 \left(\ell_u^{(k)}-1\right)}{1- 4\eta_k^2\beta^2 \ell_u^{(k)}\left(\ell_u^{(k)}-1\right)} \left(1-\frac{{B}_{u}(\widetilde{\tau}_{b}^{\downarrow,{(k)}})}{|\Upsilon_{u}(\widetilde{\tau}_{b}^{\downarrow,{(k)}})|} \right)  \frac{{(|\Upsilon_{u}(\widetilde{\tau}_{b}^{\downarrow,{(k)}})|-1)}\left(\sigma_{u}(\widetilde{\tau}_{b}^{\downarrow,{(k)}})\right)^2}{|\Upsilon_{u}(\widetilde{\tau}_{b}^{\downarrow,{(k)}})|{B}_{u}(\widetilde{\tau}_{b}^{\downarrow,{(k)}})}\Bigg)\nonumber\\
      \hspace{-2.5mm} &~~~~~~~~+\frac{24\eta_{_k}\mathfrak{B}_k}{2}\left (\frac{|\Upsilon(\widetilde{\bm{\tau}}^{\downarrow,{(k)}})|-|\Upsilon^{\mathsf{s}}(\widetilde{\bm{\tau}}^{\downarrow,{(k)}})|}{|\Upsilon(\widetilde{\bm{\tau}}^{\downarrow,{(k)}})|}\right)^2\frac{|\Upsilon(\widetilde{\bm{\tau}}^{\downarrow,{(k)}})|}{|\Upsilon_{\mathsf{min}}(\widetilde{\bm{\tau}}^{\downarrow,{(k)}})|}\left(\frac{\mathfrak{X}_2}{1- 4\eta_k^2\beta^2 \ell_{\mathsf{max}}^{(k)}\left(\ell_{\mathsf{max}}^{(k)}-1\right)}\right).
\end{align}
Assuming
\begin{align}
    \frac{8 \mathfrak{X}_1 \eta_k^2\beta^2 \left(\ell_{\mathsf{max}}^{(k)}\right)\left(\ell_{\mathsf{max}}^{(k)}-1\right)+48\mathfrak{X}_1\left (\frac{|\Upsilon(\widetilde{\bm{\tau}}^{\downarrow,{(k)}})|-|\Upsilon^{\mathsf{s}}(\widetilde{\bm{\tau}}^{\downarrow,{(k)}})|}{|\Upsilon(\widetilde{\bm{\tau}}^{\downarrow,{(k)}})|}\right)^2\frac{|\Upsilon(\widetilde{\bm{\tau}}^{\downarrow,{(k)}})|}{|\Upsilon_{\mathsf{min}}(\widetilde{\bm{\tau}}^{\downarrow,{(k)}})|}}{1- 4\eta_k^2\beta^2\ell_{\mathsf{max}}^{(k)}\left(\ell_{\mathsf{max}}^{(k)}-1\right)}<\zeta^{(k)}<1
\end{align}
yields
\begin{align}\label{eta_cond}
    \eta_k<\frac{1}{2\beta} \sqrt{\frac{\zeta^{(k)}- 48\mathfrak{X}_1\left (\frac{|\Upsilon(\widetilde{\bm{\tau}}^{\downarrow,{(k)}})|-|\Upsilon^{\mathsf{s}}(\widetilde{\bm{\tau}}^{\downarrow,{(k)}})|}{|\Upsilon(\widetilde{\bm{\tau}}^{\downarrow,{(k)}})|}\right)^2\frac{|\Upsilon(\widetilde{\bm{\tau}}^{\downarrow,{(k)}})|}{|\Upsilon_{\mathsf{min}}(\widetilde{\bm{\tau}}^{\downarrow,{(k)}})|}}{\ell_{\mathsf{max}}^{(k)}\left(\ell_{\mathsf{max}}^{(k)}-1\right)\left(2\mathfrak{X}_1+\zeta^{(k)}\right)}}.
\end{align}
Applying the above condition on \eqref{ineq:main7} and performing some algebraic manipulations give us
\begin{equation}\label{ineq:main8}
\footnotesize
\begin{aligned}
     \hspace{-2.5mm} &\left\Vert\nabla{\mathfrak{L}^{({k})}(\bm{\omega}^{(k)})}\right\Vert^2 {\leq} \frac{\mathbb{E}_k\left[\mathfrak{L}^{({k})}(\bm{\omega}^{(k)})\right]-\mathbb{E}_k\left[\mathfrak{L}^{({k})}(\bm{\omega}^{(k+1)})\right]}{\frac{\eta_{_k}\mathfrak{B}_k}{4}\left(1-\zeta^{(k)}\right)}\\
     \hspace{-2.5mm} &+\frac{1}{{\frac{\eta_{_k}\mathfrak{B}_k}{4}\left(1-\zeta^{(k)}\right)}}\left({2 \beta^2\Theta^2 \eta_k^3}\sum_{b \in \Omega} \sum_{u\in \mathcal{U}_{b}}\frac{|\Upsilon_{u}(\widetilde{\tau}_{b}^{\downarrow,{(k)}})|}{|\Upsilon(\widetilde{\bm{\tau}}^{\downarrow,{(k)}})| \ell_u^{(k)}}\frac{\left(\ell_u^{(k)}\right)\left(\ell_u^{(k)}-1\right)}{1- 4\eta_k^2\beta^2 \ell_u^{(k)}\left(\ell_u^{(k)}-1\right)}\left(1-\frac{{B}_{u}(\widetilde{\tau}_{b}^{\downarrow,{(k)}})}{|\Upsilon_{u}(\widetilde{\tau}_{b}^{\downarrow,{(k)}})|} \right)  \frac{{(|\Upsilon_{u}(\widetilde{\tau}_{b}^{\downarrow,{(k)}})|-1)}\left(\sigma_{u}(\widetilde{\tau}_{b}^{\downarrow,{(k)}})\right)^2}{|\Upsilon_{u}(\widetilde{\tau}_{b}^{\downarrow,{(k)}})|{B}_{u}(\widetilde{\tau}_{b}^{\downarrow,{(k)}})}\right)\\
     \hspace{-2.5mm} & + \frac{1}{\frac{\eta_{_k}\mathfrak{B}_k}{4}\left(1-\zeta^{(k)}\right)}\left(\frac{2 \mathfrak{X}_2 \eta_k^3\beta^2 \mathfrak{B}_k \left(\ell_{\mathsf{max}}^{(k)}\right)\left(\ell_{\mathsf{max}}^{(k)}-1\right)}{1- 4\eta_k^2\beta^2\ell_{\mathsf{max}}^{(k)}\left(\ell_{\mathsf{max}}^{(k)}-1\right)}+ \frac{\beta\eta_{_k}^2\mathfrak{B}_k^2}{2} \sum_{b \in \Omega} \sum_{u\in \mathcal{U}_{b}}\frac{\left(\widehat{\lambda}_{u}^{(k)}|\Upsilon_{u}(\widetilde{\tau}_{b}^{\downarrow,{(k)}})|\right)^2}{\left(|\Upsilon^{\mathsf{s}}(\widetilde{\bm{\tau}}^{\downarrow,{(k)}})|\right)^2 \ell^{(k)}_{u}}\left(1-\frac{{B}_{u}(\widetilde{\tau}_{b}^{\downarrow,{(k)}})}{|\Upsilon_{u}(\widetilde{\tau}_{b}^{\downarrow,{(k)}})|} \right) \frac{2(|\Upsilon_{u}(\widetilde{\tau}_{b}^{\downarrow,{(k)}})|-1)\Theta^2\left(\sigma_{u}(\widetilde{\tau}_{b}^{\downarrow,{(k)}})\right)^2}{|\Upsilon_{u}(\widetilde{\tau}_{b}^{\downarrow,{(k)}})|{B}_{u}(\widetilde{\tau}_{b}^{\downarrow,{(k)}})}\right)\\
     \hspace{-2.5mm} &+\frac{1}{\frac{\eta_{_k}\mathfrak{B}_k}{4}\left(1-\zeta^{(k)}\right)}\left(\frac{24\eta_{_k}\mathfrak{B}_k}{2}\left (\frac{|\Upsilon(\widetilde{\bm{\tau}}^{\downarrow,{(k)}})|-|\Upsilon^{\mathsf{s}}(\widetilde{\bm{\tau}}^{\downarrow,{(k)}})|}{|\Upsilon(\widetilde{\bm{\tau}}^{\downarrow,{(k)}})|}\right)^2\hspace{-1mm}\sum_{b \in \Omega} \sum_{u\in \mathcal{U}_{b}}\hspace{-1mm}\Bigg(\hspace{-1mm}\frac{4 \Theta^2 \beta^2 \eta_k^2 \left(\ell_u^{(k)}-1\right)}{1- 4\eta_k^2\beta^2 \ell_u^{(k)}\left(\ell_u^{(k)}-1\right)} \hspace{-1mm}\left(1-\frac{{B}_{u}(\widetilde{\tau}_{b}^{\downarrow,{(k)}})}{|\Upsilon_{u}(\widetilde{\tau}_{b}^{\downarrow,{(k)}})|} \right)\hspace{-1mm}  \frac{{(|\Upsilon_{u}(\widetilde{\tau}_{b}^{\downarrow,{(k)}})|-1)}\left(\sigma_{u}(\widetilde{\tau}_{b}^{\downarrow,{(k)}})\right)^2}{|\Upsilon_{u}(\widetilde{\tau}_{b}^{\downarrow,{(k)}})|{B}_{u}(\widetilde{\tau}_{b}^{\downarrow,{(k)}})}\hspace{-1mm}\Bigg)\hspace{-1mm}\right)\\
     \hspace{-2.5mm} &+\frac{1}{\frac{\eta_{_k}\mathfrak{B}_k}{4}\left(1-\zeta^{(k)}\right)}\left(\frac{24\eta_{_k}\mathfrak{B}_k}{2}\left (\frac{|\Upsilon(\widetilde{\bm{\tau}}^{\downarrow,{(k)}})|-|\Upsilon^{\mathsf{s}}(\widetilde{\bm{\tau}}^{\downarrow,{(k)}})|}{|\Upsilon(\widetilde{\bm{\tau}}^{\downarrow,{(k)}})|}\right)^2\frac{|\Upsilon(\widetilde{\bm{\tau}}^{\downarrow,{(k)}})|}{|\Upsilon_{\mathsf{min}}(\widetilde{\bm{\tau}}^{\downarrow,{(k)}})|}\left(\frac{\mathfrak{X}_2}{1- 4\eta_k^2\beta^2 \ell_{\mathsf{max}}^{(k)}\left(\ell_{\mathsf{max}}^{(k)}-1\right)}\right)\right).
\end{aligned}
\end{equation}
Taking total expectation and averaging across global aggregations leads to
\begin{equation}\label{eq:final_one_step}
\footnotesize
  \hspace{-2.5mm}
\begin{aligned}
      \hspace{-2.5mm}&\frac{1}{K} \sum_{k=0}^{K-1}\mathbb{E}\left[\left\Vert\nabla{\mathfrak{L}^{({k})}(\bm{\omega}^{(k)})}\right\Vert^2\right] {\leq} \frac{1}{K} \sum_{k=0}^{K-1}\underbrace{\left[\frac{\mathbb{E}_k\left[\mathfrak{L}^{({k})}(\bm{\omega}^{(k)})\right]-\mathbb{E}_k\left[\mathfrak{L}^{({k})}(\bm{\omega}^{(k+1)})\right]}{\frac{\eta_{_k}\mathfrak{B}_k}{4}\left(1-\zeta^{(k)}\right)}\right]}_{(a)}\\
     \hspace{-2.5mm} &+\frac{1}{K} \sum_{k=0}^{K-1}\vast[\frac{1}{{\frac{\eta_{_k}\mathfrak{B}_k}{4}\left(1-\zeta^{(k)}\right)}}\left({2 \beta^2\Theta^2 \eta_k^3\mathfrak{B}_k}\sum_{b \in \Omega} \sum_{u\in \mathcal{U}_{b}}\frac{|\Upsilon_{u}(\widetilde{\tau}_{b}^{\downarrow,{(k)}})|}{|\Upsilon(\widetilde{\bm{\tau}}^{\downarrow,{(k)}})| \ell_u^{(k)}}\frac{\left(\ell_u^{(k)}\right)\left(\ell_u^{(k)}-1\right)}{1- 4\eta_k^2\beta^2 \ell_u^{(k)}\left(\ell_u^{(k)}-1\right)}\left(1-\frac{{B}_{u}(\widetilde{\tau}_{b}^{\downarrow,{(k)}})}{|\Upsilon_{u}(\widetilde{\tau}_{b}^{\downarrow,{(k)}})|} \right)  \frac{{(|\Upsilon_{u}(\widetilde{\tau}_{b}^{\downarrow,{(k)}})|-1)}\left(\sigma_{u}(\widetilde{\tau}_{b}^{\downarrow,{(k)}})\right)^2}{|\Upsilon_{u}(\widetilde{\tau}_{b}^{\downarrow,{(k)}})|{B}_{u}(\widetilde{\tau}_{b}^{\downarrow,{(k)}})}\right)\\
      \hspace{-2.5mm}& + \frac{1}{\frac{\eta_{_k}\mathfrak{B}_k}{4}\left(1-\zeta^{(k)}\right)}\left(\frac{2 \mathfrak{X}_2 \eta_k^3\beta^2 \mathfrak{B}_k\left(\ell_{\mathsf{max}}^{(k)}\right)\left(\ell_{\mathsf{max}}^{(k)}-1\right)}{1- 4\eta_k^2\beta^2\ell_{\mathsf{max}}^{(k)}\left(\ell_{\mathsf{max}}^{(k)}-1\right)}+ \frac{\beta\eta_{_k}^2\mathfrak{B}_k^2}{2} \sum_{b \in \Omega} \sum_{u\in \mathcal{U}_{b}}\frac{\left(\widehat{\lambda}_{u}^{(k)}|\Upsilon_{u}(\widetilde{\tau}_{b}^{\downarrow,{(k)}})|\right)^2}{\left(|\Upsilon^{\mathsf{s}}(\widetilde{\bm{\tau}}^{\downarrow,{(k)}})|\right)^2 \ell^{(k)}_{u}}\left(1-\frac{{B}_{u}(\widetilde{\tau}_{b}^{\downarrow,{(k)}})}{|\Upsilon_{u}(\widetilde{\tau}_{b}^{\downarrow,{(k)}})|} \right) \frac{2(|\Upsilon_{u}(\widetilde{\tau}_{b}^{\downarrow,{(k)}})|-1)\Theta^2\left(\sigma_{u}(\widetilde{\tau}_{b}^{\downarrow,{(k)}})\right)^2}{|\Upsilon_{u}(\widetilde{\tau}_{b}^{\downarrow,{(k)}})|{B}_{u}(\widetilde{\tau}_{b}^{\downarrow,{(k)}})}\right)\\
     \hspace{-2.5mm} &+\frac{1}{\frac{\eta_{_k}\mathfrak{B}_k}{4}\left(1-\zeta^{(k)}\right)}\left(\frac{24\eta_{_k}\mathfrak{B}_k}{2}\hspace{-1mm}\left (\frac{|\Upsilon(\widetilde{\bm{\tau}}^{\downarrow,{(k)}})|-|\Upsilon^{\mathsf{s}}(\widetilde{\bm{\tau}}^{\downarrow,{(k)}})|}{|\Upsilon(\widetilde{\bm{\tau}}^{\downarrow,{(k)}})|}\right)^2\hspace{-1mm}\sum_{b \in \Omega} \sum_{u\in \mathcal{U}_{b}}\hspace{-1mm}\Bigg(\hspace{-1mm}\frac{4 \Theta^2 \beta^2 \eta_k^2 \left(\ell_u^{(k)}-1\right)}{1- 4\eta_k^2\beta^2 \ell_u^{(k)}\left(\ell_u^{(k)}-1\right)} \left(1-\frac{{B}_{u}(\widetilde{\tau}_{b}^{\downarrow,{(k)}})}{|\Upsilon_{u}(\widetilde{\tau}_{b}^{\downarrow,{(k)}})|} \right)  \frac{{(|\Upsilon_{u}(\widetilde{\tau}_{b}^{\downarrow,{(k)}})|-1)}\left(\sigma_{u}(\widetilde{\tau}_{b}^{\downarrow,{(k)}})\right)^2}{|\Upsilon_{u}(\widetilde{\tau}_{b}^{\downarrow,{(k)}})|{B}_{u}(\widetilde{\tau}_{b}^{\downarrow,{(k)}})}\hspace{-1mm}\Bigg)\hspace{-1mm}\right)\\
     \hspace{-2.5mm} &+\frac{1}{\frac{\eta_{_k}\mathfrak{B}_k}{4}\left(1-\zeta^{(k)}\right)}\left(\frac{24\eta_{_k}\mathfrak{B}_k}{2}\left (\frac{|\Upsilon(\widetilde{\bm{\tau}}^{\downarrow,{(k)}})|-|\Upsilon^{\mathsf{s}}(\widetilde{\bm{\tau}}^{\downarrow,{(k)}})|}{|\Upsilon(\widetilde{\bm{\tau}}^{\downarrow,{(k)}})|}\right)^2\frac{|\Upsilon(\widetilde{\bm{\tau}}^{\downarrow,{(k)}})|}{|\Upsilon_{\mathsf{min}}(\widetilde{\bm{\tau}}^{\downarrow,{(k)}})|}\left(\frac{\mathfrak{X}_2}{1- 4\eta_k^2\beta^2 \ell_{\mathsf{max}}^{(k)}\left(\ell_{\mathsf{max}}^{(k)}-1\right)}\right)\right)\vast].
\end{aligned}
\end{equation}
We next focus on the first term on the right hand side of~\eqref{eq:final_one_step} and aim to upper bound it. Consider wall-clock time window $[T^{(k{-}1)}, T^{(k)})$, where $T^{(k)}$ denoting the duration of global round $k$. Further, consider $T^{\mathsf{idle},(k)}_u{\triangleq} [T^{(k{-}1)}, T^{(k)}) {\setminus} T^{\mathsf{train},(k)}_u$, where $T^{\mathsf{train},(k)}_u{\triangleq} T^{(k{-}1)}{+}\Psi(\mathscr{D}_{u}^{\ndownarrow,(k)}){+}[0,\tau_{u}^{\mathsf{LC},{(k)}}]$. We assume that the data arrival/departures only happen during the idle time, so during the local training period of each FLU $u$, the dataset of $u$ is stationary. Noting that $\mathfrak{L}^{(k)}(\bm{\omega}^{(k)})$ is the loss function when global round $k$ is finished, and we bound it in terms of $\mathfrak{L}^{(k-1)}(\bm{\omega}^{(k)})$, i.e., the loss of the GM of global round $k$ (i.e., $\bm{\omega}^{(k)}$) using the total dataset at the end of global round $k-1$ (i.e., at time $T^{(k-1)}$).
\begin{align}\label{eq:driftLoss}
    \mathfrak{L}^{(k)}(\bm{\omega}^{(k)})&= \mathfrak{L}(\bm{\omega}^{(k)} | \Upsilon(T^{(k)}{+}\Psi(\mathscr{D}_{u}^{\ndownarrow,(k+1)})))\nonumber\\
    &=\sum_{b \in \Omega} \sum_{u\in \mathcal{U}_{b}} \frac{|\Upsilon_u(T^{(k)}{+}\Psi(\mathscr{D}_{u}^{\ndownarrow,(k+1)}))|}{|\Upsilon(T^{(k)}{+}\Psi(\mathscr{D}_{u}^{\ndownarrow,(k+1)}))|}\mathfrak{L}_u(\bm{\omega}^{(k)}|\Upsilon_u{(T^{(k)}{+}\Psi(\mathscr{D}_{u}^{\ndownarrow,(k+1)}))})\nonumber\\
    &=\sum_{b \in \Omega} \sum_{u\in \mathcal{U}_{b}} \frac{|\Upsilon_u(T^{(k-1)}{+}\Psi(\mathscr{D}_{u}^{\ndownarrow,(k)}))|}{|\Upsilon(T^{(k-1)}{+}\Psi(\mathscr{D}_{u}^{\ndownarrow,(k)}))|}\mathfrak{L}_u(\bm{\omega}^{(k)}|\Upsilon_u{(T^{(k-1)}{+}\Psi(\mathscr{D}_{u}^{\ndownarrow,(k)}))})\nonumber\\
    &+ \sum_{b \in \Omega}\sum_{u\in \mathcal{U}_{b}} (\lambda_{u}^{(k)}+\overline{\lambda}_{u}^{(k)})\left(\int_{t=\substack{T^{(k{-}1)}{+}\Psi(\mathscr{D}_{u}^{\ndownarrow,(k)})\\+t_{u}^{\mathsf{LC},{(k)}}}}^{T^{(k)}{+}\Psi(\mathscr{D}_{u}^{\ndownarrow,(k+1)})}\mathfrak{D}_u(t)dt\right)\nonumber\\
    &+ \sum_{b \in \Omega}\sum_{u\in \mathcal{U}_{b}} (1-(\lambda_{u}^{(k)}+\overline{\lambda}_{u}^{(k)}))\left(\int_{t=\substack{T^{(k{-}1)}{+}\Psi(\mathscr{D}_{u}^{\ndownarrow,(k)})}}^{T^{(k)}{+}\Psi(\mathscr{D}_{u}^{\ndownarrow,(k+1)})}\mathfrak{D}_u(t)dt\right)\nonumber\\
    &\leq \mathfrak{L}^{(k-1)}(\bm{\omega}^{(k)})\nonumber\\
    &+ \sum_{b \in \Omega}\sum_{u\in \mathcal{U}_{b}} (\lambda_{u}^{(k)}+\overline{\lambda}_{u}^{(k)})\left(\int_{t=\substack{T^{(k{-}1)}{+}\Psi(\mathscr{D}_{u}^{\ndownarrow,(k)})\\+t_{u}^{\mathsf{LC},{(k)}}}}^{T^{(k)}{+}\Psi(\mathscr{D}_{u}^{\ndownarrow,(k+1)})}\mathfrak{D}^{(k)}_u dt\right)\nonumber\\
    &+ \sum_{b \in \Omega}\sum_{u\in \mathcal{U}_{b}} (1-(\lambda_{u}^{(k)}+\overline{\lambda}_{u}^{(k)}))\left(\int_{t=\substack{T^{(k{-}1)}{+}\Psi(\mathscr{D}_{u}^{\ndownarrow,(k)})}}^{T^{(k)}{+}\Psi(\mathscr{D}_{u}^{\ndownarrow,(k+1)})}\mathfrak{D}^{(k)}_u dt\right)\nonumber\\
    &= \mathfrak{L}^{(k-1)}(\bm{\omega}^{(k)})\nonumber\\
    &+ \sum_{b \in \Omega}\sum_{u\in \mathcal{U}_{b}} (\lambda_{u}^{(k)}+\overline{\lambda}_{u}^{(k)})\left( (T^{(k)}{+}\Psi(\mathscr{D}_{u}^{\ndownarrow,(k+1)}) - T^{(k{-}1)}{-}\Psi(\mathscr{D}_{u}^{\ndownarrow,(k)})-t_{u}^{\mathsf{LC},{(k)}})      \mathfrak{D}^{(k)}_u \right)\nonumber\\
    &+ \sum_{b \in \Omega}\sum_{u\in \mathcal{U}_{b}} (1-(\lambda_{u}^{(k)}+\overline{\lambda}_{u}^{(k)}))\left((T^{(k)}{+}\Psi(\mathscr{D}_{u}^{\ndownarrow,(k+1)}) - T^{(k{-}1)}{-}\Psi(\mathscr{D}_{u}^{\ndownarrow,(k)})) \mathfrak{D}^{(k)}_u\right) \\
    &=\mathfrak{L}^{(k-1)}(\bm{\omega}^{(k)}) \\
    &+ \sum_{b\in\Omega} \sum_{u\in\mathcal{U}_b} \big(\lambda_u^{(k)}+\overline{\lambda}_u^{(k)}\big)\mathfrak{D}^{(k)}_u \left[ T^{(k)} + \Psi\!\big(\mathscr{D}_u^{\ndownarrow,(k+1)}\big) - \big( T^{(k-1)} + \Psi\!\big(\mathscr{D}_u^{\ndownarrow,(k)}\big) + t_u^{\mathsf{LC},(k)} \big) \right] \nonumber\\ 
    & + \sum_{b\in\Omega} \sum_{u\in\mathcal{U}_b} \big(1 - (\lambda_u^{(k)}+\overline{\lambda}_u^{(k)})\big) \mathfrak{D}^{(k)}_u \left[ T^{(k)} + \Psi\!\big(\mathscr{D}_u^{\ndownarrow,(k+1)}\big) - \big( T^{(k-1)} + \Psi\!\big(\mathscr{D}_u^{\ndownarrow,(k)}\big) \big) \right] \nonumber\\
\end{align} 

Let $\Delta T_u^{(k)} \triangleq \big(T^{(k)} - T^{(k-1)}\big) 
+ \Psi\!\big(\mathscr{D}_u^{\ndownarrow,(k+1)}\big) - \Psi\!\big(\mathscr{D}_u^{\ndownarrow,(k)}\big)$, then the above expression simplifies to
\begin{align}
\mathfrak{L}^{(k)}(\bm{\omega}^{(k)})&\leq  \mathfrak{L}^{(k-1)}(\bm{\omega}^{(k)})+\sum_{b\in\Omega} \sum_{u\in\mathcal{U}_b}
\big(\lambda_u^{(k)}+\overline{\lambda}_u^{(k)}\big)
\mathfrak{D}^{(k)}_u \big( \Delta T_u^{(k)} - t_u^{\mathsf{LC},(k)} \big) \nonumber\\
&\quad + \sum_{b\in\Omega} \sum_{u\in\mathcal{U}_b}
\big(1 - (\lambda_u^{(k)}+\overline{\lambda}_u^{(k)})\big)
\mathfrak{D}^{(k)}_u \, \Delta T_u^{(k)} \nonumber\\
&= \mathfrak{L}^{(k-1)}(\bm{\omega}^{(k)})+ \sum_{b\in\Omega} \sum_{u\in\mathcal{U}_b} 
\mathfrak{D}^{(k)}_u 
\left[ (\lambda_u^{(k)}+\overline{\lambda}_u^{(k)})\Delta T_u^{(k)} - (\lambda_u^{(k)}+\overline{\lambda}_u^{(k)}) t_u^{\mathsf{LC},(k)}
+ \Delta T_u^{(k)} - (\lambda_u^{(k)}+\overline{\lambda}_u^{(k)})\Delta T_u^{(k)} \right] \nonumber\\
&= \mathfrak{L}^{(k-1)}(\bm{\omega}^{(k)})+\sum_{b\in\Omega} \sum_{u\in\mathcal{U}_b} 
\mathfrak{D}^{(k)}_u 
\left[ \Delta T_u^{(k)} - (\lambda_u^{(k)}+\overline{\lambda}_u^{(k)}) t_u^{\mathsf{LC},(k)} \right]\\
&= \mathfrak{L}^{(k-1)}(\bm{\omega}^{(k)})+\sum_{b\in\Omega} \sum_{u\in\mathcal{U}_b} 
\mathfrak{D}^{(k)}_u 
\left[ \Delta T_u^{(k)} - \widehat{\lambda}_{u}^{(k)} t_u^{\mathsf{LC},(k)} \right].
\end{align}

where $\mathfrak{D}^{(k)}_u=\max_{t\in T^{\mathsf{Idle},(k)}_u} \mathfrak{D}_u(t)$, $\forall u$. Moreover, $\mathfrak{L}^{(-1)}(\bm{\omega}^{(0)})$ denotes the initial loss of the algorithm before model training starts. Also using the above bound recursively in term $(a)$ of~\eqref{eq:final_one_step}, we get
{\small
\begin{align}\label{eq:final_one_step2}
      \hspace{-2.5mm}&\frac{1}{K} \sum_{k=0}^{K-1}\mathbb{E}\left[\left\Vert\nabla{\mathfrak{L}^{({k})}(\bm{\omega}^{(k)})}\right\Vert^2\right] {\leq} \frac{1}{K} \sum_{k=0}^{K-1}\frac{\mathbb{E}_k\left[\mathfrak{L}^{(k-1)}(\bm{\omega}^{(k)})\right]-\mathbb{E}_k\left[\mathfrak{L}^{({k})}(\bm{\omega}^{(k+1)})\right]}{\frac{\eta_{_k}\mathfrak{B}_k}{4}\left(1-\zeta^{(k)}\right)}+\frac{\sum_{b \in \Omega}\sum_{u\in \mathcal{U}_{b}} \mathfrak{D}^{(k)}_u\left(\Delta T_u^{(k)}-\widehat{\lambda}_{u}^{(k)}t_{u}^{\mathsf{LC},{(k)}}\right) }{\frac{\eta_{_k}\mathfrak{B}_k}{4}\left(1-\zeta^{(k)}\right)}\nonumber\\
      \hspace{-2.5mm}&+\frac{1}{K} \sum_{k=0}^{K-1}\vast[\frac{1}{{\frac{\eta_{_k}}{4}\left(1-\zeta^{(k)}\right)}}\left({2 \beta^2\Theta^2 \eta_k^3}\sum_{b \in \Omega} \sum_{u\in \mathcal{U}_{b}}\frac{|\Upsilon_{u}(\widetilde{\tau}_{b}^{\downarrow,{(k)}})|}{|\Upsilon(\widetilde{\bm{\tau}}^{\downarrow,{(k)}})| }\frac{\left(\ell_u^{(k)}-1\right)}{1- 4\eta_k^2\beta^2 \ell_u^{(k)}\left(\ell_u^{(k)}-1\right)}\left(1-\frac{{B}_{u}(\widetilde{\tau}_{b}^{\downarrow,{(k)}})}{|\Upsilon_{u}(\widetilde{\tau}_{b}^{\downarrow,{(k)}})|} \right)  \frac{{(|\Upsilon_{u}(\widetilde{\tau}_{b}^{\downarrow,{(k)}})|-1)}\left(\sigma_{u}(\widetilde{\tau}_{b}^{\downarrow,{(k)}})\right)^2}{|\Upsilon_{u}(\widetilde{\tau}_{b}^{\downarrow,{(k)}})|{B}_{u}(\widetilde{\tau}_{b}^{\downarrow,{(k)}})}\right)\nonumber\\
     \hspace{-2.5mm} & + \frac{1}{\frac{\eta_{_k}}{4}\left(1-\zeta^{(k)}\right)}\left(\frac{2 \mathfrak{X}_2 \eta_k^3\beta^2 \left(\ell_{\mathsf{max}}^{(k)}\right)\left(\ell_{\mathsf{max}}^{(k)}-1\right)}{1- 4\eta_k^2\beta^2\ell_{\mathsf{max}}^{(k)}\left(\ell_{\mathsf{max}}^{(k)}-1\right)}+ \frac{\beta\eta_{_k}^2\mathfrak{B}_k}{2} \sum_{b \in \Omega} \sum_{u\in \mathcal{U}_{b}}\frac{\left(\widehat{\lambda}_{u}^{(k)}|\Upsilon_{u}(\widetilde{\tau}_{b}^{\downarrow,{(k)}})|\right)^2}{\left(|\Upsilon^{\mathsf{s}}(\widetilde{\bm{\tau}}^{\downarrow,{(k)}})|\right)^2 \ell^{(k)}_{u}}\left(1-\frac{{B}_{u}(\widetilde{\tau}_{b}^{\downarrow,{(k)}})}{|\Upsilon_{u}(\widetilde{\tau}_{b}^{\downarrow,{(k)}})|} \right) \frac{2(|\Upsilon_{u}(\widetilde{\tau}_{b}^{\downarrow,{(k)}})|-1)\Theta^2\left(\sigma_{u}(\widetilde{\tau}_{b}^{\downarrow,{(k)}})\right)^2}{|\Upsilon_{u}(\widetilde{\tau}_{b}^{\downarrow,{(k)}})|{B}_{u}(\widetilde{\tau}_{b}^{\downarrow,{(k)}})}\right)\nonumber\\
     \hspace{-2.5mm} &+\frac{1}{\frac{\eta_{_k}}{4}\left(1-\zeta^{(k)}\right)}\left(\frac{24}{2}\eta_{_k}\hspace{-1mm}\left (\frac{|\Upsilon(\widetilde{\bm{\tau}}^{\downarrow,{(k)}})|-|\Upsilon^{\mathsf{s}}(\widetilde{\bm{\tau}}^{\downarrow,{(k)}})|}{|\Upsilon(\widetilde{\bm{\tau}}^{\downarrow,{(k)}})|}\right)^2\hspace{-1mm}\sum_{b \in \Omega} \sum_{u\in \mathcal{U}_{b}}\hspace{-1mm}\Bigg(\hspace{-1mm}\frac{4 \Theta^2 \beta^2 \eta_k^2 \left(\ell_u^{(k)}-1\right)}{1- 4\eta_k^2\beta^2 \ell_u^{(k)}\left(\ell_u^{(k)}-1\right)} \left(1-\frac{{B}_{u}(\widetilde{\tau}_{b}^{\downarrow,{(k)}})}{|\Upsilon_{u}(\widetilde{\tau}_{b}^{\downarrow,{(k)}})|} \right)  \frac{{(|\Upsilon_{u}(\widetilde{\tau}_{b}^{\downarrow,{(k)}})|-1)}\left(\sigma_{u}(\widetilde{\tau}_{b}^{\downarrow,{(k)}})\right)^2}{|\Upsilon_{u}(\widetilde{\tau}_{b}^{\downarrow,{(k)}})|{B}_{u}(\widetilde{\tau}_{b}^{\downarrow,{(k)}})}\hspace{-1mm}\Bigg)\hspace{-1mm}\right)\nonumber\\
     \hspace{-2.5mm} &+\frac{1}{\frac{\eta_{_k}}{4}\left(1-\zeta^{(k)}\right)}\left(\frac{24}{2}\eta_{_k}\left (\frac{|\Upsilon(\widetilde{\bm{\tau}}^{\downarrow,{(k)}})|-|\Upsilon^{\mathsf{s}}(\widetilde{\bm{\tau}}^{\downarrow,{(k)}})|}{|\Upsilon(\widetilde{\bm{\tau}}^{\downarrow,{(k)}})|}\right)^2\frac{|\Upsilon(\widetilde{\bm{\tau}}^{\downarrow,{(k)}})|}{|\Upsilon_{\mathsf{min}}(\widetilde{\bm{\tau}}^{\downarrow,{(k)}})|}\left(\frac{\mathfrak{X}_2}{1- 4\eta_k^2\beta^2 \ell_{\mathsf{max}}^{(k)}\left(\ell_{\mathsf{max}}^{(k)}-1\right)}\right)\right)\vast].
\end{align}
}%
Performing some algebraic operations leads to
{\small
\begin{align}\label{eq:final_one_step3}
    &\frac{1}{K} \sum_{k=0}^{K-1}\mathbb{E}\left[\left\Vert\nabla{\mathfrak{L}^{({k})}(\bm{\omega}^{(k)})}\right\Vert^2\right] {\leq} \frac{4}{K} \sum_{k=0}^{K-1}\frac{\mathbb{E}_k\left[\mathfrak{L}^{(k-1)}(\bm{\omega}^{(k)})\right]-\mathbb{E}_k\left[\mathfrak{L}^{({k})}(\bm{\omega}^{(k+1)})\right]}{\eta_{_k}\mathfrak{B}_k\left(1-\zeta^{(k)}\right)}+\frac{\sum_{b \in \Omega}\sum_{u\in \mathcal{U}_{b}} \mathfrak{D}^{(k)}_u\left(\Delta T_u^{(k)}-\widehat{\lambda}_{u}^{(k)}t_{u}^{\mathsf{LC},{(k)}}\right) }{\eta_{_k}\mathfrak{B}_k\left(1-\zeta^{(k)}\right)}\nonumber\\
    &+\frac{8}{K} \sum_{k=0}^{K-1}\vast[\frac{1}{{\left(1-\zeta^{(k)}\right)}}\left({\beta^2\Theta^2 \eta_k^2}\sum_{b \in \Omega} \sum_{u\in \mathcal{U}_{b}}\frac{|\Upsilon_{u}(\widetilde{\tau}_{b}^{\downarrow,{(k)}})|}{|\Upsilon(\widetilde{\bm{\tau}}^{\downarrow,{(k)}})| }\frac{\left(\ell_u^{(k)}-1\right)}{1- 4\eta_k^2\beta^2 \ell_u^{(k)}\left(\ell_u^{(k)}-1\right)}\left(1-\frac{{B}_{u}(\widetilde{\tau}_{b}^{\downarrow,{(k)}})}{|\Upsilon_{u}(\widetilde{\tau}_{b}^{\downarrow,{(k)}})|} \right)  \frac{{(|\Upsilon_{u}(\widetilde{\tau}_{b}^{\downarrow,{(k)}})|-1)}\left(\sigma_{u}(\widetilde{\tau}_{b}^{\downarrow,{(k)}})\right)^2}{|\Upsilon_{u}(\widetilde{\tau}_{b}^{\downarrow,{(k)}})|{B}_{u}(\widetilde{\tau}_{b}^{\downarrow,{(k)}})}\right)\nonumber\\
    & + \frac{1}{\left(1-\zeta^{(k)}\right)}\left(\frac{\mathfrak{X}_2 \eta_k^2\beta^2 \left(\ell_{\mathsf{max}}^{(k)}\right)\left(\ell_{\mathsf{max}}^{(k)}-1\right)}{1- 4\eta_k^2\beta^2\ell_{\mathsf{max}}^{(k)}\left(\ell_{\mathsf{max}}^{(k)}-1\right)}+ \frac{\beta\eta_{_k}\mathfrak{B}_k}{2} \sum_{b \in \Omega} \sum_{u\in \mathcal{U}_{b}}\frac{\left(\widehat{\lambda}_{u}^{(k)}|\Upsilon_{u}(\widetilde{\tau}_{b}^{\downarrow,{(k)}})|\right)^2}{\left(|\Upsilon^{\mathsf{s}}(\widetilde{\bm{\tau}}^{\downarrow,{(k)}})|\right)^2 \ell^{(k)}_{u}}\left(1-\frac{{B}_{u}(\widetilde{\tau}_{b}^{\downarrow,{(k)}})}{|\Upsilon_{u}(\widetilde{\tau}_{b}^{\downarrow,{(k)}})|} \right) \frac{(|\Upsilon_{u}(\widetilde{\tau}_{b}^{\downarrow,{(k)}})|-1)\Theta^2\left(\sigma_{u}(\widetilde{\tau}_{b}^{\downarrow,{(k)}})\right)^2}{|\Upsilon_{u}(\widetilde{\tau}_{b}^{\downarrow,{(k)}})|{B}_{u}(\widetilde{\tau}_{b}^{\downarrow,{(k)}})}\right)\nonumber\\
    &+\frac{24}{\left(1-\zeta^{(k)}\right)}\left(\left (\frac{|\Upsilon(\widetilde{\bm{\tau}}^{\downarrow,{(k)}})|-|\Upsilon^{\mathsf{s}}(\widetilde{\bm{\tau}}^{\downarrow,{(k)}})|}{|\Upsilon(\widetilde{\bm{\tau}}^{\downarrow,{(k)}})|}\right)^2\sum_{b \in \Omega} \sum_{u\in \mathcal{U}_{b}}\Bigg(\frac{\Theta^2 \beta^2 \eta_k^2 \left(\ell_u^{(k)}-1\right)}{1- 4\eta_k^2\beta^2 \ell_u^{(k)}\left(\ell_u^{(k)}-1\right)} \left(1-\frac{{B}_{u}(\widetilde{\tau}_{b}^{\downarrow,{(k)}})}{|\Upsilon_{u}(\widetilde{\tau}_{b}^{\downarrow,{(k)}})|} \right)  \frac{{(|\Upsilon_{u}(\widetilde{\tau}_{b}^{\downarrow,{(k)}})|-1)}\left(\sigma_{u}(\widetilde{\tau}_{b}^{\downarrow,{(k)}})\right)^2}{|\Upsilon_{u}(\widetilde{\tau}_{b}^{\downarrow,{(k)}})|{B}_{u}(\widetilde{\tau}_{b}^{\downarrow,{(k)}})}\Bigg)\right)\nonumber\\
    &+\frac{6}{\left(1-\zeta^{(k)}\right)}\left(\left (\frac{|\Upsilon(\widetilde{\bm{\tau}}^{\downarrow,{(k)}})|-|\Upsilon^{\mathsf{s}}(\widetilde{\bm{\tau}}^{\downarrow,{(k)}})|}{|\Upsilon(\widetilde{\bm{\tau}}^{\downarrow,{(k)}})|}\right)^2\frac{|\Upsilon(\widetilde{\bm{\tau}}^{\downarrow,{(k)}})|}{|\Upsilon_{\mathsf{min}}(\widetilde{\bm{\tau}}^{\downarrow,{(k)}})|}\left(\frac{\mathfrak{X}_2}{1- 4\eta_k^2\beta^2 \ell_{\mathsf{max}}^{(k)}\left(\ell_{\mathsf{max}}^{(k)}-1\right)}\right)\right)\vast].
\end{align}
}%
which concludes the proof.
\newpage

\section{Proof of Theorem~\ref{th:sufficient_conditions}}\label{app:th:sufficient_conditions}
Considering \eqref{eq:final_one_step3}, we have
\begin{equation}\label{eq:final_one_stepP3}
\small
\begin{aligned}
     \hspace{-2.5mm} &\frac{1}{K} \sum_{k=0}^{K-1}\mathbb{E}\left[\left\Vert\nabla{\mathfrak{L}^{({k})}(\bm{\omega}^{(k)})}\right\Vert^2\right] {\leq} \frac{4}{K} \sum_{k=0}^{K-1}\left[\frac{\mathbb{E}_k\left[\mathfrak{L}^{(k-1)}(\bm{\omega}^{(k)})\right]-\mathbb{E}_k\left[\mathfrak{L}^{({k})}(\bm{\omega}^{(k+1)})\right]}{\eta_{_k}\mathfrak{B}_k\left(1-\zeta^{(k)}\right)}\right]\\
     \hspace{-2.5mm} &+\underbrace{\sum_{k=0}^{K-1}\left[\frac{4}{K} \frac{\sum_{b \in \Omega}\sum_{u\in \mathcal{U}_{b}} \mathfrak{D}^{(k)}_u\left(\Delta T_u^{(k)}-\widehat{\lambda}_{u}^{(k)}t_{u}^{\mathsf{LC},{(k)}}\right)}{\eta_{_k}\mathfrak{B}_k\left(1-\zeta^{(k)}\right)}\right]}_{(a)}\\
      \hspace{-2.5mm}&+\frac{8}{K} \sum_{k=0}^{K-1}\vast[\frac{1}{{\left(1-\zeta^{(k)}\right)}}\left({\beta^2\Theta^2 \eta_k^2}\sum_{b \in \Omega} \sum_{u\in \mathcal{U}_{b}}\frac{|\Upsilon_{u}(\widetilde{\tau}_{b}^{\downarrow,{(k)}})|}{|\Upsilon(\widetilde{\bm{\tau}}^{\downarrow,{(k)}})| \ell_u^{(k)}}\frac{\left(\ell_u^{(k)}\right)\left(\ell_u^{(k)}-1\right)}{1- 4\eta_k^2\beta^2 \ell_u^{(k)}\left(\ell_u^{(k)}-1\right)}\left(1-\frac{{B}_{u}(\widetilde{\tau}_{b}^{\downarrow,{(k)}})}{|\Upsilon_{u}(\widetilde{\tau}_{b}^{\downarrow,{(k)}})|} \right)  \frac{{(|\Upsilon_{u}(\widetilde{\tau}_{b}^{\downarrow,{(k)}})|-1)}\left(\sigma_{u}(\widetilde{\tau}_{b}^{\downarrow,{(k)}})\right)^2}{|\Upsilon_{u}(\widetilde{\tau}_{b}^{\downarrow,{(k)}})|{B}_{u}(\widetilde{\tau}_{b}^{\downarrow,{(k)}})}\right)\\
     \hspace{-2.5mm} & + \frac{1}{\left(1-\zeta^{(k)}\right)}\left(\frac{\mathfrak{X}_2 \eta_k^2\beta^2 \left(\ell_{\mathsf{max}}^{(k)}\right)\left(\ell_{\mathsf{max}}^{(k)}-1\right)}{1- 4\eta_k^2\beta^2\ell_{\mathsf{max}}^{(k)}\left(\ell_{\mathsf{max}}^{(k)}-1\right)}+ \frac{\beta\eta_{_k}\mathfrak{B}_k}{2} \sum_{b \in \Omega} \sum_{u\in \mathcal{U}_{b}}\frac{\left(\widehat{\lambda}_{u}^{(k)}|\Upsilon_{u}(\widetilde{\tau}_{b}^{\downarrow,{(k)}})|\right)^2}{\left(|\Upsilon^{\mathsf{s}}(\widetilde{\bm{\tau}}^{\downarrow,{(k)}})|\right)^2 \ell^{(k)}_{u}}\left(1-\frac{{B}_{u}(\widetilde{\tau}_{b}^{\downarrow,{(k)}})}{|\Upsilon_{u}(\widetilde{\tau}_{b}^{\downarrow,{(k)}})|} \right) \frac{(|\Upsilon_{u}(\widetilde{\tau}_{b}^{\downarrow,{(k)}})|-1)\Theta^2\left(\sigma_{u}(\widetilde{\tau}_{b}^{\downarrow,{(k)}})\right)^2}{|\Upsilon_{u}(\widetilde{\tau}_{b}^{\downarrow,{(k)}})|{B}_{u}(\widetilde{\tau}_{b}^{\downarrow,{(k)}})}\right)\\
      \hspace{-2.5mm}&+\frac{24}{\left(1-\zeta^{(k)}\right)}\left(\left (\frac{|\Upsilon(\widetilde{\bm{\tau}}^{\downarrow,{(k)}})|-|\Upsilon^{\mathsf{s}}(\widetilde{\bm{\tau}}^{\downarrow,{(k)}})|}{|\Upsilon(\widetilde{\bm{\tau}}^{\downarrow,{(k)}})|}\right)^2\sum_{b \in \Omega} \sum_{u\in \mathcal{U}_{b}}\Bigg(\frac{\Theta^2 \beta^2 \eta_k^2 \left(\ell_u^{(k)}-1\right)}{1- 4\eta_k^2\beta^2 \ell_u^{(k)}\left(\ell_u^{(k)}-1\right)} \left(1-\frac{{B}_{u}(\widetilde{\tau}_{b}^{\downarrow,{(k)}})}{|\Upsilon_{u}(\widetilde{\tau}_{b}^{\downarrow,{(k)}})|} \right)  \frac{{(|\Upsilon_{u}(\widetilde{\tau}_{b}^{\downarrow,{(k)}})|-1)}\left(\sigma_{u}(\widetilde{\tau}_{b}^{\downarrow,{(k)}})\right)^2}{|\Upsilon_{u}(\widetilde{\tau}_{b}^{\downarrow,{(k)}})|{B}_{u}(\widetilde{\tau}_{b}^{\downarrow,{(k)}})}\Bigg)\right)\vast]\\
      \hspace{-2.5mm}&+\underbrace{\sum_{k=0}^{K-1}\left[\frac{48}{K \left(1-\zeta^{(k)}\right)} \left (\frac{|\Upsilon(\widetilde{\bm{\tau}}^{\downarrow,{(k)}})|-|\Upsilon^{\mathsf{s}}(\widetilde{\bm{\tau}}^{\downarrow,{(k)}})|}{|\Upsilon(\widetilde{\bm{\tau}}^{\downarrow,{(k)}})|}\right)^2\frac{|\Upsilon(\widetilde{\bm{\tau}}^{\downarrow,{(k)}})|}{|\Upsilon_{\mathsf{min}}(\widetilde{\bm{\tau}}^{\downarrow,{(k)}})|}\left(\frac{\mathfrak{X}_2}{1- 4\eta_k^2\beta^2 \ell_{\mathsf{max}}^{(k)}\left(\ell_{\mathsf{max}}^{(k)}-1\right)}\right)\right]}_{(b)}.
\end{aligned}
\end{equation}

Using the first $k$ terms of geometric series we bound $(a)$ in \eqref{eq:final_one_stepP3} as follows:
\begin{align}\label{error1}
    \frac{1}{K}\sum_{k=0}^{K-1}\left[ \frac{4\sum_{b \in \Omega}\sum_{u\in \mathcal{U}_{b}} \mathfrak{D}^{(k)}_u\left(\Delta T_u^{(k)}-\widehat{\lambda}_{u}^{(k)}t_{u}^{\mathsf{LC},{(k)}}\right)  }{\eta_{_k}\mathfrak{B}_k\left(1-\zeta^{(k)}\right)}\right]\le\frac{1}{K}\sum_{k=0}^{K-1} \vartheta^k=\frac{1}{K}\frac{\vartheta^{K}-1}{\vartheta-1},
\end{align}
where $|\vartheta|<1$. Considering $\Delta T_u^{(k)} \triangleq \big(T^{(k)} - T^{(k-1)}\big) + \Psi\!\big(\mathscr{D}_u^{\ndownarrow,(k+1)}\big) - \Psi\!\big(\mathscr{D}_u^{\ndownarrow,(k)}\big)$, we have
\begin{align}\label{error1222}
     \frac{\left(4\sum_{b \in \Omega}\sum_{u\in \mathcal{U}_{b}} \mathfrak{D}^{(k)}_u \big(T^{(k)} - T^{(k-1)}\big) + 4\sum_{b \in \Omega}\sum_{u\in \mathcal{U}_{b}} \mathfrak{D}^{(k)}_u\left(\Psi\!\big(\mathscr{D}_u^{\ndownarrow,(k+1)}\big) - \Psi\!\big(\mathscr{D}_u^{\ndownarrow,(k)}\big)\right) -4\sum_{b \in \Omega}\sum_{u\in \mathcal{U}_{b}} \mathfrak{D}^{(k)}_u\widehat{\lambda}_{u}^{(k)}t_{u}^{\mathsf{LC},{(k)}}\right)  }{\eta_{_k}\mathfrak{B}_k\left(1-\zeta^{(k)}\right)}\le \vartheta^k.
\end{align}
Solving the above inequality for the duration of global training round $k$ (i.e., $T^{(k)}-T^{(k{-}1)}$) give us the following \textit{sufficient condition} ($\bm{\mathfrak{C}^{(\mathsf{L})}}$) for training latency of global round $k$ (i.e., $T^{(k)}-T^{(k{-}1)}$):
\begin{align}\label{necess:gamma_upper}
    T^{(k)}-T^{(k{-}1)} \le \frac{4\sum_{b \in \Omega}\sum_{u\in \mathcal{U}_{b}} \mathfrak{D}^{(k)}_u \left(\widehat{\lambda}_{u}^{(k)}t_{u}^{\mathsf{LC},{(k)}}-\Psi\!\big(\mathscr{D}_u^{\ndownarrow,(k+1)}\big) + \Psi\!\big(\mathscr{D}_u^{\ndownarrow,(k)}\big)\right)+ \eta_{_k}\mathfrak{B}_k\left(1-\zeta^{(k)}\right)\vartheta^k}{4\sum_{b \in \Omega}\sum_{u\in \mathcal{U}_{b}} \mathfrak{D}^{(k)}_u}.
\end{align}

Subsequently, if sufficient condition $\bm{\mathfrak{C}^{(\mathsf{L})}}$ in \eqref{necess:gamma_upper} holds, taking limit from \eqref{error1} when $K\rightarrow \infty$ leads to the following equality:
\begin{align}\label{error1111}
    \lim_{K\rightarrow \infty}\frac{1}{K}\sum_{k=0}^{K-1}\left[\frac{4\sum_{b \in \Omega}\sum_{u\in \mathcal{U}_{b}} \mathfrak{D}^{(k)}_u\left(\Delta T_u^{(k)}-\widehat{\lambda}_{u}^{(k)}t_{u}^{\mathsf{LC},{(k)}}\right)}{\eta_{_k}\mathfrak{B}_k\left(1-\zeta^{(k)}\right)}\right]\le\lim_{K\rightarrow \infty}\frac{1}{K}\sum_{k=0}^{K-1} \vartheta^k=\lim_{K\rightarrow \infty}\frac{1}{K}\frac{\vartheta^{K}-1}{\vartheta-1}= 0,
\end{align}
which concludes the proof of sufficient condition $\bm{\mathfrak{C}^{(\mathsf{L})}}$ on $T^{(k)}-T^{(k{-}1)}$.

Similarly, using the first $k$ terms of geometric series we bound term $(b)$ in \eqref{eq:final_one_stepP3} as follows:
\begin{equation}\label{error2}
    \frac{1}{K}\sum_{k=0}^{K-1}\left[\frac{48}{\left(1-\zeta^{(k)}\right)} \left (\frac{|\Upsilon(\widetilde{\bm{\tau}}^{\downarrow,{(k)}})|-|\Upsilon^{\mathsf{s}}(\widetilde{\bm{\tau}}^{\downarrow,{(k)}})|}{|\Upsilon(\widetilde{\bm{\tau}}^{\downarrow,{(k)}})|}\right)^2\frac{|\Upsilon(\widetilde{\bm{\tau}}^{\downarrow,{(k)}})|}{|\Upsilon_{\mathsf{min}}(\widetilde{\bm{\tau}}^{\downarrow,{(k)}})|}\left(\frac{\mathfrak{X}_2}{1- 4\eta_k^2\beta^2 \ell_{\mathsf{max}}^{(k)}\left(\ell_{\mathsf{max}}^{(k)}-1\right)}\right)\right]\le\frac{1}{K}\sum_{k=0}^{K-1} \varpi^k= \frac{1}{K}\frac{\varpi^{K}-1}{\varpi-1},
\end{equation}
where $|\varpi|<1$. Accordingly, we have:
\begin{equation}\label{error244444}
    \frac{48}{\left(1-\zeta^{(k)}\right)} \left (\frac{|\Upsilon(\widetilde{\bm{\tau}}^{\downarrow,{(k)}})|-|\Upsilon^{\mathsf{s}}(\widetilde{\bm{\tau}}^{\downarrow,{(k)}})|}{|\Upsilon(\widetilde{\bm{\tau}}^{\downarrow,{(k)}})|}\right)^2\frac{|\Upsilon(\widetilde{\bm{\tau}}^{\downarrow,{(k)}})|}{|\Upsilon_{\mathsf{min}}(\widetilde{\bm{\tau}}^{\downarrow,{(k)}})|}\left(\frac{\mathfrak{X}_2}{1- 4\eta_k^2\beta^2 \ell_{\mathsf{max}}^{(k)}\left(\ell_{\mathsf{max}}^{(k)}-1\right)}\right)\le \varpi^k,
\end{equation}
which give us the following \textit{sufficient condition} ($\bm{\mathfrak{C}^{(\Upsilon)}}$) for cumulative data size of recruited FLUs at global training round $k$ (i.e., $|\Upsilon^{\mathsf{s}}(\widetilde{\bm{\tau}}^{\downarrow,{(k)}})|$)
\begin{align}\label{necess:recruitment_upper}
    \left (|\Upsilon(\widetilde{\bm{\tau}}^{\downarrow,{(k)}})|-|\Upsilon^{\mathsf{s}}(\widetilde{\bm{\tau}}^{\downarrow,{(k)}})|\right)^2\le \frac{|\Upsilon(\widetilde{\bm{\tau}}^{\downarrow,{(k)}})|\varpi^k  \left(1-\zeta^{(k)}\right)|\Upsilon_{\mathsf{min}}(\widetilde{\bm{\tau}}^{\downarrow,{(k)}})|\left(1- 4\eta_k^2\beta^2 \ell_{\mathsf{max}}^{(k)}\left(\ell_{\mathsf{max}}^{(k)}-1\right)\right)}{48\mathfrak{X}_2}.
\end{align}
Subsequently, if sufficient condition $\bm{\mathfrak{C}^{(\Upsilon)}}$ in \eqref{necess:recruitment_upper} holds, taking limit from \eqref{error2} when $K\rightarrow \infty$ leads to the following equality:
\begin{align}\label{error22}
    &\lim_{K\rightarrow \infty}\frac{1}{K}\sum_{k=0}^{K-1}\left[\frac{48}{\left(1-\zeta^{(k)}\right)} \left (\frac{|\Upsilon(\widetilde{\bm{\tau}}^{\downarrow,{(k)}})|-|\Upsilon^{\mathsf{s}}(\widetilde{\bm{\tau}}^{\downarrow,{(k)}})|}{|\Upsilon(\widetilde{\bm{\tau}}^{\downarrow,{(k)}})|}\right)^2\frac{|\Upsilon(\widetilde{\bm{\tau}}^{\downarrow,{(k)}})|}{|\Upsilon_{\mathsf{min}}(\widetilde{\bm{\tau}}^{\downarrow,{(k)}})|}\left(\frac{\mathfrak{X}_2}{1- 4\eta_k^2\beta^2 \ell_{\mathsf{max}}^{(k)}\left(\ell_{\mathsf{max}}^{(k)}-1\right)}\right)\right]\nonumber\\
    &~~~~~~~~~~~~~~~~~~~~~~~~~~~~~~~~~~~~~~~~~~~~~~~~~~~~~~~~~~~~~~~~~~~~\le\lim_{K\rightarrow \infty}\frac{1}{K}\sum_{k=0}^{K-1} \varpi^k= \lim_{K\rightarrow \infty}\frac{1}{K}\frac{\varpi^{K}-1}{\varpi-1}=0.
\end{align}
Furthermore, for $\zeta^{(k)}>\widehat{\zeta}^{(k)}$, inequality \eqref{eta_cond} imposes the following condition on $|\Upsilon^{\mathsf{s}}(\widetilde{\bm{\tau}}^{\downarrow,{(k)}})|$:
\begin{equation}\label{eq:suffi_Upsi}
    \left (|\Upsilon(\widetilde{\bm{\tau}}^{\downarrow,{(k)}})|-|\Upsilon^{\mathsf{s}}(\widetilde{\bm{\tau}}^{\downarrow,{(k)}})|\right)^2<\frac{|\Upsilon(\widetilde{\bm{\tau}}^{\downarrow,{(k)}})|\widehat{\zeta}^{(k)}|\Upsilon_{\mathsf{min}}(\widetilde{\bm{\tau}}^{\downarrow,{(k)}})|}{48\mathfrak{X}_1}.
\end{equation}
Considering \eqref{necess:recruitment_upper} and \eqref{eq:suffi_Upsi}, we have 
\begin{equation}\label{suf:main:recruitment}
    \begin{aligned}
        &\left(\hspace{-0.5mm}|\Upsilon(\widetilde{\bm{\tau}}^{\downarrow,{(k)}})|{-}|\Upsilon^{\mathsf{s}}(\widetilde{\bm{\tau}}^{\downarrow,{(k)}})|\hspace{-0.5mm}\right)^2{<}\hspace{-0.5mm}|\Upsilon(\widetilde{\bm{\tau}}^{\downarrow,{(k)}})|\min\hspace{-0.5mm}\left\{\hspace{-0.5mm}\frac{\widehat{\zeta}^{(k)}|\Upsilon_{\mathsf{min}}(\widetilde{\bm{\tau}}^{\downarrow,{(k)}})|}{(48\mathfrak{X}_1)},\frac{\mathfrak{N}^{(k)}}{(48\mathfrak{X}_2)}\hspace{-0.5mm}\right\}\hspace{-0.5mm},
        \hspace{-4mm}
    \end{aligned}
\end{equation}
where
\begin{align}
    \mathfrak{N}^{(k)}&=\varpi^k \left(1-\zeta^{(k)}\right)|\Upsilon_{\mathsf{min}}(\widetilde{\bm{\tau}}^{\downarrow,{(k)}})|\left(1- 4\eta_k^2\beta^2 \ell_{\mathsf{max}}^{(k)}\left(\ell_{\mathsf{max}}^{(k)}-1\right)\right),
\end{align}
which concludes the proof of sufficient condition $\bm{\mathfrak{C}^{(\Upsilon)}}$ on $|\Upsilon^{\mathsf{s}}(\widetilde{\bm{\tau}}^{\downarrow,{(k)}})|$.

Moreover, \eqref{eq:suffi_Upsi} leads to the following condition on $\eta_k$.
\begin{equation}\label{eta_final_cond222}
    \eta_k \leq \min \left\{\frac{1}{2\beta} \sqrt{\frac{\zeta^{(k)}- \widehat{\zeta}^{(k)}}{\Delta}}, \frac{1}{2\beta}\right\},
\end{equation}
where $\Delta=\ell_{\mathsf{max}}^{(k)}(\ell_{\mathsf{max}}^{(k)}-1)(2\mathfrak{X}_1+\zeta^{(k)})$. Assume $\widehat{\ell}_{\mathsf{min}} \leq \ell^{(k)}_{\mathsf{sum}}\leq  \widehat{\ell}_{\mathsf{max}}$, $\max_{k}\{ \ell^{(k)}_{\mathsf{max}}\}\leq \ell_{\mathsf{max}}$, $\min_{k}\{ \ell^{(k)}_{\mathsf{min}}\}\leq \ell_{\mathsf{min}} $. Further, assume sufficient condition $\bm{\mathfrak{C}^{(\eta)}}$ holds which stats that 
\begin{equation}\label{eq:eta_value}
\eta_k = \frac{\alpha}{\sqrt{K \ell^{(k)}_{\mathsf{sum}}/N}},  
\end{equation}
where $\ell^{(k)}_{\mathsf{sum}}=\sum_{b\in\Omega}\sum_{u\in \mathcal{U}_{b}} \ell_u^{(k)}$, and $\alpha$ must be chosen to satisfy \eqref{eta_final_cond222} and 
\begin{equation}\label{alpha}
    \alpha<\sqrt{\frac{\widehat{\ell}_{\mathsf{min}} K}{4 N\beta^2 \ell_{\mathsf{max}}\left(\ell_{\mathsf{max}}-1\right)}}.
\end{equation}

Substituting $\eta_k$ from \eqref{eq:final_one_stepP3} into term $(b)$ in \eqref{eq:final_one_stepP3} and using the same technique utilized in inequality \eqref{error2}, term $(b)$ in \eqref{eq:final_one_stepP3} can be upper bounded as follows:
\begin{align}\label{error22}
    &\frac{1}{K}\sum_{k=0}^{K-1}\left[\frac{48}{\left(1-\zeta^{(k)}\right)} \left (\frac{|\Upsilon(\widetilde{\bm{\tau}}^{\downarrow,{(k)}})|-|\Upsilon^{\mathsf{s}}(\widetilde{\bm{\tau}}^{\downarrow,{(k)}})|}{|\Upsilon(\widetilde{\bm{\tau}}^{\downarrow,{(k)}})|}\right)^2\frac{|\Upsilon(\widetilde{\bm{\tau}}^{\downarrow,{(k)}})|}{|\Upsilon_{\mathsf{min}}(\widetilde{\bm{\tau}}^{\downarrow,{(k)}})|}\left(\frac{\mathfrak{X}_2}{1- 4\eta_k^2\beta^2 \ell_{\mathsf{max}}^{(k)}\left(\ell_{\mathsf{max}}^{(k)}-1\right)}\right)\right]\nonumber\\
    &{\le} \frac{1}{K}\sum_{k=0}^{K-1}\left[\frac{48}{ \left(1-\zeta^{(k)}\right)} \hspace{-1mm}\left(\frac{|\Upsilon(\widetilde{\bm{\tau}}^{\downarrow,{(k)}})|-|\Upsilon^{\mathsf{s}}(\widetilde{\bm{\tau}}^{\downarrow,{(k)}})|}{|\Upsilon(\widetilde{\bm{\tau}}^{\downarrow,{(k)}})|}\right)^2\hspace{-1mm}\frac{|\Upsilon(\widetilde{\bm{\tau}}^{\downarrow,{(k)}})|}{|\Upsilon_{\mathsf{min}}(\widetilde{\bm{\tau}}^{\downarrow,{(k)}})|}\hspace{-1mm}\left(\frac{\mathfrak{X}_2 \ell_{\mathsf{max}} }{\widehat{\ell}_{\mathsf{min}} K-4\alpha^2 N \beta^2 \ell_{\mathsf{max}}\left(\ell_{\mathsf{max}}-1\right)}\right)\hspace{-1mm}\right]{\le}\frac{1}{K}\sum_{k=0}^{K-1} \varpi^k= \frac{1}{K}\frac{\varpi^{K}-1}{\varpi-1}.
\end{align}
Subsequently, we have
\begin{align}\label{error224422}
    \left (\frac{|\Upsilon(\widetilde{\bm{\tau}}^{\downarrow,{(k)}})|-|\Upsilon^{\mathsf{s}}(\widetilde{\bm{\tau}}^{\downarrow,{(k)}})|}{|\Upsilon(\widetilde{\bm{\tau}}^{\downarrow,{(k)}})|}\right)^2\le \frac{\varpi^k\left(1-\zeta^{(k)}\right)|\Upsilon_{\mathsf{min}}(\widetilde{\bm{\tau}}^{\downarrow,{(k)}})|\left(\widehat{\ell}_{\mathsf{min}} K-4\alpha^2 N \beta^2 \ell_{\mathsf{max}}\left(\ell_{\mathsf{max}}-1\right)\right)}{48|\Upsilon(\widetilde{\bm{\tau}}^{\downarrow,{(k)}})| \mathfrak{X}_2 \ell_{\mathsf{max}} }.
\end{align}
Assuming sufficient conditions $\bm{\mathfrak{C}^{(\mathsf{L})}}$ and $\bm{\mathfrak{C}^{(\Upsilon)}}$, given in \eqref{necess:gamma_upper} and \eqref{necess:recruitment_upper}, hold, assuming  $\max_{k{\in}\mathcal{K}} \big\{\zeta^{(k)}\big\} \leq \zeta_{\mathsf{max}}<1$, $\eta_k = \frac{\alpha}{\sqrt{{\ell}^{(k)}_{\mathsf{sum}} K/N}}$, assuming \eqref{alpha} holds, and substituting \eqref{error1}, \eqref{error2}, and \eqref{error224422} into \eqref{eq:final_one_stepP3} give us the following upper bound:
{\small
\begin{align}\label{eq:final_one_stepP344}
      \hspace{-2.5mm}&\frac{1}{K} \sum_{k=0}^{K-1}\mathbb{E}\left[\left\Vert\nabla{\mathfrak{L}^{({k})}(\bm{\omega}^{(k)})}\right\Vert^2\right] {\leq}  \frac{4 \sqrt{\widehat{\ell}_{\mathsf{max}}}\left(\mathfrak{L}^{(-1)}(\bm{\omega}^{(0)}) - \mathfrak{L}^{(K)^\star}\right)}{\mathfrak{B}_k\left(1-\zeta_{\mathsf{max}}\right)\alpha \sqrt{KN}}+\frac{1}{K}\frac{\vartheta^{K}-1}{\vartheta-1}+\frac{1}{K}\frac{\varpi^{K}-1}{\varpi-1}\nonumber\\
     \hspace{-2.5mm} &+\frac{8}{K\left(1-\zeta_{\mathsf{max}}\right)} \sum_{k=0}^{K-1}\vast[\left(\sum_{b \in \Omega} \sum_{u\in \mathcal{U}_{b}}\frac{|\Upsilon_{u}(\widetilde{\tau}_{b}^{\downarrow,{(k)}})|}{|\Upsilon(\widetilde{\bm{\tau}}^{\downarrow,{(k)}})| \ell_u^{(k)}}\frac{{\beta^2\Theta^2 \alpha^2 N}\left(\ell_u^{(k)}\right)\left(\ell_u^{(k)}-1\right)}{{\ell}^{(k)}_{\mathsf{sum}} K- 4 \alpha^2 N\beta^2 \ell_u^{(k)}\left(\ell_u^{(k)}-1\right)}\left(1-\frac{{B}_{u}(\widetilde{\tau}_{b}^{\downarrow,{(k)}})}{|\Upsilon_{u}(\widetilde{\tau}_{b}^{\downarrow,{(k)}})|} \right)  \frac{{(|\Upsilon_{u}(\widetilde{\tau}_{b}^{\downarrow,{(k)}})|-1)}\left(\sigma_{u}(\widetilde{\tau}_{b}^{\downarrow,{(k)}})\right)^2}{|\Upsilon_{u}(\widetilde{\tau}_{b}^{\downarrow,{(k)}})|{B}_{u}(\widetilde{\tau}_{b}^{\downarrow,{(k)}})}\right)\nonumber\\
     \hspace{-2.5mm} &~~~~+ \left(\frac{\mathfrak{X}_2 \alpha^2 N\beta^2 \left(\ell_{\mathsf{max}}^{(k)}\right)\left(\ell_{\mathsf{max}}^{(k)}-1\right)}{{\ell}^{(k)}_{\mathsf{sum}} K- 4\alpha^2 N\beta^2\ell_{\mathsf{max}}^{(k)}\left(\ell_{\mathsf{max}}^{(k)}-1\right)}+ \frac{\beta\alpha\sqrt{N}\mathfrak{B}_k}{2\sqrt{{\ell}^{(k)}_{\mathsf{sum}} K}} \sum_{b \in \Omega} \sum_{u\in \mathcal{U}_{b}}\frac{\left(\widehat{\lambda}_{u}^{(k)}|\Upsilon_{u}(\widetilde{\tau}_{b}^{\downarrow,{(k)}})|\right)^2}{\left(|\Upsilon^{\mathsf{s}}(\widetilde{\bm{\tau}}^{\downarrow,{(k)}})|\right)^2 \ell^{(k)}_{u}}\left(1-\frac{{B}_{u}(\widetilde{\tau}_{b}^{\downarrow,{(k)}})}{|\Upsilon_{u}(\widetilde{\tau}_{b}^{\downarrow,{(k)}})|} \right) \frac{(|\Upsilon_{u}(\widetilde{\tau}_{b}^{\downarrow,{(k)}})|-1)\Theta^2\left(\sigma_{u}(\widetilde{\tau}_{b}^{\downarrow,{(k)}})\right)^2}{|\Upsilon_{u}(\widetilde{\tau}_{b}^{\downarrow,{(k)}})|{B}_{u}(\widetilde{\tau}_{b}^{\downarrow,{(k)}})}\right)\nonumber\\
      \hspace{-2.5mm}&~~~~+\left(1-\zeta_{\mathsf{max}}\right)\Bigg(\frac{\varpi^k|\Upsilon_{\mathsf{min}}(\widetilde{\bm{\tau}}^{\downarrow,{(k)}})|\left(\widehat{\ell}_{\mathsf{min}} K-4\alpha^2 N \beta^2 \ell_{\mathsf{max}}\left(\ell_{\mathsf{max}}-1\right)\right)}{2|\Upsilon(\widetilde{\bm{\tau}}^{\downarrow,{(k)}})| \mathfrak{X}_2 \ell_{\mathsf{max}} }\nonumber\\
     \hspace{-2.5mm} &~~~~~~~~~~~~~~~~~~~~~\times\sum_{b \in \Omega} \sum_{u\in \mathcal{U}_{b}}\frac{\Theta^2 \beta^2 \alpha^2 N\left(\ell_u^{(k)}-1\right)}{{\ell}^{(k)}_{\mathsf{sum}} K- 4\alpha^2 N\beta^2 \ell_u^{(k)}\left(\ell_u^{(k)}-1\right)} \left(1-\frac{{B}_{u}(\widetilde{\tau}_{b}^{\downarrow,{(k)}})}{|\Upsilon_{u}(\widetilde{\tau}_{b}^{\downarrow,{(k)}})|} \right)  \frac{{(|\Upsilon_{u}(\widetilde{\tau}_{b}^{\downarrow,{(k)}})|-1)}\left(\sigma_{u}(\widetilde{\tau}_{b}^{\downarrow,{(k)}})\right)^2}{|\Upsilon_{u}(\widetilde{\tau}_{b}^{\downarrow,{(k)}})|{B}_{u}(\widetilde{\tau}_{b}^{\downarrow,{(k)}})}\Bigg)\vast].
\end{align}
}%
Assume sufficient condition $\bm{\mathfrak{S}^{(\sigma)}}$ holds which stats that 
\begin{equation}\label{sigma}
    \displaystyle\max_{\substack{k\in\mathcal{K},\\b{\in}\Omega,\\u{\in}\mathcal{U}_{b}}}\left\{\left(1-\frac{{B}_{u}(\widetilde{\tau}_{b}^{\downarrow,{(k)}})}{|\Upsilon_{u}(\widetilde{\tau}_{b}^{\downarrow,{(k)}})|} \right)  \frac{{(|\Upsilon_{u}(\widetilde{\tau}_{b}^{\downarrow,{(k)}})|-1)}\left(\sigma_{u}(\widetilde{\tau}_{b}^{\downarrow,{(k)}})\right)^2}{|\Upsilon_{u}(\widetilde{\tau}_{b}^{\downarrow,{(k)}})|{B}_{u}(\widetilde{\tau}_{b}^{\downarrow,{(k)}})}\right\}\leq \sigma_{\mathsf{max}}.
\end{equation}
Assuming $\ell^{(k)}_{\mathsf{max}}=\max_{b\in\Omega,u\in\mathcal{U}_{b}}\{\ell^{(k)}_u\}$, $\ell^{(k)}_{\mathsf{min}}=\min_{b\in\Omega,u\in\mathcal{U}_{b}}\{\ell^{(k)}_u\}$, and substituting \eqref{sigma} in \eqref{eq:final_one_stepP344}, further bound \eqref{eq:final_one_stepP344} as follows:
\begin{align}\label{eq:final_one_stepP323}
    &\frac{1}{K} \sum_{k=0}^{K-1}\mathbb{E}\left[\left\Vert\nabla{\mathfrak{L}^{({k})}(\bm{\omega}^{(k)})}\right\Vert^2\right] {\leq}  \frac{4 \sqrt{\widehat{\ell}_{\mathsf{max}}}\left(\mathfrak{L}^{(-1)}(\bm{\omega}^{(0)}) - \mathfrak{L}^{(K)^\star}\right)}{\mathfrak{B}_k\left(1-\zeta_{\mathsf{max}}\right)\alpha \sqrt{KN}}+\frac{1}{K}\frac{\vartheta^{K}-1}{\vartheta-1}+\frac{1}{K}\frac{\varpi^{K}-1}{\varpi-1}\nonumber\\
    &+\frac{8}{K\left(1-\zeta_{\mathsf{max}}\right)} \sum_{k=0}^{K-1}\vast[\left({\beta^2\Theta^2 \alpha^2 N}\sum_{b \in \Omega} \sum_{u\in \mathcal{U}_{b}}\frac{|\Upsilon_{u}(\widetilde{\tau}_{b}^{\downarrow,{(k)}})|}{|\Upsilon(\widetilde{\bm{\tau}}^{\downarrow,{(k)}})| }\frac{\left(\ell_{\mathsf{max}}^{(k)}-1\right)}{{\ell}^{(k)}_{\mathsf{sum}} K- 4 \alpha^2 N\beta^2 \ell_{\mathsf{max}}^{(k)}\left(\ell_{\mathsf{max}}^{(k)}-1\right)}\sigma_{\mathsf{max}}\right)\nonumber\\
    &~~~~+ \left(\frac{\mathfrak{X}_2 \alpha^2 N\beta^2 \left(\ell_{\mathsf{max}}^{(k)}\right)\left(\ell_{\mathsf{max}}^{(k)}-1\right)}{{\ell}^{(k)}_{\mathsf{sum}} K- 4\alpha^2 N\beta^2\ell_{\mathsf{max}}^{(k)}\left(\ell_{\mathsf{max}}^{(k)}-1\right)}+ \frac{\beta\alpha\sqrt{N}\mathfrak{B}_k}{2\sqrt{{\ell}^{(k)}_{\mathsf{sum}} K}} \sum_{b \in \Omega} \sum_{u\in \mathcal{U}_{b}}\frac{\left(\widehat{\lambda}_{u}^{(k)}|\Upsilon_{u}(\widetilde{\tau}_{b}^{\downarrow,{(k)}})|\right)^2 \Theta^2}{\left(|\Upsilon^{\mathsf{s}}(\widetilde{\bm{\tau}}^{\downarrow,{(k)}})|\right)^2 \ell^{(k)}_{u}}\sigma_{\mathsf{max}}\right)\nonumber\\
    &~~~~+\left(1-\zeta_{\mathsf{max}}\right)\Bigg(\frac{\varpi^k|\Upsilon_{\mathsf{min}}(\widetilde{\bm{\tau}}^{\downarrow,{(k)}})|\left(\widehat{\ell}_{\mathsf{min}} K-4\alpha^2 N \beta^2 \ell_{\mathsf{max}}\left(\ell_{\mathsf{max}}-1\right)\right)}{2|\Upsilon(\widetilde{\bm{\tau}}^{\downarrow,{(k)}})| \mathfrak{X}_2 \ell_{\mathsf{max}} }\nonumber\\
    &~~~~~~~~~~~~~~~~~~~~~\times\sum_{b \in \Omega} \sum_{u\in \mathcal{U}_{b}}\frac{\Theta^2 \beta^2 \alpha^2 N\left(\ell_{\mathsf{max}}^{(k)}-1\right)}{{\ell}^{(k)}_{\mathsf{sum}} K- 4\alpha^2 N\beta^2 \ell_{\mathsf{max}}^{(k)}\left(\ell_{\mathsf{max}}^{(k)}-1\right)} \sigma_{\mathsf{max}}\Bigg)\vast].
\end{align}
Assuming $\widehat{\ell}_{\mathsf{min}} \leq \ell^{(k)}_{\mathsf{sum}}\leq  \widehat{\ell}_{\mathsf{max}}$, $\max_{k\in\mathcal{K}}\{ \ell^{(k)}_{\mathsf{max}}\}\leq \ell_{\mathsf{max}}$, $\min_{k\in\mathcal{K}}\{ \ell^{(k)}_{\mathsf{min}}\}\leq \ell_{\mathsf{min}} $, $ \widehat{\ell}_{\mathsf{min}} \leq \ell^{(k)}_{\mathsf{sum}}\leq  \widehat{\ell}_{\mathsf{max}}$, $\max_{k{\in}\mathcal{K}} \big\{\zeta^{(k)}\big\} \leq \zeta_{\mathsf{max}}<1$, and $\max_{k\in\mathcal{K}}\left\{\frac{|\Upsilon_{\mathsf{min}}(\widetilde{\bm{\tau}}^{\downarrow,{(k)}})|}{|\Upsilon(\widetilde{\bm{\tau}}^{\downarrow,{(k)}})|}\right\}\le |\Upsilon_{\mathsf{max}}|$, and performing some algebraic operations give us
\begin{align}\label{eq:final_one_stepP323ddg}
    &\frac{1}{K} \sum_{k=0}^{K-1}\mathbb{E}\left[\left\Vert\nabla{\mathfrak{L}^{({k})}(\bm{\omega}^{(k)})}\right\Vert^2\right] {\leq}  \frac{4 \sqrt{\widehat{\ell}_{\mathsf{max}}}\left(\mathfrak{L}^{(-1)}(\bm{\omega}^{(0)}) - \mathfrak{L}^{(K)^\star}\right)}{\mathfrak{B}_k\left(1-\zeta_{\mathsf{max}}\right)\alpha \sqrt{KN}}+\frac{4\beta\alpha\sqrt{N}\mathfrak{B}_k\sigma_{\mathsf{max}}\Theta^2}{\left(1-\zeta_{\mathsf{max}}\right)\sqrt{\widehat{\ell}_{\mathsf{min}} K} \ell_{\mathsf{min}}}+\frac{1}{K}\frac{\vartheta^{K}-1}{\vartheta-1}+\frac{1}{K}\frac{\varpi^{K}-1}{\varpi-1}\nonumber\\
    &+\frac{8\beta^2\Theta^2 \alpha^2 N\left(\ell_{\mathsf{max}}-1\right)\sigma_{\mathsf{max}}}{\left(1-\zeta_{\mathsf{max}}\right)\left(\widehat{\ell}_{\mathsf{min}}K - 4 \alpha^2 N\beta^2 \ell_{\mathsf{max}}\left(\ell_{\mathsf{max}}-1\right)\right)}+ \frac{8\mathfrak{X}_2 \alpha^2 N\beta^2 \left(\ell_{\mathsf{max}}\right)\left(\ell_{\mathsf{max}}-1\right)}{\left(1-\zeta_{\mathsf{max}}\right)\left(\widehat{\ell}_{\mathsf{min}} K- 4\alpha^2 N\beta^2\ell_{\mathsf{max}}\left(\ell_{\mathsf{max}}-1\right)\right)} \nonumber\\
    &+\Bigg(\frac{4|\Upsilon_{\mathsf{max}}|\left(\widehat{\ell}_{\mathsf{min}} K-4\alpha^2 N \beta^2 \ell_{\mathsf{max}}\left(\ell_{\mathsf{max}}-1\right)\right)\Theta^2 \beta^2 \alpha^2 N^2\left(\ell_{\mathsf{max}}-1\right)\sigma_{\mathsf{max}}}{ \mathfrak{X}_2 \ell_{\mathsf{max}}\left(\widehat{\ell}_{\mathsf{min}} K- 4\alpha^2 N\beta^2 \ell_{\mathsf{max}}\left(\ell_{\mathsf{max}}-1\right)\right)} \Bigg)\frac{1}{K} \sum_{k=0}^{K-1}\varpi^k.
\end{align}
Since we assume $|\varpi|<1$, $\frac{1}{K}\sum_{k=0}^{K-1}\varpi^k=\frac{1}{K}\frac{\varpi^{K}-1}{\varpi-1}$. Therefore, the above inequality can be further bounded as follows:
\begin{align}\label{eq:final_one_stepP323ddds}
    &\frac{1}{K} \sum_{k=0}^{K-1}\mathbb{E}\left[\left\Vert\nabla{\mathfrak{L}^{({k})}(\bm{\omega}^{(k)})}\right\Vert^2\right] {\leq}  \underbrace{\frac{4 \sqrt{\widehat{\ell}_{\mathsf{max}}}\left(\mathfrak{L}^{(-1)}(\bm{\omega}^{(0)}) - \mathfrak{L}^{(K)^\star}\right)}{\mathfrak{B}_k\left(1-\zeta_{\mathsf{max}}\right)\alpha \sqrt{KN}}}_{(a)}+\underbrace{\frac{4\beta\alpha\sqrt{N}\mathfrak{B}_k\sigma_{\mathsf{max}}\Theta^2}{\left(1-\zeta_{\mathsf{max}}\right)\sqrt{\widehat{\ell}_{\mathsf{min}} K} \ell_{\mathsf{min}}}}_{(b)}+\underbrace{\frac{1}{K}\frac{\vartheta^{K}-1}{\vartheta-1}}_{(c)}+\underbrace{\frac{1}{K}\frac{\varpi^{K}-1}{\varpi-1}}_{(d)}\nonumber\\
    &+\underbrace{\frac{8\beta^2\Theta^2 \alpha^2 N\left(\ell_{\mathsf{max}}-1\right)\sigma_{\mathsf{max}}}{\left(1-\zeta_{\mathsf{max}}\right)\left(\widehat{\ell}_{\mathsf{min}}K - 4 \alpha^2 N\beta^2 \ell_{\mathsf{max}}\left(\ell_{\mathsf{max}}-1\right)\right)}}_{(e)}+ \underbrace{\frac{8\mathfrak{X}_2 \alpha^2 N\beta^2 \left(\ell_{\mathsf{max}}\right)\left(\ell_{\mathsf{max}}-1\right)}{\left(1-\zeta_{\mathsf{max}}\right)\left(\widehat{\ell}_{\mathsf{min}} K- 4\alpha^2 N\beta^2\ell_{\mathsf{max}}\left(\ell_{\mathsf{max}}-1\right)\right)}}_{(f)} \nonumber\\
    &+\underbrace{\Bigg(\frac{4|\Upsilon_{\mathsf{max}}|\Theta^2 \beta^2 \alpha^2 N^2\left(\ell_{\mathsf{max}}-1\right)\sigma_{\mathsf{max}}(\varpi^{K}-1)}{ \mathfrak{X}_2 \ell_{\mathsf{max}}K(\varpi-1)} \Bigg)}_{(g)}.
\end{align}
Subsequently, we have
\begin{align}\label{eq:final_one_stepP323dddsdsdds}
    &\frac{1}{K} \sum_{k=0}^{K-1}\mathbb{E}\left[\left\Vert\nabla{\mathfrak{L}^{({k})}(\bm{\omega}^{(k)})}\right\Vert^2\right] {=}  \mathcal{O}\left(\frac{1}{\sqrt{K}}\right)+\mathcal{O}\left(\frac{1}{K}\right)+\mathcal{O}\left(\frac{1}{K}\right)+\mathcal{O}\left(\frac{1}{K}\right)+ \mathcal{O}\left(\frac{1}{K}\right)+\mathcal{O}\left(\frac{1}{K}\right)= \mathcal{O}\left(\frac{1}{\sqrt{K}}\right).
\end{align}

Moreover, taking limit from both hand sides of \eqref{eq:final_one_stepP323ddds} when $K\rightarrow\infty$, $(a)$, $(b)$, $(c)$, $(d)$, $(e)$, $(f)$, and $(g)$ in the \eqref{eq:final_one_stepP323ddds} become $0$. Consequently, we have 
\begin{align}
    &\lim_{K\rightarrow \infty}\frac{1}{K} \sum_{k=0}^{K-1}\mathbb{E}\left[\left\Vert\nabla{\mathfrak{L}^{({k})}(\bm{\omega}^{(k)})}\right\Vert^2\right]=0,
\end{align}
which guarantees the convergence of {\tt DCLM} and concludes the proof of Theorem~\ref{th:sufficient_conditions}.

\newpage
\section{Convergence under Sufficient Conditions Presented in Theorem \ref{th:sufficient_conditions}}\label{app:cor:sufficient_condition}
\begin{corollary}[Convergence under Sufficient Conditions]\label{cor:sufficient_condition}
Under sufficient conditions given in Theorem~\ref{th:sufficient_conditions}, the  cumulative average of the global loss function gradient for {\tt DCLM} satisfies:
\begin{equation}\label{eq:cor2}
\begin{aligned}
    &\frac{1}{K} \sum_{k=0}^{K-1}\mathbb{E}\left[\left\Vert\nabla{\mathfrak{L}^{({k})}(\bm{\omega}^{(k)})}\right\Vert^2\right] {\leq}  \frac{4 \sqrt{\widehat{\ell}_{\mathsf{max}}}\left(\mathfrak{L}^{(-1)}(\bm{\omega}^{(0)}) - \mathfrak{L}^{(K)^\star}\right)}{\mathfrak{B}_k\left(1-\zeta_{\mathsf{max}}\right)\alpha \sqrt{KN}}+\frac{4\beta\alpha\sqrt{N}\mathfrak{B}_k\sigma_{\mathsf{max}}\Theta^2}{\left(1-\zeta_{\mathsf{max}}\right)\sqrt{\widehat{\ell}_{\mathsf{min}} K} \ell_{\mathsf{min}}}+\frac{1}{K}\frac{\vartheta^{K}-1}{\vartheta-1}+\frac{1}{K}\frac{\varpi^{K}-1}{\varpi-1}\nonumber\\
    &+\frac{8\beta^2\Theta^2 \alpha^2 N\left(\ell_{\mathsf{max}}-1\right)\sigma_{\mathsf{max}}}{\left(1-\zeta_{\mathsf{max}}\right)\left(\widehat{\ell}_{\mathsf{min}}K - 4 \alpha^2 N\beta^2 \ell_{\mathsf{max}}\left(\ell_{\mathsf{max}}-1\right)\right)}+ \frac{8\mathfrak{X}_2 \alpha^2 N\beta^2 \left(\ell_{\mathsf{max}}\right)\left(\ell_{\mathsf{max}}-1\right)}{\left(1-\zeta_{\mathsf{max}}\right)\left(\widehat{\ell}_{\mathsf{min}} K- 4\alpha^2 N\beta^2\ell_{\mathsf{max}}\left(\ell_{\mathsf{max}}-1\right)\right)} \nonumber\\
    &+\Bigg(\frac{4|\Upsilon_{\mathsf{max}}|\Theta^2 \beta^2 \alpha^2 N^2\left(\ell_{\mathsf{max}}-1\right)\sigma_{\mathsf{max}}(\varpi^{K}-1)}{ \mathfrak{X}_2 \ell_{\mathsf{max}}K(\varpi-1)} \Bigg).
\end{aligned}
\end{equation}
\end{corollary}
\vspace{-3.4mm}
\begin{proof}
The proof is the same as the proof for Theorem \ref{th:sufficient_conditions} (presented in Appendix~\ref{app:th:sufficient_conditions}). The conclusion for Corollary~\ref{cor:sufficient_condition} can be founded in inequality \eqref{eq:final_one_stepP323ddds}.
\end{proof}
\vspace{-2mm}

\newpage
\section{Transforming the Optimization Problem of {\tt DCML}}\label{app:optTransform}
In the following, we provide a tractable technique to solve $\bm{\mathcal{P}}$ through utilizing \textit{geometric programming} (GP)~\cite{chiang2005geometric}. Our approach relies on a set of approximations and constraint modifications to solve the problem in an iterative manner, which enjoys strong convergence guarantees.


\subsection{Geometric Programming (GP)} \label{sec:GPtransConv}
Consider the following two key definitions: \ref{sec:GPtransConv}
\begin{definition}[Monomial]
    A {monomial} is a function $f: \mathbb{R}^n_{++}\rightarrow \mathbb{R}$:\footnote{$\mathbb{R}^n_{++}$ denotes the strictly positive quadrant of $n$-dimensional Euclidean space.} $f(\bm{y})=d y_1^{\alpha^{(1)}} y_2^{\alpha^{(2)}} \cdots y_n ^{\alpha^{(n)}}$, where $d> 0$, $\bm{y}=[y_1,\cdots,y_n]$, and $\alpha^{(j)}\in \mathbb{R}$, $\forall j$.
\end{definition}
\begin{definition}[Posynomial]
     A posynomial $g$ is defined as a sum of monomials: $g(\bm{y})= \sum_{m=1}^{M} d_m y_1^{\alpha^{(1)}_m} y_2^{\alpha^{(2)}_m} \cdots y_n ^{\alpha^{(n)}_m}$, where $d_m> 0$, $\bm{y}=[y_1,\cdots,y_n]$, and $\alpha_m^{(j)}\in \mathbb{R}$, $\forall j$.
\end{definition}
Accordingly, a standard GP is a non-convex optimization problem defined as minimizing a posynomial subject to inequality constraints on posynomials (and monomials) and equality constraints on monomials~\cite{chiang2005geometric}:
\begin{equation}\label{eq:GPformat}
\begin{aligned}
    &\min_{\bm{y}} f_0 (\bm{y})\nonumber\\
    &\textrm{s.t.} ~~~ f_i(\bm{y})\leq 1, \;\; i=1,\cdots,I,\nonumber\\
    &~~~~~~ h_l(\bm{y})=1, \;\; l=1,\cdots,L,
\end{aligned}
\end{equation}
where  $f_i(\bm{y})=\sum_{m=1}^{M_i} d_{i,m} y_1^{\alpha^{(1)}_{i,m}} y_2^{\alpha^{(2)}_{i,m}} \cdots y_n ^{\alpha^{(n)}_{i,m}}$, $\forall i$, $h_l(\bm{y})= d_l y_1^{\alpha^{(1)}_l} y_2^{\alpha^{(2)}_l} \cdots y_n ^{\alpha^{(n)}_l}$, $\forall l$, $d_{i,m}>0~\forall i,m$, and $d_l>0,\forall l$. Due to the fact that the log-sum-exp function $f(\bm{y}) = \ln \sum_{j=1}^n e^{y_j}$ is convex with the following change of variables $z_i=\ln(y_i)$, $b_{i,k}=\ln(d_{i,k})$, $b_l=\ln (d_l)$  the GP can be converted into the following convex format
\begin{equation}~\label{GPtoConvex}
\begin{aligned}
    &\min_{\bm{z}} \;\ln \sum_{m=1}^{M_0} e^{\left(\bm{\alpha}^{\top}_{0,m}\bm{z}+ b_{0,m}\right)}\nonumber\\
    &\textrm{s.t.} ~~~ \ln \sum_{m=1}^{M_i} e^{\left(\bm{\alpha}^{\top}_{i,m}\bm{z}+ b_{i,m}\right)}\leq 0 \;\; i=1,\cdots,I,\nonumber\\
    &~~~~~~~ \bm{\alpha}_l^\top \bm{z}+b_l =0\;\; l=1,\cdots,L,
\end{aligned}
\end{equation}
where $\bm{z}=[z_1,\cdots,z_n]^\top$, $\bm{\alpha}_{i,k}=\left[\alpha_{i,k}^{(1)},\alpha_{i,k}^{(2)}\cdots, \alpha_{i,k}^{(n)}\right]^\top$, $\forall i,k$, and $\bm{\alpha}_{l}=\left[\alpha_{l}^{(1)},a_{l}^{(2)}\cdots, \alpha_{l}^{(n)}\right]^\top$, $\forall l$.

Nevertheless, $\bm{\mathcal{P}}$ does not admit the format of GP due to several factors which violate the definition of posynomials, including (i) negative terms in the convergence bound \eqref{eq:gen_conv}, (ii) logarithmic functions in the calculations of the data transmission rates (see Sec.~\ref{sec:data_transmission_rate}), (iii) recursive functions such as the global round completion time (see Sec.~\ref{sec:communication_latency}) and scheduling decision constraints (see Sec.~\ref{sec:mac_scheduler}), (iv) and non-convex nested min/max functions involved in the constraints of $\bm{\mathcal{P}}$ and the calculations of the communication latencies of {\tt DCLM} (see Sec.~\ref{sec:communication_latency}). In order to address these violations, we utilize a series of approximations and constraint modifications. Among these techniques, we frequently use the arithmetic-geometric mean inequality presented in the following lemma to upper bound a posynomial with a larger monomial in value.

\begin{lemma}[\textbf{Arithmetic-geometric mean inequality}~\cite{duffin1972reversed}]\label{Lemma:ArethmaticGeometric}
    Consider a posynomial function $g(\bm{y})=\sum_{i=1}^{i'} u_i(\bm{y})$, where $u_i(\bm{y})$ is a monomial, $\forall i$. The following inequality holds:
    \begin{equation}\label{eq:approxPosMonMain}
     g(\bm{y})\geq \hat{g}(\bm{y})\triangleq \prod_{i=1}^{i'}\left( \frac{u_i(\bm{y})}{\alpha_i(\bm{z})}\right)^{\alpha_i(\bm{z})},
     \vspace{-1.5mm}
    \end{equation}
    where $\alpha_i(\bm{z})=u_i(\bm{z})/g(\bm{z})$, $\forall i$, and $\bm{z}>0$ is a fixed point.
\end{lemma}

In the following sections, we transform the objective function and constraints of optimization problem $\bm{\mathcal{P}}$ into standard GP format. Specifically, first, we focus on transforming data transmission rates (Sec.~\ref{app:data_transformation_rate}) and computation/communication latencies (Sec.~\ref{app:computation_communication_latencies}), which appear in all terms of the objective function and different constraints of the optimization problem $\bm{\mathcal{P}}$ (see \eqref{opt:main}). Next, we transform energy consumption equations into standard GP format (Sec.~\ref{app:energy_consumption}). We then present the transformation of the objective function of $\bm{\mathcal{P}}$ (Sec.~\ref{app:objective_fucntion}). Finally, we transform constraints of $\bm{\mathcal{P}}$ into standard GP format. This includes transforming (i) sufficient conditions of ML convergence (Sec.~\ref{app:objective_fucntion}), (ii) scheduling decisions constraints (Sec.~\ref{app:cons:scheduling_decisions}), (iii) PRB/Power allocations constraints (Sec.~\ref{app:cons:PRB_Power_allocation}), and (iv) GM/GPs transmissions constraints (Sec.~\ref{app:cons:GM_GP_transmissions}). 

Note that, for brevity of notation, we refer to the vector of all decision variables as $\bm{v}$.

\newpage
\subsection{Transformation of Data Transmission Rates}\label{app:data_transformation_rate}

\noindent\textbf{$\bm{\mathsf{R2F}}$ transmission rate.} Referring to Sec.~\ref{sec:data_transmission_rate}, let us revisit the data transmission rate equation of $\mathsf{R2F}$ communication mode as follows:



\begin{equation}\label{app:eq:DTR_R2F_C1}
    \mathfrak{R}^{\downarrow}_{b,r}(t_{x}){=}\min_{u\in\mathcal{U}_b}\big\{\mathfrak{R}^{\downarrow}_{b,u,r}(t_{x})|\widehat{\lambda}^{(k)}_u{=} 1,x{\in}\mathcal{N}^{(k)},\forall u{\in} \mathcal{U}_{b}\big\},
\end{equation}
where
\begin{equation} \label{app:eq:DTR_R2F_C2}
\mathfrak{R}^{\downarrow}_{b,u,r}(t_{x})= B_{b} \log_2\big(1+\Gamma^{\downarrow}_{b,u,r}(t_{x})\big),
\end{equation}
In \eqref{app:eq:DTR_R2F_C2}, $\Gamma^{\downarrow}_{b,u,r}(t_{x})$ is given by
\begin{equation}\label{app:eq:DTR_R2F_C3}
    \Gamma^{\downarrow}_{b,u,r}(t_{x}){=}\frac{\vert\xi_{b,u}(t_{x})\vert^2 \rho^{\downarrow}_{b,r}(t_{x}) P^{\mathsf{max}}_{b}}{\displaystyle\sum_{b'{\in}\Omega\setminus \{b\}} |\xi_{b',u}(t_{x})|^2 \rho^{\downarrow}_{b',r}(t_{x}) P^{\mathsf{max}}_{b'}{+} B_{b}N_{0}}.
\end{equation}

In the following, we transform \eqref{app:eq:DTR_R2F_C1} into the standard form of GP. To this end, we first aim to approximate \eqref{app:eq:DTR_R2F_C2}, through defining $\mathfrak{R}^{\downarrow}_{b,u,r}(t_{x})$ as an auxiliary decision variable which must satisfy the following constraint:

\begin{align}
    &\mathfrak{R}^{\downarrow}_{b,u,r}(t_{x})= B_{b} \log_2\big(1+\Gamma^{\downarrow}_{b,u,r}(t_{x})\big)\nonumber\\
    &{\Rightarrow} \frac{\ln(2)\mathfrak{R}^{\downarrow}_{b,u,r}(t_{x})}{B_{b}}=  \ln\big(1+\Gamma^{\downarrow}_{b,u,r}(t_{x})\big)\nonumber\\
    &{\Rightarrow} \exp\left(\frac{\ln(2)\mathfrak{R}^{\downarrow}_{b,u,r}(t_{x})}{B_{b}}\right)= \exp\left(\ln\big(1+\Gamma^{\downarrow}_{b,u,r}(t_{x})\big)\right)\nonumber\\
    &{\Rightarrow} \exp\left(\frac{C\ln(2)\mathfrak{R}^{\downarrow}_{b,u,r}(t_{x})}{C B_{b}}\right)= 1+\Gamma^{\downarrow}_{b,u,r}(t_{x})\nonumber\\
    &{\Rightarrow} \left(\exp\left(\frac{\ln(2)\mathfrak{R}^{\downarrow}_{b,u,r}(t_{x})}{C B_{b}}\right)\right)^C= 1+\Gamma^{\downarrow}_{b,u,r}(t_{x}).\nonumber
\end{align}
Using the first three terms of the Maclaurin series of the exponential function $e^x=1+x+\frac{x^2}{2!}+\frac{x^3}{3!}+\cdots$, we get
\begin{align}\label{ttttt}
    &{\Rightarrow} \left(1+\left(\frac{\ln(2)\mathfrak{R}^{\downarrow}_{b,u,r}(t_{x})}{CB_{b}}\right)+\frac{\left(\frac{\ln(2)\mathfrak{R}^{\downarrow}_{b,u,r}(t_{x})}{CB_{b}}\right)^2}{2}\right)^C= 1+\Gamma^{\downarrow}_{b,u,r}(t_{x})\nonumber\\
    &{\Rightarrow} \frac{\left(1+\left(\frac{\ln(2)\mathfrak{R}^{\downarrow}_{b,u,r}(t_{x})}{CB_{b}}\right)+\left(\frac{\ln(2)\mathfrak{R}^{\downarrow}_{b,u,r}(t_{x})}{\sqrt{2}CB_{b}}\right)^2\right)^C}{1+\Gamma^{\downarrow}_{b,u,r}(t_{x})}= 1.
\end{align}
Where $C$ is a constant utilized to regulate the accuracy of the Maclaurin series expansion of $e^x$. Note that the above approximation is the result of the following  relation:
\begin{align}\label{eq:taylor_e}
e^x=\lim_{C \rightarrow \infty}{\left(e^{\frac{x}{C}}\right)^C} \geq \lim_{C \rightarrow \infty}{\left(1+\frac{x}{C}+\frac{\left(\frac{x}{C}\right)^2}{2}\right)^C} \geq \lim_{C \rightarrow \infty}\left(1+\frac{x}{C}\right)^C = e^x.
\end{align}
Replacing \eqref{app:eq:DTR_R2F_C3} into \eqref{ttttt} leads to
\begin{align}
    \frac{\left(1+\left(\frac{\ln(2)\mathfrak{R}^{\downarrow}_{b,u,r}(t_{x})}{CB_{b}}\right)+\left(\frac{\ln(2)\mathfrak{R}^{\downarrow}_{b,u,r}(t_{x})}{\sqrt{2}CB_{b}}\right)^2\right)^C}{1+\frac{\vert\xi_{b,u}(t_{x})\vert^2 \rho^{\downarrow}_{b,r}(t_{x}) P^{\mathsf{max}}_{b}}{\displaystyle\sum_{b'{\in}\Omega\setminus \{b\}} |\xi_{b',u}(t_{x})|^2 \rho^{\downarrow}_{b',r}(t_{x}) P^{\mathsf{max}}_{b'}{+} B_{b}N_{0}}}= 1.
\end{align}
Performing some algebraic manipulations gives us
\begin{align}
    \frac{\left(\displaystyle\sum_{b'{\in}\Omega\setminus \{b\}} |\xi_{b',u}(t_{x})|^2 \rho^{\downarrow}_{b',r}(t_{x}) P^{\mathsf{max}}_{b'}{+} B_{b}N_{0}\right)\left(1+\left(\frac{\ln(2)\mathfrak{R}^{\downarrow}_{b,u,r}(t_{x})}{CB_{b}}\right)+\left(\frac{\ln(2)\mathfrak{R}^{\downarrow}_{b,u,r}(t_{x})}{\sqrt{2}CB_{b}}\right)^2\right)^C}{\displaystyle\sum_{b'{\in}\Omega} |\xi_{b',u}(t_{x})|^2 \rho^{\downarrow}_{b',r}(t_{x}) P^{\mathsf{max}}_{b'}{+} B_{b}N_{0}}= 1.\nonumber
\end{align}
Defining auxiliary decision variable $\mathfrak{J}^{\downarrow}_{b,u,r}(t_{x})$ satisfying
\begin{equation}\label{app:eq:DTR_R2F_C2_1}
    \mathfrak{J}^{\downarrow}_{b,u,r}(t_{x})=1+\left(\frac{\ln(2)\mathfrak{R}^{\downarrow}_{b,u,r}(t_{x})}{CB_{b}}\right)+\left(\frac{\ln(2)\mathfrak{R}^{\downarrow}_{b,u,r}(t_{x})}{\sqrt{2}CB_{b}}\right)^2
\end{equation}
results in
\begin{align}
    \frac{\left(\mathfrak{J}^{\downarrow}_{b,u,r}(t_{x})\right)^{C}\left(\displaystyle\sum_{b'{\in}\Omega\setminus \{b\}} |\xi_{b',u}(t_{x})|^2 \rho^{\downarrow}_{b',r}(t_{x}) P^{\mathsf{max}}_{b'}{+} B_{b}N_{0}\right)}{\displaystyle\sum_{b'{\in}\Omega} |\xi_{b',u}(t_{x})|^2 \rho^{\downarrow}_{b',r}(t_{x}) P^{\mathsf{max}}_{b'}{+} B_{b}N_{0}}= 1.\nonumber
\end{align}
However, this constraint is not in the format of GP. Therefore, we transform it into the standard GP format via introducing the following three inequalities:
\begin{equation}\label{app:eq:DTR_R2F_C2_3}
    \frac{\left(\mathfrak{J}^{\downarrow}_{b,u,r}(t_{x})\right)^{C}\left(\displaystyle\sum_{b'{\in}\Omega\setminus \{b\}} |\xi_{b',u}(t_{x})|^2 \rho^{\downarrow}_{b',r}(t_{x}) P^{\mathsf{max}}_{b'}{+} B_{b}N_{0}\right)}{\displaystyle\sum_{b'{\in}\Omega} |\xi_{b',u}(t_{x})|^2 \rho^{\downarrow}_{b',r}(t_{x}) P^{\mathsf{max}}_{b'}{+} B_{b}N_{0}}\le 1,
\end{equation}
\begin{equation}\label{app:eq:DTR_R2F_C2_4}
    \frac{\left(\mathfrak{J}^{\downarrow}_{b,u,r}(t_{x})\right)^{-C}(\mathscr{B}^{\downarrow}_{b,u,r}(t_{x}))^{-1}\left(\displaystyle\sum_{b'{\in}\Omega} |\xi_{b',u}(t_{x})|^2 \rho^{\downarrow}_{b',r}(t_{x}) P^{\mathsf{max}}_{b'}{+} B_{b}N_{0}\right)}{\displaystyle \sum_{b'{\in}\Omega\setminus \{b\}} |\xi_{b',u}(t_{x})|^2 \rho^{\downarrow}_{b',r}(t_{x}) P^{\mathsf{max}}_{b'}{+} B_{b}N_{0}}\le 1,
\end{equation}
\begin{equation}
    \mathscr{B}^{\downarrow}_{b,u,r}(t_{x})\ge 1,
\end{equation}
where $\mathscr{B}^{\downarrow}_{b,u,r}(t_{x})$ is an auxiliary decision variable added with a large penalty term to the objective function to force $\mathscr{B}^{\downarrow}_{b,u,r}(t_{x}){\rightarrow}1^+$ at the optimal point. The fractions in~\eqref{app:eq:DTR_R2F_C2_3} and \eqref{app:eq:DTR_R2F_C2_4} still need transformation since they are inequalities with posynomials in their denominator, which are not posynomials. We thus exploit arithmetic-geometric mean inequality (Lemma~\ref{Lemma:ArethmaticGeometric}) to approximate the denominators of \eqref{app:eq:DTR_R2F_C2_3} and \eqref{app:eq:DTR_R2F_C2_4} with monomials. In doing so, we approximate the denominator in \eqref{app:eq:DTR_R2F_C2_3} as follows:
\begin{align}\label{app:eq:DTR_R2F_C2_5}
    F^{\downarrow}_{b,u,r}(\bm{v})&{=}\sum_{b'{\in}\Omega} |\xi_{b',u}(t_{x})|^2 \rho^{\downarrow}_{b',r}(t_{x}) P^{\mathsf{max}}_{b'}{+} B_{b}N_{0}\nonumber\\
    &{\geq} \widehat{F}^{\downarrow}_{b,u,r}(\bm{v};\ell) {\triangleq} \prod_{b'{\in}\Omega}\left(\frac{|\xi_{b',u}(t_{x})|^2 \rho^{\downarrow}_{b',r}(t_{x})  F^{\downarrow}_{b,u,r}([\bm{v}]^{(\ell-1)})}{\left[|\xi_{b',u}(t_{x})|^2 \rho^{\downarrow}_{b',r}(t_{x}) \right]^{(\ell-1)}}\right)^{\hspace{-2mm}\frac{\left[|\xi_{b',u}(t_{x})|^2 \rho^{\downarrow}_{b',r}(t_{x}) P^{\mathsf{max}}_{b'}\right]^{(\ell-1)}}{F^{\downarrow}_{b,u,r}([\bm{v}]^{(\ell-1)})}}\nonumber\\
    &~~~~~~~~~~~~~~~~~~~~~~~~~\times\left(F^{\downarrow}_{b,u,r}([\bm{v}]^{(\ell-1)})\right)^{\frac{\left[B_{b}N_{0}\right]^{(\ell-1)}}{F^{\downarrow}_{b,u,r}([\bm{v}]^{(\ell-1)})}},
\end{align}
which gives us the following approximation of~\eqref{app:eq:DTR_R2F_C2_3}:
\begin{equation}\label{app:eq:DTR_R2F_C2_6}
     \frac{\left(\mathfrak{J}^{\downarrow}_{b,u,r}(t_{x})\right)^{C}\left(\displaystyle\sum_{b'{\in}\Omega\setminus \{b\}} |\xi_{b',u}(t_{x})|^2 \rho^{\downarrow}_{b',r}(t_{x}) P^{\mathsf{max}}_{b'}{+} B_{b}N_{0}\right)}{\widehat{F}^{\downarrow}_{b,u,r}(\bm{v};\ell)}\le 1.
\end{equation}
Similarly, we approximate the denominator in \eqref{app:eq:DTR_R2F_C2_4} as follows:
\begin{align}\label{app:eq:DTR_R2F_C2_7}
    G^{\downarrow}_{b,u,r}(\bm{v}){=}& \sum_{b'{\in}\Omega\setminus \{b\}} |\xi_{b',u}(t_{x})|^2 \rho^{\downarrow}_{b',r}(t_{x}) P^{\mathsf{max}}_{b'}{+} B_{b}N_{0}\nonumber\vspace{-15mm}\\
    &{\geq} \widehat{G}^{\downarrow}_{b,u,r}(\bm{v};\ell) {\triangleq} \prod_{b'{\in}\Omega\setminus \{b\}}\left(\frac{|\xi_{b',u}(t_{x})|^2 \rho^{\downarrow}_{b',r}(t_{x})  G^{\downarrow}_{b,u,r}([\bm{v}]^{(\ell-1)})}{\left[|\xi_{b',u}(t_{x})|^2 \rho^{\downarrow}_{b',r}(t_{x}) \right]^{(\ell-1)}}\right)^{\hspace{-2mm}\frac{\left[|\xi_{b',u}(t_{x})|^2 \rho^{\downarrow}_{b',r}(t_{x}) P^{\mathsf{max}}_{b'}\right]^{(\ell-1)}}{G^{\downarrow}_{b,u,r}([\bm{v}]^{(\ell-1)})}}\nonumber\\
    &~~~~~~~~~~~~~~~~~~~~~~~~~\times\left(G^{\downarrow}_{b,u,r}([\bm{v}]^{(\ell-1)})\right)^{\frac{\left[B_{b}N_{0}\right]^{(\ell-1)}}{G^{\downarrow}_{b,u,r}([\bm{v}]^{(\ell-1)})}},
\end{align}
which gives the following approximation of~\eqref{app:eq:DTR_R2F_C2_4}:
\begin{equation}\label{app:eq:DTR_R2F_C2_8}
     \frac{\left(\mathfrak{J}^{\downarrow}_{b,u,r}(t_{x})\right)^{-C}(\mathscr{B}^{\downarrow}_{b,u,r}(t_{x}))^{-1}\left(\displaystyle\sum_{b'{\in}\Omega} |\xi_{b',u}(t_{x})|^2 \rho^{\downarrow}_{b',r}(t_{x}) P^{\mathsf{max}}_{b'}{+} B_{b}N_{0}\right)}{ \widehat{G}^{\downarrow}_{b,u,r}(\bm{v};\ell)}\le 1.
\end{equation}

We next aim to transform \eqref{app:eq:DTR_R2F_C2_1} into standard GP format. In doing so, let us rewrite \eqref{app:eq:DTR_R2F_C2_1} as follows:
\begin{equation}\label{app:eq:DTR_R2F_C2_9}
    \frac{2(CB_{b})^2\mathfrak{J}^{\downarrow}_{b,u,r}(t_{x})}{2(CB_{b})^2+2(CB_{b})^2\ln(2)\mathfrak{R}^{\downarrow}_{b,u,r}(t_{x})+\left(\ln(2)\mathfrak{R}^{\downarrow}_{b,u,r}(t_{x})\right)^2}=1.
\end{equation}
We transform this constraint into the format of GP by splitting it into the following three inequalities:
\begin{equation}\label{app:eq:DTR_R2F_C2_10}
    \frac{2(CB_{b})^2\mathfrak{J}^{\downarrow}_{b,u,r}(t_{x})}{2(CB_{b})^2+2(CB_{b})^2\ln(2)\mathfrak{R}^{\downarrow}_{b,u,r}(t_{x})+\left(\ln(2)\mathfrak{R}^{\downarrow}_{b,u,r}(t_{x})\right)^2}\le 1,
\end{equation}
\begin{equation}\label{app:eq:DTR_R2F_C2_11}
    \frac{(\mathscr{C}^{\downarrow}_{b,u,r}(t_{x}))^{-1}\left(2(CB_{b})^2+2(CB_{b})^2\ln(2)\mathfrak{R}^{\downarrow}_{b,u,r}(t_{x})+\left(\ln(2)\mathfrak{R}^{\downarrow}_{b,u,r}(t_{x})\right)^2\right)}{2(CB_{b})^2\mathfrak{J}^{\downarrow}_{b,u,r}(t_{x})}\le 1,
\end{equation}
\begin{equation}
    \mathscr{C}^{\downarrow}_{b,u,r}(t_{x})\ge 1
\end{equation}
where $\mathscr{C}^{\downarrow}_{b,u,r}(t_{x})$ is an auxiliary decision variable and will be added with a large penalty term to the objective function to force $\mathscr{C}^{\downarrow}_{b,u,r}(t_{x}){\rightarrow}1^+$ at the optimal point. The fraction in~\eqref{app:eq:DTR_R2F_C2_10} still needs transformation since it is inequality with posynomial in its denominator, which is not posynomial. We thus exploit arithmetic-geometric mean inequality (Lemma~\ref{Lemma:ArethmaticGeometric}) to approximate the denominator of \eqref{app:eq:DTR_R2F_C2_3} with monomial. 
\begin{align}\label{app:eq:DTR_R2F_C2_12}
    &H^{\downarrow}_{b,u,r}(\bm{v}){=} 2(CB_{b})^2{+}2(CB_{b})^2\ln(2)\mathfrak{R}^{\downarrow}_{b,u,r}(t_{x}){+}\left(\ln(2)\mathfrak{R}^{\downarrow}_{b,u,r}(t_{x})\right)^2 {\geq} \widehat{H}^{\downarrow}_{b,u,r}(\bm{v};\ell) {\triangleq} \left( H^{\downarrow}_{b,u,r}([\bm{v}]^{(\ell-1)})\right)^{\hspace{-1mm}\frac{\left[2(CB_{b})^2\right]^{(\ell-1)}}{H^{\downarrow}_{b,u,r}([\bm{v}]^{(\ell-1)})}}\nonumber\\
    &{\times}\left(\hspace{-1.5mm}\frac{\mathfrak{R}^{\downarrow}_{b,u,r}(t_{x}) H^{\downarrow}_{b,u,r}([\bm{v}]^{(\ell-1)})}{\left[\mathfrak{R}^{\downarrow}_{b,u,r}(t_{x})\right]^{(\ell-1)}}\hspace{-1.5mm}\right)^{\hspace{-2mm}\frac{\left[2(CB_{b})^2\ln(2)\mathfrak{R}^{\downarrow}_{b,u,r}(t_{x})\right]^{(\ell-1)}}{H^{\downarrow}_{b,u,r}([\bm{v}]^{(\ell-1)})}}{\times}\left(\hspace{-1.5mm}\frac{\left(\mathfrak{R}^{\downarrow}_{b,u,r}(t_{x})\right)^2 H^{\downarrow}_{b,u,r}([\bm{v}]^{(\ell-1)})}{\left[\left(\mathfrak{R}^{\downarrow}_{b,u,r}(t_{x})\right)^2\right]^{(\ell-1)}}\hspace{-1.5mm}\right)^{\hspace{-2mm}\frac{\left[\left(\ln(2)\mathfrak{R}^{\downarrow}_{b,u,r}(t_{x})\right)^2\right]^{(\ell-1)}}{H^{\downarrow}_{b,u,r}([\bm{v}]^{(\ell-1)})}},
\end{align}
which results in the following approximation of~\eqref{app:eq:DTR_R2F_C2_10}:
\begin{equation}\label{app:eq:DTR_R2F_C2_14}
     \frac{2(CB_{b})^2\mathfrak{J}^{\downarrow}_{b,u,r}(t_{x})}{ \widehat{H}^{\downarrow}_{b,u,r}(\bm{v};\ell)}\le 1.
\end{equation}

We finally approximate \eqref{app:eq:DTR_R2F_C2} as follows:
\begin{tcolorbox}[ams align]
     & \frac{\left(\mathfrak{J}^{\downarrow}_{b,u,r}(t_{x})\right)^{C}\left(\displaystyle\sum_{b'{\in}\Omega\setminus \{b\}} |\xi_{b',u}(t_{x})|^2 \rho^{\downarrow}_{b',r}(t_{x}) P^{\mathsf{max}}_{b'}{+} B_{b}N_{0}\right)}{\widehat{F}^{\downarrow}_{b,u,r}(\bm{v};\ell)}\le 1,\nonumber\\
     & \frac{\left(\mathfrak{J}^{\downarrow}_{b,u,r}(t_{x})\right)^{-C}(\mathscr{B}^{\downarrow}_{b,u,r}(t_{x}))^{-1}\left(\displaystyle\sum_{b'{\in}\Omega} |\xi_{b',u}(t_{x})|^2 \rho^{\downarrow}_{b',r}(t_{x}) P^{\mathsf{max}}_{b'}{+} B_{b}N_{0}\right)}{ \widehat{G}^{\downarrow}_{b,u,r}(\bm{v};\ell)}\le 1,\nonumber\\
     & \frac{(\mathscr{C}^{\downarrow}_{b,u,r}(t_{x}))^{-1}\left(2(CB_{b})^2+2(CB_{b})^2\ln(2)\mathfrak{R}^{\downarrow}_{b,u,r}(t_{x})+\left(\ln(2)\mathfrak{R}^{\downarrow}_{b,u,r}(t_{x})\right)^2\right)}{2(CB_{b})^2\mathfrak{J}^{\downarrow}_{b,u,r}(t_{x})}\le 1,~\frac{2(CB_{b})^2\mathfrak{J}^{\downarrow}_{b,u,r}(t_{x})}{ \widehat{H}^{\downarrow}_{b,u,r}(\bm{v};\ell)}\le 1,\nonumber\\
     &\frac{1}{\mathscr{B}^{\downarrow}_{b,u,r}(t_{x})}\le 1,~~~\frac{1}{\mathscr{C}^{\downarrow}_{b,u,r}(t_{x})}\le 1.\nonumber
\end{tcolorbox}

We next aim to transform \eqref{app:eq:DTR_R2F_C1} into standard GP format. In doing so, we rewrite \eqref{app:eq:DTR_R2F_C1} through defining $\mathfrak{R}^{\downarrow}_{b,r}(t_{x})$ as an auxiliary decision variable which must satisfy the following constraint:
\begin{equation}\label{app:eq:DTR_R2F_C1_2}
    \mathfrak{R}^{\downarrow}_{b,r}(t_{x}){=}\beta^{\downarrow}_{b}\hspace{-0.3mm}(t_{x})\min_{u\in\mathcal{U}_b}\left\{\underbrace{\frac{1+\epsilon}{\widehat{\lambda}^{(k)}_u+\epsilon}-1}_{(a)}+\mathfrak{R}^{\downarrow}_{b,u,r}(t_{x})\right\},
\end{equation}
where $x{\in}\mathcal{N}^{(k)}$. Term $(a)$ in \eqref{app:eq:DTR_R2F_C1_2} is a large value when $\widehat{\lambda}^{(k)}_u=0$ and is equal to $0$ when $\widehat{\lambda}^{(k)}_u=1$. Also, $\epsilon$ is added to the numerator and denominator of the fraction in $(a)$ to avoid division by zero. Furthermore, $\beta^{\downarrow}_{b}\hspace{-0.3mm}(t_{x})=0$ means that O-RU $b$ is not scheduled to broadcast data to the recruited FLUs and thus $\mathfrak{R}^{\downarrow}_{b,r}(t_{x})$ must be $0$, which is captured through multiplying the min function by $\beta^{\downarrow}_{b}\hspace{-0.3mm}(t_{x})$. Considering the above description, \eqref{app:eq:DTR_R2F_C1_2} is equivalent to \eqref{app:eq:DTR_R2F_C1}. To transform \eqref{app:eq:DTR_R2F_C1_2} into standard GP format, we use the approximation $\min\{A, B\}\approx (A^{-p}+B^{-p})^{-\frac{1}{p}}$, which is tight when $p \gg 1$. Accordingly, we approximate \eqref{app:eq:DTR_R2F_C1_2} as follows:
\begin{equation}\label{app:eq:DTR_R2F_C1_3}
\begin{aligned}
\mathfrak{R}^{\downarrow}_{b,r}(t_{x})&{=}\beta^{\downarrow}_{b}\hspace{-0.3mm}(t_{x})\left(\sum_{u\in\mathcal{U}_b}\left(\frac{1+\epsilon}{\widehat{\lambda}^{(k)}_u+\epsilon}-1+\mathfrak{R}^{\downarrow}_{b,u,r}(t_{x})+\underbrace{(1-\beta^{\downarrow}_{b}\hspace{-0.3mm}(t_{x}))}_{(b)}\right)^{-p}\right)^{-\frac{1}{p}}\\
&=\beta^{\downarrow}_{b}\hspace{-0.3mm}(t_{x})\left(\sum_{u\in\mathcal{U}_b}\left(\frac{1+\epsilon}{\widehat{\lambda}^{(k)}_u+\epsilon}-\beta^{\downarrow}_{b}\hspace{-0.3mm}(t_{x})+\mathfrak{R}^{\downarrow}_{b,u,r}(t_{x})\right)^{-p}\right)^{-\frac{1}{p}},
\end{aligned}
\end{equation}
where term $(b)$ is added to the above equation to avoid division by zero when scheduling decision $\beta^{\downarrow}_{b}\hspace{-0.3mm}(t_{x})$ is equal to $0$. We simplify \eqref{app:eq:DTR_R2F_C1_3} through introducing an auxiliary decision variable $J^{\downarrow}_{b,r}(t_{x},p)$ which must satisfy  
\begin{equation}\label{app:eq:DTR_R2F_C1_3_inter}
    J^{\downarrow}_{b,r}(t_{x},p)=\sum_{u\in\mathcal{U}_b}\left(\frac{1+\epsilon}{\widehat{\lambda}^{(k)}_u+\epsilon}-\beta^{\downarrow}_{b}\hspace{-0.3mm}(t_{x})+\mathfrak{R}^{\downarrow}_{b,u,r}(t_{x})\right)^{-p}.
\end{equation}
Replacing $J^{\downarrow}_{b,r}(t_{x},p)$ in \eqref{app:eq:DTR_R2F_C1_3} and adding $1$ to both sides lead to
\begin{equation}\label{app:eq:DTR_R2F_C1_4}
\mathfrak{R}^{\downarrow}_{b,r}(t_{x})+1{=} \beta^{\downarrow}_{b}\hspace{-0.3mm}(t_{x})\left(J^{\downarrow}_{b,r}(t_{x},p)\right)^{-\frac{1}{p}}+1,
\end{equation}
Dividing both sides by $\beta^{\downarrow}_{b}\hspace{-0.3mm}(t_{x})\left(J^{\downarrow}_{b,r}(t_{x},p)\right)^{-\frac{1}{p}}+1$ gives us
\begin{equation}\label{app:eq:DTR_R2F_C1_5}
\frac{\mathfrak{R}^{\downarrow}_{b,r}(t_{x})+1}{\beta^{\downarrow}_{b}\hspace{-0.3mm}(t_{x})\left(J^{\downarrow}_{b,r}(t_{x},p)\right)^{-\frac{1}{p}}+1}{=} 1,
\end{equation}
Note that we added $1$ to both sides of \eqref{app:eq:DTR_R2F_C1_4} to avoid division by zero in \eqref{app:eq:DTR_R2F_C1_5} when $\beta^{\downarrow}_{b}\hspace{-0.3mm}(t_{x})=0$. We transform \eqref{app:eq:DTR_R2F_C1_5} into the format of GP through introducing the following three inequalities:
\begin{equation}\label{app:eq:DTR_R2F_C1_6}
    \frac{\mathfrak{R}^{\downarrow}_{b,r}(t_{x})+1}{\beta^{\downarrow}_{b}\hspace{-0.3mm}(t_{x})\left(J^{\downarrow}_{b,r}(t_{x},p)\right)^{-\frac{1}{p}}+1}\le 1,
\end{equation}
\begin{equation}\label{app:eq:DTR_R2F_C1_7}
    \frac{\left(\mathscr{E}^{\downarrow}_{b,r}(t_{x})\right)^{-1}\beta^{\downarrow}_{b}\hspace{-0.3mm}(t_{x})\left(J^{\downarrow}_{b,r}(t_{x},p)\right)^{-\frac{1}{p}}+1}{\mathfrak{R}^{\downarrow}_{b,r}(t_{x})+1}\le 1,
\end{equation}
\begin{equation}
    \mathscr{E}^{\downarrow}_{b,r}(t_{x})\ge 1,
\end{equation}
where $\mathscr{E}^{\downarrow}_{b,r}(t_{x})$ is an auxiliary decision variable added with a large penalty term to the objective function to force $ \mathscr{E}^{\downarrow}_{b,r}(t_{x}){\rightarrow}1^+$ at the optimal point. The fractions in~\eqref{app:eq:DTR_R2F_C1_6} and \eqref{app:eq:DTR_R2F_C1_7} still need transformation since they are inequalities with a posynomials in their denominators, which are not posynomials. We thus exploit arithmetic-geometric mean inequality (Lemma~\ref{Lemma:ArethmaticGeometric}) to approximate the denominators of \eqref{app:eq:DTR_R2F_C1_6} and \eqref{app:eq:DTR_R2F_C1_7} with monomials. In doing so, we transform \eqref{app:eq:DTR_R2F_C1_6} into standard GP format as follows:
\begin{align}\label{app:eq:DTR_R2F_C1_8}
    L^{\downarrow}_{b,r}(\bm{v})=\beta^{\downarrow}_{b}\hspace{-0.3mm}(t_{x})\left(J^{\downarrow}_{b,r}(t_{x},p)\right)^{-\frac{1}{p}}+1&\geq \widehat{L}^{\downarrow}_{b,r}(\bm{v};\ell) \triangleq \left(\frac{\beta^{\downarrow}_{b}\hspace{-0.3mm}(t_{x})\left(J^{\downarrow}_{b,r}(t_{x},p)\right)^{-\frac{1}{p}} L^{\downarrow}_{b,r}([\bm{v}]^{(\ell-1)})}{\left[\beta^{\downarrow}_{b}\hspace{-0.3mm}(t_{x})\left(J^{\downarrow}_{b,r}(t_{x},p)\right)^{-\frac{1}{p}}\right]^{(\ell-1)}}\right)^{\frac{\left[\beta^{\downarrow}_{b}\hspace{-0.3mm}(t_{x})\left(J^{\downarrow}_{b,r}(t_{x},p)\right)^{-\frac{1}{p}}\right]^{(\ell-1)}}{L^{\downarrow}_{b,r}([\bm{v}]^{(\ell-1)})}}\nonumber\\
    &~~~~~~~~~~~~~~~~~\times\left(L^{\downarrow}_{b,r}([\bm{v}]^{(\ell-1)})\right)^{\frac{1}{L^{\downarrow}_{b,r}([\bm{v}]^{(\ell-1)})}},
\end{align}
which gives us an approximation of~\eqref{app:eq:DTR_R2F_C1_6} as follows:
\begin{equation}\label{app:eq:DTR_R2F_C1_9}
    \frac{\mathfrak{R}^{\downarrow}_{b,r}(t_{x})+1}{\widehat{L}^{\downarrow}_{b,r}(\bm{v};\ell)}\le 1.
\end{equation}
Similarly, we transform \eqref{app:eq:DTR_R2F_C1_7} into standard GP format as follows:
\begin{align}\label{app:eq:DTR_R2F_C1_10}
    N^{\downarrow}_{b,r}(\bm{v})=\mathfrak{R}^{\downarrow}_{b,r}(t_{x})+1&\geq \widehat{N}^{\downarrow}_{b,r}(\bm{v};\ell) \triangleq \left(\frac{\mathfrak{R}^{\downarrow}_{b,r}(t_{x}) N^{\downarrow}_{b,r}([\bm{v}]^{(\ell-1)})}{\left[\mathfrak{R}^{\downarrow}_{b,r}(t_{x})\right]^{(\ell-1)}}\right)^{\frac{\left[\mathfrak{R}^{\downarrow}_{b,r}(t_{x})\right]^{(\ell-1)}}{N^{\downarrow}_{b,r}([\bm{v}]^{(\ell-1)})}}\nonumber\\
    &~~~~~~~~~~~~~~~~~\times\left(N^{\downarrow}_{b,r}([\bm{v}]^{(\ell-1)})\right)^{\frac{1}{N^{\downarrow}_{b,r}([\bm{v}]^{(\ell-1)})}},
\end{align}
which gives us an approximation of~\eqref{app:eq:DTR_R2F_C1_6} as follows:
\begin{equation}\label{app:eq:DTR_R2F_C1_9}
    \frac{\left(\mathscr{E}^{\downarrow}_{b,r}(t_{x})\right)^{-1}\beta^{\downarrow}_{b}\hspace{-0.3mm}(t_{x})\left(J^{\downarrow}_{b,r}(t_{x},p)\right)^{-\frac{1}{p}}+1}{\widehat{N}^{\downarrow}_{b,r}(\bm{v};\ell)}\le 1,
\end{equation}

We next aim to transform constraint \eqref{app:eq:DTR_R2F_C1_3_inter} into the standard GP format. To this end, let us rewrite \eqref{app:eq:DTR_R2F_C1_3_inter} as follows:
\begin{equation}\label{app:eq:DTR_R2F_C1_3_inter_1}
    \frac{J^{\downarrow}_{b,r}(t_{x},p)}{\sum_{u\in\mathcal{U}_b}\left(Y_{b,u,r}(t_{x})\right)^{-p}}= 1,
\end{equation}
where $Y_{b,u,r}(t_{x})$ is an auxiliary decision variable satisfying the following equality constraint:
\begin{equation}\label{app:eq:DTR_R2F_C1_3_inter_inter_1}
    Y_{b,u,r}(t_{x}) = \frac{1+\epsilon}{\widehat{\lambda}^{(k)}_u+\epsilon}-\beta^{\downarrow}_{b}\hspace{-0.3mm}(t_{x})+\mathfrak{R}^{\downarrow}_{b,u,r}(t_{x}).
\end{equation}
We transform \eqref{app:eq:DTR_R2F_C1_3_inter_1} into GP format through splitting it into the following three inequalities:
\begin{equation}\label{app:eq:DTR_R2F_C1_3_inter_2}
    \frac{J^{\downarrow}_{b,r}(t_{x},p)}{\sum_{u\in\mathcal{U}_b}\left(Y_{b,u,r}(t_{x})\right)^{-p}}\le 1,
\end{equation}
\begin{equation}\label{app:eq:DTR_R2F_C1_3_inter_3}
    \frac{\left(\mathscr{F}^{\downarrow}_{b,r}(t_{x})\right)^{-1}\left(\sum_{u\in\mathcal{U}_b}\left(Y_{b,u,r}(t_{x})\right)^{-p}\right)}{J^{\downarrow}_{b,r}(t_{x},p)}\le 1,
\end{equation}
\begin{equation}
    \mathscr{F}^{\downarrow}_{b,r}(t_{x})\ge 1,
\end{equation}
where $\mathscr{F}^{\downarrow}_{b,r}(t_{x})$ is an auxiliary decision variable added with a large penalty term to the objective function to force $\mathscr{F}^{\downarrow}_{b,r}(t_{x}){\rightarrow}1^+$ at the optimal point. However, the fraction in~\eqref{app:eq:DTR_R2F_C1_3_inter_2} still needs transformation since it is an inequality with a posynomial in the denominator, which is not a posynomial. We thus exploit arithmetic-geometric mean inequality (Lemma~\ref{Lemma:ArethmaticGeometric}) to approximate the denominator with a monomial:
\begin{align}\label{app:eq:DTR_R2F_C1_3_inter_4}
    Q^{\downarrow}_{b,r}(\bm{v})=\sum_{u\in\mathcal{U}_b}\left(Y_{b,u,r}(t_{x})\right)^{-p}&\geq \widehat{Q}^{\downarrow}_{b,r}(\bm{v};\ell) \triangleq \prod_{u\in\mathcal{U}_b}\left(\frac{Y_{b,u,r}(t_{x}) Q^{\downarrow}_{b,r}([\bm{v}]^{(\ell-1)})}{\left[Y_{b,u,r}(t_{x})\right]^{(\ell-1)}}\right)^{\frac{\left[Y_{b,u,r}(t_{x})\right]^{(\ell-1)}}{Q^{\downarrow}_{b,r}([\bm{v}]^{(\ell-1)})}},
\end{align}
which gives us an approximation of~\eqref{app:eq:DTR_R2F_C1_3_inter_2} as follows:
\begin{equation}\label{app:eq:DTR_R2F_C1_3_inter_5}
    \frac{J^{\downarrow}_{b,r}(t_{x},p)}{\widehat{Q}^{\downarrow}_{b,r}(\bm{v};\ell)}\le 1.
\end{equation}
Using the same technique utilized above, we transform \eqref{app:eq:DTR_R2F_C1_3_inter_inter_1} into GP format. To this end, we rewrite \eqref{app:eq:DTR_R2F_C1_3_inter_inter_1} as follows: 
\begin{equation}\label{app:eq:DTR_R2F_C1_3_inter_inter_2}
    \frac{\widehat{\lambda}^{(k)}_u\times Y_{b,u,r}(t_{x})+ \widehat{\lambda}^{(k)}_u\times\beta^{\downarrow}_{b}\hspace{-0.3mm}(t_{x})+\epsilon \times Y_{b,u,r}(t_{x})+\epsilon \times\beta^{\downarrow}_{b}\hspace{-0.3mm}(t_{x})}{1+\epsilon+\widehat{\lambda}^{(k)}_u\times\mathfrak{R}^{\downarrow}_{b,u,r}(t_{x})+\epsilon\times\mathfrak{R}^{\downarrow}_{b,u,r}(t_{x})} = 1.
\end{equation}
We transform \eqref{app:eq:DTR_R2F_C1_3_inter_inter_2} into GP format through splitting it into two inequalities:
\begin{equation}\label{app:eq:DTR_R2F_C1_3_inter_inter_3}
    \frac{\widehat{\lambda}^{(k)}_u\times Y_{b,u,r}(t_{x})+ \widehat{\lambda}^{(k)}_u\times\beta^{\downarrow}_{b}\hspace{-0.3mm}(t_{x})+\epsilon \times Y_{b,u,r}(t_{x})+\epsilon \times\beta^{\downarrow}_{b}\hspace{-0.3mm}(t_{x})}{1+\epsilon+\widehat{\lambda}^{(k)}_u\times\mathfrak{R}^{\downarrow}_{b,u,r}(t_{x})+\epsilon\times\mathfrak{R}^{\downarrow}_{b,u,r}(t_{x})}\le 1,
\end{equation}
\begin{equation}\label{app:eq:DTR_R2F_C1_3_inter_inter_4}
    \frac{\left(\mathscr{H}^{\downarrow}_{b,u,r}(t_{x})\right)^{-1}\left(1+\epsilon+\widehat{\lambda}^{(k)}_u\times\mathfrak{R}^{\downarrow}_{b,u,r}(t_{x})+\epsilon\times\mathfrak{R}^{\downarrow}_{b,u,r}(t_{x})\right)}{\widehat{\lambda}^{(k)}_u\times Y_{b,u,r}(t_{x})+ \widehat{\lambda}^{(k)}_u\times\beta^{\downarrow}_{b}\hspace{-0.3mm}(t_{x})+\epsilon \times Y_{b,u,r}(t_{x})+\epsilon \times\beta^{\downarrow}_{b}\hspace{-0.3mm}(t_{x})}\le 1,
\end{equation}
\begin{equation}
    \mathscr{H}^{\downarrow}_{b,u,r}(t_{x})\ge 1,
\end{equation}
where $\mathscr{H}^{\downarrow}_{b,u,r}(t_{x})$ is an auxiliary decision variable added with a large penalty term to the objective function to force $\mathscr{H}^{\downarrow}_{b,u,r}(t_{x}){\rightarrow}1^+$ at the optimal point. However, the fractions in~\eqref{app:eq:DTR_R2F_C1_3_inter_inter_3} and \eqref{app:eq:DTR_R2F_C1_3_inter_inter_4} still need transformation since they are inequalities with posynomials in their denominators, which are not posynomials. We thus exploit arithmetic-geometric mean inequality (Lemma~\ref{Lemma:ArethmaticGeometric}) to approximate the denominators in \eqref{app:eq:DTR_R2F_C1_3_inter_inter_3} and \eqref{app:eq:DTR_R2F_C1_3_inter_inter_4} with monomials. To this end, we transform \eqref{app:eq:DTR_R2F_C1_3_inter_inter_3} into GP format as follows:
\begin{align}\label{app:eq:DTR_R2F_C1_3_inter_inter_5}
    R^{\downarrow}_{b,u,r}(\bm{v}){=}&1{+}\epsilon{+}\widehat{\lambda}^{(k)}_u\times\mathfrak{R}^{\downarrow}_{b,u,r}(t_{x}){+}\epsilon\times\mathfrak{R}^{\downarrow}_{b,u,r}(t_{x}){\geq} \widehat{R}^{\downarrow}_{b,u,r}(\bm{v};\ell) {\triangleq} \left(R^{\downarrow}_{b,u,r}([\bm{v}]^{(\ell-1)})\right)^{\frac{\left[1\right]^{(\ell-1)}}{R^{\downarrow}_{b,u,r}([\bm{v}]^{(\ell-1)})}}\times \left(R^{\downarrow}_{b,u,r}([\bm{v}]^{(\ell-1)})\right)^{\frac{\left[\epsilon\right]^{(\ell-1)}}{R^{\downarrow}_{b,u,r}([\bm{v}]^{(\ell-1)})}}\nonumber\\
    &\times \left(\frac{\widehat{\lambda}^{(k)}_u\times\mathfrak{R}^{\downarrow}_{b,u,r}(t_{x}) R^{\downarrow}_{b,u,r}([\bm{v}]^{(\ell-1)})}{\left[\widehat{\lambda}^{(k)}_u\times\mathfrak{R}^{\downarrow}_{b,u,r}(t_{x})\right]^{(\ell-1)}}\right)^{\frac{\left[\widehat{\lambda}^{(k)}_u\times\mathfrak{R}^{\downarrow}_{b,u,r}(t_{x})\right]^{(\ell-1)}}{R^{\downarrow}_{b,u,r}([\bm{v}]^{(\ell-1)})}}\times \left(\frac{\mathfrak{R}^{\downarrow}_{b,u,r}(t_{x}) R^{\downarrow}_{b,u,r}([\bm{v}]^{(\ell-1)})}{\left[\mathfrak{R}^{\downarrow}_{b,u,r}(t_{x})\right]^{(\ell-1)}}\right)^{\frac{\left[\epsilon\times\mathfrak{R}^{\downarrow}_{b,u,r}(t_{x})\right]^{(\ell-1)}}{R^{\downarrow}_{b,u,r}([\bm{v}]^{(\ell-1)})}},
\end{align}
which gives us an approximation of~\eqref{app:eq:DTR_R2F_C1_3_inter_inter_3} as follows:
\begin{equation}\label{app:eq:DTR_R2F_C1_3_inter_inter_6}
    \frac{\widehat{\lambda}^{(k)}_u\times Y_{b,u,r}(t_{x})+ \widehat{\lambda}^{(k)}_u\times\beta^{\downarrow}_{b}\hspace{-0.3mm}(t_{x})+\epsilon \times Y_{b,u,r}(t_{x})+\epsilon \times\beta^{\downarrow}_{b}\hspace{-0.3mm}(t_{x})}{\widehat{R}^{\downarrow}_{b,u,r}(\bm{v};\ell)}\le 1.
\end{equation}

Similarly, we transform \eqref{app:eq:DTR_R2F_C1_3_inter_inter_4} into GP format as follows:
\begin{align}\label{app:eq:DTR_R2F_C1_3_inter_inter_7}
    &S^{\downarrow}_{b,u,r}(\bm{v})=\widehat{\lambda}^{(k)}_u\times Y_{b,u,r}(t_{x})+ \widehat{\lambda}^{(k)}_u\times\beta^{\downarrow}_{b}\hspace{-0.3mm}(t_{x})+\epsilon \times Y_{b,u,r}(t_{x})+\epsilon \times\beta^{\downarrow}_{b}\hspace{-0.3mm}(t_{x})\nonumber\\
    &~~~\geq \widehat{S}^{\downarrow}_{b,u,r}(\bm{v};\ell) \triangleq \left(\frac{\widehat{\lambda}^{(k)}_u\times Y_{b,u,r}(t_{x}) S^{\downarrow}_{b,u,r}([\bm{v}]^{(\ell-1)})}{\left[\widehat{\lambda}^{(k)}_u\times Y_{b,u,r}(t_{x})\right]^{(\ell-1)}}\right)^{\frac{\left[\widehat{\lambda}^{(k)}_u\times Y_{b,u,r}(t_{x})\right]^{(\ell-1)}}{S^{\downarrow}_{b,u,r}([\bm{v}]^{(\ell-1)})}}\times \left(\frac{\widehat{\lambda}^{(k)}_u\times\beta^{\downarrow}_{b}\hspace{-0.3mm}(t_{x}) S^{\downarrow}_{b,u,r}([\bm{v}]^{(\ell-1)})}{\left[\widehat{\lambda}^{(k)}_u\times\beta^{\downarrow}_{b}\hspace{-0.3mm}(t_{x})\right]^{(\ell-1)}}\right)^{\frac{\left[\widehat{\lambda}^{(k)}_u\times\beta^{\downarrow}_{b}\hspace{-0.3mm}(t_{x})\right]^{(\ell-1)}}{S^{\downarrow}_{b,u,r}([\bm{v}]^{(\ell-1)})}}\nonumber\\
    &~~~\times \left(\frac{Y_{b,u,r}(t_{x}) S^{\downarrow}_{b,u,r}([\bm{v}]^{(\ell-1)})}{\left[Y_{b,u,r}(t_{x})\right]^{(\ell-1)}}\right)^{\frac{\left[\epsilon \times Y_{b,u,r}(t_{x})\right]^{(\ell-1)}}{S^{\downarrow}_{b,u,r}([\bm{v}]^{(\ell-1)})}}\times \left(\frac{\beta^{\downarrow}_{b}\hspace{-0.3mm}(t_{x}) S^{\downarrow}_{b,u,r}([\bm{v}]^{(\ell-1)})}{\left[\beta^{\downarrow}_{b}\hspace{-0.3mm}(t_{x})\right]^{(\ell-1)}}\right)^{\frac{\left[\epsilon \times\beta^{\downarrow}_{b}\hspace{-0.3mm}(t_{x})\right]^{(\ell-1)}}{S^{\downarrow}_{b,u,r}([\bm{v}]^{(\ell-1)})}},
\end{align} 
which gives us an approximation of~\eqref{app:eq:DTR_R2F_C1_3_inter_inter_4} as follows:
\begin{equation}\label{app:eq:DTR_R2F_C1_3_inter_inter_8}
    \frac{\left(\mathscr{H}^{\downarrow}_{b,u,r}(t_{x})\right)^{-1}\left(1+\epsilon+\widehat{\lambda}^{(k)}_u\times\mathfrak{R}^{\downarrow}_{b,u,r}(t_{x})+\epsilon\times\mathfrak{R}^{\downarrow}_{b,u,r}(t_{x})\right)}{\widehat{S}^{\downarrow}_{b,u,r}(\bm{v};\ell)}\le 1.
\end{equation}

We finally approximate constraint~\eqref{app:eq:DTR_R2F_C1} as follows:
\begin{tcolorbox}[ams align]
     & \frac{\mathfrak{R}^{\downarrow}_{b,r}(t_{x})+1}{\widehat{L}^{\downarrow}_{b,r}(\bm{v};\ell)}\le 1,~~~\frac{\left(\mathscr{E}^{\downarrow}_{b,r}(t_{x})\right)^{-1}\beta^{\downarrow}_{b}\hspace{-0.3mm}(t_{x})\left(J^{\downarrow}_{b,r}(t_{x},p)\right)^{-\frac{1}{p}}+1}{\widehat{N}^{\downarrow}_{b,r}(\bm{v};\ell)}\le 1,\nonumber\\
     &\frac{\left(\mathscr{F}^{\downarrow}_{b,r}(t_{x})\right)^{-1}\left(\sum_{u\in\mathcal{U}_b}\left(Y_{b,u,r}(t_{x})\right)^{-p}\right)}{J^{\downarrow}_{b,r}(t_{x},p)}\le 1,~~~\frac{J^{\downarrow}_{b,r}(t_{x},p)}{\widehat{Q}^{\downarrow}_{b,r}(\bm{v};\ell)}\le 1,\nonumber\\
     & \frac{\widehat{\lambda}^{(k)}_u\times Y_{b,u,r}(t_{x})+ \widehat{\lambda}^{(k)}_u\times\beta^{\downarrow}_{b}\hspace{-0.3mm}(t_{x})+\epsilon \times Y_{b,u,r}(t_{x})+\epsilon \times\beta^{\downarrow}_{b}\hspace{-0.3mm}(t_{x})}{\widehat{R}^{\downarrow}_{b,u,r}(\bm{v};\ell)}\le 1,\nonumber\\
     &\frac{\left(\mathscr{H}^{\downarrow}_{b,u,r}(t_{x})\right)^{-1}\left(1+\epsilon+\widehat{\lambda}^{(k)}_u\times\mathfrak{R}^{\downarrow}_{b,u,r}(t_{x})+\epsilon\times\mathfrak{R}^{\downarrow}_{b,u,r}(t_{x})\right)}{\widehat{S}^{\downarrow}_{b,u,r}(\bm{v};\ell)}\le 1,\nonumber\\
     &\frac{1}{\mathscr{F}^{\downarrow}_{b,r}(t_{x})}\le 1,~~~\frac{1}{\mathscr{E}^{\downarrow}_{b,r}(t_{x})}\le 1,~~~\frac{1}{\mathscr{H}^{\downarrow}_{b,u,r}(t_{x})}\le 1.\nonumber
\end{tcolorbox}

\textbf{$\bm{\mathsf{C2R}}$ transmission rate.} Referring to Sec.~\ref{sec:data_transmission_rate}, let us revisit the data transmission rate equation of $\mathsf{C2R}$ communication mode as follows:
\begin{equation}\label{app:eq:DR_C2R_C1} 
\mathfrak{R}^{\uparrow}_{u,r}(t_{x})= B_{b} \log_2\big(1+\Gamma^{\uparrow}_{u,r}(t_{x})\big),~x{\in}\mathcal{N}^{(k)},
\end{equation}
where
\begin{equation}\label{app:eq:DR_C2R_C2}
    \Gamma^{\uparrow}_{u,r}(t_{x}){=}\frac{\vert\xi_{b,u}(t_{x})\vert^2 \rho^{\uparrow}_{u,r}(t_{x})P^{\mathsf{max}}_{u}}{\displaystyle\sum_{b' \in \Omega} \sum_{u'\in \mathcal{U}_{b'}\setminus \{u\}}|\xi_{b,u'}(t_{x})|^2 \rho^{\uparrow}_{u',r}(t_{x})P^{\mathsf{max}}_{u'}{+} B_{b}N_{0}}.
\end{equation}
In the following, we transform \eqref{app:eq:DR_C2R_C1} into the standard form of GP. To this end, we define $\mathfrak{R}^{\uparrow}_{u,r}(t_{x})$ as an auxiliary decision variable which must satisfy the following constraint:
\begin{align}
    &\mathfrak{R}^{\uparrow}_{u,r}(t_{x})= B_{b} \log_2\big(1+\Gamma^{\uparrow}_{u,r}(t_{x})\big)\nonumber\\
    &{\Rightarrow} \frac{\ln(2)\mathfrak{R}^{\uparrow}_{u,r}(t_{x})}{B_{b}}=  \ln\big(1+\Gamma^{\uparrow}_{u,r}(t_{x})\big)\nonumber\\
    &{\Rightarrow} \exp\left(\frac{\ln(2)\mathfrak{R}^{\uparrow}_{u,r}(t_{x})}{B_{b}}\right)= \exp\left(\ln\big(1+\Gamma^{\uparrow}_{u,r}(t_{x})\big)\right)\nonumber\\
    &{\Rightarrow} \exp\left(\frac{C\ln(2)\mathfrak{R}^{\uparrow}_{u,r}(t_{x})}{C B_{b}}\right)= 1+\Gamma^{\uparrow}_{u,r}(t_{x})\nonumber\\
    &{\Rightarrow} \left(\exp\left(\frac{\ln(2)\mathfrak{R}^{\uparrow}_{u,r}(t_{x})}{C B_{b}}\right)\right)^C= 1+\Gamma^{\uparrow}_{u,r}(t_{x}).
\end{align}
Using the first three terms of the Maclaurin series of the exponential function $e^x=1+x+\frac{x^2}{2!}+\frac{x^3}{3!}+\cdots$, we get
\begin{align}\label{app:eq:DR_C2R_C1_temp}
    &{\Rightarrow} \left(1+\left(\frac{\ln(2)\mathfrak{R}^{\uparrow}_{u,r}(t_{x})}{CB_{b}}\right)+\frac{\left(\frac{\ln(2)\mathfrak{R}^{\uparrow}_{u,r}(t_{x})}{CB_{b}}\right)^2}{2}\right)^C= 1+\Gamma^{\uparrow}_{u,r}(t_{x})\nonumber\\
    &{\Rightarrow} \frac{\left(1+\left(\frac{\ln(2)\mathfrak{R}^{\uparrow}_{u,r}(t_{x})}{CB_{b}}\right)+\left(\frac{\ln(2)\mathfrak{R}^{\uparrow}_{u,r}(t_{x})}{\sqrt{2}CB_{b}}\right)^2\right)^C}{1+\Gamma^{\uparrow}_{u,r}(t_{x})}= 1.
\end{align}
In \eqref{app:eq:DR_C2R_C1_temp}, $C$ is a constant utilized to regulate the accuracy of the Maclaurin series expansion of $e^x$. Replacing \eqref{app:eq:DR_C2R_C1} into \eqref{app:eq:DR_C2R_C1_temp} leads to
\begin{align}
    \frac{\left(1+\left(\frac{\ln(2)\mathfrak{R}^{\uparrow}_{u,r}(t_{x})}{CB_{b}}\right)+\left(\frac{\ln(2)\mathfrak{R}^{\uparrow}_{u,r}(t_{x})}{\sqrt{2}CB_{b}}\right)^2\right)^C}{1+\frac{\vert\xi_{b,u}(t_{x})\vert^2 \rho^{\uparrow}_{u,r}(t_{x})P^{\mathsf{max}}_{u}}{\displaystyle\sum_{b' \in \Omega} \sum_{u'\in \mathcal{U}_{b'}\setminus\{u\}}|\xi_{b,u'}(t_{x})|^2 \rho^{\uparrow}_{u',r}(t_{x})P^{\mathsf{max}}_{u'}{+} B_{b}N_{0}}}= 1.
\end{align}
Performing some algebraic manipulations gives us
\begin{align}
     \frac{\displaystyle\left(\sum_{b' \in \Omega} \sum_{u'\in \mathcal{U}_{b'}\setminus\{u\}}|\xi_{b,u'}(t_{x})|^2 \rho^{\uparrow}_{u',r}(t_{x})P^{\mathsf{max}}_{u'}{+} B_{b}N_{0}\right)\left(1+\left(\frac{\ln(2)\mathfrak{R}^{\uparrow}_{u,r}(t_{x})}{CB_{b}}\right)+\left(\frac{\ln(2)\mathfrak{R}^{\uparrow}_{u,r}(t_{x})}{\sqrt{2}CB_{b}}\right)^2\right)^C}{\displaystyle\sum_{b' \in \Omega} \sum_{u'\in \mathcal{U}_{b'}\setminus\{u\}}|\xi_{b,u'}(t_{x})|^2 \rho^{\uparrow}_{u',r}(t_{x})P^{\mathsf{max}}_{u'}{+} B_{b}N_{0}+\vert\xi_{b,u}(t_{x})\vert^2 \rho^{\uparrow}_{u,r}(t_{x})P^{\mathsf{max}}_{u}}= 1.
\end{align}
Defining auxiliary decision variable $\mathfrak{J}^{\uparrow}_{u,r}(t_{x})$ as
\begin{equation}\label{app:eq:DR_C2R_C1_1}
    \mathfrak{J}^{\uparrow}_{u,r}(t_{x})=1+\left(\frac{\ln(2)\mathfrak{R}^{\uparrow}_{u,r}(t_{x})}{CB_{b}}\right)+\left(\frac{\ln(2)\mathfrak{R}^{\uparrow}_{u,r}(t_{x})}{\sqrt{2}CB_{b}}\right)^2
\end{equation}
results in the following constraint:
\begin{align}\label{app:eq:DR_C2R_C1_2}
    \frac{\left(\mathfrak{J}^{\uparrow}_{u,r}(t_{x})\right)^C\left(\displaystyle\sum_{b' \in \Omega} \sum_{u'\in \mathcal{U}_{b'}\setminus\{u\}}|\xi_{b,u'}(t_{x})|^2 \rho^{\uparrow}_{u',r}(t_{x})P^{\mathsf{max}}_{u'}{+} B_{b}N_{0}\right)}{\displaystyle\sum_{b' \in \Omega} \sum_{u'\in \mathcal{U}_{b'}\setminus\{u\}}|\xi_{b,u'}(t_{x})|^2 \rho^{\uparrow}_{u',r}(t_{x})P^{\mathsf{max}}_{u'}{+} B_{b}N_{0}+\vert\xi_{b,u}(t_{x})\vert^2 \rho^{\uparrow}_{u,r}(t_{x})P^{\mathsf{max}}_{u}}= 1.
\end{align}
However, this constraint is not in the format of GP. Therefore, we transform it into the standard GP format via introducing the following three inequalities:
\begin{equation}\label{app:eq:DR_C2R_C1_3}
     \frac{\left(\mathfrak{J}^{\uparrow}_{u,r}(t_{x})\right)^C\left(\displaystyle\sum_{b' \in \Omega} \sum_{u'\in \mathcal{U}_{b'}\setminus\{u\}}|\xi_{b,u'}(t_{x})|^2 \rho^{\uparrow}_{u',r}(t_{x})P^{\mathsf{max}}_{u'}{+} B_{b}N_{0}\right)}{\displaystyle\sum_{b' \in \Omega} \sum_{u'\in \mathcal{U}_{b'}\setminus\{u\}}|\xi_{b,u'}(t_{x})|^2 \rho^{\uparrow}_{u',r}(t_{x})P^{\mathsf{max}}_{u'}+\vert\xi_{b,u}(t_{x})\vert^2 \rho^{\uparrow}_{u,r}(t_{x})P^{\mathsf{max}}_{u}{+} B_{b}N_{0}}\le 1,
\end{equation}
\begin{equation}\label{app:eq:DR_C2R_C1_4}
    \frac{\displaystyle\left(\mathfrak{J}^{\uparrow}_{u,r}(t_{x})\right)^{-C}\left(\mathscr{B}^{\uparrow}_{u,r}(t_{x})\right)^{-1}\left(\sum_{b' \in \Omega} \sum_{u'\in \mathcal{U}_{b'}\setminus\{u\}}|\xi_{b,u'}(t_{x})|^2 \rho^{\uparrow}_{u',r}(t_{x})P^{\mathsf{max}}_{u'}{+} B_{b}N_{0}+\vert\xi_{b,u}(t_{x})\vert^2 \rho^{\uparrow}_{u,r}(t_{x})P^{\mathsf{max}}_{u}\right)}{\displaystyle\sum_{b' \in \Omega} \sum_{u'\in \mathcal{U}_{b'}\setminus\{u\}}|\xi_{b,u'}(t_{x})|^2 \rho^{\uparrow}_{u',r}(t_{x})P^{\mathsf{max}}_{u'}{+}  B_{b}N_{0}}\le 1,
\end{equation}
\begin{equation}
    \mathscr{B}^{\uparrow}_{u,r}(t_{x})\ge 1,
\end{equation}
where $\mathscr{B}^{\uparrow}_{u,r}(t_{x})$ is an auxiliary decision variable added with a large penalty term to the objective function to force $\mathscr{B}^{\uparrow}_{u,r}(t_{x}){\rightarrow}1^+$ at the optimal point. The fractions in~\eqref{app:eq:DR_C2R_C1_3} and \eqref{app:eq:DR_C2R_C1_4} still need transformation since they are inequalities with posynomials in their denominator, which are not posynomials. We thus exploit arithmetic-geometric mean inequality (Lemma~\ref{Lemma:ArethmaticGeometric}) to approximate the denominators of \eqref{app:eq:DR_C2R_C1_3} and \eqref{app:eq:DR_C2R_C1_4} with monomials. In doing so, we approximate the denominator in \eqref{app:eq:DR_C2R_C1_3} as follows:
\begin{align}\label{app:eq:DR_C2R_C1_5}
    F^{\uparrow}_{u,r}(\bm{v}){=}&\sum_{b' \in \Omega} \sum_{u'\in \mathcal{U}_{b'}\setminus\{u\}}|\xi_{b,u'}(t_{x})|^2 \rho^{\uparrow}_{u',r}(t_{x})P^{\mathsf{max}}_{u'}+\vert\xi_{b,u}(t_{x})\vert^2 \rho^{\uparrow}_{u,r}(t_{x})P^{\mathsf{max}}_{u} + B_{b}N_{0}\nonumber\\
    &{\geq} \widehat{F}^{\uparrow}_{u,r}(\bm{v};\ell) {\triangleq} \prod_{b' \in \Omega}\prod_{u'\in \mathcal{U}_{b'}\setminus\{u\}}\left(\frac{|\xi_{b,u'}(t_{x})|^2 \rho^{\uparrow}_{u',r}(t_{x}) F^{\uparrow}_{u,r}([\bm{v}]^{(\ell-1)})}{\left[|\xi_{b,u'}(t_{x})|^2 \rho^{\uparrow}_{u',r}(t_{x})\right]^{(\ell-1)}}\right)^{\hspace{-2mm}\frac{\left[|\xi_{b,u'}(t_{x})|^2 \rho^{\uparrow}_{u',r}(t_{x})P^{\mathsf{max}}_{u'}\right]^{(\ell-1)}}{F^{\uparrow}_{u,r}([\bm{v}]^{(\ell-1)})}}\nonumber\\
    &\times\left(\frac{\vert\xi_{b,u}(t_{x})\vert^2 \rho^{\uparrow}_{u,r}(t_{x}) F^{\uparrow}_{u,r}([\bm{v}]^{(\ell-1)})}{\left[\vert\xi_{b,u}(t_{x})\vert^2 \rho^{\uparrow}_{u,r}(t_{x})\right]^{(\ell-1)}}\right)^{\frac{\left[\vert\xi_{b,u}(t_{x})\vert^2 \rho^{\uparrow}_{u,r}(t_{x})P^{\mathsf{max}}_{u}\right]^{(\ell-1)}}{F^{\uparrow}_{u,r}([\bm{v}]^{(\ell-1)})}}\times\left(F^{\uparrow}_{u,r}([\bm{v}]^{(\ell-1)})\right)^{\frac{\left[B_{b}N_{0}\right]^{(\ell-1)}}{F^{\uparrow}_{u,r}([\bm{v}]^{(\ell-1)})}},
\end{align}
which gives us the following approximation of~\eqref{app:eq:DR_C2R_C1_3}:
\begin{equation}\label{app:eq:DR_C2R_C1_6}
     \frac{\left(\mathfrak{J}^{\uparrow}_{u,r}(t_{x})\right)^C\left(\displaystyle\sum_{b' \in \Omega} \sum_{u'\in \mathcal{U}_{b'}\setminus\{u\}}|\xi_{b,u'}(t_{x})|^2 \rho^{\uparrow}_{u',r}(t_{x})P^{\mathsf{max}}_{u'}{+} B_{b}N_{0}\right)}{\widehat{F}^{\uparrow}_{u,r}(\bm{v};\ell)}\le 1.
\end{equation}
Similarly, we approximate the denominator in \eqref{app:eq:DR_C2R_C1_4} as follows:
\begin{align}\label{app:eq:DR_C2R_C1_7}
    G^{\uparrow}_{u,r}(\bm{v}){=}& \sum_{b' \in \Omega} \sum_{u'\in \mathcal{U}_{b'}\setminus\{u\}}|\xi_{b,u'}(t_{x})|^2 \rho^{\uparrow}_{u',r}(t_{x})P^{\mathsf{max}}_{u'}{+}  B_{b}N_{0}\nonumber\\
    &{\geq} \widehat{G}^{\uparrow}_{u,r}(\bm{v};\ell) {\triangleq} \prod_{b' \in \Omega}\prod_{u'\in \mathcal{U}_{b'}\setminus\{u\}}\left(\frac{|\xi_{b',u}(t_{x})|^2 \rho^{\downarrow}_{b',r}(t_{x})  G^{\uparrow}_{u,r}([\bm{v}]^{(\ell-1)})}{\left[|\xi_{b',u}(t_{x})|^2 \rho^{\downarrow}_{b',r}(t_{x}) \right]^{(\ell-1)}}\right)^{\hspace{-2mm}\frac{\left[|\xi_{b',u}(t_{x})|^2 \rho^{\downarrow}_{b',r}(t_{x}) P^{\mathsf{max}}_{b'}\right]^{(\ell-1)}}{G^{\uparrow}_{u,r}([\bm{v}]^{(\ell-1)})}}\nonumber\\
    &~~~~~~~~~~~~~~~~~~~~~~~~~\times\left(G^{\uparrow}_{u,r}([\bm{v}]^{(\ell-1)})\right)^{\frac{\left[B_{b}N_{0}\right]^{(\ell-1)}}{G^{\uparrow}_{u,r}([\bm{v}]^{(\ell-1)})}},
\end{align}
which gives the following approximation of~\eqref{app:eq:DR_C2R_C1_4}:
\begin{equation}\label{app:eq:DR_C2R_C1_8}
     \frac{\displaystyle\left(\mathfrak{J}^{\uparrow}_{u,r}(t_{x})\right)^{-C}\left(\mathscr{B}^{\uparrow}_{u,r}(t_{x})\right)^{-1}\left(\sum_{b' \in \Omega} \sum_{u'\in \mathcal{U}_{b'}\setminus\{u\}}|\xi_{b,u'}(t_{x})|^2 \rho^{\uparrow}_{u',r}(t_{x})P^{\mathsf{max}}_{u'}{+} B_{b}N_{0}+\vert\xi_{b,u}(t_{x})\vert^2 \rho^{\uparrow}_{u,r}(t_{x})P^{\mathsf{max}}_{u}\right)}{ \widehat{G}^{\uparrow}_{u,r}(\bm{v};\ell)}\le 1.
\end{equation}

We next aim to transform \eqref{app:eq:DR_C2R_C1_1} into standard GP format. In doing so, let us rewrite \eqref{app:eq:DR_C2R_C1_1} as follows:
\begin{equation}\label{app:eq:DR_C2R_C1_9}
    \frac{2(CB_{b})^2\mathfrak{J}^{\uparrow}_{u,r}(t_{x})}{2(CB_{b})^2+2(CB_{b})^2\ln(2)\mathfrak{R}^{\uparrow}_{u,r}(t_{x})+\left(\ln(2)\mathfrak{R}^{\uparrow}_{u,r}(t_{x})\right)^2}=1.
\end{equation}
We transform this constraint into the format of GP by splitting it into the following three inequalities:
\begin{equation}\label{app:eq:DR_C2R_C1_10}
    \frac{2(CB_{b})^2\mathfrak{J}^{\uparrow}_{u,r}(t_{x})}{2(CB_{b})^2+2(CB_{b})^2\ln(2)\mathfrak{R}^{\uparrow}_{u,r}(t_{x})+\left(\ln(2)\mathfrak{R}^{\uparrow}_{u,r}(t_{x})\right)^2}\le 1,
\end{equation}
\begin{equation}\label{app:eq:DR_C2R_C1_11}
    \frac{(\mathscr{C}^{\uparrow}_{u,r}(t_{x}))^{-1}\left(2(CB_{b})^2+2(CB_{b})^2\ln(2)\mathfrak{R}^{\uparrow}_{u,r}(t_{x})+\left(\ln(2)\mathfrak{R}^{\uparrow}_{u,r}(t_{x})\right)^2\right)}{2(CB_{b})^2\mathfrak{J}^{\uparrow}_{u,r}(t_{x})}\le 1,
\end{equation}
\begin{equation}
    \mathscr{C}^{\uparrow}_{u,r}(t_{x})\ge 1,
\end{equation}
where $\mathscr{C}^{\uparrow}_{u,r}(t_{x})$ is an auxiliary decision variable and will be added with a large penalty term to the objective function to force $\mathscr{C}^{\uparrow}_{u,r}(t_{x}){\rightarrow}1^+$ at the optimal point. The fraction in~\eqref{app:eq:DR_C2R_C1_10} still needs transformation since it is inequality with posynomial in its denominator, which is not posynomial. We thus exploit arithmetic-geometric mean inequality (Lemma~\ref{Lemma:ArethmaticGeometric}) to approximate the denominator of \eqref{app:eq:DR_C2R_C1_3} with monomial. 
\begin{align}\label{app:eq:DR_C2R_C1_12}
    &H^{\uparrow}_{u,r}(\bm{v}){=} 2(CB_{b})^2+2(CB_{b})^2\ln(2)\mathfrak{R}^{\uparrow}_{u,r}(t_{x})+\left(\ln(2)\mathfrak{R}^{\uparrow}_{u,r}(t_{x})\right)^2 {\geq} \widehat{H}^{\uparrow}_{u,r}(\bm{v};\ell) {\triangleq} \left(H^{\uparrow}_{u,r}([\bm{v}]^{(\ell-1)})\right)^{\hspace{-1mm}\frac{\left[2(CB_{b})^2\right]^{(\ell-1)}}{H^{\uparrow}_{u,r}([\bm{v}]^{(\ell-1)})}}\nonumber\\
    &\times\left(\frac{\mathfrak{R}^{\uparrow}_{u,r}(t_{x}) H^{\uparrow}_{u,r}([\bm{v}]^{(\ell-1)})}{\left[\mathfrak{R}^{\uparrow}_{u,r}(t_{x})\right]^{(\ell-1)}}\right)^{\frac{\left[2(CB_{b})^2\ln(2)\mathfrak{R}^{\uparrow}_{u,r}(t_{x})\right]^{(\ell-1)}}{H^{\uparrow}_{u,r}([\bm{v}]^{(\ell-1)})}}\times\left(\frac{\left(\mathfrak{R}^{\uparrow}_{u,r}(t_{x})\right)^2 H^{\uparrow}_{u,r}([\bm{v}]^{(\ell-1)})}{\left[\left(\mathfrak{R}^{\uparrow}_{u,r}(t_{x})\right)^2\right]^{(\ell-1)}}\right)^{\frac{\left[\left(\ln(2)\mathfrak{R}^{\uparrow}_{u,r}(t_{x})\right)^2\right]^{(\ell-1)}}{H^{\uparrow}_{u,r}([\bm{v}]^{(\ell-1)})}},
\end{align}
which gives the following approximation of~\eqref{app:eq:DR_C2R_C1_10}:
\begin{equation}\label{app:eq:DR_C2R_C1_14}
     \frac{2(CB_{b})^2\mathfrak{J}^{\uparrow}_{u,r}(t_{x})}{ \widehat{H}^{\uparrow}_{u,r}(\bm{v};\ell)}\le 1.
\end{equation}

We finally approximate \eqref{app:eq:DR_C2R_C1} as follows:
\begin{tcolorbox}[ams align]
     & \frac{\left(\mathfrak{J}^{\uparrow}_{u,r}(t_{x})\right)^C\left(\displaystyle\sum_{b' \in \Omega} \sum_{u'\in \mathcal{U}_{b'}\setminus\{u\}}|\xi_{b,u'}(t_{x})|^2 \rho^{\uparrow}_{u',r}(t_{x})P^{\mathsf{max}}_{u'}{+} B_{b}N_{0}\right)}{\widehat{F}^{\uparrow}_{u,r}(\bm{v};\ell)}\le 1,\nonumber\\
     & \frac{\displaystyle\left(\mathfrak{J}^{\uparrow}_{u,r}(t_{x})\right)^{-C}\left(\mathscr{B}^{\uparrow}_{u,r}(t_{x})\right)^{-1}\left(\sum_{b' \in \Omega} \sum_{u'\in \mathcal{U}_{b'}\setminus\{u\}}|\xi_{b,u'}(t_{x})|^2 \rho^{\uparrow}_{u',r}(t_{x})P^{\mathsf{max}}_{u'}{+} B_{b}N_{0}+\vert\xi_{b,u}(t_{x})\vert^2 \rho^{\uparrow}_{u,r}(t_{x})P^{\mathsf{max}}_{u}\right)}{ \widehat{G}^{\uparrow}_{u,r}(\bm{v};\ell)}\le 1,\nonumber\\
     &\frac{(\mathscr{C}^{\uparrow}_{u,r}(t_{x}))^{-1}\left(2(CB_{b})^2+2(CB_{b})^2\ln(2)\mathfrak{R}^{\uparrow}_{u,r}(t_{x})+\left(\ln(2)\mathfrak{R}^{\uparrow}_{u,r}(t_{x})\right)^2\right)}{2(CB_{b})^2\mathfrak{J}^{\uparrow}_{u,r}(t_{x})}\le 1,~~\frac{2(CB_{b})^2\mathfrak{J}^{\uparrow}_{u,r}(t_{x})}{ \widehat{H}^{\uparrow}_{u,r}(\bm{v};\ell)}\le 1,\nonumber\\
     &\frac{1}{\mathscr{B}^{\uparrow}_{u,r}(t_{x})}\le 1,~~~\frac{1}{\mathscr{C}^{\uparrow}_{u,r}(t_{x})}\le 1.\nonumber
\end{tcolorbox}

\textbf{$\bm{\mathsf{D2C}}$ transmission rate.} Referring to Sec.~\ref{sec:data_transmission_rate}, let us revisit the data transmission rate equation of $\mathsf{D2C}$ communication mode as follows:
\begin{equation} \label{app:eq:DR_D2C_C1}
\overline{\mathfrak{R}}^{\uparrow}_{u,u',r}(t_{x}){=}\overline{B}_{b} \log_2\left(1+\overline{\Gamma}^{\uparrow}_{u,u',r}(t_{x})\right),~x{\in}\mathcal{N}^{(k)},
\end{equation}
where
\begin{equation}\label{app:eq:DR_D2C_C2}
    \overline{\Gamma}^{\uparrow}_{u,u',r}(t_{x}){=}\frac{\vert\xi_{u,u'}(t_{x})\vert^2 \overline{\rho}^{\uparrow}_{u,r}(t_{x})P^{\mathsf{max}}_{u}}{\displaystyle\sum_{b' \in \Omega} \sum_{\widehat{u},\widehat{u}'\in \mathcal{U}_{b'}\setminus \{u,u'\}}|\xi_{\widehat{u},u'}(t_{x})|^2 \rho_{\widehat{u},\widehat{u}',r}(t_{x})P^{\mathsf{max}}_{\widehat{u}}{+}\overline{B}_{b}N_{0}}.
\end{equation}
In the following, we transform \eqref{app:eq:DR_D2C_C1} into the standard form of GP. To this end, we define $\overline{\mathfrak{R}}^{\uparrow}_{u,u',r}(t_{x})$ as an auxiliary decision variable that must satisfy the following constraint:

\begin{align}
    &\overline{\mathfrak{R}}^{\uparrow}_{u,u',r}(t_{x})= \overline{B}_{b} \log_2\big(1+\overline{\Gamma}^{\uparrow}_{u,u',r}(t_{x})\big)\nonumber\\
    &{\Rightarrow} \frac{\ln(2)\overline{\mathfrak{R}}^{\uparrow}_{u,u',r}(t_{x})}{\overline{B}_{b}}=  \ln\big(1+\overline{\Gamma}^{\uparrow}_{u,u',r}(t_{x})\big)\nonumber\\
    &{\Rightarrow} \exp\left(\frac{\ln(2)\overline{\mathfrak{R}}^{\uparrow}_{u,u',r}(t_{x})}{\overline{B}_{b}}\right)= \exp\left(\ln\big(1+\overline{\Gamma}^{\uparrow}_{u,u',r}(t_{x})\big)\right)\nonumber\\
    &{\Rightarrow} \exp\left(\frac{C\ln(2)\overline{\mathfrak{R}}^{\uparrow}_{u,u',r}(t_{x})}{C \overline{B}_{b}}\right)= 1+\overline{\Gamma}^{\uparrow}_{u,u',r}(t_{x})\nonumber\\
    &{\Rightarrow} \left(\exp\left(\frac{\ln(2)\overline{\mathfrak{R}}^{\uparrow}_{u,u',r}(t_{x})}{C \overline{B}_{b}}\right)\right)^C= 1+\overline{\Gamma}^{\uparrow}_{u,u',r}(t_{x}).
\end{align}
Using the first three terms of the Maclaurin series of the exponential function $e^x=1+x+\frac{x^2}{2!}+\frac{x^3}{3!}+\cdots$, we get
\begin{align}\label{app:eq:DR_D2C_C1_temp}
    &{\Rightarrow} \left(1+\left(\frac{\ln(2)\overline{\mathfrak{R}}^{\uparrow}_{u,u',r}(t_{x})}{C\overline{B}_{b}}\right)+\frac{\left(\frac{\ln(2)\overline{\mathfrak{R}}^{\uparrow}_{u,u',r}(t_{x})}{C\overline{B}_{b}}\right)^2}{2}\right)^C= 1+\overline{\Gamma}^{\uparrow}_{u,u',r}(t_{x})\nonumber\\
    &{\Rightarrow} \frac{\left(1+\left(\frac{\ln(2)\overline{\mathfrak{R}}^{\uparrow}_{u,u',r}(t_{x})}{C\overline{B}_{b}}\right)+\left(\frac{\ln(2)\overline{\mathfrak{R}}^{\uparrow}_{u,u',r}(t_{x})}{\sqrt{2}C\overline{B}_{b}}\right)^2\right)^C}{1+\overline{\Gamma}^{\uparrow}_{u,u',r}(t_{x})}= 1.
\end{align}
In \eqref{app:eq:DR_D2C_C1_temp}, $C$ is a constant utilized to regulate the accuracy of the Maclaurin series expansion of $e^x$. Replacing \eqref{app:eq:DR_D2C_C1} into \eqref{app:eq:DR_D2C_C1_temp} leads to
\begin{align}
    \frac{\left(1+\left(\frac{\ln(2)\overline{\mathfrak{R}}^{\uparrow}_{u,u',r}(t_{x})}{C\overline{B}_{b}}\right)+\left(\frac{\ln(2)\overline{\mathfrak{R}}^{\uparrow}_{u,u',r}(t_{x})}{\sqrt{2}C\overline{B}_{b}}\right)^2\right)^C}{1+\frac{\vert\xi_{u,u'}(t_{x})\vert^2 \overline{\rho}^{\uparrow}_{u,r}(t_{x})P^{\mathsf{max}}_{u}}{\displaystyle\sum_{b' \in \Omega} \sum_{\widehat{u},\widehat{u}'\in \mathcal{U}_{b'}\setminus \{u,u'\}}|\xi_{\widehat{u},u'}(t_{x})|^2 \rho_{\widehat{u},\widehat{u}',r}(t_{x})P^{\mathsf{max}}_{\widehat{u}}{+}\overline{B}_{b}N_{0}}}= 1.
\end{align}
Performing some algebraic manipulations gives us
\begin{align}
     \frac{\left(\displaystyle\sum_{b' \in \Omega} \sum_{\widehat{u},\widehat{u}'\in \mathcal{U}_{b'}\setminus \{u,u'\}}|\xi_{\widehat{u},u'}(t_{x})|^2 \rho_{\widehat{u},\widehat{u}',r}(t_{x})P^{\mathsf{max}}_{\widehat{u}}{+}\overline{B}_{b}N_{0}\right)\left(1+\left(\frac{\ln(2)\overline{\mathfrak{R}}^{\uparrow}_{u,u',r}(t_{x})}{C\overline{B}_{b}}\right)+\left(\frac{\ln(2)\overline{\mathfrak{R}}^{\uparrow}_{u,u',r}(t_{x})}{\sqrt{2}C\overline{B}_{b}}\right)^2\right)^C}{\displaystyle\sum_{b' \in \Omega} \sum_{\widehat{u},\widehat{u}'\in \mathcal{U}_{b'}\setminus \{u,u'\}}|\xi_{\widehat{u},u'}(t_{x})|^2 \rho_{\widehat{u},\widehat{u}',r}(t_{x})P^{\mathsf{max}}_{\widehat{u}}{+}\overline{B}_{b}N_{0}{+} \vert\xi_{u,u'}(t_{x})\vert^2 \overline{\rho}^{\uparrow}_{u,r}(t_{x})P^{\mathsf{max}}_{u}}= 1.
\end{align}
Defining auxiliary decision variable $\overline{\mathfrak{J}}^{\uparrow}_{u,u',r}(t_{x})$ as
\begin{equation}\label{app:eq:DR_D2C_C1_1}
    \overline{\mathfrak{J}}^{\uparrow}_{u,u',r}(t_{x})=1+\left(\frac{\ln(2)\overline{\mathfrak{R}}^{\uparrow}_{u,u',r}(t_{x})}{C\overline{B}_{b}}\right)+\left(\frac{\ln(2)\overline{\mathfrak{R}}^{\uparrow}_{u,u',r}(t_{x})}{\sqrt{2}C\overline{B}_{b}}\right)^2
\end{equation}
results in the following constraint:
\begin{align}\label{app:eq:DR_D2C_C1_2}
    \frac{\left(\overline{\mathfrak{J}}^{\uparrow}_{u,u',r}(t_{x})\right)^C\left(\displaystyle\sum_{b' \in \Omega} \sum_{\widehat{u},\widehat{u}'\in \mathcal{U}_{b'}\setminus \{u,u'\}}|\xi_{\widehat{u},u'}(t_{x})|^2 \rho_{\widehat{u},\widehat{u}',r}(t_{x})P^{\mathsf{max}}_{\widehat{u}}{+}\overline{B}_{b}N_{0}\right)}{\displaystyle\sum_{b' \in \Omega} \sum_{\widehat{u},\widehat{u}'\in \mathcal{U}_{b'}\setminus \{u,u'\}}|\xi_{\widehat{u},u'}(t_{x})|^2 \rho_{\widehat{u},\widehat{u}',r}(t_{x})P^{\mathsf{max}}_{\widehat{u}}{+}\overline{B}_{b}N_{0}{+} \vert\xi_{u,u'}(t_{x})\vert^2 \overline{\rho}^{\uparrow}_{u,r}(t_{x})P^{\mathsf{max}}_{u}}= 1.
\end{align}
However, this constraint is not in the format of GP. Therefore, we transform it into the standard GP format via introducing the following three inequalities:
\begin{equation}\label{app:eq:DR_D2C_C1_3}
      \frac{\left(\overline{\mathfrak{J}}^{\uparrow}_{u,u',r}(t_{x})\right)^C\left(\displaystyle\sum_{b' \in \Omega} \sum_{\widehat{u},\widehat{u}'\in \mathcal{U}_{b'}\setminus \{u,u'\}}|\xi_{\widehat{u},u'}(t_{x})|^2 \rho_{\widehat{u},\widehat{u}',r}(t_{x})P^{\mathsf{max}}_{\widehat{u}}{+}\overline{B}_{b}N_{0}\right)}{\displaystyle\sum_{b' \in \Omega} \sum_{\widehat{u},\widehat{u}'\in \mathcal{U}_{b'}\setminus \{u,u'\}}|\xi_{\widehat{u},u'}(t_{x})|^2 \rho_{\widehat{u},\widehat{u}',r}(t_{x})P^{\mathsf{max}}_{\widehat{u}}{+}\overline{B}_{b}N_{0}{+} \vert\xi_{u,u'}(t_{x})\vert^2 \overline{\rho}^{\uparrow}_{u,r}(t_{x})P^{\mathsf{max}}_{u}}\le 1,
\end{equation}
\begin{equation}\label{app:eq:DR_D2C_C1_4}
     \frac{\left(\overline{\mathfrak{J}}^{\uparrow}_{u,u',r}(t_{x})\right)^{-C}\left(\mathscr{B}^{\uparrow}_{u,u',r}(t_{x})\right)^{-1}\left(\displaystyle\sum_{b' \in \Omega} \sum_{\widehat{u},\widehat{u}'\in \mathcal{U}_{b'}\setminus \{u,u'\}}|\xi_{\widehat{u},u'}(t_{x})|^2 \rho_{\widehat{u},\widehat{u}',r}(t_{x})P^{\mathsf{max}}_{\widehat{u}}{+}\overline{B}_{b}N_{0}{+} \vert\xi_{u,u'}(t_{x})\vert^2 \overline{\rho}^{\uparrow}_{u,r}(t_{x})P^{\mathsf{max}}_{u}\right)}{\displaystyle\sum_{b' \in \Omega} \sum_{\widehat{u},\widehat{u}'\in \mathcal{U}_{b'}\setminus \{u,u'\}}|\xi_{\widehat{u},u'}(t_{x})|^2 \rho_{\widehat{u},\widehat{u}',r}(t_{x})P^{\mathsf{max}}_{\widehat{u}}{+}\overline{B}_{b}N_{0}}\le 1,
\end{equation}
\begin{equation}
    \mathscr{B}^{\uparrow}_{u,u',r}(t_{x})\ge 1,
\end{equation}
where $\mathscr{B}^{\uparrow}_{u,u',r}(t_{x})$ is an auxiliary decision variable added with a large penalty term to the objective function to force $\mathscr{B}^{\uparrow}_{u,u',r}(t_{x}){\rightarrow}1^+$ at the optimal point. The fractions in~\eqref{app:eq:DR_D2C_C1_3} and \eqref{app:eq:DR_D2C_C1_4} still need transformation since they are inequalities with posynomials in their denominator, which are not posynomials. We thus exploit arithmetic-geometric mean inequality (Lemma~\ref{Lemma:ArethmaticGeometric}) to approximate the denominators of \eqref{app:eq:DR_D2C_C1_3} and \eqref{app:eq:DR_D2C_C1_4} with monomials. In doing so, we approximate the denominator in \eqref{app:eq:DR_D2C_C1_3} as follows:
\begin{align}\label{app:eq:DR_D2C_C1_5}
    F^{\uparrow}_{u,u',r}(\bm{v})&{=}\sum_{b' \in \Omega} \sum_{\widehat{u},\widehat{u}'\in \mathcal{U}_{b'}\setminus \{u,u'\}}|\xi_{\widehat{u},u'}(t_{x})|^2 \rho_{\widehat{u},\widehat{u}',r}(t_{x})P^{\mathsf{max}}_{\widehat{u}}{+} \vert\xi_{u,u'}(t_{x})\vert^2 \overline{\rho}^{\uparrow}_{u,r}(t_{x})P^{\mathsf{max}}_{u}{+}\overline{B}_{b}N_{0}\nonumber\\
    &{\geq} \widehat{F}^{\uparrow}_{u,u',r}(\bm{v};\ell) {\triangleq} \prod_{b' \in \Omega\setminus \{b\}}\prod_{\widehat{u},\widehat{u}'\in \mathcal{U}_{b'}}\left(\frac{|\xi_{\widehat{u},u'}(t_{x})|^2 \rho_{\widehat{u},\widehat{u}',r}(t_{x}) F^{\uparrow}_{u,u',r}([\bm{v}]^{(\ell-1)})}{\left[|\xi_{\widehat{u},u'}(t_{x})|^2 \rho_{\widehat{u},\widehat{u}',r}(t_{x})\right]^{(\ell-1)}}\right)^{\hspace{-2mm}\frac{\left[|\xi_{\widehat{u},u'}(t_{x})|^2 \rho_{\widehat{u},\widehat{u}',r}(t_{x})P^{\mathsf{max}}_{\widehat{u}}\right]^{(\ell-1)}}{F^{\uparrow}_{u,u',r}([\bm{v}]^{(\ell-1)})}}\nonumber\\
    &\times\left(\frac{ \vert\xi_{u,u'}(t_{x})\vert^2 \overline{\rho}^{\uparrow}_{u,r}(t_{x}) F^{\uparrow}_{u,u',r}([\bm{v}]^{(\ell-1)})}{\left[ \vert\xi_{u,u'}(t_{x})\vert^2 \overline{\rho}^{\uparrow}_{u,r}(t_{x})\right]^{(\ell-1)}}\right)^{\frac{\left[ \vert\xi_{u,u'}(t_{x})\vert^2 \overline{\rho}^{\uparrow}_{u,r}(t_{x})P^{\mathsf{max}}_{u}\right]^{(\ell-1)}}{F^{\uparrow}_{u,u',r}([\bm{v}]^{(\ell-1)})}}\times\left( F^{\uparrow}_{u,u',r}([\bm{v}]^{(\ell-1)})\right)^{\frac{\left[\overline{B}_{b}N_{0}\right]^{(\ell-1)}}{F^{\uparrow}_{u,u',r}([\bm{v}]^{(\ell-1)})}},
\end{align}
which gives us the following approximation of~\eqref{app:eq:DR_D2C_C1_3}:
\begin{equation}\label{app:eq:DR_D2C_C1_6}
     \frac{\left(\overline{\mathfrak{J}}^{\uparrow}_{u,u',r}(t_{x})\right)^C\left(\displaystyle\sum_{b' \in \Omega} \sum_{\widehat{u},\widehat{u}'\in \mathcal{U}_{b'}\setminus \{u,u'\}}|\xi_{\widehat{u},u'}(t_{x})|^2 \rho_{\widehat{u},\widehat{u}',r}(t_{x})P^{\mathsf{max}}_{\widehat{u}}{+}\overline{B}_{b}N_{0}\right)}{\widehat{F}^{\uparrow}_{u,u',r}(\bm{v};\ell)}\le 1,
\end{equation}
Similarly, we approximate the denominator in \eqref{app:eq:DR_D2C_C1_4} as follows:
\begin{align}\label{app:eq:DR_D2C_C1_7}
    G^{\uparrow}_{u,u',r}(\bm{v}){=}& \sum_{b' \in \Omega} \sum_{\widehat{u},\widehat{u}'\in \mathcal{U}_{b'}\setminus \{u,u'\}}|\xi_{\widehat{u},u'}(t_{x})|^2 \rho_{\widehat{u},\widehat{u}',r}(t_{x})P^{\mathsf{max}}_{\widehat{u}}{+}\overline{B}_{b}N_{0}\nonumber\\
    &{\geq} \widehat{G}^{\uparrow}_{u,u',r}(\bm{v};\ell) {\triangleq} \prod_{b' \in \Omega\setminus \{b\}}\prod_{\widehat{u},\widehat{u}'\in \mathcal{U}_{b'}}\left(\frac{|\xi_{\widehat{u},u'}(t_{x})|^2 \rho_{\widehat{u},\widehat{u}',r}(t_{x}) G^{\uparrow}_{u,u',r}([\bm{v}]^{(\ell-1)})}{\left[|\xi_{\widehat{u},u'}(t_{x})|^2 \rho_{\widehat{u},\widehat{u}',r}(t_{x})\right]^{(\ell-1)}}\right)^{\hspace{-2mm}\frac{\left[|\xi_{\widehat{u},u'}(t_{x})|^2 \rho_{\widehat{u},\widehat{u}',r}(t_{x})P^{\mathsf{max}}_{\widehat{u}}\right]^{(\ell-1)}}{G^{\uparrow}_{u,u',r}([\bm{v}]^{(\ell-1)})}}\nonumber\\
    &~~~~~~~~~~~~~~~~~~~~~~~~~\times\left(G^{\uparrow}_{u,u',r}([\bm{v}]^{(\ell-1)})\right)^{\frac{\left[\overline{B}_{b}N_{0}\right]^{(\ell-1)}}{G^{\uparrow}_{u,u',r}([\bm{v}]^{(\ell-1)})}},
\end{align}
which gives the following approximation of~\eqref{app:eq:DR_D2C_C1_4}:
\begin{equation}\label{app:eq:DR_D2C_C1_8}
     \frac{\left(\overline{\mathfrak{J}}^{\uparrow}_{u,u',r}(t_{x})\right)^{-C}\left(\mathscr{B}^{\uparrow}_{u,u',r}(t_{x})\right)^{-1}\left(\displaystyle\sum_{b' \in \Omega} \sum_{\widehat{u},\widehat{u}'\in \mathcal{U}_{b'}\setminus \{u,u'\}}|\xi_{\widehat{u},u'}(t_{x})|^2 \rho_{\widehat{u},\widehat{u}',r}(t_{x})P^{\mathsf{max}}_{\widehat{u}}{+}\overline{B}_{b}N_{0}{+} \vert\xi_{u,u'}(t_{x})\vert^2 \overline{\rho}^{\uparrow}_{u,r}(t_{x})P^{\mathsf{max}}_{u}\right)}{ \widehat{G}^{\uparrow}_{u,u',r}(\bm{v};\ell)}\le 1.
\end{equation}
We next aim to transform \eqref{app:eq:DR_D2C_C1_1} into standard GP format. In doing so, let us rewrite \eqref{app:eq:DR_D2C_C1_1} as follows:
\begin{equation}\label{app:eq:DR_D2C_C1_9}
    \frac{2(C\overline{B}_{b})^2\overline{\mathfrak{J}}^{\uparrow}_{u,u',r}(t_{x})}{2(C\overline{B}_{b})^2+2(C\overline{B}_{b})^2\ln(2)\overline{\mathfrak{R}}^{\uparrow}_{u,u',r}(t_{x})+\left(\ln(2)\overline{\mathfrak{R}}^{\uparrow}_{u,u',r}(t_{x})\right)^2}=1.
\end{equation}
We transform this constraint into the format of GP by splitting it into the following three inequalities:
\begin{equation}\label{app:eq:DR_D2C_C1_10}
    \frac{2(C\overline{B}_{b})^2\overline{\mathfrak{J}}^{\uparrow}_{u,u',r}(t_{x})}{2(C\overline{B}_{b})^2+2(C\overline{B}_{b})^2\ln(2)\overline{\mathfrak{R}}^{\uparrow}_{u,u',r}(t_{x})+\left(\ln(2)\overline{\mathfrak{R}}^{\uparrow}_{u,u',r}(t_{x})\right)^2}\le 1,
\end{equation}
\begin{equation}\label{app:eq:DR_D2C_C1_11}
    \frac{(\mathscr{C}^{\uparrow}_{u,u',r}(t_{x}))^{-1}2(C\overline{B}_{b})^2+(\mathscr{C}^{\uparrow}_{u,u',r}(t_{x}))^{-1}2(C\overline{B}_{b})^2\ln(2)\overline{\mathfrak{R}}^{\uparrow}_{u,u',r}(t_{x})+(\mathscr{C}^{\uparrow}_{u,u',r}(t_{x}))^{-1}\left(\ln(2)\overline{\mathfrak{R}}^{\uparrow}_{u,u',r}(t_{x})\right)^2}{2(C\overline{B}_{b})^2\overline{\mathfrak{J}}^{\uparrow}_{u,u',r}(t_{x})}\le 1,
\end{equation}
\begin{equation}
    \mathscr{C}^{\uparrow}_{u,u',r}(t_{x})\ge 1,
\end{equation}
where $\mathscr{C}^{\uparrow}_{u,u',r}(t_{x})$ is an auxiliary decision variable and will be added with a large penalty term to the objective function to force $\mathscr{C}^{\uparrow}_{u,u',r}(t_{x}){\rightarrow}1^+$ at the optimal point. The fraction in~\eqref{app:eq:DR_D2C_C1_10} still needs transformation since it is inequality with posynomial in its denominator, which is not posynomial. We thus exploit arithmetic-geometric mean inequality (Lemma~\ref{Lemma:ArethmaticGeometric}) to approximate the denominator of \eqref{app:eq:DR_D2C_C1_3} with monomial. 
\begin{align}\label{app:eq:DR_D2C_C1_12}
    &H^{\uparrow}_{u,u',r}(\bm{v}){=} 2(C\overline{B}_{b})^2{+}2(C\overline{B}_{b})^2\ln(2)\overline{\mathfrak{R}}^{\uparrow}_{u,u',r}(t_{x}){+}\left(\ln(2)\overline{\mathfrak{R}}^{\uparrow}_{u,u',r}(t_{x})\right)^2 {\geq} \widehat{H}^{\uparrow}_{u,u',r}(\bm{v};\ell) {\triangleq} \left(H^{\uparrow}_{u,u',r}([\bm{v}]^{(\ell-1)})\right)^{\hspace{-1mm}\frac{\left[2(C\overline{B}_{b})^2\right]^{(\ell-1)}}{H^{\uparrow}_{u,u',r}([\bm{v}]^{(\ell-1)})}}\nonumber\\
    &{\times}\hspace{-1.5mm}\left(\frac{\overline{\mathfrak{R}}^{\uparrow}_{u,u',r}(t_{x}) H^{\uparrow}_{u,u',r}([\bm{v}]^{(\ell-1)})}{\left[\overline{\mathfrak{R}}^{\uparrow}_{u,u',r}(t_{x})\right]^{(\ell-1)}}\right)^{\hspace{-2mm}\frac{\left[2(C\overline{B}_{b})^2\ln(2)\overline{\mathfrak{R}}^{\uparrow}_{u,u',r}(t_{x})\right]^{(\ell-1)}}{H^{\uparrow}_{u,u',r}([\bm{v}]^{(\ell-1)})}}\times\left(\frac{\left(\overline{\mathfrak{R}}^{\uparrow}_{u,u',r}(t_{x})\right)^2 H^{\uparrow}_{u,u',r}([\bm{v}]^{(\ell-1)})}{\left[\left(\overline{\mathfrak{R}}^{\uparrow}_{u,u',r}(t_{x})\right)^2\right]^{(\ell-1)}}\right)^{\hspace{-2mm}\frac{\left[\left(\ln(2)\overline{\mathfrak{R}}^{\uparrow}_{u,u',r}(t_{x})\right)^2\right]^{(\ell-1)}}{H^{\uparrow}_{u,u',r}([\bm{v}]^{(\ell-1)})}},
\end{align}
which results in the following approximation of~\eqref{app:eq:DR_D2C_C1_10}:
\begin{equation}\label{app:eq:DR_D2C_C1_14}
     \frac{2(C\overline{B}_{b})^2\overline{\mathfrak{J}}^{\uparrow}_{u,u',r}(t_{x})}{ \widehat{H}^{\uparrow}_{u,u',r}(\bm{v};\ell)}\le 1.
\end{equation}
We finally approximate \eqref{app:eq:DR_D2C_C1} as follows:
\begin{tcolorbox}[ams align]
     & \frac{\left(\overline{\mathfrak{J}}^{\uparrow}_{u,u',r}(t_{x})\right)^C\left(\displaystyle\sum_{b' \in \Omega} \sum_{\widehat{u},\widehat{u}'\in \mathcal{U}_{b'}\setminus \{u,u'\}}|\xi_{\widehat{u},u'}(t_{x})|^2 \rho_{\widehat{u},\widehat{u}',r}(t_{x})P^{\mathsf{max}}_{\widehat{u}}{+}\overline{B}_{b}N_{0}\right)}{\widehat{F}^{\uparrow}_{u,u',r}(\bm{v};\ell)}\le 1,\nonumber\\
     & \frac{\left(\overline{\mathfrak{J}}^{\uparrow}_{u,u',r}(t_{x})\right)^{-C}\hspace{-1mm}\left(\mathscr{B}^{\uparrow}_{u,u',r}(t_{x})\right)^{-1}\hspace{-1mm}\left(\displaystyle\sum_{b' \in \Omega} \sum_{\widehat{u},\widehat{u}'\in \mathcal{U}_{b'}\setminus \{u,u'\}}\hspace{-1mm} |\xi_{\widehat{u},u'}(t_{x})|^2\rho_{\widehat{u},\widehat{u}',r}(t_{x})P^{\mathsf{max}}_{\widehat{u}}{+}\overline{B}_{b}N_{0}{+} \vert\xi_{u,u'}(t_{x})\vert^2 \overline{\rho}^{\uparrow}_{u,r}(t_{x})P^{\mathsf{max}}_{u}\right)}{ \widehat{G}^{\uparrow}_{u,u',r}(\bm{v};\ell)}\le 1,\nonumber\\
     &  \frac{(\mathscr{C}^{\uparrow}_{u,u',r}(t_{x}))^{-1}2(C\overline{B}_{b})^2+(\mathscr{C}^{\uparrow}_{u,u',r}(t_{x}))^{-1}2(C\overline{B}_{b})^2\ln(2)\overline{\mathfrak{R}}^{\uparrow}_{u,u',r}(t_{x})+(\mathscr{C}^{\uparrow}_{u,u',r}(t_{x}))^{-1}\left(\ln(2)\overline{\mathfrak{R}}^{\uparrow}_{u,u',r}(t_{x})\right)^2}{2(C\overline{B}_{b})^2\overline{\mathfrak{J}}^{\uparrow}_{u,u',r}(t_{x})}\le 1,\nonumber\\
     &  \frac{2(C\overline{B}_{b})^2\overline{\mathfrak{J}}^{\uparrow}_{u,u',r}(t_{x})}{ \widehat{H}^{\uparrow}_{u,u',r}(\bm{v};\ell)}\le 1,\nonumber\\
     &\frac{1}{\mathscr{B}^{\uparrow}_{u,u',r}(t_{x})}\le 1,~~~~\frac{1}{\mathscr{C}^{\uparrow}_{u,u',r}(t_{x})}\le 1.\nonumber
\end{tcolorbox}

\newpage
\subsection{Computation and communication latencies}\label{app:computation_communication_latencies}
\noindent\textbf{Computation latency.} Referring to Sec.~\ref{sec:computation_latency}, let us revisit the equation of computation latency below:
\begin{equation}\label{app:cons:computation_latency_1}
     \tau_{u}^{\mathsf{LC},{(k)}}= \frac{\ell^{(k)}_{u} a_{u} {B}_{u}(\widetilde{\tau}_{b}^{\downarrow,{(k)}})}{f^{(k)}_{u}},
\end{equation}
In the following we transform \eqref{app:cons:computation_latency_1} into standard GP format. In doing so, we define $\tau_{u}^{\mathsf{LC},{(k)}}$ as an auxiliary decision variable that must satisfy the following constraint (derived through performing some algebraic manipulations on \eqref{app:cons:computation_latency_1}):
\begin{equation}\label{app:cons:computation_latency_2}
    \frac{f^{(k)}_{u}\tau_{u}^{\mathsf{LC},{(k)}}+f^{(k)}_{u}}{\ell^{(k)}_{u} a_{u} {B}_{u}(\widetilde{\tau}_{b}^{\downarrow,{(k)}})+f^{(k)}_{u}}= 1.
\end{equation}
We transform \eqref{app:cons:computation_latency_2} into the format of GP by splitting it into the following three inequalities:
\begin{equation}\label{app:cons:computation_latency_3}
    \frac{f^{(k)}_{u}\tau_{u}^{\mathsf{LC},{(k)}}+f^{(k)}_{u}}{\ell^{(k)}_{u} a_{u} {B}_{u}(\widetilde{\tau}_{b}^{\downarrow,{(k)}})+f^{(k)}_{u}}\le 1,
\end{equation}
\begin{equation}\label{app:cons:computation_latency_4}
    \frac{\left(A^{\mathsf{LC},{(k)}}\right)^{-1}\left(\ell^{(k)}_{u} a_{u} {B}_{u}(\widetilde{\tau}_{b}^{\downarrow,{(k)}})+f^{(k)}_{u}\right)}{f^{(k)}_{u}\tau_{u}^{\mathsf{LC},{(k)}}+f^{(k)}_{u}}\le 1,
\end{equation}
\begin{equation}
    A^{\mathsf{LC},{(k)}}\ge 1
\end{equation}
where $A^{\mathsf{LC},{(k)}}$ is an auxiliary decision variable and will be added with a large penalty term to the objective function to force $A^{\mathsf{LC},{(k)}}{\rightarrow}1^+$ at the optimal point. The fractions in~\eqref{app:cons:computation_latency_3} and \eqref{app:cons:computation_latency_4} still need transformation since they are inequalities with posynomials in their denominators, which are not posynomials. We thus exploit arithmetic-geometric mean inequality (Lemma~\ref{Lemma:ArethmaticGeometric}) to approximate the denominators of \eqref{app:cons:computation_latency_3} and \eqref{app:cons:computation_latency_4} with monomials. To this end, we transform \eqref{app:cons:computation_latency_3} into GP format as follows:
\begin{align}\label{app:cons:computation_latency_5}
    F^{\mathsf{LC},{(k)}}(\bm{v}){=}& \ell^{(k)}_{u} a_{u} {B}_{u}(\widetilde{\tau}_{b}^{\downarrow,{(k)}})+f^{(k)}_{u} {\geq} \widehat{F}^{\mathsf{LC},{(k)}}(\bm{v};\ell) {\triangleq} \left(\frac{\ell^{(k)}_{u} a_{u} {B}_{u}(\widetilde{\tau}_{b}^{\downarrow,{(k)}}) F^{\mathsf{LC},{(k)}}([\bm{v}]^{(\ell-1)})}{\left[\ell^{(k)}_{u} a_{u} {B}_{u}(\widetilde{\tau}_{b}^{\downarrow,{(k)}})\right]^{(\ell-1)}}\right)^{\hspace{-2mm}\frac{\left[\ell^{(k)}_{u} a_{u} {B}_{u}(\widetilde{\tau}_{b}^{\downarrow,{(k)}})\right]^{(\ell-1)}}{F^{\mathsf{LC},{(k)}}([\bm{v}]^{(\ell-1)})}}\nonumber\\
    &~~~~~~~~~~~~~~~~~~~~~~~~~~~~~~~~~~~\times\left(\frac{f^{(k)}_{u} F^{\mathsf{LC},{(k)}}([\bm{v}]^{(\ell-1)})}{\left[f^{(k)}_{u}\right]^{(\ell-1)}}\right)^{\frac{\left[f^{(k)}_{u}\right]^{(\ell-1)}}{F^{\mathsf{LC},{(k)}}([\bm{v}]^{(\ell-1)})}},
\end{align}
which gives the following approximation of~\eqref{app:cons:computation_latency_3}:
\begin{equation}\label{app:cons:computation_latency_6}
     \frac{f^{(k)}_{u}\tau_{u}^{\mathsf{LC},{(k)}}+f^{(k)}_{u}}{ \widehat{F}^{\mathsf{LC},{(k)}}(\bm{v};\ell)}\le 1.
\end{equation}
Similarly, we transform \eqref{app:cons:computation_latency_4} into GP format as follows:
\begin{align}\label{app:cons:computation_latency_5}
    G^{\mathsf{LC},{(k)}}(\bm{v}){=}& f^{(k)}_{u}\tau_{u}^{\mathsf{LC},{(k)}}+f^{(k)}_{u} {\geq} \widehat{G}^{\mathsf{LC},{(k)}}(\bm{v};\ell) {\triangleq} \left(\frac{f^{(k)}_{u}\tau_{u}^{\mathsf{LC},{(k)}} G^{\mathsf{LC},{(k)}}([\bm{v}]^{(\ell-1)})}{\left[f^{(k)}_{u}\tau_{u}^{\mathsf{LC},{(k)}}\right]^{(\ell-1)}}\right)^{\hspace{-2mm}\frac{\left[f^{(k)}_{u}\tau_{u}^{\mathsf{LC},{(k)}}\right]^{(\ell-1)}}{G^{\mathsf{LC},{(k)}}([\bm{v}]^{(\ell-1)})}}\nonumber\\
    &~~~~~~~~~~~~~~~~~~~~~~~~~~~~~~~~~~~\times\left(\frac{f^{(k)}_{u} G^{\mathsf{LC},{(k)}}([\bm{v}]^{(\ell-1)})}{\left[f^{(k)}_{u}\right]^{(\ell-1)}}\right)^{\frac{\left[f^{(k)}_{u}\right]^{(\ell-1)}}{G^{\mathsf{LC},{(k)}}([\bm{v}]^{(\ell-1)})}},
\end{align}
which gives the following approximation of~\eqref{app:cons:computation_latency_3}:
\begin{equation}\label{app:cons:computation_latency_6}
     \frac{\left(A^{\mathsf{LC},{(k)}}\right)^{-1}\left(\ell^{(k)}_{u} a_{u} {B}_{u}(\widetilde{\tau}_{b}^{\downarrow,{(k)}})+f^{(k)}_{u}\right)}{ \widehat{G}^{\mathsf{LC},{(k)}}(\bm{v};\ell)}\le 1.
\end{equation}
We finally approximate \eqref{app:cons:computation_latency_1} as follows:
\begin{tcolorbox}[ams align]
     & \frac{f^{(k)}_{u}\tau_{u}^{\mathsf{LC},{(k)}}+f^{(k)}_{u}}{ \widehat{F}^{\mathsf{LC},{(k)}}(\bm{v};\ell)}\le 1,~~~~\frac{\left(A^{\mathsf{LC},{(k)}}\right)^{-1}\left(\ell^{(k)}_{u} a_{u} {B}_{u}(\widetilde{\tau}_{b}^{\downarrow,{(k)}})+f^{(k)}_{u}\right)}{ \widehat{G}^{\mathsf{LC},{(k)}}(\bm{v};\ell)}\le 1,~~~~\frac{1}{A^{\mathsf{LC},{(k)}}}\le 1.\nonumber
\end{tcolorbox}

\noindent\textbf{GM Broadcasting latency (Proposition~\ref{propo:broadcast}).} Referring to Sec.~\ref{sec:communication_latency}, let us revisit the equations of GM broadcasting latency below:
\begin{equation}\label{app:eq:broadcast_latency_0}
    \sum_{r \in \mathcal{R}_{b}} \varphi_{b,r}\hspace{-0.3mm}(t_{x})= 1
\end{equation}
\begin{equation}\label{app:eq:broadcast_latency_1}
   \tau_{b}^{\downarrow,{(k)}}=\hspace{-3mm}\sum_{x{=}N^{(k{-}1)}{+}1}^{N^{(k)}}\max_{r{\in}\mathcal{R}_{b}}\Big\{\min\big\{\tau_{b,r}^{\downarrow}(t_{x}),\left(t_{x{+}1}-t_{x}\right)\big\}\Big\},
\end{equation}    
where
\begin{equation}\label{app:eq:broadcast_latency_2}
    \tau_{b,r}^{\downarrow}(t_{x}){=}\frac{\big(\alpha_{\bm{\omega}} M{-}M^{\downarrow}_{b}(x)\big)\varphi_{b,r}(t_{x})}{\mathfrak{R}^{\downarrow}_{b,r}(t_{x})+1-\beta^{\downarrow}_{b}\hspace{-0.3mm}(t_{x})}.
\end{equation}
In \eqref{app:eq:broadcast_latency_2}, $M^{\downarrow}_{b}(x)$ is given by
\begin{equation}\label{app:eq:broadcast_latency_3}
   M^{\downarrow}_{b}(x){=}\sum_{z{=}N^{(k{-}1)}{+}1}^{x}\sum_{r'
   {\in}\mathcal{R}_{b}}\min\big\{\tau_{b,r'}^{\downarrow}(t_{z}),\left(t_{z{+}1}{-}t_{z}\right)\big\}\mathfrak{R}^{\downarrow}_{b,r'}(t_{z}).
\end{equation} 

In the following, we first transform \eqref{app:eq:broadcast_latency_0} into the standard form of GP. In doing so, we split \eqref{app:eq:broadcast_latency_0} into the following three inequalities:
\begin{equation}\label{app:eq:broadcast_latency_0_2}
    \sum_{r \in \mathcal{R}_{b}} \varphi_{b,r}\hspace{-0.3mm}(t_{x})\le 1,
\end{equation}
\begin{equation}\label{app:eq:broadcast_latency_0_3}
    \frac{\left(\mathscr{B}_{b}(t_{x})\right)^{-1}}{\sum_{r \in \mathcal{R}_{b}} \varphi_{b,r}\hspace{-0.3mm}(t_{x})}\le 1,
\end{equation}
\begin{equation}
    \mathscr{B}_{b}(t_{x})\ge 1,
\end{equation}
where $\mathscr{B}_{b}(t_{x})$ is an auxiliary decision variable added with a large penalty term to the objective function to force $\mathscr{B}_{b}(t_{x}){\rightarrow}1^+$ at the optimal point. The fraction in~\eqref{app:eq:broadcast_latency_0_3} still needs transformation since it is inequality with a posynomial in its denominator, which is not a posynomial. We thus exploit arithmetic-geometric mean inequality (Lemma~\ref{Lemma:ArethmaticGeometric}) to approximate the denominator in \eqref{app:eq:broadcast_latency_0_3} with a monomial.
\begin{align}\label{app:eq:broadcast_latency_0_3_1}
    A_{b}(\bm{v})&=\sum_{r \in \mathcal{R}_{b}} \varphi_{b,r}\hspace{-0.3mm}(t_{x}) \geq \widehat{A}_{b}(\bm{v};\ell) \triangleq \prod_{r \in \mathcal{R}_{b}} \left(\frac{\varphi_{b,r}\hspace{-0.3mm}(t_{x}) A_{b}([\bm{v}]^{(\ell-1)})}{\left[\varphi_{b,r}\hspace{-0.3mm}(t_{x})\right]^{(\ell-1)}}\right)^{\frac{\left[\varphi_{b,r}\hspace{-0.3mm}(t_{x})\right]^{(\ell-1)}}{A_{b}([\bm{v}]^{(\ell-1)})}},
\end{align}
which gives us an approximation of~\eqref{app:eq:broadcast_latency_0_3} as follows:
\begin{equation}\label{app:eq:broadcast_latency_0_3_3}
    \frac{\left(\mathscr{B}_{b}(t_{x})\right)^{-1}}{\widehat{A}_{b}(\bm{v};\ell)}\le 1,
\end{equation}
We finally approximate constraint~\eqref{app:eq:broadcast_latency_0} as follows:
\begin{tcolorbox}[ams align]
     &\sum_{r \in \mathcal{R}_{b}} \varphi_{b,r}\hspace{-0.3mm}(t_{x})\le 1,~~~~~~~\frac{\left(\mathscr{B}_{b}(t_{x})\right)^{-1}}{\widehat{A}_{b}(\bm{v};\ell)}\le 1,~~~~~~~\frac{1}{\mathscr{B}_{b}(t_{x})}\le 1.\nonumber
\end{tcolorbox}

We next transform \eqref{app:eq:broadcast_latency_1} into the standard form of GP. To this end, we define $\tau_{b}^{\downarrow,{(k)}}$, $\tau_{b,r}^{\downarrow}(t_{x})$, and $M^{\downarrow}_{b}(x)$ as auxiliary decision variables that must satisfy constraints \eqref{app:eq:broadcast_latency_1}, \eqref{app:eq:broadcast_latency_2}, and \eqref{app:eq:broadcast_latency_3}, respectively. Furthermore, we define an auxiliary decision variable $T_{b,r}^{\downarrow}(t_{x})$ satisfying the following constraint:
\begin{equation}\label{app:eq:broadcast_latency_min_1}
    T_{b,r}^{\downarrow}(t_{x})=\min\big\{\tau_{b,r}^{\downarrow}(t_{x}),\left(t_{x{+}1}-t_{x}\right)\big\}
\end{equation}
Accordingly, \eqref{app:eq:broadcast_latency_1} and \eqref{app:eq:broadcast_latency_3} can be rewritten as follows:
\begin{equation}\label{app:eq:broadcast_latency_1_1}
   \tau_{b}^{\downarrow,{(k)}}=\hspace{-3mm}\sum_{x{=}N^{(k{-}1)}{+}1}^{N^{(k)}}\max_{r{\in}\mathcal{R}_{b}}\Big\{T_{b,r}^{\downarrow}(t_{x})\Big\},
\end{equation}    
\begin{equation}\label{app:eq:broadcast_latency_3_1}
   M^{\downarrow}_{b}(x){=}\sum_{z{=}N^{(k{-}1)}{+}1}^{x}\sum_{r'
   {\in}\mathcal{R}_{b}}T_{b,r}^{\downarrow}(t_{z})\mathfrak{R}^{\downarrow}_{b,r'}(t_{z}).
\end{equation} 
In the following, we first aim to transform \eqref{app:eq:broadcast_latency_min_1} into the standard form of GP. To this end, using the approximation $\min\{A, B\}\approx (A^{-p}+B^{-p})^{-\frac{1}{p}}$, which is tight when $p \gg 1$, gives us
\begin{equation}\label{app:eq:broadcast_latency_min_2}
    T_{b,r}^{\downarrow}(t_{x})+\epsilon=\left(\left(\tau_{b,r}^{\downarrow}(t_{x})+\epsilon\right)^{-p}+\left(t_{x{+}1}{-}t_{x}+\epsilon\right)^{-p}\right)^{-\frac{1}{p}}.
\end{equation}
where $\epsilon$ is added to both sides to avoid division by zero. To simplify \eqref{app:eq:broadcast_latency_min_2}, we define an auxiliary decision variable $\widetilde{T}_{b,r}^{\downarrow}(t_{x})$ that satisfy the following equality constraint:
\begin{equation}\label{app:eq:broadcast_latency_min_3}
    \widetilde{T}_{b,r}^{\downarrow}(t_{x})=\left(\tau_{b,r}^{\downarrow}(t_{x})+\epsilon\right)^{-p}+\left(t_{x{+}1}{-}t_{x}+\epsilon\right)^{-p}
\end{equation}
resulting in 
\begin{equation}\label{app:eq:broadcast_latency_min_4}
    T_{b,r}^{\downarrow}(t_{x})+\epsilon=\left(\widetilde{T}_{b,r}^{\downarrow}(t_{x})\right)^{-\frac{1}{p}}.
\end{equation}
Performing some algebraic operations gives us
\begin{equation}\label{app:eq:broadcast_latency_min_5}
    T_{b,r}^{\downarrow}(t_{x})\times \left(\widetilde{T}_{b,r}^{\downarrow}(t_{x})\right)^{\frac{1}{p}}+\epsilon\times \left(\widetilde{T}_{b,r}^{\downarrow}(t_{x})\right)^{\frac{1}{p}} = 1.
\end{equation}
We transform \eqref{app:eq:broadcast_latency_min_5} into the standard GP format via introducing the following three inequalities:
\begin{equation}\label{app:eq:broadcast_latency_min_6}
     T_{b,r}^{\downarrow}(t_{x})\times \left(\widetilde{T}_{b,r}^{\downarrow}(t_{x})\right)^{\frac{1}{p}}+\epsilon\times \left(\widetilde{T}_{b,r}^{\downarrow}(t_{x})\right)^{\frac{1}{p}}\le 1,
\end{equation}
\begin{equation}\label{app:eq:broadcast_latency_min_7}
    \frac{\left(\mathscr{B}^{\downarrow,\mathsf{min}}_{b,r}(t_{x})\right)^{-1}}{T_{b,r}^{\downarrow}(t_{x})\times \left(\widetilde{T}_{b,r}^{\downarrow}(t_{x})\right)^{\frac{1}{p}}+\epsilon\times \left(\widetilde{T}_{b,r}^{\downarrow}(t_{x})\right)^{\frac{1}{p}}}\le 1,
\end{equation}
\begin{equation}
    \mathscr{B}^{\downarrow,\mathsf{min}}_{b,r}(t_{x})\ge 1,
\end{equation}
where $\mathscr{B}^{\downarrow,\mathsf{min}}_{b,r}(t_{x})$ is an auxiliary decision variable added with a large penalty term to the objective function to force $\mathscr{B}^{\downarrow,\mathsf{min}}_{b,r}(t_{x}){\rightarrow}1^+$ at the optimal point. The fraction in~\eqref{app:eq:broadcast_latency_min_7} still needs transformation since it is an inequality with a posynomial in its denominator, which is not posynomial. We thus exploit arithmetic-geometric mean inequality (Lemma~\ref{Lemma:ArethmaticGeometric}) to approximate the denominator of \eqref{app:eq:broadcast_latency_min_7} with a monomial:
\begin{align}\label{app:eq:broadcast_latency_min_8}
    W^{\downarrow,\mathsf{min}}_{b,r}(\bm{v}){=}T_{b,r}^{\downarrow}(t_{x}) \left(\widetilde{T}_{b,r}^{\downarrow}(t_{x})\right)^{\frac{1}{p}}{+}\epsilon{\times} \left(\widetilde{T}_{b,r}^{\downarrow}(t_{x})\right)^{\frac{1}{p}}{\geq} \widehat{W}^{\downarrow,\mathsf{min}}_{b,r}(\bm{v};\ell) &{\triangleq} \left(\frac{T_{b,r}^{\downarrow}(t_{x}) \left(\widetilde{T}_{b,r}^{\downarrow}(t_{x})\right)^{\frac{1}{p}} W^{\downarrow,\mathsf{min}}_{b,r}([\bm{v}]^{(\ell-1)})}{\left[T_{b,r}^{\downarrow}(t_{x}) \left(\widetilde{T}_{b,r}^{\downarrow}(t_{x})\right)^{\frac{1}{p}}\right]^{(\ell-1)}}\right)^{\hspace{-2mm}\frac{\left[T_{b,r}^{\downarrow}(t_{x}) \left(\widetilde{T}_{b,r}^{\downarrow}(t_{x})\right)^{\frac{1}{p}}\right]^{(\ell-1)}}{W^{\downarrow,\mathsf{min}}_{b,r}([\bm{v}]^{(\ell-1)})}}\nonumber\\
    &\times\left(\frac{ \left(\widetilde{T}_{b,r}^{\downarrow}(t_{x})\right)^{\frac{1}{p}} W^{\downarrow,\mathsf{min}}_{b,r}([\bm{v}]^{(\ell-1)})}{\left[ \left(\widetilde{T}_{b,r}^{\downarrow}(t_{x})\right)^{\frac{1}{p}}\right]^{(\ell-1)}}\right)^{\frac{\left[\epsilon\times \left(\widetilde{T}_{b,r}^{\downarrow}(t_{x})\right)^{\frac{1}{p}}\right]^{(\ell-1)}}{W^{\downarrow,\mathsf{min}}_{b,r}([\bm{v}]^{(\ell-1)})}},
\end{align}
which gives us the following approximation of~\eqref{app:eq:broadcast_latency_min_7}:
\begin{equation}
     \frac{\left(\mathscr{B}^{\downarrow,\mathsf{min}}_{b,r}(t_{x})\right)^{-1}}{\widehat{W}^{\downarrow,\mathsf{min}}_{b,r}(\bm{v};\ell)}\le 1,
\end{equation}

We next aim to transform \eqref{app:eq:broadcast_latency_min_3} into standard GP format. In doing so, we rewrite \eqref{app:eq:broadcast_latency_min_3} as follows:
\begin{equation}\label{app:eq:broadcast_latency_min_9}
    \frac{\widetilde{T}_{b,r}^{\downarrow}(t_{x})}{\left(\tau_{b,r}^{\downarrow}(t_{x})+\epsilon\right)^{-p}+\left(t_{x{+}1}{-}t_{x}+\epsilon\right)^{-p}}= 1.
\end{equation}
Defining two auxiliary decision variable $\widetilde{t}_{x}$ and $\widetilde{\tau}_{b,r}^{\downarrow}(t_{x})$ that satisfy 
\begin{equation}\label{app:eq:broadcast_latency_min_10}
    \widetilde{t}_{x}=t_{x{+}1}{-}t_{x}+\epsilon
\end{equation}
and
\begin{equation}\label{app:eq:broadcast_latency_min_101}
    \widetilde{\tau}_{b,r}^{\downarrow}(t_{x})=\tau_{b,r}^{\downarrow}(t_{x})+\epsilon
\end{equation}
gives us
\begin{equation}\label{app:eq:broadcast_latency_min_11}
    \frac{\widetilde{T}_{b,r}^{\downarrow}(t_{x})}{\left(\widetilde{\tau}_{b,r}^{\downarrow}(t_{x})\right)^{-p}+\left(\widetilde{t}_{x}\right)^{-p}}= 1.
\end{equation}
This constraint is not in the format of GP. Therefore, we transform it through splitting it into the following three inequalities:
\begin{equation}\label{app:eq:broadcast_latency_min_12}
    \frac{\widetilde{T}_{b,r}^{\downarrow}(t_{x})}{\left(\widetilde{\tau}_{b,r}^{\downarrow}(t_{x})\right)^{-p}+\left(\widetilde{t}_{x}\right)^{-p}}\le 1,
\end{equation}
\begin{equation}\label{app:eq:broadcast_latency_min_13}
    \frac{\left(\mathscr{C}^{\downarrow,\mathsf{min}}_{b,r}(t_{x})\right)^{-1}\left(\left(\widetilde{\tau}_{b,r}^{\downarrow}(t_{x})\right)^{-p}+\left(\widetilde{t}_{x}\right)^{-p}\right)}{\widetilde{T}_{b,r}^{\downarrow}(t_{x})}\le 1,
\end{equation}
\begin{equation}
    \mathscr{C}^{\downarrow,\mathsf{min}}_{b,r}(t_{x})\ge 1,
\end{equation}
where $\mathscr{C}^{\downarrow,\mathsf{min}}_{b,r}(t_{x})$ is an auxiliary decision variable added with a large penalty term to the objective function to force $\mathscr{C}^{\downarrow,\mathsf{min}}_{b,r}(t_{x}){\rightarrow}1^+$ at the optimal point. The fraction in~\eqref{app:eq:broadcast_latency_min_12} still needs transformation since it is an inequality with a posynomial in the denominator, which is not a posynomial. We thus exploit arithmetic-geometric mean inequality (Lemma~\ref{Lemma:ArethmaticGeometric}) to approximate the denominator with a monomial:
\begin{align}\label{app:eq:broadcast_latency_min_14}
    G^{\downarrow,\mathsf{min}}_{b,r}(\bm{v})=\left(\widetilde{\tau}_{b,r}^{\downarrow}(t_{x})\right)^{-p}+\left(\widetilde{t}_{x}\right)^{-p}\geq \widehat{G}^{\downarrow,\mathsf{min}}_{b,r}(\bm{v};\ell) &\triangleq \left(\frac{\left(\widetilde{\tau}_{b,r}^{\downarrow}(t_{x})\right)^{-p} G^{\downarrow,\mathsf{min}}_{b,r}([\bm{v}]^{(\ell-1)})}{\left[\left(\widetilde{\tau}_{b,r}^{\downarrow}(t_{x})\right)^{-p}\right]^{(\ell-1)}}\right)^{\frac{\left[\left(\widetilde{\tau}_{b,r}^{\downarrow}(t_{x})\right)^{-p}\right]^{(\ell-1)}}{G^{\downarrow,\mathsf{min}}_{b,r}([\bm{v}]^{(\ell-1)})}}\nonumber\\
    &\times\left(\frac{\left(\widetilde{t}_{x}\right)^{-p} G^{\downarrow,\mathsf{min}}_{b,r}([\bm{v}]^{(\ell-1)})}{\left[\left(\widetilde{t}_{x}\right)^{-p}\right]^{(\ell-1)}}\right)^{\frac{\left[\left(\widetilde{t}_{x}\right)^{-p}\right]^{(\ell-1)}}{G^{\downarrow,\mathsf{min}}_{b,r}([\bm{v}]^{(\ell-1)})}},
\end{align}
which gives us an approximation of~\eqref{app:eq:broadcast_latency_min_12} as follows:
\begin{equation}\label{app:eq:broadcast_latency_min_15}
    \frac{\widetilde{T}_{b,r}^{\downarrow}(t_{x})}{\widehat{G}^{\downarrow,\mathsf{min}}_{b,r}(\bm{v};\ell)}\le 1,
\end{equation}

Using the same technique used above, we next transform \eqref{app:eq:broadcast_latency_min_10} into GP format. To this end, we rewrite \eqref{app:eq:broadcast_latency_min_10} as follows:
\begin{equation}\label{app:eq:broadcast_latency_min_16}
    \frac{\widetilde{t}_{x}+t_{x}}{t_{x{+}1}+\epsilon}=1.
\end{equation}
This constraint is not in the format of GP. Therefore, we transform it through splitting it into the following three inequalities:
\begin{equation}\label{app:eq:broadcast_latency_min_17}
    \frac{\widetilde{t}_{x}+t_{x}}{t_{x{+}1}+\epsilon}\le 1,
\end{equation}
\begin{equation}\label{app:eq:broadcast_latency_min_18}
    \frac{\left(\mathscr{E}^{\downarrow,\mathsf{min}}_{b,r}(t_{x})\right)^{-1}\left(t_{x{+}1}+\epsilon\right)}{\widetilde{t}_{x}+t_{x}}\le 1,
\end{equation}
\begin{equation}
    \mathscr{E}^{\downarrow,\mathsf{min}}_{b,r}(t_{x})\ge 1,
\end{equation}
where $\mathscr{E}^{\downarrow,\mathsf{min}}_{b,r}(t_{x})$ is an auxiliary decision variable added with a large penalty term to the objective function to force $\mathscr{E}^{\downarrow,\mathsf{min}}_{b,r}(t_{x}){\rightarrow}1^+$ at the optimal point. The fractions in~\eqref{app:eq:broadcast_latency_min_17} and \eqref{app:eq:broadcast_latency_min_18} still need transformation since they are inequalities with posynomials in their denominators, which are not posynomials. We thus exploit arithmetic-geometric mean inequality (Lemma~\ref{Lemma:ArethmaticGeometric}) to approximate the denominators of \eqref{app:eq:broadcast_latency_min_17} and \eqref{app:eq:broadcast_latency_min_18} with monomials. In doing so, we transform \eqref{app:eq:broadcast_latency_min_17} into GP format as follows:
\begin{align}\label{app:eq:broadcast_latency_min_192}
    H^{\downarrow,\mathsf{min}}_{b,r}(\bm{v})=t_{x{+}1}+\epsilon \geq \widehat{H}^{\downarrow,\mathsf{min}}_{b,r}(\bm{v};\ell) &\triangleq \left(\frac{t_{x{+}1} H^{\downarrow,\mathsf{min}}_{b,r}([\bm{v}]^{(\ell-1)})}{\left[t_{x{+}1}\right]^{(\ell-1)}}\right)^{\frac{\left[t_{x{+}1}\right]^{(\ell-1)}}{H^{\downarrow,\mathsf{min}}_{b,r}([\bm{v}]^{(\ell-1)})}}\times\left(H^{\downarrow,\mathsf{min}}_{b,r}([\bm{v}]^{(\ell-1)})\right)^{\frac{\displaystyle\epsilon}{H^{\downarrow,\mathsf{min}}_{b,r}([\bm{v}]^{(\ell-1)})}},
\end{align}
which gives us an approximation of~\eqref{app:eq:broadcast_latency_min_17} as follows:
\begin{equation}\label{app:eq:broadcast_latency_min_200}
    \frac{\widetilde{t}_{x}+t_{x}}{\widehat{H}^{\downarrow,\mathsf{min}}_{b,r}(\bm{v};\ell)}\le 1.
\end{equation}
Similarly, we transform \eqref{app:eq:broadcast_latency_min_18} into GP format as follows:
\begin{align}\label{app:eq:broadcast_latency_min_19}
    H^{\downarrow}_{b,r}(\bm{v})=\widetilde{t}_{x}+t_{x} \geq \widehat{H}^{\downarrow}_{b,r}(\bm{v};\ell) &\triangleq \left(\frac{\widetilde{t}_{x} H^{\downarrow}_{b,r}([\bm{v}]^{(\ell-1)})}{\left[\widetilde{t}_{x}\right]^{(\ell-1)}}\right)^{\frac{\left[\widetilde{t}_{x}\right]^{(\ell-1)}}{H^{\downarrow}_{b,r}([\bm{v}]^{(\ell-1)})}}\times\left(\frac{t_{x} H^{\downarrow}_{b,r}([\bm{v}]^{(\ell-1)})}{\left[t_{x}\right]^{(\ell-1)}}\right)^{\frac{\left[t_{x}\right]^{(\ell-1)}}{H^{\downarrow}_{b,r}([\bm{v}]^{(\ell-1)})}},
\end{align}
which gives us an approximation of~\eqref{app:eq:broadcast_latency_min_18} as follows:
\begin{equation}\label{app:eq:broadcast_latency_min_20}
    \frac{\left(\mathscr{E}^{\downarrow,\mathsf{min}}_{b,r}(t_{x})\right)^{-1}\left(t_{x{+}1}+\epsilon\right)}{\widehat{H}^{\downarrow}_{b,r}(\bm{v};\ell)}\le 1.
\end{equation}

We next transform \eqref{app:eq:broadcast_latency_min_101} into GP format. To this end, we rewrite \eqref{app:eq:broadcast_latency_min_101} as follows:
\begin{equation}\label{app:eq:broadcast_latency_min_102}
    \frac{\widetilde{\tau}_{b,r}^{\downarrow}(t_{x})}{\tau_{b,r}^{\downarrow}(t_{x})+\epsilon}=1.
\end{equation}
This constraint is not in the format of GP. Therefore, we transform it through splitting it into the following three inequalities:
\begin{equation}\label{app:eq:broadcast_latency_min_103}
    \frac{\widetilde{\tau}_{b,r}^{\downarrow}(t_{x})}{\tau_{b,r}^{\downarrow}(t_{x})+\epsilon}\le 1,
\end{equation}
\begin{equation}\label{app:eq:broadcast_latency_min_104}
    \frac{\left(\mathscr{F}^{\downarrow,\mathsf{min}}_{b,r}(t_{x})\right)^{-1}\left(\tau_{b,r}^{\downarrow}(t_{x})+\epsilon\right)}{\widetilde{\tau}_{b,r}^{\downarrow}(t_{x})}\le 1,
\end{equation}
\begin{equation}
    \mathscr{F}^{\downarrow,\mathsf{min}}_{b,r}(t_{x})\ge 1,
\end{equation}
where $\mathscr{F}^{\downarrow,\mathsf{min}}_{b,r}(t_{x})$ is an auxiliary decision variable added with a large penalty term to the objective function to force $\mathscr{F}^{\downarrow,\mathsf{min}}_{b,r}(t_{x}){\rightarrow}1^+$ at the optimal point. The fraction in~\eqref{app:eq:broadcast_latency_min_103} still needs transformation since it is inequality with posynomial in its denominator, which is not a posynomial. We thus exploit arithmetic-geometric mean inequality (Lemma~\ref{Lemma:ArethmaticGeometric}) to approximate the denominator of \eqref{app:eq:broadcast_latency_min_103} with a monomial.
\begin{align}\label{app:eq:broadcast_latency_min_105}
    J^{\downarrow,\mathsf{min}}_{b,r}(\bm{v})=\tau_{b,r}^{\downarrow}(t_{x})+\epsilon \geq \widehat{J}^{\downarrow,\mathsf{min}}_{b,r}(\bm{v};\ell) &\triangleq \left(\frac{\tau_{b,r}^{\downarrow}(t_{x}) J^{\downarrow,\mathsf{min}}_{b,r}([\bm{v}]^{(\ell-1)})}{\left[\tau_{b,r}^{\downarrow}(t_{x})\right]^{(\ell-1)}}\right)^{\frac{\left[\tau_{b,r}^{\downarrow}(t_{x})\right]^{(\ell-1)}}{J^{\downarrow,\mathsf{min}}_{b,r}([\bm{v}]^{(\ell-1)})}}\times\left(J^{\downarrow,\mathsf{min}}_{b,r}([\bm{v}]^{(\ell-1)})\right)^{\frac{\displaystyle\epsilon}{J^{\downarrow,\mathsf{min}}_{b,r}([\bm{v}]^{(\ell-1)})}},
\end{align}
which gives us an approximation of~\eqref{app:eq:broadcast_latency_min_101} as follows:
\begin{equation}\label{app:eq:broadcast_latency_min_106}
    \frac{\tau_{b,r}^{\downarrow}(t_{x})+\epsilon}{\widehat{J}^{\downarrow,\mathsf{min}}_{b,r}(\bm{v};\ell)}\le 1.
\end{equation}

We finally approximate constraint~\eqref{app:eq:broadcast_latency_min_1} as follows:
\begin{tcolorbox}[ams align]
     &T_{b,r}^{\downarrow}(t_{x})\times \left(\widetilde{T}_{b,r}^{\downarrow}(t_{x})\right)^{\frac{1}{p}}+\epsilon\times \left(\widetilde{T}_{b,r}^{\downarrow}(t_{x})\right)^{\frac{1}{p}}\le 1,~~\frac{\left(\mathscr{B}^{\downarrow,\mathsf{min}}_{b,r}(t_{x})\right)^{-1}}{\widehat{W}^{\downarrow,\mathsf{min}}_{b,r}(\bm{v};\ell)}\le 1,\nonumber\\
     &\frac{\left(\mathscr{C}^{\downarrow,\mathsf{min}}_{b,r}(t_{x})\right)^{-1}\left(\left(\widetilde{\tau}_{b,r}^{\downarrow}(t_{x})\right)^{-p}+\left(\widetilde{t}_{x}\right)^{-p}\right)}{\widetilde{T}_{b,r}^{\downarrow}(t_{x})}\le 1,~~\frac{\widetilde{T}_{b,r}^{\downarrow}(t_{x})}{\widehat{G}^{\downarrow,\mathsf{min}}_{b,r}(\bm{v};\ell)}\le 1,\nonumber\\
     &\frac{\widetilde{t}_{x}+t_{x}}{\widehat{H}^{\downarrow,\mathsf{min}}_{b,r}(\bm{v};\ell)}\le 1,~~\frac{\left(\mathscr{E}^{\downarrow,\mathsf{min}}_{b,r}(t_{x})\right)^{-1}\left(t_{x{+}1}+\epsilon\right)}{\widehat{H}^{\downarrow}_{b,r}(\bm{v};\ell)}\le 1,\nonumber\\
     &\frac{\tau_{b,r}^{\downarrow}(t_{x})+\epsilon}{\widehat{J}^{\downarrow,\mathsf{min}}_{b,r}(\bm{v};\ell)}\le 1,~~~~\frac{\left(\mathscr{F}^{\downarrow,\mathsf{min}}_{b,r}(t_{x})\right)^{-1}\left(\tau_{b,r}^{\downarrow}(t_{x})+\epsilon\right)}{\widetilde{\tau}_{b,r}^{\downarrow}(t_{x})}\le 1,\nonumber\\
     &\frac{1}{\mathscr{B}^{\downarrow,\mathsf{min}}_{b,r}(t_{x})}\le 1,~~\frac{1}{\mathscr{C}^{\downarrow,\mathsf{min}}_{b,r}(t_{x})}\le 1,~~\frac{1}{\mathscr{E}^{\downarrow,\mathsf{min}}_{b,r}(t_{x})}\le 1,~~~\frac{1}{\mathscr{F}^{\downarrow,\mathsf{min}}_{b,r}(t_{x})}\le 1.\nonumber
\end{tcolorbox}

We next aim to transform \eqref{app:eq:broadcast_latency_2} into the standard GP format. In doing so, performing some algebraic manipulations on \eqref{app:eq:broadcast_latency_2} gives us
\begin{equation}\label{app:eq:broadcast_latency_4}
    \frac{\mathfrak{R}^{\downarrow}_{b,r}(t_{x})\tau_{b,r}^{\downarrow}(t_{x})+\tau_{b,r}^{\downarrow}(t_{x})+\varphi_{b,r}(t_{x})M^{\downarrow}_{b}(x)+1}{\varphi_{b,r}(t_{x})\alpha_{\bm{\omega}} M +\beta^{\downarrow}_{b}\hspace{-0.3mm}(t_{x})\tau_{b,r}^{\downarrow}(t_{x})+1}{=} 1.
\end{equation}
However, this constraint is not in the format of GP. Therefore, we transform it into the standard GP format via introducing the following three inequalities:
\begin{equation}\label{app:eq:broadcast_latency_5}
     \frac{\mathfrak{R}^{\downarrow}_{b,r}(t_{x})\tau_{b,r}^{\downarrow}(t_{x})+\tau_{b,r}^{\downarrow}(t_{x})+\varphi_{b,r}(t_{x})M^{\downarrow}_{b}(x)+1}{\varphi_{b,r}(t_{x})\alpha_{\bm{\omega}} M +\beta^{\downarrow}_{b}\hspace{-0.3mm}(t_{x})\tau_{b,r}^{\downarrow}(t_{x})+1}\le 1,
\end{equation}
\begin{equation}\label{app:eq:broadcast_latency_6}
    \frac{\left(\mathscr{H}^{\downarrow,\mathsf{L}}_{b,r}(t_{x})\right)^{-1}\left(\varphi_{b,r}(t_{x})\alpha_{\bm{\omega}} M +\beta^{\downarrow}_{b}\hspace{-0.3mm}(t_{x})\tau_{b,r}^{\downarrow}(t_{x})+1\right)}{\mathfrak{R}^{\downarrow}_{b,r}(t_{x})\tau_{b,r}^{\downarrow}(t_{x})+\tau_{b,r}^{\downarrow}(t_{x})+\varphi_{b,r}(t_{x})M^{\downarrow}_{b}(x)+1}\le 1,
\end{equation}
\begin{equation}
    \mathscr{H}^{\downarrow,\mathsf{L}}_{b,r}(t_{x})\ge 1,
\end{equation}
where $\mathscr{H}^{\downarrow,\mathsf{L}}_{b,r}(t_{x})$ is an auxiliary decision variable added with a large penalty term to the objective function to force $\mathscr{H}^{\downarrow,\mathsf{L}}_{b,r}(t_{x}){\rightarrow}1^+$ at the optimal point. The fractions in~\eqref{app:eq:broadcast_latency_5} and \eqref{app:eq:broadcast_latency_6} still need transformation since they are inequalities with posynomials in their denominator, which are not posynomials. We thus exploit arithmetic-geometric mean inequality (Lemma~\ref{Lemma:ArethmaticGeometric}) to approximate the denominators of \eqref{app:eq:broadcast_latency_5} and \eqref{app:eq:broadcast_latency_6} with monomials. In doing so, we approximate the denominator in \eqref{app:eq:broadcast_latency_5} as follows:
\begin{align}\label{app:eq:broadcast_latency_7}
    L^{\downarrow,\mathsf{L}}_{b,r}(\bm{v}){=}&\varphi_{b,r}(t_{x})\alpha_{\bm{\omega}} M +\beta^{\downarrow}_{b}\hspace{-0.3mm}(t_{x})\tau_{b,r}^{\downarrow}(t_{x})+1{\geq} \widehat{L}^{\downarrow,\mathsf{L}}_{b,r}(\bm{v};\ell) {\triangleq} \left(\frac{\varphi_{b,r}(t_{x})\ L^{\downarrow,\mathsf{L}}_{b,r}([\bm{v}]^{(\ell-1)})}{\left[\varphi_{b,r}(t_{x})\right]^{(\ell-1)}}\right)^{\hspace{-2mm}\frac{\left[\varphi_{b,r}(t_{x})\alpha_{\bm{\omega}} M\right]^{(\ell-1)}}{L^{\downarrow,\mathsf{L}}_{b,r}([\bm{v}]^{(\ell-1)})}}\nonumber\\
    &~~~~~~~~~~~~~~~~~~~\times\left(\frac{\beta^{\downarrow}_{b}\hspace{-0.3mm}(t_{x})\tau_{b,r}^{\downarrow}(t_{x}) L^{\downarrow,\mathsf{L}}_{b,r}([\bm{v}]^{(\ell-1)})}{\left[\beta^{\downarrow}_{b}\hspace{-0.3mm}(t_{x})\tau_{b,r}^{\downarrow}(t_{x})\right]^{(\ell-1)}}\right)^{\frac{\left[\beta^{\downarrow}_{b}\hspace{-0.3mm}(t_{x})\tau_{b,r}^{\downarrow}(t_{x})\right]^{(\ell-1)}}{L^{\downarrow,\mathsf{L}}_{b,r}([\bm{v}]^{(\ell-1)})}}\times\left(L^{\downarrow,\mathsf{L}}_{b,r}([\bm{v}]^{(\ell-1)})\right)^{\frac{1}{L^{\downarrow,\mathsf{L}}_{b,r}([\bm{v}]^{(\ell-1)})}},
\end{align}
which gives us the following approximation of~\eqref{app:eq:broadcast_latency_5}:
\begin{equation}
     \frac{\mathfrak{R}^{\downarrow}_{b,r}(t_{x})\tau_{b,r}^{\downarrow}(t_{x})+\tau_{b,r}^{\downarrow}(t_{x})+\varphi_{b,r}(t_{x})M^{\downarrow}_{b}(x)+1}{\widehat{L}^{\downarrow,\mathsf{L}}_{b,r}(\bm{v};\ell)}\le 1,
\end{equation}
Similarly, we approximate the denominator in \eqref{app:eq:broadcast_latency_6} as follows:
\begin{align}\label{app:eq:broadcast_latency_8}
    &R^{\downarrow,\mathsf{L}}_{b,r}(\bm{v}){=}\mathfrak{R}^{\downarrow}_{b,r}(t_{x})\tau_{b,r}^{\downarrow}(t_{x})+\tau_{b,r}^{\downarrow}(t_{x})+\varphi_{b,r}(t_{x})M^{\downarrow}_{b}(x)+1{\geq} \widehat{R}^{\downarrow,\mathsf{L}}_{b,r}(\bm{v};\ell) {\triangleq} \left(\frac{\mathfrak{R}^{\downarrow}_{b,r}(t_{x})\tau_{b,r}^{\downarrow}(t_{x}) R^{\downarrow,\mathsf{L}}_{b,r}([\bm{v}]^{(\ell-1)})}{\left[\mathfrak{R}^{\downarrow}_{b,r}(t_{x})\tau_{b,r}^{\downarrow}(t_{x})\right]^{(\ell-1)}}\right)^{\hspace{-2mm}\frac{\left[\mathfrak{R}^{\downarrow}_{b,r}(t_{x})\tau_{b,r}^{\downarrow}(t_{x})\right]^{(\ell-1)}}{R^{\downarrow,\mathsf{L}}_{b,r}([\bm{v}]^{(\ell-1)})}}\nonumber\\
    &\times\left(\frac{\tau_{b,r}^{\downarrow}(t_{x}) R^{\downarrow,\mathsf{L}}_{b,r}([\bm{v}]^{(\ell-1)})}{\left[\tau_{b,r}^{\downarrow}(t_{x})\right]^{(\ell-1)}}\right)^{\frac{\left[\tau_{b,r}^{\downarrow}(t_{x})\right]^{(\ell-1)}}{R^{\downarrow,\mathsf{L}}_{b,r}([\bm{v}]^{(\ell-1)})}}\left(\frac{\varphi_{b,r}(t_{x})M^{\downarrow}_{b}(x) R^{\downarrow,\mathsf{L}}_{b,r}([\bm{v}]^{(\ell-1)})}{\left[\varphi_{b,r}(t_{x})M^{\downarrow}_{b}(x)\right]^{(\ell-1)}}\right)^{\frac{\left[\varphi_{b,r}(t_{x})M^{\downarrow}_{b}(x)\right]^{(\ell-1)}}{R^{\downarrow,\mathsf{L}}_{b,r}([\bm{v}]^{(\ell-1)})}}\left(R^{\downarrow,\mathsf{L}}_{b,r}([\bm{v}]^{(\ell-1)})\right)^{\frac{1}{R^{\downarrow,\mathsf{L}}_{b,r}([\bm{v}]^{(\ell-1)})}},
\end{align}
which gives the following approximation of~\eqref{app:eq:broadcast_latency_6}:
\begin{equation}
     \frac{\left(\mathscr{H}^{\downarrow,\mathsf{L}}_{b,r}(t_{x})\right)^{-1}\left(\varphi_{b,r}(t_{x})\alpha_{\bm{\omega}} M +\beta^{\downarrow}_{b}\hspace{-0.3mm}(t_{x})\tau_{b,r}^{\downarrow}(t_{x})+1\right)}{ \widehat{R}^{\downarrow,\mathsf{L}}_{b,r}(\bm{v};\ell)}\le 1.
\end{equation}
We finally approximate \eqref{app:eq:broadcast_latency_2} as follows:
\begin{tcolorbox}[ams align]
     & \frac{\mathfrak{R}^{\downarrow}_{b,r}(t_{x})\tau_{b,r}^{\downarrow}(t_{x})+\tau_{b,r}^{\downarrow}(t_{x})+\varphi_{b,r}(t_{x})M^{\downarrow}_{b}(x)+1}{\widehat{L}^{\downarrow,\mathsf{L}}_{b,r}(\bm{v};\ell)}\le 1,\nonumber\\
     & \frac{\left(\mathscr{H}^{\downarrow,\mathsf{L}}_{b,r}(t_{x})\right)^{-1}\left(\varphi_{b,r}(t_{x})\alpha_{\bm{\omega}} M +\beta^{\downarrow}_{b}\hspace{-0.3mm}(t_{x})\tau_{b,r}^{\downarrow}(t_{x})+1\right)}{ \widehat{R}^{\downarrow,\mathsf{L}}_{b,r}(\bm{v};\ell)}\le 1,\nonumber\\
     &\frac{1}{\mathscr{H}^{\downarrow,\mathsf{L}}_{b,r}(t_{x})}\le 1.\nonumber
\end{tcolorbox}

We next aim to transform \eqref{app:eq:broadcast_latency_3_1} into GP format. To this end, we rewrite \eqref{app:eq:broadcast_latency_3_1} as follows:

\begin{equation}\label{app:eq:broadcast_latency_3_1_1}
   \frac{M^{\downarrow}_{b}(x)+1}{\displaystyle\sum_{z{=}0}^{x}\sum_{r'
   {\in}\mathcal{R}_{b}}T_{b,r'}^{\downarrow}(t_{z})\mathfrak{R}^{\downarrow}_{b,r'}(t_{z})+1}{=} 1.
\end{equation} 
This constraint is not in the format of GP. Therefore, we transform it through splitting it into the following three inequalities:
\begin{equation}\label{app:eq:broadcast_latency_3_1_4}
    \frac{M^{\downarrow}_{b}(x)+1}{\displaystyle\sum_{z{=}0}^{x}\sum_{r'
   {\in}\mathcal{R}_{b}}T_{b,r'}^{\downarrow}(t_{z})\mathfrak{R}^{\downarrow}_{b,r'}(t_{z})+1}\le 1,
\end{equation}
\begin{equation}\label{app:eq:broadcast_latency_3_1_5}
    \frac{\left(\mathscr{L}^{\downarrow,\mathsf{M}}_{b}(t_{x})\right)^{-1}\left(\displaystyle\sum_{z{=}0}^{x}\sum_{r'
   {\in}\mathcal{R}_{b}}T_{b,r'}^{\downarrow}(t_{z})\mathfrak{R}^{\downarrow}_{b,r'}(t_{z})+1\right)}{M^{\downarrow}_{b}(x)+1}\le 1,
\end{equation}
\begin{equation}
    \mathscr{L}^{\downarrow,\mathsf{M}}_{b}(t_{x})\ge 1,
\end{equation}
where $\mathscr{L}^{\downarrow,\mathsf{M}}_{b}(t_{x})$ is an auxiliary decision variable added with a large penalty term to the objective function to force $\mathscr{L}^{\downarrow,\mathsf{M}}_{b}(t_{x}){\rightarrow}1^+$ at the optimal point. The fractions in~\eqref{app:eq:broadcast_latency_3_1_4} and \eqref{app:eq:broadcast_latency_3_1_5} still need transformation since they are inequalities with posynomials in their denominators, which are not posynomials. We thus exploit arithmetic-geometric mean inequality (Lemma~\ref{Lemma:ArethmaticGeometric}) to approximate the denominators in \eqref{app:eq:broadcast_latency_3_1_4} and \eqref{app:eq:broadcast_latency_3_1_5} with monomials. To this end, we approximate the denominator in \eqref{app:eq:broadcast_latency_3_1_4} as follows:
\begin{align}\label{app:eq:broadcast_latency_3_1_6}
    S^{\downarrow,\mathsf{M}}_{b}(\bm{v})&=\sum_{z{=}0}^{x}\sum_{r'
   {\in}\mathcal{R}_{b}}T_{b,r'}^{\downarrow}(t_{z})\mathfrak{R}^{\downarrow}_{b,r'}(t_{z})+1\nonumber\\
   &\geq \widehat{S}^{\downarrow,\mathsf{M}}_{b}(\bm{v};\ell) \triangleq \prod_{z{=}0}^{x}\prod_{r'
   {\in}\mathcal{R}_{b}}\left(\frac{T_{b,r'}^{\downarrow}(t_{z})\mathfrak{R}^{\downarrow}_{b,r'}(t_{z}) S^{\downarrow,\mathsf{M}}_{b}([\bm{v}]^{(\ell-1)})}{\left[T_{b,r'}^{\downarrow}(t_{z})\mathfrak{R}^{\downarrow}_{b,r'}(t_{z})\right]^{(\ell-1)}}\right)^{\frac{\left[T_{b,r'}^{\downarrow}(t_{z})\mathfrak{R}^{\downarrow}_{b,r'}(t_{z})\right]^{(\ell-1)}}{S^{\downarrow,\mathsf{M}}_{b}([\bm{v}]^{(\ell-1)})}}\nonumber\\
   &~~~~~~~~~~~\times\left(S^{\downarrow,\mathsf{M}}_{b}([\bm{v}]^{(\ell-1)})\right)^{\frac{1}{S^{\downarrow,\mathsf{M}}_{b}([\bm{v}]^{(\ell-1)})}},
\end{align}
which gives us an approximation of~\eqref{app:eq:broadcast_latency_3_1_4} as follows:
\begin{equation}\label{app:eq:broadcast_latency_3_1_7}
    \frac{M^{\downarrow}_{b}(x)+1}{\widehat{S}^{\downarrow,\mathsf{M}}_{b}(\bm{v};\ell)}\le 1.
\end{equation}

Similarly, we approximate the denominator in \eqref{app:eq:broadcast_latency_3_1_5} as follows:
\begin{align}\label{app:eq:broadcast_latency_3_1_8}
    &Q^{\downarrow,\mathsf{M}}_{b}(\bm{v}){=}M^{\downarrow}_{b}(x)+1{\geq} \widehat{Q}^{\downarrow,\mathsf{M}}_{b}(\bm{v};\ell) {\triangleq} \left(\frac{M^{\downarrow}_{b}(x) Q^{\downarrow,\mathsf{M}}_{b}([\bm{v}]^{(\ell-1)})}{\left[M^{\downarrow}_{b}(x)\right]^{(\ell-1)}}\right)^{\frac{\left[M^{\downarrow}_{b}(x)\right]^{(\ell-1)}}{Q^{\downarrow,\mathsf{M}}_{b}([\bm{v}]^{(\ell-1)})}}\times\left(Q^{\downarrow,\mathsf{M}}_{b}([\bm{v}]^{(\ell-1)})\right)^{\frac{1}{Q^{\downarrow,\mathsf{M}}_{b}([\bm{v}]^{(\ell-1)})}},
\end{align}
which gives the following approximation of~\eqref{app:eq:broadcast_latency_6}:
\begin{equation}
     \frac{\left(\mathscr{L}^{\downarrow,\mathsf{M}}_{b}(t_{x})\right)^{-1}\left(\displaystyle\sum_{z{=}0}^{x}\sum_{r'
   {\in}\mathcal{R}_{b}}T_{b,r'}^{\downarrow}(t_{z})\mathfrak{R}^{\downarrow}_{b,r'}(t_{z})+1\right)}{\widehat{Q}^{\downarrow,\mathsf{M}}_{b}(\bm{v};\ell)}\le 1.
\end{equation}

We finally approximate constraint~\eqref{app:eq:broadcast_latency_3_1} as follows:
\begin{tcolorbox}[ams align]
     &\frac{M^{\downarrow}_{b}(x)+1}{\widehat{S}^{\downarrow,\mathsf{M}}_{b}(\bm{v};\ell)}\le 1,~~~~~~~~\frac{M^{\downarrow}_{b}(x)+1}{\widehat{Q}^{\downarrow,\mathsf{M}}_{b}(\bm{v};\ell)}\le 1,~~~~~~~\frac{1}{\mathscr{L}^{\downarrow,\mathsf{M}}_{b}(t_{x})}\le 1,\nonumber
\end{tcolorbox}

We next aim to transform \eqref{app:eq:broadcast_latency_1_1} into the standard form of GP. In doing so, we rewrite \eqref{app:eq:broadcast_latency_1_1} as follows:
\begin{equation}\label{app:eq:broadcast_latency_1_2}
\begin{aligned}
   \tau_{b}^{\downarrow,{(k)}}&=\sum_{x{=}0}^{N}\max_{r{\in}\mathcal{R}_{b}}\Big\{T_{b,r}^{\downarrow}(t_{x})+\epsilon\Big\}-\sum_{x{=}0}^{N}\epsilon\\
   &=\sum_{x{=}0}^{N}\max_{r{\in}\mathcal{R}_{b}}\Big\{T_{b,r}^{\downarrow}(t_{x})+\epsilon\Big\}-\left(N^{(k)}-N^{(k{-}1)}\right)\epsilon
\end{aligned}
\end{equation}
Performing some algebraic manipulations gives us
\begin{equation}\label{app:eq:broadcast_latency_1_2}
   \frac{\tau_{b}^{\downarrow,{(k)}}+N^{(k)}\epsilon}{\displaystyle\sum_{x{=}0}^{N}\max_{r{\in}\mathcal{R}_{b}}\Big\{T_{b,r}^{\downarrow}(t_{x})+\epsilon\Big\}+N^{(k{-}1)}\epsilon}=1.
\end{equation} 
Using the approximation $\max\{A, B\}\approx (A^{p}+B^{p})^{-\frac{1}{p}}$, which is tight when $p \gg 1$, gives us
\begin{equation}\label{app:eq:broadcast_latency_1_3}
   \frac{\tau_{b}^{\downarrow,{(k)}}+N^{(k)}\epsilon}{\displaystyle\sum_{x{=}0}^{N}\left(\sum_{r{\in}\mathcal{R}_{b}}\left(T_{b,r}^{\downarrow}(t_{x})+\epsilon\right)^{p}\right)^{-\frac{1}{p}}+N^{(k{-}1)}\epsilon}=1.
\end{equation} 
Defining an auxiliary decision variable $T_{b}^{\downarrow,\mathsf{max}}(t_{x})$ that satisfies
\begin{equation}\label{app:eq:broadcast_latency_1_4}
    T_{b}^{\downarrow,\mathsf{max}}(t_{x}) =\sum_{r{\in}\mathcal{R}_{b}}\left(T_{b,r}^{\downarrow}(t_{x})+\epsilon\right)^{p}
\end{equation}
results in
\begin{equation}\label{app:eq:broadcast_latency_1_5}
   \frac{\tau_{b}^{\downarrow,{(k)}}+N^{(k)}\epsilon}{\displaystyle\sum_{x{=}0}^{N}\left(T_{b}^{\downarrow,\mathsf{max}}(t_{x})\right)^{-\frac{1}{p}}+N^{(k{-}1)}\epsilon}=1.
\end{equation} 
This constraint is not in the format of GP. Therefore, we transform it through splitting it into the following three inequalities:
\begin{equation}\label{app:eq:broadcast_latency_1_6}
    \frac{\tau_{b}^{\downarrow,{(k)}}+N^{(k)}\epsilon}{\displaystyle\sum_{x{=}0}^{N}\left(T_{b}^{\downarrow,\mathsf{max}}(t_{x})\right)^{-\frac{1}{p}}+N^{(k{-}1)}\epsilon}\le 1,
\end{equation}

\begin{equation}\label{app:eq:broadcast_latency_1_7}
    \frac{\left(\mathscr{Q}^{\downarrow}_{b}(t_{x})\right)^{-1}\left(\displaystyle\sum_{x{=}0}^{N}\left(T_{b}^{\downarrow,\mathsf{max}}(t_{x})\right)^{-\frac{1}{p}}+N^{(k{-}1)}\epsilon\right)}{\tau_{b}^{\downarrow,{(k)}}+N^{(k)}\epsilon}\le 1,
\end{equation}

\begin{equation}
    \mathscr{Q}^{\downarrow}_{b}(t_{x})\ge 1,
\end{equation}
where $\mathscr{Q}^{\downarrow}_{b}(t_{x})$ is an auxiliary decision variable added with a large penalty term to the objective function to force $\mathscr{Q}^{\downarrow}_{b}(t_{x}){\rightarrow}1^+$ at the optimal point. The fractions in~\eqref{app:eq:broadcast_latency_1_6} and \eqref{app:eq:broadcast_latency_1_7} still need transformation since they are inequalities with posynomials in their denominators, which are not posynomials. We thus exploit arithmetic-geometric mean inequality (Lemma~\ref{Lemma:ArethmaticGeometric}) to approximate the denominators in \eqref{app:eq:broadcast_latency_1_6} and \eqref{app:eq:broadcast_latency_1_7} with monomials. To this end, we approximate the denominator in \eqref{app:eq:broadcast_latency_1_6} as follows:
\begin{align}\label{app:eq:broadcast_latency_1_8}
    &V^{\downarrow}_{b}(\bm{v})=\sum_{x{=}0}^{N}\left(T_{b}^{\downarrow,\mathsf{max}}(t_{x})\right)^{-\frac{1}{p}}+N^{(k{-}1)}\epsilon\nonumber\\
   &\geq \widehat{V}^{\downarrow}_{b}(\bm{v};\ell) \triangleq \prod_{x{=}0}^{N}\left(\frac{\left(T_{b}^{\downarrow,\mathsf{max}}(t_{x})\right)^{-\frac{1}{p}} V^{\downarrow}_{b}([\bm{v}]^{(\ell-1)})}{\left[\left(T_{b}^{\downarrow,\mathsf{max}}(t_{x})\right)^{-\frac{1}{p}}\right]^{(\ell-1)}}\right)^{\frac{\left[\left(T_{b}^{\downarrow,\mathsf{max}}(t_{x})\right)^{-\frac{1}{p}}\right]^{(\ell-1)}}{V^{\downarrow}_{b}([\bm{v}]^{(\ell-1)})}}\times\left(V^{\downarrow}_{b}([\bm{v}]^{(\ell-1)})\right)^{\frac{\left[N^{(k{-}1)}\displaystyle\epsilon\right]^{(\ell-1)}}{V^{\downarrow}_{b}([\bm{v}]^{(\ell-1)})}},
\end{align}
which gives us an approximation of~\eqref{app:eq:broadcast_latency_1_6} as follows:
\begin{equation}\label{app:eq:broadcast_latency_1_9}
    \frac{\tau_{b}^{\downarrow,{(k)}}+N^{(k)}\epsilon}{\widehat{V}^{\downarrow}_{b}(\bm{v};\ell)}\le 1.
\end{equation}

Similarly, we approximate the denominator in \eqref{app:eq:broadcast_latency_1_7} as follows:
\begin{align}\label{app:eq:broadcast_latency_1_10}
    W^{\downarrow}_{b}(\bm{v})&=\tau_{b}^{\downarrow,{(k)}}+N^{(k)}\epsilon \geq \widehat{W}^{\downarrow}_{b}(\bm{v};\ell) \triangleq \left(\frac{\tau_{b}^{\downarrow,{(k)}} W^{\downarrow}_{b}([\bm{v}]^{(\ell-1)})}{\left[\tau_{b}^{\downarrow,{(k)}}\right]^{(\ell-1)}}\right)^{\frac{\left[\tau_{b}^{\downarrow,{(k)}}\right]^{(\ell-1)}}{W^{\downarrow}_{b}([\bm{v}]^{(\ell-1)})}}\times\left(W^{\downarrow}_{b}([\bm{v}]^{(\ell-1)})\right)^{\frac{\left[N^{(k)}\displaystyle\epsilon\right]^{(\ell-1)}}{W^{\downarrow}_{b}([\bm{v}]^{(\ell-1)})}},
\end{align}
which gives us an approximation of~\eqref{app:eq:broadcast_latency_1_7} as follows:
\begin{equation}\label{app:eq:broadcast_latency_1_11}
    \frac{\left(\mathscr{Q}^{\downarrow}_{b}(t_{x})\right)^{-1}\left(\displaystyle\sum_{x{=}0}^{N}\left(T_{b}^{\downarrow,\mathsf{max}}(t_{x})\right)^{-\frac{1}{p}}+N^{(k{-}1)}\epsilon\right)}{\widehat{W}^{\downarrow}_{b}(\bm{v};\ell)}\le 1,
\end{equation}

We next transform \eqref{app:eq:broadcast_latency_1_4} into the standard form of GP. To do so, we rewrite \eqref{app:eq:broadcast_latency_1_4} as follows:
\begin{equation}\label{app:eq:broadcast_latency_1_4_1}
    \frac{T_{b}^{\downarrow,\mathsf{max}}(t_{x})}{\sum_{r{\in}\mathcal{R}_{b}}\left(T_{b,r}^{\downarrow}(t_{x})+\epsilon\right)^{p}} =1.
\end{equation}
Defining an auxiliary decision variable $T_{b,r}^{\downarrow,\epsilon}(t_{x})$ satisfying
\begin{equation}\label{app:eq:broadcast_latency_1_4_2}
    T_{b,r}^{\downarrow,\epsilon}(t_{x})=T_{b,r}^{\downarrow}(t_{x})+\epsilon
\end{equation}
leading to
\begin{equation}\label{app:eq:broadcast_latency_1_4_3}
    \frac{T_{b}^{\downarrow,\mathsf{max}}(t_{x})}{\sum_{r{\in}\mathcal{R}_{b}}\left(T_{b,r}^{\downarrow,\epsilon}(t_{x})\right)^{p}} =1.
\end{equation}
This constraint is not in the format of GP. Therefore, we transform it through splitting it into the following three inequalities:
\begin{equation}\label{app:eq:broadcast_latency_1_4_4}
    \frac{T_{b}^{\downarrow,\mathsf{max}}(t_{x})}{\sum_{r{\in}\mathcal{R}_{b}}\left(T_{b,r}^{\downarrow,\epsilon}(t_{x})\right)^{p}}\le 1,
\end{equation}

\begin{equation}\label{app:eq:broadcast_latency_1_4_5}
    \frac{\left(\mathscr{R}^{\downarrow,\mathsf{max}}_{b}(t_{x})\right)^{-1}\left(\sum_{r{\in}\mathcal{R}_{b}}\left(T_{b,r}^{\downarrow,\epsilon}(t_{x})\right)^{p}\right)}{T_{b}^{\downarrow,\mathsf{max}}(t_{x})}\le 1,
\end{equation}

\begin{equation}
    \mathscr{R}^{\downarrow,\mathsf{max}}_{b}(t_{x})\ge 1,
\end{equation}
where $\mathscr{R}^{\downarrow,\mathsf{max}}_{b}(t_{x})$ is an auxiliary decision variable added with a large penalty term to the objective function to force $\mathscr{R}^{\downarrow,\mathsf{max}}_{b}(t_{x}){\rightarrow}1^+$ at the optimal point. The fraction in~\eqref{app:eq:broadcast_latency_1_4_4} still needs transformation since it is inequality with a posynomial in its denominator, which is not a posynomial. We thus exploit arithmetic-geometric mean inequality (Lemma~\ref{Lemma:ArethmaticGeometric}) to approximate the denominator in \eqref{app:eq:broadcast_latency_1_4_4} with a monomial.
\begin{align}\label{app:eq:broadcast_latency_1_4_6}
    &Y^{\downarrow,\mathsf{max}}_{b}(\bm{v})=\sum_{r{\in}\mathcal{R}_{b}}\left(T_{b,r}^{\downarrow,\epsilon}(t_{x})\right)^{p}\geq \widehat{Y}^{\downarrow,\mathsf{max}}_{b}(\bm{v};\ell) \triangleq \prod_{r{\in}\mathcal{R}_{b}}\left(\frac{\left(T_{b,r}^{\downarrow,\epsilon}(t_{x})\right)^{p} Y^{\downarrow,\mathsf{max}}_{b}([\bm{v}]^{(\ell-1)})}{\left[\left(T_{b,r}^{\downarrow,\epsilon}(t_{x})\right)^{p}\right]^{(\ell-1)}}\right)^{\frac{\left[\left(T_{b,r}^{\downarrow,\epsilon}(t_{x})\right)^{p}\right]^{(\ell-1)}}{Y^{\downarrow,\mathsf{max}}_{b}([\bm{v}]^{(\ell-1)})}},
\end{align}
which gives us an approximation of~\eqref{app:eq:broadcast_latency_1_4_4} as follows:
\begin{equation}\label{app:eq:broadcast_latency_1_4_6}
    \frac{T_{b}^{\downarrow,\mathsf{max}}(t_{x})}{\widehat{Y}^{\downarrow,\mathsf{max}}_{b}(\bm{v};\ell)}\le 1.
\end{equation}

We next transform \eqref{app:eq:broadcast_latency_1_4_2} into GP format. To this end, we rewrite \eqref{app:eq:broadcast_latency_1_4_2} as follows:
\begin{equation}\label{app:eq:broadcast_latency_1_4_7}
    \frac{T_{b,r}^{\downarrow,\epsilon}(t_{x})}{T_{b,r}^{\downarrow}(t_{x})+\epsilon}=1.
\end{equation}
This constraint is not in the format of GP. Therefore, we transform it through splitting it into the following three inequalities:
\begin{equation}\label{app:eq:broadcast_latency_1_4_4}
    \frac{T_{b,r}^{\downarrow,\epsilon}(t_{x})}{T_{b,r}^{\downarrow}(t_{x})+\epsilon}\le 1,
\end{equation}

\begin{equation}\label{app:eq:broadcast_latency_1_4_5}
    \frac{\left(\mathscr{S}^{\downarrow,\epsilon}_{b,r}(t_{x})\right)^{-1}\left(T_{b,r}^{\downarrow}(t_{x})+\epsilon\right)}{T_{b,r}^{\downarrow,\epsilon}(t_{x})}\le 1,
\end{equation}

\begin{equation}
    \mathscr{S}^{\downarrow,\epsilon}_{b,r}(t_{x})\ge 1,
\end{equation}
where $\mathscr{S}^{\downarrow,\epsilon}_{b,r}(t_{x})$ is an auxiliary decision variable added with a large penalty term to the objective function to force $\mathscr{S}^{\downarrow,\epsilon}_{b,r}(t_{x}){\rightarrow}1^+$ at the optimal point. The fraction in~\eqref{app:eq:broadcast_latency_1_4_4} still needs transformation since it is inequality with a posynomial in its denominator, which is not a posynomial. We thus exploit arithmetic-geometric mean inequality (Lemma~\ref{Lemma:ArethmaticGeometric}) to approximate the denominator in \eqref{app:eq:broadcast_latency_1_4_4} with a monomial.
\begin{align}\label{app:eq:broadcast_latency_1_4_6}
    Z^{\downarrow,\epsilon}_{b,r}(\bm{v})&=T_{b,r}^{\downarrow}(t_{x})+\epsilon \geq \widehat{Z}^{\downarrow,\epsilon}_{b,r}(\bm{v};\ell) \triangleq \left(\frac{T_{b,r}^{\downarrow}(t_{x}) Z^{\downarrow,\epsilon}_{b,r}([\bm{v}]^{(\ell-1)})}{\left[T_{b,r}^{\downarrow}(t_{x})\right]^{(\ell-1)}}\right)^{\frac{\left[T_{b,r}^{\downarrow}(t_{x})\right]^{(\ell-1)}}{Z^{\downarrow,\epsilon}_{b,r}([\bm{v}]^{(\ell-1)})}}\times\left(Z^{\downarrow,\epsilon}_{b,r}([\bm{v}]^{(\ell-1)})\right)^{\frac{\left[\displaystyle\epsilon\right]^{(\ell-1)}}{Z^{\downarrow,\epsilon}_{b,r}([\bm{v}]^{(\ell-1)})}},
\end{align}
which gives us an approximation of~\eqref{app:eq:broadcast_latency_1_4_4} as follows:
\begin{equation}\label{app:eq:broadcast_latency_1_4_7}
    \frac{T_{b,r}^{\downarrow,\epsilon}(t_{x})}{\widehat{Z}^{\downarrow,\epsilon}_{b,r}(\bm{v};\ell)}\le 1.
\end{equation}

We finally approximate constraint~\eqref{app:eq:broadcast_latency_1_1} as follows:
\begin{tcolorbox}[ams align]
     &\frac{\tau_{b}^{\downarrow,{(k)}}+N^{(k)}\epsilon}{\widehat{V}^{\downarrow}_{b}(\bm{v};\ell)}\le 1,~~~~~~\frac{\left(\mathscr{Q}^{\downarrow}_{b}(t_{x})\right)^{-1}\left(\displaystyle\sum_{x{=}0}^{N}\left(T_{b}^{\downarrow,\mathsf{max}}(t_{x})\right)^{-\frac{1}{p}}+N^{(k{-}1)}\epsilon\right)}{\widehat{W}^{\downarrow}_{b}(\bm{v};\ell)}\le 1,\nonumber\\
     &\frac{T_{b}^{\downarrow,\mathsf{max}}(t_{x})}{\widehat{Y}^{\downarrow,\mathsf{max}}_{b}(\bm{v};\ell)}\le 1,~~~~\frac{\left(\mathscr{R}^{\downarrow,\mathsf{max}}_{b}(t_{x})\right)^{-1}\left(\sum_{r{\in}\mathcal{R}_{b}}\left(T_{b,r}^{\downarrow,\epsilon}(t_{x})\right)^{p}\right)}{T_{b}^{\downarrow,\mathsf{max}}(t_{x})}\le 1,\nonumber\\
     & \frac{T_{b,r}^{\downarrow,\epsilon}(t_{x})}{\widehat{Z}^{\downarrow,\epsilon}_{b,r}(\bm{v};\ell)}\le 1,~~~\frac{\left(\mathscr{S}^{\downarrow,\epsilon}_{b,r}(t_{x})\right)^{-1}\left(T_{b,r}^{\downarrow}(t_{x})+\epsilon\right)}{T_{b,r}^{\downarrow,\epsilon}(t_{x})}\le 1,\nonumber\\
     &\frac{1}{\mathscr{Q}^{\downarrow}_{b}(t_{x})}\le 1,~~~~\frac{1}{\mathscr{R}^{\downarrow,\mathsf{max}}_{b}(t_{x})}\le 1,~~~~\frac{1}{\mathscr{S}^{\downarrow,\epsilon}_{b,r}(t_{x})}\le 1.\nonumber
\end{tcolorbox}

\noindent\textbf{GPs dispersion latency of DPUs (Proposition~\ref{propo:DPU_uplink}).} Referring to Sec.~\ref{sec:communication_latency}, let us revisit the equations of GPs dispersion latency of DPUs below:
\begin{equation}\label{app:eq:dispersion_1}
    \sum_{u'{\in} \mathcal{U}_{b}}\sum_{r{\in}\overline{\mathcal{R}}_{b}}\overline{\psi}_{u,u',r}(t_{x}){=}1,~x\in\mathcal{N}^{(k)},
\end{equation}

\begin{equation}\label{app:eq:dispersion_2}
    \overline{\tau}_{u}^{\uparrow,(k)}{=}\hspace{-3mm}\sum_{x{=}N^{(k{-}1)}{+}1}^{N^{(k)}}\max_{\substack{u'{\in}\mathcal{U}_{b},\\r{\in}\overline{\mathcal{R}}_{b}}}\Big\{\min\big\{\overline{\tau}_{u,u',r}^{\uparrow}(t_{x}),\left(t_{x{+}1}{-}t_{x}\right)\big\}\Big\}.
\end{equation}
In \eqref{app:eq:dispersion_2}, $\overline{\tau}_{u,u',r}^{\uparrow}(t_{x})$ is calculate as follows:
\begin{equation}\label{app:eq:dispersion_3}
    \overline{\tau}_{u,u',r}^{\uparrow}(t_{x}){=}\frac{(\alpha_{\bm{\omega}} M{-}\overline{M}^{\uparrow}_{u}(x))\overline{\psi}_{u,u',r}(t_{x})}{\overline{\mathfrak{R}}^{\uparrow}_{u,u',r}(t_{x})+1{-}\overline{\beta}^{\uparrow}_u(t_{x})},
\end{equation}
where
\begin{align}\label{app:eq:dispersion_4}
    \overline{M}^{\uparrow}_{u}(x)=\hspace{-3mm}\sum_{z{=}N^{(k{-}1)}{+}1}^{x}\sum_{u'{\in}\mathcal{U}_{b}}\sum_{r'{\in}\overline{\mathcal{R}}_{b}}\min\Big\{\overline{\tau}_{u,u',r'}^{\uparrow}(t_{z}),\left(t_{z{+}1}-t_{z}\right)\Big\}\overline{\mathfrak{R}}^{\uparrow}_{u,u',r'}(t_{z}).
\end{align}

In the following, we first transform \eqref{app:eq:dispersion_1} into the standard form of GP. In doing so, we split \eqref{app:eq:dispersion_1} into the following three inequalities:
\begin{equation}\label{app:eq:dispersion_1_2}
    \sum_{u'{\in} \mathcal{U}_{b}}\sum_{r{\in}\overline{\mathcal{R}}_{b}}\overline{\psi}_{u,u',r}(t_{x})\le 1,
\end{equation}
\begin{equation}\label{app:eq:dispersion_1_3}
    \frac{\left(\mathscr{B}_{u,u',r}(t_{x})\right)^{-1}}{\sum_{u'{\in} \mathcal{U}_{b}}\sum_{r{\in}\overline{\mathcal{R}}_{b}}\overline{\psi}_{u,u',r}(t_{x})}\le 1,
\end{equation}
\begin{equation}
    \mathscr{B}_{u,u',r}(t_{x})\ge 1,
\end{equation}
where $\mathscr{B}_{u,u',r}(t_{x})$ is an auxiliary decision variable added with a large penalty term to the objective function to force $\mathscr{B}_{u,u',r}(t_{x}){\rightarrow}1^+$ at the optimal point. The fraction in~\eqref{app:eq:dispersion_1_3} still needs transformation since it is inequality with a posynomial in its denominator, which is not a posynomial. We thus exploit arithmetic-geometric mean inequality (Lemma~\ref{Lemma:ArethmaticGeometric}) to approximate the denominator in \eqref{app:eq:dispersion_1_3} with a monomial.
\begin{align}\label{app:eq:dispersion_1_3_1}
    D_{u,u',r}(\bm{v})&=\sum_{u'{\in} \mathcal{U}_{b}}\sum_{r{\in}\overline{\mathcal{R}}_{b}}\overline{\psi}_{u,u',r}(t_{x}) \geq \widehat{D}_{u,u',r}(\bm{v};\ell) \triangleq \prod_{u'{\in} \mathcal{U}_{b}}\prod_{r{\in}\overline{\mathcal{R}}_{b}} \left(\frac{\overline{\psi}_{u,u',r}(t_{x}) D_{u,u',r}([\bm{v}]^{(\ell-1)})}{\left[\overline{\psi}_{u,u',r}(t_{x})\right]^{(\ell-1)}}\right)^{\frac{\left[\overline{\psi}_{u,u',r}(t_{x})\right]^{(\ell-1)}}{D_{u,u',r}([\bm{v}]^{(\ell-1)})}},
\end{align}
which gives us an approximation of~\eqref{app:eq:dispersion_1_3} as follows:
\begin{equation}\label{app:eq:dispersion_1_3_3}
    \frac{\left(\mathscr{B}_{u,u',r}(t_{x})\right)^{-1}}{\widehat{D}_{u,u',r}(\bm{v};\ell)}\le 1.
\end{equation}
We finally approximate constraint~\eqref{app:eq:dispersion_1} as follows:
\begin{tcolorbox}[ams align]
     &\sum_{u'{\in} \mathcal{U}_{b}}\sum_{r{\in}\overline{\mathcal{R}}_{b}}\overline{\psi}_{u,u',r}(t_{x})\le 1,~~~~~~~\frac{\left(\mathscr{B}_{u,u',r}(t_{x})\right)^{-1}}{\widehat{D}_{u,u',r}(\bm{v};\ell)}\le 1,~~~~~~~\frac{1}{\mathscr{B}_{u,u',r}(t_{x})}\le 1.\nonumber
\end{tcolorbox}

We next transform \eqref{app:eq:dispersion_2} into the standard form of GP. To this end, we define $\overline{\tau}_{u}^{\uparrow,(k)}$, $\overline{\tau}_{u,u',r}^{\uparrow}(t_{x})$, and $\overline{M}^{\uparrow}_{u}(x)$ as auxiliary decision variables that must satisfy constraints \eqref{app:eq:dispersion_2}, \eqref{app:eq:dispersion_3}, and \eqref{app:eq:dispersion_4}, respectively. Furthermore, we define an auxiliary decision variable $T_{u,u',r}^{\uparrow}(t_{x})$ satisfying the following constraint:
\begin{equation}\label{app:eq:dispersion_min_1}
    T_{u,u',r}^{\uparrow}(t_{x})=\min\big\{\overline{\tau}_{u,u',r}^{\uparrow}(t_{x}),\left(t_{x{+}1}{-}t_{x}\right)\big\}.
\end{equation}
Accordingly, \eqref{app:eq:dispersion_2} and \eqref{app:eq:dispersion_4} can be rewritten as follows:
\begin{equation}\label{app:eq:dispersion_2_1}
    \overline{\tau}_{u}^{\uparrow,(k)}{=}\hspace{-3mm}\sum_{x{=}N^{(k{-}1)}{+}1}^{N^{(k)}}\max_{\substack{u'{\in}\mathcal{U}_{b},\\r{\in}\overline{\mathcal{R}}_{b}}}\Big\{T_{u,u',r}^{\uparrow}(t_{x})\Big\},
\end{equation}
and
\begin{align}\label{app:eq:dispersion_4_1}
    \overline{M}^{\uparrow}_{u}(x)=\hspace{-3mm}\sum_{z{=}N^{(k{-}1)}{+}1}^{x}\sum_{u'{\in}\mathcal{U}_{b}}\sum_{r'{\in}\overline{\mathcal{R}}_{b}}T_{u,u',r'}^{\uparrow}(t_{z})\overline{\mathfrak{R}}^{\uparrow}_{u,u',r'}(t_{z}).
\end{align}

In the following, we first aim to transform \eqref{app:eq:dispersion_min_1} into the standard form of GP. To this end, using the approximation $\min\{A, B\}\approx (A^{-p}+B^{-p})^{-\frac{1}{p}}$, which is tight when $p \gg 1$, gives us
\begin{equation}\label{app:eq:dispersion_min_2}
    T_{u,u',r}^{\uparrow}(t_{x})+\epsilon=\left(\left(\overline{\tau}_{u,u',r}^{\uparrow}(t_{x})+\epsilon\right)^{-p}+\left(t_{x{+}1}{-}t_{x}+\epsilon\right)^{-p}\right)^{-\frac{1}{p}},
\end{equation}
where $\epsilon$ is added to both sides to avoid division by zero. To simplify \eqref{app:eq:dispersion_min_2}, we define an auxiliary decision variable $\widetilde{T}_{u,u',r}^{\uparrow}(t_{x})$ that satisfy the following equality constraint:
\begin{equation}\label{app:eq:dispersion_min_3}
    \widetilde{T}_{u,u',r}^{\uparrow}(t_{x})=\left(\overline{\tau}_{u,u',r}^{\uparrow}(t_{x})+\epsilon\right)^{-p}+\left(t_{x{+}1}{-}t_{x}+\epsilon\right)^{-p}
\end{equation}
resulting in 
\begin{equation}\label{app:eq:dispersion_min_4}
    T_{u,u',r}^{\uparrow}(t_{x})+\epsilon=\left(\widetilde{T}_{u,u',r}^{\uparrow}(t_{x})\right)^{-\frac{1}{p}}.
\end{equation}
Performing some algebraic operations gives us
\begin{equation}\label{app:eq:dispersion_min_5}
    T_{u,u',r}^{\uparrow}(t_{x})\times \left(\widetilde{T}_{u,u',r}^{\uparrow}(t_{x})\right)^{\frac{1}{p}}+\epsilon\times \left(\widetilde{T}_{u,u',r}^{\uparrow}(t_{x})\right)^{\frac{1}{p}} = 1.
\end{equation}
We transform \eqref{app:eq:dispersion_min_5} into the standard GP format via introducing the following three inequalities:
\begin{equation}\label{app:eq:dispersion_min_6}
     T_{u,u',r}^{\uparrow}(t_{x})\times \left(\widetilde{T}_{u,u',r}^{\uparrow}(t_{x})\right)^{\frac{1}{p}}+\epsilon\times \left(\widetilde{T}_{u,u',r}^{\uparrow}(t_{x})\right)^{\frac{1}{p}}\le 1,
\end{equation}
\begin{equation}\label{app:eq:dispersion_min_7}
    \frac{\left(\mathscr{B}^{\uparrow,\mathsf{min}}_{u,u',r}(t_{x})\right)^{-1}}{T_{u,u',r}^{\uparrow}(t_{x})\times \left(\widetilde{T}_{u,u',r}^{\uparrow}(t_{x})\right)^{\frac{1}{p}}+\epsilon\times \left(\widetilde{T}_{u,u',r}^{\uparrow}(t_{x})\right)^{\frac{1}{p}}}\le 1,
\end{equation}
\begin{equation}
    \mathscr{B}^{\uparrow,\mathsf{min}}_{u,u',r}(t_{x})\ge 1,
\end{equation}
where $\mathscr{B}^{\uparrow,\mathsf{min}}_{u,u',r}(t_{x})$ is an auxiliary decision variable added with a large penalty term to the objective function to force $\mathscr{B}^{\uparrow,\mathsf{min}}_{u,u',r}(t_{x}){\rightarrow}1^+$ at the optimal point. The fraction in~\eqref{app:eq:dispersion_min_7} still needs transformation since it is an inequality with a posynomial in its denominator, which is not posynomial. We thus exploit arithmetic-geometric mean inequality (Lemma~\ref{Lemma:ArethmaticGeometric}) to approximate the denominator of \eqref{app:eq:dispersion_min_7} with a monomial:
\begin{align}\label{app:eq:dispersion_min_8}
    W^{\uparrow,\mathsf{min}}_{u,u',r}(\bm{v})&{=}T_{u,u',r}^{\uparrow}(t_{x}) \left(\widetilde{T}_{u,u',r}^{\uparrow}(t_{x})\right)^{\frac{1}{p}}{+}\epsilon{\times} \left(\widetilde{T}_{u,u',r}^{\uparrow}(t_{x})\right)^{\frac{1}{p}}\nonumber\\
    &{\geq} \widehat{W}^{\uparrow,\mathsf{min}}_{u,u',r}(\bm{v};\ell) {\triangleq} \left(\frac{T_{u,u',r}^{\uparrow}(t_{x}) \left(\widetilde{T}_{u,u',r}^{\uparrow}(t_{x})\right)^{\frac{1}{p}} W^{\uparrow,\mathsf{min}}_{u,u',r}([\bm{v}]^{(\ell-1)})}{\left[T_{u,u',r}^{\uparrow}(t_{x}) \left(\widetilde{T}_{u,u',r}^{\uparrow}(t_{x})\right)^{\frac{1}{p}}\right]^{(\ell-1)}}\right)^{\hspace{-2mm}\frac{\left[T_{u,u',r}^{\uparrow}(t_{x}) \left(\widetilde{T}_{u,u',r}^{\uparrow}(t_{x})\right)^{\frac{1}{p}}\right]^{(\ell-1)}}{W^{\uparrow,\mathsf{min}}_{u,u',r}([\bm{v}]^{(\ell-1)})}}\nonumber\\
    &~~~~~~~~~~~~~~~~~\times\left(\frac{ \left(\widetilde{T}_{u,u',r}^{\uparrow}(t_{x})\right)^{\frac{1}{p}} W^{\uparrow,\mathsf{min}}_{u,u',r}([\bm{v}]^{(\ell-1)})}{\left[ \left(\widetilde{T}_{u,u',r}^{\uparrow}(t_{x})\right)^{\frac{1}{p}}\right]^{(\ell-1)}}\right)^{\frac{\left[\epsilon\times \left(\widetilde{T}_{u,u',r}^{\uparrow}(t_{x})\right)^{\frac{1}{p}}\right]^{(\ell-1)}}{W^{\uparrow,\mathsf{min}}_{u,u',r}([\bm{v}]^{(\ell-1)})}},
\end{align}
which gives us the following approximation of~\eqref{app:eq:dispersion_min_7}:
\begin{equation}
     \frac{\left(\mathscr{B}^{\uparrow,\mathsf{min}}_{u,u',r}(t_{x})\right)^{-1}}{\widehat{W}^{\uparrow,\mathsf{min}}_{u,u',r}(\bm{v};\ell)}\le 1.
\end{equation}

We next aim to transform \eqref{app:eq:dispersion_min_3} into standard GP format. In doing so, we rewrite \eqref{app:eq:dispersion_min_3} as follows:
\begin{equation}\label{app:eq:dispersion_min_9}
    \frac{\widetilde{T}_{u,u',r}^{\uparrow}(t_{x})}{\left(\overline{\tau}_{u,u',r}^{\uparrow}(t_{x})+\epsilon\right)^{-p}+\left(t_{x{+}1}{-}t_{x}+\epsilon\right)^{-p}}= 1.
\end{equation}
Defining an auxiliary decision variable $\widetilde{\tau}_{u,u',r}^{\uparrow}(t_{x})$ that satisfies 
\begin{equation}\label{app:eq:dispersion_min_101}
    \widetilde{\tau}_{u,u',r}^{\uparrow}(t_{x})=\overline{\tau}_{u,u',r}^{\uparrow}(t_{x})+\epsilon
\end{equation}
and considering $\widetilde{t}_{x}$ in \eqref{app:eq:broadcast_latency_min_10} give us
\begin{equation}\label{app:eq:dispersion_min_11}
    \frac{\widetilde{T}_{u,u',r}^{\uparrow}(t_{x})}{\left(\widetilde{\tau}_{u,u',r}^{\uparrow}(t_{x})\right)^{-p}+\left(\widetilde{t}_{x}\right)^{-p}}= 1.
\end{equation}
This constraint is not in the format of GP. Therefore, we transform it through splitting it into the following three inequalities:
\begin{equation}\label{app:eq:dispersion_min_12}
    \frac{\widetilde{T}_{u,u',r}^{\uparrow}(t_{x})}{\left(\widetilde{\tau}_{u,u',r}^{\uparrow}(t_{x})\right)^{-p}+\left(\widetilde{t}_{x}\right)^{-p}}\le 1,
\end{equation}
\begin{equation}\label{app:eq:dispersion_min_13}
    \frac{\left(\mathscr{C}^{\uparrow,\mathsf{min}}_{u,u',r}(t_{x})\right)^{-1}\left(\left(\widetilde{\tau}_{u,u',r}^{\uparrow}(t_{x})\right)^{-p}+\left(\widetilde{t}_{x}\right)^{-p}\right)}{\widetilde{T}_{u,u',r}^{\uparrow}(t_{x})}\le 1,
\end{equation}
\begin{equation}
    \mathscr{C}^{\uparrow,\mathsf{min}}_{u,u',r}(t_{x})\ge 1,
\end{equation}
where $\mathscr{C}^{\uparrow,\mathsf{min}}_{u,u',r}(t_{x})$ is an auxiliary decision variable added with a large penalty term to the objective function to force $\mathscr{C}^{\uparrow,\mathsf{min}}_{u,u',r}(t_{x}){\rightarrow}1^+$ at the optimal point. The fraction in~\eqref{app:eq:dispersion_min_12} still needs transformation since it is an inequality with a posynomial in the denominator, which is not a posynomial. We thus exploit arithmetic-geometric mean inequality (Lemma~\ref{Lemma:ArethmaticGeometric}) to approximate the denominator with a monomial:
\begin{align}\label{app:eq:dispersion_min_14}
    G^{\uparrow,\mathsf{min}}_{u,u',r}(\bm{v})=\left(\widetilde{\tau}_{u,u',r}^{\uparrow}(t_{x})\right)^{-p}+\left(\widetilde{t}_{x}\right)^{-p}\geq \widehat{G}^{\uparrow,\mathsf{min}}_{u,u',r}(\bm{v};\ell) &\triangleq \left(\frac{\left(\widetilde{\tau}_{u,u',r}^{\uparrow}(t_{x})\right)^{-p} G^{\uparrow,\mathsf{min}}_{u,u',r}([\bm{v}]^{(\ell-1)})}{\left[\left(\widetilde{\tau}_{u,u',r}^{\uparrow}(t_{x})\right)^{-p}\right]^{(\ell-1)}}\right)^{\frac{\left[\left(\widetilde{\tau}_{u,u',r}^{\uparrow}(t_{x})\right)^{-p}\right]^{(\ell-1)}}{G^{\uparrow,\mathsf{min}}_{u,u',r}([\bm{v}]^{(\ell-1)})}}\nonumber\\
    &\times\left(\frac{\left(\widetilde{t}_{x}\right)^{-p} G^{\uparrow,\mathsf{min}}_{u,u',r}([\bm{v}]^{(\ell-1)})}{\left[\left(\widetilde{t}_{x}\right)^{-p}\right]^{(\ell-1)}}\right)^{\frac{\left[\left(\widetilde{t}_{x}\right)^{-p}\right]^{(\ell-1)}}{G^{\uparrow,\mathsf{min}}_{u,u',r}([\bm{v}]^{(\ell-1)})}},
\end{align}
which gives us an approximation of~\eqref{app:eq:dispersion_min_12} as follows:
\begin{equation}\label{app:eq:dispersion_min_15}
    \frac{\widetilde{T}_{u,u',r}^{\uparrow}(t_{x})}{\widehat{G}^{\uparrow,\mathsf{min}}_{u,u',r}(\bm{v};\ell)}\le 1.
\end{equation}

We next transform \eqref{app:eq:dispersion_min_101} into GP format. To this end, we rewrite \eqref{app:eq:dispersion_min_101} as follows:
\begin{equation}\label{app:eq:dispersion_min_102}
    \frac{\widetilde{\tau}_{u,u',r}^{\uparrow}(t_{x})}{\overline{\tau}_{u,u',r}^{\uparrow}(t_{x})+\epsilon}=1.
\end{equation}
This constraint is not in the format of GP. Therefore, we transform it through splitting it into the following three inequalities:
\begin{equation}\label{app:eq:dispersion_min_103}
    \frac{\widetilde{\tau}_{u,u',r}^{\uparrow}(t_{x})}{\overline{\tau}_{u,u',r}^{\uparrow}(t_{x})+\epsilon}\le 1,
\end{equation}
\begin{equation}\label{app:eq:dispersion_min_104}
    \frac{\left(\mathscr{F}^{\uparrow,\mathsf{min}}_{u,u',r}(t_{x})\right)^{-1}\left(\overline{\tau}_{u,u',r}^{\uparrow}(t_{x})+\epsilon\right)}{\widetilde{\tau}_{u,u',r}^{\uparrow}(t_{x})}\le 1,
\end{equation}
\begin{equation}
    \mathscr{F}^{\uparrow,\mathsf{min}}_{u,u',r}(t_{x})\ge 1,
\end{equation}
where $\mathscr{F}^{\uparrow,\mathsf{min}}_{u,u',r}(t_{x})$ is an auxiliary decision variable added with a large penalty term to the objective function to force $\mathscr{F}^{\uparrow,\mathsf{min}}_{u,u',r}(t_{x}){\rightarrow}1^+$ at the optimal point. The fraction in~\eqref{app:eq:dispersion_min_103} still needs transformation since it is inequality with posynomial in its denominator, which is not a posynomial. We thus exploit arithmetic-geometric mean inequality (Lemma~\ref{Lemma:ArethmaticGeometric}) to approximate the denominator of \eqref{app:eq:dispersion_min_103} with a monomial.
\begin{align}\label{app:eq:dispersion_min_105}
    J^{\uparrow,\mathsf{min}}_{u,u',r}(\bm{v})=\overline{\tau}_{u,u',r}^{\uparrow}(t_{x})+\epsilon \geq \widehat{J}^{\uparrow,\mathsf{min}}_{u,u',r}(\bm{v};\ell) &\triangleq \left(\frac{\overline{\tau}_{u,u',r}^{\uparrow}(t_{x}) J^{\uparrow,\mathsf{min}}_{u,u',r}([\bm{v}]^{(\ell-1)})}{\left[\overline{\tau}_{u,u',r}^{\uparrow}(t_{x})\right]^{(\ell-1)}}\right)^{\frac{\left[\overline{\tau}_{u,u',r}^{\uparrow}(t_{x})\right]^{(\ell-1)}}{J^{\uparrow,\mathsf{min}}_{u,u',r}([\bm{v}]^{(\ell-1)})}}\times\left(J^{\uparrow,\mathsf{min}}_{u,u',r}([\bm{v}]^{(\ell-1)})\right)^{\frac{\displaystyle\epsilon}{J^{\uparrow,\mathsf{min}}_{u,u',r}([\bm{v}]^{(\ell-1)})}},
\end{align}
which gives us an approximation of~\eqref{app:eq:dispersion_min_101} as follows:
\begin{equation}\label{app:eq:dispersion_min_106}
    \frac{\overline{\tau}_{u,u',r}^{\uparrow}(t_{x})+\epsilon}{\widehat{J}^{\uparrow,\mathsf{min}}_{u,u',r}(\bm{v};\ell)}\le 1.
\end{equation}
We finally approximate constraint~\eqref{app:eq:dispersion_min_1} as follows:
\begin{tcolorbox}[ams align]
     &T_{u,u',r}^{\uparrow}(t_{x})\times \left(\widetilde{T}_{u,u',r}^{\uparrow}(t_{x})\right)^{\frac{1}{p}}+\epsilon\times \left(\widetilde{T}_{u,u',r}^{\uparrow}(t_{x})\right)^{\frac{1}{p}}\le 1,~~\frac{\left(\mathscr{B}^{\uparrow,\mathsf{min}}_{u,u',r}(t_{x})\right)^{-1}}{\widehat{W}^{\uparrow,\mathsf{min}}_{u,u',r}(\bm{v};\ell)}\le 1,\nonumber\\
     &\frac{\left(\mathscr{C}^{\uparrow,\mathsf{min}}_{u,u',r}(t_{x})\right)^{-1}\left(\left(\widetilde{\tau}_{u,u',r}^{\uparrow}(t_{x})\right)^{-p}+\left(\widetilde{t}_{x}\right)^{-p}\right)}{\widetilde{T}_{u,u',r}^{\uparrow}(t_{x})}\le 1,~~\frac{\widetilde{T}_{u,u',r}^{\uparrow}(t_{x})}{\widehat{G}^{\uparrow,\mathsf{min}}_{u,u',r}(\bm{v};\ell)}\le 1,\nonumber\\
     &\frac{\overline{\tau}_{u,u',r}^{\uparrow}(t_{x})+\epsilon}{\widehat{J}^{\uparrow,\mathsf{min}}_{u,u',r}(\bm{v};\ell)}\le 1,~~~~\frac{\left(\mathscr{F}^{\uparrow,\mathsf{min}}_{u,u',r}(t_{x})\right)^{-1}\left(\overline{\tau}_{u,u',r}^{\uparrow}(t_{x})+\epsilon\right)}{\widetilde{\tau}_{u,u',r}^{\uparrow}(t_{x})}\le 1,\nonumber\\
     &\frac{1}{\mathscr{B}^{\uparrow,\mathsf{min}}_{u,u',r}(t_{x})}\le 1,~~\frac{1}{\mathscr{C}^{\uparrow,\mathsf{min}}_{u,u',r}(t_{x})}\le 1,~~~\frac{1}{\mathscr{F}^{\uparrow,\mathsf{min}}_{u,u',r}(t_{x})}\le 1.\nonumber
\end{tcolorbox}

We next aim to transform \eqref{app:eq:dispersion_3} into the standard GP format. In doing so, performing some algebraic manipulations on \eqref{app:eq:dispersion_3} gives us
\begin{equation}\label{app:eq:dispersion_3_4}
    \frac{\overline{\mathfrak{R}}^{\uparrow}_{u,u',r}(t_{x})\overline{\tau}_{u,u',r}^{\uparrow}(t_{x})+\overline{\tau}_{u,u',r}^{\uparrow}(t_{x})+\overline{\psi}_{u,u',r}(t_{x})\overline{M}^{\uparrow}_{u}(x)+1}{\overline{\psi}_{u,u',r}(t_{x})\alpha_{\bm{\omega}} M +\overline{\beta}^{\uparrow}_u(t_{x})\overline{\tau}_{u,u',r}^{\uparrow}(t_{x})+1}{=} 1.
\end{equation}
However, this constraint is not in the format of GP. Therefore, we transform it into the standard GP format via introducing the following three inequalities:
\begin{equation}\label{app:eq:dispersion_3_5}
     \frac{\overline{\mathfrak{R}}^{\uparrow}_{u,u',r}(t_{x})\overline{\tau}_{u,u',r}^{\uparrow}(t_{x})+\overline{\tau}_{u,u',r}^{\uparrow}(t_{x})+\overline{\psi}_{u,u',r}(t_{x})\overline{M}^{\uparrow}_{u}(x)+1}{\overline{\psi}_{u,u',r}(t_{x})\alpha_{\bm{\omega}} M +\overline{\beta}^{\uparrow}_u(t_{x})\overline{\tau}_{u,u',r}^{\uparrow}(t_{x})+1}\le 1,
\end{equation}
\begin{equation}\label{app:eq:dispersion_3_6}
    \frac{\left(\mathscr{H}^{\uparrow,\mathsf{L}}_{u,u',r}(t_{x})\right)^{-1}\left(\overline{\psi}_{u,u',r}(t_{x})\alpha_{\bm{\omega}} M +\overline{\beta}^{\uparrow}_u(t_{x})\overline{\tau}_{u,u',r}^{\uparrow}(t_{x})+1\right)}{\overline{\mathfrak{R}}^{\uparrow}_{u,u',r}(t_{x})\overline{\tau}_{u,u',r}^{\uparrow}(t_{x})+\overline{\tau}_{u,u',r}^{\uparrow}(t_{x})+\overline{\psi}_{u,u',r}(t_{x})\overline{M}^{\uparrow}_{u}(x)+1}\le 1,
\end{equation}
\begin{equation}
    \mathscr{H}^{\uparrow,\mathsf{L}}_{u,u',r}(t_{x})\ge 1,
\end{equation}
where $\mathscr{H}^{\uparrow,\mathsf{L}}_{u,u',r}(t_{x})$ is an auxiliary decision variable added with a large penalty term to the objective function to force $\mathscr{H}^{\uparrow,\mathsf{L}}_{u,u',r}(t_{x}){\rightarrow}1^+$ at the optimal point. The fractions in~\eqref{app:eq:dispersion_3_5} and \eqref{app:eq:dispersion_3_6} still need transformation since they are inequalities with posynomials in their denominator, which are not posynomials. We thus exploit arithmetic-geometric mean inequality (Lemma~\ref{Lemma:ArethmaticGeometric}) to approximate the denominators of \eqref{app:eq:dispersion_3_5} and \eqref{app:eq:dispersion_3_6} with monomials. In doing so, we approximate the denominator in \eqref{app:eq:dispersion_3_5} as follows:
\begin{align}\label{app:eq:dispersion_3_7}
    L^{\uparrow,\mathsf{L}}_{u,u',r}(\bm{v}){=}&\overline{\psi}_{u,u',r}(t_{x})\alpha_{\bm{\omega}} M +\overline{\beta}^{\uparrow}_u(t_{x})\overline{\tau}_{u,u',r}^{\uparrow}(t_{x})+1{\geq} \widehat{L}^{\uparrow,\mathsf{L}}_{u,u',r}(\bm{v};\ell) {\triangleq} \left(\frac{\overline{\psi}_{u,u',r}(t_{x})\ L^{\uparrow,\mathsf{L}}_{u,u',r}([\bm{v}]^{(\ell-1)})}{\left[\overline{\psi}_{u,u',r}(t_{x})\right]^{(\ell-1)}}\right)^{\hspace{-2mm}\frac{\left[\overline{\psi}_{u,u',r}(t_{x})\alpha_{\bm{\omega}} M\right]^{(\ell-1)}}{L^{\uparrow,\mathsf{L}}_{u,u',r}([\bm{v}]^{(\ell-1)})}}\nonumber\\
    &~~~~~~~~~~~~~~~~~~~\times\left(\frac{\overline{\beta}^{\uparrow}_u(t_{x})\overline{\tau}_{u,u',r}^{\uparrow}(t_{x}) L^{\uparrow,\mathsf{L}}_{u,u',r}([\bm{v}]^{(\ell-1)})}{\left[\overline{\beta}^{\uparrow}_u(t_{x})\overline{\tau}_{u,u',r}^{\uparrow}(t_{x})\right]^{(\ell-1)}}\right)^{\frac{\left[\overline{\beta}^{\uparrow}_u(t_{x})\overline{\tau}_{u,u',r}^{\uparrow}(t_{x})\right]^{(\ell-1)}}{L^{\uparrow,\mathsf{L}}_{u,u',r}([\bm{v}]^{(\ell-1)})}}\times\left(L^{\uparrow,\mathsf{L}}_{u,u',r}([\bm{v}]^{(\ell-1)})\right)^{\frac{1}{L^{\uparrow,\mathsf{L}}_{u,u',r}([\bm{v}]^{(\ell-1)})}},
\end{align}
which gives us the following approximation of~\eqref{app:eq:dispersion_3_5}:
\begin{equation}
     \frac{\overline{\mathfrak{R}}^{\uparrow}_{u,u',r}(t_{x})\overline{\tau}_{u,u',r}^{\uparrow}(t_{x})+\overline{\tau}_{u,u',r}^{\uparrow}(t_{x})+\overline{\psi}_{u,u',r}(t_{x})\overline{M}^{\uparrow}_{u}(x)+1}{\widehat{L}^{\uparrow,\mathsf{L}}_{u,u',r}(\bm{v};\ell)}\le 1.
\end{equation}
Similarly, we approximate the denominator in \eqref{app:eq:dispersion_3_6} as follows:
\begin{align}\label{app:eq:dispersion_3_8}
    &R^{\uparrow,\mathsf{L}}_{u,u',r}(\bm{v}){=}\overline{\mathfrak{R}}^{\uparrow}_{u,u',r}(t_{x})\overline{\tau}_{u,u',r}^{\uparrow}(t_{x})+\overline{\tau}_{u,u',r}^{\uparrow}(t_{x})+\overline{\psi}_{u,u',r}(t_{x})\overline{M}^{\uparrow}_{u}(x)+1\nonumber\\
    &{\geq} \widehat{R}^{\uparrow,\mathsf{L}}_{u,u',r}(\bm{v};\ell) {\triangleq} \left(\frac{\overline{\mathfrak{R}}^{\uparrow}_{u,u',r}(t_{x})\overline{\tau}_{u,u',r}^{\uparrow}(t_{x}) R^{\uparrow,\mathsf{L}}_{u,u',r}([\bm{v}]^{(\ell-1)})}{\left[\overline{\mathfrak{R}}^{\uparrow}_{u,u',r}(t_{x})\overline{\tau}_{u,u',r}^{\uparrow}(t_{x})\right]^{(\ell-1)}}\right)^{\hspace{-2mm}\frac{\left[\overline{\mathfrak{R}}^{\uparrow}_{u,u',r}(t_{x})\overline{\tau}_{u,u',r}^{\uparrow}(t_{x})\right]^{(\ell-1)}}{R^{\uparrow,\mathsf{L}}_{u,u',r}([\bm{v}]^{(\ell-1)})}}\times\left(\frac{\overline{\tau}_{u,u',r}^{\uparrow}(t_{x}) R^{\uparrow,\mathsf{L}}_{u,u',r}([\bm{v}]^{(\ell-1)})}{\left[\overline{\tau}_{u,u',r}^{\uparrow}(t_{x})\right]^{(\ell-1)}}\right)^{\frac{\left[\overline{\tau}_{u,u',r}^{\uparrow}(t_{x})\right]^{(\ell-1)}}{R^{\uparrow,\mathsf{L}}_{u,u',r}([\bm{v}]^{(\ell-1)})}}\nonumber\\
    &~~~~~~~~~~~~~~\times\left(\frac{\overline{\psi}_{u,u',r}(t_{x})\overline{M}^{\uparrow}_{u}(x) R^{\uparrow,\mathsf{L}}_{u,u',r}([\bm{v}]^{(\ell-1)})}{\left[\overline{\psi}_{u,u',r}(t_{x})\overline{M}^{\uparrow}_{u}(x)\right]^{(\ell-1)}}\right)^{\frac{\left[\overline{\psi}_{u,u',r}(t_{x})\overline{M}^{\uparrow}_{u}(x)\right]^{(\ell-1)}}{R^{\uparrow,\mathsf{L}}_{u,u',r}([\bm{v}]^{(\ell-1)})}}\times\left(R^{\uparrow,\mathsf{L}}_{u,u',r}([\bm{v}]^{(\ell-1)})\right)^{\frac{1}{R^{\uparrow,\mathsf{L}}_{u,u',r}([\bm{v}]^{(\ell-1)})}},
\end{align}
which gives the following approximation of~\eqref{app:eq:dispersion_3_6}:
\begin{equation}
     \frac{\left(\mathscr{H}^{\uparrow,\mathsf{L}}_{u,u',r}(t_{x})\right)^{-1}\left(\overline{\psi}_{u,u',r}(t_{x})\alpha_{\bm{\omega}} M +\overline{\beta}^{\uparrow}_u(t_{x})\overline{\tau}_{u,u',r}^{\uparrow}(t_{x})+1\right)}{ \widehat{R}^{\uparrow,\mathsf{L}}_{u,u',r}(\bm{v};\ell)}\le 1.
\end{equation}
We finally approximate \eqref{app:eq:dispersion_3} as follows:
\begin{tcolorbox}[ams align]
     & \frac{\overline{\mathfrak{R}}^{\uparrow}_{u,u',r}(t_{x})\overline{\tau}_{u,u',r}^{\uparrow}(t_{x})+\overline{\tau}_{u,u',r}^{\uparrow}(t_{x})+\overline{\psi}_{u,u',r}(t_{x})\overline{M}^{\uparrow}_{u}(x)+1}{\widehat{L}^{\uparrow,\mathsf{L}}_{u,u',r}(\bm{v};\ell)}\le 1,\nonumber\\
     & \frac{\left(\mathscr{H}^{\uparrow,\mathsf{L}}_{u,u',r}(t_{x})\right)^{-1}\left(\overline{\psi}_{u,u',r}(t_{x})\alpha_{\bm{\omega}} M +\overline{\beta}^{\uparrow}_u(t_{x})\overline{\tau}_{u,u',r}^{\uparrow}(t_{x})+1\right)}{ \widehat{R}^{\uparrow,\mathsf{L}}_{u,u',r}(\bm{v};\ell)}\le 1,\nonumber\\
     &\frac{1}{\mathscr{H}^{\uparrow,\mathsf{L}}_{u,u',r}(t_{x})}\le 1.\nonumber
\end{tcolorbox}

We next aim to transform \eqref{app:eq:dispersion_4_1} into GP format. To this end, we rewrite \eqref{app:eq:dispersion_4_1} as follows:
\begin{equation}\label{app:eq:dispersion_4_1_1}
   \frac{\overline{M}^{\uparrow}_{u}(x)+1}{\displaystyle\sum_{z{=}0}^{x}\sum_{u'{\in}\mathcal{U}_{b}}\sum_{r'{\in}\overline{\mathcal{R}}_{b}}T_{u,u',r'}^{\uparrow}(t_{z})\overline{\mathfrak{R}}^{\uparrow}_{u,u',r'}(t_{z})+1}{=} 1.
\end{equation} 
This constraint is not in the format of GP. Therefore, we transform it through splitting it into the following three inequalities:
\begin{equation}\label{app:eq:dispersion_4_1_4}
    \frac{\overline{M}^{\uparrow}_{u}(x)+1}{\displaystyle\sum_{z{=}0}^{x}\sum_{u'{\in}\mathcal{U}_{b}}\sum_{r'{\in}\overline{\mathcal{R}}_{b}}T_{u,u',r'}^{\uparrow}(t_{z})\overline{\mathfrak{R}}^{\uparrow}_{u,u',r'}(t_{z})+1}\le 1,
\end{equation}
\begin{equation}\label{app:eq:dispersion_4_1_5}
    \frac{\left(\overline{\mathscr{L}}^{\uparrow,\mathsf{M}}_{u}(t_{x})\right)^{-1}\left(\displaystyle\sum_{z{=}0}^{x}\sum_{u'{\in}\mathcal{U}_{b}}\sum_{r'{\in}\overline{\mathcal{R}}_{b}}T_{u,u',r'}^{\uparrow}(t_{z})\overline{\mathfrak{R}}^{\uparrow}_{u,u',r'}(t_{z})+1\right)}{\overline{M}^{\uparrow}_{u}(x)+1}\le 1,
\end{equation}
\begin{equation}
    \overline{\mathscr{L}}^{\uparrow,\mathsf{M}}_{u}(t_{x})\ge 1,
\end{equation}
where $\overline{\mathscr{L}}^{\uparrow,\mathsf{M}}_{u}(t_{x})$ is an auxiliary decision variable added with a large penalty term to the objective function to force $\overline{\mathscr{L}}^{\uparrow,\mathsf{M}}_{u}(t_{x}){\rightarrow}1^+$ at the optimal point. The fractions in~\eqref{app:eq:dispersion_4_1_4} and \eqref{app:eq:dispersion_4_1_5} still need transformation since they are inequalities with posynomials in their denominators, which are not posynomials. We thus exploit arithmetic-geometric mean inequality (Lemma~\ref{Lemma:ArethmaticGeometric}) to approximate the denominators in \eqref{app:eq:dispersion_4_1_4} and \eqref{app:eq:dispersion_4_1_5} with monomials. To this end, we approximate the denominator in \eqref{app:eq:dispersion_4_1_4} as follows:
\begin{align}\label{app:eq:dispersion_4_1_6}
   S^{\uparrow,\mathsf{M}}_{u}(\bm{v})&=\sum_{z{=}0}^{x}\sum_{u'{\in}\mathcal{U}_{b}}\sum_{r'{\in}\overline{\mathcal{R}}_{b}}T_{u,u',r'}^{\uparrow}(t_{z})\overline{\mathfrak{R}}^{\uparrow}_{u,u',r'}(t_{z})+1\nonumber\\
   &\geq \widehat{S}^{\uparrow,\mathsf{M}}_{u}(\bm{v};\ell) \triangleq \prod_{z{=}0}^{x}\prod_{u'{\in}\mathcal{U}_{b}}\prod_{r'{\in}\overline{\mathcal{R}}_{b}}\left(\frac{T_{u,u',r'}^{\uparrow}(t_{z})\overline{\mathfrak{R}}^{\uparrow}_{u,u',r'}(t_{z}) S^{\uparrow,\mathsf{M}}_{u}([\bm{v}]^{(\ell-1)})}{\left[T_{u,u',r'}^{\uparrow}(t_{z})\overline{\mathfrak{R}}^{\uparrow}_{u,u',r'}(t_{z})\right]^{(\ell-1)}}\right)^{\frac{\left[T_{u,u',r'}^{\uparrow}(t_{z})\overline{\mathfrak{R}}^{\uparrow}_{u,u',r'}(t_{z})\right]^{(\ell-1)}}{S^{\uparrow,\mathsf{M}}_{u}([\bm{v}]^{(\ell-1)})}}\nonumber\\
   &~~~~~~~~~~~~~~~\times\left(S^{\uparrow,\mathsf{M}}_{u}([\bm{v}]^{(\ell-1)})\right)^{\frac{1}{S^{\uparrow,\mathsf{M}}_{u}([\bm{v}]^{(\ell-1)})}},
\end{align}
which gives us an approximation of~\eqref{app:eq:dispersion_4_1_4} as follows:
\begin{equation}\label{app:eq:dispersion_4_1_7}
    \frac{\overline{M}^{\uparrow}_{u}(x)+1}{\widehat{S}^{\uparrow,\mathsf{M}}_{u}(\bm{v};\ell)}\le 1.
\end{equation}

Similarly, we approximate the denominator in \eqref{app:eq:dispersion_4_1_5} as follows:
\begin{align}\label{app:eq:dispersion_4_1_8}
    &Q^{\uparrow,\mathsf{M}}_{u}(\bm{v}){=}\overline{M}^{\uparrow}_{u}(x)+1{\geq} \widehat{Q}^{\uparrow,\mathsf{M}}_{u}(\bm{v};\ell) {\triangleq} \left(\frac{\overline{M}^{\uparrow}_{u}(x) Q^{\uparrow,\mathsf{M}}_{u}([\bm{v}]^{(\ell-1)})}{\left[\overline{M}^{\uparrow}_{u}(x)\right]^{(\ell-1)}}\right)^{\hspace{-2mm}\frac{\left[\overline{M}^{\uparrow}_{u}(x)\right]^{(\ell-1)}}{Q^{\uparrow,\mathsf{M}}_{u}([\bm{v}]^{(\ell-1)})}}\times\left(Q^{\uparrow,\mathsf{M}}_{u}([\bm{v}]^{(\ell-1)})\right)^{\frac{1}{Q^{\uparrow,\mathsf{M}}_{u}([\bm{v}]^{(\ell-1)})}},
\end{align}
which gives the following approximation of~\eqref{app:eq:broadcast_latency_6}:
\begin{equation}
     \frac{\left(\overline{\mathscr{L}}^{\uparrow,\mathsf{M}}_{u}(t_{x})\right)^{-1}\left(\displaystyle\sum_{z{=}0}^{x}\sum_{u'{\in}\mathcal{U}_{b}}\sum_{r'{\in}\overline{\mathcal{R}}_{b}}T_{u,u',r'}^{\uparrow}(t_{z})\overline{\mathfrak{R}}^{\uparrow}_{u,u',r'}(t_{z})+1\right)}{\widehat{Q}^{\uparrow,\mathsf{M}}_{u}(\bm{v};\ell)}\le 1.
\end{equation}

We finally approximate constraint~\eqref{app:eq:dispersion_4} as follows:
\begin{tcolorbox}[ams align]
     &\frac{\overline{M}^{\uparrow}_{u}(x)+1}{\widehat{S}^{\uparrow,\mathsf{M}}_{u}(\bm{v};\ell)}\le 1,~~~~\frac{\left(\overline{\mathscr{L}}^{\uparrow,\mathsf{M}}_{u}(t_{x})\right)^{-1}\left(\displaystyle\sum_{z{=}0}^{x}\sum_{u'{\in}\mathcal{U}_{b}}\sum_{r'{\in}\overline{\mathcal{R}}_{b}}T_{u,u',r'}^{\uparrow}(t_{z})\overline{\mathfrak{R}}^{\uparrow}_{u,u',r'}(t_{z})+1\right)}{\widehat{Q}^{\uparrow,\mathsf{M}}_{u}(\bm{v};\ell)}\le 1,\nonumber\\
     &\frac{1}{\overline{\mathscr{L}}^{\uparrow,\mathsf{M}}_{u}(t_{x})}\le 1,\nonumber
\end{tcolorbox}

We next aim to transform \eqref{app:eq:dispersion_2_1} into the standard form of GP. In doing so, we rewrite \eqref{app:eq:dispersion_2_1} as follows:
\begin{equation}\label{app:eq:dispersion_2_2}
\begin{aligned}
   \overline{\tau}_{u}^{\uparrow,(k)}&=\sum_{x{=}0}^{N}\max_{\substack{u'{\in}\mathcal{U}_{b},\\r{\in}\overline{\mathcal{R}}_{b}}}\Big\{T_{u,u',r}^{\uparrow}(t_{x})+\epsilon\Big\}-\sum_{x{=}0}^{N}\epsilon\\
   &=\sum_{x{=}0}^{N}\max_{\substack{u'{\in}\mathcal{U}_{b},\\r{\in}\overline{\mathcal{R}}_{b}}}\Big\{T_{u,u',r}^{\uparrow}(t_{x})+\epsilon\Big\}-N\epsilon.
\end{aligned}
\end{equation}
Performing some algebraic manipulations gives us
\begin{equation}\label{app:eq:dispersion_2_2}
   \frac{\overline{\tau}_{u}^{\uparrow,(k)}+N^{(k)}\epsilon}{\displaystyle\sum_{x{=}0}^{N}\max_{\substack{u'{\in}\mathcal{U}_{b},\\r{\in}\overline{\mathcal{R}}_{b}}}\Big\{T_{u,u',r}^{\uparrow}(t_{x})+\epsilon\Big\}+N^{(k{-}1)}\epsilon}=1.
\end{equation} 
Using the approximation $\max\{A, B\}\approx (A^{p}+B^{p})^{-\frac{1}{p}}$, which is tight when $p \gg 1$, gives us
\begin{equation}\label{app:eq:dispersion_2_3}
   \frac{\overline{\tau}_{u}^{\uparrow,(k)}+N^{(k)}\epsilon}{\displaystyle\sum_{x{=}0}^{N}\left(\sum_{u'{\in}\mathcal{U}_{b}}\sum_{r{\in}\overline{\mathcal{R}}_{b}}\left(T_{u,u',r}^{\uparrow}(t_{x})+\epsilon\right)^{p}\right)^{-\frac{1}{p}}+N^{(k{-}1)}\epsilon}=1.
\end{equation} 
Defining an auxiliary decision variable $\overline{T}_{u}^{\uparrow,\mathsf{max}}(t_{x})$ that satisfies
\begin{equation}\label{app:eq:dispersion_2_4}
    \overline{T}_{u}^{\uparrow,\mathsf{max}}(t_{x}) =\sum_{u'{\in}\mathcal{U}_{b}}\sum_{r{\in}\overline{\mathcal{R}}_{b}}\left(T_{u,u',r}^{\uparrow}(t_{x})+\epsilon\right)^{p}
\end{equation}
results in
\begin{equation}\label{app:eq:dispersion_2_5}
   \frac{\overline{\tau}_{u}^{\uparrow,(k)}+N^{(k)}\epsilon}{\displaystyle\sum_{x{=}0}^{N}\left(\overline{T}_{u}^{\uparrow,\mathsf{max}}(t_{x})\right)^{-\frac{1}{p}}+N^{(k{-}1)}\epsilon}=1.
\end{equation} 
This constraint is not in the format of GP. Therefore, we transform it through splitting it into the following three inequalities:
\begin{equation}\label{app:eq:dispersion_2_6}
    \frac{\overline{\tau}_{u}^{\uparrow,(k)}+N^{(k)}\epsilon}{\displaystyle\sum_{x{=}0}^{N}\left(\overline{T}_{u}^{\uparrow,\mathsf{max}}(t_{x})\right)^{-\frac{1}{p}}+N^{(k{-}1)}\epsilon}\le 1,
\end{equation}

\begin{equation}\label{app:eq:dispersion_2_7}
    \frac{\left(\overline{\mathscr{Q}}^{\uparrow}_{u}(t_{x})\right)^{-1}\left(\displaystyle\sum_{x{=}0}^{N}\left(\overline{T}_{u}^{\uparrow,\mathsf{max}}(t_{x})\right)^{-\frac{1}{p}}+N^{(k{-}1)}\epsilon\right)}{\overline{\tau}_{u}^{\uparrow,(k)}+N^{(k)}\epsilon}\le 1,
\end{equation}

\begin{equation}
    \overline{\mathscr{Q}}^{\uparrow}_{u}(t_{x})\ge 1,
\end{equation}
where $\overline{\mathscr{Q}}^{\uparrow}_{u}(t_{x})$ is an auxiliary decision variable added with a large penalty term to the objective function to force $\overline{\mathscr{Q}}^{\uparrow}_{u}(t_{x}){\rightarrow}1^+$ at the optimal point. The fractions in~\eqref{app:eq:dispersion_2_6} and \eqref{app:eq:dispersion_2_7} still need transformation since they are inequalities with posynomials in their denominators, which are not posynomials. We thus exploit arithmetic-geometric mean inequality (Lemma~\ref{Lemma:ArethmaticGeometric}) to approximate the denominators in \eqref{app:eq:dispersion_2_6} and \eqref{app:eq:dispersion_2_7} with monomials. To this end, we approximate the denominator in \eqref{app:eq:dispersion_2_6} as follows:
\begin{align}\label{app:eq:dispersion_2_8}
   &V^{\uparrow}_{u}(\bm{v})=\sum_{x{=}0}^{N}\left(\overline{T}_{u}^{\uparrow,\mathsf{max}}(t_{x})\right)^{-\frac{1}{p}}+N^{(k{-}1)}\epsilon\nonumber\\
   &\geq \widehat{V}^{\uparrow}_{u}(\bm{v};\ell) \triangleq \prod_{x{=}0}^{N}\left(\frac{\left(\overline{T}_{u}^{\uparrow,\mathsf{max}}(t_{x})\right)^{-\frac{1}{p}} V^{\uparrow}_{u}([\bm{v}]^{(\ell-1)})}{\left[\left(\overline{T}_{u}^{\uparrow,\mathsf{max}}(t_{x})\right)^{-\frac{1}{p}}\right]^{(\ell-1)}}\right)^{\frac{\left[\left(\overline{T}_{u}^{\uparrow,\mathsf{max}}(t_{x})\right)^{-\frac{1}{p}}\right]^{(\ell-1)}}{V^{\uparrow}_{u}([\bm{v}]^{(\ell-1)})}}\times\left(V^{\uparrow}_{u}([\bm{v}]^{(\ell-1)})\right)^{\frac{\left[N^{(k{-}1)}\displaystyle\epsilon\right]^{(\ell-1)}}{V^{\uparrow}_{u}([\bm{v}]^{(\ell-1)})}},
\end{align}
which gives us an approximation of~\eqref{app:eq:dispersion_2_6} as follows:
\begin{equation}\label{app:eq:dispersion_2_9}
    \frac{\overline{\tau}_{u}^{\uparrow,(k)}+N^{(k)}\epsilon}{\widehat{V}^{\uparrow}_{u}(\bm{v};\ell)}\le 1.
\end{equation}

Similarly, we approximate the denominator in \eqref{app:eq:dispersion_2_7} as follows:
\begin{align}\label{app:eq:dispersion_2_10}
    W^{\uparrow}_{u}(\bm{v})&=\overline{\tau}_{u}^{\uparrow,(k)}+N^{(k)}\epsilon \geq \widehat{W}^{\uparrow}_{u}(\bm{v};\ell) \triangleq \left(\frac{\overline{\tau}_{u}^{\uparrow,(k)} W^{\uparrow}_{u}([\bm{v}]^{(\ell-1)})}{\left[\overline{\tau}_{u}^{\uparrow,(k)}\right]^{(\ell-1)}}\right)^{\frac{\left[\overline{\tau}_{u}^{\uparrow,(k)}\right]^{(\ell-1)}}{W^{\uparrow}_{u}([\bm{v}]^{(\ell-1)})}}\times\left(W^{\uparrow}_{u}([\bm{v}]^{(\ell-1)})\right)^{\frac{\left[N^{(k)}\displaystyle\epsilon\right]^{(\ell-1)}}{W^{\uparrow}_{u}([\bm{v}]^{(\ell-1)})}},
\end{align}
which gives us an approximation of~\eqref{app:eq:dispersion_2_7} as follows:
\begin{equation}\label{app:eq:dispersion_2_11}
    \frac{\left(\overline{\mathscr{Q}}^{\uparrow}_{u}(t_{x})\right)^{-1}\left(\displaystyle\sum_{x{=}0}^{N}\left(\overline{T}_{u}^{\uparrow,\mathsf{max}}(t_{x})\right)^{-\frac{1}{p}}+N^{(k{-}1)}\epsilon\right)}{\widehat{W}^{\uparrow}_{u}(\bm{v};\ell)}\le 1.
\end{equation}

We next transform \eqref{app:eq:dispersion_2_4} into the standard form of GP. To do so, we rewrite \eqref{app:eq:dispersion_2_4} as follows:
\begin{equation}\label{app:eq:dispersion_2_4_1}
    \frac{\overline{T}_{u}^{\uparrow,\mathsf{max}}(t_{x})}{\sum_{u'{\in}\mathcal{U}_{b}}\sum_{r{\in}\overline{\mathcal{R}}_{b}}\left(T_{u,u',r}^{\uparrow}(t_{x})+\epsilon\right)^{p}} =1.
\end{equation}
Defining an auxiliary decision variable $\overline{T}_{u,u',r}^{\uparrow,\epsilon}(t_{x})$ satisfying
\begin{equation}\label{app:eq:dispersion_2_4_2}
    \overline{T}_{u,u',r}^{\uparrow,\epsilon}(t_{x})=T_{u,u',r}^{\uparrow}(t_{x})+\epsilon
\end{equation}
leading to
\begin{equation}\label{app:eq:dispersion_2_4_3}
    \frac{\overline{T}_{u}^{\uparrow,\mathsf{max}}(t_{x})}{\sum_{u'{\in}\mathcal{U}_{b}}\sum_{r{\in}\overline{\mathcal{R}}_{b}}\left(\overline{T}_{u,u',r}^{\uparrow,\epsilon}(t_{x})\right)^{p}} =1.
\end{equation}
This constraint is not in the format of GP. Therefore, we transform it through splitting it into the following three inequalities:
\begin{equation}\label{app:eq:dispersion_2_4_4}
    \frac{\overline{T}_{u}^{\uparrow,\mathsf{max}}(t_{x})}{\sum_{u'{\in}\mathcal{U}_{b}}\sum_{r{\in}\overline{\mathcal{R}}_{b}}\left(\overline{T}_{u,u',r}^{\uparrow,\epsilon}(t_{x})\right)^{p}}\le 1,
\end{equation}

\begin{equation}\label{app:eq:dispersion_2_4_5}
    \frac{\left(\overline{\mathscr{R}}^{\uparrow,\mathsf{max}}_{u}(t_{x})\right)^{-1}\left(\sum_{u'{\in}\mathcal{U}_{b}}\sum_{r{\in}\overline{\mathcal{R}}_{b}}\left(\overline{T}_{u,u',r}^{\uparrow,\epsilon}(t_{x})\right)^{p}\right)}{\overline{T}_{u}^{\uparrow,\mathsf{max}}(t_{x})}\le 1,
\end{equation}

\begin{equation}
    \overline{\mathscr{R}}^{\uparrow,\mathsf{max}}_{u}(t_{x})\ge 1,
\end{equation}
where $\overline{\mathscr{R}}^{\uparrow,\mathsf{max}}_{u}(t_{x})$ is an auxiliary decision variable added with a large penalty term to the objective function to force $\overline{\mathscr{R}}^{\uparrow,\mathsf{max}}_{u}(t_{x}){\rightarrow}1^+$ at the optimal point. The fraction in~\eqref{app:eq:dispersion_2_4_4} still needs transformation since it is inequality with a posynomial in its denominator, which is not a posynomial. We thus exploit arithmetic-geometric mean inequality (Lemma~\ref{Lemma:ArethmaticGeometric}) to approximate the denominator in \eqref{app:eq:dispersion_2_4_4} with a monomial.
\begin{align}\label{app:eq:dispersion_2_4_6}
    &Y^{\uparrow,\mathsf{max}}_{u}(\bm{v})=\sum_{u'{\in}\mathcal{U}_{b}}\sum_{r{\in}\overline{\mathcal{R}}_{b}}\left(\overline{T}_{u,u',r}^{\uparrow,\epsilon}(t_{x})\right)^{p}\geq \widehat{Y}^{\uparrow,\mathsf{max}}_{u}(\bm{v};\ell) \triangleq \prod_{u'{\in}\mathcal{U}_{b}}\prod_{r{\in}\overline{\mathcal{R}}_{b}}\left(\frac{\left(\overline{T}_{u,u',r}^{\uparrow,\epsilon}(t_{x})\right)^{p} Y^{\uparrow,\mathsf{max}}_{u}([\bm{v}]^{(\ell-1)})}{\left[\left(\overline{T}_{u,u',r}^{\uparrow,\epsilon}(t_{x})\right)^{p}\right]^{(\ell-1)}}\right)^{\frac{\left[\left(\overline{T}_{u,u',r}^{\uparrow,\epsilon}(t_{x})\right)^{p}\right]^{(\ell-1)}}{Y^{\uparrow,\mathsf{max}}_{u}([\bm{v}]^{(\ell-1)})}},
\end{align}
which gives us an approximation of~\eqref{app:eq:dispersion_2_4_4} as follows:
\begin{equation}\label{app:eq:dispersion_2_4_6}
    \frac{\overline{T}_{u}^{\uparrow,\mathsf{max}}(t_{x})}{\widehat{Y}^{\uparrow,\mathsf{max}}_{u}(\bm{v};\ell)}\le 1.
\end{equation}

We next transform \eqref{app:eq:dispersion_2_4_2} into GP format. To this end, we rewrite \eqref{app:eq:dispersion_2_4_2} as follows:
\begin{equation}\label{app:eq:dispersion_2_4_7}
    \frac{\overline{T}_{u,u',r}^{\uparrow,\epsilon}(t_{x})}{T_{u,u',r}^{\uparrow}(t_{x})+\epsilon}=1.
\end{equation}
This constraint is not in the format of GP. Therefore, we transform it through splitting it into the following three inequalities:
\begin{equation}\label{app:eq:dispersion_2_4_4}
    \frac{\overline{T}_{u,u',r}^{\uparrow,\epsilon}(t_{x})}{T_{u,u',r}^{\uparrow}(t_{x})+\epsilon}\le 1,
\end{equation}

\begin{equation}\label{app:eq:dispersion_2_4_5}
    \frac{\left(\overline{\mathscr{S}}^{\uparrow,\epsilon}_{u,u',r}(t_{x})\right)^{-1}\left(T_{u,u',r}^{\uparrow}(t_{x})+\epsilon\right)}{\overline{T}_{u,u',r}^{\uparrow,\epsilon}(t_{x})}\le 1,
\end{equation}

\begin{equation}
    \overline{\mathscr{S}}^{\uparrow,\epsilon}_{u,u',r}(t_{x})\ge 1,
\end{equation}
where $\overline{\mathscr{S}}^{\uparrow,\epsilon}_{u,u',r}(t_{x})$ is an auxiliary decision variable added with a large penalty term to the objective function to force $\overline{\mathscr{S}}^{\uparrow,\epsilon}_{u,u',r}(t_{x}){\rightarrow}1^+$ at the optimal point. The fraction in~\eqref{app:eq:dispersion_2_4_4} still needs transformation since it is inequality with a posynomial in its denominator, which is not a posynomial. We thus exploit arithmetic-geometric mean inequality (Lemma~\ref{Lemma:ArethmaticGeometric}) to approximate the denominator in \eqref{app:eq:dispersion_2_4_4} with a monomial.
\begin{align}\label{app:eq:dispersion_2_4_6}
    Z^{\uparrow,\epsilon}_{u,u',r}(\bm{v})&=T_{u,u',r}^{\uparrow}(t_{x})+\epsilon \geq \widehat{Z}^{\uparrow,\epsilon}_{u,u',r}(\bm{v};\ell) \triangleq \left(\frac{T_{u,u',r}^{\uparrow}(t_{x}) Z^{\uparrow,\epsilon}_{u,u',r}([\bm{v}]^{(\ell-1)})}{\left[T_{u,u',r}^{\uparrow}(t_{x})\right]^{(\ell-1)}}\right)^{\frac{\left[T_{u,u',r}^{\uparrow}(t_{x})\right]^{(\ell-1)}}{Z^{\uparrow,\epsilon}_{u,u',r}([\bm{v}]^{(\ell-1)})}}\times\left(Z^{\uparrow,\epsilon}_{u,u',r}([\bm{v}]^{(\ell-1)})\right)^{\frac{\left[\displaystyle\epsilon\right]^{(\ell-1)}}{Z^{\uparrow,\epsilon}_{u,u',r}([\bm{v}]^{(\ell-1)})}},
\end{align}
which gives us an approximation of~\eqref{app:eq:dispersion_2_4_4} as follows:
\begin{equation}\label{app:eq:dispersion_2_4_7}
    \frac{\overline{T}_{u,u',r}^{\uparrow,\epsilon}(t_{x})}{\widehat{Z}^{\uparrow,\epsilon}_{u,u',r}(\bm{v};\ell)}\le 1.
\end{equation}
We finally approximate constraint~\eqref{app:eq:dispersion_2} as follows:
\begin{tcolorbox}[ams align]
     &\frac{\overline{\tau}_{u}^{\uparrow,(k)}+N^{(k)}\epsilon}{\widehat{V}^{\uparrow}_{u}(\bm{v};\ell)}\le 1,~~~~~~\frac{\left(\overline{\mathscr{Q}}^{\uparrow}_{u}(t_{x})\right)^{-1}\left(\displaystyle\sum_{x{=}0}^{N}\left(\overline{T}_{u}^{\uparrow,\mathsf{max}}(t_{x})\right)^{-\frac{1}{p}}+N^{(k{-}1)}\epsilon\right)}{\widehat{W}^{\uparrow}_{u}(\bm{v};\ell)}\le 1,\nonumber\\
     &\frac{\overline{T}_{u}^{\uparrow,\mathsf{max}}(t_{x})}{\widehat{Y}^{\uparrow,\mathsf{max}}_{u}(\bm{v};\ell)}\le 1,~~~~\frac{\left(\overline{\mathscr{R}}^{\uparrow,\mathsf{max}}_{u}(t_{x})\right)^{-1}\left(\sum_{u'{\in}\mathcal{U}_{b}}\sum_{r{\in}\overline{\mathcal{R}}_{b}}\left(\overline{T}_{u,u',r}^{\uparrow,\epsilon}(t_{x})\right)^{p}\right)}{\overline{T}_{u}^{\uparrow,\mathsf{max}}(t_{x})}\le 1,\nonumber\\
     & \frac{\overline{T}_{u,u',r}^{\uparrow,\epsilon}(t_{x})}{\widehat{Z}^{\uparrow,\epsilon}_{u,u',r}(\bm{v};\ell)}\le 1,~~~\frac{\left(\overline{\mathscr{S}}^{\uparrow,\epsilon}_{u,u',r}(t_{x})\right)^{-1}\left(T_{u,u',r}^{\uparrow}(t_{x})+\epsilon\right)}{\overline{T}_{u,u',r}^{\uparrow,\epsilon}(t_{x})}\le 1,\nonumber\\
     &\frac{1}{\overline{\mathscr{Q}}^{\uparrow}_{u}(t_{x})}\le 1,~~~~~\frac{1}{\overline{\mathscr{R}}^{\uparrow,\mathsf{max}}_{u}(t_{x})}\le 1,~~~~\frac{1}{\overline{\mathscr{S}}^{\uparrow,\epsilon}_{u,u',r}(t_{x})}\le 1.\nonumber
\end{tcolorbox}

\noindent\textbf{GPs dispatching latency of CHUs (Proposition~\ref{propo:CHU_uplink}).} Referring to Sec.~\ref{sec:communication_latency}, let us revisit the equations of GPs dispatching latency of DPUs below:
\begin{equation}\label{app:eq:dispatch_1}
    \sum_{r {\in} \mathcal{R}_{b}} \psi_{u,r}(t_{x})=1,~x\in\mathcal{N}^{(k)},
\end{equation}

\begin{equation}\label{app:eq:dispatch_2}
     \tau_{u}^{\uparrow,{(k)}}{=}\hspace{-3mm}\sum_{x{=}N^{(k{-}1)}{+}1}^{N^{(k)}}\max_{r{\in}\mathcal{R}_{b}}\Big\{\hspace{-1mm}\min\left\{\tau_{u,r}^{\uparrow}\hspace{-0.3mm}(t_{x}),\left(t_{x{+}1}-t_{x}\right)\right\}\hspace{-1mm}\Big\}.
\end{equation}
In \eqref{app:eq:dispatch_2}, $\tau_{u,r}^{\uparrow}(t_{x})$ is given by
\begin{equation}\label{app:eq:dispatch_3}
\tau_{u,r}^{\uparrow}(t_{x}){=}\frac{\big(\alpha_{\bm{\omega}} M{-}M^{\uparrow}_{u}(x)\big)\psi_{u,r}(t_{x})}{\mathfrak{R}^{\uparrow}_{u,r}(t_{x})+1{-}\beta^{\uparrow}_u(t_{x})},
\end{equation}
where
\begin{equation}\label{app:eq:dispatch_4}
    M^{\uparrow}_{u}(x){=}\sum_{z{=}N^{(k{-}1)}{+}1}^{x}\sum_{r'{\in}\mathcal{R}_{b}}\min\big\{\tau_{u,r'}^{\uparrow}(t_{z}),\left(t_{z{+}1}{-}t_{z}\right)\big\}\mathfrak{R}^{\uparrow}_{u,r'}(t_{z}).
\end{equation}

In the following, we first transform \eqref{app:eq:dispatch_1} into the standard form of GP. In doing so, we split \eqref{app:eq:dispatch_1} into the following three inequalities:
\begin{equation}\label{app:eq:dispatch_1_2}
    \sum_{r {\in} \mathcal{R}_{b}} \psi_{u,r}(t_{x})\le 1,
\end{equation}
\begin{equation}\label{app:eq:dispatch_1_3}
    \frac{\left(\mathscr{B}_{u,r}(t_{x})\right)^{-1}}{\sum_{r {\in} \mathcal{R}_{b}} \psi_{u,r}(t_{x})}\le 1,
\end{equation}
\begin{equation}
    \mathscr{B}_{u,r}(t_{x})\ge 1,
\end{equation}
where $\mathscr{B}_{u,r}(t_{x})$ is an auxiliary decision variable added with a large penalty term to the objective function to force $\mathscr{B}_{u,r}(t_{x}){\rightarrow}1^+$ at the optimal point. The fraction in~\eqref{app:eq:dispatch_1_3} still needs transformation since it is inequality with a posynomial in its denominator, which is not a posynomial. We thus exploit arithmetic-geometric mean inequality (Lemma~\ref{Lemma:ArethmaticGeometric}) to approximate the denominator in \eqref{app:eq:dispatch_1_3} with a monomial.
\begin{align}\label{app:eq:dispatch_1_3_1}
    D^{\mathsf{C}}_{u,r}(\bm{v})&=\sum_{r {\in} \mathcal{R}_{b}} \psi_{u,r}(t_{x}) \geq \widehat{D}^{\mathsf{C}}_{u,r}(\bm{v};\ell) \triangleq \prod_{r {\in} \mathcal{R}_{b}} \left(\frac{\psi_{u,r}(t_{x}) D^{\mathsf{C}}_{u,r}([\bm{v}]^{(\ell-1)})}{\left[\psi_{u,r}(t_{x})\right]^{(\ell-1)}}\right)^{\frac{\left[\psi_{u,r}(t_{x})\right]^{(\ell-1)}}{D^{\mathsf{C}}_{u,r}([\bm{v}]^{(\ell-1)})}},
\end{align}
which gives us an approximation of~\eqref{app:eq:dispatch_1_3} as follows:
\begin{equation}\label{app:eq:dispatch_1_3_3}
    \frac{\left(\mathscr{B}_{u,r}(t_{x})\right)^{-1}}{\widehat{D}^{\mathsf{C}}_{u,r}(\bm{v};\ell)}\le 1.
\end{equation}
We finally approximate constraint~\eqref{app:eq:dispatch_1} as follows:
\begin{tcolorbox}[ams align]
     &\sum_{r {\in} \mathcal{R}_{b}} \psi_{u,r}(t_{x})\le 1,~~~~~~~\frac{\left(\mathscr{B}_{u,r}(t_{x})\right)^{-1}}{\widehat{D}^{\mathsf{C}}_{u,r}(\bm{v};\ell)}\le 1,~~~~~~~\frac{1}{\mathscr{B}_{u,r}(t_{x})}\le 1.\nonumber
\end{tcolorbox}

We next aim to transform \eqref{app:eq:dispatch_2} into the standard form of GP. To this end, we define $\tau_{u}^{\uparrow,{(k)}}$, $\tau_{u,r}^{\uparrow}(t_{x})$, and $M^{\uparrow}_{u}(x)$ as auxiliary decision variables that must satisfy constraints \eqref{app:eq:dispatch_2}, \eqref{app:eq:dispatch_3}, and \eqref{app:eq:dispatch_4}, respectively. Furthermore, we define an auxiliary decision variable $T_{u,r}^{\uparrow}(t_{x})$ satisfying the following constraint:
\begin{equation}\label{app:eq:dispatch_min_1}
    T_{u,r}^{\uparrow}(t_{x})=\min\big\{\tau_{u,r}^{\uparrow}(t_{x}),\left(t_{x{+}1}{-}t_{x}\right)\big\}.
\end{equation}
Accordingly, \eqref{app:eq:dispatch_2} and \eqref{app:eq:dispatch_4} can be rewritten as follows:
\begin{equation}\label{app:eq:dispatch_2_1}
    \tau_{u}^{\uparrow,{(k)}}{=}\hspace{-3mm}\sum_{x{=}N^{(k{-}1)}{+}1}^{N^{(k)}}\max_{r{\in}\mathcal{R}_{b}}\Big\{T_{u,r}^{\uparrow}(t_{x})\Big\},
\end{equation}
and
\begin{align}\label{app:eq:dispatch_4_1}
    M^{\uparrow}_{u}(x)=\hspace{-3mm}\sum_{z{=}N^{(k{-}1)}{+}1}^{x}\sum_{r'{\in}\mathcal{R}_{b}}T_{u,r'}^{\uparrow,\mathsf{C}}(t_{z})\mathfrak{R}^{\uparrow}_{u,r'}(t_{z}).
\end{align}

In the following, we first aim to transform \eqref{app:eq:dispatch_min_1} into the standard form of GP. To this end, using the approximation $\min\{A, B\}\approx (A^{-p}+B^{-p})^{-\frac{1}{p}}$, which is tight when $p \gg 1$, gives us
\begin{equation}\label{app:eq:dispatch_min_2}
    T_{u,r}^{\uparrow}(t_{x})+\epsilon=\left(\left(\tau_{u,r}^{\uparrow}(t_{x})+\epsilon\right)^{-p}+\left(t_{x{+}1}{-}t_{x}+\epsilon\right)^{-p}\right)^{-\frac{1}{p}},
\end{equation}
where $\epsilon$ is added to both sides to avoid division by zero. To simplify \eqref{app:eq:dispatch_min_2}, we define an auxiliary decision variable $\widetilde{T}_{u,r}^{\uparrow,\mathsf{C}}(t_{x})$ that satisfy the following equality constraint:
\begin{equation}\label{app:eq:dispatch_min_3}
    \widetilde{T}_{u,r}^{\uparrow,\mathsf{C}}(t_{x})=\left(\tau_{u,r}^{\uparrow}(t_{x})+\epsilon\right)^{-p}+\left(t_{x{+}1}{-}t_{x}+\epsilon\right)^{-p}
\end{equation}
resulting in 
\begin{equation}\label{app:eq:dispatch_min_4}
    T_{u,r}^{\uparrow}(t_{x})+\epsilon=\left(\widetilde{T}_{u,r}^{\uparrow,\mathsf{C}}(t_{x})\right)^{-\frac{1}{p}}.
\end{equation}
Performing some algebraic operations gives us
\begin{equation}\label{app:eq:dispatch_min_5}
    T_{u,r}^{\uparrow}(t_{x})\times \left(\widetilde{T}_{u,r}^{\uparrow,\mathsf{C}}(t_{x})\right)^{\frac{1}{p}}+\epsilon\times \left(\widetilde{T}_{u,r}^{\uparrow,\mathsf{C}}(t_{x})\right)^{\frac{1}{p}} = 1.
\end{equation}
We transform \eqref{app:eq:dispatch_min_5} into the standard GP format via introducing the following three inequalities:
\begin{equation}\label{app:eq:dispatch_min_6}
     T_{u,r}^{\uparrow}(t_{x})\times \left(\widetilde{T}_{u,r}^{\uparrow,\mathsf{C}}(t_{x})\right)^{\frac{1}{p}}+\epsilon\times \left(\widetilde{T}_{u,r}^{\uparrow,\mathsf{C}}(t_{x})\right)^{\frac{1}{p}}\le 1,
\end{equation}
\begin{equation}\label{app:eq:dispatch_min_7}
    \frac{\left(\mathscr{B}^{\uparrow,\mathsf{min}}_{u,r}(t_{x})\right)^{-1}}{T_{u,r}^{\uparrow}(t_{x})\times \left(\widetilde{T}_{u,r}^{\uparrow,\mathsf{C}}(t_{x})\right)^{\frac{1}{p}}+\epsilon\times \left(\widetilde{T}_{u,r}^{\uparrow,\mathsf{C}}(t_{x})\right)^{\frac{1}{p}}}\le 1,
\end{equation}
\begin{equation}
    \mathscr{B}^{\uparrow,\mathsf{min}}_{u,r}(t_{x})\ge 1,
\end{equation}
where $\mathscr{B}^{\uparrow,\mathsf{min}}_{u,r}(t_{x})$ is an auxiliary decision variable added with a large penalty term to the objective function to force $\mathscr{B}^{\uparrow,\mathsf{min}}_{u,r}(t_{x}){\rightarrow}1^+$ at the optimal point. The fraction in~\eqref{app:eq:dispatch_min_7} still needs transformation since it is an inequality with a posynomial in its denominator, which is not posynomial. We thus exploit arithmetic-geometric mean inequality (Lemma~\ref{Lemma:ArethmaticGeometric}) to approximate the denominator of \eqref{app:eq:dispatch_min_7} with a monomial:
\begin{align}\label{app:eq:dispatch_min_8}
    W^{\uparrow,\mathsf{min}}_{u,r}(\bm{v})&{=}T_{u,r}^{\uparrow}(t_{x}) \left(\widetilde{T}_{u,r}^{\uparrow,\mathsf{C}}(t_{x})\right)^{\frac{1}{p}}{+}\epsilon{\times} \left(\widetilde{T}_{u,r}^{\uparrow,\mathsf{C}}(t_{x})\right)^{\frac{1}{p}}\nonumber\\
    &{\geq} \widehat{W}^{\uparrow,\mathsf{min}}_{u,r}(\bm{v};\ell) {\triangleq} \left(\frac{T_{u,r}^{\uparrow}(t_{x}) \left(\widetilde{T}_{u,r}^{\uparrow,\mathsf{C}}(t_{x})\right)^{\frac{1}{p}} W^{\uparrow,\mathsf{min}}_{u,r}([\bm{v}]^{(\ell-1)})}{\left[T_{u,r}^{\uparrow}(t_{x}) \left(\widetilde{T}_{u,r}^{\uparrow,\mathsf{C}}(t_{x})\right)^{\frac{1}{p}}\right]^{(\ell-1)}}\right)^{\hspace{-2mm}\frac{\left[T_{u,r}^{\uparrow}(t_{x}) \left(\widetilde{T}_{u,r}^{\uparrow,\mathsf{C}}(t_{x})\right)^{\frac{1}{p}}\right]^{(\ell-1)}}{W^{\uparrow,\mathsf{min}}_{u,r}([\bm{v}]^{(\ell-1)})}}\nonumber\\
    &~~~~~~~~~~~~~~~~~\times\left(\frac{ \left(\widetilde{T}_{u,r}^{\uparrow,\mathsf{C}}(t_{x})\right)^{\frac{1}{p}} W^{\uparrow,\mathsf{min}}_{u,r}([\bm{v}]^{(\ell-1)})}{\left[ \left(\widetilde{T}_{u,r}^{\uparrow,\mathsf{C}}(t_{x})\right)^{\frac{1}{p}}\right]^{(\ell-1)}}\right)^{\frac{\left[{\displaystyle\epsilon}\times \left(\widetilde{T}_{u,r}^{\uparrow,\mathsf{C}}(t_{x})\right)^{\frac{1}{p}}\right]^{(\ell-1)}}{W^{\uparrow,\mathsf{min}}_{u,r}([\bm{v}]^{(\ell-1)})}},
\end{align}
which gives us the following approximation of~\eqref{app:eq:dispatch_min_7}:
\begin{equation}
     \frac{\left(\mathscr{B}^{\uparrow,\mathsf{min}}_{u,r}(t_{x})\right)^{-1}}{\widehat{W}^{\uparrow,\mathsf{min}}_{u,r}(\bm{v};\ell)}\le 1.
\end{equation}

We next aim to transform \eqref{app:eq:dispatch_min_3} into standard GP format. In doing so, we rewrite \eqref{app:eq:dispatch_min_3} as follows:
\begin{equation}\label{app:eq:dispatch_min_9}
    \frac{\widetilde{T}_{u,r}^{\uparrow,\mathsf{C}}(t_{x})}{\left(\tau_{u,r}^{\uparrow}(t_{x})+\epsilon\right)^{-p}+\left(t_{x{+}1}{-}t_{x}+\epsilon\right)^{-p}}= 1.
\end{equation}
Defining an auxiliary decision variable $\widetilde{\tau}_{u,r}^{\uparrow}(t_{x})$ that satisfies 
\begin{equation}\label{app:eq:dispatch_min_101}
    \widetilde{\tau}_{u,r}^{\uparrow}(t_{x})=\tau_{u,r}^{\uparrow}(t_{x})+\epsilon
\end{equation}
and considering $\widetilde{t}_{x}$ in \eqref{app:eq:broadcast_latency_min_10} give us
\begin{equation}\label{app:eq:dispatch_min_11}
    \frac{\widetilde{T}_{u,r}^{\uparrow,\mathsf{C}}(t_{x})}{\left(\widetilde{\tau}_{u,r}^{\uparrow}(t_{x})\right)^{-p}+\left(\widetilde{t}_{x}\right)^{-p}}= 1.
\end{equation}
This constraint is not in the format of GP. Therefore, we transform it through splitting it into the following three inequalities:
\begin{equation}\label{app:eq:dispatch_min_12}
    \frac{\widetilde{T}_{u,r}^{\uparrow,\mathsf{C}}(t_{x})}{\left(\widetilde{\tau}_{u,r}^{\uparrow}(t_{x})\right)^{-p}+\left(\widetilde{t}_{x}\right)^{-p}}\le 1,
\end{equation}
\begin{equation}\label{app:eq:dispatch_min_13}
    \frac{\left(\mathscr{C}^{\uparrow,\mathsf{min}}_{u,r}(t_{x})\right)^{-1}\left(\left(\widetilde{\tau}_{u,r}^{\uparrow}(t_{x})\right)^{-p}+\left(\widetilde{t}_{x}\right)^{-p}\right)}{\widetilde{T}_{u,r}^{\uparrow,\mathsf{C}}(t_{x})}\le 1,
\end{equation}
\begin{equation}
    \mathscr{C}^{\uparrow,\mathsf{min}}_{u,r}(t_{x})\ge 1,
\end{equation}
where $\mathscr{C}^{\uparrow,\mathsf{min}}_{u,r}(t_{x})$ is an auxiliary decision variable added with a large penalty term to the objective function to force $\mathscr{C}^{\uparrow,\mathsf{min}}_{u,r}(t_{x}){\rightarrow}1^+$ at the optimal point. The fraction in~\eqref{app:eq:dispatch_min_12} still needs transformation since it is an inequality with a posynomial in the denominator, which is not a posynomial. We thus exploit arithmetic-geometric mean inequality (Lemma~\ref{Lemma:ArethmaticGeometric}) to approximate the denominator with a monomial:
\begin{align}\label{app:eq:dispatch_min_14}
    G^{\uparrow,\mathsf{min}}_{u,r}(\bm{v})=\left(\widetilde{\tau}_{u,r}^{\uparrow}(t_{x})\right)^{-p}+\left(\widetilde{t}_{x}\right)^{-p}\geq \widehat{G}^{\uparrow,\mathsf{min}}_{u,r}(\bm{v};\ell) &\triangleq \left(\frac{\left(\widetilde{\tau}_{u,r}^{\uparrow}(t_{x})\right)^{-p} G^{\uparrow,\mathsf{min}}_{u,r}([\bm{v}]^{(\ell-1)})}{\left[\left(\widetilde{\tau}_{u,r}^{\uparrow}(t_{x})\right)^{-p}\right]^{(\ell-1)}}\right)^{\frac{\left[\left(\widetilde{\tau}_{u,r}^{\uparrow}(t_{x})\right)^{-p}\right]^{(\ell-1)}}{G^{\uparrow,\mathsf{min}}_{u,r}([\bm{v}]^{(\ell-1)})}}\nonumber\\
    &\times\left(\frac{\left(\widetilde{t}_{x}\right)^{-p} G^{\uparrow,\mathsf{min}}_{u,r}([\bm{v}]^{(\ell-1)})}{\left[\left(\widetilde{t}_{x}\right)^{-p}\right]^{(\ell-1)}}\right)^{\frac{\left[\left(\widetilde{t}_{x}\right)^{-p}\right]^{(\ell-1)}}{G^{\uparrow,\mathsf{min}}_{u,r}([\bm{v}]^{(\ell-1)})}},
\end{align}
which gives us an approximation of~\eqref{app:eq:dispatch_min_12} as follows:
\begin{equation}\label{app:eq:dispatch_min_15}
    \frac{\widetilde{T}_{u,r}^{\uparrow,\mathsf{C}}(t_{x})}{\widehat{G}^{\uparrow,\mathsf{min}}_{u,r}(\bm{v};\ell)}\le 1.
\end{equation}

We next transform \eqref{app:eq:dispatch_min_101} into GP format. To this end, we rewrite \eqref{app:eq:dispatch_min_101} as follows:
\begin{equation}\label{app:eq:dispatch_min_102}
    \frac{\widetilde{\tau}_{u,r}^{\uparrow}(t_{x})}{\tau_{u,r}^{\uparrow}(t_{x})+\epsilon}=1.
\end{equation}
This constraint is not in the format of GP. Therefore, we transform it through splitting it into the following three inequalities:
\begin{equation}\label{app:eq:dispatch_min_103}
    \frac{\widetilde{\tau}_{u,r}^{\uparrow}(t_{x})}{\tau_{u,r}^{\uparrow}(t_{x})+\epsilon}\le 1,
\end{equation}
\begin{equation}\label{app:eq:dispatch_min_104}
    \frac{\left(\mathscr{F}^{\uparrow,\mathsf{min}}_{u,r}(t_{x})\right)^{-1}\left(\tau_{u,r}^{\uparrow}(t_{x})+\epsilon\right)}{\widetilde{\tau}_{u,r}^{\uparrow}(t_{x})}\le 1,
\end{equation}
\begin{equation}
    \mathscr{F}^{\uparrow,\mathsf{min}}_{u,r}(t_{x})\ge 1,
\end{equation}
where $\mathscr{F}^{\uparrow,\mathsf{min}}_{u,r}(t_{x})$ is an auxiliary decision variable added with a large penalty term to the objective function to force $\mathscr{F}^{\uparrow,\mathsf{min}}_{u,r}(t_{x}){\rightarrow}1^+$ at the optimal point. The fraction in~\eqref{app:eq:dispatch_min_103} still needs transformation since it is inequality with posynomial in its denominator, which is not a posynomial. We thus exploit arithmetic-geometric mean inequality (Lemma~\ref{Lemma:ArethmaticGeometric}) to approximate the denominator of \eqref{app:eq:dispatch_min_103} with a monomial.
\begin{align}\label{app:eq:dispatch_min_105}
    J^{\uparrow,\mathsf{min}}_{u,r}(\bm{v})=\tau_{u,r}^{\uparrow}(t_{x})+\epsilon \geq \widehat{J}^{\uparrow,\mathsf{min}}_{u,r}(\bm{v};\ell) &\triangleq \left(\frac{\tau_{u,r}^{\uparrow}(t_{x}) J^{\uparrow,\mathsf{min}}_{u,r}([\bm{v}]^{(\ell-1)})}{\left[\tau_{u,r}^{\uparrow}(t_{x})\right]^{(\ell-1)}}\right)^{\frac{\left[\tau_{u,r}^{\uparrow}(t_{x})\right]^{(\ell-1)}}{J^{\uparrow,\mathsf{min}}_{u,r}([\bm{v}]^{(\ell-1)})}}\times\left(J^{\uparrow,\mathsf{min}}_{u,r}([\bm{v}]^{(\ell-1)})\right)^{\frac{\displaystyle\epsilon}{J^{\uparrow,\mathsf{min}}_{u,r}([\bm{v}]^{(\ell-1)})}},
\end{align}
which gives us an approximation of~\eqref{app:eq:dispatch_min_101} as follows:
\begin{equation}\label{app:eq:dispatch_min_106}
    \frac{\tau_{u,r}^{\uparrow}(t_{x})+\epsilon}{\widehat{J}^{\uparrow,\mathsf{min}}_{u,r}(\bm{v};\ell)}\le 1.
\end{equation}

We finally approximate constraint~\eqref{app:eq:dispatch_min_1} as follows:
\begin{tcolorbox}[ams align]
     &T_{u,r}^{\uparrow}(t_{x})\times \left(\widetilde{T}_{u,r}^{\uparrow,\mathsf{C}}(t_{x})\right)^{\frac{1}{p}}+\epsilon\times \left(\widetilde{T}_{u,r}^{\uparrow,\mathsf{C}}(t_{x})\right)^{\frac{1}{p}}\le 1,~~\frac{\left(\mathscr{B}^{\uparrow,\mathsf{min}}_{u,r}(t_{x})\right)^{-1}}{\widehat{W}^{\uparrow,\mathsf{min}}_{u,r}(\bm{v};\ell)}\le 1,\nonumber\\
     &\frac{\left(\mathscr{C}^{\uparrow,\mathsf{min}}_{u,r}(t_{x})\right)^{-1}\left(\left(\widetilde{\tau}_{u,r}^{\uparrow}(t_{x})\right)^{-p}+\left(\widetilde{t}_{x}\right)^{-p}\right)}{\widetilde{T}_{u,r}^{\uparrow,\mathsf{C}}(t_{x})}\le 1,~~\frac{\widetilde{T}_{u,r}^{\uparrow,\mathsf{C}}(t_{x})}{\widehat{G}^{\uparrow,\mathsf{min}}_{u,r}(\bm{v};\ell)}\le 1,\nonumber\\
     &\frac{\tau_{u,r}^{\uparrow}(t_{x})+\epsilon}{\widehat{J}^{\uparrow,\mathsf{min}}_{u,r}(\bm{v};\ell)}\le 1,~~~~\frac{\left(\mathscr{F}^{\uparrow,\mathsf{min}}_{u,r}(t_{x})\right)^{-1}\left(\tau_{u,r}^{\uparrow}(t_{x})+\epsilon\right)}{\widetilde{\tau}_{u,r}^{\uparrow}(t_{x})}\le 1,\nonumber\\
     &\frac{1}{\mathscr{B}^{\uparrow,\mathsf{min}}_{u,r}(t_{x})}\le 1,~~\frac{1}{\mathscr{C}^{\uparrow,\mathsf{min}}_{u,r}(t_{x})}\le 1,~~~\frac{1}{\mathscr{F}^{\uparrow,\mathsf{min}}_{u,r}(t_{x})}\le 1.\nonumber
\end{tcolorbox}

We next aim to transform \eqref{app:eq:dispatch_3} into the standard GP format. In doing so, performing some algebraic manipulations on \eqref{app:eq:dispatch_3} gives us
\begin{equation}\label{app:eq:dispatch_3_4}
    \frac{\mathfrak{R}^{\uparrow}_{u,r}(t_{x})\tau_{u,r}^{\uparrow}(t_{x})+\tau_{u,r}^{\uparrow}(t_{x})+\psi_{u,r}(t_{x})M^{\uparrow}_{u}(x)+1}{\psi_{u,r}(t_{x})\alpha_{\bm{\omega}} M +\beta^{\uparrow}_u(t_{x})\tau_{u,r}^{\uparrow}(t_{x})+1}{=} 1.
\end{equation}
However, this constraint is not in the format of GP. Therefore, we transform it into the standard GP format via introducing the following three inequalities:
\begin{equation}\label{app:eq:dispatch_3_5}
     \frac{\mathfrak{R}^{\uparrow}_{u,r}(t_{x})\tau_{u,r}^{\uparrow}(t_{x})+\tau_{u,r}^{\uparrow}(t_{x})+\psi_{u,r}(t_{x})M^{\uparrow}_{u}(x)+1}{\psi_{u,r}(t_{x})\alpha_{\bm{\omega}} M +\beta^{\uparrow}_u(t_{x})\tau_{u,r}^{\uparrow}(t_{x})+1}\le 1,
\end{equation}
\begin{equation}\label{app:eq:dispatch_3_6}
    \frac{\left(\mathscr{H}^{\uparrow,\mathsf{L}}_{u,r}(t_{x})\right)^{-1}\left(\psi_{u,r}(t_{x})\alpha_{\bm{\omega}} M +\beta^{\uparrow}_u(t_{x})\tau_{u,r}^{\uparrow}(t_{x})+1\right)}{\mathfrak{R}^{\uparrow}_{u,r}(t_{x})\tau_{u,r}^{\uparrow}(t_{x})+\tau_{u,r}^{\uparrow}(t_{x})+\psi_{u,r}(t_{x})M^{\uparrow}_{u}(x)+1}\le 1,
\end{equation}
\begin{equation}
    \mathscr{H}^{\uparrow,\mathsf{L}}_{u,r}(t_{x})\ge 1,
\end{equation}
where $\mathscr{H}^{\uparrow,\mathsf{L}}_{u,r}(t_{x})$ is an auxiliary decision variable added with a large penalty term to the objective function to force $\mathscr{H}^{\uparrow,\mathsf{L}}_{u,r}(t_{x}){\rightarrow}1^+$ at the optimal point. The fractions in~\eqref{app:eq:dispatch_3_5} and \eqref{app:eq:dispatch_3_6} still need transformation since they are inequalities with posynomials in their denominator, which are not posynomials. We thus exploit arithmetic-geometric mean inequality (Lemma~\ref{Lemma:ArethmaticGeometric}) to approximate the denominators of \eqref{app:eq:dispatch_3_5} and \eqref{app:eq:dispatch_3_6} with monomials. In doing so, we approximate the denominator in \eqref{app:eq:dispatch_3_5} as follows:
\begin{align}\label{app:eq:dispatch_3_7}
    L^{\uparrow,\mathsf{L}}_{u,r}(\bm{v}){=}&\psi_{u,r}(t_{x})\alpha_{\bm{\omega}} M +\beta^{\uparrow}_u(t_{x})\tau_{u,r}^{\uparrow}(t_{x})+1{\geq} \widehat{L}^{\uparrow,\mathsf{L}}_{u,r}(\bm{v};\ell) {\triangleq} \left(\frac{\psi_{u,r}(t_{x})\ L^{\uparrow,\mathsf{L}}_{u,r}([\bm{v}]^{(\ell-1)})}{\left[\psi_{u,r}(t_{x})\right]^{(\ell-1)}}\right)^{\hspace{-2mm}\frac{\left[\psi_{u,r}(t_{x})\alpha_{\bm{\omega}} M\right]^{(\ell-1)}}{L^{\uparrow,\mathsf{L}}_{u,r}([\bm{v}]^{(\ell-1)})}}\nonumber\\
    &~~~~~~~~~~~~~~~~~~~\times\left(\frac{\beta^{\uparrow}_u(t_{x})\tau_{u,r}^{\uparrow}(t_{x}) L^{\uparrow,\mathsf{L}}_{u,r}([\bm{v}]^{(\ell-1)})}{\left[\beta^{\uparrow}_u(t_{x})\tau_{u,r}^{\uparrow}(t_{x})\right]^{(\ell-1)}}\right)^{\frac{\left[\beta^{\uparrow}_u(t_{x})\tau_{u,r}^{\uparrow}(t_{x})\right]^{(\ell-1)}}{L^{\uparrow,\mathsf{L}}_{u,r}([\bm{v}]^{(\ell-1)})}}\times\left(L^{\uparrow,\mathsf{L}}_{u,r}([\bm{v}]^{(\ell-1)})\right)^{\frac{1}{L^{\uparrow,\mathsf{L}}_{u,r}([\bm{v}]^{(\ell-1)})}},
\end{align}
which gives us the following approximation of~\eqref{app:eq:dispatch_3_5}:
\begin{equation}
     \frac{\mathfrak{R}^{\uparrow}_{u,r}(t_{x})\tau_{u,r}^{\uparrow}(t_{x})+\tau_{u,r}^{\uparrow}(t_{x})+\psi_{u,r}(t_{x})M^{\uparrow}_{u}(x)+1}{\widehat{L}^{\uparrow,\mathsf{L}}_{u,r}(\bm{v};\ell)}\le 1.
\end{equation}
Similarly, we approximate the denominator in \eqref{app:eq:dispatch_3_6} as follows:
\begin{align}\label{app:eq:dispatch_3_8}
    &R^{\uparrow,\mathsf{L}}_{u,r}(\bm{v}){=}\mathfrak{R}^{\uparrow}_{u,r}(t_{x})\tau_{u,r}^{\uparrow}(t_{x})+\tau_{u,r}^{\uparrow}(t_{x})+\psi_{u,r}(t_{x})M^{\uparrow}_{u}(x)+1\nonumber\\
    &{\geq} \widehat{R}^{\uparrow,\mathsf{L}}_{u,r}(\bm{v};\ell) {\triangleq} \left(\frac{\mathfrak{R}^{\uparrow}_{u,r}(t_{x})\tau_{u,r}^{\uparrow}(t_{x}) R^{\uparrow,\mathsf{L}}_{u,r}([\bm{v}]^{(\ell-1)})}{\left[\mathfrak{R}^{\uparrow}_{u,r}(t_{x})\tau_{u,r}^{\uparrow}(t_{x})\right]^{(\ell-1)}}\right)^{\hspace{-2mm}\frac{\left[\mathfrak{R}^{\uparrow}_{u,r}(t_{x})\tau_{u,r}^{\uparrow}(t_{x})\right]^{(\ell-1)}}{R^{\uparrow,\mathsf{L}}_{u,r}([\bm{v}]^{(\ell-1)})}}\times\left(\frac{\tau_{u,r}^{\uparrow}(t_{x}) R^{\uparrow,\mathsf{L}}_{u,r}([\bm{v}]^{(\ell-1)})}{\left[\tau_{u,r}^{\uparrow}(t_{x})\right]^{(\ell-1)}}\right)^{\frac{\left[\tau_{u,r}^{\uparrow}(t_{x})\right]^{(\ell-1)}}{R^{\uparrow,\mathsf{L}}_{u,r}([\bm{v}]^{(\ell-1)})}}\nonumber\\
    &~~~~~~~~~~~~~~\times\left(\frac{\psi_{u,r}(t_{x})M^{\uparrow}_{u}(x) R^{\uparrow,\mathsf{L}}_{u,r}([\bm{v}]^{(\ell-1)})}{\left[\psi_{u,r}(t_{x})M^{\uparrow}_{u}(x)\right]^{(\ell-1)}}\right)^{\frac{\left[\psi_{u,r}(t_{x})M^{\uparrow}_{u}(x)\right]^{(\ell-1)}}{R^{\uparrow,\mathsf{L}}_{u,r}([\bm{v}]^{(\ell-1)})}}\times\left(R^{\uparrow,\mathsf{L}}_{u,r}([\bm{v}]^{(\ell-1)})\right)^{\frac{1}{R^{\uparrow,\mathsf{L}}_{u,r}([\bm{v}]^{(\ell-1)})}},
\end{align}
which gives the following approximation of~\eqref{app:eq:dispatch_3_6}:
\begin{equation}
     \frac{\left(\mathscr{H}^{\uparrow,\mathsf{L}}_{u,r}(t_{x})\right)^{-1}\left(\psi_{u,r}(t_{x})\alpha_{\bm{\omega}} M +\beta^{\uparrow}_u(t_{x})\tau_{u,r}^{\uparrow}(t_{x})+1\right)}{ \widehat{R}^{\uparrow,\mathsf{L}}_{u,r}(\bm{v};\ell)}\le 1.
\end{equation}
We finally approximate \eqref{app:eq:dispatch_3} as follows:
\begin{tcolorbox}[ams align]
     & \frac{\mathfrak{R}^{\uparrow}_{u,r}(t_{x})\tau_{u,r}^{\uparrow}(t_{x})+\tau_{u,r}^{\uparrow}(t_{x})+\psi_{u,r}(t_{x})M^{\uparrow}_{u}(x)+1}{\widehat{L}^{\uparrow,\mathsf{L}}_{u,r}(\bm{v};\ell)}\le 1,\nonumber\\
     & \frac{\left(\mathscr{H}^{\uparrow,\mathsf{L}}_{u,r}(t_{x})\right)^{-1}\left(\psi_{u,r}(t_{x})\alpha_{\bm{\omega}} M +\beta^{\uparrow}_u(t_{x})\tau_{u,r}^{\uparrow}(t_{x})+1\right)}{ \widehat{R}^{\uparrow,\mathsf{L}}_{u,r}(\bm{v};\ell)}\le 1,\nonumber\\
     &\frac{1}{\mathscr{H}^{\uparrow,\mathsf{L}}_{u,r}(t_{x})}\le 1.\nonumber
\end{tcolorbox}

We next aim to transform \eqref{app:eq:dispatch_4_1} into GP format. To this end, we rewrite \eqref{app:eq:dispatch_4_1} as follows:
\begin{equation}\label{app:eq:dispatch_4_1_1}
   \frac{M^{\uparrow}_{u}(x)+1}{\displaystyle\sum_{z{=}0}^{x}\sum_{r'{\in}\mathcal{R}_{b}}T_{u,r'}^{\uparrow,\mathsf{C}}(t_{z})\mathfrak{R}^{\uparrow}_{u,r'}(t_{z})+1}{=} 1.
\end{equation} 
This constraint is not in the format of GP. Therefore, we transform it through splitting it into the following three inequalities:
\begin{equation}\label{app:eq:dispatch_4_1_4}
    \frac{M^{\uparrow}_{u}(x)+1}{\displaystyle\sum_{z{=}0}^{x}\sum_{r'{\in}\mathcal{R}_{b}}T_{u,r'}^{\uparrow,\mathsf{C}}(t_{z})\mathfrak{R}^{\uparrow}_{u,r'}(t_{z})+1}\le 1,
\end{equation}
\begin{equation}\label{app:eq:dispatch_4_1_5}
    \frac{\left(\mathscr{L}^{\uparrow,\mathsf{M},\mathsf{C}}_{u}(t_{x})\right)^{-1}\left(\displaystyle\sum_{z{=}0}^{x}\sum_{r'{\in}\mathcal{R}_{b}}T_{u,r'}^{\uparrow,\mathsf{C}}(t_{z})\mathfrak{R}^{\uparrow}_{u,r'}(t_{z})+1\right)}{M^{\uparrow}_{u}(x)+1}\le 1,
\end{equation}
\begin{equation}
    \mathscr{L}^{\uparrow,\mathsf{M},\mathsf{C}}_{u}(t_{x})\ge 1,
\end{equation}
where $\mathscr{L}^{\uparrow,\mathsf{M},\mathsf{C}}_{u}(t_{x})$ is an auxiliary decision variable added with a large penalty term to the objective function to force $\mathscr{L}^{\uparrow,\mathsf{M},\mathsf{C}}_{u}(t_{x}){\rightarrow}1^+$ at the optimal point. The fractions in~\eqref{app:eq:dispatch_4_1_4} and \eqref{app:eq:dispatch_4_1_5} still need transformation since they are inequalities with posynomials in their denominators, which are not posynomials. We thus exploit arithmetic-geometric mean inequality (Lemma~\ref{Lemma:ArethmaticGeometric}) to approximate the denominators in \eqref{app:eq:dispatch_4_1_4} and \eqref{app:eq:dispatch_4_1_5} with monomials. To this end, we approximate the denominator in \eqref{app:eq:dispatch_4_1_4} as follows:
\begin{align}\label{app:eq:dispatch_4_1_6}
   S^{\uparrow,\mathsf{M},\mathsf{C}}_{u}(\bm{v})&=\sum_{z{=}0}^{x}\sum_{r'{\in}\mathcal{R}_{b}}T_{u,r'}^{\uparrow,\mathsf{C}}(t_{z})\mathfrak{R}^{\uparrow}_{u,r'}(t_{z})+1\nonumber\\
   &\geq \widehat{S}^{\uparrow,\mathsf{M},\mathsf{C}}_{u}(\bm{v};\ell) \triangleq \prod_{z{=}0}^{x}\prod_{r'{\in}\mathcal{R}_{b}}\left(\frac{T_{u,r'}^{\uparrow,\mathsf{C}}(t_{z})\mathfrak{R}^{\uparrow}_{u,r'}(t_{z}) S^{\uparrow,\mathsf{M},\mathsf{C}}_{u}([\bm{v}]^{(\ell-1)})}{\left[T_{u,r'}^{\uparrow,\mathsf{C}}(t_{z})\mathfrak{R}^{\uparrow}_{u,r'}(t_{z})\right]^{(\ell-1)}}\right)^{\frac{\left[T_{u,r'}^{\uparrow,\mathsf{C}}(t_{z})\mathfrak{R}^{\uparrow}_{u,r'}(t_{z})\right]^{(\ell-1)}}{S^{\uparrow,\mathsf{M},\mathsf{C}}_{u}([\bm{v}]^{(\ell-1)})}}\nonumber\\
   &~~~~~~~~~~~~~~~\times\left(S^{\uparrow,\mathsf{M},\mathsf{C}}_{u}([\bm{v}]^{(\ell-1)})\right)^{\frac{1}{S^{\uparrow,\mathsf{M},\mathsf{C}}_{u}([\bm{v}]^{(\ell-1)})}},
\end{align}
which gives us an approximation of~\eqref{app:eq:dispatch_4_1_4} as follows:
\begin{equation}\label{app:eq:dispatch_4_1_7}
    \frac{M^{\uparrow}_{u}(x)+1}{\widehat{S}^{\uparrow,\mathsf{M},\mathsf{C}}_{u}(\bm{v};\ell)}\le 1.
\end{equation}

Similarly, we approximate the denominator in \eqref{app:eq:dispatch_4_1_5} as follows:
\begin{align}\label{app:eq:dispatch_4_1_8}
    &Q^{\uparrow,\mathsf{M},\mathsf{C}}_{u}(\bm{v}){=}M^{\uparrow}_{u}(x)+1{\geq} \widehat{Q}^{\uparrow,\mathsf{M},\mathsf{C}}_{u}(\bm{v};\ell) {\triangleq} \left(\frac{M^{\uparrow}_{u}(x) Q^{\uparrow,\mathsf{M},\mathsf{C}}_{u}([\bm{v}]^{(\ell-1)})}{\left[M^{\uparrow}_{u}(x)\right]^{(\ell-1)}}\right)^{\hspace{-2mm}\frac{\left[M^{\uparrow}_{u}(x)\right]^{(\ell-1)}}{Q^{\uparrow,\mathsf{M},\mathsf{C}}_{u}([\bm{v}]^{(\ell-1)})}}\times\left(Q^{\uparrow,\mathsf{M},\mathsf{C}}_{u}([\bm{v}]^{(\ell-1)})\right)^{\frac{1}{Q^{\uparrow,\mathsf{M},\mathsf{C}}_{u}([\bm{v}]^{(\ell-1)})}},
\end{align}
which gives the following approximation of~\eqref{app:eq:broadcast_latency_6}:
\begin{equation}
     \frac{\left(\mathscr{L}^{\uparrow,\mathsf{M},\mathsf{C}}_{u}(t_{x})\right)^{-1}\left(\displaystyle\sum_{z{=}0}^{x}\sum_{r'{\in}\mathcal{R}_{b}}T_{u,r'}^{\uparrow,\mathsf{C}}(t_{z})\mathfrak{R}^{\uparrow}_{u,r'}(t_{z})+1\right)}{\widehat{Q}^{\uparrow,\mathsf{M},\mathsf{C}}_{u}(\bm{v};\ell)}\le 1.
\end{equation}

We finally approximate constraint~\eqref{app:eq:dispatch_4} as follows:
\begin{tcolorbox}[ams align]
     &\frac{M^{\uparrow}_{u}(x)+1}{\widehat{S}^{\uparrow,\mathsf{M},\mathsf{C}}_{u}(\bm{v};\ell)}\le 1,~~~~\frac{\left(\mathscr{L}^{\uparrow,\mathsf{M},\mathsf{C}}_{u}(t_{x})\right)^{-1}\left(\displaystyle\sum_{z{=}0}^{x}\sum_{r'{\in}\mathcal{R}_{b}}T_{u,r'}^{\uparrow,\mathsf{C}}(t_{z})\mathfrak{R}^{\uparrow}_{u,r'}(t_{z})+1\right)}{\widehat{Q}^{\uparrow,\mathsf{M},\mathsf{C}}_{u}(\bm{v};\ell)}\le 1,~~~~~\frac{1}{\mathscr{L}^{\uparrow,\mathsf{M},\mathsf{C}}_{u}(t_{x})}\le 1,\nonumber
\end{tcolorbox}

We next aim to transform \eqref{app:eq:dispatch_2_1} into the standard form of GP. In doing so, we rewrite \eqref{app:eq:dispatch_2_1} as follows:
\begin{equation}\label{app:eq:dispatch_2_2}
\begin{aligned}
   \tau_{u}^{\uparrow,{(k)}}&=\sum_{x{=}0}^{N}\max_{r{\in}\mathcal{R}_{b}}\Big\{T_{u,r}^{\uparrow}(t_{x})+\epsilon\Big\}-\sum_{x{=}0}^{N}\epsilon\\
   &=\sum_{x{=}0}^{N}\max_{r{\in}\mathcal{R}_{b}}\Big\{T_{u,r}^{\uparrow}(t_{x})+\epsilon\Big\}-\left(N^{(k)}-N^{(k{-}1)}\right)\epsilon.
\end{aligned}
\end{equation}
Performing some algebraic manipulations gives us
\begin{equation}\label{app:eq:dispatch_2_2}
   \frac{\tau_{u}^{\uparrow,{(k)}}+N^{(k)}\epsilon}{\displaystyle\sum_{x{=}0}^{N}\max_{r{\in}\mathcal{R}_{b}}\Big\{T_{u,r}^{\uparrow}(t_{x})+\epsilon\Big\}+N^{(k{-}1)}\epsilon}=1.
\end{equation} 
Using the approximation $\max\{A, B\}\approx (A^{p}+B^{p})^{-\frac{1}{p}}$, which is tight when $p \gg 1$, gives us
\begin{equation}\label{app:eq:dispatch_2_3}
   \frac{\tau_{u}^{\uparrow,{(k)}}+N^{(k)}\epsilon}{\displaystyle\sum_{x{=}0}^{N}\left(\sum_{r{\in}\mathcal{R}_{b}}\left(T_{u,r}^{\uparrow}(t_{x})+\epsilon\right)^{p}\right)^{-\frac{1}{p}}+N^{(k{-}1)}\epsilon}=1.
\end{equation} 
Defining an auxiliary decision variable $T_{u}^{\uparrow,\mathsf{max}}(t_{x})$ that satisfies
\begin{equation}\label{app:eq:dispatch_2_4}
    T_{u}^{\uparrow,\mathsf{max}}(t_{x}) =\sum_{r{\in}\mathcal{R}_{b}}\left(T_{u,r}^{\uparrow}(t_{x})+\epsilon\right)^{p}
\end{equation}
results in
\begin{equation}\label{app:eq:dispatch_2_5}
   \frac{\tau_{u}^{\uparrow,{(k)}}+N^{(k)}\epsilon}{\displaystyle\sum_{x{=}0}^{N}\left(T_{u}^{\uparrow,\mathsf{max}}(t_{x})\right)^{-\frac{1}{p}}+N^{(k{-}1)}\epsilon}=1.
\end{equation} 
This constraint is not in the format of GP. Therefore, we transform it through splitting it into the following three inequalities:
\begin{equation}\label{app:eq:dispatch_2_6}
    \frac{\tau_{u}^{\uparrow,{(k)}}+N^{(k)}\epsilon}{\displaystyle\sum_{x{=}0}^{N}\left(T_{u}^{\uparrow,\mathsf{max}}(t_{x})\right)^{-\frac{1}{p}}+N^{(k{-}1)}\epsilon}\le 1,
\end{equation}

\begin{equation}\label{app:eq:dispatch_2_7}
    \frac{\left(\mathscr{Q}^{\uparrow,\mathsf{C}}_{u}(t_{x})\right)^{-1}\left(\displaystyle\sum_{x{=}0}^{N}\left(T_{u}^{\uparrow,\mathsf{max}}(t_{x})\right)^{-\frac{1}{p}}+N^{(k{-}1)}\epsilon\right)}{\tau_{u}^{\uparrow,{(k)}}+N^{(k)}\epsilon}\le 1,
\end{equation}

\begin{equation}
    \mathscr{Q}^{\uparrow,\mathsf{C}}_{u}(t_{x})\ge 1,
\end{equation}
where $\mathscr{Q}^{\uparrow,\mathsf{C}}_{u}(t_{x})$ is an auxiliary decision variable added with a large penalty term to the objective function to force $\mathscr{Q}^{\uparrow,\mathsf{C}}_{u}(t_{x}){\rightarrow}1^+$ at the optimal point. The fractions in~\eqref{app:eq:dispatch_2_6} and \eqref{app:eq:dispatch_2_7} still need transformation since they are inequalities with posynomials in their denominators, which are not posynomials. We thus exploit arithmetic-geometric mean inequality (Lemma~\ref{Lemma:ArethmaticGeometric}) to approximate the denominators in \eqref{app:eq:dispatch_2_6} and \eqref{app:eq:dispatch_2_7} with monomials. To this end, we approximate the denominator in \eqref{app:eq:dispatch_2_6} as follows:
\begin{align}\label{app:eq:dispatch_2_8}
   &V^{\uparrow,\mathsf{C}}_{u}(\bm{v})=\sum_{x{=}0}^{N}\left(T_{u}^{\uparrow,\mathsf{max}}(t_{x})\right)^{-\frac{1}{p}}+N^{(k{-}1)}\epsilon\nonumber\\
   &\geq \widehat{V}^{\uparrow,\mathsf{C}}_{u}(\bm{v};\ell) \triangleq \prod_{x{=}0}^{N}\left(\frac{\left(T_{u}^{\uparrow,\mathsf{max}}(t_{x})\right)^{-\frac{1}{p}} V^{\uparrow,\mathsf{C}}_{u}([\bm{v}]^{(\ell-1)})}{\left[\left(T_{u}^{\uparrow,\mathsf{max}}(t_{x})\right)^{-\frac{1}{p}}\right]^{(\ell-1)}}\right)^{\frac{\left[\left(T_{u}^{\uparrow,\mathsf{max}}(t_{x})\right)^{-\frac{1}{p}}\right]^{(\ell-1)}}{V^{\uparrow,\mathsf{C}}_{u}([\bm{v}]^{(\ell-1)})}}\times\left(V^{\uparrow,\mathsf{C}}_{u}([\bm{v}]^{(\ell-1)})\right)^{\frac{\left[N^{(k{-}1)}\displaystyle\epsilon\right]^{(\ell-1)}}{V^{\uparrow,\mathsf{C}}_{u}([\bm{v}]^{(\ell-1)})}},
\end{align}
which gives us an approximation of~\eqref{app:eq:dispatch_2_6} as follows:
\begin{equation}\label{app:eq:dispatch_2_9}
    \frac{\tau_{u}^{\uparrow,{(k)}}+N^{(k)}\epsilon}{\widehat{V}^{\uparrow,\mathsf{C}}_{u}(\bm{v};\ell)}\le 1.
\end{equation}

Similarly, we approximate the denominator in \eqref{app:eq:dispatch_2_7} as follows:
\begin{align}\label{app:eq:dispatch_2_10}
    W^{\uparrow,\mathsf{C}}_{u}(\bm{v})&=\tau_{u}^{\uparrow,{(k)}}+N^{(k)}\epsilon \geq \widehat{W}^{\uparrow,\mathsf{C}}_{u}(\bm{v};\ell) \triangleq \left(\frac{\tau_{u}^{\uparrow,{(k)}} W^{\uparrow,\mathsf{C}}_{u}([\bm{v}]^{(\ell-1)})}{\left[\tau_{u}^{\uparrow,{(k)}}\right]^{(\ell-1)}}\right)^{\frac{\left[\tau_{u}^{\uparrow,{(k)}}\right]^{(\ell-1)}}{W^{\uparrow,\mathsf{C}}_{u}([\bm{v}]^{(\ell-1)})}}\times\left(W^{\uparrow,\mathsf{C}}_{u}([\bm{v}]^{(\ell-1)})\right)^{\frac{\left[N^{(k)}\displaystyle\epsilon\right]^{(\ell-1)}}{W^{\uparrow,\mathsf{C}}_{u}([\bm{v}]^{(\ell-1)})}},
\end{align}
which gives us an approximation of~\eqref{app:eq:dispatch_2_7} as follows:
\begin{equation}\label{app:eq:dispatch_2_11}
    \frac{\left(\mathscr{Q}^{\uparrow,\mathsf{C}}_{u}(t_{x})\right)^{-1}\left(\displaystyle\sum_{x{=}0}^{N}\left(T_{u}^{\uparrow,\mathsf{max}}(t_{x})\right)^{-\frac{1}{p}}+N^{(k{-}1)}\epsilon\right)}{\widehat{W}^{\uparrow,\mathsf{C}}_{u}(\bm{v};\ell)}\le 1.
\end{equation}

We next transform \eqref{app:eq:dispatch_2_4} into the standard form of GP. To do so, we rewrite \eqref{app:eq:dispatch_2_4} as follows:
\begin{equation}\label{app:eq:dispatch_2_4_1}
    \frac{T_{u}^{\uparrow,\mathsf{max}}(t_{x})}{\sum_{r{\in}\mathcal{R}_{b}}\left(T_{u,r}^{\uparrow}(t_{x})+\epsilon\right)^{p}} =1.
\end{equation}
Defining an auxiliary decision variable $T_{u,r}^{\uparrow,\epsilon}(t_{x})$ satisfying
\begin{equation}\label{app:eq:dispatch_2_4_2}
    T_{u,r}^{\uparrow,\epsilon}(t_{x})=T_{u,r}^{\uparrow}(t_{x})+\epsilon
\end{equation}
leads to
\begin{equation}\label{app:eq:dispatch_2_4_3}
    \frac{T_{u}^{\uparrow,\mathsf{max}}(t_{x})}{\sum_{r{\in}\mathcal{R}_{b}}\left(T_{u,r}^{\uparrow,\epsilon}(t_{x})\right)^{p}} =1.
\end{equation}
This constraint is not in the format of GP. Therefore, we transform it through splitting it into the following three inequalities:
\begin{equation}\label{app:eq:dispatch_2_4_4}
    \frac{T_{u}^{\uparrow,\mathsf{max}}(t_{x})}{\sum_{r{\in}\mathcal{R}_{b}}\left(T_{u,r}^{\uparrow,\epsilon}(t_{x})\right)^{p}}\le 1,
\end{equation}

\begin{equation}\label{app:eq:dispatch_2_4_5}
    \frac{\left(\mathscr{R}^{\uparrow,\mathsf{max}}_{u}(t_{x})\right)^{-1}\left(\sum_{r{\in}\mathcal{R}_{b}}\left(T_{u,r}^{\uparrow,\epsilon}(t_{x})\right)^{p}\right)}{T_{u}^{\uparrow,\mathsf{max}}(t_{x})}\le 1,
\end{equation}

\begin{equation}
    \mathscr{R}^{\uparrow,\mathsf{max}}_{u}(t_{x})\ge 1,
\end{equation}
where $\mathscr{R}^{\uparrow,\mathsf{max}}_{u}(t_{x})$ is an auxiliary decision variable added with a large penalty term to the objective function to force $\mathscr{R}^{\uparrow,\mathsf{max}}_{u}(t_{x}){\rightarrow}1^+$ at the optimal point. The fraction in~\eqref{app:eq:dispatch_2_4_4} still needs transformation since it is inequality with a posynomial in its denominator, which is not a posynomial. We thus exploit arithmetic-geometric mean inequality (Lemma~\ref{Lemma:ArethmaticGeometric}) to approximate the denominator in \eqref{app:eq:dispatch_2_4_4} with a monomial.
\begin{align}\label{app:eq:dispatch_2_4_6}
    &Y^{\uparrow,\mathsf{max},\mathsf{C}}_{u}(\bm{v})=\sum_{r{\in}\mathcal{R}_{b}}\left(T_{u,r}^{\uparrow,\epsilon}(t_{x})\right)^{p}\geq \widehat{Y}^{\uparrow,\mathsf{max},\mathsf{C}}_{u}(\bm{v};\ell) \triangleq \prod_{r{\in}\mathcal{R}_{b}}\left(\frac{\left(T_{u,r}^{\uparrow,\epsilon}(t_{x})\right)^{p} Y^{\uparrow,\mathsf{max},\mathsf{C}}_{u}([\bm{v}]^{(\ell-1)})}{\left[\left(T_{u,r}^{\uparrow,\epsilon}(t_{x})\right)^{p}\right]^{(\ell-1)}}\right)^{\frac{\left[\left(T_{u,r}^{\uparrow,\epsilon}(t_{x})\right)^{p}\right]^{(\ell-1)}}{Y^{\uparrow,\mathsf{max},\mathsf{C}}_{u}([\bm{v}]^{(\ell-1)})}},
\end{align}
which gives us an approximation of~\eqref{app:eq:dispatch_2_4_4} as follows:
\begin{equation}\label{app:eq:dispatch_2_4_6}
    \frac{T_{u}^{\uparrow,\mathsf{max}}(t_{x})}{\widehat{Y}^{\uparrow,\mathsf{max},\mathsf{C}}_{u}(\bm{v};\ell)}\le 1.
\end{equation}

We next transform \eqref{app:eq:dispatch_2_4_2} into GP format. To this end, we rewrite \eqref{app:eq:dispatch_2_4_2} as follows:
\begin{equation}\label{app:eq:dispatch_2_4_7}
    \frac{T_{u,r}^{\uparrow,\epsilon}(t_{x})}{T_{u,r}^{\uparrow}(t_{x})+\epsilon}=1.
\end{equation}
This constraint is not in the format of GP. Therefore, we transform it through splitting it into the following three inequalities:
\begin{equation}\label{app:eq:dispatch_2_4_4}
    \frac{T_{u,r}^{\uparrow,\epsilon}(t_{x})}{T_{u,r}^{\uparrow}(t_{x})+\epsilon}\le 1,
\end{equation}

\begin{equation}\label{app:eq:dispatch_2_4_5}
    \frac{\left(\mathscr{S}^{\uparrow,\epsilon}_{u,r}(t_{x})\right)^{-1}\left(T_{u,r}^{\uparrow}(t_{x})+\epsilon\right)}{T_{u,r}^{\uparrow,\epsilon}(t_{x})}\le 1,
\end{equation}

\begin{equation}
    \mathscr{S}^{\uparrow,\epsilon}_{u,r}(t_{x})\ge 1,
\end{equation}
where $\mathscr{S}^{\uparrow,\epsilon}_{u,r}(t_{x})$ is an auxiliary decision variable added with a large penalty term to the objective function to force $\mathscr{S}^{\uparrow,\epsilon}_{u,r}(t_{x}){\rightarrow}1^+$ at the optimal point. The fraction in~\eqref{app:eq:dispatch_2_4_4} still needs transformation since it is inequality with a posynomial in its denominator, which is not a posynomial. We thus exploit arithmetic-geometric mean inequality (Lemma~\ref{Lemma:ArethmaticGeometric}) to approximate the denominator in \eqref{app:eq:dispatch_2_4_4} with a monomial.
\begin{align}\label{app:eq:dispatch_2_4_6}
    Z^{\uparrow,\epsilon}_{u,r}(\bm{v})&=T_{u,r}^{\uparrow}(t_{x})+\epsilon \geq \widehat{Z}^{\uparrow,\epsilon}_{u,r}(\bm{v};\ell) \triangleq \left(\frac{T_{u,r}^{\uparrow}(t_{x}) Z^{\uparrow,\epsilon}_{u,r}([\bm{v}]^{(\ell-1)})}{\left[T_{u,r}^{\uparrow}(t_{x})\right]^{(\ell-1)}}\right)^{\frac{\left[T_{u,r}^{\uparrow}(t_{x})\right]^{(\ell-1)}}{Z^{\uparrow,\epsilon}_{u,r}([\bm{v}]^{(\ell-1)})}}\times\left(Z^{\uparrow,\epsilon}_{u,r}([\bm{v}]^{(\ell-1)})\right)^{\frac{\left[\displaystyle\epsilon\right]^{(\ell-1)}}{Z^{\uparrow,\epsilon}_{u,r}([\bm{v}]^{(\ell-1)})}},
\end{align}
which gives us an approximation of~\eqref{app:eq:dispatch_2_4_4} as follows:
\begin{equation}\label{app:eq:dispatch_2_4_7}
    \frac{T_{u,r}^{\uparrow,\epsilon}(t_{x})}{\widehat{Z}^{\uparrow,\epsilon}_{u,r}(\bm{v};\ell)}\le 1.
\end{equation}

We finally approximate constraint~\eqref{app:eq:dispatch_2} as follows:
\begin{tcolorbox}[ams align]
     &\frac{\tau_{u}^{\uparrow,{(k)}}+N^{(k)}\epsilon}{\widehat{V}^{\uparrow,\mathsf{C}}_{u}(\bm{v};\ell)}\le 1,~~~~~~\frac{\left(\mathscr{Q}^{\uparrow,\mathsf{C}}_{u}(t_{x})\right)^{-1}\left(\displaystyle\sum_{x{=}0}^{N}\left(T_{u}^{\uparrow,\mathsf{max}}(t_{x})\right)^{-\frac{1}{p}}+N^{(k{-}1)}\epsilon\right)}{\widehat{W}^{\uparrow,\mathsf{C}}_{u}(\bm{v};\ell)}\le 1,\nonumber\\
     &\frac{T_{u}^{\uparrow,\mathsf{max}}(t_{x})}{\widehat{Y}^{\uparrow,\mathsf{max},\mathsf{C}}_{u}(\bm{v};\ell)}\le 1,~~~~\frac{\left(\mathscr{R}^{\uparrow,\mathsf{max}}_{u}(t_{x})\right)^{-1}\left(\sum_{r{\in}\mathcal{R}_{b}}\left(T_{u,r}^{\uparrow,\epsilon}(t_{x})\right)^{p}\right)}{T_{u}^{\uparrow,\mathsf{max}}(t_{x})}\le 1,\nonumber\\
     & \frac{T_{u,r}^{\uparrow,\epsilon}(t_{x})}{\widehat{Z}^{\uparrow,\epsilon}_{u,r}(\bm{v};\ell)}\le 1,~~~\frac{\left(\mathscr{S}^{\uparrow,\epsilon}_{u,r}(t_{x})\right)^{-1}\left(T_{u,r}^{\uparrow}(t_{x})+\epsilon\right)}{T_{u,r}^{\uparrow,\epsilon}(t_{x})}\le 1,\nonumber\\
     &\frac{1}{\mathscr{Q}^{\uparrow,\mathsf{C}}_{u}(t_{x})}\le 1,~~~~~\frac{1}{\mathscr{R}^{\uparrow,\mathsf{max}}_{u}(t_{x})}\le 1,~~~~\frac{1}{\mathscr{S}^{\uparrow,\epsilon}_{u,r}(t_{x})}\le 1.\nonumber
\end{tcolorbox}

\noindent\textbf{Waiting time of CHUs (Proposition~\ref{propo:waiting_time}).}  Referring to Sec.~\ref{sec:communication_latency}, let us revisit the equation for waiting time of CHU $u$ below:
\begin{equation}\label{app:eq:wating_time_1}
    \tau_{u}^{\mathsf{W},{(k)}}=\max_{u'\in\mathcal{U}_{b}}\Big\{ \tau_{u'}^{\mathsf{LC},{(k)}}{+}\sum_{x{=}N^{(k{-}1)}{+}1}^{N^{(k)}}\max_{r{\in}\overline{\mathcal{R}}_{b}}\left\{\min\big\{\overline{\tau}_{u',u,r}^{\uparrow}(t_{x}),\left(t_{x{+}1}{-}t_{x}\right)\big\}\right\}\Big\}.
\end{equation}
The above equation can be rewritten as follows:
\begin{equation}\label{app:eq:wating_time_2}
    \tau_{u}^{\mathsf{W},{(k)}}=\max_{u'\in\mathcal{U}_{b}}\Big\{ \tau_{u'}^{\mathsf{LC},{(k)}}{+}\sum_{x{=}N^{(k{-}1)}{+}1}^{N^{(k)}}\max_{r{\in}\overline{\mathcal{R}}_{b}}\left\{T_{u',u,r}^{\uparrow}(t_{x})\right\}\Big\},
\end{equation}
where $T_{u',u,r}^{\uparrow}(t_{x})$ is given in \eqref{app:eq:dispersion_min_1}. Defining an auxiliary decision variable $\tau_{u,u'}^{\uparrow,{(k)}}$ satisfying
\begin{equation}\label{app:eq:wating_time_3}
    \tau_{u',u}^{\uparrow,{(k)}}=\sum_{x{=}N^{(k{-}1)}{+}1}^{N^{(k)}}\max_{r{\in}\overline{\mathcal{R}}_{b}}\left\{T_{u',u,r}^{\uparrow}(t_{x})\right\}
\end{equation}
gives us
\begin{equation}\label{app:eq:wating_time_4}
    \tau_{u}^{\mathsf{W},{(k)}}=\max_{u'\in\mathcal{U}_{b}}\Big\{ \tau_{u'}^{\mathsf{LC},{(k)}}{+}\tau_{u',u}^{\uparrow,{(k)}}\Big\}.
\end{equation}
Using the approximation $\max\{A, B\}\approx (A^{p}+B^{p})^{-\frac{1}{p}}$, which is tight when $p \gg 1$, and performing some algebraic manipulations give us
\begin{equation}\label{app:eq:wating_time_5}
    \frac{\tau_{u}^{\mathsf{W},{(k)}}}{\left(\displaystyle\sum_{u'\in\mathcal{U}_{b}}\left(\tau_{u'}^{\mathsf{LC},{(k)}}{+}\tau_{u',u}^{\uparrow,{(k)}}\right)^{p}\right)^{-\frac{1}{p}}}= 1.
\end{equation}
Defining an auxiliary decision variable $T^{\uparrow,(k)}_{u}$ satisfying
\begin{equation}\label{app:eq:wating_time_6}
    T^{\uparrow,(k)}_{u} = \sum_{u'\in\mathcal{U}_{b}}\left(\tau_{u'}^{\mathsf{LC},{(k)}}{+}\tau_{u',u}^{\uparrow,{(k)}}\right)^{p}
\end{equation}
results in 
\begin{equation}\label{app:eq:wating_time_7}
    \left(T^{\uparrow,(k)}_{u}\right)^{\frac{1}{p}}\tau_{u}^{\mathsf{W},{(k)}}= 1.
\end{equation}
This constraint is not in the format of GP. Therefore, we transform it through splitting it into the following three inequalities:
\begin{equation}\label{app:eq:wating_time_6_1_4_4}
    \left(T^{\uparrow,(k)}_{u}\right)^{\frac{1}{p}}\tau_{u}^{\mathsf{W},{(k)}}\le 1,
\end{equation}
\begin{equation}\label{app:eq:wating_time_6_1_4_5}
    \frac{1}{\mathscr{J}^{\uparrow,(k)}_{u}\left(T^{\uparrow,(k)}_{u}\right)^{\frac{1}{p}}\tau_{u}^{\mathsf{W},{(k)}}}\le 1,
\end{equation}
\begin{equation}
    \mathscr{J}^{\uparrow,(k)}_{u}\ge 1,
\end{equation}
where $\mathscr{J}^{\uparrow,(k)}_{u}$ is an auxiliary decision variable added with a large penalty term to the objective function to force $\mathscr{J}^{\uparrow,(k)}_{u}{\rightarrow}1^+$ at the optimal point.

We finally approximate constraint~\eqref{app:eq:wating_time_7} as follows:
\begin{tcolorbox}[ams align]
     &\left(T^{\uparrow,(k)}_{u}\right)^{\frac{1}{p}}\tau_{u}^{\mathsf{W},{(k)}}\le 1,~~~~\frac{1}{\mathscr{J}^{\uparrow,(k)}_{u}\left(T^{\uparrow,(k)}_{u}\right)^{\frac{1}{p}}\tau_{u}^{\mathsf{W},{(k)}}}\le 1, \frac{1}{\mathscr{J}^{\uparrow,(k)}_{u}}\le 1.\nonumber
\end{tcolorbox}

We next transform \eqref{app:eq:wating_time_6} into the standard form of GP. To do so, we rewrite \eqref{app:eq:wating_time_6} as follows:
\begin{equation}\label{app:eq:wating_time_6_1_4_1}
    \frac{T^{\uparrow,(k)}_{u}}{\sum_{u'\in\mathcal{U}_{b}}\left(\tau_{u'}^{\mathsf{LC},{(k)}}{+}\tau_{u',u}^{\uparrow,{(k)}}\right)^{p}} = 1.
\end{equation}
Defining an auxiliary decision variable $T_{u',u}^{\uparrow,(k)}$ satisfying
\begin{equation}\label{app:eq:wating_time_6_1_4_2}
    T_{u',u}^{\uparrow,(k)}=\tau_{u'}^{\mathsf{LC},{(k)}}{+}\tau_{u',u}^{\uparrow,{(k)}}
\end{equation}
leads to
\begin{equation}\label{app:eq:wating_time_6_1_4_3}
    \frac{T^{\uparrow,(k)}_{u}}{\sum_{u'\in\mathcal{U}_{b}}\left(T_{u',u}^{\uparrow,(k)}\right)^{p}} =1.
\end{equation}
This constraint is not in the format of GP. Therefore, we transform it through splitting it into the following three inequalities:
\begin{equation}\label{app:eq:wating_time_6_1_4_4}
    \frac{T^{\uparrow,(k)}_{u}}{\sum_{u'\in\mathcal{U}_{b}}\left(T_{u',u}^{\uparrow,(k)}\right)^{p}}\le 1,
\end{equation}
\begin{equation}\label{app:eq:wating_time_6_1_4_5}
    \frac{\left(\mathscr{C}^{\uparrow,(k)}_{u',u}\right)^{-1}\left(\sum_{u'\in\mathcal{U}_{b}}\left(T_{u',u}^{\uparrow,(k)}\right)^{p}\right)}{T^{\uparrow,(k)}_{u}}\le 1,
\end{equation}
\begin{equation}
    \mathscr{C}^{\uparrow,(k)}_{u',u}\ge 1,
\end{equation}
where $\mathscr{C}^{\uparrow,(k)}_{u',u}$ is an auxiliary decision variable added with a large penalty term to the objective function to force $\mathscr{C}^{\uparrow,(k)}_{u',u}{\rightarrow}1^+$ at the optimal point. The fraction in~\eqref{app:eq:wating_time_6_1_4_4} still needs transformation since it is inequality with a posynomial in its denominator, which is not a posynomial. We thus exploit arithmetic-geometric mean inequality (Lemma~\ref{Lemma:ArethmaticGeometric}) to approximate the denominator in \eqref{app:eq:wating_time_6_1_4_4} with a monomial.
\begin{align}\label{app:eq:wating_time_6_1_4_6}
    &F^{\uparrow,(k)}_{u',u}(\bm{v})=\sum_{u'\in\mathcal{U}_{b}}\left(T_{u',u}^{\uparrow,(k)}\right)^{p}\geq \widehat{F}^{\uparrow,(k)}_{u',u}(\bm{v};\ell) \triangleq \prod_{u'\in\mathcal{U}_{b}}\left(\frac{\left(T_{u',u}^{\uparrow,(k)}\right)^{p} F^{\uparrow,(k)}_{u',u}([\bm{v}]^{(\ell-1)})}{\left[\left(T_{u',u}^{\uparrow,(k)}\right)^{p}\right]^{(\ell-1)}}\right)^{\frac{\left[\left(T_{u',u}^{\uparrow,(k)}\right)^{p}\right]^{(\ell-1)}}{F^{\uparrow,(k)}_{u',u}([\bm{v}]^{(\ell-1)})}},
\end{align}
which gives us an approximation of~\eqref{app:eq:wating_time_6_1_4_4} as follows:
\begin{equation}\label{app:eq:wating_time_6_1_4_6}
    \frac{T^{\uparrow,(k)}_{u}}{\widehat{F}^{\uparrow,(k)}_{u',u}(\bm{v};\ell)}\le 1.
\end{equation}

We finally approximate constraint~\eqref{app:eq:wating_time_6} as follows:
\begin{tcolorbox}[ams align]
     &\frac{T^{\uparrow,(k)}_{u}}{\widehat{F}^{\uparrow,(k)}_{u',u}(\bm{v};\ell)}\le 1,~~~~\frac{\left(\mathscr{C}^{\uparrow,(k)}_{u',u}\right)^{-1}\left(\sum_{u'\in\mathcal{U}_{b}}\left(T_{u',u}^{\uparrow,(k)}\right)^{p}\right)}{T^{\uparrow,(k)}_{u}}\le 1,~~~~~\frac{1}{\mathscr{C}^{\uparrow,(k)}_{u',u}}\le 1.\nonumber
\end{tcolorbox}

We next transform \eqref{app:eq:wating_time_6_1_4_2} into GP format. To this end, we rewrite \eqref{app:eq:wating_time_6_1_4_2} as follows:
\begin{equation}\label{app:eq:wating_time_6_1_4_7}
    \frac{T_{u',u}^{\uparrow,(k)}}{\tau_{u'}^{\mathsf{LC},{(k)}}{+}\tau_{u',u}^{\uparrow,{(k)}}}=1.
\end{equation}
This constraint is not in the format of GP. Therefore, we transform it through splitting it into the following three inequalities:
\begin{equation}\label{app:eq:wating_time_6_1_4_4}
    \frac{T_{u',u}^{\uparrow,(k)}}{\tau_{u'}^{\mathsf{LC},{(k)}}{+}\tau_{u',u}^{\uparrow,{(k)}}}\le 1,
\end{equation}
\begin{equation}\label{app:eq:wating_time_6_1_4_5}
    \frac{\left(\mathscr{H}^{\uparrow,(k)}_{u',u}\right)^{-1}\left(\tau_{u'}^{\mathsf{LC},{(k)}}{+}\tau_{u',u}^{\uparrow,{(k)}}\right)}{T_{u',u}^{\uparrow,(k)}}\le 1,
\end{equation}
\begin{equation}
    \mathscr{H}^{\uparrow,(k)}_{u',u}\ge 1,
\end{equation}
where $\mathscr{H}^{\uparrow,(k)}_{u',u}$ is an auxiliary decision variable added with a large penalty term to the objective function to force $\mathscr{H}^{\uparrow,(k)}_{u',u}{\rightarrow}1^+$ at the optimal point. The fraction in~\eqref{app:eq:wating_time_6_1_4_4} still needs transformation since it is inequality with a posynomial in its denominator, which is not a posynomial. We thus exploit arithmetic-geometric mean inequality (Lemma~\ref{Lemma:ArethmaticGeometric}) to approximate the denominator in \eqref{app:eq:wating_time_6_1_4_4} with a monomial.
\begin{align}\label{app:eq:wating_time_6_1_4_6}
    J^{\uparrow,(k)}_{u',u}(\bm{v})&=\tau_{u'}^{\mathsf{LC},{(k)}}{+}\tau_{u',u}^{\uparrow,{(k)}} \geq \widehat{J}^{\uparrow,(k)}_{u',u}(\bm{v};\ell) \triangleq \left(\frac{\tau_{u'}^{\mathsf{LC},{(k)}} J^{\uparrow,(k)}_{u',u}([\bm{v}]^{(\ell-1)})}{\left[\tau_{u'}^{\mathsf{LC},{(k)}}\right]^{(\ell-1)}}\right)^{\frac{\left[\tau_{u'}^{\mathsf{LC},{(k)}}\right]^{(\ell-1)}}{J^{\uparrow,(k)}_{u',u}([\bm{v}]^{(\ell-1)})}}\times \left(\frac{\tau_{u',u}^{\uparrow,{(k)}} J^{\uparrow,(k)}_{u',u}([\bm{v}]^{(\ell-1)})}{\left[\tau_{u',u}^{\uparrow,{(k)}}\right]^{(\ell-1)}}\right)^{\frac{\left[\tau_{u',u}^{\uparrow,{(k)}}\right]^{(\ell-1)}}{J^{\uparrow,(k)}_{u',u}([\bm{v}]^{(\ell-1)})}},
\end{align}
which gives us an approximation of~\eqref{app:eq:wating_time_6_1_4_4} as follows:
\begin{equation}\label{app:eq:wating_time_6_1_4_7}
    \frac{T_{u',u}^{\uparrow,(k)}}{\widehat{J}^{\uparrow,(k)}_{u',u}(\bm{v};\ell)}\le 1.
\end{equation}

We finally approximate constraint~\eqref{app:eq:wating_time_6_1_4_2} as follows:
\begin{tcolorbox}[ams align]
     &\frac{T_{u',u}^{\uparrow,(k)}}{\widehat{J}^{\uparrow,(k)}_{u',u}(\bm{v};\ell)}\le 1,~~~~\frac{\left(\mathscr{H}^{\uparrow,(k)}_{u',u}\right)^{-1}\left(\tau_{u'}^{\mathsf{LC},{(k)}}{+}\tau_{u',u}^{\uparrow,{(k)}}\right)}{T_{u',u}^{\uparrow,(k)}}\le 1,~~~~~\frac{1}{\mathscr{H}^{\uparrow,(k)}_{u',u}}\le 1.\nonumber
\end{tcolorbox}

We next aim to transform \eqref{app:eq:wating_time_3} into the standard form of GP. In doing so, we rewrite \eqref{app:eq:wating_time_3} as follows:
\begin{equation}\label{app:eq:wating_time_1_2}
\begin{aligned}
   \tau_{u',u}^{\uparrow,{(k)}}&=\sum_{x{=}0}^{N}\max_{r{\in}\overline{\mathcal{R}}_{b}}\Big\{T_{u',u,r}^{\uparrow}(t_{x})+\epsilon\Big\}-\sum_{x{=}0}^{N}\epsilon\\
   &=\sum_{x{=}0}^{N}\max_{r{\in}\overline{\mathcal{R}}_{b}}\Big\{T_{u',u,r}^{\uparrow}(t_{x})+\epsilon\Big\}-N\epsilon.
\end{aligned}
\end{equation}
Performing some algebraic manipulations gives us
\begin{equation}\label{app:eq:wating_time_1_2}
   \frac{\tau_{u',u}^{\uparrow,{(k)}}+N^{(k)}\epsilon}{\displaystyle\sum_{x{=}0}^{N}\max_{r{\in}\overline{\mathcal{R}}_{b}}\Big\{T_{u',u,r}^{\uparrow}(t_{x})+\epsilon\Big\}+N^{(k{-}1)}\epsilon}=1.
\end{equation} 
Using the approximation $\max\{A, B\}\approx (A^{p}+B^{p})^{-\frac{1}{p}}$, which is tight when $p \gg 1$, gives us
\begin{equation}\label{app:eq:wating_time_1_3}
   \frac{\tau_{u',u}^{\uparrow,{(k)}}+N^{(k)}\epsilon}{\displaystyle\sum_{x{=}0}^{N}\left(\sum_{r{\in}\overline{\mathcal{R}}_{b}}\left(T_{u',u,r}^{\uparrow}(t_{x})+\epsilon\right)^{p}\right)^{-\frac{1}{p}}+N^{(k{-}1)}\epsilon}=1.
\end{equation} 
Defining an auxiliary decision variable $T_{u',u}^{\uparrow,\mathsf{max}}(t_{x})$ that satisfies
\begin{equation}\label{app:eq:wating_time_1_4}
    T_{u',u}^{\uparrow,\mathsf{max}}(t_{x}) =\sum_{r{\in}\overline{\mathcal{R}}_{b}}\left(T_{u',u,r}^{\uparrow}(t_{x})+\epsilon\right)^{p}
\end{equation}
results in
\begin{equation}\label{app:eq:wating_time_1_5}
   \frac{\tau_{u',u}^{\uparrow,{(k)}}+N^{(k)}\epsilon}{\displaystyle\sum_{x{=}0}^{N}\left(T_{u',u}^{\uparrow,\mathsf{max}}(t_{x})\right)^{-\frac{1}{p}}+N^{(k{-}1)}\epsilon}=1.
\end{equation} 
This constraint is not in the format of GP. Therefore, we transform it through splitting it into the following three inequalities:
\begin{equation}\label{app:eq:wating_time_1_6}
    \frac{\tau_{u',u}^{\uparrow,{(k)}}+N^{(k)}\epsilon}{\displaystyle\sum_{x{=}0}^{N}\left(T_{u',u}^{\uparrow,\mathsf{max}}(t_{x})\right)^{-\frac{1}{p}}+N^{(k{-}1)}\epsilon}\le 1,
\end{equation}

\begin{equation}\label{app:eq:wating_time_1_7}
    \frac{\left(\mathscr{Q}^{\uparrow}_{u',u}(t_{x})\right)^{-1}\left(\displaystyle\sum_{x{=}0}^{N}\left(T_{u',u}^{\uparrow,\mathsf{max}}(t_{x})\right)^{-\frac{1}{p}}+N^{(k{-}1)}\epsilon\right)}{\tau_{u',u}^{\uparrow,{(k)}}+N^{(k)}\epsilon}\le 1,
\end{equation}

\begin{equation}
    \mathscr{Q}^{\uparrow}_{u',u}(t_{x})\ge 1,
\end{equation}
where $\mathscr{Q}^{\uparrow}_{u',u}(t_{x})$ is an auxiliary decision variable added with a large penalty term to the objective function to force $\mathscr{Q}^{\uparrow}_{u',u}(t_{x}){\rightarrow}1^+$ at the optimal point. The fractions in~\eqref{app:eq:wating_time_1_6} and \eqref{app:eq:wating_time_1_7} still need transformation since they are inequalities with posynomials in their denominators, which are not posynomials. We thus exploit arithmetic-geometric mean inequality (Lemma~\ref{Lemma:ArethmaticGeometric}) to approximate the denominators in \eqref{app:eq:wating_time_1_6} and \eqref{app:eq:wating_time_1_7} with monomials. To this end, we approximate the denominator in \eqref{app:eq:wating_time_1_6} as follows:
\begin{align}\label{app:eq:wating_time_1_8}
   &V^{\uparrow}_{u',u}(\bm{v})=\sum_{x{=}0}^{N}\left(T_{u',u}^{\uparrow,\mathsf{max}}(t_{x})\right)^{-\frac{1}{p}}+N^{(k{-}1)}\epsilon\nonumber\\
   &\geq \widehat{V}^{\uparrow}_{u',u}(\bm{v};\ell) \triangleq \prod_{x{=}0}^{N}\left(\frac{\left(T_{u',u}^{\uparrow,\mathsf{max}}(t_{x})\right)^{-\frac{1}{p}} V^{\uparrow}_{u',u}([\bm{v}]^{(\ell-1)})}{\left[\left(T_{u',u}^{\uparrow,\mathsf{max}}(t_{x})\right)^{-\frac{1}{p}}\right]^{(\ell-1)}}\right)^{\frac{\left[\left(T_{u',u}^{\uparrow,\mathsf{max}}(t_{x})\right)^{-\frac{1}{p}}\right]^{(\ell-1)}}{V^{\uparrow}_{u',u}([\bm{v}]^{(\ell-1)})}}\times\left(V^{\uparrow}_{u',u}([\bm{v}]^{(\ell-1)})\right)^{\frac{\left[N^{(k{-}1)}\displaystyle\epsilon\right]^{(\ell-1)}}{V^{\uparrow}_{u',u}([\bm{v}]^{(\ell-1)})}},
\end{align}
which gives us an approximation of~\eqref{app:eq:wating_time_1_6} as follows:
\begin{equation}\label{app:eq:wating_time_1_9}
    \frac{\tau_{u',u}^{\uparrow,{(k)}}+N^{(k)}\epsilon}{\widehat{V}^{\uparrow}_{u',u}(\bm{v};\ell)}\le 1.
\end{equation}

Similarly, we approximate the denominator in \eqref{app:eq:wating_time_1_7} as follows:
\begin{align}\label{app:eq:wating_time_1_10}
    W^{\uparrow}_{u',u}(\bm{v})&=\tau_{u',u}^{\uparrow,{(k)}}+N^{(k)}\epsilon \geq \widehat{W}^{\uparrow}_{u',u}(\bm{v};\ell) \triangleq \left(\frac{\tau_{u',u}^{\uparrow,{(k)}} W^{\uparrow}_{u',u}([\bm{v}]^{(\ell-1)})}{\left[\tau_{u',u}^{\uparrow,{(k)}}\right]^{(\ell-1)}}\right)^{\frac{\left[\tau_{u',u}^{\uparrow,{(k)}}\right]^{(\ell-1)}}{W^{\uparrow}_{u',u}([\bm{v}]^{(\ell-1)})}}\times\left(W^{\uparrow}_{u',u}([\bm{v}]^{(\ell-1)})\right)^{\frac{\left[N^{(k)}\displaystyle\epsilon\right]^{(\ell-1)}}{W^{\uparrow}_{u',u}([\bm{v}]^{(\ell-1)})}},
\end{align}
which gives us an approximation of~\eqref{app:eq:wating_time_1_7} as follows:
\begin{equation}\label{app:eq:wating_time_1_11}
    \frac{\left(\mathscr{Q}^{\uparrow}_{u',u}(t_{x})\right)^{-1}\left(\displaystyle\sum_{x{=}0}^{N}\left(T_{u',u}^{\uparrow,\mathsf{max}}(t_{x})\right)^{-\frac{1}{p}}+N^{(k{-}1)}\epsilon\right)}{\widehat{W}^{\uparrow}_{u',u}(\bm{v};\ell)}\le 1.
\end{equation}

We next transform \eqref{app:eq:wating_time_1_4} into the standard form of GP. To do so, we rewrite \eqref{app:eq:wating_time_1_4} as follows:
\begin{equation}\label{app:eq:wating_time_1_4_1}
    \frac{T_{u',u}^{\uparrow,\mathsf{max}}(t_{x})}{\sum_{r{\in}\overline{\mathcal{R}}_{b}}\left(T_{u',u,r}^{\uparrow}(t_{x})+\epsilon\right)^{p}} =1.
\end{equation}
Defining an auxiliary decision variable $T_{u',u,r}^{\uparrow,\epsilon}(t_{x})$ satisfying
\begin{equation}\label{app:eq:wating_time_1_4_2}
    T_{u',u,r}^{\uparrow,\epsilon}(t_{x})=T_{u',u,r}^{\uparrow}(t_{x})+\epsilon
\end{equation}
leads to
\begin{equation}\label{app:eq:wating_time_1_4_3}
    \frac{T_{u',u}^{\uparrow,\mathsf{max}}(t_{x})}{\sum_{r{\in}\overline{\mathcal{R}}_{b}}\left(T_{u',u,r}^{\uparrow,\epsilon}(t_{x})\right)^{p}} =1.
\end{equation}
This constraint is not in the format of GP. Therefore, we transform it through splitting it into the following three inequalities:
\begin{equation}\label{app:eq:wating_time_1_4_4}
    \frac{T_{u',u}^{\uparrow,\mathsf{max}}(t_{x})}{\sum_{r{\in}\overline{\mathcal{R}}_{b}}\left(T_{u',u,r}^{\uparrow,\epsilon}(t_{x})\right)^{p}}\le 1,
\end{equation}

\begin{equation}\label{app:eq:wating_time_1_4_5}
    \frac{\left(\mathscr{R}^{\uparrow,\mathsf{max}}_{u',u}(t_{x})\right)^{-1}\left(\sum_{r{\in}\overline{\mathcal{R}}_{b}}\left(T_{u',u,r}^{\uparrow,\epsilon}(t_{x})\right)^{p}\right)}{T_{u',u}^{\uparrow,\mathsf{max}}(t_{x})}\le 1,
\end{equation}

\begin{equation}
    \mathscr{R}^{\uparrow,\mathsf{max}}_{u',u}(t_{x})\ge 1,
\end{equation}
where $\mathscr{R}^{\uparrow,\mathsf{max}}_{u',u}(t_{x})$ is an auxiliary decision variable added with a large penalty term to the objective function to force $\mathscr{R}^{\uparrow,\mathsf{max}}_{u',u}(t_{x}){\rightarrow}1^+$ at the optimal point. The fraction in~\eqref{app:eq:wating_time_1_4_4} still needs transformation since it is inequality with a posynomial in its denominator, which is not a posynomial. We thus exploit arithmetic-geometric mean inequality (Lemma~\ref{Lemma:ArethmaticGeometric}) to approximate the denominator in \eqref{app:eq:wating_time_1_4_4} with a monomial.
\begin{align}\label{app:eq:wating_time_1_4_6}
    &Y^{\uparrow,\mathsf{max}}_{u',u}(\bm{v})=\sum_{r{\in}\overline{\mathcal{R}}_{b}}\left(T_{u',u,r}^{\uparrow,\epsilon}(t_{x})\right)^{p}\geq \widehat{Y}^{\uparrow,\mathsf{max}}_{u',u}(\bm{v};\ell) \triangleq \prod_{r{\in}\overline{\mathcal{R}}_{b}}\left(\frac{\left(T_{u',u,r}^{\uparrow,\epsilon}(t_{x})\right)^{p} Y^{\uparrow,\mathsf{max}}_{u',u}([\bm{v}]^{(\ell-1)})}{\left[\left(T_{u',u,r}^{\uparrow,\epsilon}(t_{x})\right)^{p}\right]^{(\ell-1)}}\right)^{\frac{\left[\left(T_{u',u,r}^{\uparrow,\epsilon}(t_{x})\right)^{p}\right]^{(\ell-1)}}{Y^{\uparrow,\mathsf{max}}_{u',u}([\bm{v}]^{(\ell-1)})}},
\end{align}
which gives us an approximation of~\eqref{app:eq:wating_time_1_4_4} as follows:
\begin{equation}\label{app:eq:wating_time_1_4_6}
    \frac{T_{u',u}^{\uparrow,\mathsf{max}}(t_{x})}{\widehat{Y}^{\uparrow,\mathsf{max}}_{u',u}(\bm{v};\ell)}\le 1.
\end{equation}

We next transform \eqref{app:eq:wating_time_1_4_2} into GP format. To this end, we rewrite \eqref{app:eq:wating_time_1_4_2} as follows:
\begin{equation}\label{app:eq:wating_time_1_4_7}
    \frac{T_{u',u,r}^{\uparrow,\epsilon}(t_{x})}{T_{u',u,r}^{\uparrow}(t_{x})+\epsilon}=1.
\end{equation}
This constraint is not in the format of GP. Therefore, we transform it through splitting it into the following three inequalities:
\begin{equation}\label{app:eq:wating_time_1_4_4}
    \frac{T_{u',u,r}^{\uparrow,\epsilon}(t_{x})}{T_{u',u,r}^{\uparrow}(t_{x})+\epsilon}\le 1,
\end{equation}

\begin{equation}\label{app:eq:wating_time_1_4_5}
    \frac{\left(\mathscr{S}^{\uparrow,\epsilon}_{u',u,r}(t_{x})\right)^{-1}\left(T_{u',u,r}^{\uparrow}(t_{x})+\epsilon\right)}{T_{u',u,r}^{\uparrow,\epsilon}(t_{x})}\le 1,
\end{equation}

\begin{equation}
    \mathscr{S}^{\uparrow,\epsilon}_{u',u,r}(t_{x})\ge 1,
\end{equation}
where $\mathscr{S}^{\uparrow,\epsilon}_{u',u,r}(t_{x})$ is an auxiliary decision variable added with a large penalty term to the objective function to force $\mathscr{S}^{\uparrow,\epsilon}_{u',u,r}(t_{x}){\rightarrow}1^+$ at the optimal point. The fraction in~\eqref{app:eq:wating_time_1_4_4} still needs transformation since it is inequality with a posynomial in its denominator, which is not a posynomial. We thus exploit arithmetic-geometric mean inequality (Lemma~\ref{Lemma:ArethmaticGeometric}) to approximate the denominator in \eqref{app:eq:wating_time_1_4_4} with a monomial.
\begin{align}\label{app:eq:wating_time_1_4_6}
    Z^{\uparrow,\epsilon}_{u',u,r}(\bm{v})&=T_{u',u,r}^{\uparrow}(t_{x})+\epsilon \geq \widehat{Z}^{\uparrow,\epsilon}_{u',u,r}(\bm{v};\ell) \triangleq \left(\frac{T_{u',u,r}^{\uparrow}(t_{x}) Z^{\uparrow,\epsilon}_{u',u,r}([\bm{v}]^{(\ell-1)})}{\left[T_{u',u,r}^{\uparrow}(t_{x})\right]^{(\ell-1)}}\right)^{\frac{\left[T_{u',u,r}^{\uparrow}(t_{x})\right]^{(\ell-1)}}{Z^{\uparrow,\epsilon}_{u',u,r}([\bm{v}]^{(\ell-1)})}}\times\left(Z^{\uparrow,\epsilon}_{u',u,r}([\bm{v}]^{(\ell-1)})\right)^{\frac{\left[\displaystyle\epsilon\right]^{(\ell-1)}}{Z^{\uparrow,\epsilon}_{u',u,r}([\bm{v}]^{(\ell-1)})}},
\end{align}
which gives us an approximation of~\eqref{app:eq:wating_time_1_4_4} as follows:
\begin{equation}\label{app:eq:wating_time_1_4_7}
    \frac{T_{u',u,r}^{\uparrow,\epsilon}(t_{x})}{\widehat{Z}^{\uparrow,\epsilon}_{u',u,r}(\bm{v};\ell)}\le 1.
\end{equation}

We finally approximate constraint~\eqref{app:eq:wating_time_3} as follows:
\begin{tcolorbox}[ams align]
     &\frac{\tau_{u',u}^{\uparrow,{(k)}}+N^{(k)}\epsilon}{\widehat{V}^{\uparrow}_{u',u}(\bm{v};\ell)}\le 1,~~~~~~\frac{\left(\mathscr{Q}^{\uparrow}_{u',u}(t_{x})\right)^{-1}\left(\displaystyle\sum_{x{=}0}^{N}\left(T_{u',u}^{\uparrow,\mathsf{max}}(t_{x})\right)^{-\frac{1}{p}}+N^{(k{-}1)}\epsilon\right)}{\widehat{W}^{\uparrow}_{u',u}(\bm{v};\ell)}\le 1,\nonumber\\
     &\frac{T_{u',u}^{\uparrow,\mathsf{max}}(t_{x})}{\widehat{Y}^{\uparrow,\mathsf{max}}_{u',u}(\bm{v};\ell)}\le 1,~~~~\frac{\left(\mathscr{R}^{\uparrow,\mathsf{max}}_{u',u}(t_{x})\right)^{-1}\left(\sum_{r{\in}\overline{\mathcal{R}}_{b}}\left(T_{u',u,r}^{\uparrow,\epsilon}(t_{x})\right)^{p}\right)}{T_{u',u}^{\uparrow,\mathsf{max}}(t_{x})}\le 1,\nonumber\\
     & \frac{T_{u',u,r}^{\uparrow,\epsilon}(t_{x})}{\widehat{Z}^{\uparrow,\epsilon}_{u',u,r}(\bm{v};\ell)}\le 1,~~~\frac{\left(\mathscr{S}^{\uparrow,\epsilon}_{u',u,r}(t_{x})\right)^{-1}\left(T_{u',u,r}^{\uparrow}(t_{x})+\epsilon\right)}{T_{u',u,r}^{\uparrow,\epsilon}(t_{x})}\le 1,\nonumber\\
     &\frac{1}{\mathscr{Q}^{\uparrow}_{u',u}(t_{x})}\le 1,~~~~\frac{1}{\mathscr{R}^{\uparrow,\mathsf{max}}_{u',u}(t_{x})}\le 1,~~~~\frac{1}{\mathscr{S}^{\uparrow,\epsilon}_{u',u,r}(t_{x})}\le 1.\nonumber
\end{tcolorbox}

\textbf{Completion time of global round $k$.} Referring to Sec.~\ref{sec:communication_latency}, let us revisit the equation for completion time of global round $k$ below:
\begin{equation}\label{app:cons:completion_time_1}
    T^{(k)}= T^{(k{-}1)}+\max_{b\in \Omega,u\in \mathcal{U}_{b}}\{\Psi(\mathscr{D}_{u}^{\nuparrow,(k)})\},
\end{equation}
where 
\begin{equation}\label{app:cons:completion_time_2}
    \Psi(\mathscr{D}_{u}^{\nuparrow,(k)}){=}\lambda_{u}^{(k)}\left(\tau_{b}^{\downarrow,{(k)}}{+}\max\{\tau_{u}^{\mathsf{LC},{(k)}}, \tau_{u}^{\mathsf{W},{(k)}}\}{+} \tau_{u}^{\uparrow,{(k)}}\right).
\end{equation}

In the following we first transform \eqref{app:cons:completion_time_2} into standard GP format. In doing so, we define $\Psi(\mathscr{D}_{u}^{\nuparrow,(k)})$ as an auxiliary decision variable that must satisfy the following constraint (derived through performing some algebraic manipulations on \eqref{app:cons:completion_time_2}):
\begin{equation}\label{app:cons:completion_time_3}
    \frac{\Psi(\mathscr{D}_{u}^{\nuparrow,(k)})}{\lambda_{u}^{(k)}\tau_{b}^{\downarrow,{(k)}}{+}\max\{\lambda_{u}^{(k)}\tau_{u}^{\mathsf{LC},{(k)}}+\epsilon, \lambda_{u}^{(k)}\tau_{u}^{\mathsf{W},{(k)}}+\epsilon\}{+} \lambda_{u}^{(k)}\tau_{u}^{\uparrow,{(k)}}}{=}1.
\end{equation}
Subsequently, we rewrite \eqref{app:cons:completion_time_1} as follows:
\begin{equation}\label{app:cons:completion_time_4}
    T^{(k)}+\epsilon= T^{(k{-}1)}+\max_{b\in \Omega,u\in \mathcal{U}_{b}}\{\Psi(\mathscr{D}_{u}^{\nuparrow,(k)})\},
\end{equation}
where $\epsilon$ is added to \eqref{app:cons:completion_time_3} and \eqref{app:cons:completion_time_4} to avoid division by zero. Using the approximation $\max\{A, B\}\approx (A^{p}+B^{p})^{-\frac{1}{p}}$, which is tight when $p \gg 1$, gives us
\begin{equation}\label{app:cons:completion_time_5}
    \frac{\Psi(\mathscr{D}_{u}^{\nuparrow,(k)})}{\lambda_{u}^{(k)}\tau_{b}^{\downarrow,{(k)}} + \left(\left(\lambda_{u}^{(k)}\tau_{u}^{\mathsf{LC},{(k)}} + \epsilon\right)^{p}+ \left(\lambda_{u}^{(k)}\tau_{u}^{\mathsf{W},{(k)}}+\epsilon\right)^{p}\right)^{-\frac{1}{p}} + \lambda_{u}^{(k)}\tau_{u}^{\uparrow,{(k)}}}{=}1.
\end{equation}
Defining an auxiliary decision variable $T^{\mathsf{LW},(k)}_{u}$ satisfying 
\begin{equation}\label{app:cons:completion_time_6}
    T^{\mathsf{LW},(k)}_{u}=\left(\lambda_{u}^{(k)}\tau_{u}^{\mathsf{LC},{(k)}} + \epsilon\right)^{p}+ \left(\lambda_{u}^{(k)}\tau_{u}^{\mathsf{W},{(k)}}+\epsilon\right)^{p}
\end{equation}
gives us
\begin{equation}\label{app:cons:completion_time_7}
    \frac{\Psi(\mathscr{D}_{u}^{\nuparrow,(k)})}{\lambda_{u}^{(k)}\tau_{b}^{\downarrow,{(k)}} + \left(T^{\mathsf{LW},(k)}_{u}\right)^{-\frac{1}{p}} + \lambda_{u}^{(k)}\tau_{u}^{\uparrow,{(k)}}}{=}1.
\end{equation}
This constraint is not in the format of GP. Therefore, we transform it via introducing an auxiliary variable as two inequalities:
\begin{equation}\label{app:cons:completion_time_8}
    \frac{\Psi(\mathscr{D}_{u}^{\nuparrow,(k)})}{\lambda_{u}^{(k)}\tau_{b}^{\downarrow,{(k)}} + \left(T^{\mathsf{LW},(k)}_{u}\right)^{-\frac{1}{p}} + \lambda_{u}^{(k)}\tau_{u}^{\uparrow,{(k)}}}\le 1,
\end{equation}
\begin{equation}\label{app:cons:completion_time_9}
    \frac{\lambda_{u}^{(k)}\tau_{b}^{\downarrow,{(k)}} + \left(T^{\mathsf{LW},(k)}_{u}\right)^{-\frac{1}{p}} + \lambda_{u}^{(k)}\tau_{u}^{\uparrow,{(k)}}}{\mathscr{B}^{\nuparrow,(k)}_{u}\Psi(\mathscr{D}_{u}^{\nuparrow,(k)})}\le 1,
\end{equation}
\begin{equation}
    \mathscr{B}^{\nuparrow,(k)}_{u}\ge 1,
\end{equation}
where $\mathscr{B}^{\nuparrow,(k)}_{u}$ is added with a large penalty term to the objective function to force $\mathscr{B}^{\nuparrow,(k)}_{u}{\rightarrow}1^+$ at the optimal point. The fraction in~\eqref{app:cons:completion_time_8} still needs transformation since it is an inequality with a posynomial in the denominator, which is not a posynomial. We thus exploit arithmetic-geometric mean inequality (Lemma~\ref{Lemma:ArethmaticGeometric}) to approximate the denominator with a monomial:
\begin{align}\label{app:cons:completion_time_10}
    F^{\nuparrow,(k)}_{u}(\bm{v})&=\lambda_{u}^{(k)}\tau_{b}^{\downarrow,{(k)}} + \left(T^{\mathsf{LW},(k)}_{u}\right)^{-\frac{1}{p}} + \lambda_{u}^{(k)}\tau_{u}^{\uparrow,{(k)}}\geq \widehat{F}^{\nuparrow,(k)}_{u}(\bm{v};\ell) \triangleq \left(\frac{\lambda_{u}^{(k)}\tau_{b}^{\downarrow,{(k)}} F^{\nuparrow,(k)}_{u}([\bm{v}]^{(\ell-1)})}{\left[\lambda_{u}^{(k)}\tau_{b}^{\downarrow,{(k)}}\right]^{(\ell-1)}}\right)^{\frac{\left[\lambda_{u}^{(k)}\tau_{b}^{\downarrow,{(k)}}\right]^{(\ell-1)}}{F^{\nuparrow,(k)}_{u}([\bm{v}]^{(\ell-1)})}}\nonumber\\
    &\times \left(\frac{\left(T^{\mathsf{LW},(k)}_{u}\right)^{-\frac{1}{p}} F^{\nuparrow,(k)}_{u}([\bm{v}]^{(\ell-1)})}{\left[\left(T^{\mathsf{LW},(k)}_{u}\right)^{-\frac{1}{p}}\right]^{(\ell-1)}}\right)^{\frac{\left[\left(T^{\mathsf{LW},(k)}_{u}\right)^{-\frac{1}{p}}\right]^{(\ell-1)}}{F^{\nuparrow,(k)}_{u}([\bm{v}]^{(\ell-1)})}}\times \left(\frac{\lambda_{u}^{(k)}\tau_{u}^{\uparrow,{(k)}} F^{\nuparrow,(k)}_{u}([\bm{v}]^{(\ell-1)})}{\left[\lambda_{u}^{(k)}\tau_{u}^{\uparrow,{(k)}}\right]^{(\ell-1)}}\right)^{\frac{\left[\lambda_{u}^{(k)}\tau_{u}^{\uparrow,{(k)}}\right]^{(\ell-1)}}{F^{\nuparrow,(k)}_{u}([\bm{v}]^{(\ell-1)})}}
\end{align}
which gives us an approximation of~\eqref{app:cons:completion_time_8} as follows:
\begin{equation}\label{app:cons:completion_time_11}
    \frac{\Psi(\mathscr{D}_{u}^{\nuparrow,(k)})}{\widehat{F}^{\nuparrow,(k)}_{u}(\bm{v};\ell)}\le 1.
\end{equation}

We finally approximate constraint~\eqref{app:cons:completion_time_7} as follows:
\begin{tcolorbox}[ams align]
     &\frac{\Psi(\mathscr{D}_{u}^{\nuparrow,(k)})}{\widehat{F}^{\nuparrow,(k)}_{u}(\bm{v};\ell)}\le 1,~~~~~ \frac{\lambda_{u}^{(k)}\tau_{b}^{\downarrow,{(k)}} + \left(T^{\mathsf{LW},(k)}_{u}\right)^{-\frac{1}{p}} + \lambda_{u}^{(k)}\tau_{u}^{\uparrow,{(k)}}}{\mathscr{B}^{\nuparrow,(k)}_{u}\Psi(\mathscr{D}_{u}^{\nuparrow,(k)})}\le 1,~~~~~\frac{1}{\mathscr{B}^{\nuparrow,(k)}_{u}}\le 1.\nonumber
\end{tcolorbox}

We next aim to transform \eqref{app:cons:completion_time_6} into standard form of GP. In doing so, let us rewrite \eqref{app:cons:completion_time_6} as follows:
\begin{equation}\label{app:cons:completion_time_6_1}
    \frac{T^{\mathsf{LW},(k)}_{u}}{\left(\lambda_{u}^{(k)}\tau_{u}^{\mathsf{LC},{(k)}} + \epsilon\right)^{p}+ \left(\lambda_{u}^{(k)}\tau_{u}^{\mathsf{W},{(k)}}+\epsilon\right)^{p}}=1.
\end{equation}
Defining two auxiliary decision variables $T^{\mathsf{LC},(k)}_{u}$ and $T^{\mathsf{W},(k)}_{u}$ satisfying
\begin{equation}\label{app:cons:completion_time_6_2}
    T^{\mathsf{LC},(k)}_{u}=\lambda_{u}^{(k)}\tau_{u}^{\mathsf{LC},{(k)}} + \epsilon
\end{equation}
and 
\begin{equation}\label{app:cons:completion_time_6_3}
    T^{\mathsf{W},(k)}_{u}=\lambda_{u}^{(k)}\tau_{u}^{\mathsf{W},{(k)}}+\epsilon
\end{equation}
result in
\begin{equation}\label{app:cons:completion_time_6_4}
    \frac{T^{\mathsf{LW},(k)}_{u}}{\left(T^{\mathsf{LC},(k)}_{u}\right)^{p}+ \left(T^{\mathsf{W},(k)}_{u}\right)^{p}}=1.
\end{equation}
This constraint is not in the format of GP. Therefore, we transform it via introducing an auxiliary variable as two inequalities:
\begin{equation}\label{app:cons:completion_time_6_4_8}
    \frac{T^{\mathsf{LW},(k)}_{u}}{\left(T^{\mathsf{LC},(k)}_{u}\right)^{p}+ \left(T^{\mathsf{W},(k)}_{u}\right)^{p}}\le 1,
\end{equation}
\begin{equation}\label{app:cons:completion_time_6_4_9}
    \frac{\left(T^{\mathsf{LC},(k)}_{u}\right)^{p}+ \left(T^{\mathsf{LC},(k)}_{u}\right)^{p}}{\mathscr{C}^{\mathsf{W},(k)}_{u}T^{\mathsf{LW},(k)}_{u}}\le 1,
\end{equation}
\begin{equation}
    \mathscr{C}^{\mathsf{LW},(k)}_{u}\ge 1,
\end{equation}
where $\mathscr{C}^{\mathsf{LW},(k)}_{u}$ is added with a large penalty term to the objective function to force $\mathscr{C}^{\mathsf{LW},(k)}_{u}{\rightarrow}1^+$ at the optimal point. The fraction in~\eqref{app:cons:completion_time_6_4_8} still needs transformation since it is an inequality with a posynomial in the denominator, which is not a posynomial. We thus exploit arithmetic-geometric mean inequality (Lemma~\ref{Lemma:ArethmaticGeometric}) to approximate the denominator with a monomial:
\begin{align}\label{app:cons:completion_time_6_4_10}
    G^{\mathsf{LW},(k)}_{u}(\bm{v})=\left(T^{\mathsf{LC},(k)}_{u}\right)^{p}+ \left(T^{\mathsf{W},(k)}_{u}\right)^{p}\geq \widehat{G}^{\mathsf{LW},(k)}_{u}(\bm{v};\ell) &\triangleq \left(\frac{\left(T^{\mathsf{LC},(k)}_{u}\right)^{p} G^{\mathsf{LW},(k)}_{u}([\bm{v}]^{(\ell-1)})}{\left[\left(T^{\mathsf{LC},(k)}_{u}\right)^{p}\right]^{(\ell-1)}}\right)^{\frac{\left[\left(T^{\mathsf{LC},(k)}_{u}\right)^{p}\right]^{(\ell-1)}}{G^{\mathsf{LW},(k)}_{u}([\bm{v}]^{(\ell-1)})}}\nonumber\\
    &\times \left(\frac{\left(T^{\mathsf{W},(k)}_{u}\right)^{p} G^{\mathsf{LW},(k)}_{u}([\bm{v}]^{(\ell-1)})}{\left[\left(T^{\mathsf{LC},(k)}_{u}\right)^{p}\right]^{(\ell-1)}}\right)^{\frac{\left[\left(T^{\mathsf{LC},(k)}_{u}\right)^{p}\right]^{(\ell-1)}}{G^{\mathsf{LW},(k)}_{u}([\bm{v}]^{(\ell-1)})}},\nonumber
\end{align}
which gives us an approximation of~\eqref{app:cons:completion_time_6_4_8} as follows:
\begin{equation}\label{app:cons:completion_time_6_4_11}
    \frac{\Psi(\mathscr{D}_{u}^{\nuparrow,(k)})}{\widehat{G}^{\mathsf{LW},(k)}_{u}(\bm{v};\ell)}\le 1.
\end{equation}

We finally approximate constraint~\eqref{app:cons:completion_time_6_4} as follows:
\begin{tcolorbox}[ams align]
     &\frac{\Psi(\mathscr{D}_{u}^{\nuparrow,(k)})}{\widehat{G}^{\mathsf{LW},(k)}_{u}(\bm{v};\ell)}\le 1,~~~~~\frac{\left(T^{\mathsf{LC},(k)}_{u}\right)^{p}+ \left(T^{\mathsf{LW},(k)}_{u}\right)^{p}}{\mathscr{C}^{\mathsf{LW},(k)}_{u}T^{\mathsf{LW},(k)}_{u}}\le 1,~~~~~\frac{1}{\mathscr{C}^{\mathsf{LW},(k)}_{u}}\le 1.\nonumber
\end{tcolorbox}

We next aim to transform \eqref{app:cons:completion_time_6_2} into GP format. To do so, we rewrite \eqref{app:cons:completion_time_6_2} as follows:
\begin{equation}\label{app:cons:completion_time_6_2_1}
    \frac{T^{\mathsf{LC},(k)}_{u}}{\lambda_{u}^{(k)}\tau_{u}^{\mathsf{LC},{(k)}} + \epsilon}= 1.
\end{equation}

This constraint is not in the format of GP. Therefore, we transform it via introducing an auxiliary variable as two inequalities:
\begin{equation}\label{app:cons:completion_time_6_2_1_8}
    \frac{T^{\mathsf{LC},(k)}_{u}}{\lambda_{u}^{(k)}\tau_{u}^{\mathsf{LC},{(k)}} + \epsilon}\le 1,
\end{equation}
\begin{equation}\label{app:cons:completion_time_6_2_1_9}
    \frac{\lambda_{u}^{(k)}\tau_{u}^{\mathsf{LC},{(k)}} + \epsilon}{\mathscr{F}^{\mathsf{L},(k)}_{u} T^{\mathsf{LC},(k)}_{u}}\le 1,
\end{equation}
\begin{equation}
    \mathscr{F}^{\mathsf{L},(k)}_{u}\ge 1,
\end{equation}
where $\mathscr{F}^{\mathsf{L},(k)}_{u}$ is added with a large penalty term to the objective function to force $\mathscr{F}^{\mathsf{L},(k)}_{u}{\rightarrow}1^+$ at the optimal point. The fraction in~\eqref{app:cons:completion_time_6_2_1_8} still needs transformation since it is an inequality with a posynomial in the denominator, which is not a posynomial. We thus exploit arithmetic-geometric mean inequality (Lemma~\ref{Lemma:ArethmaticGeometric}) to approximate the denominator with a monomial:
\begin{align}\label{app:cons:completion_time_6_2_1_10}
    H^{\mathsf{LC},(k)}_{u}(\bm{v})=\lambda_{u}^{(k)}\tau_{u}^{\mathsf{LC},{(k)}} + \epsilon\geq \widehat{H}^{\mathsf{LC},(k)}_{u}(\bm{v};\ell) &\triangleq \left(\frac{\lambda_{u}^{(k)}\tau_{u}^{\mathsf{LC},{(k)}} H^{\mathsf{LC},(k)}_{u}([\bm{v}]^{(\ell-1)})}{\left[\lambda_{u}^{(k)}\tau_{u}^{\mathsf{LC},{(k)}}\right]^{(\ell-1)}}\right)^{\frac{\left[\lambda_{u}^{(k)}\tau_{u}^{\mathsf{LC},{(k)}}\right]^{(\ell-1)}}{H^{\mathsf{LC},(k)}_{u}([\bm{v}]^{(\ell-1)})}}\nonumber\\
    &\times \left(H^{\mathsf{LC},(k)}_{u}([\bm{v}]^{(\ell-1)})\right)^{\frac{\left[\displaystyle \epsilon \right]^{(\ell-1)}}{H^{\mathsf{LC},(k)}_{u}([\bm{v}]^{(\ell-1)})}},
\end{align}
which gives us an approximation of~\eqref{app:cons:completion_time_6_2_1_8} as follows:
\begin{equation}\label{app:cons:completion_time_6_2_1_11}
    \frac{T^{\mathsf{LC},(k)}_{u}}{\widehat{H}^{\mathsf{LC},(k)}_{u}(\bm{v};\ell)}\le 1.
\end{equation}

We finally approximate constraint~\eqref{app:cons:completion_time_6_2_1} as follows:
\begin{tcolorbox}[ams align]
     &\frac{T^{\mathsf{LC},(k)}_{u}}{\widehat{H}^{\mathsf{LC},(k)}_{u}(\bm{v};\ell)}\le 1,~~~~~\frac{\lambda_{u}^{(k)}\tau_{u}^{\mathsf{LC},{(k)}} + \epsilon}{\mathscr{F}^{\mathsf{L},(k)}_{u} T^{\mathsf{LC},(k)}_{u}}\le 1,~~~~~\frac{1}{\mathscr{F}^{\mathsf{L},(k)}_{u}}\le 1.\nonumber
\end{tcolorbox}

We next transform \eqref{app:cons:completion_time_6_3} into GP format. To do so, we rewrite \eqref{app:cons:completion_time_6_3} as follows:
\begin{equation}\label{app:cons:completion_time_6_3_1}
    \frac{T^{\mathsf{W},(k)}_{u}}{\lambda_{u}^{(k)}\tau_{u}^{\mathsf{W},{(k)}}+\epsilon}=1.
\end{equation}
This constraint is not in the format of GP. Therefore, we transform it via introducing an auxiliary variable as two inequalities:
\begin{equation}\label{app:cons:completion_time_6_3_1_8}
    \frac{T^{\mathsf{W},(k)}_{u}}{\lambda_{u}^{(k)}\tau_{u}^{\mathsf{W},{(k)}}+\epsilon}\le 1,
\end{equation}
\begin{equation}\label{app:cons:completion_time_6_3_1_9}
    \frac{\lambda_{u}^{(k)}\tau_{u}^{\mathsf{W},{(k)}}+\epsilon}{\mathscr{H}^{\mathsf{W},(k)}_{u}T^{\mathsf{W},(k)}_{u}}\le 1,
\end{equation}
\begin{equation}
    \mathscr{H}^{\mathsf{W},(k)}_{u}\ge 1,
\end{equation}
where $\mathscr{H}^{\mathsf{W},(k)}_{u}$ is added with a large penalty term to the objective function to force $\mathscr{H}^{\mathsf{W},(k)}_{u}{\rightarrow}1^+$ at the optimal point. The fraction in~\eqref{app:cons:completion_time_6_3_1_8} still needs transformation since it is an inequality with a posynomial in the denominator, which is not a posynomial. We thus exploit arithmetic-geometric mean inequality (Lemma~\ref{Lemma:ArethmaticGeometric}) to approximate the denominator with a monomial:
\begin{align}\label{app:cons:completion_time_6_3_1_10}
    J^{\mathsf{W},(k)}_{u}(\bm{v})= \lambda_{u}^{(k)}\tau_{u}^{\mathsf{W},{(k)}}+\epsilon \geq \widehat{J}^{\mathsf{W},(k)}_{u}(\bm{v};\ell) &\triangleq \left(\frac{\lambda_{u}^{(k)}\tau_{u}^{\mathsf{W},{(k)}} J^{\mathsf{W},(k)}_{u}([\bm{v}]^{(\ell-1)})}{\left[\lambda_{u}^{(k)}\tau_{u}^{\mathsf{W},{(k)}}\right]^{(\ell-1)}}\right)^{\frac{\left[\lambda_{u}^{(k)}\tau_{u}^{\mathsf{W},{(k)}}\right]^{(\ell-1)}}{J^{\mathsf{W},(k)}_{u}([\bm{v}]^{(\ell-1)})}}\nonumber\\
    &\times \left(J^{\mathsf{W},(k)}_{u}([\bm{v}]^{(\ell-1)})\right)^{\frac{\left[\displaystyle \epsilon \right]^{(\ell-1)}}{J^{\mathsf{W},(k)}_{u}([\bm{v}]^{(\ell-1)})}},
\end{align}
which gives us an approximation of~\eqref{app:cons:completion_time_6_3_1_8} as follows:
\begin{equation}\label{app:cons:completion_time_6_3_1_11}
    \frac{T^{\mathsf{W},(k)}_{u}}{\widehat{J}^{\mathsf{W},(k)}_{u}(\bm{v};\ell)}\le 1.
\end{equation}

We finally approximate constraint~\eqref{app:cons:completion_time_6_3_1} as follows:
\begin{tcolorbox}[ams align]
     &\frac{T^{\mathsf{W},(k)}_{u}}{\widehat{J}^{\mathsf{W},(k)}_{u}(\bm{v};\ell)}\le 1,~~~~~\frac{\lambda_{u}^{(k)}\tau_{u}^{\mathsf{W},{(k)}} + \epsilon}{\mathscr{H}^{\mathsf{W},(k)}_{u} T^{\mathsf{W},(k)}_{u}}\le 1,~~~~~\frac{1}{\mathscr{H}^{\mathsf{W},(k)}_{u}}\le 1.\nonumber
\end{tcolorbox}

We next transform \eqref{app:cons:completion_time_1} into standard GP format. To do so, we define $T^{(k)}$ as an auxiliary decision variable satisfying 
\begin{equation}\label{app:cons:completion_time_1_1}
    \frac{T^{(k)}}{T^{(k{-}1)}+\max_{b\in \Omega,u\in \mathcal{U}_{b}}\{\Psi(\mathscr{D}_{u}^{\nuparrow,(k)})\}}= 1,
\end{equation}
Using the approximation $\max\{A, B\}\approx (A^{p}+B^{p})^{-\frac{1}{p}}$, which is tight when $p \gg 1$, gives us
\begin{equation}\label{app:cons:completion_time_1_5}
    \frac{T^{(k)}}{T^{(k{-}1)}+\left(\sum_{b\in \Omega}\sum_{u\in \mathcal{U}_{b}}\left(\Psi(\mathscr{D}_{u}^{\nuparrow,(k)})\right)^{p}\right)^{-\frac{1}{p}}}= 1,
\end{equation}
Defining an auxiliary decision variable $T^{(k)}_{u}$ satisfying 
\begin{equation}\label{app:cons:completion_time_1_6}
    T^{(k)}_{u}=\sum_{b\in \Omega}\sum_{u\in \mathcal{U}_{b}}\left(\Psi(\mathscr{D}_{u}^{\nuparrow,(k)})\right)^{p}
\end{equation}
gives us
\begin{equation}\label{app:cons:completion_time_1_7}
    \frac{T^{(k)}}{T^{(k{-}1)}+\left(T^{(k)}_{u}\right)^{-\frac{1}{p}}}= 1.
\end{equation}
This constraint is not in the format of GP. Therefore, we transform it via introducing an auxiliary variable as two inequalities:
\begin{equation}\label{app:cons:completion_time_1_8}
    \frac{T^{(k)}}{T^{(k{-}1)}+\left(T^{(k)}_{u}\right)^{-\frac{1}{p}}}\le 1,
\end{equation}
\begin{equation}\label{app:cons:completion_time_1_9}
    \frac{T^{(k{-}1)}+\left(T^{(k)}_{u}\right)^{-\frac{1}{p}}}{\mathscr{G}^{(k)}_{u}T^{(k)}}\le 1,
\end{equation}
\begin{equation}
    \mathscr{G}^{(k)}_{u}\ge 1,
\end{equation}
where $\mathscr{G}^{(k)}_{u}$ is added with a large penalty term to the objective function to force $\mathscr{G}^{(k)}_{u}{\rightarrow}1^+$ at the optimal point. The fraction in~\eqref{app:cons:completion_time_1_8} still needs transformation since it is an inequality with a posynomial in the denominator, which is not a posynomial. We thus exploit arithmetic-geometric mean inequality (Lemma~\ref{Lemma:ArethmaticGeometric}) to approximate the denominator with a monomial:
\begin{align}\label{app:cons:completion_time_1_10}
    L^{(k)}_{u}(\bm{v})&=T^{(k{-}1)}+\left(T^{(k)}_{u}\right)^{-\frac{1}{p}}\geq \widehat{L}^{(k)}_{u}(\bm{v};\ell) \triangleq \left(\frac{T^{(k{-}1)} L^{(k)}_{u}([\bm{v}]^{(\ell-1)})}{\left[T^{(k{-}1)}\right]^{(\ell-1)}}\right)^{\frac{\left[T^{(k{-}1)}\right]^{(\ell-1)}}{L^{(k)}_{u}([\bm{v}]^{(\ell-1)})}}\nonumber\\
    &\times \left(\frac{\left(T^{(k)}_{u}\right)^{-\frac{1}{p}} L^{(k)}_{u}([\bm{v}]^{(\ell-1)})}{\left[\left(T^{(k)}_{u}\right)^{-\frac{1}{p}}\right]^{(\ell-1)}}\right)^{\frac{\left[\left(T^{(k)}_{u}\right)^{-\frac{1}{p}}\right]^{(\ell-1)}}{L^{(k)}_{u}([\bm{v}]^{(\ell-1)})}},
\end{align}
which gives us an approximation of~\eqref{app:cons:completion_time_1_8} as follows:
\begin{equation}\label{app:cons:completion_time_1_11}
    \frac{T^{(k)}}{\widehat{L}^{(k)}_{u}(\bm{v};\ell)}\le 1.
\end{equation}

We finally approximate constraint~\eqref{app:cons:completion_time_1_7} as follows:
\begin{tcolorbox}[ams align]
     &\frac{T^{(k)}}{\widehat{L}^{(k)}_{u}(\bm{v};\ell)}\le 1,~~~~~ \frac{T^{(k{-}1)}+\left(T^{(k)}_{u}\right)^{-\frac{1}{p}}}{\mathscr{G}^{(k)}_{u}T^{(k)}}\le 1,~~~~~\frac{1}{\mathscr{G}^{(k)}_{u}}\le 1.\nonumber
\end{tcolorbox}

We next aim to transform \eqref{app:cons:completion_time_1_6} into standard form of GP. In doing so, let us rewrite \eqref{app:cons:completion_time_1_6} as follows:
\begin{equation}\label{app:cons:completion_time_1_6_1}
    \frac{T^{(k)}_{u}}{\sum_{b\in \Omega}\sum_{u\in \mathcal{U}_{b}}\left(\Psi(\mathscr{D}_{u}^{\nuparrow,(k)})\right)^{p}}= 1.
\end{equation}
This constraint is not in the format of GP. Therefore, we transform it via introducing an auxiliary variable as two inequalities:
\begin{equation}\label{app:cons:completion_time_1_6_4_8}
    \frac{T^{(k)}_{u}}{\sum_{b\in \Omega}\sum_{u\in \mathcal{U}_{b}}\left(\Psi(\mathscr{D}_{u}^{\nuparrow,(k)})\right)^{p}}\le 1,
\end{equation}
\begin{equation}\label{app:cons:completion_time_1_6_4_9}
    \frac{\sum_{b\in \Omega}\sum_{u\in \mathcal{U}_{b}}\left(\Psi(\mathscr{D}_{u}^{\nuparrow,(k)})\right)^{p}}{\mathscr{E}^{(k)}_{u}T^{(k)}_{u}}\le 1,
\end{equation}
\begin{equation}
    \mathscr{E}^{(k)}_{u}\ge 1,
\end{equation}
where $\mathscr{E}^{(k)}_{u}$ is added with a large penalty term to the objective function to force $\mathscr{E}^{(k)}_{u}{\rightarrow}1^+$ at the optimal point. The fraction in~\eqref{app:cons:completion_time_1_6_4_8} still needs transformation since it is an inequality with a posynomial in the denominator, which is not a posynomial. We thus exploit arithmetic-geometric mean inequality (Lemma~\ref{Lemma:ArethmaticGeometric}) to approximate the denominator with a monomial:
\begin{align}\label{app:cons:completion_time_1_6_4_10}
    Q^{(k)}_{u}(\bm{v})=\sum_{b\in \Omega}\sum_{u\in \mathcal{U}_{b}}\left(\Psi(\mathscr{D}_{u}^{\nuparrow,(k)})\right)^{p}\geq \widehat{Q}^{(k)}_{u}(\bm{v};\ell) &\triangleq \prod_{b\in \Omega}\prod_{u\in \mathcal{U}_{b}}\left(\frac{\left(\Psi(\mathscr{D}_{u}^{\nuparrow,(k)})\right)^{p} Q^{(k)}_{u}([\bm{v}]^{(\ell-1)})}{\left[\left(\Psi(\mathscr{D}_{u}^{\nuparrow,(k)})\right)^{p}\right]^{(\ell-1)}}\right)^{\frac{\left[\left(\Psi(\mathscr{D}_{u}^{\nuparrow,(k)})\right)^{p}\right]^{(\ell-1)}}{Q^{(k)}_{u}([\bm{v}]^{(\ell-1)})}},
\end{align}
which gives us an approximation of~\eqref{app:cons:completion_time_1_6_4_8} as follows:
\begin{equation}\label{app:cons:completion_time_1_6_4_11}
    \frac{T^{(k)}_{u}}{\widehat{Q}^{(k)}_{u}(\bm{v};\ell)}\le 1.
\end{equation}

We finally approximate constraint~\eqref{app:cons:completion_time_1_6} as follows:
\begin{tcolorbox}[ams align]
     &\frac{T^{(k)}_{u}}{\widehat{Q}^{(k)}_{u}(\bm{v};\ell)}\le 1,~~~~~\frac{\sum_{b\in \Omega}\sum_{u\in \mathcal{U}_{b}}\left(\Psi(\mathscr{D}_{u}^{\nuparrow,(k)})\right)^{p}}{\mathscr{E}^{(k)}_{u}T^{(k)}_{u}}\le 1,~~~~~\frac{1}{\mathscr{E}^{(k)}_{u}}\le 1.\nonumber
\end{tcolorbox}

\newpage
\subsection{Energy Consumption}\label{app:energy_consumption}

\noindent\textbf{LM training:}  Referring to Sec.~\ref{sec:energy_consumption}, let us revisit the equation for computation energy (consumed energy for LM training) of FLU $u$ below:
\begin{equation}\label{app:cons:computation_energy_1}
    E_{u}^{\mathsf{LC},{(k)}}= \frac{\alpha_{u}}{2}(f^{(k)}_{u})^{3}\tau_{u}^{\mathsf{LC},{(k)}}.
\end{equation}

In the following we transform \eqref{app:cons:computation_energy_1} into standard GP format. In doing so, we define $E_{u}^{\mathsf{LC},{(k)}}$ as an auxiliary decision variable that must satisfy the following constraint (derived through performing some algebraic manipulations on \eqref{app:cons:computation_energy_1}):
\begin{equation}\label{app:cons:computation_energy_2}
    \frac{2 E_{u}^{\mathsf{LC},{(k)}}}{\alpha_{u}(f^{(k)}_{u})^{3}\tau_{u}^{\mathsf{LC},{(k)}}}= 1.
\end{equation}
We transform \eqref{app:cons:computation_energy_2} into the format of GP by splitting it into the following three inequalities:
\begin{equation}\label{app:cons:computation_energy_3}
    \frac{2 E_{u}^{\mathsf{LC},{(k)}}}{\alpha_{u}(f^{(k)}_{u})^{3}\tau_{u}^{\mathsf{LC},{(k)}}}\le 1,
\end{equation}
\begin{equation}\label{app:cons:computation_energy_4}
    \frac{\left(\mathscr{E}^{\mathsf{LC},{(k)}}_{u}\right)^{-1}\alpha_{u}(f^{(k)}_{u})^{3}\tau_{u}^{\mathsf{LC},{(k)}}}{2 E_{u}^{\mathsf{LC},{(k)}}}\le 1,
\end{equation}
\begin{equation}
    \mathscr{E}^{\mathsf{LC},{(k)}}_{u}\ge 1,
\end{equation}
where $\mathscr{E}^{\mathsf{LC},{(k)}}_{u}$ is an auxiliary decision variable and will be added with a large penalty term to the objective function to force $A\mathscr{E}^{\mathsf{LC},{(k)}}_{u}{\rightarrow}1^+$ at the optimal point.

Accordingly, we approximate \eqref{app:cons:computation_energy_1} as follows:
\begin{tcolorbox}[ams align]
     & \frac{2 E_{u}^{\mathsf{LC},{(k)}}}{\alpha_{u}(f^{(k)}_{u})^{3}\tau_{u}^{\mathsf{LC},{(k)}}}\le 1,~~~~\frac{\left(\mathscr{E}^{\mathsf{LC},{(k)}}_{u}\right)^{-1}\alpha_{u}(f^{(k)}_{u})^{3}\tau_{u}^{\mathsf{LC},{(k)}}}{2 E_{u}^{\mathsf{LC},{(k)}}}\le 1,~~~~\frac{1}{\mathscr{E}^{\mathsf{LC},{(k)}}_{u}}\le 1.\nonumber
\end{tcolorbox}

\noindent\textbf{GM broadcasting of O-RUs:} Referring to Sec.~\ref{sec:energy_consumption}, let us revisit the equation for GM broadcasting latency of O-RU $b$ below:
\begin{equation}\label{app:cons:GM_broadcasting_energy_0}
    E_{b}^{\downarrow,{(k)}}{=}\sum_{x{=}N^{(k{-}1)}{+}1}^{N^{(k)}}\sum_{r{\in} \mathcal{R}_{b}}\min\big\{\tau_{b,r}^{\downarrow}(t_{x}),\left(t_{x{+}1}{-}t_{x}\right)\big\}\rho^{\downarrow}_{b,r}\left(t_{x}\right) P^{\mathsf{max}}_{b}.
\end{equation}
The above equation can be rewritten as follows:
\begin{equation}\label{app:cons:GM_broadcasting_energy_1}
    E_{b}^{\downarrow,{(k)}}{=}\sum_{x{=}0}^{N}\sum_{r{\in} \mathcal{R}_{b}}T_{b,r}^{\downarrow}(t_{x})\rho^{\downarrow}_{b,r}\left(t_{x}\right) P^{\mathsf{max}}_{b},
\end{equation}
where $T_{b,r}^{\downarrow}(t_{x})$ is given in \eqref{app:eq:broadcast_latency_min_1}. We transform \eqref{app:cons:GM_broadcasting_energy_1} into GP format. To this end, we rewrite \eqref{app:cons:GM_broadcasting_energy_1} as follows:

\begin{equation}\label{app:cons:GM_broadcasting_energy_3_1_1}
   \frac{E_{b}^{\downarrow,{(k)}}}{\displaystyle\sum_{x{=}0}^{N}\sum_{r{\in} \mathcal{R}_{b}}T_{b,r}^{\downarrow}(t_{x})\rho^{\downarrow}_{b,r}\left(t_{x}\right) P^{\mathsf{max}}_{b}} = 1.
\end{equation} 
This constraint is not in the format of GP. Therefore, we transform it through splitting it into the following three inequalities:
\begin{equation}\label{app:cons:GM_broadcasting_energy_3_1_4}
   \frac{E_{b}^{\downarrow,{(k)}}}{\displaystyle\sum_{x{=}0}^{N}\sum_{r{\in} \mathcal{R}_{b}}T_{b,r}^{\downarrow}(t_{x})\rho^{\downarrow}_{b,r}\left(t_{x}\right) P^{\mathsf{max}}_{b}}\le 1,
\end{equation}
\begin{equation}\label{app:cons:GM_broadcasting_energy_3_1_5}
    \frac{\left(\mathscr{B}^{\downarrow,\mathsf{E}}_{b}(t_{x})\right)^{-1}\left(\displaystyle\sum_{x{=}0}^{N}\sum_{r{\in} \mathcal{R}_{b}}T_{b,r}^{\downarrow}(t_{x})\rho^{\downarrow}_{b,r}\left(t_{x}\right) P^{\mathsf{max}}_{b}\right)}{E_{b}^{\downarrow,{(k)}}}\le 1,
\end{equation}
\begin{equation}
    \mathscr{B}^{\downarrow,\mathsf{E}}_{b}(t_{x})\ge 1,
\end{equation}
where $\mathscr{B}^{\downarrow,\mathsf{E}}_{b}(t_{x})$ is an auxiliary decision variable added with a large penalty term to the objective function to force $\mathscr{B}^{\downarrow,\mathsf{E}}_{b}(t_{x}){\rightarrow}1^+$ at the optimal point. The fraction in~\eqref{app:cons:GM_broadcasting_energy_3_1_4} still needs transformation since it is inequality with a posynomial in its denominator, which is not a posynomial. We thus exploit arithmetic-geometric mean inequality (Lemma~\ref{Lemma:ArethmaticGeometric}) to approximate the denominator in \eqref{app:cons:GM_broadcasting_energy_3_1_4} with a monomial.
\begin{align}\label{app:cons:GM_broadcasting_energy_3_1_6}
   G^{\downarrow,\mathsf{E}}_{b}(\bm{v})&=\sum_{x{=}0}^{N}\sum_{r{\in} \mathcal{R}_{b}}T_{b,r}^{\downarrow}(t_{x})\rho^{\downarrow}_{b,r}\left(t_{x}\right) P^{\mathsf{max}}_{b}\nonumber\\
   &\geq \widehat{G}^{\downarrow,\mathsf{E}}_{b}(\bm{v};\ell) \triangleq \prod_{x{=}0}^{N}\prod_{r{\in} \mathcal{R}_{b}}\left(\frac{T_{b,r}^{\downarrow}(t_{x})\rho^{\downarrow}_{b,r}\left(t_{x}\right)  G^{\downarrow,\mathsf{E}}_{b}([\bm{v}]^{(\ell-1)})}{\left[T_{b,r}^{\downarrow}(t_{x})\rho^{\downarrow}_{b,r}\left(t_{x}\right) \right]^{(\ell-1)}}\right)^{\frac{\left[T_{b,r}^{\downarrow}(t_{x})\rho^{\downarrow}_{b,r}\left(t_{x}\right) P^{\mathsf{max}}_{b}\right]^{(\ell-1)}}{G^{\downarrow,\mathsf{E}}_{b}([\bm{v}]^{(\ell-1)})}}
\end{align}
which gives us an approximation of~\eqref{app:cons:GM_broadcasting_energy_3_1_4} as follows:
\begin{equation}\label{app:cons:GM_broadcasting_energy_3_1_7}
    \frac{E_{b}^{\downarrow,{(k)}}}{\widehat{G}^{\downarrow,\mathsf{E}}_{b}(\bm{v};\ell)}\le 1.
\end{equation}

We finally approximate constraint~\eqref{app:cons:GM_broadcasting_energy_1} as follows:
\begin{tcolorbox}[ams align]
     &\frac{E_{b}^{\downarrow,{(k)}}}{\widehat{G}^{\downarrow,\mathsf{E}}_{b}(\bm{v};\ell)}\le 1,~~~~~~~~\frac{\left(\mathscr{B}^{\downarrow,\mathsf{E}}_{b}(t_{x})\right)^{-1}\left(\displaystyle\sum_{x{=}0}^{N}\sum_{r{\in} \mathcal{R}_{b}}T_{b,r}^{\downarrow}(t_{x})\rho^{\downarrow}_{b,r}\left(t_{x}\right) P^{\mathsf{max}}_{b}\right)}{E_{b}^{\downarrow,{(k)}}}\le 1,~~~~~~~\frac{1}{\mathscr{B}^{\downarrow,\mathsf{E}}_{b}(t_{x})}\le 1,\nonumber
\end{tcolorbox}

\noindent\textbf{GPs dispersion by DPUs:} Referring to Sec.~\ref{sec:energy_consumption}, let us revisit the equation for GPs dispersion latency of DPU $u$ below:
\begin{equation}\label{app:cons:GPs_dispersion_latency_0}
    \overline{E}_{u}^{\uparrow,(k)}{=}\sum_{x{=}N^{(k{-}1)}{+}1}^{N^{(k)}}\sum_{u'{\in}\mathcal{U}_{b}}\sum_{r{\in} \overline{\mathcal{R}}_{b}}\min\big\{\overline{\tau}_{u,u',r}^{\uparrow}(t_{x}),\left(t_{x{+}1}{-}t_{x}\right)\big\}\overline{\rho}^{\uparrow}_{u,r}(t_{x})P^{\mathsf{max}}_{u}.
\end{equation}
The above equation can be rewritten as follows:
\begin{equation}\label{app:cons:GPs_dispersion_latency_1}
   \overline{E}_{u}^{\uparrow,(k)}{=}\sum_{x{=}0}^{N}\sum_{u'{\in}\mathcal{U}_{b}}\sum_{r{\in} \overline{\mathcal{R}}_{b}}T_{u,u',r}^{\uparrow}(t_{x})\overline{\rho}^{\uparrow}_{u,r}(t_{x})P^{\mathsf{max}}_{u},
\end{equation}
where $T_{u,u',r}^{\uparrow}(t_{x})$ is given in \eqref{app:eq:dispersion_min_1}. We transform \eqref{app:cons:GPs_dispersion_latency_1} into GP format. To this end, we rewrite \eqref{app:cons:GPs_dispersion_latency_1} as follows:
\begin{equation}\label{app:cons:GPs_dispersion_latency_3_1_1}
   \frac{\overline{E}_{u}^{\uparrow,(k)}}{\sum_{x{=}0}^{N}\sum_{u'{\in}\mathcal{U}_{b}}\sum_{r{\in} \overline{\mathcal{R}}_{b}}T_{u,u',r}^{\uparrow}(t_{x})\overline{\rho}^{\uparrow}_{u,r}(t_{x})P^{\mathsf{max}}_{u}}{=} 1.
\end{equation} 
This constraint is not in the format of GP. Therefore, we transform it through splitting it into the following three inequalities:
\begin{equation}\label{app:cons:GPs_dispersion_latency_3_1_4}
   \frac{\overline{E}_{u}^{\uparrow,(k)}}{\sum_{x{=}0}^{N}\sum_{u'{\in}\mathcal{U}_{b}}\sum_{r{\in} \overline{\mathcal{R}}_{b}}T_{u,u',r}^{\uparrow}(t_{x})\overline{\rho}^{\uparrow}_{u,r}(t_{x})P^{\mathsf{max}}_{u}}\le 1,
\end{equation}
\begin{equation}\label{app:cons:GPs_dispersion_latency_3_1_5}
    \frac{\left(\overline{\mathscr{B}}^{\uparrow,\mathsf{E}}_{u}(t_{x})\right)^{-1}\left(\sum_{x{=}0}^{N}\sum_{u'{\in}\mathcal{U}_{b}}\sum_{r{\in} \overline{\mathcal{R}}_{b}}T_{u,u',r}^{\uparrow}(t_{x})\overline{\rho}^{\uparrow}_{u,r}(t_{x})P^{\mathsf{max}}_{u}\right)}{\overline{E}_{u}^{\uparrow,(k)}}\le 1,
\end{equation}
\begin{equation}
    \overline{\mathscr{B}}^{\uparrow,\mathsf{E}}_{u}(t_{x})\ge 1,
\end{equation}
where $\overline{\mathscr{B}}^{\uparrow,\mathsf{E}}_{u}(t_{x})$ is an auxiliary decision variable added with a large penalty term to the objective function to force $\overline{\mathscr{B}}^{\uparrow,\mathsf{E}}_{u}(t_{x}){\rightarrow}1^+$ at the optimal point. The fraction in~\eqref{app:cons:GPs_dispersion_latency_3_1_4} still needs transformation since it is inequality with a posynomial in its denominator, which is not a posynomial. We thus exploit arithmetic-geometric mean inequality (Lemma~\ref{Lemma:ArethmaticGeometric}) to approximate the denominator in \eqref{app:cons:GPs_dispersion_latency_3_1_4} with a monomial.
\begin{align}\label{app:cons:GPs_dispersion_latency_3_1_6}
   G^{\uparrow,\mathsf{E},\mathsf{D}}_{u}(\bm{v})&=\sum_{x{=}0}^{N}\sum_{u'{\in}\mathcal{U}_{b}}\sum_{r{\in} \overline{\mathcal{R}}_{b}}T_{u,u',r}^{\uparrow}(t_{x})\overline{\rho}^{\uparrow}_{u,r}(t_{x})P^{\mathsf{max}}_{u}\nonumber\\
   &\geq \widehat{G}^{\uparrow,\mathsf{E},\mathsf{D}}_{u}(\bm{v};\ell) \triangleq \prod_{x{=}0}^{N}\prod_{u'{\in}\mathcal{U}_{b}}\prod_{r{\in} \overline{\mathcal{R}}_{b}}\left(\frac{T_{u,u',r}^{\uparrow}(t_{x})\overline{\rho}^{\uparrow}_{u,r}(t_{x})  G^{\uparrow,\mathsf{E},\mathsf{D}}_{u}([\bm{v}]^{(\ell-1)})}{\left[T_{u,u',r}^{\uparrow}(t_{x})\overline{\rho}^{\uparrow}_{u,r}(t_{x}) \right]^{(\ell-1)}}\right)^{\frac{\left[T_{u,u',r}^{\uparrow}(t_{x})\overline{\rho}^{\uparrow}_{u,r}(t_{x})P^{\mathsf{max}}_{u} P^{\mathsf{max}}_{b}\right]^{(\ell-1)}}{G^{\uparrow,\mathsf{E},\mathsf{D}}_{u}([\bm{v}]^{(\ell-1)})}}
\end{align}
which gives us an approximation of~\eqref{app:cons:GPs_dispersion_latency_3_1_4} as follows:
\begin{equation}\label{app:cons:GPs_dispersion_latency_3_1_7}
    \frac{\overline{E}_{u}^{\uparrow,(k)}}{\widehat{G}^{\uparrow,\mathsf{E},\mathsf{D}}_{u}(\bm{v};\ell)}\le 1.
\end{equation}

We finally approximate constraint~\eqref{app:cons:GPs_dispersion_latency_1} as follows:
\begin{tcolorbox}[ams align]
     &\frac{\overline{E}_{u}^{\uparrow,(k)}}{\widehat{G}^{\uparrow,\mathsf{E},\mathsf{D}}_{u}(\bm{v};\ell)}\le 1,~~~~~~~~\frac{\left(\overline{\mathscr{B}}^{\uparrow,\mathsf{E}}_{u}(t_{x})\right)^{-1}\left(\sum_{x{=}0}^{N}\sum_{u'{\in}\mathcal{U}_{b}}\sum_{r{\in} \overline{\mathcal{R}}_{b}}T_{u,u',r}^{\uparrow}(t_{x})\overline{\rho}^{\uparrow}_{u,r}(t_{x})P^{\mathsf{max}}_{u}\right)}{\overline{E}_{u}^{\uparrow,(k)}}\le 1,~~~~~~~\frac{1}{\overline{\mathscr{B}}^{\uparrow,\mathsf{E}}_{u}(t_{x})}\le 1,\nonumber
\end{tcolorbox}

\noindent\textbf{GPs dispatching of CHUs:} Referring to Sec.~\ref{sec:energy_consumption}, let us revisit the equation for GPs dispatching latency of CHU $u$ below:
\begin{equation}\label{app:cons:GPs_dispatching_latency_0}
    E_{u}^{\uparrow,(k)}{=}\sum_{x=N^{(k{-}1)}{+}1}^{N^{(k)}}\sum_{r{\in}\mathcal{R}_{b}}\min\big\{\tau_{u,r}^{\uparrow}(t_{x}),\left(t_{x{+}1}{-}t_{x}\right)\big\}\rho^{\uparrow}_{u,r}(t_{x})P^{\mathsf{max}}_{u}.
\end{equation}
The above equation can be rewritten as follows:
\begin{equation}\label{app:cons:GPs_dispatching_latency_1}
   E_{u}^{\uparrow,(k)}{=}\sum_{x{=}0}^{N}\sum_{r{\in}\mathcal{R}_{b}}T_{u,r}^{\uparrow}(t_{x})\rho^{\uparrow}_{u,r}(t_{x})P^{\mathsf{max}}_{u},
\end{equation}
where $T_{u,r}^{\uparrow}(t_{x})$ is given in \eqref{app:eq:dispatch_min_1}. We transform \eqref{app:cons:GPs_dispatching_latency_1} into GP format. To this end, we rewrite \eqref{app:cons:GPs_dispatching_latency_1} as follows:
\begin{equation}\label{app:cons:GPs_dispatching_latency_3_1_1}
   \frac{E_{u}^{\uparrow,(k)}}{\sum_{x{=}0}^{N}\sum_{r{\in}\mathcal{R}_{b}}T_{u,r}^{\uparrow}(t_{x})\rho^{\uparrow}_{u,r}(t_{x})P^{\mathsf{max}}_{u}}{=} 1.
\end{equation} 
This constraint is not in the format of GP. Therefore, we transform it through splitting it into the following three inequalities:
\begin{equation}\label{app:cons:GPs_dispatching_latency_3_1_4}
   \frac{E_{u}^{\uparrow,(k)}}{\sum_{x{=}0}^{N}\sum_{r{\in}\mathcal{R}_{b}}T_{u,r}^{\uparrow}(t_{x})\rho^{\uparrow}_{u,r}(t_{x})P^{\mathsf{max}}_{u}}\le 1,
\end{equation}
\begin{equation}\label{app:cons:GPs_dispatching_latency_3_1_5}
    \frac{\left(\mathscr{B}^{\uparrow,\mathsf{E}}_{u}(t_{x})\right)^{-1}\left(\sum_{x{=}0}^{N}\sum_{r{\in}\mathcal{R}_{b}}T_{u,r}^{\uparrow}(t_{x})\rho^{\uparrow}_{u,r}(t_{x})P^{\mathsf{max}}_{u}\right)}{E_{u}^{\uparrow,(k)}}\le 1,
\end{equation}
\begin{equation}
    \mathscr{B}^{\uparrow,\mathsf{E}}_{u}(t_{x})\ge 1,
\end{equation}
where $\mathscr{B}^{\uparrow,\mathsf{E}}_{u}(t_{x})$ is an auxiliary decision variable added with a large penalty term to the objective function to force $\mathscr{B}^{\uparrow,\mathsf{E}}_{u}(t_{x}){\rightarrow}1^+$ at the optimal point. The fraction in~\eqref{app:cons:GPs_dispatching_latency_3_1_4} still needs transformation since it is inequality with a posynomial in its denominator, which is not a posynomial. We thus exploit arithmetic-geometric mean inequality (Lemma~\ref{Lemma:ArethmaticGeometric}) to approximate the denominator in \eqref{app:cons:GPs_dispatching_latency_3_1_4} with a monomial.
\begin{align}\label{app:cons:GPs_dispatching_latency_3_1_6}
   G^{\uparrow,\mathsf{E},\mathsf{C}}_{u}(\bm{v})&=\sum_{x{=}0}^{N}\sum_{r{\in}\mathcal{R}_{b}}T_{u,r}^{\uparrow}(t_{x})\rho^{\uparrow}_{u,r}(t_{x})P^{\mathsf{max}}_{u}\nonumber\\
   &\geq \widehat{G}^{\uparrow,\mathsf{E},\mathsf{C}}_{u}(\bm{v};\ell) \triangleq \prod_{x{=}0}^{N}\prod_{r{\in}\mathcal{R}_{b}}\left(\frac{T_{u,u',r}^{\uparrow}(t_{x})\overline{\rho}^{\uparrow}_{u,r}(t_{x})  G^{\uparrow,\mathsf{E},\mathsf{C}}_{u}([\bm{v}]^{(\ell-1)})}{\left[T_{u,u',r}^{\uparrow}(t_{x})\overline{\rho}^{\uparrow}_{u,r}(t_{x}) \right]^{(\ell-1)}}\right)^{\frac{\left[T_{u,u',r}^{\uparrow}(t_{x})\overline{\rho}^{\uparrow}_{u,r}(t_{x})P^{\mathsf{max}}_{u} P^{\mathsf{max}}_{b}\right]^{(\ell-1)}}{G^{\uparrow,\mathsf{E},\mathsf{C}}_{u}([\bm{v}]^{(\ell-1)})}}
\end{align}
which gives us an approximation of~\eqref{app:cons:GPs_dispatching_latency_3_1_4} as follows:
\begin{equation}\label{app:cons:GPs_dispatching_latency_3_1_7}
    \frac{E_{u}^{\uparrow,(k)}}{\widehat{G}^{\uparrow,\mathsf{E},\mathsf{C}}_{u}(\bm{v};\ell)}\le 1.
\end{equation}

We finally approximate constraint~\eqref{app:cons:GPs_dispatching_latency_1} as follows:
\begin{tcolorbox}[ams align]
     &\frac{E_{u}^{\uparrow,(k)}}{\widehat{G}^{\uparrow,\mathsf{E},\mathsf{C}}_{u}(\bm{v};\ell)}\le 1,~~~~~~~~\frac{\left(\mathscr{B}^{\uparrow,\mathsf{E}}_{u}(t_{x})\right)^{-1}\left(\sum_{x{=}0}^{N}\sum_{r{\in}\mathcal{R}_{b}}T_{u,r}^{\uparrow}(t_{x})\rho^{\uparrow}_{u,r}(t_{x})P^{\mathsf{max}}_{u}\right)}{E_{u}^{\uparrow,(k)}}\le 1,~~~~~~~\frac{1}{\mathscr{B}^{\uparrow,\mathsf{E}}_{u}(t_{x})}\le 1,\nonumber
\end{tcolorbox}

\newpage
\subsection{Transformation of the Objective Function of Optimization Problem \texorpdfstring{$\mathbf{\mathcal{P}}$}{TEXT}}\label{app:objective_fucntion}
Let us revisit the objective function of optimization problem $\bm{\mathcal{P}}$:

\begin{align}
    &\mathscr{O}=\underbrace{c_1 \hspace{-.1mm}\frac{1}{K}\sum_{k=0}^{K-1}\mathbb E\left[ \big\Vert \nabla \mathfrak{L}^{({k})}(\bm{\omega}^{({k})})\big\Vert^2\right]}_{(x_1)}+ \underbrace{c_2 T^{(K)}}_{(x_2)}+\underbrace{c_3\sum_{k=0}^{K-1}\sum_{b \in \Omega} E_{b}^{\downarrow,{(k)}}}_{(x_3)}{+}\underbrace{c_3\sum_{k=0}^{K-1}\sum_{b \in \Omega}\sum_{u{\in}\mathcal{U}_{b}}\overline{\lambda}_{u}^{(k)}\big( E_{u}^{\mathsf{LC},{(k)}}+\overline{E}_{u}^{\uparrow,(k)}\big)}_{(x_4)}\nonumber \\
    &~~~~~~~~~~~~~~~~~~~~~~~~~~~~~~~~~~~~~~~+\underbrace{c_3\sum_{k=0}^{K-1}\sum_{b \in \Omega} \sum_{u{\in}\mathcal{U}_{b}} \lambda_{u}^{(k)}\big(E_{u}^{\mathsf{LC},{(k)}}+E_{u}^{\uparrow,(k)}\big)}_{(x_5)}.
\end{align}

In the above objective function, $T^{(K)}$, $E_{b}^{\downarrow,{(k)}}$, $\overline{\lambda}_{u}^{(k)}$, $E_{u}^{\mathsf{LC},{(k)}}$, $\overline{E}_{u}^{\uparrow,(k)}$, $\lambda_{u}^{(k)}$, and $E_{u}^{\uparrow,(k)}$ are decision variables, and thus, terms $(x_2)$, $(x_3)$, $(x_4)$, and $(x_5)$ are posynomials. Therefore, no additional transformations on these terms are required. In the following, we transform term $(x_1)$ of the above objective function into the standard form of GP. 

\textbf{The ML bound in the objective function:} We now turn our attention to the term associated with the ML performance in the objective function (i.e., term $(x_1)$ of the above objective function). Consider the bound presented in \eqref{eq:gen_conv}, which is revisited below.

{\footnotesize
\begin{align}\label{eq:GP_gen_conv_1}
\vspace{-5mm}
    &\frac{1}{K} \sum_{k=0}^{K-1}\mathbb{E}\left[\left\Vert\nabla{\mathfrak{L}^{({k})}(\bm{\omega}^{(k)})}\right\Vert^2\right] {\leq} \frac{4}{K} \sum_{k=0}^{K-1}\Bigg(\frac{\mathbb{E}_k\left[\mathfrak{L}^{(k-1)}(\bm{\omega}^{(k)})\right]-\mathbb{E}_k\left[\mathfrak{L}^{({k})}(\bm{\omega}^{(k+1)})\right]}{\eta_{_k}\mathfrak{B}_k\left(1-\zeta^{(k)}\right)}+\frac{\sum_{b \in \Omega}\sum_{u\in \mathcal{U}_{b}} \mathfrak{D}^{(k)}_u\left(\left((T^{(k)}-T^{(k{-}1)})-\widehat{\lambda}_{u}^{(k)}\tau_{u}^{\mathsf{LC},{(k)}}\right)\right) }{\eta_{_k}\mathfrak{B}_k\left(1-\zeta^{(k)}\right)}\Bigg)\nonumber\\
    &+\frac{8}{K} \sum_{k=0}^{K-1}\Bigg[\frac{1}{{\left(1-\zeta^{(k)}\right)}}\Bigg({\beta^2\Theta^2 \eta_k^2}\sum_{b \in \Omega} \sum_{u\in \mathcal{U}_{b}}\frac{|\Upsilon_{u}(\widetilde{\tau}_{b}^{\downarrow,{(k)}})|}{|\Upsilon(\widetilde{\bm{\tau}}^{\downarrow,{(k)}})|}\frac{\left(\ell_u^{(k)}-1\right)}{1- 4\eta_k^2\beta^2 \ell_u^{(k)}\left(\ell_u^{(k)}-1\right)}\left(1-\frac{{B}_{u}(\widetilde{\tau}_{b}^{\downarrow,{(k)}})}{|\Upsilon_{u}(\widetilde{\tau}_{b}^{\downarrow,{(k)}})|} \right)  \frac{{(|\Upsilon_{u}(\widetilde{\tau}_{b}^{\downarrow,{(k)}})|-1)}\left(\sigma_{u}(\widetilde{\tau}_{b}^{\downarrow,{(k)}})\right)^2}{|\Upsilon_{u}(\widetilde{\tau}_{b}^{\downarrow,{(k)}})|{B}_{u}(\widetilde{\tau}_{b}^{\downarrow,{(k)}})}\Bigg)\nonumber\\
    & + \frac{1}{\left(1-\zeta^{(k)}\right)}\Bigg(\frac{\mathfrak{X}_2 \eta_k^2\beta^2 \left(\ell_{\mathsf{max}}^{(k)}\right)\left(\ell_{\mathsf{max}}^{(k)}-1\right)}{1- 4\eta_k^2\beta^2\ell_{\mathsf{max}}^{(k)}\left(\ell_{\mathsf{max}}^{(k)}-1\right)}+ \frac{\Theta^2\beta\eta_{_k}\mathfrak{B}_k}{2} \sum_{b \in \Omega} \sum_{u\in \mathcal{U}_{b}}\frac{\left(\widehat{\lambda}_{u}^{(k)}|\Upsilon_{u}(\widetilde{\tau}_{b}^{\downarrow,{(k)}})|\right)^2}{\left(|\Upsilon^{\mathsf{s}}(\widetilde{\bm{\tau}}^{\downarrow,{(k)}})|\right)^2 \ell^{(k)}_{u}}\left(1-\frac{{B}_{u}(\widetilde{\tau}_{b}^{\downarrow,{(k)}})}{|\Upsilon_{u}(\widetilde{\tau}_{b}^{\downarrow,{(k)}})|} \right) \frac{(|\Upsilon_{u}(\widetilde{\tau}_{b}^{\downarrow,{(k)}})|-1)\left(\sigma_{u}(\widetilde{\tau}_{b}^{\downarrow,{(k)}})\right)^2}{|\Upsilon_{u}(\widetilde{\tau}_{b}^{\downarrow,{(k)}})|{B}_{u}(\widetilde{\tau}_{b}^{\downarrow,{(k)}})}\Bigg)\nonumber\\
    &+\frac{24}{\left(1-\zeta^{(k)}\right)}\left (\frac{|\Upsilon(\widetilde{\bm{\tau}}^{\downarrow,{(k)}})|-|\Upsilon^{\mathsf{s}}(\widetilde{\bm{\tau}}^{\downarrow,{(k)}})|}{|\Upsilon(\widetilde{\bm{\tau}}^{\downarrow,{(k)}})|}\right)^2\sum_{b \in \Omega} \sum_{u\in \mathcal{U}_{b}}\Bigg(\frac{\Theta^2 \beta^2 \eta_k^2 \left(\ell_u^{(k)}-1\right)}{1- 4\eta_k^2\beta^2 \ell_u^{(k)}\left(\ell_u^{(k)}-1\right)} \left(1-\frac{{B}_{u}(\widetilde{\tau}_{b}^{\downarrow,{(k)}})}{|\Upsilon_{u}(\widetilde{\tau}_{b}^{\downarrow,{(k)}})|} \right)  \frac{{(|\Upsilon_{u}(\widetilde{\tau}_{b}^{\downarrow,{(k)}})|-1)}\left(\sigma_{u}(\widetilde{\tau}_{b}^{\downarrow,{(k)}})\right)^2}{|\Upsilon_{u}(\widetilde{\tau}_{b}^{\downarrow,{(k)}})|{B}_{u}(\widetilde{\tau}_{b}^{\downarrow,{(k)}})}\Bigg)\nonumber\\
    &+\frac{6}{\left(1-\zeta^{(k)}\right)}\Bigg (\frac{|\Upsilon(\widetilde{\bm{\tau}}^{\downarrow,{(k)}})|-|\Upsilon^{\mathsf{s}}(\widetilde{\bm{\tau}}^{\downarrow,{(k)}})|}{|\Upsilon(\widetilde{\bm{\tau}}^{\downarrow,{(k)}})|}\Bigg)^2\frac{|\Upsilon(\widetilde{\bm{\tau}}^{\downarrow,{(k)}})|}{|\Upsilon_{\mathsf{min}}(\widetilde{\bm{\tau}}^{\downarrow,{(k)}})|}\Bigg(\frac{\mathfrak{X}_2}{1- 4\eta_k^2\beta^2 \ell_{\mathsf{max}}^{(k)}\left(\ell_{\mathsf{max}}^{(k)}-1\right)}\Bigg)\Bigg].
 \end{align}
}%

In addition to the assumptions made in Theorem~\ref{th:main}, assume $\max_{k{\in}\mathcal{K}} \big\{\zeta^{(k)}\big\} \leq \zeta_{\mathsf{max}}<1$ and $\eta_k = \frac{\alpha}{\sqrt{{\ell}^{(k)}_{\mathsf{sum}} K/N}}$. The above inequality can be simplified as follows:

{\footnotesize
\begin{align}\label{eq:GP_gen_conv_2}
\vspace{-5mm}
    &\frac{1}{K} \sum_{k=0}^{K-1}\mathbb{E}\left[\left\Vert\nabla{\mathfrak{L}^{({k})}(\bm{\omega}^{(k)})}\right\Vert^2\right] {\leq}\frac{4 \sqrt{\widehat{\ell}_{\mathsf{max}}}\left(\mathfrak{L}^{(-1)}(\bm{\omega}^{(0)}) - \mathfrak{L}^{(K)^\star}\right)}{\mathfrak{B}_k\left(1-\zeta_{\mathsf{max}}\right)\alpha \sqrt{KN}}+\frac{4}{K} \sum_{k=0}^{K-1}\frac{\sum_{b \in \Omega}\sum_{u\in \mathcal{U}_{b}} \mathfrak{D}^{(k)}_u\left(\left((T^{(k)}-T^{(k{-}1)})-\widehat{\lambda}_{u}^{(k)}\tau_{u}^{\mathsf{LC},{(k)}}\right)\right) }{\eta_{_k}\mathfrak{B}_k\left(1-\zeta^{(k)}\right)}\nonumber\\
    &+\frac{8}{K} \sum_{k=0}^{K-1}\Bigg[\frac{1}{{\left(1-\zeta^{(k)}\right)}}\Bigg({\beta^2\Theta^2 \eta_k^2}\sum_{b \in \Omega} \sum_{u\in \mathcal{U}_{b}}\frac{|\Upsilon_{u}(\widetilde{\tau}_{b}^{\downarrow,{(k)}})|}{|\Upsilon(\widetilde{\bm{\tau}}^{\downarrow,{(k)}})|}\frac{\left(\ell_u^{(k)}-1\right)}{1- 4\eta_k^2\beta^2 \ell_u^{(k)}\left(\ell_u^{(k)}-1\right)}\left(1-\frac{{B}_{u}(\widetilde{\tau}_{b}^{\downarrow,{(k)}})}{|\Upsilon_{u}(\widetilde{\tau}_{b}^{\downarrow,{(k)}})|} \right)  \frac{{(|\Upsilon_{u}(\widetilde{\tau}_{b}^{\downarrow,{(k)}})|-1)}\left(\sigma_{u}(\widetilde{\tau}_{b}^{\downarrow,{(k)}})\right)^2}{|\Upsilon_{u}(\widetilde{\tau}_{b}^{\downarrow,{(k)}})|{B}_{u}(\widetilde{\tau}_{b}^{\downarrow,{(k)}})}\Bigg)\nonumber\\
    & + \frac{1}{\left(1-\zeta^{(k)}\right)}\Bigg(\frac{\mathfrak{X}_2 \eta_k^2\beta^2 \left(\ell_{\mathsf{max}}^{(k)}\right)\left(\ell_{\mathsf{max}}^{(k)}-1\right)}{1- 4\eta_k^2\beta^2\ell_{\mathsf{max}}^{(k)}\left(\ell_{\mathsf{max}}^{(k)}-1\right)}+ \frac{\Theta^2\beta\eta_{_k}\mathfrak{B}_k}{2} \sum_{b \in \Omega} \sum_{u\in \mathcal{U}_{b}}\frac{\left(\widehat{\lambda}_{u}^{(k)}|\Upsilon_{u}(\widetilde{\tau}_{b}^{\downarrow,{(k)}})|\right)^2}{\left(|\Upsilon^{\mathsf{s}}(\widetilde{\bm{\tau}}^{\downarrow,{(k)}})|\right)^2 \ell^{(k)}_{u}}\left(1-\frac{{B}_{u}(\widetilde{\tau}_{b}^{\downarrow,{(k)}})}{|\Upsilon_{u}(\widetilde{\tau}_{b}^{\downarrow,{(k)}})|} \right) \frac{(|\Upsilon_{u}(\widetilde{\tau}_{b}^{\downarrow,{(k)}})|-1)\left(\sigma_{u}(\widetilde{\tau}_{b}^{\downarrow,{(k)}})\right)^2}{|\Upsilon_{u}(\widetilde{\tau}_{b}^{\downarrow,{(k)}})|{B}_{u}(\widetilde{\tau}_{b}^{\downarrow,{(k)}})}\Bigg)\nonumber\\
    &+\frac{24}{\left(1-\zeta^{(k)}\right)}\left (\frac{|\Upsilon(\widetilde{\bm{\tau}}^{\downarrow,{(k)}})|-|\Upsilon^{\mathsf{s}}(\widetilde{\bm{\tau}}^{\downarrow,{(k)}})|}{|\Upsilon(\widetilde{\bm{\tau}}^{\downarrow,{(k)}})|}\right)^2\sum_{b \in \Omega} \sum_{u\in \mathcal{U}_{b}}\Bigg(\frac{\Theta^2 \beta^2 \eta_k^2 \left(\ell_u^{(k)}-1\right)}{1- 4\eta_k^2\beta^2 \ell_u^{(k)}\left(\ell_u^{(k)}-1\right)} \left(1-\frac{{B}_{u}(\widetilde{\tau}_{b}^{\downarrow,{(k)}})}{|\Upsilon_{u}(\widetilde{\tau}_{b}^{\downarrow,{(k)}})|} \right)  \frac{{(|\Upsilon_{u}(\widetilde{\tau}_{b}^{\downarrow,{(k)}})|-1)}\left(\sigma_{u}(\widetilde{\tau}_{b}^{\downarrow,{(k)}})\right)^2}{|\Upsilon_{u}(\widetilde{\tau}_{b}^{\downarrow,{(k)}})|{B}_{u}(\widetilde{\tau}_{b}^{\downarrow,{(k)}})}\Bigg)\nonumber\\
    &+\frac{6}{\left(1-\zeta^{(k)}\right)}\Bigg (\frac{|\Upsilon(\widetilde{\bm{\tau}}^{\downarrow,{(k)}})|-|\Upsilon^{\mathsf{s}}(\widetilde{\bm{\tau}}^{\downarrow,{(k)}})|}{|\Upsilon(\widetilde{\bm{\tau}}^{\downarrow,{(k)}})|}\Bigg)^2\frac{|\Upsilon(\widetilde{\bm{\tau}}^{\downarrow,{(k)}})|}{|\Upsilon_{\mathsf{min}}(\widetilde{\bm{\tau}}^{\downarrow,{(k)}})|}\Bigg(\frac{\mathfrak{X}_2}{1- 4\eta_k^2\beta^2 \ell_{\mathsf{max}}^{(k)}\left(\ell_{\mathsf{max}}^{(k)}-1\right)}\Bigg)\Bigg].
 \end{align}
}%

As opposed to the bound given in \eqref{eq:GP_gen_conv_1} which depends on consecutive loss gains (i.e., term $(a)$ in \eqref{eq:GP_gen_conv_2}), the above bound only depends on the initial error $\mathfrak{L}^{(-1)}(\bm{\omega}^{(0)}) - \mathfrak{L}^{{(K)}^\star}$ (i.e., term $(a)$ in \eqref{eq:GP_gen_conv_2}). We can further simplify the above inequality as follows:

\begin{align}\label{eq:GP_gen_conv_3}
    &\frac{1}{K} \sum_{k=0}^{K-1}\mathbb{E}\left[\left\Vert\nabla{\mathfrak{L}^{({k})}(\bm{\omega}^{(k)})}\right\Vert^2\right] {\leq} \Phi,
\end{align}
where $\Phi$ is given by
{\footnotesize
\begin{align}\label{eq:GP_gen_conv_phi_1}
    &\Phi=\frac{4 \sqrt{\widehat{\ell}_{\mathsf{max}}}\left(\mathfrak{L}^{(-1)}(\bm{\omega}^{(0)})\right)}{\mathfrak{B}_k\left(1-\zeta_{\mathsf{max}}\right)\alpha \sqrt{KN}}+\frac{4}{K} \sum_{k=0}^{K-1}\frac{\sum_{b \in \Omega}\sum_{u\in \mathcal{U}_{b}} \mathfrak{D}^{(k)}_u\left(\left((T^{(k)}-T^{(k{-}1)})-\widehat{\lambda}_{u}^{(k)}\tau_{u}^{\mathsf{LC},{(k)}}\right)\right) }{\eta_{_k}\mathfrak{B}_k\left(1-\zeta^{(k)}\right)}\nonumber\\
    &+\frac{8}{K} \sum_{k=0}^{K-1}\Bigg[\frac{1}{{\left(1-\zeta^{(k)}\right)}}\Bigg({\beta^2\Theta^2 \eta_k^2}\sum_{b \in \Omega} \sum_{u\in \mathcal{U}_{b}}\frac{|\Upsilon_{u}(\widetilde{\tau}_{b}^{\downarrow,{(k)}})|}{|\Upsilon(\widetilde{\bm{\tau}}^{\downarrow,{(k)}})|}\frac{\left(\ell_u^{(k)}-1\right)}{1- 4\eta_k^2\beta^2 \ell_u^{(k)}\left(\ell_u^{(k)}-1\right)}\left(1-\frac{{B}_{u}(\widetilde{\tau}_{b}^{\downarrow,{(k)}})}{|\Upsilon_{u}(\widetilde{\tau}_{b}^{\downarrow,{(k)}})|} \right)  \frac{{(|\Upsilon_{u}(\widetilde{\tau}_{b}^{\downarrow,{(k)}})|-1)}\left(\sigma_{u}(\widetilde{\tau}_{b}^{\downarrow,{(k)}})\right)^2}{|\Upsilon_{u}(\widetilde{\tau}_{b}^{\downarrow,{(k)}})|{B}_{u}(\widetilde{\tau}_{b}^{\downarrow,{(k)}})}\Bigg)\nonumber\\
    & + \frac{1}{\left(1-\zeta^{(k)}\right)}\Bigg(\frac{\mathfrak{X}_2 \eta_k^2\beta^2 \left(\ell_{\mathsf{max}}^{(k)}\right)\left(\ell_{\mathsf{max}}^{(k)}-1\right)}{1- 4\eta_k^2\beta^2\ell_{\mathsf{max}}^{(k)}\left(\ell_{\mathsf{max}}^{(k)}-1\right)}+ \frac{\Theta^2\beta\eta_{_k}\mathfrak{B}_k}{2} \sum_{b \in \Omega} \sum_{u\in \mathcal{U}_{b}}\frac{\left(\widehat{\lambda}_{u}^{(k)}|\Upsilon_{u}(\widetilde{\tau}_{b}^{\downarrow,{(k)}})|\right)^2}{\left(|\Upsilon^{\mathsf{s}}(\widetilde{\bm{\tau}}^{\downarrow,{(k)}})|\right)^2 \ell^{(k)}_{u}}\left(1-\frac{{B}_{u}(\widetilde{\tau}_{b}^{\downarrow,{(k)}})}{|\Upsilon_{u}(\widetilde{\tau}_{b}^{\downarrow,{(k)}})|} \right) \frac{(|\Upsilon_{u}(\widetilde{\tau}_{b}^{\downarrow,{(k)}})|-1)\left(\sigma_{u}(\widetilde{\tau}_{b}^{\downarrow,{(k)}})\right)^2}{|\Upsilon_{u}(\widetilde{\tau}_{b}^{\downarrow,{(k)}})|{B}_{u}(\widetilde{\tau}_{b}^{\downarrow,{(k)}})}\Bigg)\nonumber\\
    &+\frac{24}{\left(1-\zeta^{(k)}\right)}\left (\frac{|\Upsilon(\widetilde{\bm{\tau}}^{\downarrow,{(k)}})|-|\Upsilon^{\mathsf{s}}(\widetilde{\bm{\tau}}^{\downarrow,{(k)}})|}{|\Upsilon(\widetilde{\bm{\tau}}^{\downarrow,{(k)}})|}\right)^2\sum_{b \in \Omega} \sum_{u\in \mathcal{U}_{b}}\Bigg(\frac{\Theta^2 \beta^2 \eta_k^2 \left(\ell_u^{(k)}-1\right)}{1- 4\eta_k^2\beta^2 \ell_u^{(k)}\left(\ell_u^{(k)}-1\right)} \left(1-\frac{{B}_{u}(\widetilde{\tau}_{b}^{\downarrow,{(k)}})}{|\Upsilon_{u}(\widetilde{\tau}_{b}^{\downarrow,{(k)}})|} \right)  \frac{{(|\Upsilon_{u}(\widetilde{\tau}_{b}^{\downarrow,{(k)}})|-1)}\left(\sigma_{u}(\widetilde{\tau}_{b}^{\downarrow,{(k)}})\right)^2}{|\Upsilon_{u}(\widetilde{\tau}_{b}^{\downarrow,{(k)}})|{B}_{u}(\widetilde{\tau}_{b}^{\downarrow,{(k)}})}\Bigg)\nonumber\\
    &+\frac{6}{\left(1-\zeta^{(k)}\right)}\Bigg (\frac{|\Upsilon(\widetilde{\bm{\tau}}^{\downarrow,{(k)}})|-|\Upsilon^{\mathsf{s}}(\widetilde{\bm{\tau}}^{\downarrow,{(k)}})|}{|\Upsilon(\widetilde{\bm{\tau}}^{\downarrow,{(k)}})|}\Bigg)^2\frac{|\Upsilon(\widetilde{\bm{\tau}}^{\downarrow,{(k)}})|}{|\Upsilon_{\mathsf{min}}(\widetilde{\bm{\tau}}^{\downarrow,{(k)}})|}\Bigg(\frac{\mathfrak{X}_2}{1- 4\eta_k^2\beta^2 \ell_{\mathsf{max}}^{(k)}\left(\ell_{\mathsf{max}}^{(k)}-1\right)}\Bigg)\Bigg].
 \end{align}
}%

In the following, we aim to transform \eqref{eq:GP_gen_conv_phi_1} into standard GP format.



\textbf{Convergence bound $\Phi$.}
From Sec.~\ref{sec:dynamic_dataset}, we have
\begin{align}
(\sigma_{u}(t))^2=\frac{1}{|\Upsilon_{u}(t)|-1}\sum_{\xi\in \Upsilon_{u}(t)} \Vert \bm{\xi}- \bm{\mu}_{u}(t)\Vert^2.
\end{align}
Assuming $\sum_{\xi\in \Upsilon_{u}(t)} \Vert \bm{\xi}- \bm{\mu}_{u}(t)\Vert^2\le \sigma^{\mathsf{max}}_{u}$, considering that ${B}_{u}(\widetilde{\tau}_{b}^{\downarrow,{(k)}})=\varsigma^{(k)}_u |\Upsilon_{u}(\widetilde{\tau}_{b}^{\downarrow,{(k)}})|$, and performing some algebraic manipulations on \eqref{eq:GP_gen_conv_phi_1} give us
\begin{align}\label{eq:GP_gen_conv_phi_2}
    &\Phi=\frac{4 \sqrt{\widehat{\ell}_{\mathsf{max}}}\left(\mathfrak{L}^{(-1)}(\bm{\omega}^{(0)})\right)}{\mathfrak{B}_k\left(1-\zeta_{\mathsf{max}}\right)\alpha \sqrt{KN}}+\frac{4}{K} \sum_{k=0}^{K-1}\frac{\sum_{b \in \Omega}\sum_{u\in \mathcal{U}_{b}} \mathfrak{D}^{(k)}_u\left(\left((T^{(k)}-T^{(k{-}1)})-\widehat{\lambda}_{u}^{(k)}\tau_{u}^{\mathsf{LC},{(k)}}\right)\right) }{\eta_{_k}\mathfrak{B}_k\left(1-\zeta^{(k)}\right)}\nonumber\\
    &+\frac{8}{K} \sum_{k=0}^{K-1}\Bigg[\frac{1}{{\left(1-\zeta^{(k)}\right)}}\Bigg({\beta^2\Theta^2 \eta_k^2}\sum_{b \in \Omega} \sum_{u\in \mathcal{U}_{b}}\frac{1}{|\Upsilon(\widetilde{\bm{\tau}}^{\downarrow,{(k)}})|}\frac{\left(\ell_u^{(k)}-1\right)}{1- 4\eta_k^2\beta^2 \ell_u^{(k)}\left(\ell_u^{(k)}-1\right)}\frac{\left(1-\varsigma^{(k)}_u\right)\sigma^{\mathsf{max}}_{u}}{\varsigma^{(k)}_u|\Upsilon_{u}(\widetilde{\tau}_{b}^{\downarrow,{(k)}})|}\Bigg)\nonumber\\
    & + \frac{1}{\left(1-\zeta^{(k)}\right)}\Bigg(\frac{\mathfrak{X}_2 \eta_k^2\beta^2 \left(\ell_{\mathsf{max}}^{(k)}\right)\left(\ell_{\mathsf{max}}^{(k)}-1\right)}{1- 4\eta_k^2\beta^2\ell_{\mathsf{max}}^{(k)}\left(\ell_{\mathsf{max}}^{(k)}-1\right)}+ \frac{\Theta^2\beta\eta_{_k}\mathfrak{B}_k}{2} \sum_{b \in \Omega} \sum_{u\in \mathcal{U}_{b}}\frac{\left(\widehat{\lambda}_{u}^{(k)}\right)^2}{\left(|\Upsilon^{\mathsf{s}}(\widetilde{\bm{\tau}}^{\downarrow,{(k)}})|\right)^2 \ell^{(k)}_{u}}\frac{\left(1-\varsigma^{(k)}_u\right)\sigma^{\mathsf{max}}_{u}}{\varsigma^{(k)}_u}\Bigg)\nonumber\\
    &+\frac{24}{\left(1-\zeta^{(k)}\right)}\left (\frac{|\Upsilon(\widetilde{\bm{\tau}}^{\downarrow,{(k)}})|-|\Upsilon^{\mathsf{s}}(\widetilde{\bm{\tau}}^{\downarrow,{(k)}})|}{|\Upsilon(\widetilde{\bm{\tau}}^{\downarrow,{(k)}})|}\right)^2\sum_{b \in \Omega} \sum_{u\in \mathcal{U}_{b}}\Bigg(\frac{\Theta^2 \beta^2 \eta_k^2 \left(\ell_u^{(k)}-1\right)}{1- 4\eta_k^2\beta^2 \ell_u^{(k)}\left(\ell_u^{(k)}-1\right)} \frac{\left(1-\varsigma^{(k)}_u\right)\sigma^{\mathsf{max}}_{u}}{\varsigma^{(k)}_u|\Upsilon_{u}(\widetilde{\tau}_{b}^{\downarrow,{(k)}})|^2}\Bigg)\nonumber\\
    &+\frac{6}{\left(1-\zeta^{(k)}\right)}\Bigg (\frac{|\Upsilon(\widetilde{\bm{\tau}}^{\downarrow,{(k)}})|-|\Upsilon^{\mathsf{s}}(\widetilde{\bm{\tau}}^{\downarrow,{(k)}})|}{|\Upsilon(\widetilde{\bm{\tau}}^{\downarrow,{(k)}})|}\Bigg)^2\frac{|\Upsilon(\widetilde{\bm{\tau}}^{\downarrow,{(k)}})|}{|\Upsilon_{\mathsf{min}}(\widetilde{\bm{\tau}}^{\downarrow,{(k)}})|}\Bigg(\frac{\mathfrak{X}_2}{1- 4\eta_k^2\beta^2 \ell_{\mathsf{max}}^{(k)}\left(\ell_{\mathsf{max}}^{(k)}-1\right)}\Bigg)\Bigg].
\end{align}

Furthermore, we assume that the dynamic of dataset size  of FLU $u$ is described by the following ordinary differential equation.
\begin{equation}\label{app:eq:dynamic_dataset_size2_ODE}
\frac{d|\Upsilon_{u}(t)|}{dt}{=}
\begin{cases}
   C_{u}^{\downarrow,(k)},&t{\in}[T^{(k{-}1)}, T^{(k{-}1)}{+}\tau_{b}^{\downarrow,{(k)}}),\\
   0,&t{\in}[T^{(k{-}1)}{+}\tau_{b}^{\downarrow,{(k)}},\Psi^{(k)}_{u}],\\
   C_{u}^{\uparrow,(k)},&t{\in}(\Psi^{(k)}_{u},T^{(k)}].
\end{cases}
\end{equation}
where $\Psi^{(k)}_{u}$ is the time at which LM training of FLU $u$ is completed, which is given below:
\begin{equation}\label{app:eq:psi_u_k}
    \Psi^{(k)}_{u}=T^{(k{-}1)}+\tau_{b}^{\downarrow,{(k)}}+\tau_{u}^{\mathsf{LC},{(k)}}.
\end{equation}
Integrating both hand sides of \eqref{app:eq:dynamic_dataset_size2_ODE} yields
\begin{equation}\label{eq:dynamic_dataset_size2}
\hspace{-5mm}
|\Upsilon_{u}(t)|{=}
\begin{cases}
   A^{\downarrow,(k)}_{u}+C_{u}^{\downarrow,(k)}\times \left(t-T^{(k{-}1)}\right), &t{\in}[T^{(k{-}1)}, T^{(k{-}1)}{+}\tau_{b}^{\downarrow,{(k)}}),\\
   A^{\mathsf{LC},(k)}_{u},&t{\in}[T^{(k{-}1)}{+}\tau_{b}^{\downarrow,{(k)}},\Psi^{(k)}_{u}],\\
   A^{\uparrow,(k)}_{u}+C_{u}^{\uparrow,(k)}\times \left(t-\Psi^{(k)}_{u}\right),&t{\in}(\Psi^{(k)}_{u},T^{(k)}].
\end{cases}
\hspace{-3mm}
\end{equation}
Considering boundary dataset size of FLU $u$ at time $t=T^{(k{-}1)}$ (i.e., $|\Upsilon_{u}(T^{(k{-}1)})|$) results in
\begin{equation}\label{app:eq:dataset_size_recursive}
|\Upsilon_{u}(t)|{=}
\begin{cases}
    |\Upsilon_{u}(T^{(k{-}1)})| + C_{u}^{\downarrow,(k)} \times\left(t-T^{(k{-}1)}\right), &t{\in}[T^{(k{-}1)}, T^{(k{-}1)}{+}\tau_{b}^{\downarrow,{(k)}}),\\
    |\Upsilon_{u}(T^{(k{-}1)})| + C_{u}^{\downarrow,(k)} \times \tau_{b}^{\downarrow,{(k)}},&t{\in}[T^{(k{-}1)}{+}\tau_{b}^{\downarrow,{(k)}},\Psi^{(k)}_{u}],\\
    |\Upsilon_{u}(T^{(k{-}1)})| + C_{u}^{\downarrow,(k)} \times \tau_{b}^{\downarrow,{(k)}} + C_{u}^{\uparrow,(k)} \times\left(t-\Psi^{(k)}_{u}\right),&t{\in}(\Psi^{(k)}_{u},T^{(k)}].
\end{cases}
\end{equation}
To prevent negative dataset size, we rewrite \eqref{app:eq:dataset_size_recursive} as follows:
\begin{equation}\label{app:eq:dataset_size_recursive2}
|\Upsilon_{u}(t)|{=}
\begin{cases}
    \max\left\{|\Upsilon_{u}(T^{(k{-}1)})| + C_{u}^{\downarrow,(k)} \times\left(t-T^{(k{-}1)}\right),0\right\}, &t{\in}[T^{(k{-}1)}, T^{(k{-}1)}{+}\tau_{b}^{\downarrow,{(k)}}),\\
    \max\left\{|\Upsilon_{u}(T^{(k{-}1)})| + C_{u}^{\downarrow,(k)} \times \tau_{b}^{\downarrow,{(k)}},0\right\},&t{\in}[T^{(k{-}1)}{+}\tau_{b}^{\downarrow,{(k)}},\Psi^{(k)}_{u}],\\
    \max\left\{|\Upsilon_{u}(T^{(k{-}1)})| + C_{u}^{\downarrow,(k)} \times \tau_{b}^{\downarrow,{(k)}} + C_{u}^{\uparrow,(k)} \times\left(t-\Psi^{(k)}_{u}\right),0\right\},&t{\in}(\Psi^{(k)}_{u},T^{(k)}].
\end{cases}
\end{equation}
The recursive function represented in \eqref{app:eq:dataset_size_recursive2} can be expanded to obtain a closed-form equation for the dataset size of FLU $u$, which is given below.
\begin{equation}\label{app:eq:dataset_size_closed_form_1}
|\Upsilon_{u}(t)|{=}
\begin{cases}
    \max\left\{H^{(k)}_u + C_{u}^{\downarrow,(k)} \times\left(t-T^{(k{-}1)}\right),0\right\}, &t{\in}[T^{(k{-}1)}, T^{(k{-}1)}{+}\tau_{b}^{\downarrow,{(k)}}),\\
    \max\left\{H^{(k)}_u + C_{u}^{\downarrow,(k)} \times \tau_{b}^{\downarrow,{(k)}},0\right\},&t{\in}[T^{(k{-}1)}{+}\tau_{b}^{\downarrow,{(k)}},\Psi^{(k)}_{u}],\\
    \max\left\{H^{(k)}_u + C_{u}^{\downarrow,(k)} \times\tau_{b}^{\downarrow,{(k)}} + C_{u}^{\uparrow,(k)} \times\left(t-\Psi^{(k)}_{u}\right),0\right\}, &t{\in}(\Psi^{(k)}_{u},T^{(k)}],
\end{cases}
\end{equation}
where $H^{(k)}_u$ is calculated as follows:
\begin{equation}\label{app:eq:H_u}
    H^{(k)}_u=|\Upsilon_{u}(0)|+\sum_{k'{=}1}^{k{-}1}C_{u}^{\downarrow,(k')} \times \tau_{b}^{\downarrow,{(k')}}+C_{u}^{\uparrow,(k')} \times\left(T^{(k')}{-}\Psi^{(k')}_{u}\right).
\end{equation}
In \eqref{app:eq:H_u}, $|\Upsilon_{u}(0)|$ is the initial dataset size of FLU $u$. The piecewise equation represented in \eqref{app:eq:dataset_size_closed_form_1} can be succinctly expressed as follows:
\begin{equation}\label{app:eq:dataset_size_compact_closed_form}
    |\Upsilon_{u}(t)|{=}\max\left(H^{(k)}_u{+}C_{u}^{\downarrow,(k)}\min\left(t-T^{(k{-}1)},\tau_{b}^{\downarrow,{(k)}}\right){+}C_{u}^{\uparrow,(k)}\left(\max\left(t,\Psi^{(k)}_{u}\right){-}\Psi^{(k)}_{u}\right),0\right).
\end{equation}
Let us refer to $|\Upsilon_{u}^{(k)}|$ as the dataset size of FLU $u$ utilized for LM training at global round $k$ (i.e., dataset size at time $\widetilde{\tau}_{b}^{\downarrow,{(k)}}{=}T^{(k{-}1)}{+}\tau_{b}^{\downarrow,{(k)}}$). Considering \eqref{app:eq:dataset_size_compact_closed_form}, $|\Upsilon_{u}^{(k)}|$ can be calculated as follows:
\begin{align}\label{app:eq:upsilon_u_closed}
    |\Upsilon_{u}^{(k)}|&=|\Upsilon_{u}(\widetilde{\tau}_{b}^{\downarrow,{(k)}})|=|\Upsilon_{u}(T^{(k{-}1)}{+}\tau_{b}^{\downarrow,{(k)}})|{=}\max\left(H^{(k)}_u{+}C_{u}^{\downarrow,(k)}\times\tau_{b}^{\downarrow,{(k)}},0\right).
\end{align}
Subsequently, we refer to 
\begin{equation}\label{app:eq:upsilon_total_closed}
    |\Upsilon^{(k)}|=|\Upsilon(\widetilde{\bm{\tau}}^{\downarrow,{(k)}})|=\sum_{b\in \Omega,u\in\mathcal{U}_{b}}|\Upsilon_{u}(\widetilde{\tau}_{b}^{\downarrow,{(k)}})|=\sum_{b\in \Omega,u\in\mathcal{U}_{b}}|\Upsilon_{u}^{(k)}|
\end{equation} 
and 
\begin{equation}\label{app:eq:upsilon_recruited_closed}
\hspace{-5mm}
    |\Upsilon^{\mathsf{s},(k)}|{=}|\Upsilon^{\mathsf{s}}(\widetilde{\bm{\tau}}^{\downarrow,{(k)}})|{=}\sum_{b\in \Omega,u\in\mathcal{U}_{b}}\widehat{\lambda}_{u}^{(k)} |\Upsilon_{u}(\widetilde{\tau}_{b}^{\downarrow,{(k)}})|{=}\sum_{b\in \Omega,u\in\mathcal{U}_{b}}\widehat{\lambda}_{u}^{(k)} |\Upsilon_{u}^{(k)}|
\hspace{-3mm}
\end{equation}
as the size of cumulative dataset of all FLUs and recruited FLUs, respectively, utilized at global training round $k$.

Let us define $H^{(k)}_u$, $|\Upsilon_{u}^{(k)}|$, $|\Upsilon^{(k)}|$, and $|\Upsilon^{\mathsf{s},(k)}|$ as three auxiliary decision variables satisfying \eqref{app:eq:H_u}, \eqref{app:eq:upsilon_u_closed}, \eqref{app:eq:upsilon_total_closed}, and \eqref{app:eq:upsilon_recruited_closed}, respectively. We finally rewrite \eqref{eq:GP_gen_conv_phi_1} as the following posynomial, which is in the standard form of GP:
\begin{align}\label{eq:GP_gen_conv_phi_3}
    &\Phi=\frac{4 \sqrt{\widehat{\ell}_{\mathsf{max}}}\left(\mathfrak{L}^{(-1)}(\bm{\omega}^{(0)})\right)}{\mathfrak{B}_k\left(1-\zeta_{\mathsf{max}}\right)\alpha \sqrt{KN}}+\frac{4}{K} \sum_{k=0}^{K-1}\frac{\sum_{b \in \Omega}\sum_{u\in \mathcal{U}_{b}} \mathfrak{D}^{(k)}_u\widehat{\tau}_{u}^{\mathsf{LC},{(k)}} }{\eta_{_k}\mathfrak{B}_k\left(1-\zeta^{(k)}\right)}+\frac{8}{K} \sum_{k=0}^{K-1}\Bigg[\frac{\beta^2\Theta^2 \eta_k^2}{{\left(1-\zeta^{(k)}\right)}}\sum_{b \in \Omega} \sum_{u\in \mathcal{U}_{b}}\frac{\overline{\ell}_u^{(k)}\overline{\varsigma}^{(k)}_u\sigma^{\mathsf{max}}_{u}}{|\Upsilon^{(k)}||\Upsilon_{u}^{(k)}|}\nonumber\\
    & + \frac{1}{\left(1-\zeta^{(k)}\right)}\Bigg(\frac{\mathfrak{X}_2 \eta_k^2\beta^2 \left(\ell_{\mathsf{max}}^{(k)}\right)^2}{\overline{\ell}_{\mathsf{max}}^{(k)}}+ \frac{\Theta^2\beta\eta_{_k}\mathfrak{B}_k}{2} \sum_{b \in \Omega} \sum_{u\in \mathcal{U}_{b}}\frac{\left(\widehat{\lambda}_{u}^{(k)}\right)^2 \overline{\varsigma}^{(k)}_u\sigma^{\mathsf{max}}_{u}}{\left(|\Upsilon^{\mathsf{s},(k)}|\right)^2 \ell^{(k)}_{u}}\Bigg)\nonumber\\
    &+\frac{24\left(|\overline{\Upsilon}^{(k)}|\right)^2}{\left(1-\zeta^{(k)}\right)\left(|\Upsilon^{(k)}|\right)^2}\sum_{b \in \Omega} \sum_{u\in \mathcal{U}_{b}}\Bigg( \frac{\Theta^2 \beta^2 \eta_k^2 \overline{\ell}_u^{(k)}\overline{\varsigma}^{(k)}_u\sigma^{\mathsf{max}}_{u}}{\left(|\Upsilon_{u}^{(k)}|\right)^2}\Bigg)+\frac{6\left(|\overline{\Upsilon}^{(k)}|\right)^2\mathfrak{X}_2}{\left(1-\zeta^{(k)}\right)|\Upsilon^{(k)}||\Upsilon_{\mathsf{min}}^{(k)}|\overline{\ell}_{\mathsf{max}}^{(k)}}\Bigg]\nonumber\\
    &=\frac{4 \sqrt{\widehat{\ell}_{\mathsf{max}}}\left(\mathfrak{L}^{(-1)}(\bm{\omega}^{(0)})\right)}{\alpha \sqrt{KN}\mathfrak{B}_k\left(1-\zeta_{\mathsf{max}}\right)}+\frac{4}{K} \sum_{k=0}^{K-1}\frac{\sum_{b \in \Omega}\sum_{u\in \mathcal{U}_{b}} \mathfrak{D}^{(k)}_u\widehat{\tau}_{u}^{\mathsf{LC},{(k)}} }{\eta_{_k}\mathfrak{B}_k\left(1-\zeta_{\mathsf{max}}\right)}+\frac{8}{K\left(1-\zeta_{\mathsf{max}}\right)} \sum_{k=0}^{K-1}\Bigg[\beta^2\Theta^2 \eta_k^2\sum_{b \in \Omega} \sum_{u\in \mathcal{U}_{b}}\frac{\ell_u^{(k)}\overline{\varsigma}^{(k)}_u\sigma^{\mathsf{max}}_{u}}{\overline{\ell}_{\mathsf{max}}^{(k)}|\Upsilon^{(k)}||\Upsilon_{u}^{(k)}|}\nonumber\\
    & +\Bigg(\frac{\mathfrak{X}_2 \eta_k^2\beta^2 \left(\ell_{\mathsf{max}}^{(k)}\right)^2}{\overline{\ell}_{\mathsf{max}}^{(k)}}+ \frac{\Theta^2\beta\eta_{_k}\mathfrak{B}_k}{2} \sum_{b \in \Omega} \sum_{u\in \mathcal{U}_{b}}\frac{\left(\widehat{\lambda}_{u}^{(k)}\right)^2 \overline{\varsigma}^{(k)}_u\sigma^{\mathsf{max}}_{u}}{\left(|\Upsilon^{\mathsf{s},(k)}|\right)^2 \ell^{(k)}_{u}}\Bigg)\nonumber\\
    &+\frac{24\left(|\overline{\Upsilon}^{(k)}|\right)^2}{\left(|\Upsilon^{(k)}|\right)^2}\sum_{b \in \Omega} \sum_{u\in \mathcal{U}_{b}}\Bigg( \frac{\Theta^2 \beta^2 \eta_k^2 \ell_u^{(k)}\overline{\varsigma}^{(k)}_u\sigma^{\mathsf{max}}_{u}}{\overline{\ell}_{\mathsf{max}}^{(k)}\left(|\Upsilon_{u}^{(k)}|\right)^2}\Bigg)+\frac{6\left(|\overline{\Upsilon}^{(k)}|\right)^2\mathfrak{X}_2}{|\Upsilon^{(k)}||\Upsilon_{\mathsf{min}}^{(k)}|\overline{\ell}_{\mathsf{max}}^{(k)}}\Bigg],
\end{align}
where $|\overline{\Upsilon}^{(k)}|$, $\overline{\ell}_u^{(k)}$, $\overline{\ell}_{\mathsf{max}}^{(k)}$,  $\widehat{\tau}_{u}^{\mathsf{LC},{(k)}}$, $\overline{\varsigma}^{(k)}_u$, $|\Upsilon_{\mathsf{min}}^{(k)}|$, and $\ell^{(k)}_{\mathsf{max}}$ are decision variables satisfying the following constraints: 

\begin{equation}\label{app:eq:upsilon_overline}
    |\overline{\Upsilon}^{(k)}|=|\Upsilon^{(k)}|-|\Upsilon^{\mathsf{s},(k)}|,
\end{equation}
\begin{equation}\label{app:eq:e_max_overline}
   \overline{\ell}_{\mathsf{max}}^{(k)}= 1- 4\eta_k^2\beta^2 \left(\ell_{\mathsf{max}}^{(k)}\right)^2,
\end{equation}
\begin{equation}\label{app:eq:tau_LC_hat}
    \widehat{\tau}_{u}^{\mathsf{LC},{(k)}}=(T^{(k)}-T^{(k{-}1)})-\widehat{\lambda}_{u}^{(k)}\tau_{u}^{\mathsf{LC},{(k)}},
\end{equation}
\begin{equation}\label{app:eq:varsigma_overline}
    \overline{\varsigma}^{(k)}_u=\frac{1}{\varsigma^{(k)}_u} - 1.
\end{equation}
Moreover, from Theorem~\ref{th:main}, we have
\begin{equation}\label{app:eq:upsilon_min}
    |\Upsilon_{\mathsf{min}}^{(k)}|=\min_{b{\in}\Omega,u{\in} \mathcal{U}_{b}}\{|\Upsilon_{u}^{(k)}|\},
\end{equation}
\begin{equation}\label{app:eq:ell_max}
    \ell^{(k)}_{\mathsf{max}}=\max_{b\in\Omega,u{\in} \mathcal{U}_{b}}\{\ell^{(k)}_u\},
\end{equation}

and from Theorem~\ref{th:sufficient_conditions}, the following sufficient conditions must be met:

$\bullet~\bm{\mathfrak{C}^{(L)}}$: \textbf{Sufficient Condition on Training Latency.} 
\begin{align}\label{app:suf:main:gamma_upper}
    T^{(k)}-T^{(k{-}1)} &\le \frac{\sum_{b \in \Omega}\sum_{u\in \mathcal{U}_{b}} \mathfrak{D}^{(k)}_u\widehat{\lambda}_{u}^{(k)}\tau_{u}^{\mathsf{LC},{(k)}}}{\sum_{b \in \Omega}\sum_{u\in \mathcal{U}_{b}} \mathfrak{D}^{(k)}_u}+ \frac{\eta_{_k}\mathfrak{B}_k\left(1-\zeta^{(k)}\right)\vartheta^k}{4\sum_{b \in \Omega}\sum_{u\in \mathcal{U}_{b}} \mathfrak{D}^{(k)}_u}.
\end{align}

$\bullet~\bm{\mathfrak{C}^{(\Upsilon)}}$: \textbf{Sufficient Condition on $|\Upsilon^{\mathsf{s},(k)}|$.}
\begin{equation}\label{app:suf:main:recruitment_overline}
    \left(|\overline{\Upsilon}^{(k)}|\right)^2{<}|\Upsilon^{(k)}|\min\left\{\frac{\widehat{\zeta}^{(k)}|\Upsilon_{\mathsf{min}}^{(k)}|}{48\mathfrak{X}_1},\frac{\varpi^k \left(1-\zeta^{(k)}\right)|\Upsilon_{\mathsf{min}}^{(k)}|\overline{\ell}_{\mathsf{max}}^{(k)}}{48\mathfrak{X}_2}\right\}.
\end{equation}


$\bullet~\bm{\mathfrak{C}^{(\sigma)}}$: \textbf{Sufficient Condition on Gradient Sampling Noise.}
\begin{equation}\label{app:suf:main:sigma}
    \max_{k\in\mathcal{K}}\max_{b{\in}\Omega}\max_{u{\in}\mathcal{U}_{b}}\left\{\frac{\overline{\varsigma}^{(k)}_u\sigma^{\mathsf{max}}_{u}}{|\Upsilon_{u}^{(k)}|}\right\}{\leq} \sigma_{\mathsf{max}}.
\end{equation}

We next aim to transform the above constraints into standard form of GP:

\textbf{Constraint~\eqref{app:eq:H_u}:} Performing some algebraic manipulations on \eqref{app:eq:H_u} gives us
\begin{equation}\label{app:eq:H_u_1}
    \frac{H^{(k)}_u}{|\Upsilon_{u}(0)|+\sum_{k'{=}2}^{k{-}1}\max\left\{C_{u}^{\downarrow,(k')} \times \tau_{b}^{\downarrow,{(k')}}+C_{u}^{\uparrow,(k')} \times\left(T^{(k')}{-}\Psi^{(k')}_{u}\right),0\right\}}= 1.
\end{equation}
Defining two auxiliary decision variables $Q_{u}^{+,(k')}$ and $Q_{u}^{(k')}$ satisfying 
\begin{equation}\label{app:eq:H_u_2}
    Q_{u}^{+,(k')}=\max\left\{Q_{u}^{(k')},0\right\}
\end{equation}
and
\begin{equation}\label{app:eq:H_u_3}
    Q_{u}^{(k')} = C_{u}^{\downarrow,(k')} \times \tau_{b}^{\downarrow,{(k')}}+C_{u}^{\uparrow,(k')} \times\left(T^{(k')}{-}\Psi^{(k')}_{u}\right)
\end{equation}
results in 
\begin{equation}\label{app:eq:H_u_4}
    \frac{H^{(k)}_u}{|\Upsilon_{u}(0)|+\sum_{k'{=}2}^{k{-}1}Q_{u}^{+,(k')}}= 1.
\end{equation}
This constraint is not in the format of GP; therefore, we transform it by splitting it into the following three inequalities:
\begin{equation}\label{app:eq:H_u_4_2}
    \frac{H^{(k)}_u}{|\Upsilon_{u}(0)|+\sum_{k'{=}2}^{k{-}1}Q_{u}^{+,(k')}}\le 1,
\end{equation}
\begin{equation}\label{app:eq:H_u_4_3}
    \frac{|\Upsilon_{u}(0)|+\sum_{k'{=}2}^{k{-}1}Q_{u}^{+,(k')}}{\mathscr{C}^{(k)}_{u} H^{(k)}_u}\le 1,
\end{equation}
\begin{equation}
    \mathscr{C}^{(k)}_{u}\ge 1,
\end{equation}
where $\mathscr{C}^{(k)}_{u}$ is added with a large penalty term to the objective function to force $\mathscr{C}^{(k)}_{u}{\rightarrow}1^+$ at the optimal point. The fraction in~\eqref{app:eq:H_u_4_3} still needs transformation since it is an inequality with a posynomial in the denominator, which is not a posynomial. We thus exploit arithmetic-geometric mean inequality (Lemma~\ref{Lemma:ArethmaticGeometric}) to approximate the denominator with a monomial:
\begin{align}\label{app:eq:H_u_4_333}
    G^{(k)}_{u}(\bm{v})=|\Upsilon_{u}(0)|+\sum_{k'{=}2}^{k{-}1}Q_{u}^{+,(k')} &\geq \widehat{G}^{(k)}_{u}(\bm{v};\ell) \triangleq \left(G^{(k)}_{u}([\bm{v}]^{(\ell-1)})\right)^{\frac{\left[|\Upsilon_{u}(0)|\right]^{(\ell-1)}}{G^{(k)}_{u}([\bm{v}]^{(\ell-1)})}}\times\prod_{k'{=}2}^{k{-}1}\left(\frac{Q_{u}^{+,(k')} G^{(k)}_{u}([\bm{v}]^{(\ell-1)})}{\left[Q_{u}^{+,(k')}\right]^{(\ell-1)}}\right)^{\frac{\left[Q_{u}^{+,(k')}\right]^{(\ell-1)}}{G^{(k)}_{u}([\bm{v}]^{(\ell-1)})}},
\end{align}
which gives us an approximation of~\eqref{app:eq:H_u_4_3} as follows:
\begin{equation}\label{app:eq:H_u_4_4}
    \frac{H^{(k)}_u}{\widehat{G}^{(k)}_{u}(\bm{v};\ell)}\le 1.
\end{equation}

We finally approximate constraint~\eqref{app:eq:H_u_4} as follows:
\begin{tcolorbox}[ams align]
     &\frac{H^{(k)}_u}{\widehat{G}^{(k)}_{u}(\bm{v};\ell)}\le 1,\nonumber\\
     &\frac{|\Upsilon_{u}(0)|+\sum_{k'{=}2}^{k{-}1}Q_{u}^{+,(k')}}{\mathscr{C}^{(k)}_{u} H^{(k)}_u}\le 1,\nonumber\\
     &\frac{1}{\mathscr{C}^{(k)}_{u}}\le 1.\nonumber
\end{tcolorbox}

We next aim to approximate \eqref{app:eq:H_u_2}. To this end, we rewrite \eqref{app:eq:H_u_2} as follows:
\begin{equation}\label{app:eq:H_u_2_1}
    \frac{Q_{u}^{+,(k')}}{\max\left\{Q_{u}^{(k')},\epsilon\right\}} = 1,
\end{equation}
where $\epsilon$ is added to the denominator of the above equation to avoid division by zero.
Using the approximation $\max\{A, B\}\approx (A^{p}+B^{p})^{-\frac{1}{p}}$, which is tight when $p \gg 1$, gives us
\begin{equation}\label{app:eq:H_u_2_2}
    \frac{Q_{u}^{+,(k')}}{\left(\left(Q_{u}^{(k')}\right)^{p}+(\epsilon)^{p}\right)^{-\frac{1}{p}}} = 1.
\end{equation}
The above equation can be rewritten as the following standard monomial equality constraint:
\begin{equation}\label{app:eq:H_u_2_3}
    \left(Q_{u}^{(k'),\epsilon}\right)^{\frac{1}{p}}Q_{u}^{+,(k')} = 1,
\end{equation}
where $Q_{u}^{(k'),\epsilon}$ is an auxiliary decision variable satisfying
\begin{equation}\label{app:eq:H_u_2_4}
    Q_{u}^{(k'),\epsilon} = \left(Q_{u}^{(k')}\right)^{p}+(\epsilon)^{p}.
\end{equation}

However, \eqref{app:eq:H_u_2_4} is not in the format of GP and needs to be transformed to satisfy the GP requirements. In doing so, we rewrite \eqref{app:eq:H_u_2_4} as follows:
\begin{equation}\label{app:eq:H_u_2_4_6}
        \frac{Q_{u}^{(k'),\epsilon}}{\left(Q_{u}^{(k')}\right)^{p}+(\epsilon)^{p}} = 1.
\end{equation}
This constraint is still not in the format of GP. Therefore, we transform it via splitting it into the following three inequalities:
\begin{equation}\label{app:eq:H_u_2_4_7}
    \frac{Q_{u}^{(k'),\epsilon}}{\left(Q_{u}^{(k')}\right)^{p}+(\epsilon)^{p}}\le 1,
\end{equation}
\begin{equation}\label{app:eq:H_u_2_4_8}
    \frac{\left(Q_{u}^{(k')}\right)^{p}+(\epsilon)^{p}}{\mathscr{F}^{(k)}_{u}Q_{u}^{(k'),\epsilon}}\le 1,
\end{equation}
\begin{equation}
    \mathscr{F}^{(k)}_{u}\ge 1,
\end{equation}
where $\mathscr{F}^{(k)}_{u}$ is added with a large penalty term to the objective function to force $\mathscr{F}^{(k)}_{u}{\rightarrow}1^+$ at the optimal point. The fraction in~\eqref{app:eq:H_u_2_4_7} still needs transformation since it is an inequality with a posynomial in the denominator, which is not a posynomial. We thus exploit arithmetic-geometric mean inequality (Lemma~\ref{Lemma:ArethmaticGeometric}) to approximate the denominator with a monomial:
\begin{align}\label{app:eq:H_u_2_4_8}
    V^{(k)}_{u}(\bm{v})=\left(Q_{u}^{(k')}\right)^{p}+(\epsilon)^{p} &\geq \widehat{V}^{(k)}_{u}(\bm{v};\ell) \triangleq \left(\frac{\left(Q_{u}^{(k')}\right)^{p} V^{(k)}_{u}([\bm{v}]^{(\ell-1)})}{\left[\left(Q_{u}^{(k')}\right)^{p}\right]^{(\ell-1)}}\right)^{\frac{\left[\left(Q_{u}^{(k')}\right)^{p}\right]^{(\ell-1)}}{V^{(k)}_{u}([\bm{v}]^{(\ell-1)})}}\times \left( V^{(k)}_{u}([\bm{v}]^{(\ell-1)})\right)^{\frac{\left[(\epsilon)^{p}\right]^{(\ell-1)}}{V^{(k)}_{u}([\bm{v}]^{(\ell-1)})}},
\end{align}
which gives us an approximation of~\eqref{app:eq:H_u_2_4_7} as follows:
\begin{equation}\label{app:eq:H_u_2_4_9}
    \frac{Q_{u}^{(k'),\epsilon}}{\widehat{V}^{(k)}_{u}(\bm{v};\ell)}\le 1.
\end{equation}
We finally approximate constraint~\eqref{app:eq:H_u_2} as follows:
\begin{tcolorbox}[ams align]
     &\left(Q_{u}^{(k'),\epsilon}\right)^{\frac{1}{p}}Q_{u}^{+,(k')} = 1,~~~~~\frac{Q_{u}^{(k'),\epsilon}}{\widehat{V}^{(k)}_{u}(\bm{v};\ell)}\le 1,~~~~~\frac{\left(Q_{u}^{(k')}\right)^{p}+(\epsilon)^{p}}{\mathscr{F}^{(k)}_{u}Q_{u}^{(k'),\epsilon}}\le 1,~~~~~\frac{1}{\mathscr{F}^{(k)}_{u}}\le 1.\nonumber
\end{tcolorbox}

Next, we transform \eqref{app:eq:H_u_3} into standard format of GP. In doing so, considering $\Psi^{(k)}_{u}$ presented in \eqref{app:eq:psi_u_k}, we rewrite \eqref{app:eq:H_u_3} as follows:
\begin{equation}\label{app:eq:H_u_3_1}
    \frac{Q_{u}^{(k')} +  C_{u}^{\uparrow,(k')}T^{(k'{-}1)}+ C_{u}^{\uparrow,(k')}\tau_{b}^{\downarrow,{(k')}}+ C_{u}^{\uparrow,(k')}\tau_{u}^{\mathsf{LC},{(k')}}}{C_{u}^{\downarrow,(k')}  \tau_{b}^{\downarrow,{(k')}}+ C_{u}^{\uparrow,(k')}T^{(k')}} = 1.
\end{equation}
This constraint is still not in the format of GP. Therefore, we transform it via splitting it into the following three inequalities:
\begin{equation}\label{app:eq:H_u_3_7}
    \frac{Q_{u}^{(k')} +  C_{u}^{\uparrow,(k')}T^{(k'{-}1)}+ C_{u}^{\uparrow,(k')}\tau_{b}^{\downarrow,{(k')}}+ C_{u}^{\uparrow,(k')}\tau_{u}^{\mathsf{LC},{(k')}}}{C_{u}^{\downarrow,(k')}  \tau_{b}^{\downarrow,{(k')}}+ C_{u}^{\uparrow,(k')}T^{(k')}}\le 1,
\end{equation}
\begin{equation}\label{app:eq:H_u_3_8}
    \frac{\left(\mathscr{J}^{(k)}_{u}\right)^{-1}\left(C_{u}^{\downarrow,(k')}  \tau_{b}^{\downarrow,{(k')}}+ C_{u}^{\uparrow,(k')}T^{(k')}\right)}{Q_{u}^{(k')} +  C_{u}^{\uparrow,(k')}T^{(k'{-}1)}+ C_{u}^{\uparrow,(k')}\tau_{b}^{\downarrow,{(k')}}+ C_{u}^{\uparrow,(k')}\tau_{u}^{\mathsf{LC},{(k')}}}\le 1,
\end{equation}
\begin{equation}
    \mathscr{J}^{(k)}_{u}\ge 1,
\end{equation}
where $\mathscr{J}^{(k)}_{u}$ is added with a large penalty term to the objective function to force $\mathscr{J}^{(k)}_{u}{\rightarrow}1^+$ at the optimal point. The fractions in \eqref{app:eq:H_u_3_7} and \eqref{app:eq:H_u_3_8} still need transformation since they are inequalities with posynomials in their denominators, which are not posynomials. We thus exploit arithmetic-geometric mean inequality (Lemma~\ref{Lemma:ArethmaticGeometric}) to approximate the denominator of \eqref{app:eq:H_u_3_7} and \eqref{app:eq:H_u_3_8} with monomials. To this end, we approximate \eqref{app:eq:H_u_3_7} as follows:
\begin{align}\label{app:eq:H_u_3_9}
    &M^{(k')}_{u}(\bm{v})= C_{u}^{\downarrow,(k')}\tau_{b}^{\downarrow,{(k')}}+ C_{u}^{\uparrow,(k')}T^{(k')} \nonumber\\
    &\geq \widehat{M}^{(k')}_{u}(\bm{v};\ell) \triangleq \left(\frac{\tau_{b}^{\downarrow,{(k')}} M^{(k')}_{u}([\bm{v}]^{(\ell-1)})}{\left[\tau_{b}^{\downarrow,{(k')}}\right]^{(\ell-1)}}\right)^{\frac{\left[C_{u}^{\downarrow,(k')}\tau_{b}^{\downarrow,{(k')}}\right]^{(\ell-1)}}{M^{(k')}_{u}([\bm{v}]^{(\ell-1)})}}\times \left( M^{(k')}_{u}([\bm{v}]^{(\ell-1)})\right)^{\frac{\left[(\epsilon)^{p}\right]^{(\ell-1)}}{M^{(k')}_{u}([\bm{v}]^{(\ell-1)})}}\nonumber\\
    &\times \left(\frac{T^{(k')} M^{(k')}_{u}([\bm{v}]^{(\ell-1)})}{\left[T^{(k')}\right]^{(\ell-1)}}\right)^{\frac{\left[C_{u}^{\uparrow,(k')}T^{(k')}\right]^{(\ell-1)}}{M^{(k')}_{u}([\bm{v}]^{(\ell-1)})}}\times \left( M^{(k')}_{u}([\bm{v}]^{(\ell-1)})\right)^{\frac{\left[(\epsilon)^{p}\right]^{(\ell-1)}}{M^{(k')}_{u}([\bm{v}]^{(\ell-1)})}},
\end{align}
which gives us an approximation of~\eqref{app:eq:H_u_3_7} as follows:
\begin{equation}\label{app:eq:H_u_3_10}
    \frac{Q_{u}^{(k')} +  C_{u}^{\uparrow,(k')}T^{(k'{-}1)}+ C_{u}^{\uparrow,(k')}\tau_{b}^{\downarrow,{(k')}}+ C_{u}^{\uparrow,(k')}\tau_{u}^{\mathsf{LC},{(k')}}}{\widehat{M}^{(k')}_{u}(\bm{v};\ell)}\le 1.
\end{equation}

Similarly, we approximate \eqref{app:eq:H_u_3_8} as follows:
\begin{align}\label{app:eq:H_u_3_11}
    &N^{(k')}_{u}(\bm{v})= Q_{u}^{(k')} +  C_{u}^{\uparrow,(k')}T^{(k'{-}1)}+ C_{u}^{\uparrow,(k')}\tau_{b}^{\downarrow,{(k')}}+ C_{u}^{\uparrow,(k')}\tau_{u}^{\mathsf{LC},{(k')}} \nonumber\\
    &\geq \widehat{N}^{(k')}_{u}(\bm{v};\ell) \triangleq \left(\frac{Q_{u}^{(k')} N^{(k')}_{u}([\bm{v}]^{(\ell-1)})}{\left[Q_{u}^{(k')}\right]^{(\ell-1)}}\right)^{\frac{\left[Q_{u}^{(k')}\right]^{(\ell-1)}}{N^{(k')}_{u}([\bm{v}]^{(\ell-1)})}}\times \left(\frac{T^{(k'{-}1)} N^{(k')}_{u}([\bm{v}]^{(\ell-1)})}{\left[T^{(k'{-}1)}\right]^{(\ell-1)}}\right)^{\frac{\left[ C_{u}^{\uparrow,(k')}T^{(k'{-}1)}\right]^{(\ell-1)}}{N^{(k')}_{u}([\bm{v}]^{(\ell-1)})}}\nonumber\\
    &\times \left(\frac{\tau_{b}^{\downarrow,{(k')}} N^{(k')}_{u}([\bm{v}]^{(\ell-1)})}{\left[\tau_{b}^{\downarrow,{(k')}}\right]^{(\ell-1)}}\right)^{\frac{\left[C_{u}^{\uparrow,(k')}\tau_{b}^{\downarrow,{(k')}}\right]^{(\ell-1)}}{N^{(k')}_{u}([\bm{v}]^{(\ell-1)})}} \times \left(\frac{\tau_{u}^{\mathsf{LC},{(k')}} N^{(k')}_{u}([\bm{v}]^{(\ell-1)})}{\left[\tau_{u}^{\mathsf{LC},{(k')}}\right]^{(\ell-1)}}\right)^{\frac{\left[C_{u}^{\uparrow,(k')}\tau_{u}^{\mathsf{LC},{(k')}}\right]^{(\ell-1)}}{N^{(k')}_{u}([\bm{v}]^{(\ell-1)})}},
\end{align}
which gives us an approximation of~\eqref{app:eq:H_u_3_8} as follows:
\begin{equation}\label{app:eq:H_u_3_12}
    \frac{\left(\mathscr{J}^{(k)}_{u}\right)^{-1}\left(C_{u}^{\downarrow,(k')}  \tau_{b}^{\downarrow,{(k')}}+ C_{u}^{\uparrow,(k')}T^{(k')}\right)}{\widehat{N}^{(k')}_{u}(\bm{v};\ell)}\le 1.
\end{equation}
We finally approximate constraint~\eqref{app:eq:H_u_3} as follows:
\begin{tcolorbox}[ams align]
     &\frac{Q_{u}^{(k')} +  C_{u}^{\uparrow,(k')}T^{(k'{-}1)}+ C_{u}^{\uparrow,(k')}\tau_{b}^{\downarrow,{(k')}}+ C_{u}^{\uparrow,(k')}\tau_{u}^{\mathsf{LC},{(k')}}}{\widehat{M}^{(k')}_{u}(\bm{v};\ell)}\le 1,\nonumber\\
     &\frac{\left(\mathscr{J}^{(k)}_{u}\right)^{-1}\left(C_{u}^{\downarrow,(k')}  \tau_{b}^{\downarrow,{(k')}}+ C_{u}^{\uparrow,(k')}T^{(k')}\right)}{\widehat{N}^{(k')}_{u}(\bm{v};\ell)}\le 1,\nonumber\\
     & \frac{1}{\mathscr{J}^{(k)}_{u}}\le 1.\nonumber
\end{tcolorbox}

\textbf{Constraint~\eqref{app:eq:upsilon_u_closed}:} Performing some algebraic manipulations on \eqref{app:eq:upsilon_u_closed} gives us
\begin{align}\label{app:eq:upsilon_u_closed_0}
    \frac{|\Upsilon_{u}^{(k)}|}{\max\left(H^{(k)}_u{+}C_{u}^{\downarrow,(k)}\times\tau_{b}^{\downarrow,{(k)}},0\right)}= 1.
\end{align}
Defining positive auxiliary decision variable $S^{(k)}_u$ satisfying
\begin{equation}\label{app:eq:upsilon_u_closed_1}
    S^{(k)}_u = H^{(k)}_u{+}C_{u}^{\downarrow,(k)}\times\tau_{b}^{\downarrow,{(k)}}
\end{equation}
gives us
\begin{align}\label{app:eq:upsilon_u_closed_2}
    \frac{|\Upsilon_{u}^{(k)}|}{S^{(k)}_u}= 1.
\end{align}

We transform \eqref{app:eq:upsilon_u_closed_1} into standard GP format through rewrite \eqref{app:eq:upsilon_u_closed_1} as follows:
\begin{equation}\label{app:eq:upsilon_u_closed_1_1}
    \frac{S^{(k)}_u}{H^{(k)}_u{+}C_{u}^{\downarrow,(k)}\times\tau_{b}^{\downarrow,{(k)}}} = 1.
\end{equation}
This constraint is still not in the format of GP. Therefore, we transform it via splitting it into the following three inequalities:
\begin{equation}\label{app:eq:upsilon_u_closed_1_7}
    \frac{S^{(k)}_u}{H^{(k)}_u{+}C_{u}^{\downarrow,(k)}\tau_{b}^{\downarrow,{(k)}}}\le 1,
\end{equation}
\begin{equation}\label{app:eq:upsilon_u_closed_1_8}
    \frac{H^{(k)}_u{+}C_{u}^{\downarrow,(k)}\times\tau_{b}^{\downarrow,{(k)}}}{\mathscr{Y}^{(k)}_{u}S^{(k)}_u}\le 1,
\end{equation}
\begin{equation}
    \mathscr{Y}^{(k)}_{u}\ge 1,
\end{equation}
where $\mathscr{Y}^{(k)}_{u}$ is added with a large penalty term to the objective function to force $\mathscr{Y}^{(k)}_{u}{\rightarrow}1^+$ at the optimal point. The fraction in~\eqref{app:eq:upsilon_u_closed_1_7} still needs transformation since it is an inequality with a posynomial in the denominator, which is not a posynomial. We thus exploit arithmetic-geometric mean inequality (Lemma~\ref{Lemma:ArethmaticGeometric}) to approximate the denominator with a monomial:
\begin{align}\label{app:eq:upsilon_u_closed_1_8}
    Y^{(k)}_{u}(\bm{v})=H^{(k)}_u{+}C_{u}^{\downarrow,(k)}\tau_{b}^{\downarrow,{(k)}} &\geq \widehat{Y}^{(k)}_{u}(\bm{v};\ell) \triangleq \left(\frac{H^{(k)}_u Y^{(k)}_{u}([\bm{v}]^{(\ell-1)})}{\left[H^{(k)}_u\right]^{(\ell-1)}}\right)^{\frac{\left[H^{(k)}_u\right]^{(\ell-1)}}{Y^{(k)}_{u}([\bm{v}]^{(\ell-1)})}}\times \left(\frac{\tau_{b}^{\downarrow,{(k)}} Y^{(k)}_{u}([\bm{v}]^{(\ell-1)})}{\left[\tau_{b}^{\downarrow,{(k)}}\right]^{(\ell-1)}}\right)^{\frac{\left[C_{u}^{\downarrow,(k)}\tau_{b}^{\downarrow,{(k)}}\right]^{(\ell-1)}}{Y^{(k)}_{u}([\bm{v}]^{(\ell-1)})}},
\end{align}
which gives us an approximation of~\eqref{app:eq:upsilon_u_closed_1_7} as follows:
\begin{equation}\label{app:eq:upsilon_u_closed_1_9}
    \frac{S^{(k)}_u}{\widehat{Y}^{(k)}_{u}(\bm{v};\ell)}\le 1.
\end{equation}
We finally approximate constraint~\eqref{app:eq:upsilon_u_closed_1} as follows:
\begin{tcolorbox}[ams align]
     &\frac{S^{(k)}_u}{\widehat{Y}^{(k)}_{u}(\bm{v};\ell)}\le 1,~~~~~\frac{H^{(k)}_u{+}C_{u}^{\downarrow,(k)}\times\tau_{b}^{\downarrow,{(k)}}}{\mathscr{Y}^{(k)}_{u}S^{(k)}_u}\le 1,~~~~~\frac{1}{\mathscr{Y}^{(k)}_{u}}\le 1.\nonumber
\end{tcolorbox}

\textbf{Constraint~\eqref{app:eq:upsilon_total_closed}:} Performing some algebraic manipulations on \eqref{app:eq:upsilon_total_closed} gives us
\begin{equation}\label{app:eq:upsilon_total_closed_1}
    \frac{|\Upsilon^{(k)}|}{\sum_{b\in \Omega,u\in\mathcal{U}_{b}}|\Upsilon_{u}^{(k)}|}= 1.
\end{equation} 
This constraint is not in the format of GP; therefore, we transform it by splitting it into the following three inequalities:
\begin{equation}\label{app:eq:upsilon_total_closed_2}
    \frac{|\Upsilon^{(k)}|}{\sum_{b\in \Omega,u\in\mathcal{U}_{b}}|\Upsilon_{u}^{(k)}|}\le 1,
\end{equation}
\begin{equation}\label{app:eq:upsilon_total_closed_3}
    \frac{\sum_{b\in \Omega,u\in\mathcal{U}_{b}}|\Upsilon_{u}^{(k)}|}{\mathscr{O}^{(k)}_{u}|\Upsilon^{(k)}|}\le 1,
\end{equation}
\begin{equation}
    \mathscr{O}^{(k)}_{u}\ge 1,
\end{equation}
where $\mathscr{O}^{(k)}_{u}$ is added with a large penalty term to the objective function to force $\mathscr{O}^{(k)}_{u}{\rightarrow}1^+$ at the optimal point. The fraction in~\eqref{app:eq:upsilon_total_closed_2} still needs transformation since it is an inequality with a posynomial in the denominator, which is not a posynomial. We thus exploit arithmetic-geometric mean inequality (Lemma~\ref{Lemma:ArethmaticGeometric}) to approximate the denominator with a monomial:
\begin{align}\label{app:eq:upsilon_total_closed_333}
    R^{(k)}_{u}(\bm{v})=\sum_{b\in\Omega}\sum_{u\in \mathcal{U}_{b}}|\Upsilon_{u}^{(k)}| &\geq \widehat{R}^{(k)}_{u}(\bm{v};\ell) \triangleq \prod_{b\in\Omega}\prod_{u\in \mathcal{U}_{b}}\left(\frac{|\Upsilon_{u}^{(k)}| R^{(k)}_{u}([\bm{v}]^{(\ell-1)})}{\left[|\Upsilon_{u}^{(k)}|\right]^{(\ell-1)}}\right)^{\frac{\left[|\Upsilon_{u}^{(k)}|\right]^{(\ell-1)}}{R^{(k)}_{u}([\bm{v}]^{(\ell-1)})}},
\end{align}
which gives us an approximation of~\eqref{app:eq:upsilon_total_closed_2} as follows:
\begin{equation}\label{app:eq:upsilon_total_closed_4}
    \frac{|\Upsilon^{(k)}|}{\widehat{R}^{(k)}_{u}(\bm{v};\ell)}\le 1.
\end{equation}
We finally approximate constraint~\eqref{app:eq:upsilon_total_closed} as follows:
\begin{tcolorbox}[ams align]
     & \frac{|\Upsilon^{(k)}|}{\widehat{R}^{(k)}_{u}(\bm{v};\ell)}\le 1,~~~~~~\frac{\sum_{b\in \Omega,u\in\mathcal{U}_{b}}|\Upsilon_{u}^{(k)}|}{\mathscr{O}^{(k)}_{u}|\Upsilon^{(k)}|}\le 1,~~~~~~\frac{1}{\mathscr{O}^{(k)}_{u}}\le 1.\nonumber
\end{tcolorbox}

\textbf{Constraint~\eqref{app:eq:upsilon_recruited_closed}:} Performing some algebraic manipulations on \eqref{app:eq:upsilon_recruited_closed} gives us
\begin{equation}\label{app:eq:upsilon_recruited_closed_1}
    \frac{|\Upsilon^{\mathsf{s},(k)}|}{\sum_{b\in \Omega,u\in\mathcal{U}_{b}}\widehat{\lambda}_{u}^{(k)} |\Upsilon_{u}^{(k)}|}{=} 1.
\end{equation}
This constraint is not in the format of GP; therefore, we transform it by splitting it into the following three inequalities:
\begin{equation}\label{app:eq:upsilon_recruited_closed_2}
    \frac{|\Upsilon^{\mathsf{s},(k)}|}{\sum_{b\in \Omega,u\in\mathcal{U}_{b}}\widehat{\lambda}_{u}^{(k)} |\Upsilon_{u}^{(k)}|}\le 1,
\end{equation}
\begin{equation}\label{app:eq:upsilon_recruited_closed_3}
    \frac{\sum_{b\in \Omega,u\in\mathcal{U}_{b}}\widehat{\lambda}_{u}^{(k)} |\Upsilon_{u}^{(k)}|}{\mathscr{O}^{\mathsf{s},(k)}_{u}|\Upsilon^{\mathsf{s},(k)}|}\le 1,
\end{equation}
\begin{equation}
    \mathscr{O}^{\mathsf{s},(k)}_{u}\ge 1,
\end{equation}
where $\mathscr{O}^{\mathsf{s},(k)}_{u}$ is added with a large penalty term to the objective function to force $\mathscr{O}^{\mathsf{s},(k)}_{u}{\rightarrow}1^+$ at the optimal point. The fraction in~\eqref{app:eq:upsilon_recruited_closed_2} still needs transformation since it is an inequality with a posynomial in the denominator, which is not a posynomial. We thus exploit arithmetic-geometric mean inequality (Lemma~\ref{Lemma:ArethmaticGeometric}) to approximate the denominator with a monomial:
\begin{align}\label{app:eq:upsilon_recruited_closed_333}
    R^{\mathsf{s},(k)}_{u}(\bm{v})=\sum_{b\in\Omega}\sum_{u\in \mathcal{U}_{b}}|\Upsilon_{u}^{(k)}| &\geq \widehat{R}^{\mathsf{s},(k)}_{u}(\bm{v};\ell) \triangleq \prod_{b\in\Omega}\prod_{u\in \mathcal{U}_{b}}\left(\frac{|\Upsilon_{u}^{(k)}| R^{\mathsf{s},(k)}_{u}([\bm{v}]^{(\ell-1)})}{\left[|\Upsilon_{u}^{(k)}|\right]^{(\ell-1)}}\right)^{\frac{\left[|\Upsilon_{u}^{(k)}|\right]^{(\ell-1)}}{R^{\mathsf{s},(k)}_{u}([\bm{v}]^{(\ell-1)})}},
\end{align}
which gives us an approximation of~\eqref{app:eq:upsilon_recruited_closed_2} as follows:
\begin{equation}\label{app:eq:upsilon_recruited_closed_4}
    \frac{|\Upsilon^{\mathsf{s},(k)}|}{\widehat{R}^{\mathsf{s},(k)}_{u}(\bm{v};\ell)}\le 1.
\end{equation}
We finally approximate constraint~\eqref{app:eq:upsilon_recruited_closed} as follows:
\begin{tcolorbox}[ams align]
     & \frac{|\Upsilon^{\mathsf{s},(k)}|}{\widehat{R}^{\mathsf{s},(k)}_{u}(\bm{v};\ell)}\le 1,~~~~~~\frac{\sum_{b\in \Omega,u\in\mathcal{U}_{b}}\widehat{\lambda}_{u}^{(k)} |\Upsilon_{u}^{(k)}|}{\mathscr{O}^{\mathsf{s},(k)}_{u}|\Upsilon^{\mathsf{s},(k)}|}\le 1,~~~~~~\frac{1}{\mathscr{O}^{\mathsf{s},(k)}_{u}}\le 1.\nonumber
\end{tcolorbox}

\textbf{Constraint~\eqref{app:eq:upsilon_overline}:} Performing some algebraic manipulations on \eqref{app:eq:upsilon_overline} gives us
\begin{equation}\label{app:eq:upsilon_overline_1}
    \frac{|\overline{\Upsilon}^{(k)}|+|\Upsilon^{\mathsf{s},(k)}|}{|\Upsilon^{(k)}|}=1.
\end{equation}
This constraint is not in the format of GP; therefore, we transform it by splitting it into the following three inequalities:
\begin{equation}\label{app:eq:upsilon_overline_2}
    \frac{|\overline{\Upsilon}^{(k)}|+|\Upsilon^{\mathsf{s},(k)}|}{|\Upsilon^{(k)}|}\le 1,
\end{equation}
\begin{equation}\label{app:eq:upsilon_overline_3}
    \frac{\left(\mathscr{B}^{(k)}\right)^{-1}|\Upsilon^{(k)}|}{|\overline{\Upsilon}^{(k)}|+|\Upsilon^{\mathsf{s},(k)}|}\le 1,
\end{equation}
\begin{equation}
    \mathscr{B}^{(k)}\ge 1,
\end{equation}
where $\mathscr{B}^{(k)}$ is added with a large penalty term to the objective function to force $\mathscr{B}^{(k)}{\rightarrow}1^+$ at the optimal point. The fraction in~\eqref{app:eq:upsilon_overline_3} still needs transformation since it is an inequality with a posynomial in the denominator, which is not a posynomial. We thus exploit arithmetic-geometric mean inequality (Lemma~\ref{Lemma:ArethmaticGeometric}) to approximate the denominator with a monomial:
\begin{align}\label{app:eq:upsilon_overline_333}
    F^{(k)}(\bm{v})=|\overline{\Upsilon}^{(k)}|+|\Upsilon^{\mathsf{s},(k)}| &\geq \widehat{F}^{(k)}(\bm{v};\ell) \triangleq \left(\frac{|\overline{\Upsilon}^{(k)}| F^{(k)}([\bm{v}]^{(\ell-1)})}{\left[|\overline{\Upsilon}^{(k)}|\right]^{(\ell-1)}}\right)^{\frac{\left[|\overline{\Upsilon}^{(k)}|\right]^{(\ell-1)}}{F^{(k)}([\bm{v}]^{(\ell-1)})}}\times \left(\frac{|\Upsilon^{\mathsf{s},(k)}| F^{(k)}([\bm{v}]^{(\ell-1)})}{\left[|\Upsilon^{\mathsf{s},(k)}|\right]^{(\ell-1)}}\right)^{\frac{\left[|\Upsilon^{\mathsf{s},(k)}|\right]^{(\ell-1)}}{F^{(k)}([\bm{v}]^{(\ell-1)})}},
\end{align}
which gives us an approximation of~\eqref{app:eq:upsilon_overline_3} as follows:
\begin{equation}\label{app:eq:upsilon_overline_4}
    \frac{|\Upsilon^{(k)}|}{\widehat{F}^{(k)}(\bm{v};\ell)}\le 1.
\end{equation}
We finally approximate constraint~\eqref{app:eq:upsilon_overline} as follows:
\begin{tcolorbox}[ams align]
     & \frac{|\overline{\Upsilon}^{(k)}|+|\Upsilon^{\mathsf{s},(k)}|}{|\Upsilon^{(k)}|}\le 1,~~~~~\frac{|\Upsilon^{(k)}|}{\widehat{F}^{(k)}(\bm{v};\ell)}\le 1,~~~~~\frac{1}{\mathscr{B}^{(k)}}\le 1.\nonumber
\end{tcolorbox}

\textbf{Constraint~\eqref{app:eq:e_max_overline}:} Performing some algebraic manipulations on \eqref{app:eq:e_max_overline} gives us
\begin{equation}\label{app:eq:e_max_overline_1}
   \overline{\ell}_{\mathsf{max}}^{(k)}+4\eta_k^2\beta^2 \ell_{\mathsf{max}}^{(k)}\widehat{\ell}_{\mathsf{max}}^{(k)}= 1.
\end{equation}
This constraint is not in the format of GP; therefore, we transform it by splitting it into the following three inequalities:
\begin{equation}\label{app:eq:e_max_overline_2}
    \overline{\ell}_{\mathsf{max}}^{(k)}+4\eta_k^2\beta^2 \ell_{\mathsf{max}}^{(k)}\widehat{\ell}_{\mathsf{max}}^{(k)}\le 1,
\end{equation}
\begin{equation}\label{app:eq:e_max_overline_3}
    \frac{\left(\mathscr{B}^{(k)}_{u,3}\right)^{-1}}{\overline{\ell}_{\mathsf{max}}^{(k)}+4\eta_k^2\beta^2 \ell_{\mathsf{max}}^{(k)}\widehat{\ell}_{\mathsf{max}}^{(k)}}\le 1,
\end{equation}
\begin{equation}
    \mathscr{B}^{(k)}_{u,3}\ge 1,
\end{equation}
where $\mathscr{B}^{(k)}_{u,3}$ is added with a large penalty term to the objective function to force $\mathscr{B}^{(k)}_{u,3}{\rightarrow}1^+$ at the optimal point. The fraction in~\eqref{app:eq:e_max_overline_3} still needs transformation since it is an inequality with a posynomial in the denominator, which is not a posynomial. We thus exploit arithmetic-geometric mean inequality (Lemma~\ref{Lemma:ArethmaticGeometric}) to approximate the denominator with a monomial:
\begin{align}\label{app:eq:e_max_overline_333}
    F^{(k)}_{u,3}(\bm{v})=\overline{\ell}_{\mathsf{max}}^{(k)}+4\eta_k^2\beta^2 \ell_{\mathsf{max}}^{(k)}\widehat{\ell}_{\mathsf{max}}^{(k)} &\geq \widehat{F}^{(k)}_{u,3}(\bm{v};\ell) \triangleq \left(\frac{\overline{\ell}_{\mathsf{max}}^{(k)} F^{(k)}_{u,3}([\bm{v}]^{(\ell-1)})}{\left[\overline{\ell}_{\mathsf{max}}^{(k)}\right]^{(\ell-1)}}\right)^{\frac{\left[\overline{\ell}_{\mathsf{max}}^{(k)}\right]^{(\ell-1)}}{F^{(k)}_{u,3}([\bm{v}]^{(\ell-1)})}}\times  \left(\frac{\eta_k^2\ell_{\mathsf{max}}^{(k)}\widehat{\ell}_{\mathsf{max}}^{(k)} F^{(k)}_{u,3}([\bm{v}]^{(\ell-1)})}{\left[\eta_k^2\ell_{\mathsf{max}}^{(k)}\widehat{\ell}_{\mathsf{max}}^{(k)}\right]^{(\ell-1)}}\right)^{\frac{\left[4\eta_k^2\beta^2 \ell_{\mathsf{max}}^{(k)}\widehat{\ell}_{\mathsf{max}}^{(k)}\right]^{(\ell-1)}}{F^{(k)}_{u,3}([\bm{v}]^{(\ell-1)})}},
\end{align}
which gives us an approximation of~\eqref{app:eq:e_max_overline_3} as follows:
\begin{equation}\label{app:eq:e_max_overline_4}
    \frac{\left(\mathscr{B}^{(k)}_{u,3}\right)^{-1}}{\widehat{F}^{(k)}_{u,3}(\bm{v};\ell)}\le 1.
\end{equation}
We finally approximate constraint~\eqref{app:eq:e_max_overline} as follows:
\begin{tcolorbox}[ams align]
     &\overline{\ell}_{\mathsf{max}}^{(k)}+4\eta_k^2\beta^2 \ell_{\mathsf{max}}^{(k)}\widehat{\ell}_{\mathsf{max}}^{(k)}\le 1,~~~~~\frac{\left(\mathscr{B}^{(k)}_{u,3}\right)^{-1}}{\widehat{F}^{(k)}_{u,3}(\bm{v};\ell)}\le 1,~~~~~~\frac{1}{\mathscr{B}^{(k)}_{u,3}}\le 1.\nonumber
\end{tcolorbox}

\textbf{Constraint~\eqref{app:eq:tau_LC_hat}:} Performing some algebraic manipulations on \eqref{app:eq:tau_LC_hat} gives us
\begin{equation}\label{app:eq:tau_LC_hat_1}
    \frac{\widehat{\tau}_{u}^{\mathsf{LC},{(k)}}+T^{(k{-}1)}+\widehat{\lambda}_{u}^{(k)}\tau_{u}^{\mathsf{LC},{(k)}}}{T^{(k)}} = 1.
\end{equation}
This constraint is not in the format of GP; therefore, we transform it by splitting it into the following three inequalities:
\begin{equation}\label{app:eq:tau_LC_hat_2}
    \frac{\widehat{\tau}_{u}^{\mathsf{LC},{(k)}}+T^{(k{-}1)}+\widehat{\lambda}_{u}^{(k)}\tau_{u}^{\mathsf{LC},{(k)}}}{T^{(k)}}\le 1,
\end{equation}
\begin{equation}\label{app:eq:tau_LC_hat_3}
    \frac{\left(\mathscr{B}^{(k)}_{u,4}\right)^{-1}T^{(k)}}{\widehat{\tau}_{u}^{\mathsf{LC},{(k)}}+T^{(k{-}1)}+\widehat{\lambda}_{u}^{(k)}\tau_{u}^{\mathsf{LC},{(k)}}}\le 1,
\end{equation}
\begin{equation}
    \mathscr{B}^{(k)}_{u,4}\ge 1,
\end{equation}
where $\mathscr{B}^{(k)}_{u,4}$ is added with a large penalty term to the objective function to force $\mathscr{B}^{(k)}_{u,4}{\rightarrow}1^+$ at the optimal point. The fraction in~\eqref{app:eq:tau_LC_hat_3} still needs transformation since it is an inequality with a posynomial in the denominator, which is not a posynomial. We thus exploit arithmetic-geometric mean inequality (Lemma~\ref{Lemma:ArethmaticGeometric}) to approximate the denominator with a monomial:
\begin{align}\label{app:eq:tau_LC_hat_333}
    F^{(k)}_{u,4}(\bm{v})=\widehat{\tau}_{u}^{\mathsf{LC},{(k)}}&+T^{(k{-}1)}+\widehat{\lambda}_{u}^{(k)}\tau_{u}^{\mathsf{LC},{(k)}} \geq \widehat{F}^{(k)}_{u,4}(\bm{v};\ell) \triangleq \left(\frac{\widehat{\tau}_{u}^{\mathsf{LC},{(k)}} F^{(k)}_{u,4}([\bm{v}]^{(\ell-1)})}{\left[\widehat{\tau}_{u}^{\mathsf{LC},{(k)}}\right]^{(\ell-1)}}\right)^{\frac{\left[\widehat{\tau}_{u}^{\mathsf{LC},{(k)}}\right]^{(\ell-1)}}{F^{(k)}_{u,4}([\bm{v}]^{(\ell-1)})}}\nonumber\\
    &\times \left(\frac{T^{(k{-}1)} F^{(k)}_{u,4}([\bm{v}]^{(\ell-1)})}{\left[T^{(k{-}1)}\right]^{(\ell-1)}}\right)^{\frac{\left[T^{(k{-}1)}\right]^{(\ell-1)}}{F^{(k)}_{u,4}([\bm{v}]^{(\ell-1)})}}\times \left(\frac{\widehat{\lambda}_{u}^{(k)}\tau_{u}^{\mathsf{LC},{(k)}} F^{(k)}_{u,4}([\bm{v}]^{(\ell-1)})}{\left[\widehat{\lambda}_{u}^{(k)}\tau_{u}^{\mathsf{LC},{(k)}}\right]^{(\ell-1)}}\right)^{\frac{\left[\widehat{\lambda}_{u}^{(k)}\tau_{u}^{\mathsf{LC},{(k)}}\right]^{(\ell-1)}}{F^{(k)}_{u,4}([\bm{v}]^{(\ell-1)})}},
\end{align}
which gives us an approximation of~\eqref{app:eq:tau_LC_hat_3} as follows:
\begin{equation}\label{app:eq:tau_LC_hat_4}
    \frac{\left(\mathscr{B}^{(k)}_{u,4}\right)^{-1}T^{(k)}}{\widehat{F}^{(k)}_{u,4}(\bm{v};\ell)}\le 1.
\end{equation}
We finally approximate constraint~\eqref{app:eq:tau_LC_hat} as follows:
\begin{tcolorbox}[ams align]
     &\frac{\widehat{\tau}_{u}^{\mathsf{LC},{(k)}}+T^{(k{-}1)}+\widehat{\lambda}_{u}^{(k)}\tau_{u}^{\mathsf{LC},{(k)}}}{T^{(k)}}\le 1,~~~~~\frac{\left(\mathscr{B}^{(k)}_{u,4}\right)^{-1}T^{(k)}}{\widehat{\tau}_{u}^{\mathsf{LC},{(k)}}+T^{(k{-}1)}+\widehat{\lambda}_{u}^{(k)}\tau_{u}^{\mathsf{LC},{(k)}}}\le 1,~~~~~\frac{1}{\mathscr{B}^{(k)}_{u,4}}\le 1.\nonumber
\end{tcolorbox}

\textbf{Constraint~\eqref{app:eq:varsigma_overline}:} Performing some algebraic manipulations on \eqref{app:eq:varsigma_overline} gives us
\begin{equation}\label{app:eq:varsigma_overline_1}
    \overline{\varsigma}^{(k)}_u\varsigma^{(k)}_u+\varsigma^{(k)}_u=1.
\end{equation}
This constraint is not in the format of GP; therefore, we transform it by splitting it into the following three inequalities:
\begin{equation}\label{app:eq:varsigma_overline_2}
    \overline{\varsigma}^{(k)}_u\varsigma^{(k)}_u+\varsigma^{(k)}_u\le 1,
\end{equation}
\begin{equation}\label{app:eq:varsigma_overline_3}
    \frac{\left(\mathscr{B}^{(k)}_{u,5}\right)^{-1}}{\overline{\varsigma}^{(k)}_u\varsigma^{(k)}_u+\varsigma^{(k)}_u}\le 1,
\end{equation}
\begin{equation}
    \mathscr{B}^{(k)}_{u,5}\ge 1,
\end{equation}
where $\mathscr{B}^{(k)}_{u,5}$ is added with a large penalty term to the objective function to force $\mathscr{B}^{(k)}_{u,5}{\rightarrow}1^+$ at the optimal point. The fraction in~\eqref{app:eq:varsigma_overline_3} still needs transformation since it is an inequality with a posynomial in the denominator, which is not a posynomial. We thus exploit arithmetic-geometric mean inequality (Lemma~\ref{Lemma:ArethmaticGeometric}) to approximate the denominator with a monomial:
\begin{align}\label{app:eq:varsigma_overline_333}
    F^{(k)}_{u,5}(\bm{v})=\overline{\varsigma}^{(k)}_u\varsigma^{(k)}_u+\varsigma^{(k)}_u &\geq \widehat{F}^{(k)}_{u,5}(\bm{v};\ell) \triangleq \left(\frac{\overline{\varsigma}^{(k)}_u\varsigma^{(k)}_u F^{(k)}_{u,5}([\bm{v}]^{(\ell-1)})}{\left[\overline{\varsigma}^{(k)}_u\varsigma^{(k)}_u\right]^{(\ell-1)}}\right)^{\frac{\left[\overline{\varsigma}^{(k)}_u\varsigma^{(k)}_u\right]^{(\ell-1)}}{F^{(k)}_{u,5}([\bm{v}]^{(\ell-1)})}}\times \left(\frac{\varsigma^{(k)}_u F^{(k)}_{u,5}([\bm{v}]^{(\ell-1)})}{\left[\varsigma^{(k)}_u\right]^{(\ell-1)}}\right)^{\frac{\left[\varsigma^{(k)}_u\right]^{(\ell-1)}}{F^{(k)}_{u,5}([\bm{v}]^{(\ell-1)})}},
\end{align}
which gives us an approximation of~\eqref{app:eq:varsigma_overline_3} as follows:
\begin{equation}\label{app:eq:varsigma_overline_4}
    \frac{\left(\mathscr{B}^{(k)}_{u,5}\right)^{-1}}{\widehat{F}^{(k)}_{u,5}(\bm{v};\ell)}\le 1.
\end{equation}
We finally approximate constraint~\eqref{app:eq:varsigma_overline} as follows:
\begin{tcolorbox}[ams align]
     &\overline{\varsigma}^{(k)}_u\varsigma^{(k)}_u+\varsigma^{(k)}_u\le 1,~~~~~~\frac{\left(\mathscr{B}^{(k)}_{u,5}\right)^{-1}}{\widehat{F}^{(k)}_{u,5}(\bm{v};\ell)}\le 1,~~~~~~\frac{1}{\mathscr{B}^{(k)}_{u,5}}\le 1.\nonumber
\end{tcolorbox}

\textbf{Constraint~\eqref{app:eq:upsilon_min}:} Performing some algebraic manipulations on \eqref{app:eq:upsilon_min} gives us
\begin{equation}\label{app:eq:upsilon_min_1}
    \frac{|\Upsilon_{\mathsf{min}}^{(k)}|}{\displaystyle\min_{b{\in}\Omega,u{\in} \mathcal{U}_{b}}\{|\Upsilon_{u}^{(k)}|\}}=1.
\end{equation}
Using the approximation $\min\{A, B\}\approx (A^{-p}+B^{-p})^{-\frac{1}{p}}$, which is tight when $p \gg 1$, gives us
\begin{equation}\label{app:eq:upsilon_min_2}
    \frac{|\Upsilon_{\mathsf{min}}^{(k)}|}{\left(\displaystyle\sum_{b{\in}\Omega}\sum_{u{\in} \mathcal{U}_{b}}(|\Upsilon_{u}^{(k)}|)^{-p}\right)^{-\frac{1}{p}}}=1.
\end{equation}
The above equation can be rewritten as the following standard monomial equality constraint:
\begin{equation}\label{app:eq:upsilon_min_3}
    \left(|\widetilde{\Upsilon}_{\mathsf{min}}^{(k)}|\right)^{\frac{1}{p}}|\Upsilon_{\mathsf{min}}^{(k)}| = 1,
\end{equation}
where $|\widetilde{\Upsilon}_{\mathsf{min}}^{(k)}|$ is an auxiliary decision variable satisfying
\begin{equation}\label{app:eq:upsilon_min_4}
    |\widetilde{\Upsilon}_{\mathsf{min}}^{(k)}| = \sum_{b{\in}\Omega}\sum_{u{\in} \mathcal{U}_{b}}(|\Upsilon_{u}^{(k)}|)^{-p}.
\end{equation}

However, \eqref{app:eq:upsilon_min_4} is not in the format of GP and needs to be transformed to satisfy the GP requirements. In doing so, we rewrite \eqref{app:eq:upsilon_min_4} as follows:
\begin{equation}\label{app:eq:upsilon_min_4_6}
        \frac{|\widetilde{\Upsilon}_{\mathsf{min}}^{(k)}|}{\sum_{b{\in}\Omega}\sum_{u{\in} \mathcal{U}_{b}}(|\Upsilon_{u}^{(k)}|)^{-p}} = 1.
\end{equation}
This constraint is still not in the format of GP. Therefore, we transform it via splitting it into the following three inequalities:
\begin{equation}\label{app:eq:upsilon_min_4_7}
    \frac{|\widetilde{\Upsilon}_{\mathsf{min}}^{(k)}|}{\sum_{b{\in}\Omega}\sum_{u{\in} \mathcal{U}_{b}}(|\Upsilon_{u}^{(k)}|)^{-p}}\le 1,
\end{equation}
\begin{equation}\label{app:eq:upsilon_min_4_8}
    \frac{\sum_{b{\in}\Omega}\sum_{u{\in} \mathcal{U}_{b}}(|\Upsilon_{u}^{(k)}|)^{-p}}{\mathscr{B}^{(k)}_{\mathsf{min}}|\widetilde{\Upsilon}_{\mathsf{min}}^{(k)}|}\le 1,
\end{equation}
\begin{equation}
    \mathscr{B}^{(k)}_{\mathsf{min}}\ge 1,
\end{equation}
where $\mathscr{B}^{(k)}_{\mathsf{min}}$ is added with a large penalty term to the objective function to force $\mathscr{B}^{(k)}_{\mathsf{min}}{\rightarrow}1^+$ at the optimal point. The fraction in~\eqref{app:eq:upsilon_min_4_7} still needs transformation since it is an inequality with a posynomial in the denominator, which is not a posynomial. We thus exploit arithmetic-geometric mean inequality (Lemma~\ref{Lemma:ArethmaticGeometric}) to approximate the denominator with a monomial:
\begin{align}\label{app:eq:upsilon_min_4_8}
    F^{(k)}_{\mathsf{min}}(\bm{v})=\sum_{b{\in}\Omega}\sum_{u{\in} \mathcal{U}_{b}}(|\Upsilon_{u}^{(k)}|)^{-p} &\geq \widehat{F}^{(k)}_{\mathsf{min}}(\bm{v};\ell) \triangleq \prod_{b{\in}\Omega}\prod_{u{\in} \mathcal{U}_{b}}\left(\frac{(|\Upsilon_{u}^{(k)}|)^{-p} F^{(k)}_{\mathsf{min}}([\bm{v}]^{(\ell-1)})}{\left[(|\Upsilon_{u}^{(k)}|)^{-p}\right]^{(\ell-1)}}\right)^{\frac{\left[(|\Upsilon_{u}^{(k)}|)^{-p}\right]^{(\ell-1)}}{F^{(k)}_{\mathsf{min}}([\bm{v}]^{(\ell-1)})}},
\end{align}
which gives us an approximation of~\eqref{app:eq:upsilon_min_4_7} as follows:
\begin{equation}\label{app:eq:upsilon_min_4_9}
    \frac{|\widetilde{\Upsilon}_{\mathsf{min}}^{(k)}|}{\widehat{F}^{(k)}_{\mathsf{min}}(\bm{v};\ell)}\le 1.
\end{equation}
We finally approximate constraint~\eqref{app:eq:upsilon_min} as follows:
\begin{tcolorbox}[ams align]
     &\frac{|\widetilde{\Upsilon}_{\mathsf{min}}^{(k)}|}{\widehat{F}^{(k)}_{\mathsf{min}}(\bm{v};\ell)}\le 1,~~~~~\frac{\sum_{b{\in}\Omega}\sum_{u{\in} \mathcal{U}_{b}}(|\Upsilon_{u}^{(k)}|)^{-p}}{\mathscr{B}^{(k)}_{\mathsf{min}}|\widetilde{\Upsilon}_{\mathsf{min}}^{(k)}|}\le 1,~~~~~\frac{1}{\mathscr{B}^{(k)}_{\mathsf{min}}}\le 1.\nonumber
\end{tcolorbox}

\textbf{Constraint~\eqref{app:eq:ell_max}:} Performing some algebraic manipulations on \eqref{app:eq:ell_max} gives us
\begin{equation}\label{app:eq:ell_max_1}
    \frac{\ell^{(k)}_{\mathsf{max}}}{\max_{b\in\Omega,u{\in} \mathcal{U}_{b}}\{\ell^{(k)}_u\}}=1.
\end{equation}
Using the approximation $\max\{A, B\}\approx (A^{p}+B^{p})^{-\frac{1}{p}}$, which is tight when $p \gg 1$, gives us
\begin{equation}\label{app:eq:ell_max_2}
    \frac{\ell^{(k)}_{\mathsf{max}}}{\left(\displaystyle\sum_{b{\in}\Omega}\sum_{u{\in} \mathcal{U}_{b}}(\ell^{(k)}_u)^{p}\right)^{-\frac{1}{p}}}=1.
\end{equation}
The above equation can be rewritten as the following standard monomial equality constraint:
\begin{equation}\label{app:eq:ell_max_3}
    \left(\widetilde{\ell}^{(k)}_{\mathsf{max}}\right)^{\frac{1}{p}}\ell^{(k)}_{\mathsf{max}} = 1,
\end{equation}
where $\widetilde{\ell}^{(k)}_{\mathsf{max}}$ is an auxiliary decision variable satisfying
\begin{equation}\label{app:eq:ell_max_4}
    \widetilde{\ell}^{(k)}_{\mathsf{max}} = \sum_{b{\in}\Omega}\sum_{u{\in} \mathcal{U}_{b}}(\ell^{(k)}_u)^{p}.
\end{equation}

However, \eqref{app:eq:ell_max_4} is not in the format of GP and needs to be transformed to satisfy the GP requirements. In doing so, we rewrite \eqref{app:eq:ell_max_4} as follows:
\begin{equation}\label{app:eq:ell_max_4_6}
        \frac{\widetilde{\ell}^{(k)}_{\mathsf{max}}}{\sum_{b{\in}\Omega}\sum_{u{\in} \mathcal{U}_{b}}(\ell^{(k)}_u)^{p}} = 1.
\end{equation}
This constraint is still not in the format of GP. Therefore, we transform it via splitting it into the following three inequalities:
\begin{equation}\label{app:eq:ell_max_4_7}
    \frac{\widetilde{\ell}^{(k)}_{\mathsf{max}}}{\sum_{b{\in}\Omega}\sum_{u{\in} \mathcal{U}_{b}}(\ell^{(k)}_u)^{p}}\le 1,
\end{equation}
\begin{equation}\label{app:eq:ell_max_4_8}
    \frac{\sum_{b{\in}\Omega}\sum_{u{\in} \mathcal{U}_{b}}(\ell^{(k)}_u)^{p}}{\mathscr{C}^{(k)}_{\mathsf{max}}\widetilde{\ell}^{(k)}_{\mathsf{max}}}\le 1,
\end{equation}
\begin{equation}
    \mathscr{C}^{(k)}_{\mathsf{max}}\ge 1,
\end{equation}
where $\mathscr{C}^{(k)}_{\mathsf{max}}$ is added with a large penalty term to the objective function to force $\mathscr{C}^{(k)}_{\mathsf{max}}{\rightarrow}1^+$ at the optimal point. The fraction in~\eqref{app:eq:ell_max_4_7} still needs transformation since it is an inequality with a posynomial in the denominator, which is not a posynomial. We thus exploit arithmetic-geometric mean inequality (Lemma~\ref{Lemma:ArethmaticGeometric}) to approximate the denominator with a monomial:
\begin{align}\label{app:eq:ell_max_4_8}
    G^{(k)}_{\mathsf{max}}(\bm{v})=\sum_{b{\in}\Omega}\sum_{u{\in} \mathcal{U}_{b}}(\ell^{(k)}_u)^{p} &\geq \widehat{G}^{(k)}_{\mathsf{max}}(\bm{v};\ell) \triangleq \prod_{b{\in}\Omega}\prod_{u{\in} \mathcal{U}_{b}}\left(\frac{(\ell^{(k)}_u)^{p} G^{(k)}_{\mathsf{max}}([\bm{v}]^{(\ell-1)})}{\left[(\ell^{(k)}_u)^{p}\right]^{(\ell-1)}}\right)^{\frac{\left[(\ell^{(k)}_u)^{p}\right]^{(\ell-1)}}{G^{(k)}_{\mathsf{max}}([\bm{v}]^{(\ell-1)})}},
\end{align}
which gives us an approximation of~\eqref{app:eq:ell_max_4_7} as follows:
\begin{equation}\label{app:eq:ell_max_4_9}
    \frac{\widetilde{\ell}^{(k)}_{\mathsf{max}}}{\widehat{G}^{(k)}_{\mathsf{max}}(\bm{v};\ell)}\le 1.
\end{equation}
We finally approximate constraint~\eqref{app:eq:ell_max} as follows:
\begin{tcolorbox}[ams align]
     &\frac{\widetilde{\ell}^{(k)}_{\mathsf{max}}}{\widehat{G}^{(k)}_{\mathsf{max}}(\bm{v};\ell)}\le 1,~~~~~\frac{\sum_{b{\in}\Omega}\sum_{u{\in} \mathcal{U}_{b}}(\ell^{(k)}_u)^{p}}{\mathscr{C}^{(k)}_{\mathsf{max}}\widetilde{\ell}^{(k)}_{\mathsf{max}}}\le 1,~~~~~\frac{1}{\mathscr{C}^{(k)}_{\mathsf{max}}}\le 1.\nonumber
\end{tcolorbox}

\textbf{Constraint~\eqref{app:suf:main:gamma_upper}:} Performing some algebraic manipulations on \eqref{app:suf:main:gamma_upper} gives us
\begin{align}\label{app:suf:main:gamma_upper2}
    \frac{4T^{(k)}\sum_{b \in \Omega}\sum_{u\in \mathcal{U}_{b}} \mathfrak{D}^{(k)}_u}{4T^{(k{-}1)}\sum_{b \in \Omega}\sum_{u\in \mathcal{U}_{b}} \mathfrak{D}^{(k)}_u+4\sum_{b \in \Omega}\sum_{u\in \mathcal{U}_{b}} \mathfrak{D}^{(k)}_u\widehat{\lambda}_{u}^{(k)}\tau_{u}^{\mathsf{LC},{(k)}}+\eta_{_k}\mathfrak{B}_k(1-\zeta^{(k)})\vartheta^k} &\le 1.
\end{align}
The fraction in~\eqref{app:suf:main:gamma_upper2} is not in the format of GP since it is an inequality with a posynomial in the denominator, which is not a posynomial. We thus exploit arithmetic-geometric mean inequality (Lemma~\ref{Lemma:ArethmaticGeometric}) to approximate the denominator with a monomial:
\begin{align}\label{app:suf:main:gamma_upper3}
    F^{(k)}_1(\bm{v})&=\sum_{b \in \Omega}\sum_{u\in \mathcal{U}_{b}} 4T^{(k{-}1)}\mathfrak{D}^{(k)}_u+\sum_{b \in \Omega}\sum_{u\in \mathcal{U}_{b}} 4\mathfrak{D}^{(k)}_u\widehat{\lambda}_{u}^{(k)}\tau_{u}^{\mathsf{LC},{(k)}}+\eta_{_k}\mathfrak{B}_k(1-\zeta^{(k)})\vartheta^k \nonumber\\
    &\geq \widehat{F}^{(k)}_1(\bm{v};\ell) \triangleq  \prod_{b \in \Omega}\prod_{u\in \mathcal{U}_{b}} \left(\frac{T^{(k{-}1)} F^{(k)}_1([\bm{v}]^{(\ell-1)})}{\left[T^{(k{-}1)}\right]^{(\ell-1)}}\right)^{\frac{\left[4T^{(k{-}1)}\mathfrak{D}^{(k)}_u\right]^{(\ell-1)} }{F^{(k)}_1([\bm{v}]^{(\ell-1)})}} \nonumber\\
    &{\times}\prod_{b \in \Omega}\prod_{u\in \mathcal{U}_{b}} \left(\frac{\widehat{\lambda}_{u}^{(k)}\tau_{u}^{\mathsf{LC},{(k)}} F^{(k)}_1([\bm{v}]^{(\ell-1)})}{\left[\widehat{\lambda}_{u}^{(k)}\tau_{u}^{\mathsf{LC},{(k)}}\right]^{(\ell-1)}}\right)^{\frac{\left[4\mathfrak{D}^{(k)}_u\widehat{\lambda}_{u}^{(k)}\tau_{u}^{\mathsf{LC},{(k)}}\right]^{(\ell-1)} }{F^{(k)}_1([\bm{v}]^{(\ell-1)})}} {\times}\left(\frac{\mathfrak{B}_k F^{(k)}_1([\bm{v}]^{(\ell-1)})}{\left[\mathfrak{B}_k\right]^{(\ell-1)}}\right)^{\frac{\left[\eta_{_k}\mathfrak{B}_k(1-\zeta^{(k)})\vartheta^k\right]^{(\ell-1)}}{F^{(k)}_1([\bm{v}]^{(\ell-1)})}}.
\end{align}

We finally approximate \eqref{app:suf:main:gamma_upper} as follows:
\begin{tcolorbox}[ams align]
    \frac{4T^{(k)}\sum_{b \in \Omega}\sum_{u\in \mathcal{U}_{b}} \mathfrak{D}^{(k)}_u}{\widehat{F}^{(k)}_1(\bm{v};\ell)}\le 1 \nonumber
\end{tcolorbox}

\textbf{Constraint~\eqref{app:suf:main:recruitment_overline}:} performing some algebraic manipulations on \eqref{app:suf:main:recruitment_overline} results in the following two monomial inequalities, which are in the standrad form of GP:

\begin{tcolorbox}[ams align]
     &\frac{\left(|\overline{\Upsilon}^{(k)}|\right)^2(48\mathfrak{X}_1)}{|\Upsilon^{(k)}|\widehat{\zeta}^{(k)}|\Upsilon_{\mathsf{min}}^{(k)}|}\le 1.~~~~~~~\frac{\left(|\overline{\Upsilon}^{(k)}|\right)^2(48\mathfrak{X}_1)}{|\Upsilon^{(k)}|\varpi^k \left(1-\zeta^{(k)}\right)|\Upsilon_{\mathsf{min}}^{(k)}|\overline{\ell}_{\mathsf{max}}^{(k)}}\le 1\nonumber.
\end{tcolorbox}

\textbf{Constraint~\eqref{app:suf:main:sigma}:} Performing some algebraic manipulations on \eqref{app:suf:main:sigma} gives us
\begin{equation}\label{app:suf:main:sigma_2}
    \frac{\displaystyle\max_{k\in\mathcal{K}}\max_{b{\in}\Omega}\max_{u{\in}\mathcal{U}_{b}}\left\{\frac{\overline{\varsigma}^{(k)}_u\sigma^{\mathsf{max}}_{u}}{\varsigma^{(k)}_u|\Upsilon_{u}^{(k)}|}\right\}}{\sigma_{\mathsf{max}}}{\leq} 1.
\end{equation}
Using the approximation $\max\{A, B\}\approx (A^{p}+B^{p})^{-\frac{1}{p}}$, which is tight when $p \gg 1$, gives us
\begin{equation}\label{app:suf:main:sigma_3}
    \frac{(\sigma_{\mathsf{max}})^{-1}}{\left(\displaystyle\sum_{k\in\mathcal{K}}\sum_{b{\in}\Omega}\sum_{u{\in}\mathcal{U}_{b}}\left(\frac{\overline{\varsigma}^{(k)}_u\sigma^{\mathsf{max}}_{u}}{\varsigma^{(k)}_u|\Upsilon_{u}^{(k)}|}\right)^{p}\right)^{\frac{1}{p}}}{\leq} 1.
\end{equation}
The above inequality is not in the format of GP since it is not a posynomial. In order to transform \eqref{app:suf:main:sigma_3} into GP format, we write it via introducing an auxiliary decision variable $A^{(k)}_{u}$ satisfying
\begin{equation}\label{app:suf:main:sigma_4}
    A^{(k)}_{u}=\sum_{k\in\mathcal{K}}\sum_{b{\in}\Omega}\sum_{u{\in}\mathcal{U}_{b}}\left(\frac{\overline{\varsigma}^{(k)}_u\sigma^{\mathsf{max}}_{u}}{\varsigma^{(k)}_u|\Upsilon_{u}^{(k)}|}\right)^{p},
\end{equation}
which gives us the following posynomial:
\begin{equation}\label{app:suf:main:sigma_5}
    \left(A^{(k)}_{u}\right)^{-\frac{1}{p}}(\sigma_{\mathsf{max}})^{-1}{\leq} 1.
\end{equation}
However, \eqref{app:suf:main:sigma_4} is not in the format of GP and needs to be transformed to satisfy the GP requirements. In doing so, we rewrite \eqref{app:suf:main:sigma_4} as follows:
\begin{equation}\label{app:suf:main:sigma_6}
    \frac{A^{(k)}_{u}}{\displaystyle\sum_{k\in\mathcal{K}}\sum_{b{\in}\Omega}\sum_{u{\in}\mathcal{U}_{b}}\left(\frac{\overline{\varsigma}^{(k)}_u\sigma^{\mathsf{max}}_{u}}{\varsigma^{(k)}_u|\Upsilon_{u}^{(k)}|}\right)^{p}}=1.
\end{equation}
This constraint is still not in the format of GP. Therefore, we transform it via splitting it into the following three inequalities:
\begin{equation}\label{app:suf:main:sigma_7}
    \frac{A^{(k)}_{u}}{\displaystyle\sum_{k\in\mathcal{K}}\sum_{b{\in}\Omega}\sum_{u{\in}\mathcal{U}_{b}}\left(\frac{\overline{\varsigma}^{(k)}_u\sigma^{\mathsf{max}}_{u}}{\varsigma^{(k)}_u|\Upsilon_{u}^{(k)}|}\right)^{p}}\le 1,
\end{equation}
\begin{equation}\label{app:suf:main:sigma_8}
    \frac{\displaystyle\sum_{k\in\mathcal{K}}\sum_{b{\in}\Omega}\sum_{u{\in}\mathcal{U}_{b}}\left(\frac{\overline{\varsigma}^{(k)}_u\sigma^{\mathsf{max}}_{u}}{\varsigma^{(k)}_u|\Upsilon_{u}^{(k)}|}\right)^{p}}{\mathscr{B}^{(k)}_{u}A^{(k)}_{u}}\le 1,
\end{equation}
\begin{equation}
    \mathscr{B}^{(k)}_{u}\ge 1,
\end{equation}
where $\mathscr{B}^{(k)}_{u}$ is added with a large penalty term to the objective function to force $\mathscr{B}^{(k)}_{u}{\rightarrow}1^+$ at the optimal point. The fraction in~\eqref{app:suf:main:sigma_7} still needs transformation since it is an inequality with a posynomial in the denominator, which is not a posynomial. We thus exploit arithmetic-geometric mean inequality (Lemma~\ref{Lemma:ArethmaticGeometric}) to approximate the denominator with a monomial:
\begin{align}\label{app:suf:main:sigma_8}
    F^{(k)}_{u}(\bm{v})=\sum_{k\in\mathcal{K}}\sum_{b{\in}\Omega}\sum_{u{\in}\mathcal{U}_{b}}\left(\frac{\overline{\varsigma}^{(k)}_u\sigma^{\mathsf{max}}_{u}}{\varsigma^{(k)}_u|\Upsilon_{u}^{(k)}|}\right)^{p} &\geq \widehat{F}^{(k)}_{u}(\bm{v};\ell) \triangleq \prod_{k\in\mathcal{K}}\prod_{b{\in}\Omega}\prod_{u{\in}\mathcal{U}_{b}}\left(\frac{\left(\frac{\overline{\varsigma}^{(k)}_u}{\varsigma^{(k)}_u|\Upsilon_{u}^{(k)}|}\right)^{p} F^{(k)}_{u}([\bm{v}]^{(\ell-1)})}{\left[\left(\frac{\overline{\varsigma}^{(k)}_u}{\varsigma^{(k)}_u|\Upsilon_{u}^{(k)}|}\right)^{p}\right]^{(\ell-1)}}\right)^{\frac{\left[\left(\frac{\overline{\varsigma}^{(k)}_u\sigma^{\mathsf{max}}_{u}}{\varsigma^{(k)}_u|\Upsilon_{u}^{(k)}|}\right)^{p}\right]^{(\ell-1)}}{F^{(k)}_{u}([\bm{v}]^{(\ell-1)})}},
\end{align}
which gives us an approximation of~\eqref{app:suf:main:sigma_7} as follows:
\begin{equation}\label{app:suf:main:sigma_9}
    \frac{A^{(k)}_{u}}{\widehat{F}^{(k)}_{u}(\bm{v};\ell)}\le 1.
\end{equation}

We finally approximate constraint~\eqref{app:suf:main:sigma_4} as follows:
\begin{tcolorbox}[ams align]
     &\left(A^{(k)}_{u}\right)^{-\frac{1}{p}}(\sigma_{\mathsf{max}})^{-1}{\leq} 1,~~~~\frac{A^{(k)}_{u}}{\widehat{F}^{(k)}_{u}(\bm{v};\ell)}\le 1,~~~~\frac{\displaystyle\sum_{k\in\mathcal{K}}\sum_{b{\in}\Omega}\sum_{u{\in}\mathcal{U}_{b}}\left(\frac{\overline{\varsigma}^{(k)}_u\sigma^{\mathsf{max}}_{u}}{\varsigma^{(k)}_u|\Upsilon_{u}^{(k)}|}\right)^{p}}{\mathscr{B}^{(k)}_{u}A^{(k)}_{u}}\le 1,~~~~\frac{1}{\mathscr{B}^{(k)}_{u}}\le 1.\nonumber
\end{tcolorbox}

\newpage
\subsection{Transformation of Scheduling Decisions Constraints}\label{app:cons:scheduling_decisions}
Let us revisit the constraints of Sec.~\ref{sec:scheduling_decisions} in the following:

\textbf{$\bm{\mathsf{SDC}_1}$: GM broadcasting by O-RUs.}  
\begin{equation}\label{app:cons:mac_1}
   \bigg(\frac{\tau_{b}^{\downarrow,{(k)}}}{t_{x}}{-}1\bigg)\times \left(1{-}\beta^{\downarrow}_{b}(t_{x})\right) \leq  0.
\end{equation}
\begin{equation}\label{app:cons:mac_2}
    \bigg(1{-}\frac{\tau_{b}^{\downarrow,{(k)}}}{t_{x}}\bigg)\times\beta^{\downarrow}_{b}(t_{x})\le 0.
\end{equation}

\textbf{$\bm{\mathsf{SDC}_2}$: Local GPs uploading by DPUs.}
\begin{equation}\label{app:cons:mac_4}
   \bigg(\frac{\overline{\lambda}_{u}^{(k)}\Psi(\overline{\mathscr{D}}_{u}^{\nuparrow,(k)})}{t_{x}}{-}1\bigg)\bigg(1{-}\frac{\overline{\lambda}_{u}^{(k)}\Psi(\overline{\mathscr{D}}_{u}^{\uparrow,(k)})}{t_{x}}\bigg)\big(1{-}\overline{\beta}^{\uparrow}_u(t_{x})\big){\leq} 0,
\end{equation}
\begin{equation}\label{app:cons:mac_5}
 \max\bigg\{\frac{\overline{\lambda}_{u}^{(k)}\Psi(\overline{\mathscr{D}}_{u}^{\uparrow,(k)})}{t_{x}}{-}1,1{-}\frac{\overline{\lambda}_{u}^{(k)}\Psi(\overline{\mathscr{D}}_{u}^{\nuparrow,(k)})}{t_{x}}\bigg\}\times\overline{\beta}^{\uparrow}_u(t_{x}){\leq} 0.
\end{equation}
\par\textbf{$\bm{\mathsf{SDC}_3}$: Local GPs uploading by CHUs.}
\begin{equation}\label{app:cons:mac_6}
    \beta^{\uparrow}_u(t_{x})-\lambda_{u}^{(k)} \le 0.
\end{equation}
\begin{equation}\label{app:cons:mac_7}
   \bigg(\frac{\Psi(\mathscr{D}_{u}^{\nuparrow,(k)})}{t_{x}}{-}1\bigg)\bigg(1{-}\frac{\Psi(\mathscr{D}_{u}^{\uparrow,(k)})}{t_{x}}\bigg)\left(1{-}\beta^{\uparrow}_u(t_{x})\right){\leq}0,
\end{equation}
\begin{equation}\label{app:cons:mac_8}
\max\bigg\{\frac{\Psi(\mathscr{D}_{u}^{\uparrow,(k)})}{t_{x}}{-}1,1{-}\frac{\Psi(\mathscr{D}_{u}^{\nuparrow,(k)})}{t_{x}}\bigg\}\times\beta^{\uparrow}_u(t_{x}){\leq} 0.
\end{equation}

In the following, we aim to transform \eqref{app:cons:mac_1}, \eqref{app:cons:mac_2}, 
 \eqref{app:cons:mac_4}, \eqref{app:cons:mac_5}, \eqref{app:cons:mac_6}, \eqref{app:cons:mac_7}, and \eqref{app:cons:mac_8} into GP format:

\noindent\textbf{Constraint~\eqref{app:cons:mac_1}:} We first rewrite \eqref{app:cons:mac_1} as follows:
\begin{equation}\label{app:cons:mac_1_1}
   \frac{\tau_{b}^{\downarrow,{(k)}}}{t_{x}}-\frac{\tau_{b}^{\downarrow,{(k)}}}{t_{x}}\beta^{\downarrow}_{b}(t_{x})-1+\beta^{\downarrow}_{b}(t_{x}) \leq  0.
\end{equation}
Performing some algebraic operations gives us
\begin{equation}\label{app:cons:mac_1_2}
   \frac{\tau_{b}^{\downarrow,{(k)}}+t_{x}\beta^{\downarrow}_{b}(t_{x})}{t_{x}+\tau_{b}^{\downarrow,{(k)}}\beta^{\downarrow}_{b}(t_{x})} \leq 1.
\end{equation}
The fraction in~\eqref{app:cons:mac_1_2} is not in the format of GP since it is an inequality with a posynomial in the denominator, which is not a posynomial. We thus exploit arithmetic-geometric mean inequality (Lemma~\ref{Lemma:ArethmaticGeometric}) to approximate the denominator with a monomial:
\begin{align}\label{app:cons:mac_1_3}
    H_1(\bm{v})=t_{x}+\tau_{b}^{\downarrow,{(k)}}\beta^{\downarrow}_{b}(t_{x})&\geq \widehat{H}_1(\bm{v};\ell) \triangleq \left(H_1([\bm{v}]^{(\ell-1)})\right)^{\frac{\left[t_{x}\right]^{(\ell-1)} }{H_1([\bm{v}]^{(\ell-1)})}} \times\left(\frac{\tau_{b}^{\downarrow,{(k)}}\beta^{\downarrow}_{b}(t_{x}) H_1([\bm{v}]^{(\ell-1)})}{\left[\tau_{b}^{\downarrow,{(k)}}\beta^{\downarrow}_{b}(t_{x})\right]^{(\ell-1)}}\right)^{\frac{\left[\tau_{b}^{\downarrow,{(k)}}\beta^{\downarrow}_{b}(t_{x})\right]^{(\ell-1)}}{H_1([\bm{v}]^{(\ell-1)})}}
\end{align}

We finally approximate the constraint as follows:
\begin{tcolorbox}[ams align]
    \frac{\tau_{b}^{\downarrow,{(k)}}+t_{x}\beta^{\downarrow}_{b}(t_{x})}{\widehat{H}_1(\bm{v};\ell)} \leq 1,~\forall b\in\Omega,~\forall x\in\mathcal{N}^{(k)}\nonumber
\end{tcolorbox}

\noindent\textbf{Constraint~\eqref{app:cons:mac_2}:} We first rewrite \eqref{app:cons:mac_2} as follows:
\begin{equation}\label{app:cons:mac_2_1}
   \bigg(1{-}\frac{\tau_{b}^{\downarrow,{(k)}}}{t_{x}}\bigg)\times\beta^{\downarrow}_{b}(t_{x})\le \epsilon.
\end{equation}
Performing some algebraic operations gives us
\begin{equation}\label{app:cons:mac_2_2}
    \frac{t_{x}\beta^{\downarrow}_{b}(t_{x})}{\epsilon t_{x}+\beta^{\downarrow}_{b}(t_{x})\tau_{b}^{\downarrow,{(k)}}}\le 1.
\end{equation}
The fraction in~\eqref{app:cons:mac_2_2} is not in the format of GP since it is an inequality with a posynomial in the denominator, which is not a posynomial. We thus exploit arithmetic-geometric mean inequality (Lemma~\ref{Lemma:ArethmaticGeometric}) to approximate the denominator with a monomial:
\begin{align}\label{app:cons:mac_2_3}
    H_2(\bm{v})=\epsilon t_{x}+\beta^{\downarrow}_{b}(t_{x})\tau_{b}^{\downarrow,{(k)}}&\geq \widehat{H}_2(\bm{v};\ell) \triangleq \left( H_2([\bm{v}]^{(\ell-1)})\right)^{\frac{\left[\epsilon t_{x}\right]^{(\ell-1)} }{H_2([\bm{v}]^{(\ell-1)})}} \times\left(\frac{\tau_{b}^{\downarrow,{(k)}}\beta^{\downarrow}_{b}(t_{x}) H_2([\bm{v}]^{(\ell-1)})}{\left[\tau_{b}^{\downarrow,{(k)}}\beta^{\downarrow}_{b}(t_{x})\right]^{(\ell-1)}}\right)^{\frac{\left[\tau_{b}^{\downarrow,{(k)}}\beta^{\downarrow}_{b}(t_{x})\right]^{(\ell-1)}}{H_2([\bm{v}]^{(\ell-1)})}}.
\end{align}

We finally approximate the constraint as follows:
\begin{tcolorbox}[ams align]
    \frac{t_{x}\beta^{\downarrow}_{b}(t_{x})}{\widehat{H}_2(\bm{v};\ell)} \leq 1,~\forall b\in\Omega,~\forall x\in\mathcal{N}^{(k)}\nonumber
\end{tcolorbox}



\noindent\textbf{Constraint~\eqref{app:cons:mac_4}:} Performing some algebraic operations on \eqref{app:cons:mac_4} gives us
\begin{equation}\label{app:cons:mac_4_2}
     \frac{t_{x}\overline{\lambda}_{u}^{(k)}\Psi(\overline{\mathscr{D}}_{u}^{\nuparrow,(k)}){+}t_{x}\overline{\lambda}_{u}^{(k)}\Psi(\overline{\mathscr{D}}_{u}^{\uparrow,(k)}){+}(t_{x})^2\overline{\beta}^{\uparrow}_u(t_{x}){+}\overline{\beta}^{\uparrow}_u(t_{x})(\overline{\lambda}_{u}^{(k)})^2\Psi(\overline{\mathscr{D}}_{u}^{\uparrow,(k)})\Psi(\overline{\mathscr{D}}_{u}^{\nuparrow,(k)})}{(t_{x})^2{+}(\overline{\lambda}_{u}^{(k)})^2\Psi(\overline{\mathscr{D}}_{u}^{\uparrow,(k)})\Psi(\overline{\mathscr{D}}_{u}^{\nuparrow,(k)}){+}t_{x}\overline{\beta}^{\uparrow}_u(t_{x})\overline{\lambda}_{u}^{(k)}\Psi(\overline{\mathscr{D}}_{u}^{\nuparrow,(k)}){+}t_{x}\overline{\beta}^{\uparrow}_u(t_{x})\overline{\lambda}_{u}^{(k)}\Psi(\overline{\mathscr{D}}_{u}^{\uparrow,(k)})}{\leq}1.
\end{equation}
The fraction in~\eqref{app:cons:mac_4_2} is not in the format of GP since it is an inequality with a posynomial in the denominator, which is not a posynomial. We thus exploit arithmetic-geometric mean inequality (Lemma~\ref{Lemma:ArethmaticGeometric}) to approximate the denominator with a monomial:
\begin{align}\label{app:cons:mac_4_3}
    &H_4(\bm{v})=(t_{x})^2{+}(\overline{\lambda}_{u}^{(k)})^2\Psi(\overline{\mathscr{D}}_{u}^{\uparrow,(k)})\Psi(\overline{\mathscr{D}}_{u}^{\nuparrow,(k)}){+}t_{x}\overline{\beta}^{\uparrow}_u(t_{x})\overline{\lambda}_{u}^{(k)}\Psi(\overline{\mathscr{D}}_{u}^{\nuparrow,(k)}){+}t_{x}\overline{\beta}^{\uparrow}_u(t_{x})\overline{\lambda}_{u}^{(k)}\Psi(\overline{\mathscr{D}}_{u}^{\uparrow,(k)})\nonumber\\
    &\geq \widehat{H}_4(\bm{v};\ell) \triangleq \left(H_4([\bm{v}]^{(\ell-1)})\right)^{\frac{\left[(t_{x})^2\right]^{(\ell-1)} }{H_4([\bm{v}]^{(\ell-1)})}} \times\left(\frac{(\overline{\lambda}_{u}^{(k)})^2\Psi(\overline{\mathscr{D}}_{u}^{\uparrow,(k)})\Psi(\overline{\mathscr{D}}_{u}^{\nuparrow,(k)}) H_4([\bm{v}]^{(\ell-1)})}{\left[(\overline{\lambda}_{u}^{(k)})^2\Psi(\overline{\mathscr{D}}_{u}^{\uparrow,(k)})\Psi(\overline{\mathscr{D}}_{u}^{\nuparrow,(k)})\right]^{(\ell-1)}}\right)^{\frac{\left[(\overline{\lambda}_{u}^{(k)})^2\Psi(\overline{\mathscr{D}}_{u}^{\uparrow,(k)})\Psi(\overline{\mathscr{D}}_{u}^{\nuparrow,(k)})\right]^{(\ell-1)}}{H_4([\bm{v}]^{(\ell-1)})}}\nonumber\\
    &\times \left(\frac{\overline{\beta}^{\uparrow}_u(t_{x})\overline{\lambda}_{u}^{(k)}\Psi(\overline{\mathscr{D}}_{u}^{\nuparrow,(k)}) H_4([\bm{v}]^{(\ell-1)})}{\left[\overline{\beta}^{\uparrow}_u(t_{x})\overline{\lambda}_{u}^{(k)}\Psi(\overline{\mathscr{D}}_{u}^{\nuparrow,(k)})\right]^{(\ell-1)}}\right)^{\frac{\left[t_{x}\overline{\beta}^{\uparrow}_u(t_{x})\overline{\lambda}_{u}^{(k)}\Psi(\overline{\mathscr{D}}_{u}^{\nuparrow,(k)})\right]^{(\ell-1)} }{H_4([\bm{v}]^{(\ell-1)})}} \times \left(\frac{\overline{\beta}^{\uparrow}_u(t_{x})\overline{\lambda}_{u}^{(k)}\Psi(\overline{\mathscr{D}}_{u}^{\uparrow,(k)}) H_4([\bm{v}]^{(\ell-1)})}{\left[\overline{\beta}^{\uparrow}_u(t_{x})\overline{\lambda}_{u}^{(k)}\Psi(\overline{\mathscr{D}}_{u}^{\uparrow,(k)})\right]^{(\ell-1)}}\right)^{\frac{\left[t_{x}\overline{\beta}^{\uparrow}_u(t_{x})\overline{\lambda}_{u}^{(k)}\Psi(\overline{\mathscr{D}}_{u}^{\uparrow,(k)})\right]^{(\ell-1)} }{H_4([\bm{v}]^{(\ell-1)})}}.
\end{align}

We finally approximate the constraint as follows:
\begin{tcolorbox}[ams align]
\hspace{-4mm}
    \frac{t_{x}\overline{\lambda}_{u}^{(k)}\Psi(\overline{\mathscr{D}}_{u}^{\nuparrow,(k)}){+}t_{x}\overline{\lambda}_{u}^{(k)}\Psi(\overline{\mathscr{D}}_{u}^{\uparrow,(k)}){+}(t_{x})^2\overline{\beta}^{\uparrow}_u(t_{x}){+}\overline{\beta}^{\uparrow}_u(t_{x})(\overline{\lambda}_{u}^{(k)})^2\Psi(\overline{\mathscr{D}}_{u}^{\uparrow,(k)})\Psi(\overline{\mathscr{D}}_{u}^{\nuparrow,(k)})}{\widehat{H}_4(\bm{v};\ell)} \leq 1.\nonumber
\hspace{-3mm}
\end{tcolorbox}

\noindent\textbf{Constraint~\eqref{app:cons:mac_5}:} We first rewrite \eqref{app:cons:mac_5} as the following two inequalities:
\begin{equation}\label{app:cons:mac_5_1}
  \left(\frac{\Psi(\overline{\mathscr{D}}_{u}^{\uparrow,(k)})}{t_{x}}{-}1\right)\times\overline{\beta}^{\uparrow}_u(t_{x}){\leq} \epsilon,
\end{equation}
and
\begin{equation}\label{app:cons:mac_5_2}
  \left(1{-}\frac{\Psi(\overline{\mathscr{D}}_{u}^{\nuparrow,(k)})}{t_{x}}\right)\times\overline{\beta}^{\uparrow}_u(t_{x}){\leq} \epsilon.
\end{equation}
Performing some algebraic operations on the above two inequalities gives us
\begin{equation}\label{app:cons:mac_5_3}
  \frac{\overline{\beta}^{\uparrow}_u(t_{x})\Psi(\overline{\mathscr{D}}_{u}^{\uparrow,(k)})}{\epsilon t_{x}+t_{x}\overline{\beta}^{\uparrow}_u(t_{x})}{\leq} 1,
\end{equation}
and
\begin{equation}\label{app:cons:mac_5_4}
  \frac{t_{x}\overline{\beta}^{\uparrow}_u(t_{x})}{\epsilon t_{x}+\overline{\beta}^{\uparrow}_u(t_{x})\Psi(\overline{\mathscr{D}}_{u}^{\nuparrow,(k)})}{\leq} 1.
\end{equation}
The fraction in~\eqref{app:cons:mac_5_3} and~\eqref{app:cons:mac_5_4} are not in the format of GP since they are inequalities with a posynomial in the denominator, which are not posynomials. We thus exploit arithmetic-geometric mean inequality (Lemma~\ref{Lemma:ArethmaticGeometric}) to approximate the denominator of the above inequalities with monomials:
\begin{align}\label{app:cons:mac_5_5}
    H_5(\bm{v})=\epsilon t_{x}+t_{x}\overline{\beta}^{\uparrow}_u(t_{x})&\geq \widehat{H}_5(\bm{v};\ell) \triangleq \left( H_5([\bm{v}]^{(\ell-1)})\right)^{\frac{\left[\epsilon t_{x}\right]^{(\ell-1)} }{H_5([\bm{v}]^{(\ell-1)})}} \times\left(\frac{\overline{\beta}^{\uparrow}_u(t_{x}) H_5([\bm{v}]^{(\ell-1)})}{\left[\overline{\beta}^{\uparrow}_u(t_{x})\right]^{(\ell-1)}}\right)^{\frac{\left[t_{x}\overline{\beta}^{\uparrow}_u(t_{x})\right]^{(\ell-1)}}{H_5([\bm{v}]^{(\ell-1)})}},
\end{align}
and
\begin{align}\label{app:cons:mac_5_6}
    H_6(\bm{v}){=}\epsilon t_{x}{+}\overline{\beta}^{\uparrow}_u(t_{x})\Psi(\overline{\mathscr{D}}_{u}^{\nuparrow,(k)})&{\geq} \widehat{H}_6(\bm{v};\ell) {\triangleq} \left(H_6([\bm{v}]^{(\ell-1)})\right)^{\frac{\left[\epsilon t_{x}\right]^{(\ell-1)} }{H_6([\bm{v}]^{(\ell-1)})}} \times\left(\frac{\overline{\beta}^{\uparrow}_u(t_{x})\Psi(\overline{\mathscr{D}}_{u}^{\nuparrow,(k)}) H_6([\bm{v}]^{(\ell-1)})}{\left[\overline{\beta}^{\uparrow}_u(t_{x})\Psi(\overline{\mathscr{D}}_{u}^{\nuparrow,(k)})\right]^{(\ell-1)}}\right)^{\frac{\left[\overline{\beta}^{\uparrow}_u(t_{x})\Psi(\overline{\mathscr{D}}_{u}^{\nuparrow,(k)})\right]^{(\ell-1)}}{H_6([\bm{v}]^{(\ell-1)})}}.
\end{align}

We finally approximate \eqref{app:cons:mac_5_3} and \eqref{app:cons:mac_5_4}, respectively, as follows:
\begin{tcolorbox}[ams align]
    \frac{\overline{\beta}^{\uparrow}_u(t_{x})\Psi(\overline{\mathscr{D}}_{u}^{\uparrow,(k)})}{\widehat{H}_5(\bm{v};\ell)} \leq 1,~\forall u \in \mathcal{U}_{b},~\forall b\in\Omega,~\forall x\in\mathcal{N}^{(k)}, \forall k\in\mathcal{K}.\nonumber
\end{tcolorbox}
and
\begin{tcolorbox}[ams align]
    \frac{t_{x}\overline{\beta}^{\uparrow}_u(t_{x})}{\widehat{H}_6(\bm{v};\ell)} \leq 1,~\forall u \in \mathcal{U}_{b},~\forall b\in\Omega,~\forall x\in\mathcal{N}^{(k)}, \forall k\in\mathcal{K}.\nonumber
\end{tcolorbox}

\noindent\textbf{Constraint~\eqref{app:cons:mac_6}:} We first rewrite \eqref{app:cons:mac_6} as follows:
\begin{equation}\label{app:cons:mac_6_1}
  \beta^{\uparrow}_u(t_{x})\le\epsilon + \lambda_{u}^{(k)} 
\end{equation}
Performing some algebraic operations gives us
\begin{equation}\label{app:cons:mac_6_2}
    \frac{\beta^{\uparrow}_u(t_{x})}{\epsilon + \lambda_{u}^{(k)}}\le 1.
\end{equation}
The fraction in~\eqref{app:cons:mac_6_2} is not in the format of GP since it is an inequality with a posynomial in the denominator, which is not a posynomial. We thus exploit arithmetic-geometric mean inequality (Lemma~\ref{Lemma:ArethmaticGeometric}) to approximate the denominator with a monomial:
\begin{align}\label{app:cons:mac_6_3}
    H_6(\bm{v})=\epsilon + \lambda_{u}^{(k)}&\geq \widehat{H}_6(\bm{v};\ell) \triangleq \left(H_6([\bm{v}]^{(\ell-1)})\right)^{\frac{\displaystyle\epsilon}{H_6([\bm{v}]^{(\ell-1)})}} \times\left(\frac{\lambda_{u}^{(k)} H_6([\bm{v}]^{(\ell-1)})}{\left[\lambda_{u}^{(k)}\right]^{(\ell-1)}}\right)^{\frac{\left[\lambda_{u}^{(k)}\right]^{(\ell-1)}}{H_6([\bm{v}]^{(\ell-1)})}}.
\end{align}

We finally approximate the constraint as follows:
\begin{tcolorbox}[ams align]
    \frac{\beta^{\uparrow}_u(t_{x})}{\widehat{H}_6(\bm{v};\ell)} \leq 1,~\forall u \in \mathcal{U}_{b},~\forall b\in\Omega,~\forall x\in\mathcal{N}^{(k)}, \forall k\in\mathcal{K}.\nonumber
\end{tcolorbox}

\noindent\textbf{Constraint~\eqref{app:cons:mac_7}:} Performing some algebraic operations on \eqref{app:cons:mac_7} gives us
\begin{equation}\label{app:cons:mac_7_2}
     \frac{t_{x}\Psi(\mathscr{D}_{u}^{\nuparrow,(k)}){+}t_{x}\Psi(\mathscr{D}_{u}^{\uparrow,(k)}){+}(t_{x})^2\beta^{\uparrow}_u(t_{x}){+}\beta^{\uparrow}_u(t_{x})\Psi(\mathscr{D}_{u}^{\uparrow,(k)})\Psi(\mathscr{D}_{u}^{\nuparrow,(k)})}{(t_{x})^2{+}\Psi(\mathscr{D}_{u}^{\uparrow,(k)})\Psi(\mathscr{D}_{u}^{\nuparrow,(k)}){+}t_{x}\beta^{\uparrow}_u(t_{x})\Psi(\mathscr{D}_{u}^{\nuparrow,(k)}){+}t_{x}\beta^{\uparrow}_u(t_{x})\Psi(\mathscr{D}_{u}^{\uparrow,(k)})}{\leq}1.
\end{equation}
The fraction in~\eqref{app:cons:mac_7_2} is not in the format of GP since it is an inequality with a posynomial in the denominator, which is not a posynomial. We thus exploit arithmetic-geometric mean inequality (Lemma~\ref{Lemma:ArethmaticGeometric}) to approximate the denominator with a monomial:
\begin{align}\label{app:cons:mac_7_3}
    H_7(\bm{v})&=(t_{x})^2{+}\Psi(\mathscr{D}_{u}^{\uparrow,(k)})\Psi(\mathscr{D}_{u}^{\nuparrow,(k)}){+}t_{x}\beta^{\uparrow}_u(t_{x})\Psi(\mathscr{D}_{u}^{\nuparrow,(k)}){+}t_{x}\beta^{\uparrow}_u(t_{x})\Psi(\mathscr{D}_{u}^{\uparrow,(k)})\nonumber\\
    &\geq \widehat{H}_7(\bm{v};\ell) \triangleq \left(H_7([\bm{v}]^{(\ell-1)})\right)^{\frac{\left[(t_{x})^2\right]^{(\ell-1)} }{H_7([\bm{v}]^{(\ell-1)})}} \times\left(\frac{\Psi(\mathscr{D}_{u}^{\uparrow,(k)})\Psi(\mathscr{D}_{u}^{\nuparrow,(k)}) H_7([\bm{v}]^{(\ell-1)})}{\left[\Psi(\mathscr{D}_{u}^{\uparrow,(k)})\Psi(\mathscr{D}_{u}^{\nuparrow,(k)})\right]^{(\ell-1)}}\right)^{\frac{\left[\Psi(\mathscr{D}_{u}^{\uparrow,(k)})\Psi(\mathscr{D}_{u}^{\nuparrow,(k)})\right]^{(\ell-1)}}{H_7([\bm{v}]^{(\ell-1)})}}\nonumber\\
    &\times \left(\frac{\beta^{\uparrow}_u(t_{x})\Psi(\mathscr{D}_{u}^{\nuparrow,(k)}) H_7([\bm{v}]^{(\ell-1)})}{\left[\beta^{\uparrow}_u(t_{x})\Psi(\mathscr{D}_{u}^{\nuparrow,(k)})\right]^{(\ell-1)}}\right)^{\frac{\left[t_{x}\beta^{\uparrow}_u(t_{x})\Psi(\mathscr{D}_{u}^{\nuparrow,(k)})\right]^{(\ell-1)} }{H_7([\bm{v}]^{(\ell-1)})}} \times \left(\frac{\beta^{\uparrow}_u(t_{x})\Psi(\mathscr{D}_{u}^{\uparrow,(k)}) H_7([\bm{v}]^{(\ell-1)})}{\left[\beta^{\uparrow}_u(t_{x})\Psi(\mathscr{D}_{u}^{\uparrow,(k)})\right]^{(\ell-1)}}\right)^{\frac{\left[t_{x}\beta^{\uparrow}_u(t_{x})\Psi(\mathscr{D}_{u}^{\uparrow,(k)})\right]^{(\ell-1)} }{H_7([\bm{v}]^{(\ell-1)})}}.
\end{align}

We finally approximate the constraint as follows:
\begin{tcolorbox}[ams align]
\hspace{-4mm}
    \frac{t_{x}\Psi(\mathscr{D}_{u}^{\nuparrow,(k)}){+}t_{x}\Psi(\mathscr{D}_{u}^{\uparrow,(k)}){+}(t_{x})^2\beta^{\uparrow}_u(t_{x}){+}\beta^{\uparrow}_u(t_{x})\Psi(\mathscr{D}_{u}^{\uparrow,(k)})\Psi(\mathscr{D}_{u}^{\nuparrow,(k)})}{\widehat{H}_7(\bm{v};\ell)} \leq 1,~\forall u \in \mathcal{U}_{b},~\forall b\in\Omega,~\forall x\in\mathcal{N}^{(k)}, \forall k\in\mathcal{K}.\nonumber
\hspace{-3mm}
\end{tcolorbox}

\noindent\textbf{Constraint~\eqref{app:cons:mac_8}:} We first rewrite \eqref{app:cons:mac_8} as the following two inequalities:
\begin{equation}\label{app:cons:mac_8_1}
  \left(\frac{\overline{\lambda}_{u}^{(k)}\Psi(\mathscr{D}_{u}^{\uparrow,(k)})}{t_{x}}{-}1\right)\times\beta^{\uparrow}_u(t_{x})+1{\leq} 1,
\end{equation}
and
\begin{equation}\label{app:cons:mac_8_2}
  \left(1{-}\frac{\overline{\lambda}_{u}^{(k)}\Psi(\mathscr{D}_{u}^{\nuparrow,(k)})}{t_{x}}\right)\times\beta^{\uparrow}_u(t_{x})+1{\leq} 1.
\end{equation}
Performing some algebraic operations on the above two inequalities gives us
\begin{equation}\label{app:cons:mac_8_3}
    \frac{\beta^{\uparrow}_u(t_{x})\overline{\lambda}_{u}^{(k)}\Psi(\mathscr{D}_{u}^{\uparrow,(k)})+t_{x}}{t_{x}+t_{x}\beta^{\uparrow}_u(t_{x})}{\leq} 1,
\end{equation}
and
\begin{equation}\label{app:cons:mac_8_4}
    \frac{\beta^{\uparrow}_u(t_{x})t_{x}+t_{x}}{t_{x}+\beta^{\uparrow}_u(t_{x})\overline{\lambda}_{u}^{(k)}\Psi(\mathscr{D}_{u}^{\nuparrow,(k)})}{\leq} 1.
\end{equation}
The fraction in~\eqref{app:cons:mac_8_3} and~\eqref{app:cons:mac_8_4} are not in the format of GP since they are inequalities with a posynomial in the denominator, which are not posynomials. We thus exploit arithmetic-geometric mean inequality (Lemma~\ref{Lemma:ArethmaticGeometric}) to approximate the denominator of the above inequalities with monomials:
\begin{align}\label{app:cons:mac_8_5}
    H_8(\bm{v})= t_{x}+t_{x}\beta^{\uparrow}_u(t_{x})&\geq \widehat{H}_8(\bm{v};\ell) \triangleq \left(H_8([\bm{v}]^{(\ell-1)})\right)^{\frac{\left[ t_{x}\right]^{(\ell-1)} }{H_8([\bm{v}]^{(\ell-1)})}} \times\left(\frac{\beta^{\uparrow}_u(t_{x}) H_8([\bm{v}]^{(\ell-1)})}{\left[\beta^{\uparrow}_u(t_{x})\right]^{(\ell-1)}}\right)^{\frac{\left[t_{x}\beta^{\uparrow}_u(t_{x})\right]^{(\ell-1)}}{H_8([\bm{v}]^{(\ell-1)})}},
\end{align}
and
\begin{align}\label{app:cons:mac_8_6}
    H_9(\bm{v})&{=} t_{x}{+}\beta^{\uparrow}_u(t_{x})\overline{\lambda}_{u}^{(k)}\Psi(\mathscr{D}_{u}^{\nuparrow,(k)})\nonumber\\
    &{\geq} \widehat{H}_9(\bm{v};\ell) {\triangleq} \left(H_9([\bm{v}]^{(\ell-1)})\right)^{\frac{\left[ t_{x}\right]^{(\ell-1)} }{H_9([\bm{v}]^{(\ell-1)})}} \left(\frac{\beta^{\uparrow}_u(t_{x})\overline{\lambda}_{u}^{(k)}\Psi(\mathscr{D}_{u}^{\nuparrow,(k)}) H_9([\bm{v}]^{(\ell-1)})}{\left[\beta^{\uparrow}_u(t_{x})\overline{\lambda}_{u}^{(k)}\Psi(\mathscr{D}_{u}^{\nuparrow,(k)})\right]^{(\ell-1)}}\right)^{\frac{\left[\beta^{\uparrow}_u(t_{x})\overline{\lambda}_{u}^{(k)}\Psi(\mathscr{D}_{u}^{\nuparrow,(k)})\right]^{(\ell-1)}}{H_9([\bm{v}]^{(\ell-1)})}}.
\end{align}

We finally approximate \eqref{app:cons:mac_8_3} and \eqref{app:cons:mac_8_4},respectively, as follows:
\begin{tcolorbox}[ams align]
    \frac{\beta^{\uparrow}_u(t_{x})\overline{\lambda}_{u}^{(k)}\Psi(\mathscr{D}_{u}^{\uparrow,(k)})+t_{x}}{\widehat{H}_8(\bm{v};\ell)} \leq 1,~\forall u \in \mathcal{U}_{b},~\forall b\in\Omega,~\forall x\in\mathcal{N}^{(k)}, \forall k\in\mathcal{K}.\nonumber
\end{tcolorbox}
and
\begin{tcolorbox}[ams align]
    \frac{\beta^{\uparrow}_u(t_{x})t_{x}+t_{x}}{\widehat{H}_9(\bm{v};\ell)} \leq 1,~\forall u \in \mathcal{U}_{b},~\forall b\in\Omega,~\forall x\in\mathcal{N}^{(k)}, \forall k\in\mathcal{K}.\nonumber
\end{tcolorbox}

\newpage
\subsection{Transformation of PRB/Power Allocations Constraints}\label{app:cons:PRB_Power_allocation}
Let us revisit the constraints of Sec.~\ref{sec:resource_allocation_mac} and Sec.~\ref{sec:power_allocation} in the following:

\textbf{Dynamic PRB allocation:}
\begin{equation}\label{app:cons:mac_r_1}
        \overline{\varrho}_{u,u',r}(t_{x}){-}\min\big\{\overline{\lambda}_{u}^{(k)},\overline{\beta}^{\uparrow}_u(t_{x}),\lambda_{u}^{(k)}\big\} {\le} 0.
\end{equation}
\begin{equation}\label{app:cons:mac_r_2}
        \overline{\varrho}_{u,u',r}(t_{x}){+}\max_{z{\in}\{\mathcal{N}^{(k{-}1)}{+}1,\cdots,x\}}\{\beta^{\uparrow}_u(t_{z})\} {\le} 1.
\end{equation}
\begin{equation}\label{app:cons:mac_r_3}
\begin{aligned}
     \sum_{u{\in} \mathcal{U}_{b}}\sum_{r{\in}\overline{\mathcal{R}}_{b}} \overline{\varrho}_{u,u',r}(t_{x})\ge 1.
\end{aligned}
\end{equation}
\begin{equation}\label{app:cons:mac_r_4}
        \varrho_{u,r}(t_{x}){-}\min\big\{\lambda_{u}^{(k)},\beta^{\uparrow}_u(t_{x})\big\}{\le}0.
\end{equation}

\textbf{Dynamic Power Allocation:}
\begin{equation}\label{app:cons:mac_p_1}
    \begin{aligned}
        0\le \beta^{\downarrow}_{b}(t_{x})-\rho^{\downarrow}_{b,r}(t_{x})<1.
    \end{aligned}
\end{equation}
\begin{equation}\label{app:cons:mac_p_2}
    \begin{aligned}
      0\le \varrho_{u,r}(t_{x})-\rho^{\uparrow}_{u,r}(t_{x})< 1.
    \end{aligned}
\end{equation}
\begin{equation}\label{app:cons:mac_p_3}
    \begin{aligned}
      0\le \overline{\varrho}_{u,u',r}(t_{x})-\overline{\rho}^{\uparrow}_{u,r}(t_{x})< 1.
    \end{aligned}
\end{equation}

In the following, we aim to transform \eqref{app:cons:mac_r_1}, \eqref{app:cons:mac_r_2}, \eqref{app:cons:mac_r_3}, 
 \eqref{app:cons:mac_r_4}, \eqref{app:cons:mac_p_1}, \eqref{app:cons:mac_p_2}, and \eqref{app:cons:mac_p_3} into GP format:

\noindent\textbf{Constraint~\eqref{app:cons:mac_r_1}:} We first rewrite \eqref{app:cons:mac_r_1} as the following three inequalities:
\begin{equation}\label{app:cons:mac_r_1_1}
        \overline{\varrho}_{u,u',r}(t_{x}){-}\overline{\lambda}_{u}^{(k)} {\le} \epsilon,
\end{equation}
\begin{equation}\label{app:cons:mac_r_1_2}
        \overline{\varrho}_{u,u',r}(t_{x}){-}\overline{\beta}^{\uparrow}_u(t_{x}){\le} \epsilon,
\end{equation}
and
\begin{equation}\label{app:cons:mac_r_1_3}
        \overline{\varrho}_{u,u',r}(t_{x}){-}\lambda_{u}^{(k)} {\le} \epsilon.
\end{equation}

Performing some algebraic operations gives us
\begin{equation}\label{app:cons:mac_r_1_4}
        \frac{\overline{\varrho}_{u,u',r}(t_{x})}{\epsilon+\overline{\lambda}_{u}^{(k)}} {\le} 1,
\end{equation}
\begin{equation}\label{app:cons:mac_r_1_5}
        \frac{\overline{\varrho}_{u,u',r}(t_{x})}{\epsilon+\overline{\beta}^{\uparrow}_u(t_{x})}{\le} 1,
\end{equation}
and
\begin{equation}\label{app:cons:mac_r_1_6}
        \frac{\overline{\varrho}_{u,u',r}(t_{x})}{\epsilon+\lambda_{u}^{(k)}} {\le} 1.
\end{equation}
The fraction in~\eqref{app:cons:mac_r_1_4}, \eqref{app:cons:mac_r_1_5}, and \eqref{app:cons:mac_r_1_6} are not in the format of GP since they are inequalities with a posynomial in their denominators, which are not posynomials. We thus exploit arithmetic-geometric mean inequality (Lemma~\ref{Lemma:ArethmaticGeometric}) to approximate the denominator of the above inequalities with the following monomials, respectively:
\begin{align}\label{app:cons:mac_r_1_7}
    G_{1}(\bm{v})=\epsilon+\overline{\lambda}_{u}^{(k)}&\geq \widehat{G}_{1}(\bm{v};\ell) \triangleq \left(G_{1}([\bm{v}]^{(\ell-1)})\right)^{\frac{\displaystyle\epsilon}{G_{1}([\bm{v}]^{(\ell-1)})}} \times\left(\frac{\overline{\lambda}_{u}^{(k)} G_{1}([\bm{v}]^{(\ell-1)})}{\left[\overline{\lambda}_{u}^{(k)}\right]^{(\ell-1)}}\right)^{\frac{\left[\overline{\lambda}_{u}^{(k)}\right]^{(\ell-1)}}{G_{1}([\bm{v}]^{(\ell-1)})}},
\end{align}
\begin{align}\label{app:cons:mac_r_1_8}
    G_{2}(\bm{v})=\epsilon+\overline{\beta}^{\uparrow}_u(t_{x})&\geq \widehat{G}_{2}(\bm{v};\ell) \triangleq \left(G_{2}([\bm{v}]^{(\ell-1)})\right)^{\frac{\displaystyle\epsilon}{G_{2}([\bm{v}]^{(\ell-1)})}} \times\left(\frac{\overline{\beta}^{\uparrow}_u(t_{x}) G_{2}([\bm{v}]^{(\ell-1)})}{\left[\overline{\beta}^{\uparrow}_u(t_{x})\right]^{(\ell-1)}}\right)^{\frac{\left[\overline{\beta}^{\uparrow}_u(t_{x})\right]^{(\ell-1)}}{G_{2}([\bm{v}]^{(\ell-1)})}},
\end{align}
and
\begin{align}\label{app:cons:mac_r_1_9}
    G_{3}(\bm{v})=\epsilon+\lambda_{u}^{(k)}&\geq \widehat{G}_{3}(\bm{v};\ell) \triangleq \left(G_{3}([\bm{v}]^{(\ell-1)})\right)^{\frac{\displaystyle\epsilon}{G_{3}([\bm{v}]^{(\ell-1)})}} \times\left(\frac{\lambda_{u}^{(k)} G_{3}([\bm{v}]^{(\ell-1)})}{\left[\lambda_{u}^{(k)}\right]^{(\ell-1)}}\right)^{\frac{\left[\lambda_{u}^{(k)}\right]^{(\ell-1)}}{G_{3}([\bm{v}]^{(\ell-1)})}}.
\end{align}

We finally approximate \eqref{app:cons:mac_r_1_4}, \eqref{app:cons:mac_r_1_5}, and \eqref{app:cons:mac_r_1_6} as follows:
\begin{tcolorbox}[ams align]
    \frac{\overline{\varrho}_{u,u',r}(t_{x})}{\widehat{G}_{1}(\bm{v};\ell)} \leq 1,~\forall u,u' \in \mathcal{U}_{b},~\forall b\in\Omega,~\forall x\in\mathcal{N}^{(k)}, \forall k\in\mathcal{K}.\nonumber
\end{tcolorbox}
\begin{tcolorbox}[ams align]
    \frac{\overline{\varrho}_{u,u',r}(t_{x})}{\widehat{G}_{2}(\bm{v};\ell)} \leq 1,~\forall u,u' \in \mathcal{U}_{b},~\forall b\in\Omega,~\forall x\in\mathcal{N}^{(k)}, \forall k\in\mathcal{K}.\nonumber
\end{tcolorbox}
\begin{tcolorbox}[ams align]
    \frac{\overline{\varrho}_{u,u',r}(t_{x})}{\widehat{G}_{3}(\bm{v};\ell)} \leq 1,~\forall u,u' \in \mathcal{U}_{b},~\forall b\in\Omega,~\forall x\in\mathcal{N}^{(k)}, \forall k\in\mathcal{K}.\nonumber
\end{tcolorbox}

\noindent\textbf{Constraint~\eqref{app:cons:mac_r_2}:} We first rewrite \eqref{app:cons:mac_r_2} as follows:
\begin{equation}\label{app:cons:mac_r_2_1}
    \overline{\varrho}_{u,u',r}(t_{x}){+}\max_{z{\in}\{\mathcal{N}^{(k{-}1)}{+}1,\cdots,x\}}\{\beta^{\uparrow}_u(t_{z})+\epsilon\} {\le} 1+\epsilon.
\end{equation}
To transform \eqref{app:cons:mac_r_2_1}, we use the approximation $\max\{A, B\}\approx (A^{p}+B^{p})^{-\frac{1}{p}}$, which is tight when $p \gg 1$ that gives us
\begin{equation}\label{app:cons:mac_r_2_2}
    \overline{\varrho}_{u,u',r}(t_{x}){+}\underbrace{\left(\sum_{z{=}N^{(k{-}1)}{+}1}^{x}\left(\beta^{\uparrow}_u(t_{z})+\epsilon\right)^{p}\right)^{-\frac{1}{p}}}_{(a)} {\le} 1+\epsilon.
\end{equation}
Term $(a)$ in the above inequality is not in the format of GP since it is not a posynomial. In order to transform \eqref{app:cons:mac_r_2_1} into GP format, we write it via introducing an auxiliary decision variable $B_{u,x}$ satisfying
\begin{equation}\label{app:cons:mac_r_2_3}
    B_{u,x}=\sum_{z{=}0}^{x}\left(\beta^{\uparrow}_u(t_{z})+\epsilon\right)^{p},
\end{equation}
which gives us the following posynomial:
\begin{equation}\label{app:cons:mac_r_2_4}
    \frac{\overline{\varrho}_{u,u',r}(t_{x})}{1+\epsilon}{+}\frac{\left(B_{u,x}\right)^{-\frac{1}{p}}}{1+\epsilon} {\le} 1.
\end{equation}
However, \eqref{app:cons:mac_r_2_3} is not in the format of GP and needs to be transformed in order to conform to the GP requirements. In doing so, we rewrite \eqref{app:cons:mac_r_2_3} as follows:
\begin{equation}\label{app:cons:mac_r_2_5}
    \frac{B_{u,x}}{\sum_{z{=}0}^{x}\left(\beta^{\uparrow}_u(t_{z})+\epsilon\right)^{p}}=1.
\end{equation}
This constraint is still not in the format of GP. Therefore, we transform it via introducing an auxiliary variable as two inequalities:
\begin{equation}\label{app:cons:mac_r_2_6}
    \frac{B_{u,x}}{\sum_{z{=}0}^{x}\left(\beta^{\uparrow}_u(t_{z})+\epsilon\right)^{p}}\le 1,
\end{equation}
\begin{equation}\label{app:cons:mac_r_2_7}
    \frac{\sum_{z{=}0}^{x}\left(\beta^{\uparrow}_u(t_{z})+\epsilon\right)^{p}}{A^{\mathsf{U}}\times B_{u,x}}\le 1,
\end{equation}
\begin{equation}
    A^{\mathsf{U}}\ge 1,
\end{equation}
where $A^{\mathsf{U}}$ is added with a large penalty term to the objective function to force $ A^{\mathsf{U}}{\rightarrow}1^+$ at the optimal point. The fraction in~\eqref{app:cons:mac_r_2_6} still needs transformation since it is an inequality with a posynomial in the denominator, which is not a posynomial. We thus exploit arithmetic-geometric mean inequality (Lemma~\ref{Lemma:ArethmaticGeometric}) to approximate the denominator with a monomial:
\begin{align}\label{app:cons:mac_r_2_8}
    G_{4}(\bm{v})=\sum_{z{=}0}^{x}\left(\beta^{\uparrow}_u(t_{z})+\epsilon\right)^{p}&\geq \widehat{G}_{4}(\bm{v};\ell) \triangleq \prod_{z{=}0}^{x}\left(\frac{\left(\beta^{\uparrow}_u(t_{z})+\epsilon\right)^{p} G_{4}([\bm{v}]^{(\ell-1)})}{\left[\left(\beta^{\uparrow}_u(t_{z})+\epsilon\right)^{p}\right]^{(\ell-1)}}\right)^{\frac{\left[\left(\beta^{\uparrow}_u(t_{z})+\epsilon\right)^{p}\right]^{(\ell-1)}}{G_{4}([\bm{v}]^{(\ell-1)})}},
\end{align}
which gives us an approximation of~\eqref{app:cons:mac_r_2_6} as follows:
\begin{equation}\label{app:cons:mac_r_2_9}
    \frac{B_{u,x}}{\widehat{G}_{4}(\bm{v};\ell)}\le 1.
\end{equation}

We finally approximate constraint~\eqref{app:cons:mac_r_2} as follows:
\begin{tcolorbox}[ams align]
     &\frac{\overline{\varrho}_{u,u',r}(t_{x})}{1+\epsilon}{+}\frac{\left(B_{u,x}\right)^{-\frac{1}{p}}}{1+\epsilon} {\le} 1 ,~\forall u,u' \in \mathcal{U}_{b},~\forall b\in\Omega,~\forall x\in\mathcal{N}^{(k)}, \forall k\in\mathcal{K},\nonumber\\
     &\frac{\sum_{z{=}0}^{x}\left(\beta^{\uparrow}_u(t_{z})+\epsilon\right)^{p}}{A^{\mathsf{U}}\times B_{u,x}}\le 1,\\
     &\frac{B_{u,x}}{\widehat{G}_{4}(\bm{v};\ell)}\le 1,\nonumber\\
     & \frac{1}{A^{\mathsf{U}}}\le 1.\nonumber
\end{tcolorbox}

\noindent\textbf{Constraint~\eqref{app:cons:mac_r_3}:} We first rewrite \eqref{app:cons:mac_r_3} as follows:
\begin{equation}\label{app:cons:mac_r_3_1}
    \frac{1}{\sum_{u{\in} \mathcal{U}_{b}}\sum_{r{\in}\overline{\mathcal{R}}_{b}} \overline{\varrho}_{u,u',r}(t_{x})}\le 1.
\end{equation}
The fraction in~\eqref{app:cons:mac_r_3_1} is not in the format of GP since it is an inequality with a posynomial in the denominator, which is not a posynomial. We thus exploit arithmetic-geometric mean inequality (Lemma~\ref{Lemma:ArethmaticGeometric}) to approximate the denominator with a monomial:
\begin{align}\label{app:cons:mac_r_3_2}
    G_{5}(\bm{v})=\sum_{u{\in} \mathcal{U}_{b}}\sum_{r{\in}\overline{\mathcal{R}}_{b}} \overline{\varrho}_{u,u',r}(t_{x})&\geq \widehat{G}_{5}(\bm{v};\ell) \triangleq \prod_{u{\in} \mathcal{U}_{b}}\prod_{r{\in}\overline{\mathcal{R}}_{b}}\left(\frac{\overline{\varrho}_{u,u',r}(t_{x}) G_{5}([\bm{v}]^{(\ell-1)})}{\left[\overline{\varrho}_{u,u',r}(t_{x})\right]^{(\ell-1)}}\right)^{\frac{\left[\overline{\varrho}_{u,u',r}(t_{x})\right]^{(\ell-1)}}{G_{5}([\bm{v}]^{(\ell-1)})}}.
\end{align}

We finally approximate the constraint as follows:
\begin{tcolorbox}[ams align]
    \frac{1}{\widehat{G}_{5}(\bm{v};\ell)} \leq 1,~u'\in\mathcal{U}_{b},~\forall b\in\Omega,~\forall x\in\mathcal{N}^{(k)},~k\in\mathcal{K}.\nonumber
\end{tcolorbox}

\noindent\textbf{Constraint~\eqref{app:cons:mac_r_4}:} We first rewrite \eqref{app:cons:mac_r_4} as the following two inequalities:
\begin{equation}\label{app:cons:mac_r_4_1}
        \varrho_{u,r}(t_{x}){-}\lambda_{u}^{(k)}{\le}\epsilon,
\end{equation}
and
\begin{equation}\label{app:cons:mac_r_4_2}
        \varrho_{u,r}(t_{x}){-}\beta^{\uparrow}_u(t_{x}){\le}\epsilon.
\end{equation}

Performing some algebraic operations on the above two inequalities gives us
\begin{equation}\label{app:cons:mac_r_4_1}
        \frac{\varrho_{u,r}(t_{x})}{\epsilon+\lambda_{u}^{(k)}}{\le}1,
\end{equation}
and
\begin{equation}\label{app:cons:mac_r_4_2}
        \frac{\varrho_{u,r}(t_{x})}{\epsilon+\beta^{\uparrow}_u(t_{x})}{\le} 1.
\end{equation}
The fraction in~\eqref{app:cons:mac_r_4_1} and \eqref{app:cons:mac_r_4_2} are not in the format of GP since they are inequalities with a posynomial in their denominators, which are not posynomials. We thus exploit arithmetic-geometric mean inequality (Lemma~\ref{Lemma:ArethmaticGeometric}) to approximate the denominator of the above inequalities with the following two monomials, respectively:
\begin{align}\label{app:cons:mac_r_4_3}
    G_{6}(\bm{v})=\epsilon+\lambda_{u}^{(k)}&\geq \widehat{G}_{6}(\bm{v};\ell) \triangleq \left(G_{6}([\bm{v}]^{(\ell-1)})\right)^{\frac{\displaystyle\epsilon}{G_{6}([\bm{v}]^{(\ell-1)})}} \times\left(\frac{\lambda_{u}^{(k)} G_{6}([\bm{v}]^{(\ell-1)})}{\left[\lambda_{u}^{(k)}\right]^{(\ell-1)}}\right)^{\frac{\left[\lambda_{u}^{(k)}\right]^{(\ell-1)}}{G_{6}([\bm{v}]^{(\ell-1)})}}
\end{align}
and
\begin{align}\label{app:cons:mac_r_4_4}
    G_{7}(\bm{v})=\epsilon+\beta^{\uparrow}_u(t_{x})&\geq \widehat{G}_{7}(\bm{v};\ell) \triangleq \left(G_{7}([\bm{v}]^{(\ell-1)})\right)^{\frac{\displaystyle\epsilon}{G_{7}([\bm{v}]^{(\ell-1)})}} \times\left(\frac{\beta^{\uparrow}_u(t_{x}) G_{7}([\bm{v}]^{(\ell-1)})}{\left[\beta^{\uparrow}_u(t_{x})\right]^{(\ell-1)}}\right)^{\frac{\left[\beta^{\uparrow}_u(t_{x})\right]^{(\ell-1)}}{G_{7}([\bm{v}]^{(\ell-1)})}}.
\end{align}

We finally approximate \eqref{app:cons:mac_r_4_1} and \eqref{app:cons:mac_r_4_2} as follows:
\begin{tcolorbox}[ams align]
    \frac{\varrho_{u,r}(t_{x})}{\widehat{G}_{6}(\bm{v};\ell)} \leq 1,~\forall u,u' \in \mathcal{U}_{b},~\forall b\in\Omega,~\forall x\in\mathcal{N}^{(k)}, \forall k\in\mathcal{K}.\nonumber
\end{tcolorbox}
\begin{tcolorbox}[ams align]
    \frac{\varrho_{u,r}(t_{x})}{\widehat{G}_{7}(\bm{v};\ell)} \leq 1,~\forall u,u' \in \mathcal{U}_{b},~\forall b\in\Omega,~\forall x\in\mathcal{N}^{(k)}, \forall k\in\mathcal{K}.\nonumber
\end{tcolorbox}

\noindent\textbf{Constraint~\eqref{app:cons:mac_p_1}:} We first rewrite \eqref{app:cons:mac_p_1} as the following two inequalities:
\begin{equation}\label{app:cons:mac_p_1_1}
    \begin{aligned}
        0\le \beta^{\downarrow}_{b}(t_{x})+\epsilon-\rho^{\downarrow}_{b,r}(t_{x})-\epsilon,
    \end{aligned}
\end{equation}
and
\begin{equation}\label{app:cons:mac_p_1_2}
    \begin{aligned}
        \beta^{\downarrow}_{b}(t_{x})-\rho^{\downarrow}_{b,r}(t_{x})\le 1-\epsilon.
    \end{aligned}
\end{equation}
Performing some algebraic operations on the above two inequalities gives us

\begin{equation}\label{app:cons:mac_p_1_1}
    \begin{aligned}
        \frac{\rho^{\downarrow}_{b,r}(t_{x})+\epsilon}{\beta^{\downarrow}_{b}(t_{x})+\epsilon}\le 1
    \end{aligned}
\end{equation}
and
\begin{equation}\label{app:cons:mac_p_1_2}
    \begin{aligned}
        \frac{\beta^{\downarrow}_{b}(t_{x})+\epsilon}{1+\rho^{\downarrow}_{b,r}(t_{x})}\le 1.
    \end{aligned}
\end{equation}

The fraction in~\eqref{app:cons:mac_p_1_2} is not in the format of GP since it is an inequality with a posynomial in the denominator, which is not a posynomial. We thus exploit arithmetic-geometric mean inequality (Lemma~\ref{Lemma:ArethmaticGeometric}) to approximate the denominator with a monomial:
\begin{align}\label{app:cons:mac_p_1_3}
    G_8(\bm{v})=1+\rho^{\downarrow}_{b,r}(t_{x})&\geq \widehat{G}_8(\bm{v};\ell) \triangleq \left(G_8([\bm{v}]^{(\ell-1)})\right)^{\frac{1}{G_8([\bm{v}]^{(\ell-1)})}}\left(\frac{\rho^{\downarrow}_{b,r}(t_{x}) G_8([\bm{v}]^{(\ell-1)})}{\left[\rho^{\downarrow}_{b,r}(t_{x})\right]^{(\ell-1)}}\right)^{\frac{\left[\rho^{\downarrow}_{b,r}(t_{x})\right]^{(\ell-1)} }{G_8([\bm{v}]^{(\ell-1)})}}.
\end{align}
We finally approximate the constraint \eqref{app:cons:mac_p_1} as follows:
\begin{tcolorbox}[ams align]
    &\frac{\rho^{\downarrow}_{b,r}(t_{x})+\epsilon}{\beta^{\downarrow}_{b}(t_{x})+\epsilon}\le 1,~\frac{\beta^{\downarrow}_{b}(t_{x})+\epsilon}{\widehat{G}_8(\bm{v};\ell)}{\le}1,~\forall r\in\mathcal{R}_{b},~\forall b\in\Omega,~\forall x\in\mathcal{N}^{(k)},~\forall k\in\mathcal{K}.\nonumber
\end{tcolorbox}

\noindent\textbf{Constraint~\eqref{app:cons:mac_p_1}:} We first rewrite \eqref{app:cons:mac_p_1} as the following two inequalities:
\begin{equation}\label{app:cons:mac_p_1_1}
    \begin{aligned}
        0\le \beta^{\downarrow}_{b}(t_{x})+\epsilon-\rho^{\downarrow}_{b,r}(t_{x})-\epsilon
    \end{aligned}
\end{equation}
and
\begin{equation}\label{app:cons:mac_p_1_2}
    \begin{aligned}
        \beta^{\downarrow}_{b}(t_{x})-\rho^{\downarrow}_{b,r}(t_{x})\le 1-\epsilon.
    \end{aligned}
\end{equation}
Performing some algebraic operations on the above two inequalities gives us

\begin{equation}\label{app:cons:mac_p_1_1}
    \begin{aligned}
        \frac{\rho^{\downarrow}_{b,r}(t_{x})+\epsilon}{\beta^{\downarrow}_{b}(t_{x})+\epsilon}\le 1
    \end{aligned}
\end{equation}
and
\begin{equation}\label{app:cons:mac_p_1_2}
    \begin{aligned}
        \frac{\beta^{\downarrow}_{b}(t_{x})+\epsilon}{1+\rho^{\downarrow}_{b,r}(t_{x})}\le 1.
    \end{aligned}
\end{equation}

The fraction in~\eqref{app:cons:mac_p_1_2} is not in the format of GP since it is an inequality with a posynomial in the denominator, which is not a posynomial. We thus exploit arithmetic-geometric mean inequality (Lemma~\ref{Lemma:ArethmaticGeometric}) to approximate the denominator with a monomial:
\begin{align}\label{app:cons:mac_p_1_3}
    G_8(\bm{v})=1+\rho^{\downarrow}_{b,r}(t_{x})&\geq \widehat{G}_8(\bm{v};\ell) \triangleq \left(G_8([\bm{v}]^{(\ell-1)})\right)^{\frac{1}{G_8([\bm{v}]^{(\ell-1)})}}\left(\frac{\rho^{\downarrow}_{b,r}(t_{x}) G_8([\bm{v}]^{(\ell-1)})}{\left[\rho^{\downarrow}_{b,r}(t_{x})\right]^{(\ell-1)}}\right)^{\frac{\left[\rho^{\downarrow}_{b,r}(t_{x})\right]^{(\ell-1)} }{G_8([\bm{v}]^{(\ell-1)})}}.
\end{align}
We finally approximate the constraint \eqref{app:cons:mac_p_1} as follows:
\begin{tcolorbox}[ams align]
    &\frac{\rho^{\downarrow}_{b,r}(t_{x})+\epsilon}{\beta^{\downarrow}_{b}(t_{x})+\epsilon}\le 1,~\frac{\beta^{\downarrow}_{b}(t_{x})+\epsilon}{\widehat{G}_8(\bm{v};\ell)}{\le}1,~\forall r\in\mathcal{R}_{b},~\forall b\in\Omega,~\forall x\in\mathcal{N}^{(k)},~\forall k\in\mathcal{K}.\nonumber
\end{tcolorbox}

\noindent\textbf{Constraint~\eqref{app:cons:mac_p_2}:} We first rewrite \eqref{app:cons:mac_p_2} as the following two inequalities:
\begin{equation}\label{app:cons:mac_p_2_1}
    \begin{aligned}
        0\le \varrho_{u,r}(t_{x})+\epsilon-\rho^{\uparrow}_{u,r}(t_{x})-\epsilon
    \end{aligned}
\end{equation}
and
\begin{equation}\label{app:cons:mac_p_2_2}
    \begin{aligned}
         \varrho_{u,r}(t_{x})-\rho^{\uparrow}_{u,r}(t_{x})\le 1-\epsilon.
    \end{aligned}
\end{equation}
Performing some algebraic operations on the above two inequalities gives us

\begin{equation}\label{app:cons:mac_p_2_1}
    \begin{aligned}
        \frac{\rho^{\uparrow}_{u,r}(t_{x})+\epsilon}{\varrho_{u,r}(t_{x})+\epsilon}\le 1
    \end{aligned}
\end{equation}
and
\begin{equation}\label{app:cons:mac_p_2_2}
    \begin{aligned}
        \frac{\varrho_{u,r}(t_{x})+\epsilon}{1+\rho^{\uparrow}_{u,r}(t_{x})}\le 1.
    \end{aligned}
\end{equation}

The fraction in~\eqref{app:cons:mac_p_2_2} is not in the format of GP since it is an inequality with a posynomial in the denominator, which is not a posynomial. We thus exploit arithmetic-geometric mean inequality (Lemma~\ref{Lemma:ArethmaticGeometric}) to approximate the denominator with a monomial:
\begin{align}\label{app:cons:mac_p_2_3}
    G_9(\bm{v})=1+\rho^{\uparrow}_{u,r}(t_{x})&\geq \widehat{G}_9(\bm{v};\ell) \triangleq \left(G_9([\bm{v}]^{(\ell-1)})\right)^{\frac{1}{G_9([\bm{v}]^{(\ell-1)})}}\left(\frac{\rho^{\uparrow}_{u,r}(t_{x}) G_9([\bm{v}]^{(\ell-1)})}{\left[\rho^{\uparrow}_{u,r}(t_{x})\right]^{(\ell-1)}}\right)^{\frac{\left[\rho^{\uparrow}_{u,r}(t_{x})\right]^{(\ell-1)} }{G_9([\bm{v}]^{(\ell-1)})}}.
\end{align}
We finally approximate the constraint \eqref{app:cons:mac_p_2} as follows:
\begin{tcolorbox}[ams align]
    &\frac{\rho^{\uparrow}_{u,r}(t_{x})+\epsilon}{\varrho_{u,r}(t_{x})+\epsilon}\le 1,~\frac{\varrho_{u,r}(t_{x})+\epsilon}{\widehat{G}_9(\bm{v};\ell)}{\le}1,~\forall r\in\mathcal{R}_{b},~\forall b\in\Omega,~\forall x\in\mathcal{N}^{(k)},~\forall k\in\mathcal{K}.\nonumber
\end{tcolorbox}

\noindent\textbf{Constraint~\eqref{app:cons:mac_p_3}:} We first rewrite \eqref{app:cons:mac_p_3} as the following two inequalities:
\begin{equation}\label{app:cons:mac_p_3_1}
    \begin{aligned}
        0\le \overline{\varrho}_{u,u',r}(t_{x})+\epsilon-\overline{\rho}^{\uparrow}_{u,r}(t_{x})-\epsilon
    \end{aligned}
\end{equation}
and
\begin{equation}\label{app:cons:mac_p_3_2}
    \begin{aligned}
        \overline{\varrho}_{u,u',r}(t_{x})-\overline{\rho}^{\uparrow}_{u,r}(t_{x})\le 1-\epsilon.
    \end{aligned}
\end{equation}
Performing some algebraic operations on the above two inequalities gives us

\begin{equation}\label{app:cons:mac_p_3_1}
    \begin{aligned}
         \frac{\overline{\rho}^{\uparrow}_{u,r}(t_{x})+\epsilon}{\overline{\varrho}_{u,u',r}(t_{x})+\epsilon}\le 1
    \end{aligned}
\end{equation}
and
\begin{equation}\label{app:cons:mac_p_3_2}
    \begin{aligned}
        \frac{\overline{\varrho}_{u,u',r}(t_{x})+\epsilon}{1+\overline{\rho}^{\uparrow}_{u,r}(t_{x})}\le 1. 
    \end{aligned}
\end{equation}

The fraction in~\eqref{app:cons:mac_p_3_2} is not in the format of GP since it is an inequality with a posynomial in the denominator, which is not a posynomial. We thus exploit arithmetic-geometric mean inequality (Lemma~\ref{Lemma:ArethmaticGeometric}) to approximate the denominator with a monomial:
\begin{align}\label{app:cons:mac_p_3_3}
    G_{10}(\bm{v})=1+\overline{\rho}^{\uparrow}_{u,r}(t_{x})&\geq \widehat{G}_{10}(\bm{v};\ell) \triangleq \left(G_{10}([\bm{v}]^{(\ell-1)})\right)^{\frac{1}{G_{10}([\bm{v}]^{(\ell-1)})}}\left(\frac{\overline{\rho}^{\uparrow}_{u,r}(t_{x}) G_{10}([\bm{v}]^{(\ell-1)})}{\left[\overline{\rho}^{\uparrow}_{u,r}(t_{x})\right]^{(\ell-1)}}\right)^{\frac{\left[\overline{\rho}^{\uparrow}_{u,r}(t_{x})\right]^{(\ell-1)} }{G_{10}([\bm{v}]^{(\ell-1)})}}.
\end{align}
We finally approximate the constraint \eqref{app:cons:mac_p_3} as follows:
\begin{tcolorbox}[ams align]
    &\frac{\overline{\rho}^{\uparrow}_{u,r}(t_{x})+\epsilon}{\overline{\varrho}_{u,u',r}(t_{x})+\epsilon}\le 1,~\frac{\overline{\varrho}_{u,u',r}(t_{x})+\epsilon}{\widehat{G}_{10}(\bm{v};\ell)}{\le}1,~\forall r\in\mathcal{R}_{b},~\forall b\in\Omega,~\forall x\in\mathcal{N}^{(k)},~\forall k\in\mathcal{K}.\nonumber
\end{tcolorbox}

\newpage
\subsection{Transformation of Constraints of GM and GPs Transmissions}\label{app:cons:GM_GP_transmissions}
Let us revisit the constraints of Sec.~\ref{sec:GM_GPS_transmissions} in the following:

\begin{equation}\label{app:cons:LGPD0}
    0\le \beta^{\downarrow}_{b}\hspace{-0.3mm}(t_{x})-\varphi_{b,r}\hspace{-0.3mm}(t_{x})<1,~x{\in}\mathcal{N}^{(k)}.
\end{equation}
\begin{equation}\label{app:cons:LGPD1}
\begin{aligned}
   0 \le \overline{\varrho}_{u,u',r}(t_{x})-\overline{\psi}_{u,u',r}(t_{x})< 1,~x{\in}\mathcal{N}^{(k)}.
\end{aligned}
\end{equation}
\begin{equation}\label{app:cons:LGPD2}
\begin{aligned}
   \overline{\psi}_{u,u,r}(t_{x})= 0,~~x{\in}\mathcal{N}^{(k)},\forall r \in \overline{\mathcal{R}}_{b}.
\end{aligned}
\end{equation}
\begin{equation}\label{app:cons:LMD1}
        0\le \varrho_{u,r}(t_{x})-\psi_{u,r}(t_{x})< 1,~x{\in}\mathcal{N}^{(k)}.
\end{equation}

In the following, we aim to transform \eqref{app:cons:LGPD0}, \eqref{app:cons:LGPD1}, \eqref{app:cons:LGPD2}, and \eqref{app:cons:LMD1} into GP format:

\noindent\textbf{Constraint~\eqref{app:cons:LGPD0}:} We first rewrite \eqref{app:cons:LGPD0} as the following two inequalities:
\begin{equation}\label{app:cons:LGPD0_1}
\begin{aligned}
   \varphi_{b,r}\hspace{-0.3mm}(t_{x}) \le \beta^{\downarrow}_{b}\hspace{-0.3mm}(t_{x}),
\end{aligned}
\end{equation}
and
\begin{equation}\label{app:cons:LGPD0_2}
\begin{aligned}
    \beta^{\downarrow}_{b}\hspace{-0.3mm}(t_{x})-\varphi_{b,r}\hspace{-0.3mm}(t_{x})\le 1-\epsilon.
\end{aligned}
\end{equation}
Performing some algebraic operations on\eqref{app:cons:LGPD0_2} gives us
\begin{equation}\label{app:cons:LGPD0_4}
\begin{aligned}
    \frac{\beta^{\downarrow}_{b}\hspace{-0.3mm}(t_{x})+\epsilon}{1+\varphi_{b,r}\hspace{-0.3mm}(t_{x})}\le 1.
\end{aligned}
\end{equation}

The fraction in~\eqref{app:cons:LGPD0_4} is not in the format of GP since it is an inequality with a posynomial in the denominator, which is not a posynomial. We thus exploit arithmetic-geometric mean inequality (Lemma~\ref{Lemma:ArethmaticGeometric}) to approximate the denominator with a monomial:
\begin{align}\label{app:cons:LGPD0_5}
    J_{1}(\bm{v})=1+\varphi_{b,r}\hspace{-0.3mm}(t_{x})&\geq \widehat{J}_{1}(\bm{v};\ell) \triangleq \left(J_{1}([\bm{v}]^{(\ell-1)})\right)^{\frac{1}{J_{1}([\bm{v}]^{(\ell-1)})}}\left(\frac{\varphi_{b,r}\hspace{-0.3mm}(t_{x}) J_{1}([\bm{v}]^{(\ell-1)})}{\left[\varphi_{b,r}\hspace{-0.3mm}(t_{x})\right]^{(\ell-1)}}\right)^{\frac{\left[\varphi_{b,r}\hspace{-0.3mm}(t_{x})\right]^{(\ell-1)} }{J_{1}([\bm{v}]^{(\ell-1)})}}.
\end{align}
We finally approximate the constraint \eqref{app:cons:LGPD0} as follows:
\begin{tcolorbox}[ams align]
    &\frac{\varphi_{b,r}\hspace{-0.3mm}(t_{x})}{\beta^{\downarrow}_{b}\hspace{-0.3mm}(t_{x})} \le 1,~\frac{\beta^{\downarrow}_{b}\hspace{-0.3mm}(t_{x})+\epsilon}{\widehat{J}_{1}(\bm{v};\ell)}{\le}1,~\forall r\in\mathcal{R}_{b},~\forall b\in\Omega,~\forall x\in\mathcal{N}^{(k)},~\forall k\in\mathcal{K}.\nonumber
\end{tcolorbox}

\noindent\textbf{Constraint~\eqref{app:cons:LGPD1}:} We first rewrite \eqref{app:cons:LGPD1} as the following two inequalities:
\begin{equation}\label{app:cons:LGPD1_1}
\begin{aligned}
   \overline{\psi}_{u,u',r}(t_{x}) \le \overline{\varrho}_{u,u',r}(t_{x})
\end{aligned}
\end{equation}
and
\begin{equation}\label{app:cons:LGPD1_2}
\begin{aligned}
    \overline{\varrho}_{u,u',r}(t_{x})-\overline{\psi}_{u,u',r}(t_{x})\le 1-\epsilon.
\end{aligned}
\end{equation}
Performing some algebraic operations on \eqref{app:cons:LGPD1_2} gives us
\begin{equation}\label{app:cons:LGPD1_4}
\begin{aligned}
    \frac{\overline{\varrho}_{u,u',r}(t_{x})+\epsilon}{1+\overline{\psi}_{u,u',r}(t_{x})}\le 1.
\end{aligned}
\end{equation}

The fraction in~\eqref{app:cons:LGPD1_4} is not in the format of GP since it is an inequality with a posynomial in the denominator, which is not a posynomial. We thus exploit arithmetic-geometric mean inequality (Lemma~\ref{Lemma:ArethmaticGeometric}) to approximate the denominator with a monomial:
\begin{align}\label{app:cons:LGPD1_5}
    J_{1}(\bm{v})=1+\overline{\psi}_{u,u',r}(t_{x})&\geq \widehat{J}_{1}(\bm{v};\ell) \triangleq \left(J_{1}([\bm{v}]^{(\ell-1)})\right)^{\frac{1}{J_{1}([\bm{v}]^{(\ell-1)})}}\left(\frac{\overline{\psi}_{u,u',r}(t_{x}) J_{1}([\bm{v}]^{(\ell-1)})}{\left[\overline{\psi}_{u,u',r}(t_{x})\right]^{(\ell-1)}}\right)^{\frac{\left[\overline{\psi}_{u,u',r}(t_{x})\right]^{(\ell-1)} }{J_{1}([\bm{v}]^{(\ell-1)})}}.
\end{align}
We finally approximate the constraint \eqref{app:cons:mac_p_3} as follows:
\begin{tcolorbox}[ams align]
    &\frac{\overline{\psi}_{u,u',r}(t_{x})}{\overline{\varrho}_{u,u',r}(t_{x})} \le 1,~\frac{\overline{\varrho}_{u,u',r}(t_{x})+\epsilon}{\widehat{J}_{1}(\bm{v};\ell)}{\le}1,~\forall r\in\mathcal{R}_{b},~\forall b\in\Omega,~\forall x\in\mathcal{N}^{(k)},~\forall k\in\mathcal{K}.\nonumber
\end{tcolorbox}

\noindent\textbf{Constraint~\eqref{app:cons:LMD1}:} We first rewrite \eqref{app:cons:LMD1} as the following two inequalities:
\begin{equation}\label{app:cons:LMD1_1}
        \psi_{u,r}(t_{x})\le \varrho_{u,r}(t_{x})
\end{equation}
and
\begin{equation}\label{app:cons:LMD1_2}
        \varrho_{u,r}(t_{x})-\psi_{u,r}(t_{x})\le 1-\epsilon.
\end{equation}
Performing some algebraic operations on \eqref{app:cons:LMD1_2} gives us

\begin{equation}\label{app:cons:LMD1_4}
        \frac{\varrho_{u,r}(t_{x})+\epsilon}{1+\psi_{u,r}(t_{x})}\le 1.
\end{equation}

The fraction in~\eqref{app:cons:LMD1_4} is not in the format of GP since it is an inequality with a posynomial in the denominator, which is not a posynomial. We thus exploit arithmetic-geometric mean inequality (Lemma~\ref{Lemma:ArethmaticGeometric}) to approximate the denominator with a monomial:
\begin{align}\label{app:cons:LMD1_5}
    J_{2}(\bm{v})=1+\psi_{u,r}(t_{x})&\geq \widehat{J}_{2}(\bm{v};\ell) \triangleq \left(J_{2}([\bm{v}]^{(\ell-1)})\right)^{\frac{1}{J_{2}([\bm{v}]^{(\ell-1)})}}\left(\frac{\psi_{u,r}(t_{x})}{\left[\psi_{u,r}(t_{x})\right]^{(\ell-1)}}\right)^{\frac{\left[\psi_{u,r}(t_{x})\right]^{(\ell-1)} }{J_{2}([\bm{v}]^{(\ell-1)})}}.
\end{align}
We finally approximate the constraint \eqref{app:cons:LMD1} as follows:
\begin{tcolorbox}[ams align]
    &\frac{\psi_{u,r}(t_{x})}{\varrho_{u,r}(t_{x})}\le 1,~\frac{\varrho_{u,r}(t_{x})+\epsilon}{\widehat{J}_{1}(\bm{v};\ell)}{\le}1,~\forall r\in\mathcal{R}_{b},~\forall b\in\Omega,~\forall x\in\mathcal{N}^{(k)},~\forall k\in\mathcal{K}.\nonumber
\end{tcolorbox}

\newpage
\subsection{Pseudo Code of Our GP Optimization Solver}\label{app:cons:sudo}
\begin{algorithm}[h]
 	\caption{Proposed Signomial programming and GP-based optimization solver for problem~$\bm{\mathcal{P}}$}\label{alg:cent}
 	\SetKwFunction{Union}{Union}\SetKwFunction{FindCompress}{FindCompress}
 	\SetKwInOut{Input}{input}\SetKwInOut{Output}{output}
 	 	{\footnotesize
 	\Input{Convergence criterion.}
 	 Set the algorithm iteration counter to be $\ell=0$.\\
 	 Choose an initial feasible point as the starting point $\bm{v}^{[0]}$.\\
          Obtain the highlighted approximations, indicated by green boxes, from page \pageref{app:data_transformation_rate} to page \pageref{app:cons:LMD1_5}, using the current solution $\bm{v}^{[\ell]}$.\label{midAlg1}\\
          Use those approximations instead of their original constraints in Problem~$\bm{\mathcal{P}}$, which results in a GP problem.\\
          Utilize logarithmic variable transformation to convert the resultant GP problem into a convex problem, following the instructions in Appendix~\ref{sec:GPtransConv}.\\
          Increment $\ell$ by one, i.e., $\ell=\ell+1$. \\
          Utilize state-of-the-art convex optimization solver framework (such as CVXPY~\cite{diamond2016cvxpy}) to solve the obtained convex problem and derive $\bm{v}^{[\ell]}$.\label{Alg:Gpconvexste}\\
 	 \If{successive solutions $\bm{v}^{[\ell-1]}$ and $\bm{v}^{[\ell]}$ fail to satisfy the defined convergence criterion}{
 	\textrm{Go to line~\ref{midAlg1} and re-iterate over the steps using the most recent solution $\bm{v}^{[\ell]}$.}\\\Else{Set the most recent solution as the final solution of the problem $\bm{v}^{\star}=\bm{v}^{[\ell]}$.\label{Alg:point2}\\
 	}}
 	  }
 	  \vspace{-.1mm}
  \end{algorithm}

 \begin{proposition}[Convergence of the Optimization Solver]\label{propo:KKT}
    Our GP-based solver (i.e. Algorithm~\ref{alg:cent} presented in Appendix.~\ref{app:cons:sudo}) generates a sequence of solutions for $\bm{\mathcal{P}}$ using the approximations~ (a-i) arithmetic-geometric mean inequality \eqref{eq:approxPosMonMain}, (a-ii) Taylor-power approximation \eqref{eq:taylor_e}, and (a-iii) sum-power approximations $\min\{A, B\}\approx (A^{-p}+B^{-p})^{-\frac{1}{p}}$ and $\max\{A, B\}\approx (A^{p}+B^{p})^{-\frac{1}{p}}$. Under tight approximations of (a-ii) and (a-iii), these solutions (i.e., $\bm{v}^\star$) converge to the Karush–Kuhn–Tucker (KKT) conditions of $\bm{\mathcal{P}}$ upon using (a-i).
\end{proposition}
\begin{proof}

Under the tight approximations of (a-ii) and (a-iii), Algorithm~\ref{alg:cent} leads to an \textit{inner approximation}~\cite{GeneralInnerApp} of~$\bm{\mathcal{P}}$. It is thus sufficient to examine the three characteristics mentioned in~\cite{GeneralInnerApp} for our solver to prove the convergence, which can be done using the methods in~\cite{SignomialGlobal} and omitted for brevity.
\end{proof}

\newpage
\section{Complexity Analysis and Feasibility of {\tt DCLM}}\label{app:complexity}
\subsection{Complexity Analysis of {\tt DCLM}}
{ As we demonstrated in Appendix~\ref{app:optTransform}, our optimization problem $(\bm{\mathcal{P}})$, presented in \eqref{opt:main}, can be transformed to a GP optimization problem. This GP can be further reformulated as a convex optimization problem through a logarithmic transformation~\cite{nesterov1994interior, chiang2005geometric}.  CVXPY then leverages interior-point methods (IPMs) to solve the convex reformulation of this GP problem by iteratively updating the solution using the Karush-Kuhn-Tucker (KKT) conditions (see Proposition~5 in Appendix~\ref{app:optTransform}). In this approach, each iteration will involve solving a Newton system derived from the KKT conditions, which requires solving a linear system with a complexity of approximately \(O(n^2)\), where \(n\) is the size of the problem, i.e., the sum of the number of variables and constraints (for details refer to \cite{nesterov1994interior}). Subsequently, the number of iterations required for IPMs in convex optimization to achieve a desired solution accuracy of \(\epsilon > 0\)  is given by \(O\left(\sqrt{n} \log\left(\frac{1}{\epsilon}\right)\right)\) \cite{nesterov1994interior}. As a result, the complexity of solving GP optimization problem using CVXPY is \(O\left(n^{2.5} \log\left(\frac{1}{\epsilon}\right)\right)\), which is polynomial with respect to $n$ making it a good candidate for large problem sizes. Further, note that implementing {\tt DCLM} does not need following all the detailed mathematical steps presented in Appendix~\ref{app:optTransform}; instead, \textit{it relies solely on the closed-form expressions}, which are highlighted in green boxes throughout Appendix~\ref{app:optTransform}, ensuring that implementation of {\tt DCLM} is mathematically tractable.}

\subsection{Feasibility of Implementing {\tt DCLM} in Real-World Scenarios}\label{app:complexity:feasibility}
{ Feasibility of realizing our proposed {\tt DCLM} method in real world scenarios can be investigated from two angles. First, as we discuss in Appendix Sec. A of Appendix~\ref{app:complexity}, the complexity of solving optimization problem $\bm{\mathcal{P}}$ using CVXPY is \textit{polynomial}, making {\tt DCLM} computationally feasible. It is worth noting that modeling an engineering optimization problem using geometric programming (GP), particularly for large-scale wireless networks entailing large numbers of variables and constraints, is widely exercised~\cite{chiang2005geometric}.

Second, implementing {\tt DCLM} does not require conducting of  all the detailed mathematical steps presented in Appendix~\ref{app:optTransform}; instead, it relies \underline{\textit{solely on the closed-form expressions}}, which are highlighted in green boxes throughout Appendix~\ref{app:optTransform}, ensuring that implementation of {\tt DCLM} is practical. Please note that the detailed mathematical steps in Appendix~\ref{app:optTransform} are intended to provide a comprehensive theoretical reference to demonstrate how such closed-form expressions can be obtained for the complex optimization problems of the form considered in this paper. This step-by-step analysis can be further used as a reference for future works.}

\end{document}